\documentclass[10pt,titlepage,oneside,final]{book}
 
\newcommand{\href}[1]{#1} 

\usepackage{ifthen}
\newboolean{PrintVersion}
\setboolean{PrintVersion}{false} 

\usepackage{nomencl} 
\usepackage{amsmath,amssymb,amstext} 
\usepackage[pdftex]{graphicx} 
\usepackage[usenames,dvipsnames]{xcolor}

\usepackage[pdftex,pagebackref=true]{hyperref} 
\renewcommand*{\backref}[1]{}
\renewcommand*{\backrefalt}[4]{$[$%
    \ifcase #1 Not cited.%
          \or cited on p.~#2.%
          \else cited on pp.~#2.%
    \fi%
    $]$}
\hypersetup{
    plainpages=false,       
    unicode=false,          
    pdftoolbar=true,        
    pdfmenubar=true,        
    pdffitwindow=false,     
    pdfstartview={FitH},    
    pdftitle={Conical Designs and Categorical Jordan Algebraic Post-Quantum Theories},    
    pdfauthor={Matthew A. Graydon},    
    pdfsubject={Theoretical Physics},  
    pdfnewwindow=true,      
    colorlinks=true,        
    linkcolor=blue,         
    citecolor=green,        
    filecolor=magenta,      
    urlcolor=blue!90!black         
}
\ifthenelse{\boolean{PrintVersion}}{   
\hypersetup{	
    citecolor=black,%
    filecolor=black,%
    linkcolor=black,%
    urlcolor=black}
}{} 

\let\origdoublepage\cleardoublepage
\newcommand{\clearemptydoublepage}{%
  \clearpage{\pagestyle{empty}\origdoublepage}}
\let\cleardoublepage\clearemptydoublepage

\usepackage{bookmark,epigraph,float,nicefrac,mathrsfs,eucal,dsfont}

\usepackage{amsthm}
\usepackage[capitalize,noabbrev]{cleveref}
\newtheoremstyle{mgStyle}{\medskipamount}{\medskipamount}{}{}{\bfseries\itshape}{}{0.5em}{\thmname{#1} \thmnumber{#2}}
\usepackage{enumerate}

\theoremstyle{mgStyle}
\newtheorem{theorem}{\textit{Theorem}}[section]
\newtheorem{definition}[theorem]{\textit{Definition}}
\newtheorem{corollary}[theorem]{\textit{Corollary}}
\newtheorem{lemma}[theorem]{\textit{Lemma}}
\newtheorem{proposition}[theorem]{\textit{Proposition}}

\newcommand*{\jProd}{\raisebox{-0.88ex}{\scalebox{2.5}{$\cdot$}}}
\usepackage{bbm}
\usepackage[all,cmtip]{xy}
\usepackage{marvosym}

\newcommand{\cc}{{\cal C}}
\newcommand{\Cu}{C^{\ast}_{u}} 
\newcommand{\Tr}{\mbox{Tr}}
\newcommand{\tr}{\mbox{tr}}
\newcommand{\R}{\mathbb{R}}

\newcommand{\M}{\mathbf{M}}
\newcommand{\N}{\mathbf{N}}

\newcommand{\C}{\mathbb{C}}
\newcommand{\Q}{\mathbb{H}}
\newcommand{\Cat}{\mathcal{C}}

\newcommand{\V}{\mathbf{V}}

\newcommand{\beq}{\begin{equation}}
\newcommand{\eeq}{\end{equation}}

\newcommand{\EJC}{\mbox{\bf EJC}}

\newcommand{\CJP}{\mbox{\bf CJP}}

\newcommand{\RSE}{\mbox{\bf RSE}}
\newcommand{\URUE}{\mbox{\bf URUE}}

\newcommand{\CQM}{\mbox{$\mathbb C$QM}}
\newcommand{\stalg}{\mbox{\bf $\ast$-Alg}}
\newcommand{\EJA}{\mbox{\bf EJA}}
\newcommand{\Cstar}{\mbox{\bf Cstar}}
\newcommand{\InvQM}{\mbox{\red \bf InvQM}}

\newcommand{\IP}{\langle \  | \ \rangle}

\renewcommand{\bar}{\overline}
\renewcommand{\hat}{\widehat}
\newcommand{\hotimes}{\widetilde{\otimes}}
\newcommand{\otilde}{\widetilde{\otimes}}

\newcommand{\id}{\mbox{id}}
\newcommand{\0}{\mathbf{0}}
\newcommand{\sa}{{sa}}

\newcommand{\Aut}{\mbox{Aut}}
\newcommand{\Sym}{\mbox{Sym}}

\newcommand{\potimes}{\circledcirc}

\newcommand{\blue}{}
\newcommand{\redd}{}
\newcommand{\red}{}
\newcommand{\green}{}
\newcommand{\magenta}{}

\newcommand{\tempout}[1]{{}}
\newcommand{\ffootnote}{\tempout}
\begin{document}

\pagestyle{empty}
\pagenumbering{roman}
\begin{center}
\begin{large}{\textbf{Preface to the arXiv version}}\end{large}
\end{center}
This thesis was accepted by the University of Waterloo (Ontario, Canada) on January 6, 2017. The main body of this thesis has two parts. \cref{partI} centres on collaborative work with D.\ M.\ Appleby (\href{https://arxiv.org/abs/1507.05323}{arXiv:1507.05323 [quant-ph]} and \href{https://arxiv.org/abs/1507.07881}{arXiv:1507.07881 [quant-ph]}). \cref{partII} centres on collaborative work with H.\ Barnum and A.\ Wilce (\href{https://arxiv.org/abs/1507.06278}{arXiv:1507.06278 [quant-ph]} and \href{https://arxiv.org/abs/1606.09331}{arXiv:1606.09331 [quant-ph]}). In posting this thesis on the arXiv, I hope that it might serve as a useful resource for those interested in its subjects.
\begin{flushright}
Matthew A.\ Graydon\\
Waterloo, Ontario, Canada, 2017
\end{flushright}
\cleardoublepage
\begin{titlepage}
        \begin{center}
        \vspace*{1.0cm}
			
        \LARGE
        {\bf Conical Designs and Categorical Jordan Algebraic Post-Quantum Theories}\\
        \vspace*{1.0cm}

        \normalsize
        by \\

        \vspace*{1.0cm}

        \large
        Matthew A.\ Graydon  \\

        \vspace*{3.0cm}

        \normalsize
        A thesis \\
        presented to the University of Waterloo \\ 
        in fulfillment of the \\
        thesis requirement for the degree of \\
        Doctor of Philosophy \\
        in \\
        Physics - Quantum Information \\

        \vspace*{2.0cm}

        Waterloo, Ontario, Canada, 2017 \\

        \vspace*{1.0cm}

        \copyright\ Matthew A.\ Graydon 2017 \\
        \end{center}
\end{titlepage}

\pagestyle{plain}
\setcounter{page}{2}

\cleardoublepage 

  \noindent
This thesis consists of material all of which I authored or co-authored: see Statement of Contributions included in the thesis. This is a true copy of the thesis, including any required final revisions, as accepted by my examiners.

  \bigskip
  
  \noindent
I understand that my thesis may be made electronically available to the public.
\cleardoublepage
\begin{center}
\textbf{Statement of Contributions}
\end{center}

\noindent The material in \cref{partI} centres on collaborative work with D.\ M.\ Appleby:
\begin{enumerate}
\item M.\ A.\ Graydon and D.\ M.\ Appleby:\\ ``\textit{Quantum Conical Designs}.''\\ Published in \cite{Graydon2015a}; electronic preprint available at \href{http://arxiv.org/abs/1507.05323}{arXiv:1507.05323 [quant-ph]}
\item M.\ A.\ Graydon and D.\ M.\ Appleby:\\ ``\textit{Entanglement and Designs}.''\\ Published in \cite{Graydon2016}; electronic preprint available at \href{http://arxiv.org/abs/1507.07881}{arXiv:1507.07881 [quant-ph]}
\end{enumerate}

\noindent The manuscripts for \cite{Graydon2015a} and \cite{Graydon2016} were largely prepared by myself. Portions of \cite{Graydon2015a} and \cite{Graydon2016} have been adapted for inclusion herein with the consent of D.\ M.\ Appleby.

\noindent The material in \cref{partII} centres on collaborative work with H.\ Barnum and A.\ Wilce:
\begin{enumerate}
\setcounter{enumi}{2}
\item H.\ Barnum, M.\ A.\ Graydon, and A.\ Wilce:\\ ``\textit{Some Nearly Quantum Theories}.''\\ Published in \cite{Barnum2015}; electronic preprint available at \href{http://arxiv.org/abs/1507.06278}{arXiv:1507.06278 [quant-ph]}
\item H.\ Barnum, M.\ A.\ Graydon, and A.\ Wilce:\\ ``\textit{Composites and Categories of Euclidean Jordan Algebras}.''\\
Forthcoming \cite{Barnum2016b}; electronic preprint available at \href{http://arxiv.org/abs/1606.09331}{arXiv:1606.09331 [quant-ph]}
\end{enumerate}

\noindent The manuscripts for \cite{Barnum2015} and \cite{Barnum2016b} were largely prepared by H.\ Barnum and A.\ Wilce. Portions of \cite{Barnum2015} and \cite{Barnum2016b} have been adapted for inclusion herein with the consent of H.\ Barnum and A.\ Wilce.


\cleardoublepage


\begin{center}\textbf{Abstract}\end{center}

\noindent Physical theories can be characterized in terms of their state spaces and their evolutive equations. The kinematical structure and the dynamical structure of finite dimensional quantum theory are, in light of the Choi-Jamio{\l}kowski isomorphism, one and the same --- namely the homogeneous self-dual cones of positive semi-definite linear endomorphisms on finite dimensional complex Hilbert spaces. From the perspective of category theory, these cones are the sets of morphisms in finite dimensional quantum theory as a dagger compact closed category. Understanding the intricate geometry of these cones and charting the wider landscape for their host category is imperative for foundational physics. 

\noindent In Part I of this thesis, we study the shape of finite dimensional quantum theory in terms of quantum information. We introduce novel geometric structures inscribed within quantum cones: conical $t$-designs. Conical $t$-designs are a natural, strictly inclusive generalization of complex projective $t$-designs. We prove that symmetric informationally complete measurements of arbitrary rank (\textsc{sim}s), and full sets of mutually unbiased measurements of arbitrary rank (\textsc{mum}s) are conical 2-designs.\! \textsc{sim}s and \textsc{mum}s correspond to highly symmetric polytopes within the Bloch body. The same holds for the entire class of homogeneous conical 2-designs; moreover, we establish necessary and sufficient conditions for a Bloch polytope to represent a homogeneous conical 2-design. Furthermore, we show that infinite families of such designs exist in all finite dimensions. It turns out that conical 2-designs are naturally adapted to a geometric description of bipartite entanglement. We prove that a quantum measurement is a conical 2-design if and only if there exists a (regular) entanglement monotone whose restriction to pure states is a function of the norm of the probability vector over the outcomes of the bipartite measurement formed from its tensor products. In that case the concurrence is such a monotone. In addition to monotones, we formulate entanglement witnesses in terms of geometric conditions on the aforementioned conical 2-design probabilities.

\noindent In Part II of this thesis, we move beyond quantum theory within the vein of Euclidean Jordan algebras (\textsc{eja}s). In light of the Koecher-Vinberg theorem, the positive cones of \textsc{eja}s are the only homogeneous self-dual cones to be found in a finite dimensional setting. We consider physical theories based on \textsc{eja}s subject to nonsignalling axioms regarding their compositional structure. We prove that any Jordanic composite is a Jordan ideal of Hanche-Olsen's universal tensor product. Consequently, no Jordanic composite exists having the exceptional Jordan algebra as a direct summand, nor does any such composite exist if either factor is exceptional. So we focus on special \textsc{eja}s of self-adjoint matrices over the real, complex, and quaternionic division rings. We demonstrate that these can be organized in a natural way as a symmetric monoidal category, albeit one that is not compact closed. We then construct a related category $\InvQM$ of embedded \textsc{eja}s, having fewer objects but more morphisms, that is dagger compact closed. This category unifies finite dimensional real, complex and quaternionic quantum theories, except that the composite of two complex quantum systems comes with an extra classical bit. Our notion of composite requires neither tomographic locality, nor preservation of purity under monoidal products. The categories we construct include examples in which both of these conditions fail. Our unification cannot be extended to include any finite dimensional spin factors (save the rebit, qubit, and quabit) without destroying compact closure.

\cleardoublepage


\begin{center}\textbf{Acknowledgements}\end{center}

\noindent Dr.\ Robert W.\ Spekkens supervised my PhD. In addition to funding the final year of my degree, Rob bestowed upon me a much more precious gift: the gift of complete academic freedom. It is so rare to be able to live out one's dreams without constraints, but Rob made just that happen in my life. I am, indeed, tremendously fortunate to have been one of his students. The organization and expression of this thesis have both been directly influenced by Rob's invaluable feedback and sage advice. Thank you, Rob.

\noindent My PhD Defence Committee consisted of Dr.\ Ross Duncan, Dr.\ Norbert L{\"u}tkenhaus, Dr.\ Kevin Resch, Dr.\ Robert W.\ Spekkens, and Dr.\ Jon Yard. I feel privileged to have been granted their feedback. I sincerely thank each of them for their time, and for agreeing to act in their respective roles for my defence.

\noindent My PhD Advisory Committee consisted of Dr.\ Joseph Emerson, Dr.\ Marco Piani, Dr.\ Kevin Resch, and Dr.\ Robert W.\ Spekkens. In particular, Dr.\ Kevin Resch acted as co-supervisor for my PhD degree, and I thank him especially for this. I thank Dr.\ Joseph Emerson, Dr.\ Marco Piani, and Dr.\ Robert W.\ Spekkens for their questions on my PhD comprehensive examination. Thank you all for your lessons and guidance.

\noindent My PhD was funded in part by the Government of Canada through an NSERC Alexander Graham Bell Canada Graduate Scholarship, and by the Province of Ontario through an Ontario Graduate Scholarship. I also received funding from the University of Waterloo via President's Graduate Scholarships and Marie Curie Graduate Student Awards. The Institute for Quantum Computing granted me an IQC Entrance Award. My primary office was at the Perimeter Institute for Theoretical Physics. Research at Perimeter Institute is supported by the Government of Canada through the Department of Innovation, Science and Economic Development and by the Province of Ontario through the Ministry of Research \& Innovation.

\noindent This thesis is a direct consequence of intense collaborative work with three very special human beings, acknowledged alphabetically: Dr.\ Marcus Appleby, Dr.\ Howard Barnum, and Dr.\ Alexander Wilce. I view these people as mentors and friends, and I am extraordinarily lucky to have met them. I look forward to continuing our collaborations in the future. Thank you, Marcus. Thank you, Howard. Thank you, Alex.

\noindent Dr.\ Christopher A.\ Fuchs introduced me to quantum foundations. Chris is the reason why I was drawn to designs. Over the course of years, he has instilled in me his ideas regarding the importance of geometry in quantum theory. No one has made a greater impact on my development in phySICs. Thank you, Chris.

\noindent My friends and academic brothers Dr.\ Hoan Bui Dang and Dr.\ Gelo Noel M. Tabia are wonderful people. I have learned so much from them, both personally and academically. These guys filled my graduate school years with some really fun times. Thank you, Hoan. Thank you, Gelo. 

\noindent Dad and Betty have been unwavering pillars of support for me throughout my studies. I could not have done this without their help. I love them both very much. Thank you, Dad. Thank you, Betty.

\noindent Sarah and JB are pretty cool too. My younger sister is many years ahead of me where it counts, and her husband is the kind of man that I want to be like. They inspire me. Thank you, Sarah. Thank you, JB.

\noindent This thesis would not exist without my dear wife, Candice. Her support has literally kept me alive, and she carried me to the finish line, as it were, during the final stages of writing this thesis. Candice is a beautiful human being. We fell in love as teenagers, and we have since passed through our high school, BSc, MSc, and PhD years; always together. I will be forever grateful for our love. Thank you, Candice.

\cleardoublepage


\begin{center}\textbf{Dedication}\end{center}

\begin{center}
\textit{I dedicate this thesis with love to my wife,\\ Candice.}\\
\end{center}
\cleardoublepage

\renewcommand\contentsname{Table of Contents}
\tableofcontents
\cleardoublepage
\phantomsection


\addtocontents{toc}{\vspace{-1cm}}
 \addcontentsline{toc}{chapter}{\textbf{List of Symbols}}
 \cleardoublepage
 \phantomsection 
 \newpage
\noindent\begin{huge}{\textbf{List\footnote{This list is not exhaustive.} of Symbols}}\end{huge}

\noindent\begin{tabular}{ll}
$\mathcal{H}_{d}$\hspace{1cm}& complex Hilbert space of finite dimension $d$\\[0.2cm]
$\mathcal{H}_{d}^{\star}$\hspace{1cm}& dual complex Hilbert space of all continuous $\mathbb{C}$-linear functionals on $\mathcal{H}_{d}$\\[0.2cm]
$\langle\cdot|\cdot\rangle$\hspace{1cm}& inner product on $\mathcal{H}_{d}$\\[0.2cm]
$\mathcal{S}\big(\mathcal{H}_{d}\big)$\hspace{1cm}& unit sphere in $\mathcal{H}_{d}$\\[0.2cm]
$\mathrm{U}\big(\mathcal{H}_{d}\big)$\hspace{1cm}& complex unitary group of degree $d$\\[0.2cm]
$\mathcal{L}\big(\mathcal{H}_{d}\big)$\hspace{1cm}& complex Hilbert space of linear endomorphisms on $\mathcal{H}_{d}$\\[0.2cm]
$\mathds{1}_{d}$\hspace{1cm}& identity endomorphism on $\mathcal{H}_{d}$\\[0.2cm]
$\langle\!\langle\cdot|\cdot\rangle\!\rangle$\hspace{1cm}& inner product on $\mathcal{L}\big(\mathcal{H}_{d}\big)$\\[0.2cm]
$\mathcal{L}_{\text{sa}}\big(\mathcal{H}_{d}\big)$\hspace{1cm}& real subspace of self-adjoint linear endomorphisms on $\mathcal{H}_{d}$\\[0.2cm]
$\mathcal{L}_{\text{sa},0}\big(\mathcal{H}_{d}\big)$\hspace{1cm}& null trace subspace of $\mathcal{L}_{\text{sa}}\big(\mathcal{H}_{d}\big)$\\[0.2cm]
$\mathcal{L}_{\text{sa},1}\big(\mathcal{H}_{d}\big)$\hspace{1cm}& unit trace hyperplane of $\mathcal{L}_{\text{sa}}\big(\mathcal{H}_{d}\big)$\\[0.2cm]
$\mathcal{L}_{\text{sa}}\big(\mathcal{H}_{d}\big)_{+}$\hspace{1cm}& positive cone of $\mathcal{L}_{\text{sa}}\big(\mathcal{H}_{d}\big)$\\[0.2cm]
$\mathcal{Q}(\mathcal{H}_{d}\big)$\hspace{1cm}& convex set of quantum states, \textit{i.e}.\ $\mathcal{L}_{\text{sa}}\big(\mathcal{H}_{d}\big)_{+}\cap\;\mathcal{L}_{\text{sa},1}\big(\mathcal{H}_{d}\big)$\\[0.2cm]
$\text{Pur}\mathcal{Q}(\mathcal{H}_{d}\big)$\hspace{1cm}& manifold of pure quantum states\\[0.2cm]
$\mathcal{E}(\mathcal{H}_{d}\big)$\hspace{1cm}& convex set of quantum effects, \textit{i.e}.\ $\mathcal{E}\equiv\big\{E\in\mathcal{L}_{\text{sa}}(\mathcal{H}_{d})_{+}\;\boldsymbol{|}\; \mathds{1}_{d}-E\geq 0 \big\}$\\[0.2cm]
$\boldsymbol{\mathcal{B}}(\mathcal{H}_{d})$\hspace{1cm}& the Bloch body, \textit{i.e}.\ all those $B\in\mathcal{L}_{\text{sa},0}(\mathcal{H}_{d})$ such that $\mathcal{Q}(\mathcal{H}_{d})\ni(B+\mathds{1}_{d})/d$\\[0.2cm]
$\boldsymbol{\mathcal{B}}(\mathcal{H}_{d})_{\text{in}}$\hspace{1cm}& the inball, \textit{i.e}.\ the largest ball centred on zero in $\mathcal{L}_{\text{sa},0}(\mathcal{H}_{d})$ contained in $\boldsymbol{\mathcal{B}}(\mathcal{H}_{d})$\\[0.2cm]
$\boldsymbol{\mathcal{B}}(\mathcal{H}_{d})_{\text{out}}$\hspace{1cm}& the outball, \textit{i.e}.\ the smallest ball centred on zero in $\mathcal{L}_{\text{sa},0}(\mathcal{H}_{d})$ containing $\boldsymbol{\mathcal{B}}(\mathcal{H}_{d})$\\[0.2cm]
$\boldsymbol{\mathcal{S}}(\mathcal{H}_{d})_{\text{in}}$\hspace{1cm}& the surface of the inball\\[0.2cm]
$\boldsymbol{\mathcal{S}}(\mathcal{H}_{d})_{\text{out}}$\hspace{1cm}& the surface of the outball\\[0.2cm]
$\mathsf{Lin}(d_{\mathrm{A}},d_{\mathrm{B}})$\hspace{1cm}& complex Hilbert space of linear homomorphisms taking $\mathcal{L}\big(\mathcal{H}_{d_{\mathrm{A}}}\big)$ into $\mathcal{L}\big(\mathcal{H}_{d_{\mathrm{B}}}\big)$\\[0.2cm]
$\mathbf{I}_{d}$\hspace{1cm}& identity endomorphism on $\mathcal{L}(\mathcal{H}_{d})$\\[0.2cm]
$\langle\!\langle\!\langle\!\langle\cdot|\cdot\rangle\!\rangle\!\rangle\!\rangle$\hspace{1cm}& inner product on $\mathsf{Lin}(d_{\mathrm{A}},d_{\mathrm{B}})$\\[0.2cm]
$\mathsf{TP}(d_{\mathrm{A}},d_{\mathrm{B}})$\hspace{1cm}& trace preserving elements of $\mathsf{Lin}(d_{\mathrm{A}},d_{\mathrm{B}})$\\[0.2cm]
$\mathsf{CP}(d_{\mathrm{A}},d_{\mathrm{B}})$\hspace{1cm}& completely positive elements of $\mathsf{Lin}(d_{\mathrm{A}},d_{\mathrm{B}})$\\[0.2cm]
$\mathsf{CPTP}(d_{\mathrm{A}},d_{\mathrm{B}})$\hspace{1cm}& completely positive trace preserving elements of $\mathsf{Lin}(d_{\mathrm{A}},d_{\mathrm{B}})$\\[0.2cm]
$\mathcal{H}_{d_{\mathrm{A}}}\otimes\mathcal{H}_{d_{\mathrm{B}}}$\hspace{1cm}& tensor product of complex Hilbert spaces of finite dimensions $d_{\mathrm{A}}$ and $d_{\mathrm{B}}$\\[0.2cm]
$\text{Sep}\mathcal{Q}(\mathcal{H}_{d_{\mathrm{A}}}\otimes\mathcal{H}_{d_{\mathrm{B}}})$\hspace{1cm}& convex set of separable quantum states within $\mathcal{Q}(\mathcal{H}_{d_{\mathrm{A}}}\otimes\mathcal{H}_{d_{\mathrm{B}}})$\\[0.2cm]
$\boldsymbol{\mathcal{J}}$\hspace{1cm}& Choi-Jamio{\l}kowski isomorphism injecting $\mathsf{CPTP}(d_{\mathrm{A}},d_{\mathrm{B}})$ onto $\mathcal{Q}(\mathcal{H}_{d_{\mathrm{A}}}\otimes\mathcal{H}_{d_{\mathrm{B}}})$
\end{tabular}	

\cleardoublepage
\noindent\begin{tabular}{ll}
$\Pi_{\text{sym}}$\hspace{1cm}& rank $d(d+1)/2$ projector onto the symmetric subspace of $\mathcal{H}_{d_{\mathrm{A}}}\otimes\mathcal{H}_{d_{\mathrm{B}}}$\\[0.2cm]
$\Pi_{\text{asym}}$\hspace{1cm}& rank $d(d-1)/2$ projector onto the antisymmetric subspace of $\mathcal{H}_{d_{\mathrm{A}}}\otimes\mathcal{H}_{d_{\mathrm{B}}}$\\[0.2cm]
$\mathsf{E}$\hspace{1cm}& arbitrary entanglement monotone\\[0.2cm]
$\mathsf{e}$\hspace{1cm}& restriction of $\mathsf{E}$ to pure states\\[0.2cm]
$\mathsf{C}$\hspace{1cm}& concurrence\\[0.2cm]
$\mathsf{c}$\hspace{1cm}& restriction of $\mathsf{c}$ to pure states\\[0.2cm]
$\jProd$\hspace{1cm}& the Jordan product\\[0.2cm]
$\mathcal{A}_{+}$\hspace{1cm}& the positive cone of a Jordan algebra $\mathcal{A}$\\[0.2cm]
$\mathfrak{j}(X)$\hspace{1cm}& the Jordan algebraic closure of a subset $X$ of a Jordan algebra\\[0.2cm]
$\mathfrak{c}(Y)$\hspace{1cm}& the C$^{*}$\!-algebraic closure of a subset $Y$ of a C$^{*}$\!-algebra\\[0.2cm]
$\otimes_{\mathbb{R}}$\hspace{1cm}& tensor product of finite dimensional vector spaces over $\mathbb{R}$\\[0.2cm]
$\tilde{\otimes}$\hspace{1cm}& universal tensor product of Euclidean Jordan algebras \\[0.2cm]
$\odot$\hspace{1cm}& canonical tensor product of Euclidean Jordan algebras \\[0.2cm]
$\potimes$\hspace{1cm}& set of pure tensors in a canonical tensor product\\[0.2cm]
$\oplus$\hspace{1cm}& direct sum of algebras, groups, vector spaces, and so on\\[0.2cm]
$\mathcal{M}_{n}(\mathbb{D})_{\text{sa}}$& $n\times n$ self-adjoint matrices over a classical division algebra $\mathbb{D}\in\{\mathbb{R},\mathbb{C},\mathbb{H},\mathbb{O}\}$\\[0.2cm]
$\mathcal{V}_{k}$\hspace{1cm}& finite $(k+1)$-dimensional spin factor\\[0.2cm]
$C^{*}_{u}(\mathcal{A})$\hspace{1cm}& the universal C$^{*}$\!-algebra of a Euclidean Jordan algebra $\mathcal{A}$\\[0.2cm]
$\psi_{\mathcal{A}}$\hspace{1cm}& the universal embedding $\mathcal{A}\longrightarrow C^{*}_{u}(\mathcal{A})_{\text{sa}}$\\[0.2cm]
$\Phi_{\mathcal{A}}$\hspace{1cm}& the canonical involution on $C^{*}_{u}(\mathcal{A})_{\text{sa}}$\\[0.2cm]
$C^{*}_{s}(\mathcal{A})$\hspace{1cm}& the standard C$^{*}$\!-algebra of a Euclidean Jordan algebra $\mathcal{A}$\\[0.2cm]
$\pi_{\mathcal{A}}$\hspace{1cm}& the standard embedding $\mathcal{A}\longrightarrow C^{*}_{s}(\mathcal{A})_{\text{sa}}$\\[0.2cm]
$\text{ob}(\mathscr{C})$\hspace{1cm}& the objects of a category $\mathscr{C}$\\[0.2cm]
$\text{hom}(\mathscr{C})$\hspace{1cm}& the morphisms of a category $\mathscr{C}$\\[0.2cm]
$\text{inc}$\hspace{1cm}& inclusion, \textit{e.g.} $\text{inc}:C^{*}_{s}(\mathcal{A})_{\text{sa}}\longrightarrow C^{*}_{s}(\mathcal{A})::x^{*}=x\longmapsto x$
\end{tabular}

\cleardoublepage
\pagenumbering{arabic}

\chapter{Prologue}
\label{prologue}
\epigraphhead[40]
	{
		\epigraph{``Indeed from our present standpoint, physics is to be regarded not so much as the study of something a priori given, but as the development of methods for ordering and surveying human experience.''}{---\textit{Niels Bohr}\\ Essays 1958-1962 on\\ Atomic Physics and Human Knowledge}
	}
Consider a spherical cone in $\mathbb{R}^{4}$. Let us introduce a Euclidean coordinate system for the ambient space as follows: one coordinate runs along the symmetry axis of the cone; the other three parametrize perpendicular hyperplanes cutting out three-dimensional balls. This shape and its setting appear in fundamental physics. Within the context of special relativity, we have a light cone in Minkowski spacetime, with exterior regions beyond the reach of causal influence. In this case, our physical interpretation of the coordinate system is intimately connected with our everyday sensory experience. This connection with ordinary life is so strong that it can be difficult to resist viewing the subject matter of relativity from the perspective of na{\"i}ve realism. Within the context of quantum theory, however, there is a much different story to be told. Here we have, for instance, the cone of unnormalized quantum states for the spin of an electron. The unit trace hyperplane cuts out a Bloch ball parametrized by expectation values of the usual three Pauli spin observables: a rather peculiar probabilistic setting. The normative character of quantum theory is such that certainty regarding the outcome of one Pauli observable should come with complete ignorance regarding the other two. Put otherwise, quantum theory imposes severe complementarity constraints on coherent degrees of belief regarding certain aspects of physical systems. From this primitive example, from a simple shape, we glimpse the essence of a strange statistical structure, namely quantum theory.

\noindent If we are to appreciate the full character of quantum theory, then we must consider higher-dimensional physical systems. For instance, the Bell-Kochen-Specker theorem \cite{Bell1966}\cite{Kochen1967} applies only for physical systems admitting three or more mutually orthonormal quantum states (we remind the reader that orthogonal states for electron spin correspond to anti-podal points on the Bloch ball.) For such systems, the Bell-Kochen-Specker theorem establishes that quantum theory is incompatible with a seemingly innocuous idea: measurements reveal pre-existing facts about physical systems. This idea is quite plainly implied by the word `measurement' in ordinary language; a word that Bell argued should be banned altogether in quantum theory \cite{Bell1990}. The proof given by Kochen and Specker is based on the subtle geometry of pure quantum states for a three-level system. This geometry is also the basis for the original proof of Gleason's theorem \cite{Gleason1957}, which establishes a necessary and sufficient characterization of probability measures over the outcomes of sharp quantum measurements. So once again, fundamental aspects of the quantum are revealed via the analysis of a shape, albeit this time a shape much more complicated than a simple ball. If we take a step further and consider nontrivial composite physical systems, then the story of quantum theory becomes stranger still: we meet entanglement.
\newpage
\noindent In the words of Bell \cite{Bell1981}: ``The philosopher on the street, who has not suffered a course in quantum mechanics, is quite unimpressed with Einstein-Podolsky-Rosen correlations.'' Bell's celebrated theorem \cite{Bell1964}, however, establishes that the logical conjunction of quantum theory and Einstein's principle of local action \cite{Einstein1948} and Einstein-Podolsky-Rosen realism \cite{Einstein1935} is false. There are at least several compelling reasons to accept Einstein's principle of local action. First, the mutually independent existence of spatially distant objects is an important premise for experimental science. Without it, indomitable confounding variables would undermine conclusions drawn from observations of change coinciding with human manipulations of laboratory conditions. In fact, Einstein went so far as to say its complete abolition would make the formulation of empirically testable laws in the usual sense impossible \cite{Einstein1948}. Second, if one abandons Einstein's principle of local action, then free will begins to erode --- or, as phrased more forcefully by Fuchs \cite{Fuchs2010}: ``\dots if one is willing to throw away one's belief in systems' autonomy from each other, why would one ever believe in one's own autonomy?'' Third, adopting the assumption of distinct physical entities is provenly pragmatic; for example, consider the Standard Model of particle physics. Einstein's principle is a polarizing subject in quantum foundations. In light of the unprecedented empirical success of quantum theory --- in particular for the recent series of Bell tests \cite{Hensen2015}\cite{Giustina2015}\cite{Shalm2015} --- one may be drawn to seriously reconsider Einstein-Podolsky-Rosen realism. Fine \cite{Fine1999}, Fuchs-Mermin-Schack \cite{Fuchs2014}, H{\"a}nsch \cite{Hansch2015}, Unruh \cite{Unruh2013}, and Zukowski-Brukner \cite{Zukowski2014} have explicitly emphasized that quantum theory and locality are logically consistent. On that view, Bohr's vision of physics (as quoted in our epigraph) may be embraced.

\noindent Entanglement provides a beautiful example of the rich interplay between mathematics, theoretical physics, experimental physics, and technological development. These areas are not part of a hierarchical structure --- insights from each inspire the others. Entanglement first emerged in a theoretical setting \cite{Einstein1935}\cite{Schrodinger1935a}, later to be confirmed by experiment \cite{Aspect1982}, and then envisioned as a vital resource for quantum computing \cite{Jozsa2003}, quantum cryptography \cite{Ekert1991} and quantum communication \cite{Lloyd1997} prior to physical realizations of quantum teleportation \cite{Ma2012}, quantum key distribution \cite{Scheidl2009}, and a compilation \cite{Martin2012} of Shor's algorithm \cite{Shor1994}. From a physical perspective, entanglement serves to motivate the general dynamical structure of quantum theory: the cones of completely positive linear transformations on complex Hilbert spaces. In finite dimensions, Choi \cite{Choi1975} and Jamio{\l}kowski \cite{Jamiolkowski1972} fully characterized these cones in the realm of pure mathematics: they are linearly isomorphic to the cones of unnormalized quantum states describing bipartite physical systems. The Choi-Jamio{\l}kowski isomorphism between quantum cones provides a deep physical insight: the kinematical structure and the dynamical structure of quantum theory are equivalent. From a critical examination of the foundations of quantum theory, Einstein-Podolsky-Rosen thus initiated an ongoing multidisciplinary revolution, with quantum foundations and quantum information enjoying mutualistic symbiosis.

\noindent Quantum complementarity, contextuality, and correlations emerge from the geometry of quantum cones; thus, to study these shapes is to study the foundations for quintessential quantum physics. In \cref{partI} of this thesis, we derive geometric insights into quantum theory in the light of novel structures: \textit{conical designs} \cite{Graydon2015a}. From one perspective, conical designs are highly symmetric polytopes in quantum state space. We fully characterize these polytopes; moreover, we establish their existence in all finite dimensions. From another perspective, conical designs can be viewed as a natural generalization of complex projective designs \cite{Neumaier1981}\cite{Hoggar1982}\cite{Zauner1999}\cite{Scott2006}, which are vital constructs in quantum information theory. For instance, complex projective designs find natural applications in quantum state tomography \cite{Paris2004} and measurement-based quantum cloning \cite{Gisin1997}. They also have important applications \cite{Fuchs2003}\cite{Englert2004}\cite{Renes2004c}\cite{Renes2005}\cite{Durt2008} in quantum cryptography. We show that conical designs are naturally adapted to the description of entanglement \cite{Graydon2016}. In particular, we demonstrate a fundamental connection between conical designs and the theory of entanglement monotones \cite{Vidal2000}. In general, the first part of this thesis is devoted to these special quantum shapes and their connections with quantum information. 

\noindent It is natural to wonder what lies beyond quantum theory; to ask: what comes next? Feynman reminds us \cite{Feynman1963}: ``In fact, everything we know is only some kind of approximation, because we know that we do not know all the laws as yet. Therefore, things must be learned only to be unlearned again or, more likely, to be corrected.'' These words appear on the first page of Feynman's famous lectures delivered at the California Institute of Technology during 1961-1963. Over a half century later, they continue to describe the situation in contemporary physics. A glaring case in point is the incomplete reconciliation of quantum theory with general relativity. This problem, perhaps, may be resolved from a deeper understanding of quantum theory. Indeed, the kinematical and dynamical structure of quantum theory is that of quantum cones; consequently, it is imperative to consider their essential characteristics.

\noindent In addition to the quantum channel-state duality established via the Choi-Jamio{\l}kowski isomorphism, there exists another important unifying feature of quantum theory, namely self-duality \cite{Barnum2016}. In quantum theory, the outcomes of physical measurements are associated with elements of the cone of positive semi-definite linear functionals on unnormalized quantum states. These elements are called effects. In arbitrary finite Hilbert dimension, the cones of unnormalized quantum states and unnormalized quantum effects are identical --- they are \textit{self-dual}. From a general information-theoretic perspective, M\"uller and Ududec proved that self-duality follows from the structure of reversible computation \cite{Muller2012}. From a pure mathematical point of view, self-dual cones can be further classified in terms of the structure of their automorphism groups. In particular, quantum cones enjoy homogeneity \cite{Satake1972}: the linear automorphism group of a quantum cone acts transitively on its interior. In this respect, one says that quantum cones are \textit{homogeneous}. Operationally, homogeneity implies that any nonsingular quantum state can be mapped to any other via a reversible process. Homogeneity and self-duality are thus fundamental aspects of quantum information theory; moreover, they are essential characteristics of quantum theory in general.

\noindent If one is to look beyond quantum theory for new physics, then a logical and conservative approach is to consider physical theories sharing some of its essential characteristics. Within the vast landscape of general probabilistic theories \cite{Barnum2012}, homogeneity and self-duality do not uniquely specify quantum theory. These crucial features do, however, appreciably narrow the field: in finite dimensions, we arrive in the closed neighbourhood of Jordan-algebraic probabilistic theories. Our arrival therein is a consequence of the Koecher-Vinberg theorem \cite{Koecher1958}\cite{Vinberg1960}, which is a very deep result in the theory of operator algebras. In finite dimensions, Koecher and Vinberg independently proved that the only homogeneous self-dual cones are the positive cones of finite dimensional formally real Jordan algebras. By definition, these algebras are equipped with a nice self-dualizing inner product, hence their shorter name: \textit{Euclidean Jordan algebras}. From the Jordan-von Neumann-Wigner classification theorem \cite{Jordan1934} one has that any Euclidean Jordan algebra is isomorphic to a direct sum of algebras from the following list: self-adjoint matrices over the real, complex, and quaternionic division rings, spin factors, and the exceptional Jordan algebra of $3\times 3$ self-adjoint octonionic matrices. Physical theories built thereupon (including the case of classical probability theory built from direct sums of the trivial algebra) are quantum theory's `closest cousins' \cite{Ududec2012}.

\noindent Starting from arbitrary general (operational) probabilistic theories, Barnum-M{\"u}ller-Ududec have recently \textit{derived} homogeneity and self-duality from three compelling physical postulates \cite{Barnum2014}. Their fourth and final postulate yields quantum cones exactly. The Barnum-M{\"u}ller-Ududec postulates are phrased solely in terms of single systems, which markedly distinguishes their derivation from those of Hardy \cite{Hardy2001}\cite{Hardy2011}, Daki{\'c}-Brukner \cite{Dakic2011}, Chiribella-D'Ariano-Perinotti \cite{Chiribella2011}, and Masanes-M{\"u}ller \cite{Masanes2011}. Barnum-M{\"u}ller-Ududec leave open a very delicate question: is it possible to formulate reasonable \textit{composites} of physical models based on Euclidean Jordan algebras? In \cref{partII} of this thesis, we answer that question in the affirmative: we construct dagger compact closed \textit{Jordanic categories} \cite{Barnum2015}\cite{Barnum2016b}.

\noindent Category theory was introduced by Eilenberg and Mac \!Lane in \cite{Eilenberg1945}, wherein they announced: ``In a metamathematical sense our theory provides general concepts applicable to all branches of mathematics, and so contributes to the current trend towards uniform treatment of different mathematical disciplines.'' From its initial roots in algebraic topology \cite{Eilenberg1952} and homological algebra \cite{Eilenberg1956}, category theory grew to span the realm of pure mathematics; moreover, the very foundations of mathematics can be understood categorically \cite{MacLane1997}, which is an interesting alternative to the familiar Zermelo-Fraenkel set-theoretic framework. Category theory is formulated at a very high level of abstraction. From that perspective, instead of studying the structure of one particular group, one considers the category of \textit{all} groups and group homomorphisms: \textbf{Grp}. Likewise, rather than examining a specific vector space over a field $\mathbb{K}$, one treats the category of all such vector spaces and linear transformations: $\mathbf{K}$-$\mathbf{Vect}$. Moving one level higher, one considers the category of all categories: $\mathbf{Cat}$. Abstract thinking is very powerful. From the general conception of a class, one captures the essence of all particular instantiations. This is especially true in physics. Thus category theory and its level of abstraction naturally interface with the study of quantum foundations.

\noindent Abramsky and Coecke launched categorical quantum mechanics in \cite{Abramsky2004}, setting in motion a paradigm shift  for quantum foundations and quantum information science. In addition to reformulating the usual pure variant\footnote{By the \textit{pure variant} of quantum theory, one refers to the usual textbook description in terms of state vectors, unitary evolution, and self-adjoint observables.} of finite dimensional quantum theory in terms of compact closed categories with biproducts, Abramsky and Coecke explicitly demonstrated how quantum teleportation, logic-gate teleportation, and entanglement swapping protocols were captured at this abstract level. The relevant category for their work was $\mathbf{FdHilb}$. Soon thereafter, Peter Selinger introduced the CPM construction in \cite{Selinger2005}. Selinger proved, in particular, that $\mathbf{CPM}(\mathbf{FdHilb})$ is a dagger compact closed category of finite dimensional complex Hilbert spaces with completely positive linear transformations thereupon, \textit{i.e}.\ the \textit{general variant} of quantum theory familiar to quantum information theorists. In short, quantum theory \textit{is} a dagger compact closed category. The Choi-Jamio{\l}kowski isomorphism, for instance, is a concrete form of dagger compact closure. These lines of thought motivate our construction of dagger compact closed \textit{Jordanic categories} in \cref{partII} of this thesis. \noindent Our construction of Jordanic physical theories as dagger compact closed categories is based on an axiomatic derivation of all possible \textit{compositional structures} on Euclidean Jordan algebras. We prove that any such Jordanic composite is a Jordan ideal in Hanche-Olsen's universal tensor product \cite{HancheOlsen1983}. The categories we construct describe physics beyond the realm of quantum theory.

\noindent In quantum theory, local measurements on bipartite systems suffice to determine a unique global state: this property is known as \textit{tomographic locality}. The aforementioned axiomatizations of quantum theory due to Hardy \cite{Hardy2011}, Daki{\'c}-Brukner \cite{Dakic2011}, and Masanes-M{\"u}ller \cite{Masanes2011} explicitly invoke tomographic locality, while Chiribella-D'Ariano-Perinotti \cite{Chiribella2011} and Hardy's earlier work \cite{Hardy2001} invoke equivalent conditions. We do not demand tomographic locality from our composites; moreover, we prove that quantum theory is the only subcategory of our construction wherein tomographic locality holds. Another important feature of quantum theory is \textit{preservation of purity}: the composite of two pure morphisms yields a pure morphism, where `morphism' can be taken as either `state' or `transformation.' In fact, preservation of purity is asserted as an axiom by Chiribella and Scandolo \cite{Chiribella2015} to prove a general version of the Lo-Popescu theorem \cite{Lo2001} within the context of such purity preserving general probabilistic theories. We do not demand preservation of purity, and we construct composites violating this principle --- specifically, quaternionic composites. Therefore, our Jordanic theories depart from quantum theory in at least two significant ways concerning composite physical systems. On the other hand, our theories all enjoy homogeneity, self-duality, and a generalized version of the Choi-Jamio{\l}kowski isomorphism. On that view, we remain within the general neighbourhood of quantum theory. 

\noindent We outline the balance of this thesis in the following section.

\section{Outline} 

\noindent We partition the main body of this thesis into two parts. Following \cref{partI}, we pause for an interlude in \cref{interlude}, which leads us into \cref{partII}. We close with an epilogue in \cref{epilogue}. Some particularly technical details required to render this thesis self-contained are relegated to \cref{appP1} and \cref{appP2}. In this thesis, we cite two hundred eighty \hyperref[theBib]{references}.

\underline{\cref{partI}: Conical Designs}
 
\indent \cref{introPartI} is introductory and divided into three sections. In \cref{quantumTheory}, we review elements of quantum theory and set our notation. In \cref{grpTheory}, we recall prerequisite group theory for the sequel, setting additional notation. Our primary aim therein is to detail the canonical product representation of the complex unitary group of degree $d$. The well known irreducible components of this unitary representation, namely the symmetric and antisymmetric subspaces, feature prominently in the balance of \cref{partI}. In \cref{projectiveDesigns}, we review complex projective $t$-designs, with a strong emphasis on the case $t=2$. We prove a novel result, namely \cref{pureDesCon}, which characterizes the extreme points of quantum state space in terms of simple conditions on probabilities for the outcomes of quantum measurements formed from an arbitrary complex projective 2-design. We also recall facts concerning \textsc{sic}s (see \cref{sicDef}) and \textsc{mub}s (see \cref{mubDef}.)

\indent \cref{designsOnQuantumCones} is based on \cite{Graydon2015a} and divided into five sections. In \cref{blochBod}, we review the generalized Bloch representation of quantum state spaces. In \cref{simmumery}, we review \textsc{sim}s (see \cref{simDef}) and \textsc{mum}s (see \cref{mumDef}) and provide a unified geometric proof for their existence in all finite dimensions. \textsc{sim}s and \textsc{mum}s are arbitrary rank generalizations of \textsc{sic}s and \textsc{mub}s, respectively. In \cref{desQC}, we introduce conical 2-designs. We establish five equivalent characterizations thereof via \cref{desCons} and we detail their essential properties. We also prove that \textsc{sim}s and \textsc{mum}s are conical 2-designs. In \cref{hc2d}, we focus on the subclass of homogeneous conical 2-designs (see \cref{hcdDef}.) We characterize homogeneous conical 2-designs via \cref{blochPoly} and we prove that all varieties homogeneous conical 2-designs exist in all finite dimensions via \cref{existThm}. In \cref{inSearchOf}, we outline a program to seek out new varieties of complex projective 2-designs. We also lift the problem of constructing a homogeneous conical 2-design to the problem of constructing a 1-design on a higher dimensional real vector space via \cref{liftThm}.

\indent \cref{entanglementConicalDesigns} is based on \cite{Graydon2016} and divided into four sections. In \cref{prelims}, we first review entanglement monotones. We then define a novel concept: \textit{regular} entanglement monotones (see \cref{regDef}.) We then prove \cref{conRegLem} thereby establishing that the concurrence is regular. In \cref{monotones}, we prove \cref{monDesThm}, which establishes a fundamental and elementary connection between the theory of regular entanglement monotones and the theory of conical 2-designs. Our proof of \cref{monDesThm} is founded on our novel \cref{nluiLem}. In \cref{witnesses}, we develop and generalize previous work relating entanglement witnesses and certain conical 2-designs. In \cref{decomps}, we explore a connection linking conical 2-designs with Werner states and isotropic states. The former are invariant under the action of the canonical product representation of the complex unitary group of the degree $d$, the latter are invariant under the action of $U\otimes \overline{U}$, where henceforth overline denotes complex conjugation with respect to a fixed basis for the underlying finite dimensional complex Hilbert space.

\indent \cref{conclusionPartI} is our conclusion for \cref{partI}.
\newpage

\underline{\cref{partII}: Categorical Jordan Algebraic Post-Quantum Theories}
 
\indent \cref{introPartII} is introductory and divided into three sections. We do not assume any prior knowledge of the material presented in this chapter. In \cref{catPrelims}, we review elements of category theory. We first recall definitions of categories, functors, and natural transformations. We then build up to the definition of a dagger compact closed category (see \cref{dagCom}.) In \cref{jordPrelims}, we review Jordan algebraic prerequisites. In particular, we focus on Euclidean Jordan algebras (see \cref{ejaDef}), which are the ambient spaces for states and effects in the Jordan algebraic post-quantum theories considered in the chapters that follow. We define \textit{standard} representations for all reversible Euclidean Jordan algebras. These representations are in terms of Jordan subalgebras of the self-adjoint parts of C$^{*}$\!-algebras. In \cref{cStarPrelims}, we recall \textit{universal} representations of Euclidean Jordan algebras from the literature, and we present some elementary calculations. The standard and universal representations facilitate the sequel.

\indent \cref{compositesEJA} is based on portions of \cite{Barnum2015} and \cite{Barnum2016b} pertaining to composites of Euclidean Jordan algebras. We divide this chapter into three sections. In \cref{physMot}, we review the framework of general probabilistic theories. Along the way, we specialize to our case of interest: Jordan algebraic general probabilistic theories. We introduce a general definition for \textit{composites} of models for physical systems in general probabilistic theories, \cref{def: dynamical composites}. The axioms in our definition reflect the physical principle of nonsignaling. In \cref{sec: composites Jordan}, we consider the structure of composites in general probabilistic theories based on Euclidean Jordan algebras. We prove \cref{thm: composites special}, which establishes that the composite of two nontrivial simple Euclidean Jordan algebras always admits a representation within the self-adjoint part of a C$^{*}$\!-algebra. An immediate \cref{cor: no composite with exceptional} is that composites involving the exceptional Jordan algebra do not exist. We then prove that any composite is a Jordan ideal of Hanche-Olsen's universal tensor product (\cref{thm: simple composites ideals}.) Next, we define a \textit{canonical tensor product} for Euclidean Jordan algebras (see \cref{def: canonical tensor product}) and we prove that canonical tensor products yield composites (\cref{prop: canonical product composite}.) In \cref{canTPcomps}, we explicitly compute all canonical tensor products involving reversible Euclidean Jordan algebras. 

\indent \cref{categoriesEJA} is based on portions of \cite{Barnum2015} and \cite{Barnum2016b} pertaining to categorical Jordan algebraic post-quantum theories. We divide this chapter into two sections. In \cref{jpmSec}, we prove that the canonical tensor product is associative (\cref{prop: associativity}.) We then consider the behaviour of the canonical tensor product over direct sums. This sets the stage for the sequel. In \cref{sec: categories EJC}, we construct our Jordanic categories. We introduce \cref{cjpMorph} for \textit{completely positive Jordan-preserving maps}, which is a natural extension of the notion of complete positivity within our Jordanic physical theories. We then prove \cref{ex: no states}, which rules out the inclusion of `higher' spin factors in our Jordanic categories. Therefore we restrict our attention to the unification of real, complex, and quaternionic quantum theory. Our main result is \cref{cor: InvQM is dagger compact}, establishing that our categorical unification of these theories is dagger compact closed.  

\indent \cref{conclusionPartII} is our conclusion for \cref{partII}.
\part{Conical Designs}
\label{partI}

\chapter{Setting the Stage (Part I)}
\label{introPartI}

\epigraphhead[40]
	{
		\epigraph{``But however the development proceeds in detail, the path so far traced by the quantum theory indicates that an understanding of those still unclarified features of atomic physics can only be acquired by foregoing visualization and objectification to an extent greater than that customary hitherto.''}{---\textit{Werner Heisenberg}\\ 1933 Nobel Lecture:\\ The Development of Quantum Mechanics}
	}

Let $\mathcal{H}_{d}$ denote a finite $d$-dimensional complex Hilbert space. In $\mathcal{H}_{d}$, the intersection of any $1$-dimensional subspace with the unit sphere $\mathcal{S}(\mathcal{H}_{d})$ is isomorphic to a circle in $2$-dimensional Euclidean space, and the quotient space formed from $\mathcal{S}(\mathcal{H}_{d})$ modulo this $\mathrm{U}(1)$ symmetry is \cite{Bengtsson2007} the complex projective space $\mathbb{CP}^{d-1}$: the manifold of \textit{pure quantum states}. The geometry of pure quantum states is extremely intricate, except for the case of $\mathbb{CP}^{1}$, which is isomorphic to a sphere in 3-dimensional Euclidean space. Remarkably, in arbitrary finite Hilbert dimension, there exist \cite{Seymour1984} instances from a class of highly symmetric substructures defined on the manifold of pure quantum states, namely \textit{complex projective $t$-designs} \cite{Neumaier1981}\cite{Hoggar1982}\cite{Zauner1999}\cite{Scott2006}. In this introductory chapter, we recall elements of quantum theory in \cref{quantumTheory} and group theory in \cref{grpTheory}. In \cref{projectiveDesigns}, we meet complex projective $t$-designs and discuss the geometry of quantum state space in their light. The purpose of this chapter is to set the stage with preliminary physical notions and prerequisite mathematical apparatus for the sequel. Before proceeding with technical matters, we shall now outline a physical interpretation along the lines suggested by Heisenberg in the epigraph above.

\noindent For emphasis, we recall once again from Heisenberg's 1933 Nobel Lecture \cite{Heisenberg1933}: ``The very fact that the formalism of quantum  mechanics cannot be interpreted as visual description of a phenomenon occurring in space and time shows that quantum mechanics is in no way concerned with the objective determination of space-time  phenomena.'' Indeed, any objectification, or \textit{reification} \cite{Mermin2009} of quantum states and other formal elements of quantum theory is subject to severe conceptual difficulties; for instance, consider the infamous parables of Wigner's friend \cite{Wigner1995} and Schr{\"o}dinger's cat\footnote{Remember, \cite{Schrodinger1935b}: ``\textit{Man kann auch ganz burleske Fälle konstruiere.} [One can also construct very burlesque cases.]''} \cite{Schrodinger1935b}. On the contrary, the zeitgeist of the present and third quantum foundations revolution \cite{Hardy2011b} revolves around the conception of quantum theory as a \textit{theory of information}. The structure of this information processing framework is highly subtle. Calling back to the epigraph in our prologue, one may adopt the view wherein the structure of this Bohrian \cite{Bohr1987} method for `ordering and surveying human experience' \textit{is} the physics of our natural world. Why is our world this way? This is obviously one of the very deepest questions of all. In the present part of this thesis, we undertake the much more modest goal of illuminating the information-geometric structure of finite dimensional quantum theory. The conical designs introduced in \cref{designsOnQuantumCones}, which include complex projective $t$-designs as a special case, shed new light on this subject. We thus gain novel insight into the mathematical form of quantum information regarding physical systems; hence the contents of \cref{partI} of this thesis are to held within the spirit of the third quantum foundations revolution. 

\section{Elements of Quantum Theory}
\label{quantumTheory}
Henceforth, for brevity, \textit{quantum theory} refers to finite dimensional quantum theory formulated over finite $d$-dimensional complex Hilbert spaces in the usual manner to be reviewed presently (or any equivalent formulation.) Most readers will be intimately familiar with quantum theory. The primary purpose of this section is to introduce our notation and to provide some important definitions. For modern comprehensive treatments of this subject, we point to, for instance, the excellent textbook of Nielsen and Chuang \cite{Nielsen2010} and the more advanced lecture notes of John Watrous \cite{Watrous2011}. We shall begin in Hilbert space.

\noindent Let $\mathcal{H}_{d}$ denote a finite $d$-dimensional complex Hilbert space with $\langle\cdot|\cdot\rangle:\mathcal{H}_{d}\times\mathcal{H}_{d}\longrightarrow\mathbb{C}$ its inner product. By definition \cite{vonNeumann1955}, $\mathcal{H}_{d}$ is complete with respect to the induced norm $\|\cdot\|:\mathcal{H}_{d}\longrightarrow\mathbb{R}_{\geq 0}::\psi\longmapsto\sqrt{\langle\psi|\psi\rangle}$ (as is any finite dimensional inner product space over $\mathbb{C}$ \cite{Roman2008}) in the sense that every Cauchy sequence in $\mathcal{H}_{d}$ converges with respect to this norm to an element of $\mathcal{H}_{d}$. We denote and define the \textit{unit sphere} in $\mathcal{H}_{d}$ via $\mathcal{S}(\mathcal{H}_{d})\equiv\{\psi\in\mathcal{H}_{d}\;\boldsymbol{|}\;\|\psi\|=1\}$. Next, let $\mathcal{H}_{d}^{\star}$ denote the dual Hilbert space of all continuous linear functionals on $\mathcal{H}_{d}$.  In light of the Riesz representation theorem \cite{Bachman1966}, $\forall f\in\mathcal{H}_{d}^{\star}\;\exists!\phi\in\mathcal{H}_{d}$ such that $\forall\psi\in\mathcal{H}_{d}\;f(\psi)=\langle\phi|\psi\rangle$; moreover $\mathcal{H}_{d}^{\star}\cong\mathcal{H}_{d}$ as Hilbert spaces. We shall frequently adopt the Dirac notation $|\psi\rangle$ and $\langle\phi|$ for elements of $\mathcal{H}_{d}$ and $\mathcal{H}_{d}^{\star}$, respectively. Furthermore, we will usually denote $\mathbb{C}$-scalar multiplication on the right, \textit{e.g}.\ by $|\psi\rangle\lambda$ for $\lambda\in\mathbb{C}$. 

\noindent Let $\mathcal{L}(\mathcal{H}_{d})$ denote the C$^{*}$\!-algebra\footnote{We shall return to C$^{*}$\!-algebras in \cref{partII}; see \cref{cStarDef}} \cite{Alfsen2012} of all linear functions $A:\mathcal{H}_{d}\longrightarrow\mathcal{H}_{d}::\psi\longmapsto A\psi$, where, in particular, the algebraic product is functional composition denoted by juxtaposition, \textit{i.e}.\ $AB\equiv A\circ B$, and the relevant $\mathbb{C}$-antilinear involution $^{*}:\mathcal{L}(\mathcal{H}_{d})\longrightarrow \mathcal{L}(\mathcal{H}_{d})$ is realized at the level of matrices via the composition of transposition and complex conjugation (both with respect to some fixed orthonormal basis for $\mathcal{H}_{d}$.) Forgetting its algebraic structure, $\mathcal{L}(\mathcal{H}_{d})$ is itself a finite $d^{2}$-dimensional complex Hilbert space with respect to the Hilbert-Schmidt inner product $\langle\!\langle\cdot|\cdot\rangle\!\rangle:\mathcal{L}(\mathcal{H}_{d})\times\mathcal{L}(\mathcal{H}_{d})\longrightarrow\mathbb{C}::(A,B)\longmapsto\mathrm{Tr}(A^{*}B)$, where of course `$\mathrm{Tr}$' denotes the usual trace functional. As such, we shall sometimes, but only when convenient, adopt the double Dirac notation $|A\rangle\!\rangle$ and $\langle\!\langle B|$ for elements of $\mathcal{L}(\mathcal{H}_{d})$ and its dual Hilbert space $\mathcal{L}(\mathcal{H}_{d})^{\star}$, respectively. We will be primarily interested in the subset of self-adjoint elements in $\mathcal{L}(\mathcal{H}_{d})$, which we denote as define via $\mathcal{L}_{\text{sa}}(\mathcal{H}_{d})=\{A\in\mathcal{L}(\mathcal{H}_{d})\;\boldsymbol{|}\;A=A^{*}\}$, which is a $d^{2}$-dimensional vector space over $\mathbb{R}$. 

\noindent Let $\mathrm{span}_{\mathbb{K}}\mathcal{X}$ denote the $\mathbb{K}$-linear span of a subset $\mathcal{X}$ of a vector space over a field $\mathbb{K}$.

\begin{definition}{\cite{Rockafellar1970}}\label{rock1Def}
\textit{Let $\mathcal{X}$ be a finite dimensional vector space over $\mathbb{R}$. A} convex set \textit{is a subset $\mathcal{C}\subseteq\mathcal{X}$ such that} $x_{1},x_{2}\in\mathcal{C}$ \textit{and} $\lambda\in[0,1]\implies x_{1}(1-\lambda)+x_{2}\lambda\in\mathcal{C}$\textit{. A} cone \textit{in $\mathcal{X}$ is a subset $\mathcal{K}\subseteq\mathcal{X}$ such that $k\in\mathcal{K}$ and $\lambda\in\mathbb{R}_{>0}\implies k\lambda \in\mathcal{K}$. A} pointed cone \textit{is a cone such that} $\mathcal{K}\cap-\mathcal{K}=\{0\}$\textit{. A} generating cone \textit{is a cone such that} $\mathrm{span}_{\mathbb{R}}\mathcal{K}=\mathcal{X}$\textit{. A} convex cone \textit{is a cone that is also a convex set.}
\label{coneDef}
\end{definition}

\noindent In quantum theory, we take $\mathcal{X}=\mathcal{L}_{\text{sa}}(\mathcal{H}_{d})$ and consider for each\footnote{For later convenience, we include the trivial case $d=1$; $\mathbb{R}$ will be our monoidal unit in \cref{partII}.} $d\in\mathbb{N}$ the generating pointed convex cone of all positive semi-definite elements within, which we denote and define as follows
\begin{equation}
\mathcal{L}_{\text{sa}}(\mathcal{H}_{d})_{+}\equiv\big\{A\in\mathcal{L}_{\text{sa}}(\mathcal{H}_{d})\;\boldsymbol{|}\;\forall\psi\in\mathcal{H}_{d}\;\langle\psi|A\psi\rangle\geq 0\big\}\text{.}
\end{equation} That $\mathcal{L}_{\text{sa}}(\mathcal{H}_{d})_{+}$ is a convex cone follows immediately from $\mathbb{R}$-linearity of the inner product $\langle\cdot|\cdot\rangle$. Trivially, it is pointed. Furthermore, any self-adjoint linear function can be decomposed as the difference of two positive semi-definite linear functions, so $\mathcal{L}_{\text{sa}}(\mathcal{H}_{d})_{+}$ is generating.
\begin{definition} 
\textit{A} quantum cone \textit{is} $\mathcal{L}_{\text{sa}}(\mathcal{H}_{d})_{+}$ \textit{for some} $d\in\mathbb{N}$.
\end{definition}  
\begin{definition} 
\textit{Let} $\mathcal{L}_{\text{sa},1}(\mathcal{H}_{d})\equiv\{A\in\mathcal{L}_{\text{sa}}(\mathcal{H}_{d})\;\boldsymbol{|}\;\mathrm{Tr}A=1\}$ \textit{be the unit trace hyperplane in} $\mathcal{L}_{\text{sa}}(\mathcal{H}_{d})$\textit{. The set of} quantum states\textit{, denoted $\mathcal{Q}(\mathcal{H}_{d})$, is the intersection of this hyperplane with the quantum cone, i.e}.\
$\mathcal{Q}(\mathcal{H}_{d})\equiv\mathcal{L}_{\text{sa}}(\mathcal{H}_{d})_{+}\cap\mathcal{L}_{\text{sa},1}(\mathcal{H}_{d})=\{\rho\in\mathcal{L}_{\text{sa}}(\mathcal{H}_{d})_{+}\;\boldsymbol{|}\;\mathrm{Tr}\rho=1\}$.
\label{stateDef}
\end{definition}
\noindent It is obvious from the foregoing definition is that $\mathcal{Q}(\mathcal{H}_{d})$ is a convex set. It is also easy to see that $\mathcal{Q}(\mathcal{H}_{d})$ is topologically compact \cite{Runde2005} in $\mathcal{L}_{\text{sa}}(\mathcal{H}_{d})$ taken as a metric space \cite{Searcoid2006} with respect to the distance induced by the norm inherited from $\mathcal{L}(\mathcal{H}_{d})$. Indeed, $\mathcal{Q}(\mathcal{H}_{d})$ is bounded, for it is contained in the unit ball. Next, let $f:\mathcal{L}_{\text{sa}}(\mathcal{H}_{d})\longrightarrow\mathbb{R}::A\longmapsto \inf_{\psi\in\mathcal{S}(\mathcal{H}_{d})}\langle\psi|A\psi\rangle$, which is a linear and hence continuous function between finite dimensional metric spaces. Therefore $\mathcal{Q}(\mathcal{H}_{d})$ is closed, for it is precisely the intersection of the $f$-preimage of $[0,\infty)$ and the $\mathrm{Tr}$-preimage of $\{1\}$. So, by the Heine-Borel theorem \cite{Leonard2015} $\mathcal{Q}(\mathcal{H}_{d})$ is compact. It follows from the Krein-Milman theorem \cite{Krein1940} that $\mathcal{Q}(\mathcal{H}_{d})$ is precisely the convex hull of its extreme points.

\begin{definition}{\cite{Rockafellar1970}}\label{rock2Def}
\textit{Let $\mathcal{C}\subseteq\mathcal{X}$ be a convex set. An} extreme point \textit{of $\mathcal{C}$ is an element $x\in\mathcal{C}$ such that, for $\lambda\in(0,1)$ and $x_{1},x_{2}\in\mathcal{C}$, $x=x_{1}(1-\lambda)+x_{2}\lambda\iff x=x_{1}=x_{2}$. A} convex combination \textit{of} $x_{1},\dots,x_{n}\in\mathcal{X}$ \textit{is} $x_{1}\lambda_{1}+\dots+x_{n}\lambda_{n}$ \textit{with}  $\lambda_{1},\dots,\lambda_{n}\in\mathbb{R}_{\geq 0}$ \textit{such that} $\lambda_{1}+\dots+\lambda_{n}=1$. \textit{Let subset} $\mathcal{Y}\subseteq\mathcal{X}$. \textit{The} convex hull of $\mathcal{Y}$\textit{, denoted} $\mathrm{conv}\mathcal{Y}$\textit{, is the set of all convex combinations of elements from $\mathcal{Y}$.}
\label{convDef}
\end{definition}

\noindent Recall that a \textit{unit rank projector} is $\pi\in\mathcal{L}_{\text{sa}}(\mathcal{H}_{d})$ such that $\pi^{2}=\pi$ and $\mathrm{Tr}\pi=1$, in which case we can write $\pi=|\psi\rangle\langle\psi|$ for some $|\psi\rangle\in\mathcal{S}(\mathcal{H}_{d})$. It is well known \cite{Bengtsson2007} that $\rho\in\mathcal{Q}(\mathcal{H}_{d})$ is an extreme point if and only if $\rho$ is a unit rank projector. A \textit{pure quantum state} is precisely an extreme point of $\mathcal{Q}(\mathcal{H}_{d})$. We denote the set of all pure states by $\text{Pur}\mathcal{Q}(\mathcal{H}_{d})$. In accordance with our preceding observations, any quantum state can therefore be decomposed as a convex combination of pure quantum states. A common physical interpretation of this fact is that any quantum state can be viewed as a probabilistic mixture of pure states; hence, a quantum state that is not pure is called \textit{mixed}. Lastly, note that the unit rank projectors corresponding to $|\psi\rangle e^{i\theta}$ are identical for all angles $\theta$, hence our earlier discussion regarding $\mathbb{CP}^{d-1}$.
\begin{definition}
\textit{Let} $\mathds{1}_{d}$ \textit{denote the identity function on $\mathcal{H}_{d}$, i.e}.\ $\mathds{1}_{d}:\mathcal{H}_{d}\longrightarrow\mathcal{H}_{d}::\psi\longmapsto\psi$. \textit{The set of} quantum effects\textit{, denoted} $\mathcal{E}(\mathcal{H}_{d})$, \textit{is the set of all} $\mathcal{L}_{\text{sa}}(\mathcal{H}_{d})_{+}\ni E\leq\mathds{1}_{d}$ \textit{with respect to the L{\"o}wner ordering, i.e}.\ $\mathcal{E}(\mathcal{H}_{d})\equiv\{E\in\mathcal{L}_{\text{sa}}(\mathcal{H}_{d})_{+}\;\boldsymbol{|}\;\mathds{1}_{d}-E\in\mathcal{L}_{\text{sa}}(\mathcal{H}_{d})_{+}\}$.
\label{effectDef}
\end{definition}

\noindent In quantum theory, one is primarily interested in certain subsets of effects that model physical measurement devices. Strictly speaking, a \textit{positive operator valued measure }is \cite{Busch1997} a Borel measure \cite{Kelley2012} $\mu$ on the Borel $\sigma$-algebra of some set $\mathscr{S}$ into the quantum cone, such that $\mu(\mathscr{S})=\mathds{1}_{d}$. In quantum information theory, the set comprised of the images of the singletons of $\mathscr{S}$ under $\mu$ is commonly called (the acronym for) a positive operator valued measure. This is mostly harmless of course; however, we shall remove any ambiguity by formally introducing `\textsc{povm}' as follows, wherein we further restrict to our specific case of interest: finite discrete positive operator valued measures.

\begin{definition}
\textit{A} \textsc{povm} \textit{is a subset} $\{E_{1},\dots,E_{n}\}\subset\mathcal{E}(\mathcal{H}_{d})$ \textit{such that} $\sum_{j=1}^{n}E_{j}=\mathds{1}_{d}$.
\label{povmDef}
\end{definition}

\noindent Consider the following physical situation. A physical system is input to a measurement device, which outputs\footnote{If nondestructive, this device outputs a physical system with an associated state computed via the \textit{L{\"u}ders Rule} \cite{Luders1950}.} one classical outcome $j\in\{1,\dots,n\}$. In quantum theory, one models this situation as follows. A quantum state $\rho\in\mathcal{Q}(\mathcal{H}_{d})$ is associated with the input system. A quantum effect $E_{j}$ is associated with each classical outcome such that $\{E_{j}\}$ is a \textsc{povm}. With these associations, based on the physics of the particular situation, the probability of observing measurement outcome $j$ is computed via the \textit{Born rule}: 
\begin{equation}
p_{j}=\mathrm{Tr}(E_{j}\rho)\text{.}
\label{bornRule}
\end{equation}
\noindent It is via the Born rule that one derives the probabilities for arbitrary measurement outcomes encoded in a given quantum state. A quantum state can thus be viewed as a compendium of information regarding the physical system with which it is associated; moreover, the Born rule links these abstract compendia of probabilities with physical experience. On that view, quantum states are something like probability distributions, colloquially speaking. In fact, this statement can be made mathematically precise. Indeed, informationally complete \cite{Prugovevcki1977}\cite{Busch1991} \textsc{povm}s exist \cite{Fuchs2002} for all cases of finite Hilbert dimension, so one can outright identify quantum states with probability distributions. 
\begin{definition}
\textit{An} informationally complete \textsc{povm} \textit{is a} \textsc{povm} $\{E_{j}\}\subset\mathcal{E}(\mathcal{H}_{d})$ \textit{such that} $\mathrm{span}_{\mathbb{R}}\{E_{j}\}=\mathcal{L}_{\text{sa}}(\mathcal{H}_{d})$. \textit{A} minimal informationally complete \textsc{povm} \textit{is an informationally complete} \textsc{povm} \textit{such that} $\mathrm{card}\{E_{j}\}=d^{2}$\textit{, where} $\mathrm{card}\mathscr{S}$ \textit{denotes the cardinality of a set }$\mathscr{S}$.
\label{icDef}
\end{definition}
\noindent If $\{E_{j}\}$ is an informationally complete \textsc{povm}, then the corresponding probabilities $p_{j}$ computed via the Born rule from an arbitrary $\rho\in\mathcal{Q}(\mathcal{H}_{d})$ uniquely determine $\rho$ within the convex set of quantum states. Of course, this follows immediately from the fact that $\{E_{j}\}$ span the ambient real vector space of self-adjoint linear functions; moreover, one can \cite{dAriano2004} explicitly form an expansion for any quantum state in terms of $p_{j}$. In the general case of an arbitrary informationally complete \textsc{povm}, such an expansion may be horrendous; hence, while it may physically appealing to replace the Hilbert space formalism with an explicit probability calculus, what remains may be rather complicated. On this point Wootters emphasizes \cite{Wootters1986}: ``It is obviously possible to devise a formulation of quantum mechanics without probability amplitudes. One is never forced to use any quantities in one's theory other than the raw results of measurements. However, there is no reason to expect such a formulation to be anything other than extremely ugly.'' Nevertheless, in \cite{Wootters1986}, Wootters makes key advances towards a mathematically beautiful formulation of quantum theory without Hilbert space. Wootters employs a specific variety of complex projective 2-design (\textsc{mub}s see \cref{mubDef}) as the primary tool for his construction. Incidentally, any complex projective 2-design facilitates a clean expansion (Eq.~\eqref{pdEpn}) for $\rho$ in terms of certain $p_{j}$. We shall see this in \cref{projectiveDesigns}.

\noindent \cref{stateDef} and \cref{effectDef} respectively define quantum states and quantum effects in a general way; however, quantum theory is equipped to handle the explicit treatment of systems and measurement devices consisting of separate parts. For instance, a physical system may be viewed as the composite of two or more subsystems, and a joint measurement may be carried out by two or more spacelike separated parties. The mathematical construction underlying a quantum theoretic model of such cases is the tensor product of finite dimensional complex vector spaces. Before moving forward with a precise definition, note that a subset $\mathcal{W}\subseteq\mathcal{Z}$ of a complex vector space $\mathcal{Z}$ is said to \textit{generate} $\mathcal{Z}$ if the closure of $\mathcal{W}$ under the ambient vector space operations is equal to $\mathcal{Z}$.

\begin{definition}\cite{Yokonuma1992}\label{tensorDef}
\textit{Let} $\mathcal{V},\mathcal{W}$ \textit{be finite dimensional vector spaces over $\mathbb{C}$. A} tensor product \textit{of} $\mathcal{V}$ \textit{with} $\mathcal{W}$ \textit{is a pair} $(\mathcal{Z}_{0},\otimes)$, \textit{with} $\mathcal{Z}_{0}$ \textit{a finite dimensional vector space over $\mathbb{C}$ and} $\otimes:\mathcal{V}\times\mathcal{W}\longrightarrow\mathcal{Z}_{0}$ \textit{a} $\mathbb{C}$\textit{-bilinear function such that} 
$\otimes(\mathcal{V}\times\mathcal{W})$ \textit{generates} $\mathcal{Z}_{0}$, \textit{and such that for any} $\mathbb{C}$\textit{-bilinear function} $f:\mathcal{V}\times\mathcal{W}\longrightarrow\mathcal{Z}$\textit{, with} $\mathcal{Z}$ \textit{a finite dimensional vector space over $\mathbb{C}$,} $\exists!$ $\mathbb{C}$\textit{-bilinear function} $f_{0}:\mathcal{Z}_{0}\longrightarrow\mathcal{Z}$ \textit{satisfying} $f=f_{0}\circ\otimes$.
\end{definition} 

\noindent A tensor product of $\mathcal{V}$ with $\mathcal{W}$ always exists and is in fact unique up to vector space isomorphism \cite{Yokonuma1992}. By standard abuse of notation and terminology, we shall write $\mathcal{Z}_{0}=\mathcal{V}\otimes\mathcal{W}$ and $\otimes(v,w)=v\otimes w$, and refer to $\mathcal{V}\otimes\mathcal{W}$ as the tensor product of $\mathcal{V}$ with $\mathcal{W}$. Concerning dimensionality, note $\mathrm{dim}_{\mathbb{C}}\mathcal{V}\otimes\mathcal{W}=\mathrm{dim}_{\mathbb{C}}\mathcal{V}\mathrm{dim}_{\mathbb{C}}\mathcal{W}$. One refers to elements of the form $v\otimes w\in\mathcal{V}\otimes\mathcal{W}$ as \textit{pure tensors}. The \textit{tensor product of linear functions} $f:\mathcal{V}\longrightarrow V'$ and $g:\mathcal{W}\longrightarrow{W}'$, denoted $f\otimes g$, is defined via the linear extension of its action on pure tensors, \textit{i.e}.\ $f\otimes g:\mathcal{V}\otimes\mathcal{W}\longrightarrow\mathcal{V}'\otimes\mathcal{W}'::v\otimes w\longmapsto f(v)\otimes g(w)$. Lastly, one notes that the tensor product construction is associative. We are now ready to consider physical composites and quantum channels.
\newpage
\noindent In physics, as in ordinary life, the assumption of distinct physical systems, of separate pieces of the world, is a common practice. In an attempt to refine the primitive word `system' with further details, one could do little more than replace it with a synonymous word like `thing,' `object,' `entity,' and so on. In quantum theory, one often speaks of physical `systems and subsystems', though the prefix \textit{sub} does not necessarily imply some formally aggregate structure, rather these terms only indicate that one is considering a system that is (at least) two systems: a \textit{composite}.   As we have mentioned, in quantum theory one associates a quantum state with a physical system. Implicit in such an association is also the association of a Hilbert space. If one considers two systems, with associated finite dimensional complex Hilbert spaces $\mathcal{H}_{d_{\mathrm{A}}}$ and $\mathcal{H}_{d_{\mathrm{B}}}$, respectively, then $\mathcal{H}_{d_{\mathrm{A}}}\otimes\mathcal{H}_{d_{\mathrm{B}}}$ is the arena for quantum theory of the composite system. In particular, a quantum state for the composite is any $\rho\in\mathcal{Q}(\mathcal{H}_{d_{\mathrm{A}}}\otimes\mathcal{H}_{d_{\mathrm{B}}})$. Such states are either separable or entangled.

\begin{definition}
\textit{A} separable \textit{quantum state is any} $\rho\in\mathcal{Q}(\mathcal{H}_{d_{\mathrm{A}}}\otimes\mathcal{H}_{d_{\mathrm{B}}})$ \textit{such that} $\rho=\sum_{i}\sigma_{i}\otimes \tau_{i}p_{i}$, \textit{where} $0\leq p_{i}\leq 1$ \textit{sum to unity and} $\sigma_{i}\in\mathcal{Q}(\mathcal{H}_{d_{\mathrm{A}}})$ \textit{and} $\tau_{i}\in\mathcal{Q}(\mathcal{H}_{d_{\mathrm{B}}})$\textit{. We denote the set of all separable states by} $\mathrm{Sep}\mathcal{Q}(\mathcal{H}_{d_{\mathrm{A}}}\otimes\mathcal{H}_{d_{\mathrm{B}}})$\textit{. An} entangled \textit{quantum state is any} $\rho\in \mathcal{Q}(\mathcal{H}_{d_{\mathrm{A}}}\otimes\mathcal{H}_{d_{\mathrm{B}}})\setminus\mathrm{Sep}\mathcal{Q}(\mathcal{H}_{d_{\mathrm{A}}}\otimes\mathcal{H}_{d_{\mathrm{B}}})$.
\label{sepDef}
\end{definition}
\noindent In \cref{entanglementConicalDesigns}, we consider entanglement in the light conical designs. For the sake of continuity, we shall defer preliminaries regarding entanglement monotones \cite{Vidal2000} and entanglement witnesses \cite{Terhal2000} until then. Having just formally met \cref{sepDef}, however, a few general, technical remarks are in order. The geometry of entangled quantum states is very rich, and only partially understood \cite{Horodecki2009}. For instance, in the language of computational complexity theory \cite{Aaronson2013}, the problem of deciding whether an arbitrary quantum state is separable is NP-hard \cite{Gurvits2003}\cite{Gharibian2008}. Nevertheless, in the simplest cases $d_{\mathrm{A}}=d_{\mathrm{B}}=2$ and $\{d_{\mathrm{A}},d_{\mathrm{B}}\}=\{2,3\}$, and only in these cases, a closed necessary and sufficient condition for separability has been discovered: the \textit{Peres-Horodecki criterion} \cite{Peres1996}\cite{Horodecki1996}. For a precise statement of this criterion, we shall need to introduce the partial transpose. Let $\mathrm{I}_{d}:\mathcal{L}(\mathcal{H}_{d})\longrightarrow\mathcal{L}(\mathcal{H}_{d})::A\longmapsto A$ denote the identity function on $\mathcal{L}(\mathcal{H}_{d})$. Next, relative to some fixed orthonormal basis for $\mathcal{H}_{d}$, let $\mathrm{T}_{d}:\mathcal{L}(\mathcal{H}_{d})\longrightarrow\mathcal{L}(\mathcal{H}_{d})::A\longmapsto A^{\mathrm{T}}$ denote the usual matrix transpose. The \textit{partial transpose} is then denoted and defined via the linear extension of $\mathrm{I}_{d_{\mathrm{A}}}\otimes\mathrm{T}_{d_{\mathrm{B}}}:\mathcal{L}(\mathcal{H}_{d_{\mathrm{A}}}\otimes\mathcal{H}_{d_{\mathrm{B}}})\longrightarrow\mathcal{L}(\mathcal{H}_{d_{\mathrm{A}}}\otimes\mathcal{H}_{d_{\mathrm{B}}})::A\otimes B\longmapsto A\otimes B^{\mathrm{T}}$. In the aforementioned simplest cases, the Peres-Horodecki criterion establishes that $\rho\in\mathcal{Q}(\mathcal{H}_{d_{\mathrm{A}}}\otimes\mathcal{H}_{d_{\mathrm{B}}})$ is separable if and only if $\mathrm{I}_{d_{\mathrm{A}}}\otimes\mathrm{T}_{d_{\mathrm{B}}}(\rho)\in\mathcal{Q}(\mathcal{H}_{d_{\mathrm{A}}}\otimes\mathcal{H}_{d_{\mathrm{B}}})$. Immediately, then, one sees that the transpose is not completely positive. 
\begin{definition}\label{cpDef}
\textit{A} positive map \textit{is a linear function} $\Lambda:\mathcal{L}(\mathcal{H}_{d_{\mathrm{A}}})\longrightarrow\mathcal{L}(\mathcal{H}_{d_{\mathrm{B}}})::\mathcal{L}_{\text{sa}}(\mathcal{H}_{d_{\mathrm{A}}})_{+}\longmapsto\mathcal{L}(\mathcal{H}_{d_{\mathrm{B}}})_{+}$\textit{. A} completely positive (CP) map \textit{is a positive map} $\Lambda$ \textit{such that} $\forall d_{\mathrm{C}}\in\mathbb{N}$ $\Lambda\otimes\mathrm{I}_{d_{\mathrm{C}}}$ \textit{is a positive map. We denote the sets of all positive and CP maps by} $\mathsf{Pos}(d_{\mathrm{A}},d_{\mathrm{B}})$ \textit{and} $\mathsf{CP}(d_{\mathrm{A}},d_{\mathrm{B}})$\textit{, respectively.}
\end{definition}

\noindent It turns out \cite{Horodecki1996} that $\rho\in\mathcal{Q}(\mathcal{H}_{d_{\mathrm{A}}}\otimes\mathcal{H}_{d_{\mathrm{B}}})$ is separable if and only if $\mathrm{I}_{d_{\mathrm{A}}}\otimes\Lambda(\rho)\in\mathcal{L}_{\text{sa}}(\mathcal{H}_{d_{\mathrm{A}}}\otimes\mathcal{H}_{d_{\mathrm{C}}})_{+}$ for \textit{every} positive map $\Lambda:\mathcal{L}(\mathcal{H}_{d_{\mathrm{B}}})\longmapsto\mathcal{L}(\mathcal{H}_{d_{\mathrm{C}}})$. The key to the proof of the Peres-Horodecki criterion is that when $d_{\mathrm{A}}=2$ and $d_{\mathrm{B}}\in\{2,3\}$ every positive map $\Lambda:\mathcal{L}(\mathcal{H}_{d_{\mathrm{A}}})\longmapsto\mathcal{L}(\mathcal{H}_{d_{\mathrm{B}}})$ is \cite{Stormer1963}\cite{Woronowicz1976} \textit{decomposable}, \textit{i.e}.\ $\Lambda=\Lambda_{1}+\Lambda_{2}\circ\mathrm{T}_{d_{\mathrm{B}}}$ with $\Lambda_{1},\Lambda_{2}$ CP maps. What complicates the situation in higher dimensions is that not every CP map is decomposable; moreover, in such dimensions, there exist \cite{Kossakowski2003} infinite families of CP maps that are not decomposable. Thus, the structure of positive and CP maps is itself quite subtle. One can study the geometry of these maps by considering the ambient finite $d_{\mathrm{A}}^{2}d_{\mathrm{B}}^{2}$-dimensional complex Hilbert space of all $\mathbb{C}$-linear functions $\Theta:\mathcal{L}(\mathcal{H}_{d_{\mathrm{A}}})\longrightarrow\mathcal{L}(\mathcal{H}_{d_{\mathrm{B}}})$, which we shall denote $\mathsf{Lin}(d_{\mathrm{A}},d_{\mathrm{B}})$. Here, the double Dirac notation for $|A\rangle\!\rangle\in\mathcal{L}(\mathcal{H}_{d})$ becomes especially useful, for one can express any $\Theta\in\mathsf{Lin}(d_{\mathrm{A}},d_{\mathrm{B}})$ as $\Theta=\sum_{j,k}|A_{j}\rangle\!\rangle\phi_{j,k}\langle\!\langle B_{k}|$ in terms of arbitrary orthonormal bases $\{|A_{j}\rangle\!\rangle\}$ and $\{|B_{k}\rangle\!\rangle\}$  for $\mathcal{L}(\mathcal{H}_{d_{\mathrm{A}}})$ and $\mathcal{L}(\mathcal{H}_{d_{\mathrm{B}}})$, respectively. With overline denoting complex conjugation, the adjoint is then $\Theta^{*}\equiv\sum_{j,k}|B_{k}\rangle\!\rangle\overline{\phi_{j,k}}\langle\!\langle A_{j}|$, and the Hilbert-Schmidt inner product is $\langle\!\langle\!\langle\!\langle\Theta|\Xi\rangle\!\rangle\!\rangle\!\rangle\equiv\mathrm{Tr}(\Theta^{*}\Xi)$, where we have introduced quadruple Dirac notation for elements of $\mathsf{Lin}(d_{\mathrm{A}},d_{\mathrm{B}})$. Naturally, one can allow only real scalars and consider $\mathsf{Lin}(d_{\mathrm{A}},d_{\mathrm{B}})$ as a vector space over $\mathbb{R}$, with quantum channels forming convex subset thereof.

\begin{definition}
\textit{A} trace preserving (TP) map \textit{is any} $\Theta\in\mathsf{Lin}(d_{\mathrm{A}},d_{\mathrm{B}})$ \textit{such that} $\mathrm{Tr}\Theta(A)=\mathrm{Tr}A$ \textit{for all} $A\in\mathcal{L}(\mathcal{H}_{d_{\mathrm{A}}})$\textit{. A} quantum channel \textit{is a completely positive trace preserving} \textit{map. We respectively denote the sets of all} TP \textit{maps and quantum channels by} $\mathsf{TP}(d_{\mathrm{A}},d_{\mathrm{B}})$ \textit{and} $\mathsf{CPTP}(d_{\mathrm{A}},d_{\mathrm{B}})$\textit{.}
\label{chanDef}
\end{definition}

\noindent Quantum information regarding physical systems can change, for instance, in light of measurements, or the destruction of subsystems, or when the physical system in question is influenced by an external force field, and so on. Quantum channels exact these changes on quantum states; moreover they are the general dynamical structure of quantum theory. Geometrically, the convex set of quantum channels is the intersection of the convex cone $\mathsf{CP}(d_{\mathrm{A}},d_{\mathrm{B}})$ with the hyperplane $\mathsf{TP}(d_{\mathrm{A}},d_{\mathrm{B}})$ in $\mathsf{Lin}(d_{\mathrm{A}},d_{\mathrm{B}})$ considered as a vector space over $\mathbb{R}$. It is not a coincidence that one has in mind the exact same geometric picture for quantum states: a convex cone sliced by a hyperplane perpendicular to the symmetry axis of the cone. Indeed, in light of the Choi-Jamio{\l}kowski isomorphism \cite{Choi1975}\cite{Jamiolkowski1972}, the convex sets of quantum states and quantum channels admit isomorphic geometries. Henceforth, let $\mathcal{L}(\mathcal{H}_{d_{\mathrm{A}}},\mathcal{H}_{d_{\mathrm{A}}})$ denote finite $d_{\mathrm{A}}d_{\mathrm{B}}$-dimensional complex Hilbert space of set of $\mathbb{C}$-linear functions $X:\mathcal{H}_{d_{\mathrm{A}}}\longrightarrow\mathcal{H}_{d_{\mathrm{B}}}$ equipped with the usual Hilbert-Schmidt inner product, and where $X^{*}:\mathcal{H}_{d_{\mathrm{B}}}\longrightarrow\mathcal{H}_{d_{\mathrm{A}}}$ denotes the usual adjoint function.

\begin{theorem} (Choi-Jamio{\l}kowski \cite{Choi1975}\cite{Jamiolkowski1972}) \textit{Let} $\{|e_{1}\rangle,\dots,|e_{d_{\mathrm{A}}}\rangle\}$ \textit{be an arbitrary orthonormal basis for a finite $d_{\mathrm{A}}$-dimensional complex Hilbert space} $\mathcal{H}_{d_{\mathrm{A}}}$. \textit{Define} $\boldsymbol{\mathcal{J}}:\mathsf{Lin}(d_{\mathrm{A}},d_{\mathrm{B}})\longrightarrow \mathcal{L}(\mathcal{H}_{d_{\mathrm{A}}}\otimes\mathcal{H}_{d_{\mathrm{B}}})$ \textit{via}
$d_{\mathrm{A}}\boldsymbol{\mathcal{J}}(\Lambda)=\sum_{j,k=1}^{d_{\mathrm{A}}}\Lambda\big(|e_{j}\rangle\langle e_{k}|\big)\otimes|e_{j}\rangle\langle e_{k}|\text{.}$
\textit{Then} $\Lambda\in\mathsf{CPTP}(d_{\mathrm{A}},d_{\mathrm{B}})$ \textit{if and only if} $\boldsymbol{\mathcal{J}}(\Lambda)\in\mathcal{Q}(\mathcal{H}_{d_{\mathrm{A}}}\otimes\mathcal{H}_{d_{\mathrm{B}}})$\textit{. In that case,} $\exists \{X_{i}\}\subset\mathcal{L}(\mathcal{H}_{d_{\mathrm{A}}},\mathcal{H}_{d_{\mathrm{A}}})$ \textit{such that} $\Lambda::A\longmapsto \sum_{i} X_{i}AX_{i}^{*}$\text{, where} $\sum_{i}X_{i}^{*}X_{i}=\mathds{1}_{d_{\mathrm{A}}}$.
\label{cjThm}
\end{theorem}

\noindent The $X_{i}$ appearing in the statement of \cref{cjThm} are referred to as \textit{Kraus operators}. Of course, $\mathsf{Lin}(d_{\mathrm{A}},d_{\mathrm{B}})$ and $\mathcal{L}(\mathcal{H}_{d_{\mathrm{A}}}\otimes\mathcal{H}_{d_{\mathrm{B}}})$ admit identical Hilbert dimension $d_{\mathrm{A}}^{2}d_{\mathrm{B}}^{2}$ and are therefore isomorphic as finite dimensional complex Hilbert spaces. The beauty of the \cref{cjThm}, for which we refer the reader to the original references \cite{Choi1975}\cite{Jamiolkowski1972} for proof, is that the image of the restriction of the linear bijection $\boldsymbol{\mathcal{J}}$ to the set of quantum channels $\mathsf{CPTP}(d_{\mathrm{A}},d_{\mathrm{B}})$ is exactly the set of quantum states $\mathcal{Q}(\mathcal{H}_{d_{\mathrm{A}}}\otimes\mathcal{H}_{d_{\mathrm{B}}})$. Thus, to study the shape of quantum states is to study the shape of quantum channels, and \textit{vice versa}. Naturally, then, just as one has the notion of a separable state, one also has the notion of a separable channel \cite{Rains1997}.

\begin{definition}
\textit{A} separable quantum channel \textit{is any completely positive trace preserving map} $\Theta:\mathcal{L}(\mathcal{H}_{d_{\mathrm{A}}})\otimes\mathcal{L}(\mathcal{H}_{d_{\mathrm{B}}})\longrightarrow\mathcal{L}(\mathcal{H}_{d_{\mathrm{C}}})\otimes\mathcal{L}(\mathcal{H}_{d_{\mathrm{D}}})$ \textit{such that} $\Theta(A)=\sum_{i}X_{i}AX_{i}^{\dagger}$ \textit{where all of the Kraus operators are of the form} $X_{i}=Y_{i}\otimes Z_{i}$ \textit{for some} $Y_{i}\in\mathcal{L}(\mathcal{H}_{d_{\mathrm{A}}},\mathcal{H}_{d_{\mathrm{C}}})$ \textit{and} $Z_{i}\in\mathcal{L}(\mathcal{H}_{d_{\mathrm{B}}},\mathcal{H}_{d_{\mathrm{D}}})$\textit{. We denote the set of all separable quantum channels by} $\mathrm{Sep}\mathsf{CPTP}(d_{\mathrm{A}};d_{\mathrm{B}},d_{\mathrm{C}};d_{\mathrm{D}})$\textit{.}
\label{sepChan}
\end{definition}

\noindent With a little work \cite{Watrous2011}, one can see that $\boldsymbol{\mathcal{J}}:\mathrm{Sep}\mathsf{CPTP}(d_{\mathrm{A}};d_{\mathrm{B}},d_{\mathrm{C}};d_{\mathrm{D}})\longrightarrow\mathrm{Sep}\mathcal{Q}(\mathcal{H}_{d_{\mathrm{A}}};\mathcal{H}_{d_{\mathrm{B}}},\mathcal{H}_{d_{\mathrm{C}}};\mathcal{H}_{d_{\mathrm{D}}})$, where the image is defined via the obvious generalization of \cref{sepDef}. We shall call on the foregoing definition in \cref{entanglementConicalDesigns}. We shall also call on the following definition.

\begin{definition}
\textit{A} product \textsc{povm} \textit{is any} \textsc{povm} $\{E_{j}\}\subset\mathcal{E}(\mathcal{H}_{d_{\mathrm{A}}}\otimes\mathcal{H}_{d_{\mathrm{A}}})$ \textit{where all of the effects are of the form} $E_{j}=F_{j}\otimes G_{j}$ \textit{for some} $F_{j}\in\mathcal{E}(\mathcal{H}_{d_{\mathrm{A}}})$ \textit{and} $F_{j}\in\mathcal{E}(\mathcal{H}_{d_{\mathrm{B}}})$\textit{. We denote the set of all product} \textsc{povm}s \textit{by} $\mathrm{Prod}\mathcal{E}(\mathcal{H}_{d_{\mathrm{A}}}\otimes\mathcal{H}_{d_{\mathrm{A}}})$\textit{. A} joint \textsc{povm} \textit{is any} \textsc{povm} $\{E_{\alpha}\}\in\mathcal{E}(\mathcal{H}_{d_{\mathrm{A}}}\otimes\mathcal{H}_{d_{\mathrm{A}}})\setminus\mathrm{Prod}\mathcal{E}(\mathcal{H}_{d_{\mathrm{A}}}\otimes\mathcal{H}_{d_{\mathrm{A}}})$\textit{.}
\label{prodPovm}
\end{definition}

\noindent There are a lot of additional, interesting things that one can say about quantum theory; however, we have now met most of essential material on quantum theory that is required to understand the chapters to follow. In closing this section, we note that the \textit{partial trace} is denoted and defined via the linear extension of its action on pure tensors according to $\mathrm{Tr}_{\mathcal{H}_{d_{\mathrm{A}}}}:\mathcal{L}(\mathcal{H}_{d_{\mathrm{A}}}\otimes\mathcal{H}_{d_{\mathrm{B}}})\longrightarrow\mathcal{L}(\mathcal{H}_{d_{\mathrm{B}}})::A\otimes B\longmapsto B\mathrm{Tr}A$, with $\mathrm{Tr}_{\mathcal{H}_{d_{\mathrm{B}}}}$ defined analogously. We now proceed with \cref{grpTheory}, which contains some necessary prerequisite group theory.



\section{Group Theoretic Prerequisites}
\label{grpTheory}
\noindent Group theory is a deep subject. There exist a host of textbooks on the subject, including \cite{Kirillov1976}\cite{Schwarz1996}\cite{Hall2015}. In this section, we collect important prerequisites for \cref{partI} of this thesis.

\noindent Let $\mathrm{G},\mathrm{H}$ be groups. A \textit{group homomorphism} is a function $\sigma:\mathrm{G}\longrightarrow\mathrm{H}::g\longmapsto \sigma_{g}$ such that $\forall g,g'\in\mathrm{G}$ one has that $\sigma_{g}\sigma_{g'}=\sigma_{gg'}$, where juxtaposition denotes the relevant group operations, interpreted from context. Let $\mathcal{V},\mathcal{W}$ be vector spaces. A \textit{linear homomorphism} is a linear function $T:\mathcal{V}\longrightarrow\mathcal{W}::v\longmapsto T(v)$. The terms \textit{monomorphism}, \textit{epimorphism}, and \textit{isomorphism} are reserved for injective, surjective, and bijective functions, respectively. The term \textit{endomorphism} is reserved for functions whose domain and codomain coincide. The term \textit{automorphism} is reserved for endomorphic isomorphisms. Group homomorphisms and linear homomorphisms are instances of a general categorical notion introduced in \cref{partII}.

\noindent A \textit{group representation} of $\mathrm{G}$ is a pair $(\sigma,\mathcal{V})$ where $\sigma:\mathrm{G}\longrightarrow\mathrm{GL}(\mathcal{V})$ is a group homomorphism and $\mathrm{GL}(\mathcal{V})$ is the group of invertible linear endomorphisms on $\mathcal{V}$. A \textit{unitary representation} of $\mathrm{G}$ is a pair $(\sigma,\mathcal{H})$ where $\mathcal{H}$ is a complex Hilbert space \cite{vonNeumann1955} and $\sigma:\mathrm{G}\longrightarrow\mathrm{U}(\mathcal{H})$, with $\mathrm{U}(\mathcal{H})$ the group of unitary linear endomorphisms on $\mathcal{H}$. Now, suppose $\mathcal{X}\subseteq\mathcal{V}$ is a vector subspace of $\mathcal{V}$ and $T$ is a linear endomorphism on $\mathcal{V}$. If $\forall x\in\mathcal{X}$ one has that $T(x)\in\mathcal{X}$, then $\mathcal{X}$ is said to be an \textit{invariant subspace} of $T$. In this case, the \textit{restriction of $T$ to $\mathcal{X}$} is the unique linear endomorphism $T_{|\mathcal{X}}:\mathcal{X}\longrightarrow\mathcal{X}::x\longmapsto T(x)$. Invariant $\mathcal{X}$ is said to be \textit{proper} when $\mathcal{X}\notin\{\{0\},\mathcal{V}\}$. Next, let $(\sigma,\mathcal{H})$ be a unitary representation of $\mathrm{G}$. If there exists a Hilbert subspace $\mathcal{X}\subseteq\mathcal{H}$ such that $\forall g\in\mathrm{G}$ and $\forall x\in\mathcal{X}$ one has that $\sigma_{g}(x)\in\mathcal{X}$, then $\mathcal{X}$ is said to be \textit{invariant}, \textit{i.e}.\ $\mathcal{X}$ is an invariant subspace of $\sigma_{g}$ for all $g\in\mathrm{G}$. If $\mathcal{X}$ is not proper, then the representation is said to be \textit{irreducible}; otherwise, it is said to be \textit{reducible}. The \textit{restriction of} $\pi$ \textit{to} $\mathcal{X}$ is defined via $\sigma_{|\mathcal{X}}:\mathrm{G}\longrightarrow\mathrm{GL}(\mathcal{X})::g\longmapsto \sigma_{g_{|\mathcal{X}}}$. 

\noindent A \textit{matrix group} is a subgroup $\mathrm{G}\subseteq\mathrm{GL}(\mathcal{H}_{d})\subset\mathcal{M}_{d}(\mathbb{C})$ where $\mathcal{M}_{d}(\mathbb{C})$ is the complex Hilbert space of $d\times d$ complex matrices equipped with the usual Hilbert-Schmidt inner product. A \textit{matrix Lie group} is a matrix group such that if a sequence $g_{i}\in\mathrm{G}$ converges entrywise to $x\in\mathcal{M}_{d}(\mathbb{C})$, then either $x\in\mathrm{G}$ else $x$ is not invertible. A \textit{compact matrix Lie group} is a matrix Lie group which is a closed, bounded subset of $\mathcal{M}_{d}(\mathbb{C})$. In particular, $\mathrm{U}(\mathcal{H}_{d})$ is a compact matrix Lie group. If $\sigma:\mathrm{G}\longrightarrow\mathrm{GL}(\mathcal{H}_{d})$ and $\sigma':\mathrm{G}\longrightarrow\mathrm{GL}(\mathcal{H}_{d'})$ are representations of a matrix Lie group $\mathrm{G}$, then the \textit{direct sum} of $(\sigma,\mathcal{H}_{d})$ and $(\sigma',\mathcal{H}_{d'})$ is the representation $(\sigma\oplus\sigma',\mathcal{H}_{d}\oplus\mathcal{H}_{d'})$ defined via $\sigma\oplus\sigma':\mathrm{G}\longrightarrow\mathcal{H}_{d}\oplus\mathcal{H}_{d'}::g\longmapsto \sigma_{g}\oplus \sigma_{g}$. This direct sum construction is associative. Any finite dimensional unitary representation $(\sigma,\mathcal{H}_{d})$ of a compact matrix Lie group $\mathrm{G}$ is such that $(\sigma,\mathcal{H}_{d})$ decomposes as the direct sum of irreducible representations, called \textit{irreducible components}.

\noindent The $t$\textit{-fold tensor product} of $\mathcal{H}_{d}$ is the complex Hilbert space $\mathcal{H}_{d^{t}}\equiv\mathcal{H}_{d}^{\otimes^{t}}$. Let $\mathrm{S}_{t}$ be the group of permutations $\mathrm{p}$ on $\{1,2,\dots,t\}$, and consider its standard unitary representation $(\sigma_{t},\mathcal{H}_{d^{t}})$ defined via 
\begin{equation}
\sigma_{t}:S_{t}\longrightarrow\mathrm{U}(\mathcal{H}_{d^{t}})::\mathrm{p}\longmapsto U_{\mathrm{p}}:::U_{\mathrm{p}}\big(|e_{r_{1}}\rangle\otimes\cdots\otimes|e_{r_{t}}\rangle\big)=|e_{r_{\mathrm{p}(1)}}\rangle\otimes\cdots\otimes|e_{r_{\mathrm{p}(t)}}\rangle\text{,}
\end{equation} 
where $r,r_{j}\in\{1,\dots,d\}$ and $|e_{r}\rangle\in\mathcal{H}_{d}$ constitute an orthonormal basis for $\mathcal{H}_{d}$ and $U_{p}$ is extended linearly. The \textit{$t$-fold product representation} of $\mathrm{U}(\mathcal{H}_{d})$ is the unitary representation $(\sigma^{(t)},\mathcal{H}_{d^{t}})$ defined via
\begin{equation}
\sigma^{(t)}:\mathrm{U}(\mathcal{H}_{d})\longrightarrow\mathrm{U}(\mathcal{H}_{d^{t}})::U\longmapsto U^{\otimes^{t}}.
\end{equation}
In light of Schur-Weyl duality \cite{Goodman2009}, there is a bijection between the irreducible components of $\sigma_{t}$ and $\sigma^{(t)}$. The relevant irreducible component for \cref{cpdDef} is the restriction of $\sigma^{(t)}$ to the invariant irreducible \textit{totally symmetric subspace} of $\mathcal{H}_{d^{t}}$, which is defined \cite{Harrow2013} along with its orthogonal projector thereonto via
\begin{eqnarray}
\mathcal{H}_{\text{sym}}^{(t)}=\text{span}_{\mathbb{C}}\Big\{|\Psi\rangle\in\mathcal{H}_{d^{t}}\;\Big|\;\forall p\in\mathrm{S}_{n}\;U_{p}|\Psi\rangle=|\Psi\rangle\Big\}
\hspace{1cm}\Pi_{\text{sym}}^{(t)}&=&\frac{1}{t!}\sum_{p\in\mathrm{S}_{t}}U_{p}\text{.}
\end{eqnarray}
\noindent In \cref{partI} of this thesis, we will be primarily interested in complex projective 2-designs and their arbitrary rank arbitrary trace generalization to be introduced in \cref{designsOnQuantumCones}. For those purposes, the relevant group representation is the $2$-fold product representation of $\mathrm{U}(\mathcal{H}_{d})$, which we will call the \textit{product representation} for short. Furthermore, for notational convenience we henceforth let $\mathcal{H}\equiv\mathcal{H}_{d}\otimes\mathcal{H}_{d}$ and $\sigma^{(2)}\equiv\sigma$ so that the product representation is written as
\begin{equation}
\sigma:\mathrm{U}(\mathcal{H}_{d})\longrightarrow\mathrm{U}(\mathcal{H})::U\longmapsto U\otimes U\text{.}
\label{prodRep}
\end{equation}
The product representation is vital for the sections and chapters to follow. Happily, in this case, one can easily understand the irreducible components via basic linear algebra. Indeed, let us introduce an arbitrary orthonormal basis with elements $|e_{r}\rangle\in\mathcal{H}_{d}$ and recall the usual swap operator $\mathrm{W}:\mathcal{H}\longrightarrow\mathcal{H}$ defined
\begin{equation}\label{swapOp}
\mathrm{W}=\sum_{r=1}^{d}\sum_{s=1}^{d}|e_{s}\rangle\langle e_{r}|\otimes|e_{r}\rangle\langle e_{s}|\text{.}
\end{equation}
Observe that $\forall|\psi\rangle,|\phi\rangle\in\mathcal{H}_{d}$ one has $\mathrm{W}(|\psi\rangle\otimes|\phi\rangle)=|\phi\rangle\otimes|\psi\rangle$, hence the operator's name. Also note that $\mathrm{W}$ is a self-adjoint unitary linear endomorphism with degenerate eigenvalues $\{+1,-1\}$. The corresponding eigenspaces are the \textit{symmetric} and \textit{antisymmetric} subspaces, which are denoted and defined as follows
\begin{eqnarray}
\mathcal{H}_{\text{sym}}&=&\text{span}_{\mathbb{C}}\Bigg\{|f^{+}_{r,s}\rangle\equiv\frac{|e_{r}\rangle\otimes|e_{s}\rangle+|e_{s}\rangle\otimes|e_{r}\rangle}{\sqrt{2}}:s>r\in\{1,\dots,d\}\Bigg\}\nonumber\\
&\oplus&\text{span}_{\mathbb{C}}\Bigg\{|f^{+}_{r,r}\rangle\equiv |e_{r}\rangle\otimes |e_{r}\rangle:r\in\{1,\dots,d\}\Bigg\}\text{,}\label{hSym}\\
\mathcal{H}_{\text{asym}}&=&\text{span}_{\mathbb{C}}\Bigg\{|f^{-}_{r,s}\rangle\equiv\frac{|e_{r}\rangle\otimes|e_{s}\rangle-|e_{s}\rangle\otimes|e_{r}\rangle}{\sqrt{2}}:s>r\in\{1,\dots,d\}\Bigg\}\label{hAsym}\text{,}
\end{eqnarray}
The union of the two sets appearing in Eq.~\eqref{hSym} is an orthonormal basis for $\mathcal{H}_{\text{sym}}$. The set appearing in Eq.~\eqref{hAsym} is an orthonormal basis for $\mathcal{H}_{\text{asym}}$. The corresponding orthogonal projectors (with respect to the Hilbert-Schmidt inner product) onto the orthogonal symmetric and antisymmetric subspaces are given by
\begin{eqnarray}
\Pi_{\text{sym}}=\frac{1}{2}\Big(\mathds{1}_{d^{2}}+\mathrm{W}\Big)\hspace{1cm}\Pi_{\text{asym}}=\frac{1}{2}\Big(\mathds{1}_{d^{2}}-\mathrm{W}\Big)\text{,}
\label{piSymAsym}
\end{eqnarray}
Evidently, from the traces of $\Pi_{\text{sym}}$ and $\Pi_{\text{asym}}$, one has that
\begin{eqnarray}
\mathrm{dim}_{\mathbb{C}}\mathcal{H}_{\text{sym}}&=&\frac{d(d+1)}{2}\\
\mathrm{dim}_{\mathbb{C}}\mathcal{H}_{\text{asym}}&=&\frac{d(d-1)}{2}\text{.}
\end{eqnarray}
In light of the foregoing analysis we see that $\mathcal{H}=\mathcal{H}_{\text{sym}}\oplus\mathcal{H}_{\text{asym}}$. Furthermore, it is readily apparent that
\begin{eqnarray}
\forall U\in\mathrm{U}(\mathcal{H}_{d})\;\forall|\psi_{+}\rangle\in\mathcal{H}_{\text{sym}}\hspace{0.5cm}(U\otimes U)|\psi_{+}\rangle\in\mathcal{H}_{\text{sym}}\\
\forall U\in\mathrm{U}(\mathcal{H}_{d})\;\forall|\psi_{-}\rangle\in\mathcal{H}_{\text{asym}}\hspace{0.5cm}(U\otimes U)|\psi_{-}\rangle\in\mathcal{H}_{\text{asym}}\text{.}
\end{eqnarray}
Put otherwise, $\mathcal{H}_{\text{sym}}$ and $\mathcal{H}_{\text{asym}}$ are invariant subspaces of the product representation. Further still, the restrictions of the product representation to the symmetric and antisymmetric subspaces are irreducible. This is a well known fact, which follows immediately from Schur-Weyl duality. Let us elevate this fact to a lemma, which we shall call on later. We provide an an elementary proof in \cref{lem2p2p1Proof}.
\begin{lemma}\label{irrepLem}
\textit{The restrictions of the product representation of $\mathrm{U}(\mathcal{H}_{d})$ to $\mathcal{H}_{\text{sym}}$ and $\mathcal{H}_{\text{asym}}$ are irreducible.}
\end{lemma}
\noindent We will call on \cref{irrepLem} at several key moments in the sections and chapter to follow. We will also call on Schur's lemma \cite{Schur1905}, which is a very important result for group representation theory. There are many variants of Schur's lemma. We shall recall from \cite{Hall2015} one particular variant that suits our purposes.

\begin{lemma}\label{schurLem} (\textit{Schur's Lemma})
Let $\sigma:\mathrm{G}\longrightarrow\mathrm{GL}(\mathcal{\mathcal{H}_{d}})::g\longmapsto \sigma_{g}$ be an irreducible representation. Let $T:\mathcal{H}_{d}\longrightarrow\mathcal{H}_{d}$ be a linear endomorphism such that $g\in\mathrm{G}\implies[\sigma_{g},T]=0$, where $[\cdot\hspace{0.07cm},\cdot]$ denotes the usual commutator. Then $T=\mathds{1}_{d}\lambda$ with $\lambda\in\mathbb{C}$.
\end{lemma}

\begin{corollary}
Let $T:\mathcal{H}\longrightarrow\mathcal{H}$ such that $\mathcal{H}_{\text{sym}}$ and $\mathcal{H}_{\text{asym}}$ are invariant subspaces of $T$. If $\forall U\in\mathrm{U}(\mathcal{H}_{d})$ one has that $[U\otimes U,T]=0$, then $T=\Pi_{\text{sym}}\lambda_{\text{sym}}+\Pi_{\text{asym}}\lambda_{\text{asym}}$ for some $\lambda_{\text{sym}},\lambda_{\text{asym}}\in\mathbb{C}$.
\end{corollary}
\begin{proof}
Immediate consequence of \cref{irrepLem} and \cref{schurLem}.
\end{proof}

\section{Complex Projective Designs}
\label{projectiveDesigns}
The notion of a \textit{spherical $t$-design} in dimension $d$ --- that is, a finite set of points on the unit hypersphere in $\mathbb{R}^{d}$ such that the average value of any real polynomial $f$ of degree $t$ or less on the set is equal to the average value of $f$ over the entire sphere --- was first introduced by Delsarte, Goethals and Seidel in \cite{Delsarte1977}. In \cite{Neumaier1981}, Neumaier extended this notion to general metric spaces. Shortly thereafter, in \cite{Hoggar1982}, Hoggar considered the notion of $t$-designs in projective spaces over the classical division algebras $\mathbb{R}$, $\mathbb{C}$, $\mathbb{H}$, and $\mathbb{O}$. Such designs were studied further by Hoggar \cite{Hoggar1984}\cite{Hoggar1989}\cite{Hoggar1992}, and later by Levenshtein \cite{Levenshtein1998}. In this section, we are specifically interested in an extension of the notion of a spherical $t$-design into the arena of finite dimensional complex projective spaces. Formally, given a finite $d$-dimensional complex Hilbert space $\mathcal{H}_{d}$, the complex projective space $\mathbb{CP}^{d-1}$ is represented by the set of all unit rank projectors in $\mathcal{Q}(\mathcal{H}_{d})$, that is
\begin{equation}
\mathbb{CP}^{d-1}\cong\left\{\pi_{\alpha}\equiv|\psi_{\alpha}\rangle\langle \psi_{\alpha}|\;\boldsymbol{|}\;\psi_{\alpha}\in\mathcal{S}(\mathcal{H}_{d})\right\}\text{.}
\end{equation} 
The unique unitarily invariant probability measure $\mu_{H}$ on $\mathbb{CP}^{d-1}$ is induced by the Haar measure \cite{Diestel2014} on $\mathrm{U}(\mathcal{H}_{d})$ and admits $\forall U\in\mathrm{U}(\mathcal{H}_{d})$ and $\forall \pi_{\alpha}\in\mathbb{CP}^{d-1}$
\begin{equation}
\mu_{H}\!\left(U\pi_{\alpha}U^{*}\right)=\mu_{H}\!\left(\pi_{\alpha}\right)\text{.}
\end{equation}
In analogy with the case of spherical $t$-designs, a complex projective $t$-design can be regarded as a set of points on the unit sphere in $\mathcal{H}_{d}$ such that the average value of any real polynomial of degree $t$ or less over that set coincides with the average over the entire sphere. There are, however, many equivalent \cite{Konig1999} definitions of complex projective $t$-designs. We shall adopt the general definition given by Scott \cite{Scott2006}. 
\begin{definition}\label{cpdDef} \cite{Scott2006} 
\textit{Let $t\in\mathbb{N}$. Let $\mathcal{D}=\{\pi_{\alpha}\}$ be a set of unit rank projectors onto rays in $\mathcal{H}_{d}$ endowed with a probability measure $\omega:\mathfrak{B}(\mathcal{D})\longrightarrow[0,1]$, where $\mathfrak{B}(\mathcal{S})$ is the Borel $\sigma$-algebra of $\mathcal{S}$. One says that $(\mathcal{D},\omega)$ is a} complex projective $t$-design\textit{ when under Lebesgue-Stieltjes integration}
\begin{equation}
\int \text{d}\omega(\pi_{\alpha})\pi_{\alpha}^{\otimes^{t}}=\frac{t!(d-1)!}{(d+t-1)!}\Pi_{\text{sym}}^{(t)}\text{.}
\end{equation}
\end{definition}

\noindent Complex projective $t$-designs \cite{Neumaier1981}\cite{Hoggar1982}\cite{Zauner1999}\cite{Scott2006} are highly symmetric structures on the $(2d-2)$-dimensional manifolds of pure quantum states in quantum theory with vital applications \cite{Adamson2010}\cite{FernandezPerez2011}\cite{Zhu2012}\cite{Tabia2012}\cite{Bennett1984}\cite{Bruss1998}\\\cite{Fuchs2003}\cite{Renes2005}\cite{Mafu2013}\cite{Bennett1992}\cite{Bennett1993}\cite{Spengler2012}\cite{Rastegin2013}\cite{Chen2015}\cite{Cerf2002} (with references given being representative only) in quantum information science. In Hilbert dimension $d=2$, the manifold of pure quantum states is equivalent to the Bloch sphere in $\mathbb{R}^{3}$; although, when $d>2$ the geometry of pure quantum states is \textit{much} more subtle. The construction of complex projective $t$-designs is therefore nontrivial. Remarkably, Seymour and Zaslavsky \cite{Seymour1984} proved that complex projective $t$-designs exist for any independent choices of $t$ and $d$. In \cite{Hayashi2005}, Hayashi-Hashimoto-Horibe provide explicit constructions based on Gauss-Legendre quadratures \cite{Davis2007}; however, the cardinality of their constructions scales exponentially in $d$. The construction of more efficient\footnote{In \cite{Ambainis2007}, Ambainis and Emerson introduce the notion of \textit{approximate} complex projective $t$-designs, and provide explicit constructions whose cardinality scales linearly with $d$.} complex projective $t$-designs remains an open problem. In quantum information theory, the vast majority of literature on the subject concerns the case of finite $\mathcal{S}$, uniform $\omega$, and $t=2$. Indeed, all of the aforementioned references to applications concern instances from this case, which we shall call \textit{projective 2-designs} for short to avoid carrying the term `finite uniform complex projective 2-designs.'
\begin{definition}\label{pdDef}
\textit{A} projective 2-design \textit{is a finite uniform complex projective 2-design, i.e.}
\begin{equation}
\sum_{\alpha=1}^{n}\pi_{\alpha}\otimes\pi_{\alpha}=\frac{2n}{d(d+1)}\Pi_{\text{sym}},\hspace{1cm}n\equiv\text{card}\mathcal{D}\text{.}
\label{pdDefEq}
\end{equation}
\end{definition}
\noindent It is well known that any projective 2-design defines a \textsc{povm} via uniform subnormalization, that is, if $\{\pi_{\alpha}\}\subset\mathcal{Q}(\mathcal{H}_{d})$ is a projective 2-design, then $\{\frac{d}{n}\pi_{\alpha}\}\subset\mathcal{E}(\mathcal{H}_{d})$ is a \textsc{povm}. For the proof, simply note from Eq.~\eqref{piSymAsym} that the partial trace of Eq.~\eqref{pdDefEq} with respect to either factor in $\mathcal{H}_{d}\otimes\mathcal{H}_{d}$ yields
\begin{eqnarray}
\sum_{\alpha=1}^{n}\pi_{\alpha}&=&\frac{2n}{d(d+1)}\left(\frac{d}{2}\mathds{1}+\frac{1}{2}\sum_{r=1}^{d}\sum_{s=1}^{d}|e_{s}\rangle\delta_{r,s}\langle e_{r}|\right)\nonumber\\
&=&\frac{n}{d}\mathds{1}_{d}.
\label{desPovm}
\end{eqnarray}
We call such a \textsc {povm} a \textit{projective 2-design} \textsc{povm}. It is also easy to see that any projective 2-design \textsc{povm} is informationally complete (this follows from our more general \cref{desCons} and \cref{rank1Lem}.) Explicitly, any $\rho\in\mathcal{Q}(\mathcal{H}_{d})$ can \cite{Scott2006} be expanded in terms of $p_{\alpha}\equiv d\mathrm{Tr}(\rho\pi_{\alpha})/n$ and $\pi_{\alpha}$ as follows
\begin{equation}
\rho=\sum_{\alpha=1}^{n}\left((d+1)p_{\alpha}-\frac{d}{n}\right)\pi_{\alpha}\text{.}
\label{pdEpn}
\end{equation}
For completeness, we note that Eq.~\eqref{pdEpn} follows from Eq.~\eqref{desPovm}, for in light Eq.~\eqref{piSymAsym} and Eq.~\eqref{pdDefEq}
\begin{eqnarray}\sum_{\alpha=1}^{n}\pi_{\alpha}\rho\otimes\pi_{\alpha}&=&\frac{n}{d(d+1)}\sum_{r=1}^{d}\sum_{s=1}^{d}\Big(\big(|e_{r}\rangle\langle e_{r}|\rho\otimes|e_{s}\rangle\langle e_{s}|\big)+\big(|e_{s}\rangle\langle e_{r}|\rho\otimes|e_{r}\rangle\langle e_{s}|\big)\Big)\label{forPT}\text{.}
\end{eqnarray}
Tracing Eq.~\eqref{forPT} over the first tensor factor in $\mathcal{H}_{d}\otimes\mathcal{H}_{d}$ yields 
\begin{eqnarray}
\sum_{\alpha=1}^{n}p_{\alpha}\pi_{\alpha}&=&\frac{1}{(d+1)}\sum_{r=1}^{d}\sum_{s=1}^{d}\Big(|e_{s}\rangle\langle e_{r}|\rho e_{r}\rangle\langle e_{s}|+|e_{r}\rangle\langle e_{r}|\rho e_{s}\rangle\langle e_{s}|\Big)\nonumber\\
&=&\frac{1}{(d+1)}\sum_{s=1}^{d}\Big(|e_{s}\rangle\mathrm{Tr}(\rho)\langle e_{s}|+\mathds{1}_{d}\rho |e_{s}\rangle\langle e_{s}|\Big)\nonumber\\
&=&\frac{1}{(d+1)}\Big(\mathds{1}_{d}+\rho\Big)\label{ppi}\text{,}
\end{eqnarray} 
Following through the foregoing equations with $\rho$ replaced by an arbitrary self-adjoint linear endomorphism, $p_{\alpha}$ may no longer be such that $0\leq p_{\alpha}\leq 1$, however one finds that $\forall A\in\mathcal{L}_{\text{sa}}(\mathcal{H}_{d})$
\begin{equation}
A=\sum_{\alpha=1}^{n}\left((d+1)p_{\alpha}-\frac{d}{n}\mathrm{Tr}(A)\right)\pi_{\alpha}\text{,}\hspace{1cm} p_{\alpha}\equiv \frac{d\mathrm{Tr}\big(A\pi_{\alpha}\big)}{n}\text{.}
\end{equation}
Any projective 2-design \textsc{povm} formed from a projective 2-design $\mathcal{D}$ thus facilitates an injection
\begin{equation}
\mathfrak{i}_{\mathcal{D}}:\mathcal{L}_{\text{sa}}(\mathcal{H}_{d})\longrightarrow\mathbb{R}^{n}::A\longmapsto \vec{p}\equiv(p_{1},\dots,p_{n})\text{.}
\label{desInj}
\end{equation}
It is not hard to prove that $\mathfrak{i}_{\mathcal{D}}$ is an injection. We shall follow a general technique pointed out by Watrous \cite{Watrous2011}. Indeed, $np_{\alpha}/d$ is precisely the Hilbert-Schmidt inner product $\langle\!\langle\pi_{\alpha}|A\rangle\!\rangle$; moreover, in light of \cref{desCons}, $\mathrm{span}_{\mathbb{R}}\{\pi_{\alpha}\}=\mathcal{L}_{\text{sa}}(\mathcal{H}_{d})$. Thus, if $\forall\alpha\in\{1,\dots,n\}$ one has that $\langle\!\langle\pi_{\alpha}|A\rangle\!\rangle=\langle\!\langle\pi_{\alpha}|B\rangle\!\rangle$ for some $A,B\in\mathcal{L}_{\text{sa}}(\mathcal{H}_{d})$, then $\forall Z\in\mathcal{L}_{\text{sa}}(\mathcal{H}_{d})$ it follows that $\langle\!\langle Z|A-B\rangle\!\rangle=0$, so $A=B$. Therefore $\mathfrak{i}_{\mathcal{D}}$ is injective; moreover, it is a linear monomorphism in light of linearity of the trace functional.

\noindent With $\mathcal{D}$ any projective $2$-design, the linear monomorphism $\mathfrak{i}_{\mathcal{D}}:\mathcal{L}_{\text{sa}}\longrightarrow\mathbb{R}^{n}$ defined in Eq.~\eqref{desInj} injects quantum state space into the $n$\textit{-probability simplex} $\Delta_{n}\equiv\{\vec{p}\in\mathbb{R}_{\geq0}^{n}\;\boldsymbol{|}\;\sum_{\alpha}p_{\alpha}=1\}$, \textit{i.e}. $\mathfrak{i}_{\mathcal{D}}::\mathcal{Q}(\mathcal{H}_{d})\longmapsto\Delta_{n}$. In light of the linearity of $\mathfrak{i}_{\mathcal{D}}$, the image of the quantum state space is clearly a convex set. Furthermore, again from linearity, $\vec{p}\in\mathfrak{i}_{\mathcal{D}}\big(\mathcal{Q}(\mathcal{H}_{d})\big)$ is an extreme point if and only if there exists a unit rank projector $\pi\in\text{Pur}\mathcal{Q}(\mathcal{H}_{d})$ (\textit{i.e}.\ a pure quantum state) such that $\vec{p}=\mathfrak{i}_{\mathcal{D}}(\pi)$. Further still from linearity, $\mathfrak{i}_{\mathcal{D}}\big(\mathcal{Q}(\mathcal{H}_{d})\big)$ is precisely the convex hull of its extreme points. It is to be observed, however, that the restriction of $\mathfrak{i}_{\mathcal{D}}$ to quantum state space does not cover the entire $n$-probability simplex. Put otherwise, not every probability vector corresponds to a quantum state. Indeed, appealing to linearity of the trace and the fact $\mathrm{Tr}(A\otimes B)=\mathrm{Tr}A\mathrm{Tr}B$, we see that $\mathfrak{i}_{\mathcal{D}}\big(\mathcal{Q}(\mathcal{H}_{d})\big)$ is bounded by a sphere of radius $2d/(n(d+1))$, specifically
\begin{eqnarray}
\sum_{\alpha=1}^{n}p_{\alpha}^{2}&=&\frac{d^{2}}{n^{2}}\sum_{\alpha=1}^{n}\Big(\mathrm{Tr}(\rho\pi_{\alpha})\Big)^{2}\nonumber\\
&=&\frac{d^{2}}{n^{2}}\sum_{\alpha=1}^{n}\mathrm{Tr}\Big(\rho\pi_{\alpha}\otimes\rho\pi_{\alpha}\Big)\nonumber\\
&=&\frac{2d}{n(d+1)}\mathrm{Tr}\Big(\big(\rho\otimes\rho\big)\Pi_{\text{sym}}\Big)\nonumber\nonumber\\
&=&\frac{d}{n(d+1)}\left(1+\sum_{r=1}^{d}\sum_{s=1}^{d}\langle e_{s}|\rho|e_{r}\rangle\langle e_{r}|\rho|e_{s}\rangle\right)\nonumber\\
&=&\frac{d}{n(d+1)}\big(1+\mathrm{Tr}\rho^{2}\big)\text{.}
\label{pdRad}
\end{eqnarray}
Incidentally, from Eq.~\eqref{pdRad}, one gleans that projective 2-design \textsc{povm} probabilities are never certain, for $\|\vec{p}\|<1\implies\text{max}\{p_{\alpha}\}<1$. Also, Eq.~\eqref{pdRad} establishes simple linear relationship between the usual Euclidean norm on $\mathbb{R}^{n}$ and the Hilbert-Schmidt norm on $\mathcal{L}_{\text{sa}}(\mathcal{H}_{d})$. Indeed, each pure state, which thus satisifies $\mathrm{Tr}\rho^{2}=1$, is a point on a sphere in $\mathbb{R}^{n}$ of the aforementioned radius. We are about to see, however, that not every point on the sphere corresponds to a pure state. Indeed, when $d>2$, the intersection of this sphere with the simplex contains points outside of $\mathfrak{i}_{\mathcal{D}}\big(\mathcal{Q}(\mathcal{H}_{d})\big)$. First, we shall need to recall the following lemma, which provides necessary and sufficient conditions for a self-adjoint linear endomorphism on $\mathcal{H}_{d}$ to be an extreme point of the convex set of quantum states $\mathcal{Q}(\mathcal{H}_{d})$.

\begin{lemma} (Flammia-Jones-Linden \cite{Flammia2004}\cite{Jones2005}) \textit{Let} $A\in\mathcal{L}_{\text{sa}}(\mathcal{H}_{d})$\textit{. Then} $A$ \textit{is a unit rank projector (i.e.\ pure quantum state) if and only if} $\mathrm{Tr}(A^{2})=\mathrm{Tr}(A^{3})=1$\textit{.}
\label{FJL}
\end{lemma}
\noindent We refer the reader to \cite{Jones2005} for the proof of \cref{FJL} (with \cite{Flammia2004} being unpublished.) It is worth pointing out that when $d=2$, and only in this case, the conditions $\mathrm{Tr}(A)=\mathrm{Tr}(A^{2})=1$ are necessary and sufficient conditions for $A$ to be a pure quantum state. Indeed, with $\lambda_{1},\lambda_{2}$ the eigenvalues of $A$, these conditions imply, respectively, that $\lambda_{1}+\lambda_{2}=1$ and $\lambda_{1}^{2}+\lambda_{2}^{2}=1$. Therefore, when $d=2$, $\vec{p}\in\Delta_{n}$ corresponds to a pure quantum state if and only if $\|\vec{p}\|=4/3n$. When $d>2$, the cubic condition in \cref{FJL} identifies the subset of points on sphere $\|\vec{p}\|=2d/(n(d+1))$ corresponding to pure quantum states. Remarkably, the conditions in \cref{FJL} lead to a very simple characterization of the extreme points of the quantum state space in terms of projective 2-design \textsc{povm} probabilities, specifically the characterization we provide in \cref{pureDesCon}. Fuchs and Schack \cite{Fuchs2013} derived this characterization in the special case where the underlying projective 2-design is a \textsc{sic} (see \cref{sicDef}.) The fact that our characterization holds for an arbitrary projective $2$-design is therefore of some independent interest.

\begin{corollary}\label{pureDesCon}\textit{Let} $\{E_{\alpha}\equiv d\pi_{\alpha}/n\}\subset\mathcal{E}(\mathcal{H}_{d})$ \textit{be a projective 2-design} \textsc{povm}\textit{, i.e.} $\{\pi_{1},\dots,\pi_{n}\}\subset\mathcal{Q}(\mathcal{H}_{d})$ \textit{is a projective 2-design. Let} $\rho\in\mathcal{Q}(\mathcal{H}_{d})$\textit{ and} $p_{\alpha}\equiv\mathrm{Tr}(\rho E_{\alpha})$\textit{. Then} $\rho$ \textit{is pure if and only if}
\begin{eqnarray}
\sum_{\alpha=1}^{n}p_{\alpha}^{2}&=&\frac{2d}{n(d+1)}\text{,}\label{quad}\\
\sum_{\alpha=1}^{n}\sum_{\beta=1}^{n}\sum_{\gamma=1}^{n}p_{\alpha}p_{\beta}p_{\gamma}\mathrm{Tr}(\pi_{\alpha}\pi_{\beta}\pi_{\gamma})&=&\frac{d+7}{(d+1)^{3}}\text{.}\label{cube}
\end{eqnarray}
\end{corollary}
\begin{proof}Eq.~\eqref{quad} is an immediate consequence of Eq.~\eqref{pdRad} and \cref{FJL}. It remains to establish that the cubic condition Eq.~\eqref{cube} holds. Let arbitrary $\rho\in\mathcal{Q}(\mathcal{H}_{d})$. Then, via Eq.~\eqref{pdEpn}, we observe that
\begin{eqnarray}
\mathrm{Tr}\big(\rho^{3}\big)=(d+1)^{3}\sum_{\alpha=1}^{n}\sum_{\beta=1}^{n}\sum_{\gamma=1}^{n}p_{\alpha}p_{\beta}p_{\gamma}\mathrm{Tr}(\pi_{\alpha}\pi_{\beta}\pi_{\gamma})-3(d+1)^{2}\sum_{\alpha=1}^{n}\sum_{\beta=1}^{n}p_{\alpha}p_{\beta}\mathrm{Tr}(\pi_{\alpha}\pi_{\beta})+2d+3\text{.}
\end{eqnarray}
Now from Eq.~\eqref{ppi}, together with linearity of the trace, we find that
\begin{equation}
\sum_{\alpha=1}^{n}\sum_{\beta=1}^{n}p_{\alpha}p_{\beta}\mathrm{Tr}(\pi_{\alpha}\pi_{\beta})=\sum_{\beta=1}^{n}p_{\beta}\mathrm{Tr}\Bigg(\Big(\sum_{\alpha=1}^{n}p_{\alpha}\pi_{\alpha}\Big)\pi_{\beta}\Bigg)=\frac{1}{(d+1)}\left(1+\frac{n}{d}\sum_{\beta=1}^{n}p_{\beta}^{2}\right)\text{.}
\end{equation}
Hence,
\begin{equation}
\sum_{\alpha=1}^{n}\sum_{\beta=1}^{n}\sum_{\gamma=1}^{n}p_{\alpha}p_{\beta}p_{\gamma}\mathrm{Tr}(\pi_{\alpha}\pi_{\beta}\pi_{\gamma})=\frac{\mathrm{Tr}\big(\rho^{3}\big)-2d-3}{(d+1)^{3}}+\frac{3}{(d+1)^{2}}\left(1+\frac{n}{d}\sum_{\beta=1}^{n}p_{\beta}^{2}\right)\text{.}
\label{forCube}
\end{equation}
Therefore in light of Eq.~\eqref{pdRad} we see that
\begin{equation}
\sum_{\alpha=1}^{n}\sum_{\beta=1}^{n}\sum_{\gamma=1}^{n}p_{\alpha}p_{\beta}p_{\gamma}\mathrm{Tr}(\pi_{\alpha}\pi_{\beta}\pi_{\gamma})=\frac{\mathrm{Tr}\big(\rho^{3}\big)+3\mathrm{Tr}\big(\rho^{2}\big)+3+d}{(d+1)^{3}}\text{.}
\label{forCube2}
\end{equation}
In the case where $\rho$ is pure, and only in this case, Eq.~\eqref{forCube2} immediately simplifies to Eq.~\eqref{cube}\text{.}
\end{proof}
\noindent With \cref{pureDesCon}, we have established necessary and sufficient conditions for a probability distribution $\vec{p}$ to be an extreme point of the image of the quantum state space, \textit{i.e}.\ the convex set $\mathfrak{i}_{\mathcal{D}}\big(\mathcal{Q}(\mathcal{H}_{d})\big)$. It is worth emphasizing that the quadratic and cubic conditions, Eq.~\eqref{quad} and Eq.~\eqref{cube}, respectively, apply quite generally: they hold for \textit{any} projective 2-design \textsc{povm}. Furthermore, it is only in the quadratic condition that any information regarding the particular design factors in, namely the cardinality $n$. It is natural to expect that more progress can be made by selecting particularly nice varieties of projective 2-design. We are about to meet the two canonical candidates. First, it will be convenient for us to recall the following lemma, which appears in \cite{Renes2004a}. In light of the fact that any projective 2-design \textsc{povm} is informationally complete, we have improved the statement of this result as it is found in \cite{Renes2004a} so that $n\geq d^{2}$, instead of $n\geq d$ as it appears in \cite{Renes2004a} (this is inconsequential for the proof given therein.)

\begin{lemma}\label{RenesLem} (Renes \textit{et al.} \cite{Renes2004a}) \textit{A set of unit rank projectors} $\{\pi_{1},\dots,\pi_{n}\}\subset\mathcal{Q}(\mathcal{H}_{d})$ \textit{with} $n\geq d^{2}$ \textit{is a projective 2-design if and only if}
\begin{equation}
\sum_{\alpha=1}^{n}\sum_{\beta=1}^{n}\Big(\mathrm{Tr}(\pi_{\alpha}\pi_{\beta})\Big)^{2}=\frac{2n^{2}}{d(d+1)}\text{.}
\end{equation}
\end{lemma}
\noindent We refer the reader to \cite{Renes2004a} for a proof of \cref{RenesLem} (see also Levenshtein \cite{Levenshtein1998} and K{\"o}nig \cite{Konig1999}.) Equipped with \cref{RenesLem}, one now has a necessary and sufficient condition that can be used to verify whether a set of unit rank projectors constitutes a projective 2-design. For example, \textit{suppose} there exists a set of unit rank projectors $\{\pi_{1},\dots,\pi_{d^{2}}\}\subset\mathcal{Q}(\mathcal{H}_{d})$ such that $\forall\alpha,\beta\in\{1,\dots,d^{2}\}$ one has that
\begin{eqnarray}
\mathrm{Tr}\big(\pi_{\alpha}\pi_{\beta}\big)=\frac{\delta_{\alpha,\beta}d+1}{d+1}\text{.}
\label{sicProds}
\end{eqnarray}
Then 
\begin{equation}
\sum_{\alpha=1}^{d^{2}}\sum_{\beta=1}^{d^{2}}\Big(\mathrm{Tr}(\pi_{\alpha}\pi_{\beta})\Big)^{2}=\sum_{\alpha=1}^{d^{2}}\sum_{\beta=1}^{d^{2}}\left(\frac{\delta_{\alpha,\beta}d+1}{(d+1)}\right)^{2}=\frac{\sum_{\alpha=1}^{n}\sum_{\beta=1}^{n}d^{2}\delta_{\alpha,\beta}+2d\delta_{\alpha,\beta}+1}{(d+1)^{2}}=\frac{2d^{3}}{(d+1)}\text{,}
\end{equation}
so the putative $\{\pi_{1},\dots,\pi_{d^{2}}\}$ constitute a projective 2-design in light of \cref{RenesLem}; moreover, one of \textit{minimal cardinality}, \textit{i.e}.\ $d^{2}$. In \cite{Scott2006}, Scott proves that any projective 2-design of minimal cardinality must be symmetric in the sense of Eq.~\eqref{sicProds}. It follows that \textsc{sic}s, which are formally defined below, are the unique projective 2-designs of minimal cardinality.
\begin{definition}\label{sicDef} \textit{A} \textsc{sic} \textit{is a set of} $d^{2}$ \textit{unit rank projectors (i.e.\ pure states)} $\{\pi_{1},\dots,\pi_{d^{2}}\}\subset\mathcal{Q}(\mathcal{H}_{d})$ \textit{such that} $\forall\alpha,\beta\in\{1,\dots,d^{2}\}$ $\mathrm{Tr}(\pi_{\alpha}\pi_{\beta})=(\delta_{\alpha,\beta}d+1)/(d+1)$\text{. A} \textsc{sic povm} \textit{is a set of} $d^{2}$ \textit{effects} $\{E_{\alpha}\equiv \pi_{\alpha}/d\}\subset\mathcal{E}(\mathcal{H}_{d})$ \textit{such that} $\{\pi_{\alpha}\}$ \textit{is a} \textsc{sic}\textit{.}
\end{definition}

\noindent The \textsc{sic} \textit{existence problem} --- that is, the question whether \textsc{sic}s exist for all $d\in\mathbb{N}$ --- remains tantalizingly open, notwithstanding significant attention, particularly from within the quantum information community \cite{Zauner1999}\cite{Renes2004a}\cite{Renes2004c}\cite{Grassl2004}\cite{Grassl2005}\cite{Appleby2005}\cite{Flammia2006}\cite{Grassl2008}\cite{Appleby2009a}\cite{Appleby2009b}\cite{Scott2010}\cite{Appleby2011}\cite{Appleby2012}\cite{Zhu2012}\cite{Appleby2013}\cite{Tabia2013}\cite{Appleby2014b}\cite{Khatirinejad2008a}\cite{Appleby2014}\cite{Appleby2015}\cite{Zhu2015}\cite{Dang2015} (a\\ representative list of references.) The \textsc{sic} existence problem, in the words and original emphasis of Appleby, Dang, and Fuchs \cite{Appleby2014}: ``\textit{seems} the sort of thing one might find as an exercise in a linear-algebra textbook.'' For instance, given that a \textsc{sic} is defined by unit rank projectors, the question of whether or not a \textsc{sic} exists in $\mathcal{Q}(\mathcal{H}_{d})$ is equivalent to the question of whether or not $d^{2}$ equiangular lines \cite{Khatirinejad2008b}\cite{Hoggar1982} exist in $\mathcal{H}_{d}$; a \textit{seemingly} simple geometric question. Consider, however, the following equivalent geometric characterization: a \textsc{sic} is exactly defined by the vertices of a regular $d^{2}$-vertex simplex within the \textit{null trace subspace}
\begin{equation} 
\mathcal{L}_{\text{sa},0}(\mathcal{H}_{d})\equiv\{B\in\mathcal{L}_{\text{sa}}(\mathcal{H}_{d})\;\boldsymbol{|}\;\mathrm{Tr}B=0\}\text{,}
\end{equation} 
where the vertices correspond to pure quantum states via the usual identification $B=d\rho-\mathds{1}_{d}$. We shall return to a full analysis of quantum state spaces in terms traceless generators (what we call the Bloch body) in \cref{blochBod}. Presently, we need only observe that it is, of course, trivial to inscribe such a simplex into the unit ball of the ambient real vector subspace $\mathcal{L}_{\text{sa},0}(\mathcal{H}_{d})$; however, it is very far from obvious that one can rotate this simplex such that its vertices lie on the measure zero $(2d-2)$-dimensional submanifold corresponding to pure quantum states. This picture provides one with a sense for why the \textsc{sic} existence problem is so hard. It is therefore remarkable, and perhaps even shocking, that \textsc{sic}s \textit{do} exist in many cases, particularly for most finite Hilbert dimensions naturally accessible by quantum computers \cite{Ladd2010}. 

\noindent Exact \textsc{sic}s are known to exist for $d\in\{2,3,4,5,6,7,8,9,10,11,12,13,14,15,16,19,24,28,35,48\}$ (see in particular\footnote{For concrete numerical solutions, we especially recommend \cite{Scott2010}.} \cite{Zauner1999}\cite{Grassl2004}\cite{Grassl2005}\cite{Appleby2005}\cite{Scott2010}\cite{Appleby2012}\cite{Appleby2014b}) and numerical solutions have been constructed for all $d\leq 121$ \cite{Renes2004a}\cite{Scott2010}\cite{Scott2014}. In 1999, Zauner \cite{Zauner1999} conjectured that for \textit{all} $d$ there exists a fiducial unit rank projector whose orbit under the action of the generalized Weyl-Heisenberg group \cite{Appleby2005} forms a \textsc{sic}. Almost all known \textsc{sic}s have been constructed in this manner (if $d$ is prime, they must be as such \cite{Zhu2010}.)

\noindent There are many reasons for why one might be interested in the \textsc{sic} existence problem. For instance, from a practical perspective, \textsc{sic} \textsc{povm}s are \cite{Scott2006} \textit{the} optimal quantum measurements for the purposes of nonadaptive sequential linear quantum state tomography \cite{Paris2004} and measurement-based quantum cloning \cite{Gisin1997}. They also have important applications \cite{Fuchs2003}\cite{Englert2004}\cite{Renes2004c}\cite{Renes2005}\cite{Durt2008} in quantum cryptography \cite{Bennett1984}. From a purely mathematical point of view, \textsc{sic}s have interesting connections with Lie algebras \cite{Appleby2011}, Jordan algebras \cite{Appleby2015}, Galois field theory \cite{Appleby2013}\cite{Dang2015}. In the context of quantum foundations, \textsc{sic}s are a cornerstone for QBism \cite{Fuchs2010}\cite{Fuchs2013}. \textsc{sic}s also have applications in \textit{classical} signal processing \cite{Howard2005}. In \cref{entanglementConicalDesigns}, we shall see that \textsc{sic}s are intimately connected with the theory of quantum entanglement.

\noindent Mutually unbiased bases are another class of symmetric structures in Hilbert space.

\begin{definition} (\cite{Schwinger1960}\cite{Ivanovic1981}\cite{Wootters1989}\cite{Durt2010}) Mutually unbiased \textit{bases are sets of orthonormal bases $\mathcal{B}_{b}=\{e_{1,b},\dots,e_{d,b}\}\subset\mathcal{H}_{d}$, where $b\in\{1,\dots,N\}$ and $N\in\mathbb{N}$, such that $|\langle e_{\alpha,b}|e_{\alpha',b'}\rangle|^{2}=1/d$ for all $\alpha,\alpha'\in\{1,\dots,d\}$ and for all $b\neq b'$. }
\end{definition}

\noindent If $N=d+1$, then the corresponding unit rank projectors $\pi_{\alpha,b}\equiv|e_{\alpha,b}\rangle\langle e_{\alpha,b}|$ define a \textsc{mub} as follows. 

\begin{definition}\label{mubDef} \textit{A} \textsc{mub} \textit{is a set of} $d(d+1)$ \textit{unit rank projectors (i.e.\ pure quantum states)} $\mathcal{Q}(\mathcal{H}_{d})\supset\{\pi_{\alpha,b}\;\boldsymbol{|}\;\alpha\in\{1,\dots,d\},b\in\{1,\dots,d+1\}\}$ \textit{such that} $\forall\alpha,\alpha'\in\{1,\dots,d\}$ \textit{and} $\forall b,b'\in\{1,\dots,d+1\}$
\begin{equation}
\mathrm{Tr}\big(\pi_{\alpha,b}\pi_{\alpha',b'}\big)=\begin{cases}\delta_{\alpha,\alpha'} & b=b'\\ \frac{1}{d} & b\neq b'\end{cases}\text{.}
\end{equation}
\textit{A} \textsc{mub} \textsc{povm} \textit{is a set of} $d(d+1)$ \textit{effects} $\{E_{\alpha,b}\equiv \pi_{\alpha,b}/(d+1)\}\subset\mathcal{E}(\mathcal{H}_{d})$ \textit{such that} $\{\pi_{\alpha,b}\}$ \textit{is a} \textsc{mub}.
\end{definition}

\noindent It turns out that $d+1$ is \cite{Wootters1989} the maximum possible number of mutually unbiased bases that might be found in a finite $d$-dimensional complex Hilbert space, hence sets of $d+1$ mutually unbiased bases are said to be \textit{full}. In fact, the unit rank projectors corresponding to mutually unbiased basis elements form a projective 2-design if and only if $N=d+1$. For the proof\footnote{Barnum was the first to prove \cite{Barnum2000} that \textsc{mub}s are projective 2-designs.}, one simply notes that $N(N-1)+Nd=2N^{2}d/(d+1)$ if and only if $N=d+1$, calls to \cref{RenesLem}, and computes
\begin{equation}
\sum_{\alpha=1}^{d}\sum_{\alpha'=1}^{d}\sum_{b=1}^{N}\sum_{b'=1}^{N}\Big(\mathrm{Tr}(\pi_{\alpha,b}\pi_{\alpha',b'})\Big)^{2}=\frac{1}{d^{2}}\sum_{\alpha=1}^{d}\sum_{\alpha'=1}^{d}\sum_{b=1}^{N}\sum_{b'=1}^{N}\Big(1+\delta_{b,b'}\big(d^{2}\delta_{\alpha,\alpha'}-1\big)\Big)=N(N-1)+Nd\text{.}
\label{mubCalc}
\end{equation}

\noindent We point the reader to the excellent review article \cite{Durt2010} for many interesting mathematical connections and physical applications related to \textsc{mub}s. Geometrically, \textsc{mub}s correspond to $d+1$ orthogonal regular $d$-vertex simplices in $\mathcal{L}_{\text{sa,0}}(\mathcal{H}_{d})$. When $d$ is a prime power, \textsc{mub}s exist \cite{Wootters1989}. No \textsc{mub}s have been found in any other cases; moreover, there is a growing body of numerical evidence \cite{Grassl2004}\cite{Butterley2007}\cite{Brierley2008}\cite{Raynal2011} supporting Zauner's conjecture \cite{Zauner1999} that \textsc{mub}s do \textit{not} exist when $d=6$.  

\noindent A terminological digression is in order. The name `symmetric informationally complete positive operator valued measures' was introduced in \cite{Renes2004a} for what we have just called a \textsc{sic} \textsc{povm}. Indeed, \textsc{sic} \textsc{povm}s are by definition symmetric in the sense of \eqref{sicProds}, and informationally complete (as is any projective 2-design \textsc{povm}.) There exists \cite{Appleby2007}, however, a more general (\textit{i.e}.\ arbitrary rank) variety of minimal informationally complete \textsc{povm}s that are also symmetric in the sense of admitting constant mutual Hilbert-Schmidt inner products, what Appleby and the author have dubbed \textsc{sim}s for `symmetric informationally complete measurements' in \cite{Graydon2015a} (see \cref{simDef}). In this language, a \textsc{sic} \textsc{povm} is by definition a \textsc{sim} \textit{of unit rank}. We shall have the occasion to consider the general case of \textsc{sim}s in detail in \cref{designsOnQuantumCones}. \textsc{mub}s, like \textsc{sic}s, have arbitrary rank counterparts --- mutually unbiased measurements (\textsc{mum}s, see \cref{mumDef}) --- that do \cite{Kalev2014} exist for all $d\in\mathbb{N}$. Incidentally, \textsc{sim}s and \textsc{sic}s and \textsc{mum}s and \textsc{mub}s are examples of conical 2-designs. The stage is set.

\chapter{Enter Conical Designs}
\label{designsOnQuantumCones}

\epigraphhead[40]
	{
		\epigraph{``\dots deep beneath the rolling waves,\\ in labyrinths of coral caves\dots''}{---\textit{Pink Floyd}\\ Echoes (1971)}
	}
	
In \cite{Graydon2015a}, the author and D.\ M.\ Appleby introduced conical $t$-designs. This chapter centres on that publication. 
\setcounter{theorem}{0}
\begin{definition} \textit{Let} $d,n,t\in\mathbb{N}$\textit{. A} conical $t$-design \textit{is a set of elements of a quantum cone, denoted} $\{A_{1},\dots,A_{n}\}\subset\mathcal{L}_{\text{sa}}(\mathcal{H}_{d})_{+}$\textit{, such that} 
\begin{equation}
\mathrm{span}_{\mathbb{R}}\{A_{1},\dots,A_{n}\}=\mathcal{L}_{\text{sa}}(\mathcal{H}_{d})\textit{,}
\label{spanDemand}
\end{equation}
\textit{and} $\forall U\in\text{U}(\mathcal{H}_{d})$
\begin{equation}
\left[\sum_{j=1}^{n}A_{j}^{\otimes ^{t}},U^{\otimes^{t}}\right]=0\text{,}
\end{equation}
\textit{where} $\forall A,B\in\mathcal{L}(\mathcal{H}_{d})$ $[A,B]\equiv AB-BA$ \textit{and} $A^{\otimes^{1}}=A$\textit{, } $A^{\otimes^{2}}=A\otimes A$\textit{, }$A^{\otimes^{3}}=A\otimes A\otimes A$\textit{, and so on.}
\end{definition}

\noindent The demand in Eq.~\eqref{spanDemand} that a conical $t$-design span the ambient space for the relevant quantum cone is inspired by our detailed consideration of conical 2-designs, in particular their five equivalent characterizations in our upcoming \cref{desCons}. If one restricts to the case $\forall j\in\{1,\dots,n\}$ $\text{rank}A_{j}=1$, then, up to scaling, one recovers the definition of a complex projective $t$-design; moreover, in this case, the spanning condition becomes redundant for $t>1$. For a proof of the later statement, one notes that (i) any complex projective $t$-design is \cite{Scott2006} also a complex projective $t'$-design for all $\mathbb{N}\ni t'\leq t$, and (ii), as we have shown in \cref{projectiveDesigns}, every complex projective 2-design \textsc{povm} is informationally complete. A proof of (i) follows from Schur-Weyl duality, as discussed in \cref{grpTheory}, and Lemma 1 in \cite{Scott2006}. 

\noindent For the remainder of this thesis, as in \cite{Graydon2015a} and \cite{Graydon2016}, we shall focus our attention entirely on the case $t=2$, deferring a consideration of higher conical $t$-designs to later work. The Bloch body will be central to our analysis. Accordingly, we summarize its structure in \cref{blochBod}. We provide a unified geometric picture of \textsc{sim}s and \textsc{mum}s in \cref{simmumery}. We formally introduce conical 2-designs in \cref{desQC} and derive their basic properties. The class of conical 2-designs is large. In \cref{hc2d}, we focus on a particular subclass that we call \textit{homogeneous conical designs}, which includes arbitrarily weighted complex projective 2-designs, as well as \textsc{sim}s and \textsc{mums}. We fully characterize homogeneous conical 2-designs in \cref{blochPoly}, and we prove they exist in \cref{existThm}. Conical 2-designs shed new light on the problem of finding new varieties of complex \textit{projective} 2-designs. We outline a program to seek out such varieties in \cref{inSearchOf}.



\newpage
\section{The Bloch Body}\label{blochBod}
\noindent The generalized Bloch representation \cite{Harriman1978}\cite{Mahler1995}\cite{Siennicki2001}\cite{Kimura2003}\cite{Schirmer2004}\cite{Kimura2005}\cite{Dietz2006}\cite{Appleby2007} of finite dimensional quantum state spaces provides a particularly intuitive way of thinking about designs, both projective and conical. In this thesis, we shall adopt the coordinate free point of view, wherein traceless self-adjoint endormorphisms themselves are the vectors corresponding to quantum states. One may always appeal to, however a concrete represenation. For instance, the generalized Gell-Mann matrices\footnote{Incidentally, the generalized Gell-Mann matrices appear in some of our proofs given later in \cref{partII}; accordingly, we provide an elementary review in \cref{gmApp}.} \cite{Bertlmann2008} can be chosen as a basis for all self-adjoint traceless matrices. Indeed, when $d=2$, this latter point of view yields the familiar Bloch vector representation for qubits in terms of the vector of projections onto the Pauli matrices. In our case, the key is to observe that any quantum state $\rho\in\mathcal{Q}(\mathcal{H}_{d})$ can be expressed as $\rho=(B+\mathds{1}_{d})/d$, with traceless $B\in\mathcal{L}_{\text{sa},0}(\mathcal{H}_{d})$; hence the following definition.

\begin{definition}\label{blochDef}\textit{Let} $d\in\mathbb{N}_{\geq 2}$\textit{. The} Bloch body\textit{, denoted} $\boldsymbol{\mathcal{B}}(\mathcal{H}_{d})$\textit{, consists of those traceless self-adjoint linear endomorphisms} $B\in\mathcal{L}_{\text{sa,0}}(\mathcal{H}_{d})$ \textit{such that} $(B+\mathds{1}_{d})/d$ \textit{is a quantum state, i.e.}
\begin{equation} \boldsymbol{\mathcal{B}}(\mathcal{H}_{d})\equiv\left\{B\in\mathcal{L}_{\text{sa,0}}(\mathcal{H}_{d})\;\boldsymbol{\Big|}\frac{1}{d}\big(B+\mathds{1}_{d}\big)\in\mathcal{Q}(\mathcal{H}_{d})\right\}\text{.}
\end{equation}
\end{definition}

\noindent We shall now render $\mathcal{L}_{\text{sa,0}}(\mathcal{H}_{d})$ a $(d^{2}-1)$-dimensional normed inner product space over $\mathbb{R}$. Of course, the traceless subspace is as such if one allows it to inherit the ambient structure from $\mathcal{L}_{\text{sa}}(\mathcal{H}_{d})$; however, it will be convenient to instead equip $\mathcal{L}_{\text{sa,0}}(\mathcal{H}_{d})$ with the following scaled inner product and its induced norm
\begin{eqnarray}
\langle\!\langle\cdot|\cdot\rangle\!\rangle_{\boldsymbol{\mathcal{B}}}:\mathcal{L}_{\text{sa,0}}(\mathcal{H}_{d})\times\mathcal{L}_{\text{sa,0}}(\mathcal{H}_{d})\longrightarrow\mathbb{R}::(B_{1},B_{2})\longmapsto\frac{1}{d(d-1)}\langle\!\langle B_{1}|B_{2}\rangle\!\rangle=\frac{1}{d(d-1)}\mathrm{Tr}\big(B_{1}B_{2}\big)\text{,}\\
\|\cdot\|_{\boldsymbol{\mathcal{B}}}:\mathcal{L}_{\text{sa,0}}(\mathcal{H}_{d})\longrightarrow\mathbb{R}_{\geq 0}::B\longmapsto\frac{1}{\sqrt{d(d-1)}}\|B\|=\frac{1}{\sqrt{d(d-1)}}\mathrm{Tr}\big(B^{2}\big)\text{.}
\end{eqnarray} 
For our purposes, it will be important to single out two balls centred on the origin of $\mathcal{L}_{\text{sa},0}(\mathcal{H}_{d})$.
Define the \textit{inball} and the \textit{outball}, respectively, via
\begin{eqnarray}
\boldsymbol{\mathcal{B}}(\mathcal{H}_{d})_{\text{in}}&\equiv&\Big\{B\in\mathcal{L}_{\text{sa},0}(\mathcal{H}_{d})\;\boldsymbol{\Big|}\;\|B\|_{\boldsymbol{\mathcal{B}}}\leq\frac{1}{d-1}\Big\}\text{,}\label{inBall}\\
\boldsymbol{\mathcal{B}}(\mathcal{H}_{d})_{\text{out}}&\equiv&\Big\{B\in\mathcal{L}_{\text{sa},0}(\mathcal{H}_{d})\;\boldsymbol{\Big|}\;\|B\|_{\boldsymbol{\mathcal{B}}}\leq1\Big\}\label{outBall}\text{.}
\end{eqnarray}
The following inclusions, for which we refer the reader to \cite{Appleby2007} for proof, hold for any $d\in\mathbb{N}$
\begin{equation}
\boldsymbol{\mathcal{B}}(\mathcal{H}_{d})_{\text{in}}\subseteq\boldsymbol{\mathcal{B}}(\mathcal{H}_{d})\subseteq\boldsymbol{\mathcal{B}}(\mathcal{H}_{d})_{\text{out}}\text{.}
\label{inclusions}
\end{equation}
Furthermore, the inball is the largest ball centred on the origin contained in the Bloch body, and the outball is the smallest ball centred on the origin containing the Bloch body \cite{Appleby2007}. We denote and define the \textit{insphere} $\boldsymbol{\mathcal{S}}(\mathcal{H}_{d})_{\text{in}}$ and the \textit{outsphere} $\boldsymbol{\mathcal{S}}(\mathcal{H}_{d})_{\text{out}}$ to be the surfaces of the inball and the outball, respectively. The manifold of pure quantum states is \cite{Appleby2007} isomorphic to the intersection of the Bloch body and the outsphere:
\begin{equation}
\boldsymbol{\mathcal{B}}(\mathcal{H}_{d})\cap\boldsymbol{\mathcal{S}}(\mathcal{H}_{d})_{\text{out}}=\left\{B\in\mathcal{L}_{\text{sa,0}}(\mathcal{H}_{d})\;\boldsymbol{\Big|}\;\frac{1}{d}\big(B+\mathds{1}_{d}\big)\text{ is a unit rank projector}\right\}\text{.}
\label{pureBlochs}
\end{equation}
Evidently, if $d=2$,
\begin{equation}
\boldsymbol{\mathcal{B}}(\mathcal{H}_{2})_{\text{in}}=\boldsymbol{\mathcal{B}}(\mathcal{H}_{2})=\boldsymbol{\mathcal{B}}(\mathcal{H}_{2})_{\text{out}}
\end{equation} 
and one recovers the familiar Bloch ball of quantum states for qubits. If $d>2$, then the Bloch body is \textit{much} more intricate than a simple ball \cite{Bengtsson2013}, however conical designs shed new light on its structure. In closing this section, we introduce a final definition that will be used in the sequel: let $\mathbf{I}_{d}:\mathcal{L}(\mathcal{H}_{d})\longrightarrow\mathcal{L}(\mathcal{H}_{d})::A\longmapsto A$ be the identity superoperator and define the \textit{Bloch projector} via
\begin{equation}
\Pi_{\boldsymbol{\mathcal{B}}}\equiv\mathbf{I}_{d}-|\mathds{1}_{d}\rangle\!\rangle\frac{1}{d}\langle\!\langle\mathds{1}_{d}|\text{.}
\label{blochProjector}
\end{equation}
\section{SIMs and MUMs}\label{simmumery}
In this section, we recall the arbitrary rank generalizations of \textsc{sic} \textsc{povm}s and \textsc{mub} \textsc{povm}s introduced by Appleby \cite{Appleby2007} and Kalev and Gour \cite{Kalev2014}, respectively. We will require the following elementary proposition.

\begin{definition}\label{simDef}(Appleby \cite{Appleby2007}) \textit{A} \textsc{sim} \textit{is a} \textsc{povm} $\{E_{1},\dots,E_{d^{2}}\}\subset\mathcal{E}(\mathcal{H}_{d})$ \textit{such that} $\forall\alpha,\beta\in\{1,\dots,d^{2}\}$
\begin{equation}
\mathrm{Tr}\big(E_{\alpha}E_{\beta}\big)=\frac{d^{2}\kappa^{2}\delta_{\alpha,\beta}+d+1-\kappa^{2}}{d^{3}(d+1)}\text{,}
\label{simCon}
\end{equation}
\textit{where} $\kappa\in(0,1]$ \textit{is the} \textsc{sim} contraction parameter\text{.}
\end{definition}

\begin{definition}(Kalev and Gour \cite{Kalev2014})\label{mumDef} \textit{A} \textsc{mum} \textit{is a} \textsc{povm} $\big\{E_{\alpha,b}\;\boldsymbol{|}\;\alpha\in\{1,\dots,d\}\wedge b\in\{1,\dots,d+1\}\big\}\subset\mathcal{E}(\mathcal{H}_{d})$ \textit{such that}
\begin{equation}
\mathrm{Tr}\big(E_{\alpha,b}E_{\alpha',b'}\big)=\begin{cases} \frac{d\eta^{2}\delta_{\alpha,\alpha'}+1-\eta^{2}}{d(d+1)^{2}} & b=b' \\ \frac{1}{d(d+1)^{2}} & b\neq b' \end{cases}\text{,}
\label{mumCon}
\end{equation}
\textit{where} $\eta\in(0,1]$ \textit{is the} \textsc{mum} contraction parameter.
\end{definition}

\noindent We remind the reader of our discussion closing \cref{introPartI}, where we pointed out that \textsc{sim} reads as `symmetric informationally complete measurement' and \textsc{mum} reads as `mutually unbiased measurement.' We will now establish the existence of \textsc{sim}s and \textsc{mum}s for all finite dimensinal quantum state spaces. We use simple geometric arguments. Appleby follows a similiar route to prove the existence of \textsc{sim}s in \cite{Appleby2007}. Kalev and Gour expound a more complicated algebraic argument to prove the existence of \textsc{mum}s in \cite{Kalev2014}. Since the Bloch body will be central to our discussion in the sequel, we present a unified geometric picture of \textsc{sim}s and \textsc{mum}s herein. We begin with the following elementary geometric proposition. 

\begin{proposition}\label{simplexProp}\textit{Let $\mathbb{N}\ni m>1$. Let $\cdot:\mathbb{R}^{m}\times\mathbb{R}^{m}\longrightarrow\mathbb{R}::(v,w)\longmapsto v\cdot w$ be an inner product on $\mathbb{R}^{m}$. Then $\forall\kappa\in\mathbb{R}_{> 0}$ and $\forall n\in\mathbb{N}$ with $2\leq n\leq m+1$ there exists $\{v_{1},\dots,v_{n}\}\subset\mathbb{R}^{m}$ such that $\forall \alpha,\beta\in\{1,\dots,n\}$}
\begin{eqnarray}
v_{\alpha}\cdot v_{\beta}=\kappa^{2}\frac{\delta_{\alpha,\beta}n-1}{n-1}\textit{,}\label{inProd}\\
\sum_{\alpha=1}^{n}v_{\alpha}=0\textit{,}\label{sumZero}\\
\text{dim}_{\mathbb{R}}\big(\text{span}_{\mathbb{R}}\{v_{1},\dots,v_{n}\}\big)=n-1\textit{.}\label{dimSpan}
\end{eqnarray}
\begin{proof} We proceed by induction on $m$. Let $m=2$. If $n=2$, then choose $\{v_{1},v_{2}\}$ to be the endpoints of a line segment of length $2\kappa$ centred on the origin. If $n=3$, then choose $\{v_{1},v_{2},v_{3}\}$ to be the vertices of an equilateral triangle with side length $\sqrt{3}\kappa$ centred on the origin. In both of these cases, the statements in Eq.~\eqref{sumZero} and Eq.~\eqref{dimSpan} follow trivially. Now, suppose the desired result holds for some fixed $m\in\mathbb{N}$. Let $\tilde{m}=m+1$. In light of the obvious inclusion of $\mathbb{R}^{m}$ as a subspace of $\mathbb{R}^{\tilde{m}}$, it remains only to show that the desired result holds for $n=\tilde{m}+1$. It follows from our supposition that there exists $\{v_{1},\dots,v_{\tilde{m}}\}\subset\mathbb{R}^{\tilde{m}}$ such that Eq.~\eqref{inProd}, Eq.~\eqref{sumZero}, and Eq.~\eqref{dimSpan} hold, in particular for some fixed $\kappa\in\mathbb{R}_{>0}$. Therefore $\forall\tilde{\kappa}\in\mathbb{R}_{>0}$ there exists $v\in\mathbb{R}^{\tilde{m}}$ such that $v\cdot v=\tilde{\kappa}^{2}$ and $\forall\alpha\in\{1,\dots,\tilde{m}\}$ $v\cdot v_{\alpha}=0$. Define $\forall\alpha\in\{1,\dots,\tilde{m}\}$ 
\begin{equation}
\tilde{v}_{\alpha}=v_{\alpha}\frac{\tilde{\kappa}}{\kappa}\sqrt{\frac{\tilde{m}^{2}-1}{\tilde{m}^{2}}}-v\frac{1}{\tilde{m}}\implies \forall\alpha,\beta\in\{1,\dots,m+1\}\; \tilde{v}_{\alpha}\cdot\tilde{v}_{\beta}=\tilde{\kappa}^{2}\frac{\delta_{\alpha,\beta}(\tilde{m}+1)-1}{\tilde{m}}\text{.}
\label{tildeVs}
\end{equation}
Then $\{\tilde{v}_{1},\dots,\tilde{v}_{\tilde{m}}\}\cup\{v\}$ satisfies the desired properties: in particular, Eq.~\eqref{inProd} and Eq.~\eqref{sumZero} follow immediately by the construction in Eq.~\eqref{tildeVs}, and Eq.~\eqref{dimSpan} follows from the fact that $v$ was chosen to be in the orthogonal complement of the span of $\{v_{1},\dots,v_{m+1}\}$.
\end{proof}
\end{proposition}

\noindent Our proof of \cref{simplexProp} establishes the existence of regular simplices in $\mathbb{R}^{m}$. Of course, these shapes have been studied extensively \cite{Coxeter1973}. Indeed, in $\mathbb{R}^{3}$ with $n=4$ we find, for instance, examples dating much farther back in the Egyptian pyramids. We give a special name to such simplices when considering the particular $(d^{2}-1)$-dimensional inner product spaces over $\mathbb{R}$ of traceless self-adjoint linear endomorphisms on $\mathcal{H}_{d}$.

\begin{definition}\textit{Let} $\kappa\in\mathbb{R}_{>0}$\textit{. Let} $2\leq d\in\mathbb{N}$\textit{. Let} $n\in\mathbb{N}$\textit{. An} $(n,\kappa)$-regular simplex in $\mathcal{L}_{\text{sa,0}}(\mathcal{H}_{d})$ \textit{is} $\{B_{1},\dots,B_{n}\}\subset\mathcal{L}_{\text{sa,0}}(\mathcal{H}_{d})$ \textit{such that the following three properties hold:} 
\begin{eqnarray}
\forall\alpha,\beta\in\{1,\dots,n\}\;\langle\!\langle 
B_{\alpha}|B_{\beta}\rangle\!\rangle_{\boldsymbol{\mathcal{B}}}=\kappa^{2}\frac{\delta_{\alpha,\beta}n-1}{n-1}\text{,}\\
\sum_{\alpha=1}^{n}B_{\alpha}=0\text{,}\\
\text{dim}_{\mathbb{R}}\big(\text{span}_{\mathbb{R}}\{B_{1},\dots,B_{n}\}\big)=n-1\textit{.}
\end{eqnarray}
\end{definition}

\begin{corollary}\label{blochSimplices}\textit{Let $\kappa\in\mathbb{R}_{>0}$.} \textit{Let} $2\leq d\in\mathbb{N}$\textit{. Then $\forall n\in\mathbb{N}$ with $2\leq n\leq d^{2}$ there exists an $(n,\kappa)$-regular simplex in} $\mathcal{L}_{\text{sa},0}(\mathcal{H}_{d})$\textit{.}
\begin{proof}
Immediate consequence of \cref{simplexProp} and $\text{dim}_{\mathbb{R}}\mathcal{L}_{\text{sa},0}(\mathcal{H}_{d})=d^{2}-1$.
\end{proof}
\end{corollary}

\noindent In light of \cref{blochSimplices} and the inclusion $\boldsymbol{\mathcal{B}}_{\text{in}}\subseteq\boldsymbol{\mathcal{B}}$, the existence of \textsc{sim}s and \textsc{mum}s is trivial. In order to see why, we require some preliminary oberservations. 

\noindent A \textsc{sim} is by \cref{simDef} a \textsc{povm}, and thus from linearity of the trace we see $\forall\alpha\in\{1,\dots,d^{2}\}$ that
\begin{eqnarray}
\mathrm{Tr}E_{\alpha}&=&\mathrm{Tr}\left(E_{\alpha}\sum_{\beta=1}^{d^{2}}E_{\beta}\right)\\
&=&\frac{d^{2}\kappa^{2}+d+1-\kappa^{2}}{d^{3}(d+1)}+\frac{(d^{2}-1)(d+1-\kappa^{2})}{d^{3}(d+1)}\\
&=&\frac{1}{d}\text{.}
\end{eqnarray}
Therefore, for any \textsc{sim} with elements $E_{\alpha}$, there exist corresponding traceless $B_{\alpha}\in\mathcal{L}_{\text{sa},0}(\mathcal{H}_{d})$ such that
\begin{equation}
E_{\alpha}=\frac{(B_{\alpha}+\mathds{1}_{d})}{d^{2}}\textit{.}
\label{simBlochs}
\end{equation} 
The value of the inner products in Eq.~\eqref{simCon} may at first seem obscure; however, in terms of the Bloch body picture we see from  Eq.~\eqref{simBlochs} that Eq.~\eqref{simCon} translates $\forall\alpha,\beta\in\{1,\dots,d^{2}\}$ to
\begin{eqnarray}
\langle\!\langle B_{\alpha}|B_{\beta}\rangle\!\rangle_{\boldsymbol{\mathcal{B}}}&=&\frac{1}{d(d-1)}\mathrm{Tr}\Big(\big(d^{2}E_{\alpha}-\mathds{1}_{d}\big)\big(d^{2}E_{\beta}-\mathds{1}_{d}\big)\Big)\\
&=&\frac{1}{d(d-1)}\left(d^{4}\frac{d^{2}\kappa^{2}\delta_{\alpha,\beta}+d+1-\kappa^{2}}{d^{3}(d+1)}-2d+d\right)\\
&=&\frac{\kappa^{2}\big(d^{2}\delta_{\alpha,\beta}-1\big)}{d^{2}-1}\text{,}
\end{eqnarray}
hence, a \textsc{sim} is exactly a $(d^{2},\kappa)$-regular simplex within the Bloch body. Incidentally, it now follows from \cref{blochSimplices} that any \textsc{sim} is informationally complete. Indeed, the Bloch vectors for a \textsc{sim} necessarily span the null trace subspace, and a \textsc{sim} is a \textsc{povm} so the sum of the \textsc{sim} elements spans the remaining 1-dimensional subspace of the full ambient space of self-adjoint linear endomorphisms on $\mathcal{H}_{d}$. Therefore \textsc{sim}s are symmetric --- in the sense of Eq.~\eqref{simCon} --- informationally complete quantum measurements; hence their name. Furthermore, a \textsc{sim} is a \textsc{sic} \textsc{povm} if and only if $\kappa=1$. As previously mentioned, in light of Eq.~\eqref{inclusions} and \cref{blochSimplices}, it is of course a trivial fact that \textsc{sim}s exist in all cases of finite Hilbert dimension, for one can choose any 
\begin{equation}
\kappa\leq 1/(d-1)\text{.}
\end{equation} 
We have proved the following propositions.

\begin{proposition}\textsc{sim}s \textit{exist for all quantum cones.}
\end{proposition}

\begin{proposition}\textit{Every} \textsc{sim} \textit{is informationally complete.}
\end{proposition}

\noindent For completeness, we now follow a similiar approach for \textsc{mum}s. A \textsc{mum} is by \cref{mumDef} a \textsc{povm}, and thus from linearity of the trace we see $\forall\alpha\in\{1,\dots,d\}$ and $\forall b\in\{1,\dots,d+1\}$ that
\begin{eqnarray}
\mathrm{Tr}E_{\alpha,b}&=&\mathrm{Tr}\left(E_{\alpha,b}\sum_{\alpha'=1}^{d}\sum_{b'=1}^{d+1}E_{\alpha',b}\right)\\
&=&\frac{d\eta^{2}+1-\eta^{2}}{d(d+1)^{2}}+\frac{(d-1)(1-\eta^{2})}{d(d+1)^{2}}+\frac{d^{2}}{d(d+1)^{2}}\\
&=&\frac{1}{d+1}
\end{eqnarray}
Therefore, for any \textsc{mum} with elements $E_{\alpha,b}$, there exist corresponding traceless $B_{\alpha,b}\in\mathcal{L}_{\text{sa},0}(\mathcal{H}_{d})$ such that
\begin{equation}
E_{\alpha}=\frac{(B_{\alpha,b}+\mathds{1}_{d})}{d(d+1)}\textit{.}
\label{mumBlochs}
\end{equation} 

\noindent We see from Eq.~\eqref{mumBlochs} for $\forall\alpha,\alpha'\in\{1,\dots,d\}$ and $\forall b,b'\in\{1,\dots,d+1\}$ that Eq.~\eqref{mumCon} translates to
\begin{eqnarray}
\langle\!\langle B_{\alpha,b}|B_{\alpha',b'}\rangle\!\rangle_{\boldsymbol{\mathcal{B}}}&=&\frac{1}{d(d-1)}\mathrm{Tr}\Big(\big(d(d+1)E_{\alpha,b}-\mathds{1}_{d}\big)\big(d(d+1)E_{\alpha',b'}-\mathds{1}_{d}\big)\Big)\\
&=&\begin{cases} \frac{1}{d(d-1)}\left(\frac{d^{2}(d+1)^{2}(d\eta^{2}\delta_{\alpha,\alpha'}+1-\eta^{2})}{d(d+1)^{2}} -\frac{2d(d+1)}{d+1}+d\right) & b=b'\\ \frac{1}{d(d-1)}\left(\frac{d^{2}(d+1)^{2}}{d(d+1)^{2}}-\frac{2d(d+1)}{d+1}+d\right) & b\neq b'\end{cases}\\
&=&\begin{cases} \frac{\eta^{2}(d\delta_{\alpha,\alpha'}-1)}{d-1}& b=b'\\ 0 & b\neq b'\end{cases}\text{,}
\end{eqnarray}
hence a \textsc{mum} is the union of $(d+1)$ mutually orthogonal $(d,\eta)$-regular simplices within the Bloch body. Incidentally, it now follows from \cref{blochSimplices} that any \textsc{mum} is informationally complete; moreover, $\forall b\in\{1,\dots,d+1\}$ that each set $\big\{E_{\alpha,b}(d+1)\;\boldsymbol{|}\;\alpha\in\{1\dots,d\}\big\}$ is a \textsc{povm}. A proof of the latter statement follows directly from that fact that the simplex vertices sum to zero. A proof of the former statement runs as follows: each simplex spans a $(d-1)$-dimensional subspace of the null trace subspace, there are $(d+1)$ mutually orthogonal such simplices, and the remaining 1-dimensional subspace of the full ambient space of self-adjoint linear endomorphisms on $\mathcal{H}_{d}$ is spanned by the sum of the \textsc{mum} elements. Therefore \textsc{mum}s are mutually unbiased --- in the sense of Eq.~\eqref{mumCon} --- informationally complete quantum measurements; hence their name. Furthermore, a \textsc{mum} is a \textsc{mub} \textsc{povm} if and only if $\eta=1$. Once again, it is trivial  that \textsc{mum}s exist in all cases of finite Hilbert dimension, for one can choose any 
\begin{equation}
\eta\leq 1/(d-1)\text{.}
\end{equation} 
We have proved the following propositions.

\begin{proposition}\textsc{mum}s \textit{exist for all quantum cones.}
\end{proposition}

\begin{proposition}\textit{Every} \textsc{mum} \textit{is informationally complete.}
\end{proposition}

\noindent  In \cite{Appleby2007}, Appleby constructs \textsc{sim}s with a nontrivial contraction parameter 
\begin{equation}
\kappa_{\text{Appleby}}=1/\sqrt{d+1}
\end{equation} 
for all odd dimensions; moreover, Appleby was the first to point out that \textsc{sim}s exist $\forall d\in\mathbb{N}$ in light of the foregoing Bloch-geometric picture. In \cite{Gour2014}, Gour and Kalev prove that \textsc{sim}s exist using a more complicated algebraic argument; moreover, Gour and Kalev construct \textsc{sim}s with nontrivial contraction parameters that, when $d$ is even, improve Appleby's constructions. 

\noindent In \cite{Kalev2014}, were \textsc{mum}s were first introduced, Kalev and Gour construct nontrivial \textsc{mum}s with\footnote{Kalev and Gour adopt a different convention for the \textsc{mum} contraction parameter; we have translated their result.} 
\begin{equation}
\eta_{\text{Kalev-Gour}}=\sqrt{2/d(d-1)}\text{.}
\end{equation}

\noindent In the following section, we consider a distinguished class of geometric substructures of quantum cones; a class wherein \textsc{sim}s and \textsc{mum}s arise as special cases. We thus provide an affirmative answer to the question posed by Dall'Arno~\cite{DallArno14}, whether \textsc{sim}s and \textsc{mum}s are particular instances of a more general class of objects.     

\newpage
\section{Conical 2-Designs}\label{desQC}
A projective $2$-design is a finite set of unit rank projectors $\{\pi_{1},\dots,\pi_{n}\}\subset\mathcal{Q}(\mathcal{H}_{d})$ such that the sum over the second tensor powers of the design elements commutes with the action of the product representation of the complex unitary group of degree $d$. It is natural to ask what can be said for a finite set of arbitrary elements of the quantum cone with this property. \cref{desCons} answers that question. It will be convenient to introduce some notation. Let $\{e_{1},\dots,e_{d}\}\subset\mathcal{H}_{d}$ be a fixed orthonormal basis, relative to which we denote and define \textit{transposition} and \textit{complex conjugation} on $\mathcal{L}(\mathcal{H}_{d})$, together with a chosen \textit{maximally entangled ket}, respectively, via
\begin{eqnarray}
\mathbf{T}:\mathcal{L}(\mathcal{H}_{d})\longrightarrow\mathcal{L}(\mathcal{H}_{d})::A\longmapsto A^{\mathrm{T}}\equiv\sum_{r=1}^{d}\sum_{s=1}^{d}|e_{r}\rangle\langle e_{s}|A e_{r}\rangle\langle e_{s}|\text{,}\\
\mathbf{C}:\mathcal{L}(\mathcal{H}_{d})\longrightarrow\mathcal{L}(\mathcal{H}_{d})::A\longmapsto
\overline{A}\equiv\sum_{r=1}^{d}\sum_{s=1}^{d}|e_{r}\rangle\overline{\langle e_{r}|A e_{s}\rangle}\langle e_{s}|\text{,}\\
\mathcal{H}_{d}\otimes\mathcal{H}_{d}\ni|\Phi_{+}\rangle=\frac{1}{\sqrt{d}}\sum_{r=1}^{d}|e_{r}\rangle\otimes|e_{r}\rangle\text{.}\label{maxKet}
\end{eqnarray}
\begin{theorem}\label{desCons}\textit{Let} $n\in\mathbb{N}$\textit{. Let} $\{A_{1},\dots,A_{n}\}\subset\mathcal{L}_{\text{sa}}(\mathcal{H}_{d})_{+}$\textit{. Then the following statements are equivalent.}
\begin{enumerate}[(i)]
\item\label{con1}$\forall U\in\mathrm{U}(\mathcal{H}_{d})$
\begin{equation}
\left[U\otimes U,\sum_{j=1}^{n}A_{j}\otimes A_{j}\right]=0\textit{.}
\label{con1Eq}
\end{equation}
\item\label{con2}$\exists k_{s}\geq k_{a}\geq 0$ \textit{such that}
\begin{equation}
\sum_{j=1}^{n}A_{j}\otimes A_{j}=\Pi_{\text{sym}}k_{s}+\Pi_{\text{asym}}k_{a}\textit{.}
\label{con2Eq}
\end{equation}
\item\label{con3}$\exists k_{+}\geq k_{-}\geq 0$ \textit{such that}
\begin{equation}
\sum_{j=1}^{n}A_{j}\otimes \overline{A_{j}}=\mathds{1}_{d}\otimes\mathds{1}_{d}k_{+}+|\Phi_{+}\rangle dk_{-}\langle\Phi_{+}|\textit{.}
\label{con3Eq}
\end{equation}
\item\label{con4}$\exists k_{+}\geq k_{-}\geq 0$ \textit{such that}
\begin{equation}
\sum_{j=1}^{n}|A_{j}\rangle\!\rangle\langle\!\langle\overline{A_{j}}|=|\mathds{1}_{d}\rangle\!\rangle k_{+}\langle\!\langle\mathds{1}_{d}|+\mathbf{T}k_{-}\textit{.}
\label{con4Eq}
\end{equation}
\item\label{con5}$\exists k_{+}\geq k_{-}\geq 0$ \textit{such that}
\begin{equation}
\sum_{j=1}^{n}|A_{j}\rangle\!\rangle\langle\!\langle A_{j}|=|\mathds{1}_{d}\rangle\!\rangle k_{+}\langle\!\langle\mathds{1}_{d}|+\mathbf{I}_{d}k_{-}\textit{.}
\label{con5Eq}
\end{equation}
\end{enumerate}
\textit{If these equivalent conditions are satisfied, then the quantities} $k_{+}$ \textit{and} $k_{-}$ \textit{appearing in conditions} (\ref{con3})-(\ref{con5}) \textit{take the same values, and are related to the quantities} $k_{s}$ \textit{and} $k_{a}$ \textit{in condition} (\ref{con2}) \textit{by} 
\begin{equation}
k_{\pm}=(k_{s}\pm k_{a})/2\textit{.}
\end{equation}
\textit{Furthermore,} $\text{span}_{\mathbb{R}}\{A_{1},\dots,A_{n}\}=\mathcal{L}_{\text{sa}}(\mathcal{H}_{d})$ \textit{if and only if} $k_{s}>k_{a}$, \textit{equivalently,} $k_{-}> 0$\textit{.}
\newpage
\begin{proof}
We will first establish that (\ref{con2})$\iff$(\ref{con3}). Observe that Eq.~\eqref{con2Eq} can be written as
\begin{equation}
\sum_{j=1}^{n}A_{j}\otimes A_{j}=\big(\mathds{1}_{d}\otimes\mathds{1}_{d}\big)\frac{k_{s}+k_{a}}{2}+\mathrm{W}\frac{k_{s}-k_{a}}{2}\text{,}
\label{for23}
\end{equation}
where $\mathrm{W}$ is the swap operator defined in Eq.~\eqref{swapOp}. Taking the partial transpose of Eq.~\eqref{for23} with respect to $\{e_{1},\dots,e_{d}\}$ over the second tensor factor in $\mathcal{L}(\mathcal{H}_{d})\otimes\mathcal{L}(\mathcal{H}_{d})$  we obtain
\begin{eqnarray}
\sum_{j=1}^{n}A_{j}\otimes \overline{A_{j}}&=&\big(\mathds{1}_{d}\otimes\mathds{1}_{d}^{\mathrm{T}}\big)\frac{k_{s}+k_{a}}{2}+\left(\sum_{r=1}^{d}\sum_{s=1}^{d}|e_{s}\rangle\langle e_{r}|\otimes (|e_{r}\rangle\langle e_{s}|)^{\mathrm{T}}\right)\frac{k_{s}-k_{a}}{2}\\
&=&\big(\mathds{1}_{d}\otimes\mathds{1}_{d}\big)\frac{k_{s}+k_{a}}{2}+\left(\sum_{s=1}^{d}|e_{s}\rangle\otimes |e_{s}\rangle\right)\frac{k_{s}-k_{a}}{2}\left(\sum_{r=1}^{d}\langle e_{r}|\otimes\langle e_{r}|\right)\\[0.3cm]
&=&\big(\mathds{1}_{d}\otimes\mathds{1}_{d}\big)\frac{k_{s}+k_{a}}{2}+|\Phi_{+}\rangle\!\rangle \frac{d(k_{s}-k_{a})}{2}\langle\!\langle \Phi_{+}|\text{.}
\end{eqnarray}
where we have used the fact that $A^{\mathrm{T}}=\overline{A}$ for every $A\in\mathcal{L}_{\text{sa}}(\mathcal{H}_{d})\supset\mathcal{L}_{\text{sa}}(\mathcal{H}_{d})_{+}$. With $k_{\pm}\equiv(k_{s}\pm k_{a})/2$ we see from Eq.~\eqref{maxKet} (\ref{con2})$\iff$(\ref{con3}). Next, recall the Choi-Jamio{\l}kowski \cref{cjThm}, and observe that the linear bijection defined therein can be written in terms of $\mathcal{L}(\mathcal{H}_{d})\ni |E_{r,s}\rangle\!\rangle\equiv |e_{r}\rangle\langle e_{s}|$ as
\begin{equation}
\boldsymbol{\mathcal{J}}:\mathsf{Lin}(d,d)\longrightarrow \mathcal{L}(\mathcal{H}_{d}\otimes\mathcal{H}_{d})::\Lambda\longmapsto\frac{1}{d}\sum_{r=1}^{d}\sum_{s=1}^{d}\Lambda\big(|E_{r,s}\rangle\!\rangle\big)\otimes |E_{r,s}\rangle\!\rangle\text{.}
\end{equation}
The inverse of $\boldsymbol{\mathcal{J}}$ is then defined via the linear extension of its action on pure tensors in $\mathcal{L}(\mathcal{H}_{d}\otimes\mathcal{H}_{d})$
\begin{equation}
\boldsymbol{\mathcal{J}}^{-1}:\mathcal{L}(\mathcal{H}_{d}\otimes\mathcal{H}_{d})\longrightarrow \mathsf{Lin}(d,d)::|A\rangle\!\rangle\otimes |B\rangle\!\rangle\longmapsto |A\rangle\!\rangle d\langle\!\langle B^{T}|\text{.}
\label{invJ}
\end{equation}
Incidentally, this is an example of where double Dirac notation is especially useful; indeed one has that
\begin{eqnarray}
\boldsymbol{\mathcal{J}}\Big(\boldsymbol{\mathcal{J}}^{-1}\big(|A\rangle\!\rangle\otimes |B\rangle\!\rangle\big)\Big)&=&\boldsymbol{\mathcal{J}}\Big(|A\rangle\!\rangle d\langle\!\langle B^{\mathrm{T}}|\Big)\\
&=&\sum_{r=1}^{d}\sum_{s=1}^{d}|A\rangle\!\rangle\langle\!\langle B^{\mathrm{T}}|E_{r,s}\rangle\!\rangle\otimes|E_{r,s}\rangle\!\rangle\\
&=&\sum_{r=1}^{d}\sum_{s=1}^{d}|A\rangle\!\rangle\otimes|E_{r,s}\rangle\!\rangle\langle\!\langle B^{\mathrm{T}}|E_{r,s}\rangle\!\rangle\\
&=&|A\rangle\!\rangle\otimes \left(\sum_{r=1}^{d}\sum_{s=1}^{d}|e_{r}\rangle\mathrm{Tr}\Big(B^{T}|e_{r}\rangle\langle e_{s}|\Big)\langle e_{s}|\right)\\
&=&|A\rangle\!\rangle\otimes \left(\sum_{r=1}^{d}\sum_{s=1}^{d}|e_{r}\rangle\langle e_{r}|B e_{s}\rangle\langle e_{s}|\right)\\[0.3cm]
&=&|A\rangle\!\rangle\otimes |B\rangle\!\rangle\text{.}
\end{eqnarray}
From Eq.~\eqref{swapOp} and Eq.~\eqref{invJ} we have
\begin{eqnarray}
\boldsymbol{\mathcal{J}}^{-1}\big(\mathrm{W}\big)&=&\sum_{r=1}^{d}\sum_{s=1}^{d}|E_{s,r}\rangle\!\rangle d\langle\!\langle E_{s,r}|\\
\implies \forall X\in\mathcal{L}(\mathcal{H}_{d})\hspace{0.5cm}\Big(\boldsymbol{\mathcal{J}}^{-1}\big(\mathrm{W}\big)\Big)(X)&=&\sum_{r=1}^{d}\sum_{s=1}^{d}|E_{s,r}\rangle\!\rangle d\langle\!\langle E_{s,r}|X\rangle\!\rangle\nonumber\\
&=&\sum_{r=1}^{d}\sum_{s=1}^{d}|E_{s,r}\rangle\!\rangle d\mathrm{Tr}\Big(|e_{s}\rangle\langle e_{r}|X\Big)\nonumber\\
&=&\sum_{r=1}^{d}\sum_{s=1}^{d}|e_{s}\rangle\langle e_{r}|Xe_{s}\rangle d\langle e_{r}|\nonumber\\[0.2cm]
&=&X^{\mathrm{T}}d\text{.}
\label{wJ}
\end{eqnarray}
From Eq.~\eqref{maxKet} and Eq.~\eqref{invJ} we have
\begin{eqnarray}\boldsymbol{\mathcal{J}}^{-1}\big(|\Phi_{+}\rangle\langle\Phi_{+}|\big)&=&\sum_{r=1}^{d}\sum_{s=1}^{d}|E_{s,r}\rangle\!\rangle d\langle\!\langle E_{r,s}|\\
\implies \forall X\in\mathcal{L}(\mathcal{H}_{d})\hspace{0.5cm}\Big(\boldsymbol{\mathcal{J}}^{-1}\big(|\Phi_{+}\rangle\langle\Phi_{+}|\big)\Big)(X)&=&\sum_{r=1}^{d}\sum_{s=1}^{d}|E_{s,r}\rangle\!\rangle \langle\!\langle E_{r,s}|X\rangle\!\rangle\nonumber\\
&=&\sum_{r=1}^{d}\sum_{s=1}^{d}|E_{s,r}\rangle\!\rangle \mathrm{Tr}\Big(|e_{r}\rangle\langle e_{s}|X\Big)\nonumber\\
&=&\sum_{r=1}^{d}\sum_{s=1}^{d}|e_{s}\rangle\langle e_{s}|Xe_{r}\rangle \langle e_{r}|\nonumber\\[0.2cm]
&=&X\text{.}
\label{phiJ}
\end{eqnarray}
Therefore, Eq.~\eqref{invJ}, Eq.~\eqref{wJ}, and Eq.~\eqref{phiJ} respectively yield
\begin{eqnarray}
\boldsymbol{\mathcal{J}}^{-1}\big(\mathds{1}_{d}\otimes\mathds{1}_{d}\big)&=&|\mathds{1}_{d}\rangle\!\rangle d\langle\!\langle\mathds{1}_{d}|\text{,}\\
\boldsymbol{\mathcal{J}}^{-1}\big(\mathrm{W}\big)&=&\mathbf{T}d\text{,}\\
\boldsymbol{\mathcal{J}}^{-1}\big(|\Phi_{+}\rangle\langle\Phi_{+}|\big)&=&\mathbf{I}_{d}\text{.}
\end{eqnarray}
Consequently, applying $\boldsymbol{\mathcal{J}}^{-1}$ to both sides of Eq.~\eqref{con3Eq} yields Eq.~\eqref{con5Eq}, and applying $\boldsymbol{\mathcal{J}}^{-1}$ to both sides of Eq.~\eqref{for23} yields Eq.~\eqref{con4Eq}. The implication (\ref{con2})$\implies$(\ref{con1}) is immediate from \cref{irrepLem}, and the implication (\ref{con2})$\implies$(\ref{con1}) follows from Schur's \cref{schurLem}; hence, (\ref{con1})$\iff$(\ref{con2})$\iff$(\ref{con3})$\iff$(\ref{con4})$\iff$(\ref{con5}). 

\noindent To see that $k_{s}\geq k_{a}\geq 0$, let $|\Psi\rangle$ be an arbitrary normalized element of the antisymmetric subspace of $\mathcal{H}_{d}\otimes\mathcal{H}_{d}$ and observe that $\{A_{1},\dots,A_{n}\}\in\mathcal{L}_{\text{sa}}(\mathcal{H}_{d})_{+}$ implies that
\begin{eqnarray}
k_{a}=\sum_{j=1}^{n}\langle\Psi|A_{j}\otimes A_{j}\Psi\rangle\geq 0\text{;}
\end{eqnarray}
moreover, partially transposing and applying $\boldsymbol{\mathcal{J}}^{-1}$ to Eq.~\eqref{con2Eq} yields Eq.~\eqref{con5Eq} with $k_{\pm}=(k_{s}\pm k_{a})/2$. 
\newpage
\noindent Let $B\in\mathcal{L}_{\text{sa},0}(\mathcal{H}_{d})$, so $\langle\!\langle\mathds{1}_{d}|B\rangle\!\rangle=\mathrm{Tr}B=0$, be normalized, \textit{i.e}.\ $\|B\|=\sqrt{\langle\!\langle B|B\rangle\!\rangle}=1$. Then from Eq.~\eqref{con5Eq}
\begin{eqnarray}
\sum_{j=1}^{n}|A_{j}\rangle\!\rangle\langle\!\langle A_{j}|B\rangle\!\rangle\langle\!\langle B|&=&|\mathds{1}_{d}\rangle\!\rangle k_{+}\langle\!\langle\mathds{1}_{d}|B\rangle\!\rangle\langle\!\langle B|+\mathbf{I}_{d}\big(|B\rangle\!\rangle \langle\!\langle B|\big)\label{forFirst}\\
\implies\hspace{0.4cm}0\leq \sum_{j=1}^{n}|\langle\!\langle A_{j}|B\rangle\!\rangle|^{2}&=&k_{-}\label{forNull}
\end{eqnarray}
where the implication follows from tracing Eq.~\eqref{forFirst}; hence, $k_{s}\geq k_{a}$. 

\noindent Lastly, we recall the fact that $\mathrm{span}_{\mathbb{R}}\{A_{1},\dots,A_{n}\}=\mathcal{L}_{\text{sa}}(\mathcal{H}_{d})$ if and only if $\forall A\in\mathcal{L}_{\text{sa}}(\mathcal{H}_{d})$ such that $A\neq 0$ one has
\begin{equation}
\sum_{j=1}^{n}\langle\!\langle A|A_{j}\rangle\!\rangle\langle\!\langle A_{j}|A\rangle\!\rangle>0\text{.}
\label{posDef}
\end{equation}
That Eq.~\eqref{posDef} holds under the assumption that $\mathrm{span}_{\mathbb{R}}\{A_{1},\dots,A_{n}\}=\mathcal{L}_{\text{sa}}(\mathcal{H}_{d})$ is immediate. Conversely, Eq.~\eqref{posDef} implies that the orthogonal complement to  $\mathrm{span}_{\mathbb{R}}\{A_{1},\dots,A_{n}\}$ is null. If $k_{-}=0$, then any $B\in\mathcal{L}_{\text{sa},0}(\mathcal{H}_{d})$ is in the null space by Eq.~\eqref{forNull}. Therefore, if $k_{s}=k_{a}$, then $\{A_{1},\dots,A_{n}\}$ is not a spanning set. If $k_{s}>k_{a}$, then Eq.~\eqref{posDef} holds in light of Eq.~\eqref{con5Eq} from the observation that $A\neq 0\implies\sum_{j=1}^{n}|\langle\!\langle A_{j}|A\rangle\!\rangle|^{2}\geq k_{-}\|A\|^{2}>0$, which follows from taking $\langle\!\langle A|\cdot|A\rangle\!\rangle$ on Eq.~\eqref{con5Eq}.
\end{proof}
\end{theorem}
\noindent \cref{desCons} prompts the following definition.
\begin{definition}\label{c2dDef}\textit{A} conical 2-design \textit{is a finite set} $\{A_{1},\dots,A_{n}\}\subset\mathcal{L}_{\text{sa}}(\mathcal{H}_{d})_{+}$ \textit{satisfying the five equivalent conditions in \cref{desCons} with }$k_{s}>k_{a}$ \textit{and such that} $\forall j\in\{1,\dots,n\}$ $A_{j}\neq 0$.
\end{definition}
\noindent The requirement that $A_{j}$ be nonzero is not essential, and is made for convenience only. We demand $k_{s}>k_{a}$ so that $\mathrm{span}_{\mathbb{R}}\{A_{1},\dots,A_{n}\}=\mathcal{L}_{\text{sa}}(\mathcal{H}_{d})$; hence for any conical 2-design $n\geq d^{2}$. Any $L\in\mathcal{L}_{\text{sa}}(\mathcal{H}_{d})$ is therefore a linear combination of the $A_{j}$, explicitly
\begin{equation}
L=\frac{1}{k_{-}}\sum_{j=1}^{n}\left(\mathrm{Tr}\big(A_{j}L\big)-\frac{k_{+}\mathrm{Tr}\big(A_{j}\big)\mathrm{Tr}\big(L\big)}{dk_{+}+k_{-}}\right)A_{j}\text{.}
\label{cdExp}
\end{equation}
For a proof of Eq.~\eqref{cdExp}, we simply note that the action of Eq.~\eqref{con5Eq} on $\mathds{1}_{d}$ and $L$ respectively implies
\begin{eqnarray}
\sum_{j=1}^{n}|A_{j}\rangle\!\rangle\mathrm{Tr}\big(A_{j}\big)&=&|\mathds{1}_{d}\rangle\!\rangle dk_{+}+|\mathds{1}_{d}\rangle\!\rangle k_{-}\text{,}\\
\sum_{j=1}^{n}|A_{j}\rangle\!\rangle\mathrm{Tr}\big(A_{j}L\big)&=&|\mathds{1}_{d}\rangle\!\rangle\mathrm{Tr}\big(L\big)k_{+}+|L\rangle\!\rangle k_{-}\text{.}
\end{eqnarray}
Of course, the expansion in Eq.~\eqref{cdExp} is unique if $n=d^{2}$; otherwise not.

\noindent The requirement that a conical 2-design be a \textit{finite} subset of a quantum cone is somewhat restrictive; however, in this thesis, we shall only consider such designs. Recently, Brandsen-Dall'Arno-Szymusiak \cite{Brandsen2016} have explored a definition of conical designs that lifts the restriction of finite cardinality imposed by the author and Appleby in \cite{Graydon2015a}. Incidentally, the Brandsen-Dall'Arno-Szymusiak definition restricts, in the finite case, to a subclass of the \textit{homogeneous} conical 2-designs introduced in \cite{Graydon2015a} and discussed at length in \cref{hc2d}.
\newpage
\noindent Obviously, any finite complex projective 2-design (not necessarily uniformly weighted) is a conical 2-design. In fact, we have the following proposition. Note that we denote the rank of $A\in\mathcal{L}(\mathcal{H}_{d})$ via $\text{rank}A$.

\begin{proposition}\label{rank1Lem}\textit{Let} $d\in\mathbb{N}$. \textit{Let} $\{A_{1},\dots,A_{n}\}\subset\mathcal{L}_{\text{sa}}(\mathcal{H}_{d})_{+}$ \textit{be a conical 2-design. Then} $\forall j\in\{1,\dots,n\}$ $\text{rank}A_{j}=1$ \textit{if and only if} $k_{a}=0$\textit{.}
\begin{proof}
From the inclusion $\mathcal{L}_{\text{sa}}(\mathcal{H}_{d})_{+}\subset\mathcal{L}_{\text{sa}}(\mathcal{H}_{d})$ one has for any $A_{j}\in\mathcal{L}_{\text{sa}}(\mathcal{H}_{d})_{+}$ that $\mathrm{Tr}\big(A_{j}^{2}\big)=\big(\mathrm{Tr}A_{j}\big)^{2}$ if and only if $\mathrm{rank}A_{j}=1$, in which case both quantities are the square of the sole nonzero eigenvalue of $A_{j}$. Tracing Eq.~\eqref{con2Eq} and Eq.~\eqref{con5Eq} we have respectively that
\begin{eqnarray}
\sum_{j=1}^{n}\big(\mathrm{Tr}A_{j}\big)^{2}&=&k_{s}d(d+1)/2+k_{a}d(d-1)/2\text{,}\\
\sum_{j=1}^{n}\mathrm{Tr}\big(A_{j}^{2}\big)&=&(k_{s}+k_{a})d/2+(k_{s}-k_{a})d^{2}/2\text{,}\\
\implies\sum_{j=1}^{n}\Big(\big(\mathrm{Tr}A_{j}\big)^{2}-\mathrm{Tr}\big(A_{j}^{2}\big)\Big)\hspace{0.5cm}&=&d(d-1)k_{a}\text{,}
\end{eqnarray}
So $k_{a}=0$ if and only if $\forall j\in\{1,\dots,n\}$ $\text{rank}A_{j}=1$.
\end{proof}
\end{proposition}

\noindent Recalling \cref{cpdDef}, we see from \cref{rank1Lem} that a conical 2-design is a finite complex projective 2-design (up to an appropriate scaling) if and only if $k_{a}=0$, equivalently if and only if $k_{+}=k_{-}$. As we will now show, \textsc{sim}s and \textsc{mum}s are also examples of conical 2-designs.

\begin{proposition}\label{simProp}\textit{Let} $\mathfrak{e}\equiv\{E_{1},\dots,E_{d^{2}}\}\subset\mathcal{E}(\mathcal{H}_{d})\subset\mathcal{L}_{\text{sa}}(\mathcal{H}_{d})_{+}$ \textit{be a} \textsc{sim} \textit{as in \cref{simDef}. Then} $\mathfrak{e}$ \textit{is a conical 2-design.}
\begin{proof}
Define $\mathrm{M}=\sum_{\alpha=1}^{d^{2}}|E_{\alpha}\rangle\!\rangle\langle\!\langle E_{\alpha}|$. Then $\forall \beta\in\{1,\dots,d^{2}\}$ it follows from Eq.~\eqref{simCon} that
\begin{eqnarray}
\mathrm{M}|E_{\beta}\rangle\!\rangle&=&\sum_{\alpha=1}^{d^{2}}|E_{\alpha}\rangle\!\rangle\langle\!\langle E_{\alpha}|E_{\beta}\rangle\!\rangle\nonumber\\
&=&\sum_{\alpha=1}^{d^{2}}|E_{\alpha}\rangle\!\rangle\nonumber\left(\frac{d+1-\kappa^{2}+\delta_{\alpha,\beta}\kappa^{2}d^{2}}{d^{3}(d+1)}\right)\\
&=&|\mathds{1}_{d}\rangle\!\rangle\frac{d+1-\kappa^{2}}{d^{2}(d+1)}+|E_{\beta}\rangle\!\rangle\frac{\kappa^{2}}{d(d+1)}\text{,}
\label{Mop}
\end{eqnarray}
where the second equality follows from the fact that $\sum_{\alpha=1}^{n}E_{\alpha}=\mathds{1}_{d}$. Recall that a \textsc{sim} is informationally complete, \textit{i.e}.\ $\mathrm{span}_{\mathbb{R}}\{E_{1},\dots,E_{d^{2}}\}=\mathcal{L}_{\text{sa}}(\mathcal{H}_{d})$. The action of our $\mathrm{M}\in\mathcal{L}\big(\mathcal{L}(\mathcal{H}_{d})\big)$ is thus completely determined by its action on $\{E_{1},\dots,E_{d^{2}}\}$. Therefore from Eq.~\eqref{Mop} we deduce
\begin{equation}
\sum_{\alpha=1}^{d^{2}}|E_{\alpha}\rangle\!\rangle\langle\!\langle E_{\alpha}|=|\mathds{1}_{d}\rangle\!\rangle\frac{d+1-\kappa^{2}}{d(d+1)}\langle\!\langle\mathds{1}_{d}|+\mathbf{I}_{d}\frac{\kappa^{2}}{d(d+1)}\text{.}
\end{equation}
Recalling $\kappa\in(0,1]$ we complete the proof in light of \cref{simDef} and \cref{desCons}, in particular from Eq.~\eqref{con5Eq}.
\end{proof}
\end{proposition}

\begin{proposition}\label{mumProp}\textit{Let} $\mathfrak{f}\equiv\big\{E_{\alpha,b}\;\boldsymbol{|}\;\alpha\in\{1,\dots,d\}\wedge b\in\{1,\dots,d+1\}\big\}\subset\mathcal{E}(\mathcal{H}_{d})\subset\mathcal{L}_{\text{sa}}(\mathcal{H}_{d})_{+}$ \textit{be a} \textsc{mum} \textit{as in \cref{mumDef}. Then} $\mathfrak{f}$ \textit{is a conical 2-design.}
\begin{proof}
Define $\mathrm{N}=\sum_{\alpha=1}^{d}\sum_{b=1}^{d+1}|E_{\alpha,b}\rangle\!\rangle\langle\!\langle E_{\alpha,b}|$. Then $\forall \alpha'\in\{1,\dots,d\}$ and $\forall b'\in\{1,\dots,d+1\}$ it follows from Eq.~\eqref{mumCon} that
\begin{eqnarray}
\mathrm{N}|E_{\alpha',b'}\rangle\!\rangle&=&\sum_{\alpha=1}^{d}\sum_{b=1}^{d+1}|E_{\alpha,b}\rangle\!\rangle\langle\!\langle E_{\alpha,b}|E_{\alpha',b'}\rangle\!\rangle\nonumber\\
&=&\sum_{\alpha=1}^{d}|E_{\alpha,b'}\rangle\!\rangle\left(\frac{d\eta^{2}\delta_{\alpha,\alpha'}+1-\eta^{2}}{d(d+1)^{2}}\right)+\sum_{\alpha=1}^{d}\sum_{b\neq b'}|E_{\alpha,b}\rangle\!\rangle\left(\frac{1}{d(d+1)^{2}}\right)\nonumber\\
&=&|\mathds{1}_{d}\rangle\!\rangle\frac{d+1-\eta^{2}}{d(d+1)^{3}}+|E_{\alpha',b'}\rangle\!\rangle\frac{\eta^{2}}{(d+1)^{2}}
\label{Nop}
\end{eqnarray}
where the second equality follows from the fact that $\sum_{\alpha=1}^{n}E_{\alpha,b}(d+1)=\mathds{1}_{d}$. Recall that a \textsc{mum} is informationally complete, \textit{i.e}.\ $\mathrm{span}_{\mathbb{R}}\mathfrak{f}=\mathcal{L}_{\text{sa}}(\mathcal{H}_{d})$. The action of our $\mathrm{N}\in\mathcal{L}\big(\mathcal{L}(\mathcal{H}_{d})\big)$ is thus completely determined by its action on $\{E_{\alpha,b}\}$. Therefore from Eq.~\eqref{Nop} we deduce
\begin{equation}
\sum_{\alpha=1}^{d}\sum_{b=1}^{d+1}|E_{\alpha}\rangle\!\rangle\langle\!\langle E_{\alpha}|=|\mathds{1}_{d}\rangle\!\rangle\frac{d+1-\eta^{2}}{d(d+1)^{2}}\langle\!\langle\mathds{1}_{d}|+\mathbf{I}_{d}\frac{\eta^{2}}{(d+1)^{2}}\text{.}
\end{equation}
Recalling $\eta\in(0,1]$ we complete the proof in light \cref{mumDef} of \cref{desCons}, in particular from Eq.~\eqref{con5Eq}.
\end{proof}
\end{proposition}

\noindent Projective 2-design \textsc{povm}s are obviously conical 2-designs in light of \cref{rank1Lem}. Our proofs of \cref{simProp} and \cref{mumProp} establish that \textsc{sim}s and \textsc{mum}s are also examples of \textsc{povm}s that are conical 2-designs. We call a \textsc{povm} that is also a conical 2-design a \textit{conical 2-design} \textsc{povm}. In fact, \textsc{sim}s are the unique conical 2-design \textsc{povm}s of the minimum possible cardinality; they are \textit{minimal}.

\begin{theorem}\label{minConPovm}\textit{Let} $\mathfrak{e}\equiv\{E_{1},\dots,E_{d^{2}}\}\subset\mathcal{E}(\mathcal{H}_{d})\subset\mathcal{L}_{\text{sa}}(\mathcal{H}_{d})_{+}$ \textit{be a} \textsc{povm}\textit{ of cardinality $d^{2}$. Then} $\mathfrak{e}$ \textit{is a conical 2-design if and only if} $\mathfrak{e}$ \textit{is a} \textsc{sim}\textit{.}
\begin{proof}
Our proof of \cref{simProp} establishes sufficiency. We now prove necessity. Let $\mathfrak{e}$ be as in the statement of the theorem, that is $\mathfrak{e}\equiv\{E_{1},\dots, E_{d^{2}}\}\subset\mathcal{E}(\mathcal{H}_{d})\subset\mathcal{L}_{\text{sa}}$ is a conical 2-design \textsc{povm}. Then the partial trace of Eq.~\eqref{con5Eq} over either factor yields
\begin{eqnarray}
\sum_{\alpha=1}^{d}E_{\alpha}\mathrm{Tr}(E_{\alpha})=\mathds{1}_{d}(dk_{+}+k_{-})=\sum_{\alpha=1}^{d}E_{\alpha}(dk_{+}+k_{-})\text{,}
\label{constantTrace}
\end{eqnarray}
where the second equality follows from the fact that $\mathfrak{e}$ is a \textsc{povm}. Since $\text{card}\mathfrak{e}=\text{dim}_{\mathbb{R}}\mathcal{L}_{\text{sa}}(\mathcal{H}_{d})$, conical 2-design $\mathfrak{e}$ is in fact a basis for $\mathcal{L}_{\text{sa}}(\mathcal{H}_{d})$. Eq.~\eqref{constantTrace} thus implies $\forall\alpha\in\{1,\dots,d^{2}\}$ that \begin{equation}
\mathrm{Tr}E_{\alpha}=dk_{+}+k_{-}\text{;}
\label{forOneD}
\end{equation} 
moreover, since $\mathfrak{e}$ is a \textsc{povm}, this constant value must be 
\begin{equation}
\mathrm{Tr}E_{\alpha}=\frac{1}{d}\text{.}
\label{oneD}
\end{equation}
With Eq.~\eqref{forOneD} and Eq.~\eqref{oneD} we see that $d^{2}k_{+}+dk_{-}=1$. Now, taking into account the inequalities $k_{+}\geq k_{-}>0$, we must have $0<k_{-}/k_{+}\leq 1$; hence, for some $\kappa\in(0,1]$ we deduce
\begin{eqnarray}
k_{+}&=&\frac{d+1-\kappa^{2}}{d^{2}(d+1)}\\
k_{-}&=&\frac{\kappa^{2}}{d(d+1)}\text{.}
\end{eqnarray}
By another application of Eq.~\eqref{con5Eq} we see $\forall\beta\in\{1,\dots,d^{2}\}$ that
\begin{equation}
\sum_{\alpha=1}^{d^{2}}|E_{\alpha}\rangle\!\rangle\mathrm{Tr}\big(E_{\alpha}E_{\beta}\big)=|\mathds{1}_{d}\rangle\!\rangle\frac{k_{+}}{d}+|E_{\beta}\rangle\!\rangle k_{-}=\sum_{\alpha=1}^{d^{2}}|E_{\alpha}\rangle\!\rangle\left(\frac{k_{+}}{d}+\delta_{\alpha,\beta}k_{-}\right)\text{,}
\label{forConc}
\end{equation}
where the second equality follows again from that fact that $\mathfrak{e}$ is a \textsc{povm}. Once again, since $\mathfrak{e}$ is a basis for $\mathcal{L}_{\text{sa}}(\mathcal{H}_{d})$, we conclude from Eq.~\eqref{forConc} that $\forall\alpha,\beta\in\{1,\dots,d^{2}\}$
\begin{equation}
\mathrm{Tr}\big(E_{\alpha}E_{\beta}\big)=\frac{d^{2}\kappa^{2}\delta_{\alpha,\beta}+d+1-\kappa^{2}}{d^{3}(d+1)}\text{,}
\end{equation}
and so recalling \cref{simDef} we have established that $\mathfrak{e}$ is a \textsc{sim}.
\end{proof}
\end{theorem}

\noindent It is not possible to prove an equally strong statement for \textsc{mum}s, \textit{i.e}.\ one cannot improve \cref{mumProp} to state that a conical 2-design \textsc{povm} $\mathfrak{f}$ is a \textsc{mum} if and only if $\mathrm{card}\mathfrak{f}=d(d+1)$. Indeed, if $\{E_{1},\dots,E_{d^{2}}\}$ is a \textsc{sim}, then the $d(d+1)$ following effects obviously constitute a conical 2-design \textsc{povm}
\begin{equation}
\mathcal{E}(\mathcal{H}_{d})\ni |F_{\alpha}\rangle\!\rangle=\begin{cases} |E_{\alpha}\rangle\!\rangle\frac{1}{2} & 1\leq \alpha\leq d^{2} \\ |\mathds{1}_{d}\rangle\!\rangle\frac{1}{2d} & d^{2}+1\leq \alpha\leq d(d+1) \end{cases}\text{,}
\end{equation}
which is not mutually unbiased and therefore not a \textsc{mum}.

\noindent At this stage it will be useful to introduce some additional notation. Let $\{A_{1},\dots,A_{n}\}\subset\mathcal{L}_{\text{sa}}(\mathcal{H}_{d})$ be an arbitrary conical 2-design, and $\forall j\in\{1,\dots,n\}$ define
\begin{equation} 
t_{j}\equiv\mathrm{Tr}A_{j}\text{.} 
\end{equation}
Thus, $\forall j\in\{1,\dots,n\}$ there exists $B_{j}\in\boldsymbol{\mathcal{B}}$ so that $A_{j}$ has the Bloch representation
\begin{equation}
A_{j}=\big(\mathds{1}_{d}+B_{j}\big)\frac{t_{j}}{d}\label{desBlochs}\text{.}
\end{equation} 
Next, $\forall j\in\{1,\dots,n\}$ define 
\begin{equation}
\kappa_{j}\equiv\|B_{j}\|_{\boldsymbol{\mathcal{B}}}\text{.}
\end{equation} 
Define the \textsc{rms} trace, $t$, and weighted \textsc{rms} Bloch norm, $\kappa$, via
\begin{eqnarray}
t&\equiv&\sqrt{\frac{1}{n}\sum_{j=1}^{n}t_{j}^{2}}\text{,}\label{tDef}\\
\kappa&\equiv&\sqrt{\frac{1}{nt^{2}}\sum_{j=1}^{n}t_{j}^{2}\kappa_{j}^{2}}\label{kDef}\text{.}
\end{eqnarray}
As with the particular cases of \textsc{sim}s and \textsc{mum}s, we refer to $\kappa$ as the \textit{contraction parameter}. It is easy to see that $0\leq \kappa\leq 1$. Indeed, the foregoing equations yield the following inequalities in light of Eq.~\eqref{outBall}
\begin{equation}
0\leq\text{min}_{j}\{\kappa_{j}\}= \sqrt{\frac{\text{min}_{j}\{\kappa_{j}^{2}\}}{nt^{2}}\sum_{j=1}^{n}t_{j}^{2}}\leq\sqrt{\frac{1}{nt^{2}}\sum_{j=1}^{n}t_{j}^{2}\kappa_{j}^{2}}=\kappa
\end{equation}
and
\begin{equation}
\kappa\leq \sqrt{\frac{\text{max}_{j}\{\kappa_{j}^{2}\}}{nt^{2}}\sum_{j=1}^{n}t_{j}^{2}}=\text{max}_{j}\{\kappa_{j}\}\leq 1\text{.}
\end{equation}
Furthermore, in light of Eq.~\eqref{pureBlochs}, we see that $\kappa=1$ if and only if $\forall j\in\{1,\dots,n\}$ $\mathrm{rank}A_{j}=1$. 

\noindent Now, the traces of Eq.~\eqref{con2Eq} and Eq.~\eqref{con5Eq} read respectively as
\begin{eqnarray}
\frac{1}{2}d(d+1)k_{s}+\frac{1}{2}d(d-1)k_{a}&=&\sum_{j=1}^{n}\big(\mathrm{Tr}A_{j}\big)^{2}\text{,}\\
\frac{1}{2}d(d+1)k_{s}-\frac{1}{2}d(d-1)k_{a}&=&\sum_{j=1}^{n}\mathrm{Tr}\big(A_{j}^{2}\big)\text{,}
\end{eqnarray}
so from Eq.~\eqref{tDef} and Eq.~\eqref{kDef} we then see
\begin{eqnarray}
k_{s}&=&\frac{nt^{2}}{d^{2}}\left(1+\frac{(d-1)\kappa^{2}}{d+1}\right)\label{ksKaps}\\
k_{a}&=&\frac{nt^{2}(1-\kappa^{2})}{d^{2}}\label{kaKaps}\text{.}
\end{eqnarray}
The partial trace of Eq.~\eqref{con2Eq} over either factor thus reads as
\begin{equation}
\sum_{j=1}^{n}A_{j}t_{j}=\mathds{1}_{d}\frac{nt^{2}}{d}\text{.}
\label{sumAjTj}
\end{equation}
Therefore, for any conical 2-design $\{A_{1},\dots,A_{j}\}\subset\mathcal{L}_{\text{sa}}(\mathcal{H}_{d})_{+}$, the effects
\begin{equation}
\mathcal{E}(\mathcal{H}_{d})\ni E_{j}\equiv A_{j}\frac{dt_{j}}{nt^{2}}
\label{conDesPovm}
\end{equation}
constitute a \textsc{povm}. In the case wherein $\forall j,k\in\{1,\dots,n\}$ $\mathrm{Tr}A_{j}=\mathrm{Tr}A_{k}$, the effects defined in Eq.~\eqref{conDesPovm} are also a conical 2-design. In light of Eq.~\eqref{sumAjTj}, we see that Bloch vectors defined in Eq.~\eqref{desBlochs} satisfy
\begin{eqnarray}
\sum_{j=1}^{n}B_{j}t_{j}^{2}&=&\sum_{j=1}^{n}\left(A_{j}dt_{j}-\mathds{1}_{d}t_{j}^{2}\right)\nonumber\\
&=&0\text{.}
\label{blochSum}
\end{eqnarray}
Furthermore, from Eq.~\eqref{con5Eq}, together with Eq.~\eqref{ksKaps} and Eq.~\eqref{kaKaps}, we find that
\begin{eqnarray}
\sum_{j=1}^{n}|B_{j}\rangle\!\rangle t_{j}^{2}\langle\!\langle B_{j}|&=&\sum_{j=1}^{n}\left(|A_{j}\rangle\!\rangle d-|\mathds{1}_{d}\rangle\!\rangle t_{j}\right)\left(d\langle\!\langle A_{j}|-t_{j}\langle\!\langle\mathds{1}_{d}|\right)\nonumber\\
&=& \sum_{j=1}^{n}|A_{j}\rangle\!\rangle d^{2}\langle\!\langle A_{j}|-\sum_{j=1}^{n}\Big(|\mathds{1}_{d}\rangle\!\rangle dt_{j}\langle\!\langle A_{j}|+|A_{j}\rangle\!\rangle dt_{j}\langle\!\langle \mathds{1}_{d}|-|\mathds{1}_{d}\rangle\!\rangle t_{j}^{2}\langle\!\langle\mathds{1}_{d}|\Big)\nonumber\\
&=&|\mathds{1}_{d}\rangle\!\rangle\langle\!\langle\mathds{1}_{d}|\frac{d^{2}(k_{s}+k_{a})}{2}+\mathbf{I}_{d}\frac{d^{2}(k_{s}-k_{a})}{2}-|\mathds{1}_{d}\rangle\!\rangle nt^{2}\langle\!\langle\mathds{1}_{d}|\nonumber\\[0.3cm]
&=&\Pi_{\boldsymbol{\mathcal{B}}}\frac{ndt^{2}\kappa^{2}}{d+1}\label{blochProjSum}\text{,}
\end{eqnarray}
where $\Pi_{\boldsymbol{\mathcal{B}}}$ is the Bloch projector defined in Eq.~\eqref{blochProjector}. We now depart from the general case.

\section{Homogeneous Conical 2-Designs}\label{hc2d}
\noindent The class of all conical 2-designs is large. In order to make further progress, we shall in this section consider the special case of conical 2-designs $\{A_{1},\dots,A_{n}\}\subset\mathcal{L}_{\text{sa}}(\mathcal{H}_{d})_{+}$ that are homogeneous following sense.

\begin{definition}\label{hcdDef}\textit{A} homogeneous conical 2-design \textit{is a conical 2-design} $\{A_{1},\dots,A_{n}\}\subset\mathcal{L}_{\text{sa}}(\mathcal{H}_{d})_{+}$ \textit{for which} $\exists t\in\mathbb{R}_{> 0}$ \textit{and} $\exists \kappa\in(0,1]$ \textit{such that} $\forall j\in\{1,\dots,n\}$ $\mathrm{Tr}A_{j}=t$ \textit{and} $\|B_{j}\|_{\boldsymbol{\mathcal{B}}}=\kappa$\textit{, where} $B_{j}$ \textit{are the corresponding Bloch vectors defined in Eq.~\eqref{desBlochs}.}
\end{definition}

\noindent The fact that $\|B_{j}\|_{\boldsymbol{\mathcal{B}}}$ is constant implies that $\mathrm{Tr}(A_{j}^{2})$ is constant. Finite complex projective 2-designs are specifically those homogeneous conical 2-designs with $t=\kappa=1$, or equivalently from Eq.~\eqref{ksKaps} and Eq.~\eqref{kaKaps}, $k_{s}=2n/(d(d+1))$ and $k_{a}=0$. Furthermore, our proofs of \cref{simProp} and \cref{mumProp} establish that \textsc{sim}s and \textsc{mum}s are homogeneous conical 2-designs. We have shown in the foregoing section that \textsc{sim}s and \textsc{mum}s correspond to highly symmetric polytopes within the Bloch body: a single regular simplex in the former case, the convex hull of $(d+1)$ orthogonal regular simplices in the latter. What is the shape of the Bloch polytope corresponding to an arbitrary homogeneous conical 2-design? \cref{blochPoly} answers that question. Indeed, these shapes are fully specified by the angles formed by the Bloch vectors. Therefore, we study the Gram matrix $G$ with entries 
\begin{equation}
G_{j,k}=\langle\!\langle B_{j}|B_{k}\rangle\!\rangle 
\end{equation}
relative to some arbitrary orthonormal basis for $\mathbb{R}^{n}$.

\begin{theorem}\label{blochPoly}\textit{Let} $n,d\in\mathbb{N}$\textit{. Let} $\{B_{1},\dots,B_{n}\}\subset\boldsymbol{\mathcal{B}}(\mathcal{H}_{d})$\textit{. Then the following statements are equivalent.}
\begin{enumerate}[(i)]
\item\label{hcon1} $\forall t\in\mathbb{R}_{>0}$ $\Big\{A_{j}\equiv \frac{(\mathds{1}_{d}+B_{j})t}{d}\;\boldsymbol{\Big|}\; j\in\{1,\dots,n\}\Big\}$ \textit{is a homogeneous conical 2-design.}
\item\label{hcon2} \textit{The Gram matrix} $G$ \textit{of} $\{B_{1},\dots,B_{n}\}$ \textit{is of the form} $G=P\lambda$ \textit{where} $\lambda\in\mathbb{R}_{>0}$ \textit{and} $P$ \textit{is a real rank} $d^{2}-1$ \textit{projector with all diagonal entries equal and such that} $\forall j\in\{1,\dots,n\}$ $\sum_{k=1}^{n}P_{j,k}=0$\textit{.}
\end{enumerate}
\textit{If these equivalent conditions are satisfied, then} $\lambda\leq nd/(d+1)$ \textit{and} $\forall j,k\in\{1,\dots,n\}$ $P_{j,k}\leq (d^{2}-1)/n$\textit{, with equality when} $j=k$\textit{. The associated homogeneous conical 2-designs have contraction parameter}
\begin{equation}
\kappa=\sqrt{\frac{\lambda(d+1)}{nd}}\textit{.}
\label{lamKap}
\end{equation}
\begin{proof}
We first prove that (\ref{hcon1})$\implies$(\ref{hcon2}). Let $\{A_{1},\dots,A_{n}\}\subset\mathcal{L}_{\text{sa}}(\mathcal{H}_{d})_{+}$ be a homogeneous conical 2-design as in statement (\ref{hcon1}), with contraction parameter $\kappa\in(0,1]$. Then it follows from Eq.~\eqref{blochProjSum} that with $\lambda\equiv(nd\kappa^{2})/(d+1)$ one has that
\begin{eqnarray}
\sum_{j=1}^{n}|B_{j}\rangle\!\rangle\langle\!\langle B_{j}|=\Pi_{\boldsymbol{\mathcal{B}}}\lambda\text{;}
\label{forTrace}
\end{eqnarray}
hence $\forall j,k\in\{1,\dots,n\}$
\begin{eqnarray}
G^{2}_{j,k}=\sum_{l=1}^{n}\langle\!\langle B_{j}|B_{l}\rangle\!\rangle\langle\!\langle B_{l}|B_{k}\rangle\!\rangle=G_{j,k}\lambda\text{.}
\end{eqnarray}
Therefore $P\equiv G/\lambda$ is a real projector; moreover we see that $\text{rank}P=d^{2}-1$ from the trace of Eq.~\eqref{forTrace}
\begin{equation}
\mathrm{Tr}P=\sum_{j=1}^{n}G_{j,j}\frac{1}{\lambda}=\mathrm{Tr}\Big(\Pi_{\boldsymbol{\mathcal{B}}}\Big)=d^{2}-1\text{.}
\end{equation}
Thus $\forall j\in\{1,\dots,n\}$
\begin{equation}
P_{j,j}=\langle\!\langle B_{j}|B_{j}\rangle\!\rangle\frac{1}{\lambda}=\frac{d^{2}-1}{n}\text{,}
\end{equation} 
\textit{i.e}.\ the diagonal entries of $P$ are equal. Finally, it follows from Eq.~\eqref{blochSum} that 
\begin{equation}
\sum_{j=1}^{n}B_{j}=0\text{,}
\end{equation} 
which implies $\forall j\in\{1,\dots,n\}$ that 
\begin{equation}
\sum_{k=1}^{n}P_{j,k}=0\text{,}
\end{equation}
where we have used the fact that $P_{j,k}=G_{j,k}\lambda=\mathrm{Tr}(B_{j}B_{k})\lambda$. So (\ref{hcon1})$\implies$(\ref{hcon2}).

\noindent We now prove that (\ref{hcon2})$\implies$(\ref{hcon1}). Let the Gram matrix of $\{B_{1},\dots,B_{n}\}\in\boldsymbol{\mathcal{B}}(\mathcal{H}_{d})$ be as in statement (\ref{hcon2}). It then follows from $\text{rank}G=d^{2}-1=\text{dim}_{\mathbb{R}}\mathcal{L}_{\text{sa,0}}(\mathcal{H}_{d})$ that $\text{span}_{\mathbb{R}}\{B_{1},\dots,B_{n}\}=\mathcal{L}_{\text{sa,0}}(\mathcal{H}_{d})$. Define $\mathrm{N}=\sum_{j=1}^{n}|B_{j}\rangle\!\rangle\langle\!\langle B_{j}|$. Then $\forall j,k\in\{1,\dots,n\}$
\begin{equation}
\langle\!\langle B_{j}|\mathrm{N}B_{k}\rangle\!\rangle=\sum_{l=1}^{n}\langle\!\langle B_{j}|B_{l}\rangle\!\rangle\langle\!\langle B_{l}|B_{k}\rangle\!\rangle=\sum_{l=1}^{n}P_{j,l}P_{l,k}\lambda^{2}=P_{j,k}\lambda^{2}=G_{j,k}\lambda=\langle\!\langle B_{j}|B_{k}\rangle\!\rangle\lambda\text{.}
\label{littleCalc}
\end{equation}
Since $\text{span}_{\mathbb{R}}\{B_{1},\dots,B_{n}\}=\mathcal{L}_{\text{sa,0}}(\mathcal{H}_{d})$, and $\mathrm{N}|\mathds{1}_{d}\rangle\!\rangle=0$, Eq.~\eqref{littleCalc} implies 
\begin{equation}
\sum_{j=1}^{n}|B_{j}\rangle\!\rangle\langle\!\langle B_{j}|=\Pi_{\boldsymbol{\mathcal{B}}}\lambda\text{.}
\label{Bsums}
\end{equation} 
The fact that $\forall j\in\{1,\dots,n\}$ $\sum_{k=1}^{n}P_{j,k}=0$ implies that $\sum_{k=1}^{n}B_{k}=0$, so 
if we define $A_{j}=(\mathds{1}_{d}+B_{j})t/d$ for an arbitrary fixed $t\in\mathbb{R}_{>0}$, then from Eq.~\eqref{Bsums} and the definition of the Bloch projector in Eq.~\eqref{blochProjector} we see that these $A_{j}$ satisfy
\begin{eqnarray}
\sum_{j=1}^{n}|A_{j}\rangle\!\rangle\langle\!\langle A_{j}|=\frac{t^{2}(nd-\lambda)}{d^{3}}|\mathds{1}_{d}\rangle\!\rangle\langle\!\langle\mathds{1}_{d}|+\mathbf{I}_{d}\frac{t^{2}\lambda}{d^{2}}\text{.}
\end{eqnarray}
Thus, if we can show that $\lambda\leq nd(d+1)$, then it will follows from \cref{desCons} that $\{A_{1},\dots,A_{n}\}$ is a conical 2-design. To see that this is the case, observe the that fact that the rank $d^{2}-1$ projector $P$ is constant on the diagonal implies that $\forall j\in\{1,\dots,n\}$ $nP_{j,j}=\mathrm{Tr}P=d^{2}-1$. Consequently,
\begin{equation}
1\geq \|B_{j}\|_{\boldsymbol{\mathcal{B}}}^{2}=\frac{P_{j,j}\lambda}{d(d-1)}=\frac{(d+1)\lambda}{nd}
\label{Bjnorm}
\end{equation}
and the claim follows; moreover, the conical 2-design is homogeneous with contraction parameter
\begin{equation}
\kappa=\sqrt{\frac{\lambda(d+1)}{nd}}\text{.}
\end{equation}
We complete the proof by noting from Eq.~\eqref{Bjnorm} that the Cauchy-Schwarz inequality yields
\begin{equation}
|P_{j,k}|=|\langle\!\langle B_{j}|B_{k}\rangle\!\rangle|\frac{1}{\lambda}\leq\|B_{j}\|\|B_{k}\|\frac{1}{\lambda}=\frac{d^{2}-1}{n}\text{,}
\end{equation}
\end{proof}
\end{theorem}

\noindent \cref{blochPoly} does not specify whether \textit{every} projector satisfying the conditions stated therein is propotional to the Gram matrix of a homogeneous conical 2-design. We call projectors for which this \textit{is} the case \textit{homogeneous conical 2-design projectors}. The candidates are those satisfying the following, where $\mathcal{M}_{n}(\mathbb{R})_{\text{sa}}$ is the set of $n\times n$ real symmetric matrices.

\begin{definition}\label{canProj}\textit{. Let} $\mathbb{N}\ni n,d\geq d^{2}$\textit{. A} candidate homogeneous conical 2-design projector \textit{is an} $n\times n$ \textit{rank} $d^{2}-1$ \textit{projector} $P\in\mathcal{M}_{n}(\mathbb{R})_{\text{sa}}$ \textit{such that} $\forall j,k\in\{1,\dots,n\}$ $\sum_{l=1}^{n}P_{j,k}=0$ \textit{and} $P_{j,k}\leq(d^{2}-1)/n$ \textit{with equality when} $j=k$\textit{.}
\end{definition}

\noindent We now seek to to prove that \textit{all} candidate homogeneous conical 2-design projectors are associated with homogeneous conical 2-designs via
condition (\ref{hcon2}) in \cref{blochPoly}. For each fixed Hilbert dimension $d$ and each fixed natural number $n\geq d^{2}$ we denote the set of all candidate homogeneous conical 2-design projectors by $\mathcal{P}(n,d)$. For each $P\in\mathcal{P}(n,d)$ we denote the set of all $n$-tuples of Bloch vectors that can be associated with $P$ by $\boldsymbol{\mathcal{S}}(P)$, \textit{i.e}.\ 
\begin{equation}
\boldsymbol{\mathcal{B}}(\mathcal{H}_{d})^{\times^{n}}\ni \vec{B}\equiv\{B_{1},\dots,B_{n}\}\in\boldsymbol{\mathcal{S}}(P)\iff \exists\lambda\in\mathbb{R}_{>0}\;\forall j,k\in\{1,\dots,n\}\;\langle\!\langle B_{j}|B_{k}\rangle\!\rangle=P_{j,k}\lambda\text{.}
\end{equation} 
For each $\vec{B}\in\boldsymbol{\mathcal{S}}(P)$ we define
\begin{equation} 
\kappa_{\vec{B}}\equiv\|B_{1}\|_{\boldsymbol{\mathcal{B}}}=\dots=\|B_{n}\|_{\boldsymbol{\mathcal{B}}}\text{,}
\end{equation}
and 
\begin{equation}
\boldsymbol{k}(P)=\{\kappa_{\vec{B}}\;|\;\vec{B}\in\boldsymbol{\mathcal{S}}(P)\}\textit{.}
\end{equation} 
The convexity of the Bloch body implies that if $\vec{B}\in\boldsymbol{\mathcal{S}}(P)$ then $\forall\eta\in(0,1]$ $\vec{B}\eta\in\boldsymbol{\mathcal{S}}(P)$, so $\boldsymbol{\mathcal{S}}(P)$ is either empty or infinite. It follows that $\kappa\in\boldsymbol{k}(P)\implies (0,\kappa]\subseteq\boldsymbol{k}(P)$. Therefore, defining
\begin{equation}
c_{P}=\begin{cases}\text{sup} \boldsymbol{k}(P) & \boldsymbol{k}(P)\neq \emptyset \\ 0 & \boldsymbol{k}(P)=\emptyset\end{cases}\text{,}
\end{equation}
we see that $(0,c_{P})\subseteq\boldsymbol{k}(P)\subseteq(0,c_{P}]$. In fact, $\boldsymbol{k}(P)=(0,c_{P}]$. For the proof, note that the case $c_{P}=0$ is trivial, so we proceed with an analysis of the nontrivial case $c_{P}>0$. Choose a sequence $\vec{B}_{a}\in\boldsymbol{\mathcal{S}}(P)$ such that $\{\kappa_{\vec{B}_{a}}\}$ is a monotone nondecreasing sequence converging to $c_{p}$. Since the Bloch body is compact, we can choose a convergent subsequence such that 
\begin{equation}
\lim_{a\rightarrow\infty}\vec{B}_{a}=\vec{B}\in\boldsymbol{\mathcal{B}}(\mathcal{H}_{d})^{\times^{n}}\text{.}
\end{equation}
We then have
\begin{equation}
\langle\!\langle B_{j}|B_{k}\rangle\!\rangle=\lim_{a\rightarrow\infty}\frac{nd\kappa^{2}_{\vec{B}_{a}}}{d+1}P_{j,k}=\frac{ndc_{p}^{2}}{d+1}P_{j,k}\text{.}
\end{equation}
So $\vec{B}\in\boldsymbol{\mathcal{S}}(P)$ and $c_{p}=\kappa_{\vec{B}}\in\boldsymbol{k}(P)$. 

\noindent We are now in a position to prove the following.
\begin{theorem}\label{existThm}\textit{Let} $d,n\in\mathbb{N}_{\geq 2}$ \textit{such that} $n\geq d^{2}$\textit{. Let} $P\in\mathcal{P}(n,d)$\it{. Then}
\begin{equation}
c_{p}\geq\frac{1}{d-1}
\label{cpBound}
\end{equation}
\textit{In particular} $\boldsymbol{\mathcal{S}}(P)\neq\emptyset$\textit{.}
\begin{proof} As in the statement of the theorem, let arbitrary $P\in\mathcal{P}(n,d)$. Then it follows from \cref{blochPoly} that $P$ is a real $n\times n$ matrix such that $P^{2}=P$ and $\mathrm{Tr}P=d^{2}-1$. Therefore there are $d^{2}-1$ orthonormal vectors $\{\vec{u}_{1},\dots,\vec{u}_{d^{2}-1}\}\subset\mathbb{R}^{n}$. Denoting the $j^{th}$ component of vector $\vec{u}_{a}$ by $\vec{u}_{a,j}$, we have
\begin{equation}
P_{j,k}=\sum_{a=1}^{d^2-1} \vec{u}_{a,j}\vec{u}_{a,j}\text{.}
\label{pjk}
\end{equation}
Let $\{D_{1},\dots,D_{d^2-1}\}\subset\mathcal{L}_{\text{sa,0}}(\mathcal{H}_{d})$ be an orthonormal basis and define $\vec{B}=(B_{1},\dots,B_{n})\in\mathcal{L}_{\text{sa,0}}(\mathcal{H}_{d})^{\times^{n}}$ by
\begin{equation}
|B_{j}\rangle\!\rangle=\sum_{a=1}^{d^{2}-1}|D_{a}\rangle\!\rangle\vec{u}_{a,j}\sqrt{\frac{nd}{(d+1)(d-1)^2}}\text{.}
\label{constructionBj}
\end{equation}
Then in light of Eq.~\eqref{pjk} we see that
\begin{eqnarray}
\langle\!\langle B_{j}|B_{k}\rangle\!\rangle&=&\sum_{a=1}^{d^{2}-1}\sum_{b=1}^{d^{2}-1}\sqrt{\frac{nd}{(d+1)(d-1)^2}}\vec{u}_{b,j}\langle\!\langle D_{b}|D_{a}\rangle\!\rangle\vec{u}_{a,k}\sqrt{\frac{nd}{(d+1)(d-1)^2}}\nonumber\\
&=&\frac{nd}{(d+1)(d-1)^2}P_{j,k}
\end{eqnarray}
In particular, from \cref{blochPoly} it follows that
\begin{eqnarray}
\|B_{j}\|_{\boldsymbol{\mathcal{B}}}&=&\sqrt{\frac{\langle\!\langle B_{j}|B_{k}\rangle\!\rangle}{d(d-1)}}\nonumber\\
&=&\frac{1}{d-1}\text{,}
\end{eqnarray} 
so $B_{j}\in\boldsymbol{\mathcal{B}}(\mathcal{H}_{d})_{\text{in}}$. Therefore $\vec{B}\in \boldsymbol{\mathcal{S}}(P)$ and $1/(d-1)=\kappa_{\vec{B}}\leq c_{P}$.
\end{proof}
\end{theorem}
\noindent \cref{existThm} can be regarded as a generalization of the \textsc{sim} and \textsc{mum} existence proofs. Indeed, our \cref{blochPoly} provides the necessary and sufficient conditions for a subset of the Bloch body to define a homogeneous conical 2-design via a complete characterization of candidate homogeneous conical 2-design projectors, and our proof \cref{existThm} establishes that a corresponding homogeneous conical 2-design exists for \textit{every} candidate homogeneous conical 2-design projector; moreover, given the convexity of the Bloch body and Eq.~\eqref{cpBound}, infinite families of such designs exist. Furthermore, given an arbitrary candidate homogeneous conical 2-design projector, Eq.~\eqref{constructionBj} provides an explicit construction of the corresponding homogeneous conical 2-design. It should be pointed out, of course, that the designs constructed through Eq.~\eqref{constructionBj} are `trivial' in the sense that the each $B_{j}$ resides within the inball, and all of the points within the inball correspond to quantum states! We emphasize, however, that in light of \cref{rank1Lem} and the Seymour-Zaslavsky  Theorem \cite{Seymour1984}, nontrivial homogeneous conical 2-designs, for which $1\geq \kappa>1/(d-1)$, exist for all cases of finite Hilbert dimension. In the following section, we outline a program for finding new varieties of \textit{projective} 2-designs, that is homogeneous conical 2-designs with $\kappa=1$.

\section{In Search of \dots New Designs}\label{inSearchOf}

\noindent Our characterization of homogeneous conical 2-designs in \cref{blochPoly} sheds new light on the structure of complex \textit{projective} 2-designs. In terms of the generalized Bloch representation, and in light of \cref{rank1Lem}, we see that complex projective 2-designs are precisely those Bloch polytopes described by $P\in\mathcal{P}(n,d)$ with $c_{P}=1$. To remind the reader, this follows from the fact that $c_{p}=1$ is equivalent to $\kappa=1$, and by Eq.~\eqref{kaKaps}, $k_{a}=0$ is equivalent to $\kappa=1$. This suggests the following two-step program: (1) Classify the polytopes described by the projectors in $\mathcal{P}(n,d)$; (2) Identify those polytopes for which $c_{P}=1$. This program is, of course, extremely ambitious, for its completion would carry with it, as a minor corollary, solutions to the \textsc{sic} and \textsc{mum} existence problems. Some partial results may, however, be useful. It might, for instance, be useful if one could exclude some of the projectors in $\mathcal{P}(n,d)$ as definitely not having $c_{P}=1$. One obvious way to do this  is to exploit the fact \cite{Kimura2005} that each vertex of the polytope corresponding to a projective design must be diametrically opposite a face which is tangential to the insphere.  Having narrowed down the set of candidates, one might then investigate the remaining polytopes numerically, to see if any of them correspond to projective designs in low dimension. Based on these putative numerical insights, one might be able to make the jump to analytical arguments. 

\noindent We conclude this chapter with a result that establishes that the problem of constructing a homogeneous conical 2-design is equivalent to the problem of constructing a 1-design on a higher dimensional vector space over $\mathbb{R}$.
\newpage
\noindent\begin{theorem}\label{liftThm}\textit{Let} $n,d\in\mathbb{N}$\textit{. Let} $\{B_{1},\dots,B_{n}\}\subset\boldsymbol{\mathcal{B}}(\mathcal{H}_{d})$\textit{. Then the following statements are equivalent.}
\begin{enumerate}[(i)]
\item\label{liftCon1} $\exists t\in\mathbb{R}_{>0}$ \textit{such that} $\Big\{A_{j}\equiv \frac{(\mathds{1}_{d}+B_{j})t}{d}\;\boldsymbol{\Big|}\; j\in\{1,\dots,n\}\Big\}$ \textit{is a homogeneous conical 2-design.}
\item\label{liftCon2}  $\forall j\in\{1,\dots,n\}$ $\|B_{j}\|=\kappa$ \textit{and} $\exists\lambda\in\mathbb{R}_{>0}$ \textit{such that} $\sum_{j=1}^{n}|B_{j}\rangle\!\rangle\langle\!\langle B_{j}|=\Pi_{\boldsymbol{\mathcal{B}}}\lambda$ \textit{and} $\sum_{j=1}^{n}B_{j}=0$\textit{.}
\end{enumerate}
\begin{proof} The implication (\ref{liftCon1})$\implies$(\ref{liftCon2}) is established by Eq.~\eqref{blochSum} and Eq.~\eqref{blochProjSum}. To prove the converse, let $\{B_{1},\dots,B_{n}\}$ be as in the statement of condition (\ref{liftCon2}) and define $P\in\mathcal{M}_{n}(\mathbb{R})_{\text{sa}}$ with entries defined $\forall j,k\in\{1,\dots,n\}$ via 
\begin{equation}
P_{j,k}=\langle\!\langle B_{j}|B_{k}\rangle\!\rangle\frac{1}{\lambda}\text{.}
\end{equation}
Then
\begin{eqnarray}
P^{2}_{j,k}=\sum_{l=1}^{n}P_{j,l}P_{l,k}=\sum_{l=1}^{n}\langle\!\langle B_{j}|B_{l}\rangle\!\rangle\langle\!\langle B_{l}|B_{k}\rangle\!\rangle\frac{1}{\lambda^{2}}=\langle\!\langle B_{j}|\Pi_{\boldsymbol{\mathcal{B}}}B_{k}\rangle\!\rangle\frac{1}{\lambda}=P_{j,k}\text{,}\\
\mathrm{Tr}P=\sum_{j=1}^{n}\langle\!\langle B_{j}|B_{l}\rangle\!\rangle\frac{1}{\lambda}=\sum_{j=1}^{n}\mathrm{Tr}\big(|B_{j}\rangle\!\rangle\langle\!\langle B_{j}|\big)\frac{1}{\lambda}=\mathrm{Tr}\big(\Pi_{\boldsymbol{\mathcal{B}}}\big)=d^{2}-1\text{,}
\end{eqnarray}
so $P$ is a real rank $d^{2}-1$ projector. The fact that  $\forall j\in\{1,\dots,n\}$ $\|B_{j}\|=\kappa$ implies that $P$ is constant on the diagonal; moreover the fact that $\sum_{j=1}^{n}B_{j}=0$ implies that $\forall j\in\{1,\dots,n\}$ $\sum_{k=1}^{n}P_{j,k}=0$. The implication (\ref{liftCon2})$\implies$(\ref{liftCon1}) thus follows from \cref{blochPoly}.
\end{proof}
\end{theorem}

\chapter{Entanglement and Designs}
\label{entanglementConicalDesigns}

\epigraphhead[40]
	{
		\epigraph{``\dots here's Tom with the weather.''}{---\textit{Bill Hicks}\\Revelations (1993)}
	}

In \cite{Graydon2016}, the author and D.\ M.\ Appleby consider the structure of entanglement in the light of conical 2-designs. This chapter centres on that publication.

\noindent The conical 2-designs that we introduced in \cref{designsOnQuantumCones} (see \cref{c2dDef}) are naturally adapted to a information-geometric description of bipartite quantum entanglement. In this chapter, we will prove that the outcome probabilities computed from an arbitrary bipartite quantum state for a quantum measurement formed from the tensor product of local conical 2-design \textsc{povm}s capture important entanglement properties of the state in question. Specifically, we will show that conical 2-designs are intimately linked with physical theory concerning the quantification of entanglement. The notion of entanglement as a resource for quantum communication and quantum computation was pinpointed by Wootters in \cite{Wootters1998a}. It is thus desirable to quantify the amount entanglement manifest in a given quantum state. In \cite{Bennett1996}, Bennett-DiVincenzo-Smolin-Wootters proposed the \textit{monotonicity axiom}: a measure of entanglement does not increase under Local Operations and Classical Communication (\textsc{locc}) \cite{Chitambar2014}. Shortly thereafter \cite{Vedral1997}, Vedral-Plenio-Rippin-Knight introduced an axiomatic framework for entanglement measures, which in addition to monotonicity under \textsc{locc}, includes additional desiderata. For instance, Vedral-Plenio-Rippin-Knight proposed that any measure of entanglement ought to vanish on separable states. Adopting a stronger version of the monotonicity axiom,  Vidal introduced an important class of functions, namely \textit{entanglement monotones} \cite{Vidal2000}, and derived their essential properties. The main result of the present chapter, \cref{monDesThm}, establishes conical 2-designs are both necessary and sufficient for casting (regular) entanglement monotones from the length of product measurement probability vectors.

\noindent We discussed entanglement in \cref{prologue} and \cref{quantumTheory}. In order to understand the main results presented in this chapter, we shall require reminders of some additional well known definitions and facts concerning the theory of entanglement. We shall recall these preliminary notions in \cref{prelims}. Specifically, therein, we shall discuss entanglement monotones \cite{Vidal2000}. Furthermore, we shall define (see \cref{regDef}) a novel concept: \textit{regular} entanglement monotones. We conclude \cref{prelims} with a proof that the concurrence \cite{Wootters1998b}\cite{Rungta2001}\cite{Rungta2003} is regular. \cref{monotones} is devoted to \cref{monDesThm} and its proof. In \cref{witnesses}, we develop and generalize recent work \cite{Wiesniak2011}\cite{Spengler2012}\cite{Chen2014}\cite{Chen2015}\cite{Liu2015}\cite{Shen2015}\cite{Kalev2013} relating entanglement witnesses and designs. We conclude with \cref{decomps} by exploring the obvious connection between conical 2-designs and Werner states \cite{Werner1989}.

\newpage
\section{Entanglement Montones}\label{prelims}
In this section, we first collect prerequisites regarding entanglement monotones. We then introduce the notion of a \textit{regular} entanglement monotone, and we prove that the concurrence \cite{Wootters1998b}\cite{Rungta2001}\cite{Rungta2003} is regular.

\noindent An entanglement monotone is, very crudely speaking, a function defined on quantum states, whose value on a particular state is indicative of the degree of entanglement manifest therein. Naturally, one desires that such a function enjoys properties pertaining to physical notions. For instance, the function (denoted $\mathsf{E}$) ought to evaluate to a constant on all separable states, say zero, for interpretational clarity. One might also be tempted to say, furthermore, that the function ought to vanish only on separable states. Phrased more precisely, wherein and henceforth we restrict our attention to bipartite entanglement,
\begin{equation}
\mathsf{E}(\rho)=0\iff\rho\in\text{Sep}\mathcal{Q}\big(\mathcal{H}_{d_{\mathrm{A}}}\otimes\mathcal{H}_{d_{\mathrm{B}}}\big)\text{.}
\label{vanSep}
\end{equation}
In addition to Eq.~\eqref{vanSep}, one might require that $\mathsf{E}$ be invariant under local unitary transformations, \textit{i.e}.\ $\forall U\in\mathrm{U}(\mathcal{H}_{d_{\mathrm{A}}})$ and $\forall V\in\mathrm{U}(\mathcal{H}_{d_{\mathrm{B}}})$ and $\forall\rho\in\mathcal{Q}(\mathcal{H}_{d_{\mathrm{A}}}\otimes\mathcal{H}_{d_{\mathrm{B}}})$
\begin{equation}
\mathsf{E}(\rho)=\mathsf{E}\Big((U\otimes V)\rho(U^{*}\otimes V^{*})\Big)\text{,}
\label{invarLU}
\end{equation}
which is in fact\footnote{This fact follows from the transitivity of $\mathrm{U}(\mathcal{H}_{d})$ on $\mathcal{S}(\mathcal{H}_{d})$ and the Schmidt decomposition introduced in \cref{schmidtDecomp}.} a necessary and sufficient condition for $\mathsf{E}$ to remain constant on all separable pure states. Physically, Eq.~\eqref{invarLU} encapsulates the intuition that any measure of entanglement ought to be invariant under independent free evolutions of the subsystems in question. More generally, it is natural to ask that $\mathsf{E}$ not increase under quantum channels associated with \textit{local operations and classical communication}: those processes resulting from the sequential composition of physical operations performed locally on the subsystems in question, which may depend on intermediate classical communication between the local parties performing those operations. Mathematically, the definition of the set of all such \textsc{locc} \textit{channels} is rather complicated \cite{Chitambar2014}. For our purposes it will suffice to leave the definition implicit, and to note that the separable quantum channels defined in \cref{sepChan} form \cite{Watrous2011} a strict superset of all \textsc{locc} channels. The aforementioned natural demand on $\mathsf{E}$ is thus stated as follows: for all \textsc{locc} channels $\Theta$
\begin{equation}
\mathsf{E}(\rho)\geq \mathsf{E}\big(\Theta(\rho)\big)\text{.}
\label{monCon}
\end{equation}
Eqs.~\eqref{vanSep}, \eqref{invarLU}, and \eqref{monCon} are precisely the conditions expounded by Vedral-Plenio-Rippin-Knight \cite{Vedral1997} in their seminal work on axiomatic formulations of \textit{entanglement measures}. The condition given by Eq.~\eqref{monCon} is now known as the \textit{monotonicity axiom}. Historically, the monotonicity axiom was in fact suggested slightly earlier (although not in an axiomatic context) by Bennett-DiVincenzo-Smolin-Wootters in \cite{Bennett1996}. One could instead demand the stronger condition that $\mathsf{E}$ not increase \textit{on average} under \textsc{locc}, \textit{i.e}.\ for any \textsc{locc} channel with Kraus operators $A_{j}$ and with $p_{j}\equiv\mathrm{Tr}\big(A_{j}\rho A_{j}^{*}\big)$ and $\rho_{j}\equiv(A_{j}\rho A_{j}^{*})/p_{j}$ one demands that
\begin{equation}
\mathsf{E}(\rho)\geq \sum_{j}\mathsf{E}(\rho_{j})p_{j}\text{.}
\label{strongMon}
\end{equation}
Eq.~\eqref{strongMon} was proposed by Vidal \cite{Vidal2000} and is now known as the \textit{strong monotonicity axiom}. Vidal's definition of \textit{entanglement monotone} asks for the strong monotonicity axiom together with \textit{convexity}, that is for any convex decomposition $\rho=\sum_{k}\rho_{k}q_{k}$ convexity is the demand that
\begin{equation}
\mathsf{E}(\rho)\leq \sum_{k}\mathsf{E}(\rho_{k})q_{k}\text{.}
\label{conMon}
\end{equation}
Eq.~\eqref{conMon} ensures, for instance, that $\mathsf{E}$ is monotonic under information loss concerning the results of local operators; moreover, convexity turns out to be a very powerful tool for the proofs given by Vidal in \cite{Vidal2000} and for the facts that we are about to recall. First, let us promote Vidal's entanglement monotones to the status of a formal definition, wherein we continue to restrict our treatment of this subject to the bipartite case.
\begin{definition}\label{eMonDef}(\text{Vidal} \cite{Vidal2000}) \textit{An} entanglement monotone \textit{is a function} $\mathsf{E}$\textit{, with domain bipartite quantum states and range} $\mathbb{R}_{\geq 0}$\textit{, which obeys the strong monotonicity axiom (Eq.~\eqref{strongMon}) and convexity (Eq.~\eqref{conMon}).}
\end{definition}
\noindent It is a matter of convention that we take the range of an entanglement monotone to be $\mathbb{R}_{+}$. Indeed, Vidal proved that any entanglement monotone is constant on separable states, so any function obeying strong monotonicity and convexity can easily be scaled to satisfy the aforementioned range condition. Of course, any entanglement monotone enjoys the weaker form of monotonicity expressed in Eq.~\eqref{monCon}
\begin{equation}
\mathsf{E}\big(\Theta(\rho)\big)=\mathsf{E}\left(\sum_{j}\rho_{j}p_{j}\right)\leq \sum_{k}\mathsf{E}(\rho_{j})q_{j}\leq \mathsf{E}(\rho)\text{,}
\end{equation}
where the first inequality comes from convexity and the second from strong monotonicity. Additionally, Vidal proved that any entanglement monotone enjoys local unitary invariance, \textit{i.e}.\ Eq.~\eqref{invarLU}. It is not the case that any entanglement monotone vanishes \textit{only} on separable states \cite{Horodecki2009} ; however, the particular monotone that we will consider later does have this property.

\noindent Before proceeding further, we now make the simplifying assumption that $d_{\mathrm{A}}=d_{\mathrm{B}}\equiv d$. If $d_{\mathrm{A}}\neq d_{\mathrm{B}}$, then one can view the quantum state space corresponding to $\mathrm{min}\{d_{\mathrm{A}},d_{\mathrm{B}}\}$ as a restricted subset of the larger, and most of what follows can be easily generalized along these lines. It will be, however, very convenient to adopt this formal restriction for the purposes of notation, and more importantly for consistency with the statements of results that we shall recall from the literature. On this restriction, with $d_{\mathrm{A}}=d_{\mathrm{B}}=d$, we shall write $\mathrm{Tr}_{1}$ for the partial trace $\mathrm{Tr}_{\mathcal{H}_{d_{\mathrm{A}}}}$ and $\mathrm{Tr}_{2}$ for the partial trace $\mathrm{Tr}_{\mathcal{H}_{d_{\mathrm{B}}}}$ over $\mathcal{L}(\mathcal{H}_{d_{\mathrm{A}}}\otimes\mathcal{H}_{d_{\mathrm{B}}})=\mathcal{L}(\mathcal{H}_{d}\otimes\mathcal{H}_{d})$.

\begin{theorem}(Vidal \cite{Vidal2000})\label{vidalThm}\textit{The restriction of any entanglement monotone to pure states} $\text{Pur}\mathcal{Q}(\mathcal{H}_{d}\otimes\mathcal{H}_{d})$ \textit{is given by a local unitarily invariant concave function of the partial trace} $\mathrm{Tr}_{2}\big(|\Psi\rangle\langle\Psi|\big)\in\mathcal{Q}(\mathcal{H}_{d})$\textit{.}
\end{theorem}
\noindent We refer the reader to \cite{Vidal2000} for a proof of \cref{vidalThm}, and we remind the reader that a \textit{concave function} is a function $f:\mathcal{X}\longrightarrow\mathbb{R}$, with $\mathcal{X}$ a subset of a vector space over $\mathbb{R}$, such that $\forall x,y\in\mathcal{X}$ and $\forall\lambda\in[0,1]$ $f(x\lambda+y(1-\lambda))\geq f(x)\lambda+f(y)(1-\lambda)$. A key ingredient for Vidal's proof is the Schmidt decomposition, which we shall recall presently, and which will feature prominently in some of our later proofs herein.
\begin{theorem}\cite{Nielsen2010}\label{schmidtDecomp} $\forall\Psi\in\mathcal{S}\big(\mathcal{H}_{d}\otimes\mathcal{H}_{d}\big)\exists\lambda_{1},\dots,\lambda_{d}\in[0,1]$ \textit{together with orthonormal bases} $\{e_{1},\dots,e_{d}\}\subset\mathcal{H}_{d}$ \textit{and} $\{f_{1},\dots,f_{d}\}\subset\mathcal{H}_{d}$ \textit{such that} $\Psi=(e_{1}\otimes f_{1})\lambda_{1}+\dots+(e_{d}\otimes f_{d})\lambda_{d}$\textit{; moreover,} $\lambda_{1}^{2}+\dots+\lambda_{d}^{2}=1$\textit{.}
\end{theorem}
\noindent As a matter of terminology, one refers to $\{e_{j}\}$ and $\{f_{j}\}$ as \textit{Schmidt bases}, and to $\lambda_{j}$ as \textit{Schmidt coefficients}. From \cref{schmidtDecomp} it is easy to see that the partial trace of $|\Psi\rangle\langle\Psi|$ over either local Hilbert space is diagonal, with entries given by the squares of the Schmidt coefficients. Consequently, it is but a matter of convention that one traces over $\mathcal{H}_{d_{\mathrm{B}}}$ in the statement of \cref{vidalThm}; moreover, from this observation we see that the restriction of any entanglement monotone to pure states is given by a concave function of the squares of the Schmidt coefficients, which are invariant under local unitary transformations.

\noindent In addition to fully characterizing the restrictions of entanglement monotones to pure states, Vidal proved that \textit{any} local unitarily invariant concave function defined on the partial traces of pure states can be extended to an entanglement monotone via the convex roof construction \cite{Uhlmann2010}, which we shall detail presently.

\noindent Let
\begin{equation}
\mathsf{e}:\mathcal{S}\big(\mathcal{H}_{d}\otimes\mathcal{H}_{d})\longrightarrow\mathbb{R}::\Psi\longmapsto f\big(\mathrm{Tr}_{2}(|\Psi\rangle\langle\Psi|)\big)\text{,}
\end{equation}
where $f:\mathcal{Q}(\mathcal{H}_{d})\longrightarrow\mathbb{R}$ is unitarily invariant and concave. Next, $\forall\rho\in\mathcal{Q}(\mathcal{H}_{d}\otimes\mathcal{H}_{d})$ define
\begin{equation}
\Upsilon_{\rho}\equiv\big\{(\Psi_{j},p_{j})\in\mathcal{S}(\mathcal{H}_{d_{\mathrm{A}}}\otimes\mathcal{H}_{d_{\mathrm{B}}})\times[0,1]\;\boldsymbol{|}\;\rho=\sum_{j}|\Psi_{j}\rangle p_{j}\langle\Psi_{j}|\big\}\text{,}
\end{equation}
which is to say, colloquially, that $\Upsilon_{\rho}$ is the set of all convex decompositions of $\rho$. Now, extend $\mathsf{e}$ to
\begin{equation}
\mathsf{E}:\mathcal{Q}(\mathcal{H}_{d}\otimes\mathcal{H}_{d})\longrightarrow\mathbb{R}::\rho\longmapsto \inf_{\Upsilon_{\rho}}\sum_{j}\mathsf{e}\big(\Psi_{j}\big)p_{j}\text{.}
\end{equation}
Vidal proved \cite{Vidal2000} that any such $\mathsf{E}$ constructed in this way is an entanglement monotone. In particular, we will be interested in the concurrence \cite{Wootters1998b}\cite{Rungta2001}\cite{Rungta2003}, which is defined as follows.
\begin{definition}\label{conDef}\textit{Let} 
\begin{equation}
\mathsf{c}:\mathcal{S}\big(\mathcal{H}_{d}\otimes\mathcal{H}_{d})\longrightarrow\mathbb{R}::\Psi\longmapsto\sqrt{2-2\mathrm{Tr}\Big(\big(\mathrm{Tr}_{2}|\Psi\rangle\langle\Psi|\big)^{2}\Big)}\textit{.}
\label{pureCon}
\end{equation}
\textit{The} concurrence \textit{is}
\begin{equation}
\mathsf{C}:\mathcal{Q}(\mathcal{H}_{d}\otimes\mathcal{H}_{d})\longrightarrow\mathbb{R}_{\geq 0}::\rho\longmapsto\inf_{\Upsilon_{\rho}}\sum_{j}\mathsf{c}\big(\Psi_{j}\big)p_{j}\text{.}
\end{equation}
\end{definition}
\noindent In light of the foregoing discussion and the following lemma, the concurrence is an entanglement monotone.
\begin{lemma}\label{conLem}\textit{The function} $\mathsf{c}$ \textit{defined in Eq.~\eqref{pureCon} enjoys unitary invariance and concavity.}
\end{lemma}
\noindent We relegate a simple proof of \cref{conLem} to \cref{conProof}. The concurrence is simply related to the entanglement of formation \cite{Bennett1996}, and has some useful properties. Computing entanglement monotones for an arbitrary mixed state can be a difficult problem.  The concurrence has the useful feature that it has easily computable lower and upper bounds \cite{Mintert2004a}\cite{Minert2004b}.  In the case of a pair of qubits there is an explicit formula for the concurrence itself \cite{Wootters1998b}.  Another useful property is the fact that it vanishes if and only if the state is separable. As an aside, we note that the concurrence illuminates the monogamous character of entanglement \cite{Coffman2000}. Finally, the concurrence of a pure bipartite state $|\Psi\rangle\langle\Psi|$ is simply related to its Schmidt coefficients as follows
\begin{equation}
\mathsf{C}\big(|\Psi\rangle\langle\Psi|\big)=\sqrt{2-2\sum_{j=1}^{d}\lambda_{j}^{4}}\text{,}
\end{equation}
since
\begin{eqnarray}
\mathrm{Tr}\Big(\big(\mathrm{Tr}_{2}|\Psi\rangle\langle\Psi|\big)^{2}\Big)&=&\sum_{j=1}^{d}\sum_{k=1}^{d}\mathrm{Tr}\left(\Big(\mathrm{Tr}_{2}\big(|e_{j}\rangle\lambda_{j}\langle e_{k}|\otimes |f_{j}\rangle\lambda_{k}\langle f_{k}|\big)\Big)^{2}\right)\nonumber\\
&=&\sum_{j=1}^{d}\mathrm{Tr}\left(\Big(|e_{j}\rangle\lambda_{j}^{2}\langle e_{j}|\Big)^{2}\right)\nonumber\\
&=&\sum_{j=1}^{d}\lambda_{j}^{4}\text{.}
\label{conSch}
\end{eqnarray}
In light of Eq.~\eqref{conSch} and convexity of the concurrence, we see that $\forall\rho\in\mathcal{Q}(\mathcal{H}_{d}\otimes\mathcal{H}_{d})$
\begin{equation}
0\leq \mathsf{C}(\rho)\leq\sqrt{\frac{2(d-1)}{d}}\equiv x_{\text{max}}\text{.}
\label{conBounds}
\end{equation}
Eq.~\eqref{conSch} also reveals an important topological property of the concurrence, to be introduced presently. First, recall that the interior of a subset $\mathcal{X}$ of the unit sphere in $\mathcal{H}_{d}$ is all of those points in $\mathcal{X}$ for which an epsilon ball strictly contained in $\mathcal{X}$ can be centred thereupon. Formally, we have the following.
 
\begin{definition}\textit{Let subset} $\mathcal{X}\subseteq\mathcal{S}(\mathcal{H}_{d}\otimes\mathcal{H}_{d})$. \textit{The} interior \textit{of} $\mathcal{X}$ \textit{is denoted and defined as follows:} $\mathrm{int}\mathcal{X}\equiv\{\Psi\in\mathcal{X}\;\boldsymbol{|}\;\exists \epsilon >0\text{ s.t.\ } B_{\epsilon}(\Psi)\subseteq\mathcal{X}\}\textit{,}$\textit{where} 
$B_{\epsilon}(\Psi)\equiv\{\Phi\in\mathcal{S}(\mathcal{H}_{d}\otimes\mathcal{H}_{d})\;\boldsymbol{|}\;\|\Psi-\Phi\|<\epsilon\}$\textit{.}
\end{definition}

\noindent Any entanglement monotone assigns a value to each pure state. When these values are such that the set of pure states that yield a given value corresponds to a set with empty interior in $\mathcal{S}(\mathcal{H}_{d})$ via $|\Psi\rangle\langle\Psi|\leftrightarrow\Psi$, we give the monotone a special name.

\begin{definition}\label{regDef}
\textit{A} regular entanglement monotone \textit{is an entanglement monotone} $\mathsf{E}:\mathcal{Q}(\mathcal{H}_{d}\otimes\mathcal{H}_{d})\longrightarrow\mathbb{R}_{\geq 0}$ \textit{such that} $\forall x\in\mathbb{R}_{\geq 0}$ \textit{the interior of} $\mathsf{E}_{x}\equiv\{\Psi\in\mathcal{S}(\mathcal{H}_{d}\otimes\mathcal{H}_{d})\;\boldsymbol{|}\;\mathsf{E}(|\Psi\rangle\langle\Psi|)=x\}$ \textit{is empty.}
\end{definition}
\noindent The author and D.\ M.\ Appleby introduced \cref{regDef} in \cite{Graydon2016}. Colloquially speaking, regularity means that an entanglement monotone is maximally sensitive to changes in pure state entanglement. Phrased more precisely, a regular entanglement monotone partitions the unit sphere $\mathcal{S}(\mathcal{H}_{d}\otimes\mathcal{H}_{d})$ into disjoint sets, each of which coincides with its own boundary. Each such set corresponds to a distinct degree of entanglement. Incidentally, this establishes that a constant function is not a regular entanglement monotone, so every regular entanglement monotone enjoys nontrivial local unitary invariance as in the upcoming \cref{nluiDef}. We now prove that the concurrence is regular.
\begin{lemma}\label{conRegLem}\textit{The concurrence is a regular entanglement monotone.}
\begin{proof} It is very well known that the concurrence is an entanglement monotone. Theorem 2 in \cite{Vidal2000} and \cref{conLem} cements the proof. It remains to establish regularity. We will consider two mutually exclusive and exhaustive cases: separable pure states with $\mathsf{C}=0$ and entangled pure states with $0<\mathsf{C}\leq x_{\text{max}}$, where $x_{\text{max}}$ is as defined in Eq.~\eqref{conBounds}. We will abuse notation and write $\mathsf{C}(\Psi)$ for the concurrence of a pure state $|\Psi\rangle\langle\Psi|$. We first concern ourselves with the case of separable pure states and consider 
\begin{equation}
\mathsf{C}_{0}=\{\Psi\in\mathcal{S}(\mathcal{H}_{d}\otimes\mathcal{H}_{d})\;\boldsymbol{|}\;\mathsf{C}(\Psi)=0\}\text{.}
\end{equation} 
Any $\Psi\in\mathsf{C}_{0}$ is of the form $\psi\otimes\phi$ for some $\psi,\phi\in\mathcal{S}(\mathcal{H}_{d})$. We will now prove that $\mathrm{int}\mathsf{C}_{0}$ is empty via \textit{reductio ad absurdum}. Suppose $\exists \Psi\in\mathrm{int}\mathsf{C}_{0}$. Then $\exists\epsilon>0$ such that $ B_{\epsilon}(\Psi)\subseteq\mathsf{C}_{0}$. By definition $B_{\epsilon}(\Psi)\subset\mathcal{S}(\mathcal{H}_{d}\otimes\mathcal{H}_{d})$, so $\epsilon\leq 2$. We will construct a $\mathcal{S}(\mathcal{H}_{d}\otimes\mathcal{H}_{d})\ni\Phi\in B_{\epsilon}(\Psi)$ with $\mathsf{C}(\Phi)>0$. Let $\psi^{\perp},\phi^{\perp}\in\mathcal{S}(\mathcal{H}_{d})$ be such that $\langle\psi|\psi^{\perp}\rangle=0$ and $\langle\phi|\phi^{\perp}\rangle=0$. Next, let $\mathbb{R}_{>0}\ni\lambda_{1}=1-\epsilon^{2}/32$ and $\mathbb{R}_{>0}\ni\lambda_{2}=\sqrt{\epsilon^{2}/16-\epsilon^{4}/1024}$. Now, construct $\Phi=(\psi\otimes\phi)\lambda_{1}+(\psi^{\perp}\otimes\phi^{\perp})\lambda_{2}$. Clearly $\Phi\in\mathcal{S}(\mathcal{H}_{d}\otimes\mathcal{H}_{d})$. One also has that $\Phi\in B_{\epsilon}(\Psi)$, for
\begin{eqnarray}
\|\Psi-\Phi\|&=&\sqrt{\langle\Psi|\Psi\rangle+\langle\Phi|\Phi\rangle-\langle\Psi|\Phi\rangle-\langle\Phi|\Psi\rangle}\nonumber\\
&=&\sqrt{2(1-\lambda_{1})}\nonumber\\
&=&\epsilon/4\nonumber\\
&<&\epsilon\text{.}
\end{eqnarray}
However, $(\psi\otimes\phi)\lambda_{1}+(\psi^{\perp}\otimes\phi^{\perp})\lambda_{2}$ is a Schmidt decomposition, thus we see $\mathsf{C}(\Phi)=\sqrt{2}\sqrt{1-\lambda_{1}^{4}-\lambda_{2}^{4}}>0$. We have reached a contradiction. Therefore our supposition, namely that $\emptyset\neq\mathrm{int}\mathsf{C}_{0}$, is false. We conclude that the interior of $\mathsf{C}_{0}$ is empty.
\newpage
\noindent We now concern ourselves with the case of entangled pure states and consider
\begin{equation}
\mathsf{C}_{x}=\{\Psi\in\mathcal{S}(\mathcal{H}_{d}\otimes\mathcal{H}_{d})\;\boldsymbol{|}\;\mathsf{C}(\Psi)=x\}\text{,} 
\end{equation}
for some arbitrary nonzero $x\in\mathbb{R}$. If $\mathsf{C}_{x}$ is empty we are done. So we proceed with an analysis of the subcase where $\mathsf{C}_{x}$ is nonempty, in which case $x$ is bounded from above according to Eq.~\eqref{conBounds}. Again, we will prove that $\mathrm{int}\mathsf{C}_{x}$ is empty via \textit{reductio ad absurdum}. Suppose $\exists \Psi\in\mathrm{int}\mathsf{C}_{x}$. Then $\exists\epsilon>0$ such that $B_{\epsilon}(\Psi)\subseteq\mathsf{C}_{x}$. Once more, note by definition $B_{\epsilon}(\Psi)\subset\mathcal{S}(\mathcal{H}_{d}\otimes\mathcal{H}_{d})$, so $\epsilon\leq 2$. We will now construct a $\mathcal{S}(\mathcal{H}_{d}\otimes\mathcal{H}_{d})\ni\Phi\in B_{\epsilon}(\Psi)$ with $\mathsf{C}(\Phi)\neq x$. To begin, note that $\Psi\in\mathrm{int}\mathsf{C}_{x}$ implies that there exist at least two nonzero Schmidt coefficients in its Schmidt decomposition, which without the loss of generality we take to be $\lambda_{1},\lambda_{2}\in\mathbb{R}_{>0}$
\begin{equation}
\Psi=(e_{1}\otimes f_{1})\lambda_{1}+(e_{2}\otimes f_{2})\lambda_{2}+\sum_{r>2}^{d}(e_{r}\otimes f_{r})\lambda_{r}\text{.}
\end{equation} For $\Psi$ as such, introduce the constant $R^{2}\equiv\lambda_{1}^{2}+\lambda_{2}^{2}$. Note that $0<R\leq1$. It follows that $\lambda_{1}=R\cos\theta$ and $\lambda_{2}=R\sin\theta$ for some $\theta$ in the open interval $(0,\pi/2)$. Let $\tau\in(0,\pi/2)$ be such that $\tau=\text{arccos}(1-\epsilon^{2}/32)$. Next, introduce $\mathbb{R}\ni\mu_{1}=R\cos(\theta-\tau)$ and $\mathbb{R}\ni\mu_{2}=R\sin(\theta-\tau)$. Now, construct 
\begin{equation}
\Phi=(e_{1}\otimes f_{1})\mu_{1}+(e_{2}\otimes f_{2})\mu_{2}+\sum_{r>2}^{d}(e_{r}\otimes f_{r})\lambda_{r}\text{.}
\label{phiDef}
\end{equation} 
Clearly $\Phi\in\mathcal{S}(\mathcal{H}_{d}\otimes\mathcal{H}_{d})$. Also, notice that
\begin{eqnarray}
\langle\Psi|\Phi\rangle&=&R^{2}\cos\theta\cos(\theta-\tau)+R^{2}\sin\theta\sin(\theta-\tau)+\sum_{r>2}^{d}\lambda_{r}^{2}\nonumber\\
&=&R^{2}\cos\theta\cos(\theta-\tau)+R^{2}\sin\theta\sin(\theta-\tau)+1-R^{2}\nonumber\\[0.2cm]
&=&\langle\Phi|\Psi\rangle\text{.}
\end{eqnarray}
Recalling the trigonometric identity $\cos\theta\cos(\theta-\tau)+\sin\theta\sin(\theta-\tau)=\cos(\tau)$, we see that $\Phi\in B_{\epsilon}(\Psi)$, for with $R\leq 1$ we have
\begin{eqnarray}
\|\Psi-\Phi\|&=&\sqrt{\langle\Psi|\Psi\rangle+\langle\Phi|\Phi\rangle-\langle\Psi|\Phi\rangle-\langle\Phi|\Psi\rangle}\nonumber\\
&=&\sqrt{2\big(1-R^{2}\cos\theta\cos(\theta-\tau)-R^{2}\sin\theta\sin(\theta-\tau)-1+R^{2}\big)}\nonumber\\
&=&\sqrt{2R^{2}\big(1-\cos\theta\cos(\theta-\tau)-\sin\theta\sin(\theta-\tau)\big)}\nonumber\\
&=&\sqrt{2R^{2}\big(1-\cos\tau\big)}\nonumber\\
&=&\sqrt{2R^{2}\big(1-1+\epsilon^{2}/32\big)}\nonumber\\[0.2cm]
&=&R\epsilon/4\nonumber\\[0.2cm]
&\leq& \epsilon/4\nonumber\\[0.2cm]
&<& \epsilon\text{.}
\end{eqnarray}
We now seek to establish the contradiction $\mathsf{C}(\Phi)\neq x$. With $\tau$ thus far only restricted such that $\tau\in(0,\pi/2)$ with $\tau=\text{arccos}(1-\epsilon^{2}/32)$, it is very important to point out that Eq.~\eqref{phiDef} is not necessarily a Schmidt decomposition, because it is possible that one, or both of $\mu_{1}$ and $\mu_{2}$ are negative. Let us therefore introduce the following orthonormal basis
$\{\tilde{f}_{1}\equiv f_{1}\mathrm{sign}(\mu_{1}),\tilde{f}_{2}\equiv f_{2}\mathrm{sign}(\mu_{2}),f_{3},\dots,f_{d}\}$, so that $(e_{1}\otimes \tilde{f}_{1})|\mu_{1}|+(e_{2}\otimes \tilde{f_{2}})|\mu_{2}|+\sum_{r>2}^{d}(e_{r}\otimes f_{r})\lambda_{r}$ \textit{is} a Schmidt decomposition. With this observation, and from the trigonometric identities $8|\cos(\theta-\tau)|^{4}=8\cos^{4}(\theta-\tau)=3+4\cos(2\theta-2\tau)+\cos(4\theta-4\tau)$ and 
$8|\sin(\theta-\tau)|^{4}=8\sin^{4}(\theta-\tau)=3-4\cos(2\theta-2\tau)+\cos(4\theta-4\tau)$, we see that
\begin{eqnarray}
\mathsf{C}(\Phi)=\sqrt{2\left(1-|\mu_{1}|^{4}-|\mu_{2}|^{4}-\sum_{r>2}^{d}\lambda_{r}^{4}\right)}=\sqrt{2\left(1-\frac{3R^{4}}{4}-\frac{R^{4}\cos(4\theta-4\tau)}{4}-\sum_{r>2}^{d}\lambda_{r}^{4}\right)}\text{,}
\end{eqnarray}
whereas
\begin{equation}
x=\mathsf{C}(\Psi)=\sqrt{2\left(1-\frac{3R^{4}}{4}-\frac{R^{4}\cos(4\theta)}{4}-\sum_{r>2}^{d}\lambda_{r}^{4}\right)}
\end{equation}
Our supposition implies that $\cos(4\theta)=\cos(4\theta-4\tau)$. Therefore $\tau=n\pi/2$ for some $n\in\mathbb{Z}$. This, however, is impossible because $\tau=\text{arccos}(1-\epsilon^{2}/32)$ and $0<\epsilon\leq 2$. We have reached a contradiction. Therefore our supposition, namely that $\emptyset\neq\mathrm{int}\mathsf{C}_{x}$, is false. We conclude that the interior of $\mathsf{C}_{x}$ is empty. All cases have been examined. 
\end{proof}
\end{lemma}
\noindent We are now ready for the sequel.
\section{Monotones and Designs}\label{monotones}
In this section, we first state \cref{monDesThm}. Our proof of \cref{monDesThm} is founded on \cref{nluiLem}. The proof of \cref{nluiLem} is intense, hence, for motivational purposes, we first state \cref{monDesThm}, then state \cref{nluiLem}, then prove \cref{monDesThm}, and then prove \cref{nluiLem}.
\begin{theorem}\label{monDesThm}\textit{Let} $\{E_{1},\dots,E_{n}\}\subset\mathcal{E}(\mathcal{H}_{d})$ \textit{be a} \textsc{povm}\textit{. Let} $\Psi\in\mathcal{S}(\mathcal{H}_{d}\otimes\mathcal{H}_{d})$\textit{. Define} $\forall\alpha,\beta\in\{1,\dots,n\}$ \textit{the probabilities} \begin{equation}
p_{\alpha,\beta}\equiv\langle\Psi|E_{\alpha}\otimes E_{\beta}|\Psi\rangle\textit{,}
\end{equation}
\textit{and let} $\|\vec{p}\|$ \textit{denote the Euclidean norm of the corresponding probability vector, i.e.}
\begin{equation}
\|\vec{p}\|=\sqrt{\sum_{\alpha=1}^{n}\sum_{\beta=1}^{n}p_{\alpha,\beta}^{2}}\textit{.}
\end{equation}
\textit{Then} $\{E_{1},\dots,E_{n}\}$ \textit{is a conical 2-design if and only if there exists a regular entanglement monotone whose restriction to pure states} $\text{Pur}\mathcal{Q}(\mathcal{H}_{d}\otimes\mathcal{H}_{d})$ \textit{is a function of} $\|\vec{p}\|$\textit{. In particular, the latter condition is satisfied if and only if the restriction of the concurrence is a function of} $\|\vec{p}\|$\textit{. Specifically}
\begin{equation}
\mathsf{C}\big(|\Psi\rangle\langle\Psi|\big)=2\sqrt{\frac{\hspace{0.35cm}k_{s}^{2}-\|\vec{p}\|^{2}}{k_{s}^{2}-k_{a}^{2}}}\textit{,}
\label{conSimplified}
\end{equation}
\textit{where} $k_{s}>k_{a}\geq 0$ \textit{are as in \cref{desCons}.}
\end{theorem}

\noindent We now build the following definition.
\begin{definition}\label{nluiDef}\textit{A} nontrivial local unitary invariant \textit{is a function} $f:\mathcal{S}\big(\mathcal{H}_{d}\otimes\mathcal{H}_{d}\big)\longrightarrow\mathbb{R}$ \textit{such that}:
\begin{enumerate}[(i)]
\item\label{lui1} $\exists\Psi\in\mathcal{S}\big(\mathcal{H}_{d}\otimes\mathcal{H}_{d}\big)$ \textit{and} $\exists\Phi\in\mathcal{S}\big(\mathcal{H}_{d}\otimes\mathcal{H}_{d}\big)$ \textit{such that} $f(\Psi)\neq f(\Phi)$.
\item\label{lui2} $\forall\Psi\in\mathcal{S}\big(\mathcal{H}_{d}\otimes\mathcal{H}_{d}\big)$ \textit{and} $\forall U,V\in\mathrm{U}\big(\mathcal{H}_{d}\big)$\textit{it holds that} $f\big((U\otimes V)\Psi\big)=f(\Psi)$\textit{.}
\end{enumerate}
\end{definition}
\noindent Condition (\ref{lui2}) in \cref{nluiDef} simply expresses that $f$ is invariant under the action of local unitary transformations, while condition (\ref{lui1}) demands such invariance is nontrivial, ruling out constant functions. \noindent We are now ready to state the following lemma.

\begin{lemma}\label{nluiLem}\textit{Let} $\{E_{1},\dots,E_{n}\}\subset\mathcal{E}(\mathcal{H}_{d})$ \textit{be a} \textsc{povm}\textit{. Let} $\Psi\in\mathcal{S}(\mathcal{H}_{d}\otimes\mathcal{H}_{d})$\textit{. Define} $\forall\alpha,\beta\in\{1,\dots,n\}$ \textit{the probabilities} \begin{equation}
p_{\alpha,\beta}\equiv\langle\Psi|E_{\alpha}\otimes E_{\beta}|\Psi\rangle\textit{,}
\end{equation}
\textit{and let} $\|\vec{p}\|$ \textit{denote the Euclidean norm of the corresponding probability vector, i.e.}
\begin{equation}
\|\vec{p}\|=\sqrt{\sum_{\alpha=1}^{n}\sum_{\beta=1}^{n}p_{\alpha,\beta}^{2}}\textit{.}
\end{equation}
\textit{Let $f$ denote the function} 
\begin{equation}
f:\mathcal{S}(\mathcal{H}_{d}\otimes\mathcal{H}_{d})\longrightarrow\mathbb{R}_{+}::\Psi\longmapsto \|\vec{p}\|\textit{.}
\end{equation} 
\textit{Then} $f$ \textit{is a nontrivial local unitary invariant if and only if} $\{E_{1},\dots,E_{n}\}$ \textit{is a conical 2-design.}
\end{lemma}

\noindent\textit{Proof of \cref{monDesThm}}. Our proof of sufficiency in \cref{nluiLem} establishes necessity in this case; in particular, it follows immediately from Eq.~\eqref{forSuf} and from \cref{conDef} and Eq.~\eqref{conSch} that
\begin{equation}
\mathsf{C}(|\Psi\rangle\langle\Psi)=2\sqrt{\frac{\hspace{0.34cm}k_s^2-\|\vec{p}\|^2}{k_s^2-k_a^2} }\text{,} 
\end{equation}
with the concurrence $\mathsf{C}$ being a regular entanglement monotone in light of \cref{conRegLem}.

\noindent We now prove sufficiency. Suppose $\mathsf{E}$ is a regular entanglement monotone such that  $\mathsf{E}=g(\|\vec{p}\|)$ on $\text{Pur}\mathcal{Q}(\mathcal{H}_{d}\otimes\mathcal{H}_{d})$, for some function $g$.  
We first show that $\|\vec{p}\|$ must be a nontrivial local unitary invariant. From \cref{nluiLem}, to show that $\{E_{1},\dots,E_{n}\}$ is a conical 2-design, it suffices to show that $\|\vec{p}\|$ is a nontrivial local unitary invariant. Our proof is by contradiction.  Suppose either $\|\vec{p}\|$ is trivial or else not a local unitary invariant.  The first possibility  would imply that $\mathsf{E}$ was constant on $\text{Pur}\mathcal{Q}(\mathcal{H}_{d}\otimes\mathcal{H}_{d})$, contradicting the fact that $\mathsf{E}$ is assumed regular.  The second possibility would imply that, for a fixed $\Psi\in\mathcal{S}(\mathcal{H}_{d}\otimes\mathcal{H}_{d})$, the map $(U,V) \in \mathrm{U}(\mathcal{H}_{d})\times\mathrm{U}(\mathcal{H}_{d})\longmapsto\sqrt{\langle \Psi | \otimes \langle \Psi | \; W_{23} (N^U\otimes N^V) W_{23} \; | \Psi\rangle \otimes |\Psi \rangle}\in \mathbb{R}$ took at least two distinct values, say $a< b$, where $W_{2,3}$ is the unitary that swaps the second and third tensor factors of $\mathcal{H}_{d}\otimes\mathcal{H}_{d}\otimes\mathcal{H}_{d}\otimes\mathcal{H}_{d}$, and where for all $U\in\text{U}(\mathcal{H}_{d})$ we define $N^{U}=\sum_{\alpha=1}^{n}U^{*}E_{\alpha}U\otimes U^{*}E_{\alpha}U$ so that one simply has $\|\vec{p}\|^{2}=\langle\Psi|\otimes\langle\Psi| W_{2,3}(N^{U}\otimes N^{V})W_{2,3}|\Psi\rangle\otimes|\Psi\rangle$. The continuity of the map together with the connectedness of the group $\mathrm{U}(\mathcal{H}_{d})\times \mathrm{U}(\mathcal{H}_{d})$ would then imply that it took every value in the interval $[a,b]$.  On the other hand, the restriction of $\mathsf{E}$ to $\text{Pur}\mathcal{Q}(\mathcal{H}_{d}\otimes\mathcal{H}_{d})$ is a local unitary invariant in light of \cref{vidalThm}.  So $g$ would have to take the same constant value, call it $x$, for every element of $[a,b]$.  This would mean that the set $\{\Psi\in \mathcal{S}(\mathcal{H}_{d}\otimes\mathcal{H}_{d}) \colon \mathsf{E}(|\Psi\rangle\langle\Psi|)=x\})$ contained the non-empty open set
$\{\Psi\in \mathcal{S}(\mathcal{H}_{d}\otimes\mathcal{H}_{d})\;\boldsymbol{|}\; a < \| \vec{p} \| < b \}$, again contradicting the fact that $\mathsf{E}$ is assumed regular.
\begin{flushright}
$\qed$
\end{flushright}
\newpage
\noindent \cref{monDesThm} establishes that a \textsc{povm} is a conical 2-design if and only if there exists what we call a regular entanglement monotone whose restriction to the pure states is a function of the norm of the probability vector. Furthermore, given that $\mathsf{C}(|\Psi\rangle\langle\Psi|)^2$ is necessarily quadratic in the probabilities, it seems fair to say that Eq.~\eqref{conSimplified} is about the simplest expression conceivable.  A simple description of entanglement in terms of probabilities is important in those theoretical approaches which seek to formulate quantum mechanics in purely probabilistic language, for instance the approaches in QBism \cite{Fuchs2010}\cite{Fuchs2013} and GPTs (see for example \cite{Barnum2012}.)  The simplicity of our result---the fact that conical 2-designs are naturally adapted to the description of  entanglement---means that  it is likely to be important for other reasons also. It may be possible to strengthen theorem \cref{monDesThm}, so that it states that $\{E_{1},\dots,E_{n}\}\subset\mathcal{E}(\mathcal{H}_{d})$ is a conical 2-design if and only if there is any entanglement monotone at all (\textit{i.e}.\ not necessarily a regular one) whose restriction to pure states is a function of $\|\vec{p}\|$.  However, we have not been able to prove it.

\noindent Some preparatory work is required \cref{nluiLem}. We begin with the following proposition, which will be central in the proof of \cref{nluiLem}. Henceforth, $C_{0}([0,2\pi]^{d})$ denotes the complex vector space of continuous $\mathbb{C}$-valued functions on the Cartesian product $[0,2\pi]^{d}$. Furthermore, we shall denote vectors in the spaces $\mathbb{Z}^{d}\subset\mathbb{R}^{d}\supset[0,2\pi]^{d}$ using bold font.

\begin{proposition}\label{applebyProp} \textit{Let} $\mathbf{n}_{1},\dots,\mathbf{n}_{m}\in\mathbb{Z}^{d}$ \textit{be distinct. Then the set of $\mathbb{C}$-valued functions} $\{e^{i\mathbf{n}_{j}\cdot\boldsymbol{\theta}}\}$ \textit{are linearly independent in} $C_{0}([0,2\pi]^{d})$.
\begin{proof}
Define the following square integrable functions in $\mathcal{L}_{2}([0,2\pi]^{d})$ with $\mathbf{n}\in\mathbb{Z}^{d}$
\begin{equation}
f_{\mathbf{n}}(\boldsymbol{\theta})=\frac{1}{(2\pi^{d/2})}e^{i\mathbf{n}\cdot\boldsymbol{\theta}}
\end{equation}
We have
\begin{equation}
\langle f_{\mathbf{n}'},f_{\mathbf{n}}\rangle=\frac{1}{2\pi^{d}}\int_{0}^{2\pi}\cdots\int_{0}^{2\pi}e^{i(n_{1}-n_{1}')\theta_{1}}\cdots e^{i(n_{d}-n_{d}')\theta_{d}}d\theta_{1}\cdots d\theta_{d}
\end{equation}
Set $\tilde{n}_{j}=n_{j}-n_{j}'\in\mathbb{Z}$. Thus,
\begin{equation}
\int_{0}^{2\pi}e^{i(n_{j}-n_{j}')\theta_{j}}d\theta_{j}=\int_{0}^{2\pi}e^{i\tilde{n}_{j}\theta_{j}}d\theta_{j}=\int_{0}^{2\pi}\big(\text{cos}(\tilde{n}_{j}\theta_{j})+i\text{sin}(\tilde{n}_{j}\theta_{j})\big)d\theta_{j}=2\pi\delta_{n_{j},n_{j}'}
\end{equation}
Consequently,
\begin{equation}
\langle f_{\mathbf{n}'},f_{\mathbf{n}}\rangle=\delta_{\mathbf{n},\mathbf{n}'}
\end{equation}
So $f_{\mathbf{n}_{j}}$ are an orthonormal set in the complex inner product space $\mathcal{L}_{2}([0,2\pi]^{d})$. Thus, any finite subset with elements $f_{\mathbf{n}_{1}},\dots,f_{\mathbf{n}_{m}}$ (with $\mathbf{n}_{1},\dots,\mathbf{n}_{m}$ distinct in $\mathbb{Z}^{d}$) is a linearly independent subset in $\mathcal{L}_{2}([0,2\pi]^{d})$. Turning our attention to the complex vector space $C_{0}([0,2\pi]^{d})$, suppose that $a_{1},\dots,a_{m}\in\mathbb{C}$ are such that
\begin{equation}
\forall\boldsymbol{\theta}\in[0,2\pi]^{d}:a_{1}f_{\mathbf{n}_{1}}(\boldsymbol{\theta})+\dots+a_{m}f_{\mathbf{n}_{m}}(\boldsymbol{\theta})=0
\end{equation}
Then $a_{1},\dots,a_{m}$ are such that
\begin{equation}
a_{1}f_{\mathbf{n}_{1}}+\dots+a_{m}f_{\mathbf{n}_{m}}=0
\end{equation}
as an element in $\mathcal{L}_{2}([0,2\pi]^{d})$. By the contrapositive, linear independence in $\mathcal{L}_{2}([0,2\pi]^{d})$ implies the desired result.
\end{proof}
\end{proposition}

\noindent We are now ready for our proof \cref{nluiLem}, which runs for the next eight pages.

\noindent\textit{Proof of \cref{nluiLem}}. We first prove sufficiency. Consider arbitrary $\Psi\in\mathcal{H}_{d}$, with Schmidt decomposition
\begin{equation}
\Psi=\sum_{r=1}^{d}\big(e_{r}\otimes f_{r}\big)\lambda_{r}\text{.}
\end{equation}
as in \cref{schmidtDecomp}. Let $\{E_{1},\dots,E_{n}\}$ be a conical 2-design as in \cref{c2dDef}. Then
\begin{eqnarray}
\sum_{\alpha,\beta=1}^{n}p_{\alpha,\beta}^{2}&=&\sum_{\alpha,\beta=1}^{n}\left(\mathrm{Tr}\Big(\big(E_{\alpha}\otimes E_{\beta}\big)\sum_{r,s=1}^{d}\big(|e_{r}\rangle\langle e_{s}|\otimes|f_{r}\rangle\langle f_{s}|\big)\lambda_{r}\lambda_{s}\Big)\right)^{2}\nonumber\\
&=&\sum_{\alpha,\beta=1}^{n}\sum_{r,s,t,v=1}^{d}\lambda_{r}\lambda_{s}\lambda_{t}\lambda_{v}\langle e_{s}|E_{\alpha}|e_{r}\rangle\langle e_{v}|E_{\alpha}|e_{t}\rangle\langle f_{s}|E_{\beta}|f_{r}\rangle\langle f_{v}|E_{\beta}|f_{t}\rangle\nonumber\\
&=&\sum_{\alpha,\beta=1}^{n}\sum_{r,s,t,v=1}^{d}\Big(\big(\langle e_{s}|\otimes\langle e_{v}|\big)\big(E_{\alpha}\otimes E_{\alpha}\big)\big(|e_{r}\rangle\otimes|e_{t}\rangle\big)\Big)\Big(\big(\langle f_{s}|\otimes\langle f_{v}|\big)\big(E_{\beta}\otimes E_{\beta}\big)\big(|f_{r}\rangle\otimes|f_{t}\rangle\big)\Big)\nonumber\\
&=&\sum_{r,s,t,v=1}^{d}\lambda_{r}\lambda_{s}\lambda_{t}\lambda_{v}\left(\left(\frac{k_{s}+k_{a}}{2}\right)\delta_{rs}\delta_{tv}+\left(\frac{k_{s}-k_{a}}{2}\right)\delta_{rv}\delta_{st}\right)^{2}
\nonumber\\[0.2cm]
&=&\frac{1}{2}\left(k_{s}^{2}+k_{a}^{2}\right)+\frac{1}{2}\left(k_{s}^{2}-k_{a}^{2}\right)\sum_{r=1}^{d}\lambda_{r}^{4}\text{.}
\end{eqnarray}
Thus
\begin{equation}
\|\vec{p}\|\equiv f(\Psi)=\sqrt{\frac{1}{2}\left(k_{s}^{2}+k_{a}^{2}\right)+\frac{1}{2}\left(k_{s}^{2}-k_{a}^{2}\right)\sum_{r=1}^{d}\lambda_{r}^{4}}\text{.}
\label{forSuf}
\end{equation}
Invariance of the Schmidt coefficients under local unitary transformations and the fact that $k_{s}>k_{a}$ implies a nontrivial dependence on the Schmidt coefficients, which completes our proof of sufficiency.

\noindent We now prove necessity. Let $f$ be a nontrivial local unitary invariant as in \cref{nluiDef}. Let us \textit{choose} $|\Psi\rangle\in\mathcal{S}(\mathcal{H}_{d}\otimes\mathcal{H}_{d})$ where the Schmidt bases for each tensor factor are identical to one another, \textit{i.e}.\
\begin{equation}
\Psi=\sum_{r=1}^{d}\big(e_{r} \otimes e_{r}\big)\lambda_{r}
\end{equation}
Then
\begin{eqnarray}
\|\vec{p}_{\Psi}\|^{2}&\equiv&\sum_{\alpha=1}^{n}\sum_{\beta=1}^{n}\langle\Psi|E_{\alpha}\otimes E_{\beta}|\Psi\rangle^{2}\nonumber\\
&=&\langle\Psi|\otimes\langle\Psi|\left(\sum_{\alpha=1}^{n}\sum_{\beta=1}^{n}(E_{\alpha}\otimes E_{\beta})\otimes(E_{\alpha}\otimes E_{\beta})\right)|\Psi\rangle\otimes|\Psi\rangle\nonumber\\
&=&\sum_{r,s,u,v=1}^{d}\sum_{\alpha=1}^{n}\sum_{\beta=1}^{n}\lambda_{r}\lambda_{s}\lambda_{u}\lambda_{v}\langle e_{r}|E_{\alpha}|e_{u}\rangle\langle e_{r}|E_{\beta}|e_{u}\rangle\langle e_{s}|E_{\alpha}|e_{v}\rangle\langle e_{s}|E_{\alpha}|e_{v}\rangle\nonumber\\
&=&\sum_{r,s,u,v=1}^{d}\lambda_{r}\lambda_{s}\lambda_{u}\lambda_{v}\left(\langle e_{r}|\otimes\langle e_{s}|\left(\sum_{\alpha=1}^{d}E_{\alpha}\otimes E_{\alpha}\right)|e_{u}\rangle\otimes |e_{v}\rangle\right)^{2}
\end{eqnarray}
Let $U,V\in\mathrm{U}(\mathcal{H}_{d})$. Let $|\Phi\rangle=(U\otimes V)|\Psi\rangle$. We compute
\begin{eqnarray}
\|\vec{p}_{\Phi}\|^{2}&\equiv&\sum_{\alpha=1}^{n}\sum_{\beta=1}^{n}\langle\Phi|E_{\alpha}\otimes E_{\beta}|\Phi\rangle^{2}\nonumber\\
&=&\sum_{\alpha=1}^{n}\sum_{\beta=1}^{n}\langle\Psi|\otimes\langle\Psi|\left((U^{*}E_{\alpha}U\otimes V^{*}E_{\beta}V)\otimes(U^{*}E_{\alpha}U\otimes V^{*}E_{\beta}V)\right)|\Psi\rangle\otimes|\Psi\rangle\nonumber\\
&=&\sum_{r,s,u,v=1}^{d}\lambda_{r}\lambda_{s}\lambda_{u}\lambda_{v}\sum_{\alpha=1}^{n}\sum_{\beta=1}^{n}\langle e_{r}|U^{*}E_{\alpha}U|e_{u}\rangle\langle e_{r}|V^{*}E_{\beta}V|e_{u}\rangle\langle e_{s}|U^{*}E_{\alpha}U|e_{v}\rangle\langle e_{s}|V^{*}E_{\beta}V|e_{v}\rangle\nonumber\\
&=&\sum_{r,s,u,v=1}^{d}\lambda_{r}\lambda_{s}\lambda_{u}\lambda_{v}\left(\langle e_{r}|\otimes\langle e_{s}|\left(\sum_{\alpha=1}^{d}U^{*}E_{\alpha}U\otimes U^{*}E_{\alpha}U\right)|e_{u}\rangle\otimes |e_{v}\rangle\right)\nonumber\\
&&\hspace{2.15cm}\times\left(\langle e_{r}|\otimes\langle e_{s}|\left(\sum_{\beta=1}^{d}V^{*}E_{\beta}V\otimes V^{*}E_{\beta}V\right)|e_{u}\rangle\otimes |e_{v}\rangle\right)
\end{eqnarray}
Denote positive semi-definite
\begin{equation}
X\equiv\sum_{\alpha=1}^{n}E_{\alpha}\otimes E_{\alpha}\in\mathcal{L}(\mathcal{H}_{d}\otimes\mathcal{H}_{d})_{+}
\end{equation}
Define $\forall U\in\mathrm{U}(\mathcal{H}_{d})$
\begin{eqnarray}
X^{U}_{rs,uv}&=&\langle e_{r}|\otimes\langle e_{s}|\left(\sum_{\alpha=1}^{d}U^{*}E_{\alpha}U\otimes U^{*} E_{\alpha}U\right)|e_{u}\rangle\otimes|e_{v}\rangle\\
&=&\langle e_{r}|U^{*}\otimes\langle e_{s}|U^{*}\left(\sum_{\alpha=1}^{d}E_{\alpha}\otimes E_{\alpha}\right)U|e_{u}\rangle\otimes U|e_{v}\rangle\text{.}
\end{eqnarray}
Then, by our standing assumption: 
\begin{equation}
\|\vec{p}_{\Psi}\|=\|\vec{p}_{\Phi}\|\text{;}
\end{equation} 
hence we have $\forall\vec{\lambda}\in\mathbb{R}_{+}^{d}$ with $\|\vec{\lambda}\|=1$ and for $\forall U,V\in\mathrm{U}(\mathcal{H}_{d})$ that
\begin{equation}
\sum_{r,s,u,v=1}^{d}\lambda_{r}\lambda_{s}\lambda_{u}\lambda_{v}X_{rs,uv}^{\mathds{1}}X_{rs,uv}^{\mathds{1}}=\sum_{r,s,u,v=1}^{d}\lambda_{r}\lambda_{s}\lambda_{u}\lambda_{v}X_{rs,uv}^{U}X_{rs,uv}^{V}\text{.}
\end{equation}
Let $P\in\mathrm{U}(\mathcal{H}_{d}):|e_{r}\rangle\longmapsto e^{i\theta_{r}}|e_{r}\rangle$ for $\theta_{r}\in[0,2\pi]$. Then $VP$ is another unitary and we have
\begin{equation}
\sum_{r,s,u,v=1}^{d}\lambda_{r}\lambda_{s}\lambda_{u}\lambda_{v}X_{rs,uv}^{\mathds{1}}X_{rs,uv}^{\mathds{1}}=\sum_{r,s,u,v=1}^{d}\lambda_{r}\lambda_{s}\lambda_{u}\lambda_{v}e^{i(\theta_{u}+\theta_{v}-\theta_{r}-\theta_{s})}X_{rs,uv}^{U}X_{rs,uv}^{V}
\end{equation}
Consider a fixed ordered set $\{r;s;u;v\}\subset\{1,\dots,d\}^{4}$. Let $\{r;s;u;v\}\longmapsto\{r_{p};s_{p};u_{p};v_{p}\}$ represent the action of a permutation $p$ on the fixed ordered set $\{r;s;u;v\}$\footnote{See \cref{permApp} for further explanation on our permutation notation.}. The monomials $\lambda_{r}\lambda_{s}\lambda_{u}\lambda_{v}$ constitute a basis for the complex vector space of homogeneous polynomials of degree four. Therefore, for each fixed ordered set $\{r;s;u;v\}$, summing over all unique permutations we have for all arbitrary $\boldsymbol{\theta}\in[0,2\pi]^{d}$
\begin{equation}
\sum_{p}X_{r_{p}s_{p},u_{p}v_{p}}^{\mathds{1}}X_{r_{p}s_{p},u_{p}v_{p}}^{\mathds{1}}=\sum_{p}e^{i(\theta_{u_{p}}+\theta_{v_{p}}-\theta_{r_{p}}-\theta_{s_{p}})}X_{r_{p}s_{p},u_{p}v_{p}}^{U}X_{r_{p}s_{p},u_{p}v_{p}}^{V}
\label{standing}
\end{equation}
We seek to prove 
\begin{equation}
\sum_{\alpha=1}^{n}E_{\alpha}\otimes E_{\alpha}=k_{s}\Pi_{\text{sym}}+k_{a}\Pi_{\text{asym}}\hspace{0.1cm}\text{with}\hspace{0.1cm}k_{s}>k_{a}\geq 0\text{.}
\end{equation}
Our proof will be via case-by-case analysis.\\[0.3cm]
There are five mutually exclusive and exhaustive cases of existence for the \textit{multiset} $\{r,s,u,v\}\subset\{1,\dots,d\}^{4}$
\begin{enumerate}
\item[]\textbf{Case 1}:  $\{r,s,u,v\}=\{r,r,r,r\}$ for some $r\in\{1,\dots,d\}$
\item[]\textbf{Case 2}:  $\{r,s,u,v\}=\{r,r,r,s\}$ for some distinct $r,s\in\{1,\dots,d\}$
\item[]\textbf{Case 3}: $\{r,s,u,v\}=\{r,r,s,s\}$ for some distinct $r,s\in\{1,\dots,d\}$
\item[]\textbf{Case 4}:  $\{r,s,u,v\}=\{r,r,s,u\}$ for some distinct $r,s,u\in\{1,\dots,d\}$
\item[]\textbf{Case 5}:  $\{r,s,u,v\}=\{r,s,u,v\}$ for some distinct $r,s,u,v\in\{1,\dots,d\}$
\end{enumerate}
\noindent Prior to our case-by-case analysis, we note two mutually exclusive and exhaustive possibilities.\\[0.3cm]
$\clubsuit$ \textbf{Possibility 1} There \textit{is not} a permutation $p$ such that $\{r,s\}=\{u,v\}$ as multisets.\\[0.3cm]
This possibility covers Case 2, Case 4, and Case 5. Let $S\in\mathrm{U}(\mathcal{H}_{d}):|e_{r}\rangle\longmapsto (-1)^{\delta_{rs}}|e_{r}\rangle$. Then VS is another unitary and we have from Eq.~\eqref{standing}
\begin{equation}
0=\sum_{p}e^{i(\theta_{u_{p}}+\theta_{v_{p}}-\theta_{r_{p}}-\theta_{s_{p}})}X_{r_{p}s_{p},u_{p}v_{p}}^{U}X_{r_{p}s_{p},u_{p}v_{p}}^{V}
\label{pos1}
\end{equation}
$\clubsuit$ \textbf{Possibility 2} There \textit{is} a permutation $p$ such that $\{r,s\}=\{u,v\}$ as multisets.\\[0.3cm]  
This possibility covers Case 1 and Case 3. Eq.~\eqref{standing} reads explicitly as
\begin{eqnarray}
0&=&e^{i2(\theta_{s}-\theta_{r})}X_{rr,ss}^{U}X_{rr,ss}^{V}+e^{i2(\theta_{r}-\theta_{s})}X_{ss,rr}^{U}X_{ss,rr}^{V} \nonumber\\
&+&\Big(X_{rs,rs}^{U}X_{rs,rs}^{V}+X_{sr,sr}^{U}X_{sr,sr}^{V}+X_{rs,sr}^{U}X_{rs,sr}^{V}+X_{sr,rs}^{U}X_{sr,rs}^{V}\Big)\nonumber \\
&-&\Big(X_{rr,ss}^{\mathds{1}}X_{rr,ss}^{\mathds{1}}+X_{ss,rr}^{\mathds{1}}X_{ss,rr}^{\mathds{1}}+X_{rs,rs}^{\mathds{1}}X_{rs,rs}^{\mathds{1}}+X_{sr,sr}^{\mathds{1}}X_{sr,sr}^{\mathds{1}}+X_{rs,sr}^{\mathds{1}}X_{rs,sr}^{\mathds{1}}+X_{sr,rs}^{\mathds{1}}X_{sr,rs}^{\mathds{1}}\Big)
\label{pos2}
\end{eqnarray}
We now proceed with case-by-case analysis.\\[0.3cm]
$\spadesuit$ \underline{\textbf{Case 1} $\{r,s,u,v\}=\{r,r,r,r\}$ for some $r\in\{1,\dots,d\}$}\\[0.3cm]
In this case Eq.~\eqref{pos2} collapses to 
\begin{equation}
X_{rr,rr}^{U}X_{rr,rr}^{V}=X_{rr,rr}^{\mathds{1}}X_{rr,rr}^{\mathds{1}}
\label{case1}
\end{equation}
$X$ is a positive semi-definite. Therefore $\forall U\in\mathrm{U}(\mathcal{H}_{d}):X_{rr,rr}^{U}\geq 0$. Select $U=V$ in Eq.~\eqref{case1}:
\begin{equation}
\big(X_{rr,rr}^{U}\big)^{2}=\big(X_{rr,rr}^{\mathds{1}}\big)^{2}\implies \forall U\in\mathrm{U}(\mathcal{H}_{d}):X_{rr,rr}^{U}=\gamma_{r}\label{gamEq}
\end{equation}
for some constant $\gamma_{r}\in\mathbb{R}_{+}$. However, we can permute the basis by taking $U$ to $UW_{r,s}$; so $\gamma_{r}$ is independent of $r$. Thus, we will write simply $\gamma$.\\[0.3cm]
$\spadesuit$ \underline{\textbf{Case 2} $\{r,s,u,v\}=\{r,r,r,s\}$ for some distinct $r,s\in\{1,\dots,d\}$}\\[0.3cm]
In this case Eq.~\eqref{pos1} reads explicitly as
\begin{eqnarray}
0&=&e^{i(\theta_{s}-\theta_{r})}X_{rr,rs}^{U}X_{rr,rs}^{V}+e^{i(\theta_{s}-\theta_{r})}X_{rr,sr}^{U}X_{rr,sr}^{V}+e^{i(\theta_{r}-\theta_{s})}X_{rs,rr}^{U}X_{rs,rr}^{V}+e^{i(\theta_{r}-\theta_{s})}X_{sr,rr}^{U}X_{sr,rr}^{V}\\
&=&e^{i(\theta_{s}-\theta_{r})}\Big(X_{rr,rs}^{U}X_{rr,rs}^{V}+X_{rr,sr}^{U}X_{rr,sr}^{V}\Big)+e^{i(\theta_{r}-\theta_{s})}\Big(X_{rs,rr}^{U}X_{rs,rr}^{v}+X_{sr,rr}^{U}X_{sr,rr}^{V}\Big)
\end{eqnarray}
Therefore from Lemma 1.2\footnote{See \cref{appDMLemma}for further details on applying \cref{applebyProp}.}
\begin{eqnarray}
0&=&X_{rr,rs}^{U}X_{rr,rs}^{V}+X_{rr,sr}^{U}X_{rr,sr}^{V}\label{case2a}\\
0&=&X_{rs,rr}^{U}X_{rs,rr}^{V}+X_{sr,rr}^{U}X_{sr,rr}^{V}\label{case2b}
\end{eqnarray}
Select $U=V$ in Eq.~\eqref{case2a} and Eq.~\eqref{case2b} 
\begin{eqnarray}
0&=&\big(X_{rr,rs}^{U}\big)^{2}+\big(X_{rr,sr}^{U}\big)^{2}\label{case2a1}\\
0&=&\big(X_{rs,rr}^{U}\big)^{2}+\big(X_{sr,rr}^{U}\big)^{2}\label{case2b1}
\end{eqnarray}
By construction
\begin{eqnarray}
X_{rr,rs}^{U}&=&\sum_{\alpha=1}^{n}\langle e_{r}|U^{*}E_{\alpha}U|e_{r}\rangle\langle e_{r}|U^{*}E_{\alpha}U|e_{s}\rangle=\sum_{\alpha=1}^{n}\langle e_{r}|U^{*}E_{\alpha}U|e_{s}\rangle\langle e_{r}|U^{*}E_{\alpha}U|e_{r}\rangle=X_{rr,sr}^{U}\\
X_{rs,rr}^{U}&=&\sum_{\alpha=1}^{n}\langle e_{r}|U^{*}E_{\alpha}U|e_{r}\rangle\langle e_{s}|U^{*}E_{\alpha}U|e_{r}\rangle=\sum_{\alpha=1}^{n}\langle e_{s}|U^{*}E_{\alpha}U|e_{r}\rangle\langle e_{r}|U^{*}E_{\alpha}U|e_{r}\rangle=X_{sr,rr}^{U}
\end{eqnarray}
Thus Eq.~\eqref{case2a1} and Eq.~\eqref{case2b1} yield
\begin{equation}
\forall U\in\mathrm{U}(\mathcal{H}_{d}):X_{rr,rs}^{U}=X_{rr,sr}^{U}=X_{rs,rr}^{U}=X_{sr,rr}^{U}=0
\end{equation}
\noindent$\spadesuit$ \underline{\textbf{Case 3} $\{r,s,u,v\}=\{r,r,s,s\}$ for some distinct $r,s\in\{1,\dots,d\}$}\\[0.3cm]
In this case Eq.~\eqref{pos2} in conjunction with Lemma 1.2 yields
\begin{eqnarray}
0&=&X_{rr,ss}^{U}X_{rr,ss}^{V}\label{case3a}\\
0&=&X_{ss,rr}^{U}X_{ss,rr}^{V}\label{case3b}
\end{eqnarray}
and
\begin{eqnarray}
&&X_{rs,rs}^{U}X_{rs,rs}^{V}+X_{sr,sr}^{U}X_{sr,sr}^{V}+X_{rs,sr}^{U}X_{rs,sr}^{V}+X_{sr,rs}^{U}X_{sr,rs}^{V}\nonumber\\
&=&X_{rr,ss}^{\mathds{1}}X_{rr,ss}^{\mathds{1}}+X_{ss,rr}^{\mathds{1}}X_{ss,rr}^{\mathds{1}}+X_{rs,rs}^{\mathds{1}}X_{rs,rs}^{\mathds{1}}+X_{sr,sr}^{\mathds{1}}X_{sr,sr}^{\mathds{1}}+X_{rs,sr}^{\mathds{1}}X_{rs,sr}^{\mathds{1}}+X_{sr,rs}^{\mathds{1}}X_{sr,rs}^{\mathds{1}}\label{case3c}
\end{eqnarray}
Select $U=V$ in Eq.~\eqref{case3a} and Eq.~\eqref{case3b} to obtain
\begin{eqnarray}
\forall U\in\mathrm{U}(\mathcal{H}_{d}):X^{U}_{rr,ss}=X^{U}_{ss,rr}=0
\end{eqnarray}
In particular, $X_{rr,ss}^{\mathds{1}}=X_{ss,rr}^{\mathds{1}}=0$. By construction
\begin{eqnarray}
X_{rs,rs}^{U}&=&\sum_{\alpha=1}^{n}\langle e_{r}|U^{*}E_{\alpha}U|e_{r}\rangle\langle e_{s}|U^{*}E_{\alpha}U|e_{s}\rangle=\sum_{\alpha=1}^{n}\langle e_{s}|U^{*}E_{\alpha}U|e_{s}\rangle\langle e_{r}|U^{*}E_{\alpha}U|e_{r}\rangle=X_{sr,sr}^{U}\\
X_{rs,sr}^{U}&=&\sum_{\alpha=1}^{n}\langle e_{r}|U^{*}E_{\alpha}U|e_{s}\rangle\langle e_{s}|U^{*}E_{\alpha}U|e_{r}\rangle=\sum_{\alpha=1}^{n}\langle e_{s}|U^{*}E_{\alpha}U|e_{r}\rangle\langle e_{r}|U^{*}E_{\alpha}U|e_{s}\rangle=X_{sr,rs}^{U}
\end{eqnarray}
Thus Eq.~\eqref{case3c} reads
\begin{equation}
X_{rs,rs}^{U}X_{rs,rs}^{V}+X_{rs,sr}^{U}X_{rs,sr}^{V}=X_{rs,rs}^{\mathds{1}}X_{rs,rs}^{\mathds{1}}+X_{rs,sr}^{\mathds{1}}X_{rs,sr}^{\mathds{1}}
\label{case3c2}
\end{equation}
Also\footnote{Skip this in the paper; jump straight to the real vectors} by construction with $E_{\alpha}^{U}=U^{*}E_{\alpha}U=(E_{\alpha}^{U})^{*}$ for arbitrary $U\in\mathrm{U}(\mathcal{H}_{d})$ and $\overline{(\cdot)}$ complex conjugation
\begin{eqnarray}
\overline{X^{U}_{sr,rs}}&=&\overline{\langle e_{s}|\otimes \langle e_{r}|\left(\sum_{\alpha=1}^{n}E_{\alpha}^{U}\otimes E_{\alpha}^{U}\right)|e_{r}\rangle\otimes|e_{s}\rangle}
=\sum_{\alpha=1}^{n}\overline{\langle e_{s}|E_{\alpha}^{U}|e_{r}\rangle}\overline{\langle e_{r}|E_{\alpha}^{U}|e_{s}\rangle}
=\sum_{\alpha=1}^{n}\langle E_{\alpha}^{U}e_{r}|e_{s}\rangle\langle E_{\alpha}^{U}e_{s}|e_{r}\rangle\nonumber\\
&=&\sum_{\alpha=1}^{n}\langle e_{r}|E_{\alpha}^{U}e_{s}\rangle\langle e_{s}|E_{\alpha}^{U}e_{r}\rangle
=\sum_{\alpha=1}^{n}\langle e_{s}|E_{\alpha}^{U}e_{r}\rangle\langle e_{r}|E_{\alpha}^{U}e_{s}\rangle
=\langle e_{s}|\otimes\langle e_{r}|\left(\sum_{\alpha=1}^{n}E_{\alpha}^{U}\otimes E_{\alpha}^{U}\right)|e_{r}\rangle\otimes|e_{s}\rangle\\[0.2cm]
&=&X_{sr,rs}^{U}
\end{eqnarray}
Introduce $\vec{u},\vec{v}\in\mathbb{R}^{2}$ with
\begin{equation}
\vec{u}=\begin{pmatrix} X^{U}_{rs,rs} \\ X^{U}_{rs,sr}\end{pmatrix}\hspace{1cm}\vec{v}=\begin{pmatrix} X^{V}_{rs,rs} \\ X^{V}_{rs,sr}\end{pmatrix}\
\end{equation}
Then from Eq.~\eqref{case3c2}
\begin{equation}
\|\vec{u}\|^{2}=\|\vec{v}\|^{2}=\vec{u}\cdot\vec{v}
\end{equation}
So $\|\vec{u}\|=\|\vec{v}\|$ and the familiar equality $\vec{u}\cdot\vec{v}=\|\vec{u}\|\|\vec{v}\|\text{cos}\theta$ yields for appropriate constants $c_{rs,rs}$ and $k_{rs,sr}$
\begin{eqnarray}
\forall U\in\mathrm{U}(\mathcal{H}_{d}):X^{U}_{rs,rs}=c_{rs,rs}\\
\forall U\in\mathrm{U}(\mathcal{H}_{d}):X^{U}_{rs,sr}=k_{rs,sr}
\end{eqnarray}
Moreover, permuting the bases (this is necessary and well defined if $d>2$) via $U\longmapsto UW_{r,u}$ with 
\begin{equation}
W_{r,u}|e_{a}\rangle=\begin{cases} |e_{u}\rangle & a=r\\ |e_{r}\rangle & a=u \\ |e_{a}\rangle & \text{otherwise} \end{cases}
\end{equation}
we see that $c_{rs,rs}$ and $k_{rs,sr}$ are independent of $r$. A similar argument concludes that $c_{rs,rs}$ and $k_{rs,sr}$ are independent of $s$. Thus we denote these constants simply by $c$ and $k$. Hence
\begin{eqnarray}
\forall U\in\mathrm{U}(\mathcal{H}_{d}):&X^{U}_{rs,rs}=X_{sr,sr}^{U}=&c\in\mathbb{R}_{+}\label{cEq}\\
\forall U\in\mathrm{U}(\mathcal{H}_{d}):&X^{U}_{rs,sr}=X_{sr,rs}^{U}=&k\in\mathbb{R}\label{kEq}
\end{eqnarray}
%
$\spadesuit$ \underline{\textbf{Case 4} $\{r,s,u,v\}=\{r,r,s,u\}$ for some distinct $r,s,u\in\{1,\dots,d\}$}\\[0.3cm]
In this case Eq.~\eqref{pos1} reads explicitly as
\begin{eqnarray}
0&=&e^{i(\theta_{s}+\theta_{u}-2\theta_{r})}\big(X_{rr,su}^{U}X_{rr,su}^{V}+X_{rr,us}^{U}X_{rr,us}^{V}\big)+e^{i(2\theta_{r}-\theta_{s}-\theta_{u})}\big(X_{su,rr}^{U}X_{su,rr}^{V}+X_{us,rr}^{U}X_{us,rr}^{V}\big)\\
&+&e^{i(\theta_{u}-\theta_{s})}\big(X_{rs,ru}^{U}X_{rs,ru}^{V}+X_{sr,ur}^{U}X_{sr,ur}^{V}+X_{sr,ru}^{U}X_{sr,ru}^{V}+X_{rs,ur}^{U}X_{rs,ur}^{V})\\
&+&e^{i(\theta_{s}-\theta_{u})}\big(X_{ru,rs}^{U}X_{ru,rs}^{V}+X_{ur,sr}^{U}X_{ur,sr}^{V}+X_{ur,rs}^{U}X_{ur,rs}^{V}+X_{ru,sr}^{U}X_{ru,sr}^{V}\big)\\
\end{eqnarray}
Therefore from Lemma 1.2
\begin{eqnarray}
0&=&X_{rr,su}^{U}X_{rr,su}^{V}+X_{rr,us}^{U}X_{rr,us}^{V}\label{4a}\\
0&=&X_{su,rr}^{U}X_{su,rr}^{V}+X_{us,rr}^{U}X_{us,rr}^{V}\label{4b}\\
0&=&X_{rs,ru}^{U}X_{rs,ru}^{V}+X_{sr,ur}^{U}X_{sr,ur}^{V}+X_{sr,ru}^{U}X_{sr,ru}^{V}+X_{rs,ur}^{U}X_{rs,ur}^{V}\label{4c}\\
0&=&X_{ru,rs}^{U}X_{ru,rs}^{V}+X_{ur,sr}^{U}X_{ur,sr}^{V}+X_{ur,rs}^{U}X_{ur,rs}^{V}+X_{ru,sr}^{U}X_{ru,sr}^{V}\label{4d}
\end{eqnarray}
By construction $X_{rr,su}^{U}=X_{rr,us}^{U}$ and $X_{su,rr}^{U}=X_{us,rr}^{U}$ thus from Eq.~\eqref{4a} and Eq.~\eqref{4b} with $U=V$ we have
\begin{eqnarray}
\forall U\in\mathrm{U}(\mathcal{H}_{d}):X_{rr,su}^{U}=X_{rr,us}^{U}=X_{su,rr}^{U}=X_{us,rr}^{U}=0
\end{eqnarray}
By construction 
\begin{eqnarray}
X_{rs,ru}^{U}=X_{sr,ur}^{U}\hspace{0.2cm}\text{and}\hspace{0.2cm}X_{rs,ur}^{U}=X_{sr,ru}^{U}\\
X_{ru,rs}^{U}=X_{ur,sr}^{U}\hspace{0.2cm}\text{and}\hspace{0.2cm}X_{ru,sr}^{U}=X_{ur,rs}^{U}
\end{eqnarray}
Thus from Eq.~\eqref{4c} and Eq.~\eqref{4d}
\begin{eqnarray}
0&=&X_{rs,ru}^{U}X_{rs,ru}^{V}+X_{rs,ur}^{U}X_{rs,ur}^{V}\label{4c2}\\
0&=&X_{ru,rs}^{U}X_{ru,rs}^{V}+X_{ru,sr}^{U}X_{ru,sr}^{V}\label{4d2}
\end{eqnarray}
Observe $X_{rs,ru}^{U}=\overline{X_{ru,rs}^{U}}$. Also, observe $X_{rs,ur}=\overline{X_{ur,rs}}=\overline{X_{ru,sr}}$. Replace $V$ with $UW_{s,u}$ in Eq.~\eqref{4c2} to obtain the first two implications; the final implication following from the foregoing  observations
\begin{equation}
0=|X_{rs,ru}^{U}|^{2}+|X_{rs,ur}^{U}|^{2}\implies X_{rs,ru}^{U}=0 \text{ and } X_{rs,ur}^{U}=0
\implies X_{ru,rs}^{U}=0 \text{ and } X_{ru,sr}^{U}=0
\end{equation}
Therefore
\begin{eqnarray}
\forall U\in\mathrm{U}(\mathcal{H}_{d}):X_{rs,ru}^{U}=X_{sr,ur}^{U}=X_{sr,ru}^{U}=X_{rs,ur}^{U}=0\\
\forall U\in\mathrm{U}(\mathcal{H}_{d}):X_{ru,rs}^{U}=X_{ur,sr}^{U}=X_{ur,rs}^{U}=X_{ru,sr}^{U}=0
\end{eqnarray}
$\spadesuit$ \underline{\textbf{Case 5} $\{r,s,u,v\}=\{r,s,u,v\}$ for some distinct $r,s,u,v\in\{1,\dots,d\}$}\\[0.3cm]
In Case 5, Eq.~\eqref{pos1} reads explicitly as
\begin{eqnarray}
0&=&e^{i(\theta_{u}+\theta_{v}-\theta_{r}-\theta_{s})}\big(X_{rs,uv}^{U}X_{rs,uv}^{V}+X_{rs,vu}^{U}X_{rs,vu}^{V}+X_{sr,uv}^{U}X_{sr,uv}^{V}+X_{sr,vu}^{U}X_{sr,vu}^{V}\big)\\
&+&e^{i(\theta_{s}+\theta_{v}-\theta_{r}-\theta_{u})}\big(X_{ru,sv}^{U}X_{ru,sv}^{V}+X_{ru,vs}^{U}X_{ru,vs}^{V}+X_{ur,sv}^{U}X_{ur,sv}^{V}+X_{ur,vs}^{U}X_{ur,vs}^{V}\big)\\
&+&e^{i(\theta_{s}+\theta_{u}-\theta_{r}-\theta_{v})}\big(X_{rv,su}^{U}X_{rv,su}^{V}+X_{rv,us}^{U}X_{rv,us}^{V}+X_{vr,su}^{U}X_{vr,su}^{V}+X_{vr,us}^{U}X_{vr,us}^{V}\big)\\
&+&e^{i(\theta_{r}+\theta_{v}-\theta_{s}-\theta_{u})}\big(X_{su,rv}^{U}X_{su,rv}^{V}+X_{su,vr}^{U}X_{su,vr}^{V}+X_{us,rv}^{U}X_{us,rv}^{V}+X_{us,vr}^{U}X_{us,vr}^{V}\big)\\
&+&e^{i(\theta_{r}+\theta_{u}-\theta_{s}-\theta_{v})}\big(X_{sv,ru}^{U}X_{sv,ru}^{V}+X_{sv,ur}^{U}X_{sv,ur}^{V}+X_{vs,ru}^{U}X_{vs,ru}^{V}+X_{vs,ur}^{U}X_{vs,ur}^{V}\big)\\
&+&e^{i(\theta_{r}+\theta_{s}-\theta_{u}-\theta_{v})}\big(X_{uv,rs}^{U}X_{uv,rs}^{V}+X_{uv,sr}^{U}X_{uv,sr}^{V}+X_{vu,rs}^{U}X_{vu,rs}^{V}+X_{vu,sr}^{U}X_{vu,sr}^{V}\big)
\end{eqnarray}
Therefore from Lemma 1.2
\begin{eqnarray}
0&=&X_{rs,uv}^{U}X_{rs,uv}^{V}+X_{rs,vu}^{U}X_{rs,vu}^{V}+X_{sr,uv}^{U}X_{sr,uv}^{V}+X_{sr,vu}^{U}X_{sr,vu}^{V}\label{f1}\\
0&=&X_{ru,sv}^{U}X_{ru,sv}^{V}+X_{ru,vs}^{U}X_{ru,vs}^{V}+X_{ur,sv}^{U}X_{ur,sv}^{V}+X_{ur,vs}^{U}X_{ur,vs}^{V}\label{f2}\\
0&=&X_{rv,su}^{U}X_{rv,su}^{V}+X_{rv,us}^{U}X_{rv,us}^{V}+X_{vr,su}^{U}X_{vr,su}^{V}+X_{vr,us}^{U}X_{vr,us}^{V}\label{f3}\\
0&=&X_{su,rv}^{U}X_{su,rv}^{V}+X_{su,vr}^{U}X_{su,vr}^{V}+X_{us,rv}^{U}X_{us,rv}^{V}+X_{us,vr}^{U}X_{us,vr}^{V}\label{f4}\\
0&=&X_{sv,ru}^{U}X_{sv,ru}^{V}+X_{sv,ur}^{U}X_{sv,ur}^{V}+X_{vs,ru}^{U}X_{vs,ru}^{V}+X_{vs,ur}^{U}X_{vs,ur}^{V}\label{f5}\\
0&=&X_{uv,rs}^{U}X_{uv,rs}^{V}+X_{uv,sr}^{U}X_{uv,sr}^{V}+X_{vu,rs}^{U}X_{vu,rs}^{V}+X_{vu,sr}^{U}X_{vu,sr}^{V}\label{f6}
\end{eqnarray}
By construction
\begin{eqnarray}
X_{rs,uv}^{U}=X_{sr,vu}^{U}&\text{and}&X_{rs,vu}^{U}=X_{sr,uv}^{U}\\
X_{ru,sv}^{U}=X_{ur,vs}^{U}&\text{and}&X_{ru,vs}^{U}=X_{ur,sv}^{U}\\
X_{rv,su}^{U}=X_{vr,us}^{U}&\text{and}&X_{rv,us}^{U}=X_{vr,su}^{U}\\
X_{su,rv}^{U}=X_{us,vr}^{U}&\text{and}&X_{su,vr}^{U}=X_{us,rv}^{U}\\
X_{sv,ru}^{U}=X_{vs,ur}^{U}&\text{and}&X_{sv,ur}^{U}=X_{vs,ru}^{U}\\
X_{uv,rs}^{U}=X_{vu,sr}^{U}&\text{and}&X_{uv,sr}^{U}=X_{vu,rs}^{U}
\end{eqnarray}
Therefore Eq.~\eqref{f1} through Eq.~\eqref{f6} read
\begin{eqnarray}
0&=&X_{rs,uv}^{U}X_{rs,uv}^{V}+X_{rs,vu}^{U}X_{rs,vu}^{V}\label{f1b}\\
0&=&X_{ru,sv}^{U}X_{ru,sv}^{V}+X_{ru,vs}^{U}X_{ru,vs}^{V}\label{f2b}\\
0&=&X_{rv,su}^{U}X_{rv,su}^{V}+X_{rv,us}^{U}X_{rv,us}^{V}\label{f3b}\\
0&=&X_{su,rv}^{U}X_{su,rv}^{V}+X_{su,vr}^{U}X_{su,vr}^{V}\label{f4b}\\
0&=&X_{sv,ru}^{U}X_{sv,ru}^{V}+X_{sv,ur}^{U}X_{sv,ur}^{V}\label{f5b}\\
0&=&X_{uv,rs}^{U}X_{uv,rs}^{V}+X_{uv,sr}^{U}X_{uv,sr}^{V}\label{f6b}
\end{eqnarray}
Set $V=W_{u,v}U$ in Eq.~\eqref{f1b} to obtain
\begin{equation}
0=X_{rs,uv}^{U}X_{rs,vu}^{U}
\end{equation}
Thus $X_{rs,uv}^{U}$ [logical \textbf{or}] $X_{rs,vu}^{U}=0$. The logical \textbf{or} may be improved to an \textbf{and} by noting that $X_{rs,uv}^{U}=0$ for all $U\in\mathrm{U}(\mathcal{H}_{d})$ implies that $X_{rs,vu}^{U}=0$ as well by taking $U\longmapsto W_{v,u}U$, and \textit{visa versa}. All cases can be handled in this way and we conclude for all permutations $p$ of the fixed ordered set $\{r,s,u,v\}$
\begin{equation}
\forall U\in\mathrm{U}(\mathcal{H}_{d}):X_{r_{p}s_{p},u_{p}v_{p}}=0\text{.}
\end{equation}
This concludes our case analysis. We now relate the constant $\gamma$ in Eq.~\eqref{gamEq} to the constants $c$ and $k$ in Eqs.~\eqref{cEq} and \eqref{kEq}. Fix arbitrary $U\in\mathrm{U}(\mathcal{H}_{d})$. Select $V\in\mathrm{U}(\mathcal{H}_{d})$ such that $V=UW_{+12}$ where $W_{+12}|e_{1}\rangle=(|e_{1}\rangle+|e_{2}\rangle)/\sqrt{2}$. Then
\begin{eqnarray}
\gamma=X_{11,11}^{V}&=&\langle e_{1}|W_{+12}^{*}\otimes\langle e_{1}|W_{+12}^{*}\left(\sum_{\alpha=1}^{n}E_{\alpha}^{U}\otimes E_{\alpha}^{U}\right)W_{+12}|e_{1}\rangle\otimes W_{+12}|e_{1}\rangle\\
&=&\frac{1}{4}\sum_{r,s,u,v=1}^{2}X_{rs,uv}^{U}\\[0.3cm]
&=&\frac{1}{4}\Big(X_{11,11}^{U}+X_{11,12}^{U}+X_{11,21}^{U}+X_{11,22}^{U}+X_{12,11}^{U}+X_{12,12}^{U}+X_{12,21}^{U}+X_{12,22}^{U}\Big)\nonumber\\
&+&\frac{1}{4}\Big(X_{21,11}^{U}+X_{21,12}^{U}+X_{21,21}^{U}+X_{21,22}^{U}+X_{22,11}^{U}+X_{22,12}^{U}+X_{22,21}^{U}+X_{22,22}^{U}\Big)\\[0.4cm]
&=&\frac{1}{2}\big(\gamma+c+k\big)
\end{eqnarray}
Therefore $\gamma=c+k$ and we conclude
\begin{equation}
\forall U\in\mathrm{U}(\mathcal{H}_{d}):X^{U}_{rs,uv}=c\delta_{ru}\delta_{sv}+k\delta_{rv}\delta_{su}\text{.}
\end{equation}
In particular we select $U=\mathds{1}$ and to conclude
\begin{equation}
\sum_{\alpha=1}^{n}E_{\alpha}\otimes E_{\alpha}=k_{s}\Pi_{\text{sym}}+k_{a}\Pi_{\text{asym}}
\end{equation}
with $c=(k_{s}+k_{a})/2$ and $k=(k_{s}-k_{a})/2$. Demanding \textit{nontrivial} local unitary invariance, our proof of sufficiency then implies that $k_{s}>k_{a}\geq 0$\text{.}
\begin{flushright}
$\qed$
\end{flushright}
\newpage
\noindent Our proof of \cref{nluiLem} completes the proof of \cref{monDesThm}. 

\noindent Recently, some rather different and interesting connections have been found between projective 2-designs and entanglement \cite{Wiesniak2011}\cite{Spengler2012}\cite{Chen2014}\cite{Kalev2013}, as well as between \textsc{sim}s and \textsc{mum}s and entanglement \cite{Chen2015}\cite{Liu2015}\cite{Shen2015}. Having established our \cref{monDesThm}, we shall now move to the following section and compare it with the results in the \cite{Spengler2012}\cite{Chen2014}\cite{Chen2015}\cite{Liu2015}\cite{Shen2015}. Although the authors do not present it this way what is done in these papers is, in effect, to show that there is a natural way to construct entanglement witnesses out of \textsc{mum}s and \textsc{sim}s (as noted in \cite{Chruscinski2014}). We will generalize and develop this previous work presently.

\section{Witnesses and Designs}\label{witnesses}

\noindent We begin by reminding the reader of the following definition, and point to the related review papers \cite{Chruscinski2014} and \cite{Guhne2009}.

\begin{definition}\label{witDef} \cite{Terhal2000}\textit{An} entanglement witness \textit{is a self-adjoint linear endomorphism} $A\in\mathcal{L}_{\text{sa}}(\mathcal{H}_{d}\otimes\mathcal{H}_{d})$ \textit{such that for all separable quantum states} $\rho_{s}\in\text{Sep}\mathcal{Q}(\mathcal{H}_{d}\otimes\mathcal{H}_{d})$ \textit{one has that} $\mathrm{Tr}(A\rho_{s})\geq 0$\textit{, while there exists an entangled quantum state} $\rho_{e}$ \textit{such that} $\mathrm{Tr}(A\rho_{e})<0$ \textit{.}
\end{definition}

\noindent In order to compress our notation, let $\mathcal{S}\equiv\mathcal{S}\big(\mathcal{H}_{d}\otimes\mathcal{H}_{d}\big)$ denote the unit sphere in $\mathcal{H}_{d}\otimes\mathcal{H}_{d}$ and let $\mathcal{S}_{s}$ denote the subset of product state vectors $\psi\otimes\phi$; hence the subscript `$s$' is for separable. Given an arbitrary self-adjoint linear endomorphism $A\in\mathcal{L}_{\text{sa}}\big(\mathcal{H}\otimes \mathcal{H}\big)$ define
\begin{eqnarray}
s_A^{-} &=& \inf_{|\Psi\rangle \in \mathcal{S}_s} \bigl\{\langle \Psi|A|\Psi\rangle\bigr\}\\
s_A^{+} &=& \sup_{|\Psi\rangle \in \mathcal{S}_s} \bigl\{\langle \Psi|A|\Psi\rangle\bigr\}\\
e_A^{-} &=& \inf_{|\Psi\rangle \in \mathcal{S}} \bigl\{\langle \Psi|A|\Psi\rangle\bigr\}\\
e_A^{+} &=& \sup_{|\Psi\rangle \in \mathcal{S}} \bigl\{\langle \Psi|A|\Psi\rangle\bigr\}
\end{eqnarray}
Then $s_A^{+}I - A$ is an entanglement witness if and only if $e_A^{+}> s_A^{+}$, in which case we will say $A$ detects entanglement from above. Furthermore, $A-s_A^{-} I$ is an entanglement witness if and only if $e_A^{-}< s_A^{-}$, in which case we will say $A$ detects entanglement from below. If, on the other hand, $s_A^{\pm}= e_A^{\pm}$, then a measurement of $A$ cannot detect entanglement. In \cite{Spengler2012}\cite{Chen2014}\cite{Chen2015} the authors only consider detection from above, although detection from below is also possible, as we will see.

\noindent Let $\{E_{1},\dots,E_{n}\}$ be an arbitrary conical $2$-design and let $|e_j\rangle$ be some fixed orthonormal basis in $\mathcal{H}_{d}$.  Define
\begin{eqnarray}
N = \sum_{\alpha}E_{\alpha} \otimes E_{\alpha}\mathrm{,}\\ 
N^{\rm{PT}} = \sum_{\alpha} E_{\alpha} \otimes E^{\rm{T}}_{\alpha}\mathrm{.}
\end{eqnarray}
where $N^{\rm{PT}}$ (respectively $E^{\rm{T}}_{\alpha}$) is the partial transpose (respectively transpose) of $N$ (respectively $E_{\alpha}$) relative to the basis $|e_j\rangle$.  It follows from the definition of a conical 2-design (\cref{c2dDef}) and Eq.~\eqref{for23} that
\begin{eqnarray}
N = k_{+} \mathds{1}_{d}\otimes\mathds{1}_{d} + k_{-} \mathrm{W}\mathrm{,} &\hspace{1cm}  N^{\rm{PT}} &= k_{+}\mathds{1}_{d}\otimes\mathds{1}_{d} + d k_{-} |\Phi_{+}\rangle \langle \Phi_{+} | \mathrm{,}
\label{nWit}
\end{eqnarray}
where $k_{\pm} = (k_s \pm k_a)/2$, $\mathrm{W}$ is the unitary which swaps the two factors in $\mathcal{H}\otimes \mathcal{H}$, and $|\Phi_{+}\rangle$ is the maximally entangled ket $(1/\sqrt{d})\sum_{j} |e_j\rangle \otimes |e_j\rangle$.  One easily sees
\begin{eqnarray}
s_N^{-}  = k_{+}\mathrm{,}  &\hspace{1cm} s_N^{+} = k_{+} + k_{-}\mathrm{,}
\\
e_N^{-}  = k_{+}-k_{-}\mathrm{,} &\hspace{1cm}  e_N^{+} = k_{+}+k_{-}\mathrm{,}
\end{eqnarray}
where we have used $k_{s}>k_{a}\geq 0$. So $N$ can detect entanglement from below, but not from above.  On the other hand
\begin{eqnarray}
s_{N^{\rm{PT}}}^{-}  =k_{+} \mathrm{,} &\hspace{1.6cm} s_{N^{\rm{PT}}}^{+}  = k_{+}+k_{-}\mathrm{,}
\\
e_{N^{\rm{PT}}}^{-} = k_{+}\mathrm{,} &\hspace{1.6cm} e_{N^{\rm{PT}}}^{+}  = k_{+}+dk_{-}\mathrm{.}
\end{eqnarray}
So $N^{\rm{PT}}$ can detect entanglement from above, but not from below.  

\noindent Let us note that in \cite{Spengler2012}\cite{Chen2014}\cite{Chen2015}\cite{Liu2015} the authors calculate $s_{N}^{+}$ for \textsc{mum}s and \textsc{sim}s, but not $e_{N}^{+}$.  Consequently, they do not draw the conclusion that $N$ cannot detect entanglement from above.

\noindent In \cite{Spengler2012}\cite{Chen2014}\cite{Chen2015}\cite{Liu2015}, the authors also consider operators of the more general form
\begin{equation}
N_{\rm{gen}} = \sum_{\alpha} E_{\alpha}\otimes F_{\alpha}
\end{equation}
where $E_{\alpha}$, $F_{\alpha}$ are distinct \textsc{mum}s or \textsc{sim}s.  Calculating  $s_{N_{\rm{gen}}}^{\pm}$ and $e_{N_{\rm{gen}}}^{\pm}$ for arbitrary  pairs of conical designs is beyond our present scope.  In this connection let us note that the authors of \cite{Spengler2012}\cite{Chen2014}\cite{Chen2015}\cite{Liu2015} do not calculate them either (although \cite{Spengler2012}\cite{Chen2014}\cite{Chen2015} do calculate a bound for $s_{N_{\rm{gen}}}^{+}$) valid for pairs of \textsc{mum}s or \textsc{sim}s having the same  contraction parameter, extended in \cite{Liu2015} to pairs of \textsc{mum}s having different contraction parameters).  

\noindent Liu \emph{et al.} \cite{Liu2015}  consider the application of \textsc{mum}s to multipartite entanglement.  However, they do not show that their  bound is violated for any non-separable states.

\noindent Shen \emph{et al.}~\cite{Shen2015} consider detection criteria which are nonlinear in the density matrix: specifically, they are quadratic in the design probabilities. It is easily seen that their criteria generalize to the statements that, for any conical 2-design, and any separable density matrix $\rho$, the following quantities are both bounded above by the same constant 
\begin{eqnarray}
\Bigl| \mathrm{Tr}\bigl(N(\rho-\rho_1\otimes \rho_2)\bigr) \Bigr|\leq k_{-} \sqrt{\bigl(1-\mathrm{Tr}(\rho_1^2)\bigr)\bigl(1-\mathrm{Tr}(\rho_2^2)\bigr)}\label{quad1}\\
\Bigl| \mathrm{Tr}\bigl(N^{\rm{PT}}(\rho-\rho_1\otimes \rho_2)\bigr) \Bigr|\leq k_{-} \sqrt{\bigl(1-\mathrm{Tr}(\rho_1^2)\bigr)\bigl(1-\mathrm{Tr}(\rho_2^2)\bigr)}\label{quad2}\text{,}
\end{eqnarray} 
with $\rho_j$ is the reduced density matrix for the $j^{\rm{th}}$ subsystem. It is easily verified that every entangled state detected by the linear  criterion 
\begin{equation}
\mathrm{Tr}(\rho N^{\rm{PT}})> s_{N^{\rm{PT}}}^{+}\label{lin1}
\end{equation} 
is also detected by the corresponding quadratic criterion expressed in Eq.~\eqref{quad1}. On the other hand there exist states which are detected by the linear criterion 
\begin{equation}
\mathrm{Tr}(\rho N) < s_N^{-}\label{lin2}
\end{equation}
but which are not detected by the corresponding quadratic criterion expressed in Eq.~\eqref{quad2}. Consider, for instance, the entangled Werner state~\cite{Werner1989}
\begin{eqnarray}
\rho_{W} = \frac{2(1-p)}{d(d+1)}\Pi_{\mathrm{sym}} + \frac{2p}{d(d-1)} \Pi_{\mathrm{asym}}, &\hspace{1cm} \frac{1}{2} < p \le 1
\end{eqnarray}
The linear criterion involving $N$ in Eq.~\eqref{lin1} detects the entanglement for all values of $p$ whereas the quadratic criterion only detects it for $p> (d-1)/d$.  This does not conflict with Shen \emph{et al.}'s analysis because they do not consider the possibility of detection from below.

\noindent The witnesses corresponding to $N$ and $N^{\rm{PT}}$ in Eq.~\eqref{nWit} are
\begin{eqnarray}
N-s_N^{-}\mathds{1}_{d}\otimes\mathds{1}_{d} &=& k_{-} \mathrm{W}  
\\
s_{N^{\rm{PT}}}^{+} \mathds{1}_{d}\otimes\mathds{1}_{d}- N^{\rm{PT}} &=& k_{-}\bigl( I -d |\Phi{+}\rangle \langle \Phi_{+} |\bigr)
\end{eqnarray}
The fact that these are witnesses is, of course, well known (see, e.g., Example 3.1 in \cite{Chruscinski2014};  Eq.~(28) in \cite{Guhne2009}). The novelty of \cite{Spengler2012}\cite{Chen2014}\cite{Chen2015}\cite{Liu2015} is that they show that the witnesses have simple expressions in terms of the probabilities obtained from local measurements. The simplicity of the expressions means that they have an obvious theoretical interest. They also have obvious experimental interest. Indeed, obtaining the probabilities empirically amounts to performing full-state tomography\footnote{It is to be observed, however, that one may be able to truncate the design and still have an entanglement witness (as is shown  in Appendix A of Spengler \emph{et al.}~\cite{Spengler2012} for  \textsc{mub}s).  As Spengler \emph{et al.} note this may be experimentally useful, especially when $d$ is large.}; however, with expressions for well known entanglement monotones in given in terms of probabilities for measurement outcomes, one can calculate these witnesses directly from emperical data, and thus avoid errors introduced by the additional processing entailed in finding the best-fit quantum states.

\noindent To summarize, we have shown in the previous section that a well known entanglement monotone has a simple expression in terms of design probabilities while \cite{Spengler2012}\cite{Chen2014}\cite{Chen2015}\cite{Liu2015} have done the same for two well-known entanglement witnesses.  The crucial difference is that \cref{monDesThm} gives a condition which is both necessary and sufficient for a given \textsc{povm} to be a conical 2-design, whereas continuity means that the inequalities $s_N^{-} > e_N^{-} $, $s_{N^{\rm{PT}}}^{+} < e_{N^{\rm{PT}}}^{+}$ will remain valid even after the $\{E_{1},\dots,E_{n}\}$ have been significantly and randomly perturbed.  Consequently, the inequalities proved in \cite{Spengler2012}\cite{Chen2014}\cite{Chen2015}\cite{Liu2015} can detect entanglement for a wide variety of \textsc{povm}s which are not conical 2-designs.  In that sense the connection we exhibit, between designs and entanglement, is tighter than the one exhibited in \cite{Spengler2012}\cite{Chen2014}\cite{Chen2015}\cite{Liu2015}.

\section{Invariant States and Designs}\label{decomps}
It is to be observed that, up to normalization, the right-hand sides of Eq.~\eqref{con2Eq} and Eq.~\eqref{con3Eq} are, respectively, separable Werner states~\cite{Werner1989}, and  separable isotropic states~\cite{Horodecki1999}. This merits a little discussion.

\noindent A \textit{Werner state} is one of the form
\begin{equation}
\rho_W = \Pi_{\rm{sym}}k_{s}  + \Pi_{\rm{asym}}k_{a} 
\label{wernerDef}
\end{equation}
with
\begin{eqnarray} 
k_s &=&\frac{2(1-p)}{d(d+1)}\text{,}\\ 
k_a &=& \frac{2p}{d(d-1)}\text{,}
\end{eqnarray}
for some $p\in [0,1]$. The state is entangled if and only if $p \in (1/2,1]$. The entangled states are the ones of most interest since, in addition to Werner's original motivation, it can be shown \cite{Horodecki1999} that the existence of bound-entangled \textsc{npt} states is equivalent to the existence of bound-entangled Werner states. The existence of the latter is still an  open question, but there are indications \cite{DiVincenzo2000a}\cite{Dur2000} that the entanglement becomes bound as one approaches the cross-over point at $p=1/2$. As we remarked in the introduction conical designs can be used to provide simple decompositions of all Werner states, both separable and entangled (although it remains to be seen how interesting they are). However, we will here confine ourselves to the point, which is already apparent from the definition, that they provide simple convex decompositions of some of the separable states. In this connection let us observe that, although less interesting, the problem of decomposing a separable Werner state is not straightforward, and  has attracted some notice in the literature \cite{Azuma2006}\cite{Unanyan2007}\cite{Tsai2013}. Conical 2-designs cast additional light on the problem.

\noindent Let us define a \textit{symmetric convex decomposition} of a separable Werner state to be one of the form
\begin{equation}
\rho_W = \sum_{j=1}^m (\rho_j \otimes \rho_j)\lambda_j
\label{wernerDecompA}
\end{equation}
where $\rho_j \in \mathcal{Q}(\mathcal{H}_{d})$, $\lambda_j \in (0,1]$ and $\sum_j \lambda_j = 1$.  We will say that the decomposition is \textit{homogeneous} if $\lambda_j = 1/m$ for all $j$, and that it is \textit{pure} if the $\rho_j$ are all pure.  It follows from \cref{desCons} that $\rho_W$ does not have a symmetric convex decomposition if  $k_s < k_a$ or, equivalently, if $p > (d-1)/(2d)$.  If $p=(d-1)/(2d)$ then $\rho_W$ is  the maximally mixed state, so the existence of a symmetric convex decomposition is trivial.  If $p< (d-1)/(2d)$ then Eq.~(\ref{wernerDecompA}) is equivalent to the statement that the operators $A_j = \rho_j\sqrt{\lambda_j}$ are a conical 2-design.  

\noindent \cref{existThm} establishes that homogeneous conical 2-designs exist for all $d$ and all $\kappa \in (0,1]$.   We conclude that a separable Werner state has a symmetric convex decomposition if and only if $0 \le p \le (d-1)/(2d)$.  Furthermore, if $p$ is in this interval the decomposition can always be chosen to be homogeneous.  Finally, \cref{rank1Lem} that a  conical design is rank 1 if and only if $k_a =0$ (in which case it is essentially the same thing as complex 2-projective design.)  So $\rho_W$ has a pure symmetric convex decomposition if and only if $p=0$.  

\noindent We have thus shown that the interval $0\le p \le 1/2$ splits into two sub-intervals separated by the maximally mixed state at $p=(d-1)/2d$. States in the sub-interval $0 \le p \le (d-1)/(2d)$ do have  symmetric convex decompositions; states in the subinterval $ (d-1)/(2d) < p \le 1/2$ do not.
One motivation for studying separable Werner states is the hope that, by looking at the states immediately  below the cross-over at $p=1/2$, one may get some insight into the bound-entangled states conjectured to exist just above it.  From this point of view the most interesting feature of our discussion is the negative statement, that states immediately below the cross-over cannot be put into the simple form of Eq.~(\ref{wernerDecompA}).

\noindent In the case $p< (d-1)/(2d)$ we define an \textit{ideal} convex decomposition to be one which  is symmetric, homogeneous and such that  $m$ achieves its minimum value of $d^2$.  A homogeneous conical 2-design is a \textsc{povm} up to rescaling, so we can use \cref{minConPovm} to conclude that an ideal convex decomposition must be of the form
\begin{equation}
\rho_W = \sum_{j=1}^{d^2} E_j \otimes E_j
\end{equation}
where the $E_j$ constitute a \textsc{sim}.  In view of our discussion in \cref{hc2d} this gives us the following reformulation of the \textsc{sic}-existence problem:  A \textsc{sic} exists in dimension $d$ if and only if every Werner state with $0 \le p < (d-1)/2d$ has an ideal convex decomposition. 

\noindent Conical 2-designs can also be used to give simple decompositions of a subset of the isotropic states introduced in \cite{Horodecki1999}. An \textit{isotropic state} is one of the form
\begin{equation}
\rho_{I} =  \frac{1-F}{d^2-1}  I +\frac{d^2F -1}{d^2-1} |\Phi_{+}\rangle \langle \Phi_{+} |
\end{equation}
with $F\in [0,1]$ and $|\Phi_{+}\rangle$ the maximally entangled state defined in Eq.~\eqref{maxKet}. They are separable for $F\in [0,1/d]$ and entangled for $F \in (1/d, 1]$ (they are not, however, bound-entangled for any value of $F$).  
We define a \textit{symmetric convex decomposition} of an isotropic state to be one of the form
\begin{equation}
\rho_{I} = \sum_{j=1}^m (\rho_j \otimes \overline{\rho_j})\lambda_{j}
\label{isotropicDecompA}
\end{equation}
where $\rho_j \in \mathcal{Q}(\mathcal{H}_{d})$, $\lambda_j \in (0,1]$, and $\sum_j \lambda_j =1$. Symmetric convex decompositions of isotropic states are in bijective correspondence with symmetric convex decompositions of Werner states. In fact, let 
\begin{equation}
\Pi_{\rm{sym}}k_{s} + \Pi_{\rm{asym}}k_{a} = \sum_j (\rho_j \otimes \rho_j)\lambda_j 
\end{equation}
be a symmetric convex decomposition of a Werner state with $p$ in the interval $[0,(d-1)/(2d)]$.  Taking the partial transpose on both sides gives
\begin{equation}
k_{+} I + d k_{-} |\Phi_{+}\rangle \langle \Phi_{+} | =  \sum_j (\rho_j \otimes \overline{\rho_j})\lambda_j 
\end{equation}
where $k_{\pm} = (k_s \pm k_a)/2$.  The fact that $0 \le p \le (d-1)/(2d)$ means $1/(d(d+1)) \le k_{+} \le 1/d^2$.  So we obtain in this way a symmetric convex decomposition of every isotropic state with $1/d^2 \le F \le 1/d$.  Reversing the argument it can be seen that, if one had a symmetric convex decomposition of an isotropic state with $0 \le F < 1/d^2$, then taking the partial transpose would give a symmetric convex decomposition of a Werner state with $(d-1)/(2d) < p \le 1/2$---which we have shown to be impossible.

\noindent Similarly to the Werner case, we see that the interval $0 \le F \le 1/d$ corresponding to the separable states splits into two sub-intervals, situated either side of the maximally mixed state at $F=1/d^2$.   States in the sub-interval $1/d^2 \le F \le 1/d$ do  have  symmetric convex decompositions;  states  in the sub-interval $ 0 \le F < 1/d^2$ do not.  The difference with the Werner case is that it is now the states \emph{with} a symmetric convex decomposition which lie next to the set of entangled states.

\chapter{Conclusion (Part I)}
\label{conclusionPartI}

\epigraphhead[40]
	{
		\epigraph{``You gotta stop and think about it, to really get the pleasure about the complexity --- the inconceivable nature of nature.''}{---\textit{Richard P.\ Feynman}\\ Fun to Imagine (1983)}
	}

\noindent Quantum state space is almost inconceivably vast. When the underlying Hilbert space dimension $d$ is larger than two, the compact convex set of quantum states $\mathcal{Q}(\mathcal{H}_{d})$ is a highly intricate geometric body within the ambient real vector space of self-adjoint linear endomorphisms $\mathcal{L}_{\text{sa}}(\mathcal{H}_{d})$. The family of all such shapes $\{\mathcal{Q}(\mathcal{H}_{d})\;\boldsymbol{|} d\in\mathbb{N}_{\geq 2}\}$ encapsulates the essence of quantum theory in its entirety. Indeed, quantum effects constitute those regions below the identity in the homogeneous self-dual cones $\mathcal{L}_{\text{sa}}(\mathcal{H}_{d})_{+}$ generated by quantum state spaces; moreover, when $d$ is a composite number, $\mathcal{Q}(\mathcal{H}_{d})$ can itself be viewed as a set of quantum channels via the Choi-Jamio{\l}kowski isomorphism. The case of composite $d$ is also especially interesting in light of entanglement. If we are to fully grasp quantum theory, we must therefore understand these shapes. In all cases, $\mathcal{Q}(\mathcal{H}_{d})$ is $d^{2}-1$-dimensional. This quadratic scaling with the Hilbert space dimension means that the shape of quantum state space is well beyond the scope of everyday intuition, even in the simplest nontrivial case $d=3$. In \cref{partI} of this thesis, we have developed a deeper understanding of these mysterious bodies. A novel variety of symmetric subshapes has emerged.

\noindent In \cref{designsOnQuantumCones} we introduced a new class of geometric structures in quantum theory, conical designs, which are natural generalizations of  projective designs.  We showed that symmetric informationally complete quantum measurements (\textsc{sim}s) and mutually unbiased quantum measurements (\textsc{mum}s) are special cases, as are weighted projective designs (up to rescaling).  We began by establishing the basic properties of conical 2-designs. In particular we gave five equivalent conditions for a subset of a quantum cone to be a conical 2-design (\cref{desCons}). We then turned to the special case of homogeneous conical 2-designs, and analyzed their Bloch geometry.  In the Bloch body picture  \textsc{sim}s and full sets of \textsc{mum}s form simple, highly symmetric polytopes (a single regular simplex in the case of \textsc{sim}s; the convex hull of a set of orthogonal regular simplices in the case of \textsc{mum}s).  We showed that the same is true of an arbitrary homogeneous conical design. Moreover, we derived necessary and sufficient conditions for a given polytope to be such a design (\cref{blochPoly} and \cref{existThm}). We also showed how the problem of constructing a homogeneous \emph{two}-design in a complex vector space reduces to the problem of constructing a spherical \emph{one}-design in a higher dimensional real vector space (\cref{liftThm}).

\noindent In \cref{entanglementConicalDesigns} we showed that conical designs are deeply implicated in the description of entanglement. We showed that a \textsc{povm} is a conical 2-design if and only if there is a regular entanglement monotone which is a function of $\|\vec{p}\|$.  We went on to extend the results in \cite{Spengler2012}\cite{Chen2014}\cite{Chen2015}\cite{Liu2015}\cite{Shen2015}, and to compare them with our \cref{monDesThm}.  In particular we showed that there is a natural way to construct entanglement witnesses out of an arbitrary conical design.  However, the connection between witnesses and designs is less tight than the one between monotones and designs, in the sense explained in the last section.  Our work naturally suggests the question, whether there are similar connections between multipartite entanglement and conical $t$-designs with $t>2$. 

\noindent There are other questions which might be interesting to investigate.  Firstly, there  is our suggestion in \cref{inSearchOf}, that the results there proved  could be used to search systematically for new \emph{projective} designs.  Secondly, all known examples of \textsc{sic}s and full sets of \textsc{mub}s  have important group covariance properties \cite{Dang2015}.  One would like to know how far this holds true in the more general setting of homogeneous conical designs.  Thirdly, one would like to extend the analysis to conical $t$-designs with $t>2$ via Schur-Weyl duality.  Fourthly, it is to be observed that the full class of conical 2-designs is itself a convex set.  It might be interesting to  explore the geometry of that set.  For instance, one might try to characterize the extreme points.  Finally,  it would be interesting to investigate conical designs in the larger context of general probabilistic theories~\cite{Barnum2012}. 
\bookmarksetup{startatroot}
\chapter{Interlude}
\label{interlude}
\epigraphhead[40]
	{
		\epigraph{``Beyond the horizon of the place we lived when we were young\\In a world of magnets and miracles\\Our thoughts strayed constantly and without boundary\\The ringing of the division bell had begun.''}{---\textit{Pink Floyd}\\ High Hopes (1994)}
	}
	
\noindent On the third day of December 1919, an article with the headline ``Einstein Expounds His New Theory'' was published in the New York Times \cite{NYT1919}. Of the aforementioned article's subheadlines, one was particularly apt: ``Improves on Newton.'' Indeed, history provides a rich library of cases wherein firmly established scientific notions were revolutionized by deep new insights into nature. In some cases, the \textit{status quo} has been abandoned altogether: consider, for instance, the death toll of the Michelson-Morley experiment \cite{Michelson1887} for the luminiferous Ether, or the fall of geocentrism during the Copernican revolution \cite{Kuhn1957}. In other cases, pre-existing physics was enveloped within a more general framework for understanding nature; the old became a special case of the new. In fact, Newton's law of universal gravitation falls out of Einstein's general relativity in the limit of small gravitational potentials \cite{Landau1975}, and classical probability theory can be viewed as a subset of quantum theory by considering diagonal states and effects. These stark examples point to the fluidity of physics, which is subject to the flux of novel ideas and new observations.

\noindent Remarkably, and some might even say shockingly, a century of careful experiments has yet to produce a single observation standing in significant conflict with quantum theory --- we stand with no analogue to the ultraviolet catastrophe, nor the perihelion of Mercury. Could it be that we have reached the end of physics? The absurdity of this question is easily established. Obviously, to borrow Feynman's words recalled in \cref{prologue} \cite{Feynman1963} ``\dots we know that we do not know all the laws as yet''; hence, the enterprise of quantum gravity programs seeking to reconcile quantum theory and general relativity. Furthermore, the history of physics alone suggests that we ought to remain alert and prepared for an experiment that will shake quantum theory to its core. At the very least, it is virtually inconceivable that quantum theory shall escape the same fate of past physical theories that were enveloped by more general frameworks. A serious question, then, concerns exactly where and how nature will force us to depart from quantum theory.

\noindent Quantum theory is situated within a vast landscape of general probabilistic theories, which are also known as \textit{foil theories} \cite{Chiribella2016}. Present experiments do not point us in any particular direction as we seek to depart from the comfort of quantum theory. We are thus left with the power of thought, at least for now. A most beautiful example of the power of thought is given by Einstein's prediction \cite{Einstein1916} for Sir Arthur Eddington's famous observation \cite{Dyson1919} of the deflection of light by the Sun. Einstein's revolutionary prediction was of course born from a revolutionary physical theory. Why has general relativity been so successful? A partial answer is that Einstein founded his theory of general relativity on deep physical principles. Recently, there has been much interest in singling out deep physical principles for the foundation of quantum theory. The information-theoretic axiomatizations in \cite{Hardy2001}\cite{Hardy2011}\cite{Dakic2011}\cite{Chiribella2011}\cite{Masanes2011} that we pointed out in \cref{prologue} are representative of this program seeking to pinpoint quantum theory within the `forest of foil theories,' as it were. Relaxing various axioms in such derivations of quantum theory is a natural way to identify foil theories with a solid conceptional basis. In particular, relaxing axioms regarding composite physical systems in \cite{Hardy2001}\cite{Hardy2011}\cite{Dakic2011}\cite{Chiribella2011}\cite{Masanes2011} naturally suggests the foil theory known as \textit{real quantum theory}. Simply put, real quantum theory is the restriction of quantum theory wherein scalars are taken from the field $\mathbb{R}$ of real numbers, instead of the complex field $\mathbb{C}$. Real quantum theory does not enjoy tomographic locality, which to remind the reader means that local measurements on components of a composite system do not suffice to determine a unique global state in real quantum theory. Instead, real quantum theory is \textit{bilocally tomographic} \cite{Hardy-Wootters}: joint measurements on bipartite subsystems of a composite system suffice to determine a unique global state. A related fact is that, in real quantum theory, entanglement is not monogamous \cite{Wootters2012}. Real quantum theory does, however, retain some important features of quantum theory, such as the superposition principle, the existence of maximally entangled states, and the absence of Sorkin's third-order interference \cite{Sorkin1994}. Furthermore, the lattices of unit rank projectors in real quantum theory are as in quantum theory, \textit{i.e}.\ nondistributive and orthomodular \cite{birkhoff1936}.

\noindent The real field $\mathbb{R}$ and the complex field $\mathbb{C}$ have a well known cousin: the quaternions $\mathbb{H}$. Strictly speaking, the quaternions form a ring, because quaternionic multiplication is not commutative. Any quaternion $h\in\mathbb{H}$ can be written in terms of its \textit{constituents} $a_{0},a_{1},a_{2},a_{3}\in\mathbb{R}$ and the \textit{quaternionic basis elements} $\{1,i,j,k\}$ as $h=1a_{0}+ia_{1}+ja_{2}+ka_{3}$ where 
\begin{equation}
i^{2}=j^{2}=k^{2}=ijk=-1\text{.}
\end{equation} 
Naturally, one can imagine a foil theory formulated by allowing scalars from the complex field $\mathbb{C}$ in quantum theory to be replaced by quaternions. One has to be careful. There are severe problems if one na{\"i}vely allows for quaternionic scalars whilst retaining the balance of the formalism of quantum theory. For instance, even at the level of single systems, the usual Born rule breaks down in this na{\"i}ve formulation \cite{Graydon2013}. At the level of composite systems, one runs into serious trouble as well, because the Kronecker product of two self-adjoint quaternionic matrices need not be self-adjoint \cite{Graydon2011}. The former issue, that of the Born rule, can be completely rectified. The solution given by the author in \cite{Graydon2013} rests on the observation that self-adjoint quaternionic matrices form a Jordan algebra (see \cref{jAlgDef}); furthermore, within this context, the author's MSc thesis \cite{Graydon2011} details \textit{quaternionic quantum theory} at the level of single systems. 

\noindent The framework of Jordan algebras is central in recent derivations of quantum theory given by Barnum-M{\"u}ller-Ududec \cite{Barnum2014} and Wilce \cite{Wilce09}\cite{Wilce11}\cite{Wilce12}. Relaxing the fourth postulate expounded by Barnum-M{\"u}ller-Ududec yields, in particular, quaternionic quantum theory at the level of single systems. The derivations given by Wilce also yield, in particular, quaternionic quantum theory at the level of single systems. As with real quantum theory, we thus have solid conceptual principles pointing to quaternionic quantum theory within the forest of foil theories. Unlike real quantum theory, however, an explicit description of composite physical systems in quaternionic quantum theory has remained unknown. One major contribution of \cref{partII} is a complete resolution \cite{Barnum2015}\cite{Barnum2016b} of this so-called `tensor product problem.' In fact, we accomplish more: we unite all three quantum theories (real, complex, and quaternionic) within a common framework. This naturally opens a novel avenue of research pertaining to the characteristics of composites in quaternionic quantum theory in particular, as well as all `hybrid' composites. From one perspective, we view these particular results as providing a testbed for new physics founded on solid conceptual principles. For instance, perhaps one could render the vague idea of a `tomographic locality witness' precise, and propose an experimental test. From another perspective, we view these results as providing a rigorous foil theory to improve our understanding of quantum theory itself.

\noindent The \textit{Spekkens toy theory} introduced by Spekkens in \cite{Spekkens2007} is a foil theory that has played a central role in understanding the foundations of quantum theory in the light of quantum information \cite{Fuchs2001}. \textit{Many}\footnote{For instance \cite{Spekkens2007}: ``\dots the noncommutativity of measurements, interference, the multiplicity of convex decompositions of a mixed state, the impossibility of discriminating nonorthogonal states, the impossibility of a universal state inverter, the distinction between bi-partite and tri-partite entanglement, the monogamy of pure entanglement, no cloning, no broadcasting, remote steering, teleportation, dense coding, mutually unbiased bases \dots''} features of quantum theory are found to have analogues within the Spekkens toy theory, which is defined by an epistemic restriction: given a maximal state of knowledge concerning a physical system, there are an equal number of questions about the physical system in question that are answered and unanswered. The long list of features common to quantum theory and the Spekkens toy theory supports an epistemic interpretation of the quantum state. This particular way of thinking about quantum states is not new \textit{per se}; although, it has certainly been a cornerstone for quantum foundations in the past decade, wherein remarkable progress has been made towards understanding quantum theory as a theory of information. The Spekkens toy theory does not reproduce violations of Bell inequalities, nor the existence of a Kochen-Specker theorem. Consequently, this particular foil theory is not a suitable candidate to build on quantum theory to predict new physics. Of course, Spekkens never claimed such a purpose. What is important to observe is that the existence of his foil theory has shed tremendous light on quantum theory. Our foil theories based on Jordan algebras hope to serve a purpose along these lines as well.

\noindent We drew the reader's attention to category theory in \cref{prologue}. Within the framework of category theory (to be formally introduced in \cref{catPrelims}), Coecke and Edwards showed the Spekkens toy theory can be viewed as a category: $\mathbf{Spek}$ \cite{Coecke2011}. Furthermore, Coecke and Edwards proved that $\mathbf{Spek}$ is a proper subcategory of $\mathbf{FRel}$. The precise definition of these categories is not necessary for us to recall. What is crucial to point out, however, is that $\mathbf{FRel}$ and its subcategory $\mathbf{Spek}$ are dagger compact closed categories (see \cref{dagCom}), just like quantum theory \textit{aka} the category $\mathbf{CPM}(\mathbf{FdHilb})$. Once again, we have therefore met foil theories, \textit{i.e}.\ $\mathbf{FRel}$ and $\mathbf{Spek}$, that share features with quantum theory. Specifically these theories admit the same abstract compositional structure of quantum theory, \textit{i.e}.\ that of a dagger compact closed category. It has remained an open question as to whether there exist additional foil theories structured as dagger compact closed categories. A major result in \cref{partII} is our construction of the dagger compact closed category $\mathbf{InvQM}$ \textit{aka} our unification of real, complex, and quaternionic quantum theories.

\noindent We now draw the reader's attention to the following fact. In \textit{any} dagger compact closed category, there exists an isomorphism between the sets of morphisms \textit{aka} physical process between objects \textit{aka} physical systems $\mathcal{A}$ and $\mathcal{B}$ and bipartite states on the composite system $\mathcal{AB}$. A formal proof of this fact, accompanied by all of the mathematical definitions can be found in the PhD thesis of Duncan \cite{Duncan2006} (specifically see Duncan's Lemma 2.21.) This isomorphism is manifest in quantum theory as the Choi-Jamio{\l}kowski isomorphism. Therefore, in light of our \cref{cor: InvQM is dagger compact}, $\mathbf{InvQM}$ enjoys such an isomorphism. We thus establish that the usual Choi-Jamio{\l}kowski isomorphism in quantum theory carries over to our unified framework for real, complex, quaternionic quantum theories. Put otherwise, our unified theory enjoys channel-state duality. In \cite{Chiribella2011}, Chiribella-D'Ariano-Perinotti prove that channel-state duality is in fact a \textit{consequence} of their purification postulate for general probabilistic theories. In this sense, the dagger compact closed structure of $\mathbf{InvQM}$ is motivated by a single physical principle. Additionally, in \cite{Bartlett2012}, Bartlett-Rudolph-Spekkens show that Gaussian quantum mechanics (wherein channel-state duality holds) follows from Liouville mechanics with an epistemic restriction. Channel-state duality in general, therefore, is a feature pointed to by physical principles. 

\noindent We now move to \cref{partII}. The end result is \cref{cor: InvQM is dagger compact}. Along the way we shall meet various aspects of Jordan algebras and discover the behaviour of composite systems in our $\mathbb{R}$-$\mathbb{C}$-$\mathbb{H}$ post-quantum theory.
\part{Categorical Jordan Algebraic Post-Quantum Theories}
\label{partII}
\chapter{Setting the Stage (Part II)}
\label{introPartII}

\epigraphhead[40]
	{
		\epigraph{``The miracle of the appropriateness of the language of mathematics for the formulation of the laws of physics is a wonderful gift which we neither understand nor deserve.''}{---\textit{Eugene Wigner}\\ The Unreasonable Effectiveness of\\ Mathematics in the Natural Sciences (1960)}
	}

\noindent Einstein's principles for special relativity can be phrased in ordinary language, in an incredibly simple way: the speed of light is constant, and physics is too. Naturally, those who wish to seriously study the subject of special relativity must then be introduced to mathematical notions such as the Minkowski metric, the Lorentz transformations, the Poincar{\'e} group, and so on. Of course, mathematics plays an absolutely crucial role in rigorous formulations of theoretical physics. This is not to say, however, that the two are synonymous. Mathematics and theoretical physics instead enjoy a rich interplay: revolutionary physical postulates prompt the development of new mathematics, and the development of new mathematics furnishes frameworks for novel physical theories. For instance, category theory initially arose within the context of pure mathematics \cite{Eilenberg1945}, only to much later become a complete framework for the formulation of quantum theoretic information processing protocols \cite{Abramsky2004}. Reciprocally, Abramsky and Coecke's full theory of categorical quantum mechanics \cite{Abramsky2009} returned the novel notion of a dagger compact closed category\footnote{In \cite{Abramsky2009}, Abramsky and Coecke use term \textit{strongly compact closed categories}. We follow Selinger's convention \cite{Selinger2005}.} to pure mathematics. Transparently, our technical work in \cref{compositesEJA} and \cref{categoriesEJA} provides a concrete example of new mathematics inspired by quantum physical theory as formulated by Abramsky and Coecke. Accordingly, in this introductory chapter, we review prerequisite elements of category theory in \cref{catPrelims}. The composites and categories constructed in the sequels are based on Euclidean Jordan algebras, hence our review of Jordanic preliminaries in \cref{jordPrelims}. We close this chapter by recalling technical tools from the theory of universal representations of Euclidean Jordan algebras in \cref{cStarPrelims}.

\noindent Ultimately, for motivation of our technical work, we return to physics. The net result of our mathematics is a novel physical theory uniting quantum theory and post-quantum theories over the real number field $\mathbb{R}$ and quaternionic division ring $\mathbb{H}$. At the level of single systems, such theories were recently pinpointed by the information-theoretic axioms expounded by Barnum-M{\"u}ller-Ududec in \cite{Barnum2014}. For quaternionic quantum theories, however, the formulation of a theory involving physical composites has remained elusive ever since the first formulation by Finkelstein \textit{et al.} \cite{Finkelstein1962}. Our novel Jordanic categories resolve this issue. Our motivations run deeper than simply completing the $\mathbb{R}$-$\mathbb{C}$-$\mathbb{H}$ program. Indeed, the twin pillars of modern physics may arise from a deeper, common level of physical principles. Holding nonsignaling sacrosanct, our categorical Jordan algebraic post-quantum theories expand the perimeter encompassing quantum theory within the vast landscape of all general probabilistic theories \cite{Barnum2012}, establishing a new realm for the discovery of novel physical ideas. We now proceed with our technical program, holding these thoughts in mind.
\newpage
\section{Elements of Category Theory}
\label{catPrelims}
Category theory provides a rich foundation for modern mathematical science. The standard reference text is Mac\! Lane's \textit{Categories for the Working Mathematician} \cite{MacLane1997}. Coecke and Paquette's \textit{Categories for the Practising Physicist} \cite{Coecke2009} presents a thorough review of category theory in the context of quantum physics. 
There also exist very deep PhD theses in quantum information with extensive introductions to category theory, in particular \cite{Duncan2006}\cite{Paquette2008}\cite{Edwards2009}\cite{Lal2012}. In this section, we collect categorical prerequisites for the sequel.

\begin{definition} \textit{A} metagraph $\mathscr{C}$ \textit{consists of two classes:} $\text{ob}\big(\mathscr{C}\big)\ni A,B,C,\ldots$\textit{, and} $\text{hom}\big(\mathscr{C}\big)\ni f,g,h,\ldots$\textit{; together with two functions} $\text{dom}:\text{hom}\big(\mathscr{C}\big)\longrightarrow\text{ob}\big(\mathscr{C}\big)$\textit{ and} $\text{cod}:\text{hom}\big(\mathscr{C}\big)\longrightarrow\text{ob}\big(\mathscr{C}\big)$\textit{.}
\end{definition}

\noindent We remind the reader that the notion of a \textit{class} is \cite{Devlin1979} more general than the notion of a \textit{set} (for instance, the collection of all sets is not itself a set.) The elements of $\text{ob}(\mathscr{C})$ are referred to as \textit{objects}, and the elements of $\text{hom}(\mathscr{C})$ are referred to as \textit{morphisms}. One writes $\text{hom}(A,B)$ for the collection of all morphisms from $A$ to $B$. The framework of metagraphs provides a graphical calculus, wherein one depicts a morphism $f$ with objects $A=\text{dom}(f)$ and $B=\text{cod}(f)$ as follows
\begin{equation}
A\stackrel{f}{\longrightarrow}B
\end{equation} 
\begin{definition}\label{catDef} \textit{A} category \textit{consists of a metagraph} $\mathscr{C}$ \textit{together with a function} 
\begin{eqnarray}
\mathbbmss{id}:\text{ob}\big(\mathscr{C}\big)\longrightarrow\text{hom}\big(\mathscr{C}\big)::A\longmapsto 1_{A}\;\;\text{where}\;\;1_{A}:A\longrightarrow A\textit{,}
\end{eqnarray}
\textit{and a function named} composition\textit{, denoted $\circ$ and defined via}
\begin{equation} \circ:\Big\{(f,g)\in\text{hom}\big(\mathscr{C}\big)\times\text{hom}\big(\mathscr{C}\big)\;\boldsymbol{|}\;\text{dom}(g)=\text{cod}(f)\Big\}\longrightarrow\text{hom}\big(\mathscr{C}\big)\textit{,}
\end{equation}
\textit{where $\circ(f,g)$ is written $g\circ f:\text{dom}(f)\longrightarrow\text{cod}(g)$ and such that the following diagrams commute\footnote{A diagram is said to \textit{commute} when all directed paths from object $A$ to object $B$ yield equivalent action via composition.}:}
\begin{equation}\label{idAxiom}
\xymatrix{A\ar[r]^{f}\ar[rd]_{f}&B\ar[d]^{1_{B}}\ar[rd]^{g} & \\ & B \ar[r]_{g} & C}
\end{equation}
\begin{equation}\label{comAxiom}
\xymatrix{A\ar[r]^{f}\ar[rd]_{g\circ f}& B\ar[d]^{g}\\ & C}
\end{equation}
\begin{equation}\label{assocAxiom}
\xymatrixcolsep{1in}\xymatrix{A\ar[r]^{h\circ(g\circ f)=(h\circ g)\circ f}\ar[d]_{f}\ar[rd]|!{[d];[r]}\hole_(0.7){g\circ f} & D \\ B\ar[r]_{g}\ar[ru]^(0.7){h\circ g}  & C\ar[u]_{h}}
\end{equation}
\end{definition}
\noindent The commutative diagrams in Eq.~\eqref{idAxiom}, Eq.~\eqref{comAxiom}, and Eq.~\eqref{assocAxiom} impose, respectively: left-right neutrality of the \textit{identity morphisms}; \textit{compositional stucture}; and \textit{associativity}. For example, $\mathbf{Set}$, the category with objects all sets and morphisms all suitably defined functions thereof, enjoys these properties.
\newpage
\noindent A functor is, crudely speaking, a morphism of categories. Formally, one has the following.
\begin{definition}\label{funDef} \textit{Let} $\mathscr{C},\mathscr{D}$ \textit{be categories. A} covariant functor $\mathcal{F}:\mathscr{C}\longrightarrow\mathscr{D}$ \textit{is defined by two functions, both denoted $\mathcal{F}$ for notational simplicity, such that} $\forall A\in\text{ob}(\mathscr{C})$ \textit{and} $\forall f,g\in\text{hom}(\mathscr{C})$
\begin{eqnarray}
\mathcal{F}:\text{hom}\big(\mathscr{C}\big)\longrightarrow\text{hom}\big(\mathscr{D}\big)::f\longmapsto\mathcal{F}(f)\;\;\text{where}\;\;\mathcal{F}(f)::\mathcal{F}\big(\text{dom}(f)\big)\longmapsto \mathcal{F}\big(\text{cod}(f)\big)\text{,}\\
\mathcal{F}:\text{ob}\big(\mathscr{C}\big)\longrightarrow\text{ob}\big(\mathscr{D}\big)::A\longmapsto \mathcal{F}(A)\textit{,}
\end{eqnarray}
\textit{such that}
\begin{eqnarray}
\mathcal{F}(1_{A})=1_{\mathcal{F}(A)}\textit{,}\\
\mathcal{F}(g\circ f)=\mathcal{F}(g)\circ\mathcal{F}(f)\textit{.}
\end{eqnarray}
\textit{A} contravariant functor $\mathcal{G}:\mathscr{C}\longrightarrow\mathscr{D}$ \textit{is defined by two functions, both denoted $\mathcal{G}$ for notational simplicity, such that} $\forall C\in\text{ob}\mathscr{C}$ \textit{and} $\forall h,l\in\text{hom}\mathscr{C}$
\begin{eqnarray}
\mathcal{G}:\text{hom}\big(\mathscr{C}\big)\longrightarrow\text{hom}\big(\mathscr{D}\big)::h\longmapsto\mathcal{G}(h)\;\;\text{where}\;\;\mathcal{G}(h)::\mathcal{G}\big(\text{cod}(h)\big)\longmapsto\mathcal{G}\big(\text{dom}(h)\big)\text{,}\\
\mathcal{G}:\text{ob}\big(\mathscr{C}\big)\longrightarrow\text{ob}\big(\mathscr{D}\big)::C\longmapsto \mathcal{G}(C)\textit{,}
\end{eqnarray}
\textit{such that}
\begin{eqnarray}
\mathcal{G}(1_{C})=1_{\mathcal{G}(C)}\textit{,}\\
\mathcal{G}(h\circ l)=\mathcal{G}(l)\circ\mathcal{G}(h)\textit{.}
\end{eqnarray}
\end{definition}
\noindent Phrased more informally, a functor maps commutative diagrams in $\mathscr{C}$ to commutative diagrams in $\mathscr{D}$, with covariant ones preserving the directionality of arrows and contravariant ones reversing them. We shall be particularly interested in functors of product categories, which are defined as follows. 
\begin{definition}\textit{Let} $\mathscr{C}_{1},\mathscr{C}_{2}$ \textit{be categories. The} product category $\mathscr{C}_{1}\times\mathscr{C}_{2}$ \textit{is the category with}
\begin{eqnarray}
\text{ob}\big(\mathscr{C}_{1}\times\mathscr{C}_{2}\big)\equiv\text{ob}\big(\mathscr{C}_{1}\big)\times\text{ob}\big(\mathscr{C}_{2}\big)\\ \text{hom}\big(\mathscr{C}_{1}\times\mathscr{C}_{2}\big)\equiv\text{hom}\big(\mathscr{C}_{1}\big)\times\text{hom}\big(\mathscr{C}_{2}\big)\text{,}
\end{eqnarray}
\textit{wherein for all suitably composable morphisms}
\begin{eqnarray}
(f_{1},f_{2})\circ (g_{1},g_{2})\equiv(f_{1}\circ g_{1},f_{2}\circ g_{2})\textit{,}
\end{eqnarray}
\textit{and wherein} $\forall A_{1}\in\text{ob}\mathscr{C}_{1}$ \textit{and} $\forall A_{2}\in\text{ob}\mathscr{C}_{2}$
\begin{eqnarray}
1_{(A_{1},A_{2})}\equiv(1_{A_{1}},1_{A_{2}})\textit{.}
\end{eqnarray}
\end{definition}
\begin{definition}\label{biFunDef} \textit{A} bifunctor \textit{is a functor with domain a product category.}
\end{definition}

\noindent For example, the category $\mathbf{FdHilb}$ --- with objects all finite dimensional complex Hilbert spaces and morphisms all linear functions thereof --- is the range of the usual \textit{tensor product bifunctor}, that is $\otimes:\mathbf{FdHilb}\times\mathbf{FdHilb}\longrightarrow\mathbf{FdHilb}$, where the definition of $\otimes$ is as given in \cref{tensorDef}. The tensor product bifunctor enjoys many nice properties, which can be described in terms of morphisms on iterations of the functor itself. We shall require the notion of natural transformations to render such a discussion precise.

\newpage
\noindent A natural transformation is, crudely speaking, a morphism of functors. Formally, one has the following.

\begin{definition}\label{natDef} \textit{Let} $\mathcal{F},\mathcal{G}:\mathscr{C}\longrightarrow\mathscr{D}$ \textit{be covariant functors. A} natural transformation $\tau:\mathcal{F}\longrightarrow\mathcal{G}$
\textit{is}
\begin{equation}
\tau:\text{ob}\big(\mathscr{C}\big)\longrightarrow\text{hom}\big(\mathscr{D}\big)::A\longmapsto \tau_{A}:::\mathcal{F}(A)\longrightarrow \mathcal{G}(A)
\end{equation}
\textit{such that the following diagram commutes for all suitable} $\text{dom}(f)=A\stackrel{f}{\longrightarrow}B=\text{cod}(f)$ \textit{in} $\mathscr{C}$
\begin{equation}
\xymatrix{\mathcal{F}(A)\ar[r]^{\mathcal{F}(f)}\ar[d]_{\tau_{A}} &\mathcal{F}(B)\ar[d]^{\tau_{B}}\\ \mathcal{G}(A)\ar[r]_{\mathcal{G}(f)} & \mathcal{G}(B)}\textit{.}
\end{equation}
\textit{The functions} $\tau_{\mathrm{A}}$ \textit{and} $\tau_{\mathrm{B}}$ \textit{are called} components\textit{. A} natural isomorphism \textit{is a natural transformation such that}
\begin{equation}
\exists\tau^{-1}:\mathcal{G}\longrightarrow\mathcal{F}::\tau_{A}\circ\tau^{-1}_{A}=1_{\mathcal{F}(A)}\;\text{and}\;\tau^{-1}_{A}\circ\tau_{A}=1_{\mathcal{G}(A)}\textit{.}
\end{equation}
\end{definition}

\noindent The formal notions of category, functor, and natural transformation can be eloquently summarized in the following informal diagram due to Baez and Dolan \cite{Baez1995}:
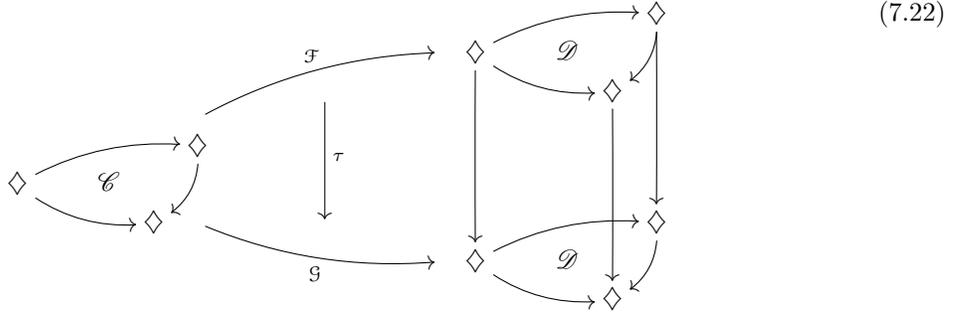
\begin{equation}
\xymatrixrowsep{0.01cm}
\xymatrixcolsep{0.1cm}
\xymatrix
{& & & & & & & & & & & & & & & & & & & & & \diamondsuit\ar@/^/[ldd]\ar@/^0.01pc/[dddddd]\\
 & & & & & & & & & & & & & & & &\diamondsuit\ar@/^/[rrrrru]\ar@/_/[rrrrd]\ar@/_0.01pc/[dddddd] & & & \mathscr{D} & &\\
& & & & & & & & & & \ar[dddd]^{\tau} & & & & & & & & & & \diamondsuit\ar@/^0.01pc/[dddddd]  &\\
& & & & & \ar@/^/[rrrrrrrrrruu]^{\mathcal{F}} & & &  & & & & & & & & & & & & &\\
& & & & & \diamondsuit\ar@/^/[ldd] & & & & & & & & & & & & & & & &\\
\diamondsuit\ar@/^/[rrrrru]\ar@/_/[rrrrd] & &  & \mathscr{C} & & & & & & & & & & & & & & & & & &\\
& & & & \diamondsuit & \ar@/_/[rrrrrrrrrrd]_{\mathcal{G}}& & & & & & & & & & & & & & & & \diamondsuit\ar@/^/[ldd]\\
& & & & & & & & & & & & & & & &\diamondsuit\ar@/^/[rrrrru]\ar@/_/[rrrrd] & & & \mathscr{D} & &\\
& & & & & & & & & & & & & & & & & & & & \diamondsuit  &
}
\end{equation}

\noindent Categories admit a natural physical interpretation, wherein objects correspond to physical systems and morphisms correspond to physical processes. This line of thinking has been explored extensively within the field of categorical quantum mechanics \cite{Abramsky2004}\cite{Abramsky2009}. One revolutionary insight derived in this vein is that the compositional structure of quantum theory is precisely that of a symmetric monoidal category \cite{Benabou1963}\cite{Benabou1964}; moreover one that is dagger compact closed \cite{Selinger2005}. We shall now formally define these notions.

\begin{definition}\label{dagCatDef} \textit{A} dagger category \textit{is a category} $\mathscr{C}$ \textit{equipped with a contravariant strictly involutive identity-on-objects endofunctor} $\dagger:\mathscr{C}\longrightarrow\mathscr{C}::f\longmapsto f^{\dagger}$\textit{, that is} $\forall A\in\text{ob}(\mathscr{C})$ \textit{and} $\forall f\in\text{hom}(\mathscr{C})$  
\begin{eqnarray}
\dagger(A)=A\text{,}\\
\dagger\big(\dagger(f)\big)=f.
\end{eqnarray}
\textit{A} unitary \textit{morphism} \textit{is} $f\in\text{hom}\mathscr{C}$ \textit{such that}
\begin{equation}
\dagger(f)\circ f=1_{\text{dom}(f)}=f\circ\dagger(f)\text{.}
\end{equation}
\textit{A} self-adjoint \textit{morphism} \textit{is} $f\in\text{hom}\mathscr{C}$ \textit{such that}
\begin{equation}
\dagger(f)=f\text{.}
\end{equation}
\end{definition}
\newpage
\noindent One way to meet the following definition is to think of $\Box$ as a ``tensor product'' and overlay a mental `$\otimes$'. 
\begin{definition}\label{smcDef} \textit{A} symmetric monoidal category is a septuple, $(\mathscr{C},\text{I},\Box,\alpha,\lambda,\rho,\sigma)$ \textit{where} $\mathscr{C}$ \textit{is a category as in \cref{catDef}}; $\text{I}\in\text{ob}\big(\mathscr{C}\big)$ \textit{is a distinguished} monoidal unit; $\Box:\mathscr{C}\times\mathscr{C}\rightarrow\mathscr{C}$ \textit{is a bifunctor as in \cref{biFunDef}}; and $\alpha,\lambda,\rho,\sigma$ \textit{are natural isomorphisms as in \cref{natDef} with components}
	\begin{eqnarray}
	 \alpha_{A,B,C}:\Box\big(A,\Box(B,C)\big)\longrightarrow \Box\big(\Box(A,B),C\big)\\
		\lambda_{A}:A\longrightarrow\Box(\text{I},A)\\
		\rho_{A}:A\longrightarrow\Box(A,\text{I})\\
		\sigma_{A,B}:\Box(A,B)\longrightarrow\Box(B,A)
\end{eqnarray}
\textit{such that the following diagrams commute}
\begin{equation}\label{assocCo}
\xymatrix{ &\Box\Big(\Box\big(A,B\big),\Box\;\big(C,D\big)\Big)\ar[rd]^{\alpha_{A,B,\Box\big(C,D\big)}}&\\
\Box\Big(\Box\big(\Box(A,B),C\big),D\Big)\ar[ru]^{\alpha_{\Box\big(A,B\big),C,D}}\ar[d]^{\Box\big({\alpha}_{A,B,C},1_{D}\big)}& & \Box\Big(A,\Box\big(B,(\Box(C,D)\big)\Big)\\
\Box\Big(\Box\big(A,\Box(B,C)\big),D\Big)\ar[rr]^{\alpha_{A,\Box(B,C),D}}& & \Box\Big(A,\Box\big(\Box(B,C),D\big)\Big)\ar[u]^{\Box\big(1_{A},\alpha_{B,C,D}\big)}}
\end{equation}
\begin{equation}\label{unitorCo}
\xymatrix{
\Big(\Box\big(A,\text{I}\big),B\Big)\ar[rr]_{\alpha_{A,\text{I},B}}& &\Box\Big(A,\Box\hspace{0.01in}\big(\text{I},B\big)\Big)\ar[ld]^{\Box\big(1_{A},\lambda_{B}^{-1}\big)}\\
&\Box\big(A,B\big)\ar[lu]^{\Box\big(\rho_{A},1_{B}\big)}&
}
\end{equation}
\begin{equation}\label{alpSigCo}\xymatrixrowsep{0.15in}\xymatrixcolsep{0.6in}
\xymatrix{
\Box\big(A,\Box(B,C)\big)\ar[r]^{\alpha_{A,B,C}}\ar[dd]^{\alpha_{A,B,C}}&\Box\big(\Box(A,B),C\big)\ar[r]^{\sigma_{\Box(A,B),C}}&\Box\big(C,\Box(A,B)\big)\ar[dd]^{\alpha_{C,A,B}}\\
& & \\
\Box\big(A,\Box(C,B)\big)\ar[r]_{\alpha^{-1}(A,C,B)} &\Box\big(\Box(A,C),B\big)\ar[r]_{\Box\big(\sigma_{A,C},1_{C}\big)} & \Box\big(\Box(C,A),B\big)
}
\end{equation}
\begin{equation}\label{simCo}
\xymatrix{
\Box(A,B)\ar[rd]^{1_{\Box(A,B)}}\ar[d]_{\sigma_{AB}}\ar[d] \\ \Box(B,A)\ar[r]_{\sigma_{BA}} & \Box(A,B)}
\end{equation}
\end{definition}
\begin{equation}\label{rhoSigmaLambdaCo}
\xymatrix{
& A\ar[d]^{\rho_{A}}\ar[ld]_{\lambda_{A}} \\ \Box(\text{I},A) & \Box(A,\text{I})\ar[l]^{\sigma_{A,\text{I}}}
}
\end{equation}
\newpage
\noindent Eqs.~\eqref{assocCo},\eqref{unitorCo},\eqref{alpSigCo},\eqref{simCo}, and \eqref{rhoSigmaLambdaCo} are the \textit{coherence conditions} for a symmetric monoidal category. The natural isomorphisms $\alpha$ and $\sigma$ are given special names: the \textit{associator} and the \textit{symmetor}, respectively (the latter is not conventional.) The bifunctor $\Box$ also receives a special name: the \textit{monoidal product}. A very deep result in category theory is that the coherence conditions guarantee that any diagram, with edges permuted instances of the monoidal product and edges permuted expansions of instances of the associator and the symmetor, commutes. This is the content of Mac\! Lane's celebrated coherence theorem \cite{MacLane1963}.

\noindent The category $\mathbf{FdHilb}$ equipped with the usual tensor product $\otimes$ is a symmetric monoidal category with monoidal unit $\mathbb{C}$ and the obvious natural isomorphisms. Moving one level higher in quantum abstraction, the category $\mathbf{CPM}(\mathbf{FdHilb})$ with objects again finite dimensional complex Hilbert spaces and morphisms completely positive linear homomorphisms thereupon is a symmetric monoidal category, where the monoidal product and monoidal unit are inherited from $\mathbf{FdHilb}$, and where the natural isomorphisms are again obvious in light of the underlying Hilbert space isomorphisms. On this view, quantum theory is, in a certain sense, a symmetric monoidal category. This precise mathematical statement expounds the compositional structure of quantum theory. The usual adjoint provides $\mathbf{CPM}(\mathbf{FdHilb})$ with the structure of a dagger symmetric monoidal category, whose definition we now promote to formal status. 

\begin{definition}\label{dagSymDef} \textit{A} dagger symmetric monoidal category \textit{is an octuple} $(\mathscr{C},\text{I},\Box,\dagger,\alpha,\lambda,\rho,\sigma)$\textit{, where the septuple} $(\mathscr{C},\text{I},\Box,\alpha,\lambda,\rho,\sigma)$ \textit{is a symmetric monoidal category as in \cref{smcDef}, with the natural isomorphisms} $\alpha,\lambda,\rho,\sigma$ \textit{unitary, and where the couple} $(\mathscr{C},\dagger)$ \textit{is a dagger category as in \cref{dagCatDef}.}
\end{definition}

\noindent We now push forward with a categorical notion introduced by Kelly and Laplaza \cite{Kelly1980}.

\begin{definition}\label{ccCatDef} \textit{A} compact closed category \textit{is a symmetric monoidal category} $(\mathscr{C},\text{I},\Box,\alpha,\lambda,\rho,\sigma)$ \textit{as in \cref{smcDef} where} $\forall A\in\text{ob}(\mathscr{C})$ $\exists A^{\star}\!\in\text{ob}(\mathscr{C})$\textit{,} $\exists\eta_{A},\epsilon_{A}\in\text{hom}(\mathscr{C})$ \textit{with} $\eta_{A}:\text{I}\longrightarrow \Box(A^{\star},A)$ \textit{and} $\epsilon_{A}:\Box(A,A^{\star})\longrightarrow\text{I}$ \textit{such that the following diagrams commute}
\begin{equation}
\xymatrixcolsep{1.08in}\xymatrixrowsep{0.2in}\xymatrix{A\ar[dd]^{1_{A}}\ar[r]^{\rho_{A}}&\Box(A,\text{I})\ar[r]^{\Box(1_{A},\eta_{A})}&\Box\big(A,\Box(A^{\star},A)\big)\ar[dd]^{\alpha_{A,A^{\star}\!\!,A}}\\
& & \\
A&\Box(\text{I},A)\ar[l]^{\lambda_{A}^{-1}}&\Box\big(\Box(A,A^{\star}),A\big)\ar[l]^{\Box(\epsilon_{A},1_{A})}}
\end{equation}
\begin{equation}
\xymatrixcolsep{1in}\xymatrixrowsep{0.2in}\xymatrix{A^{\star}\ar[dd]^{1_{A^{\star}}}\ar[r]^{\rho_{A^{\star}}}&\Box(A^{\star},\text{I})\ar[r]^{\Box(1_{A^{\star}},\eta_{A^{\star}})}&\Box\big(A^{\star},\Box(A,A^{\star})\big)\ar[dd]^{\alpha_{A^{\star},A\!\!,A^{\star}}}\\
& & \\
A^{\star}&\Box(\text{I},A^{\star})\ar[l]^{\lambda_{A^{\star}}^{-1}}&\Box\big(\Box(A^{\star},A),A^{\star}\big)\ar[l]^{\Box(\epsilon_{A^{\star}},1_{A^{\star}})}}
\end{equation}
\end{definition}



\begin{definition}\label{dagCom}\textit{A} dagger compact closed category \textit{is a compact closed category as in \cref{ccCatDef} equipped with the structure of a dagger category as in \cref{dagCatDef} such that} $\forall A\in\text{ob}(\mathscr{C})$ $\eta_{A}=\sigma_{A,A^{\star}}\circ\dagger(\epsilon_{A})$
\end{definition}

\noindent We will recall the dagger compact closed structure of quantum theory in \cref{sec: categories EJC}. It will in fact be useful for us to accomplish this task within a more general framework connected to the algebras to be discussed in the following section.

\newpage
\section{Jordan Algebraic Prerequisites}
\label{jordPrelims}
The notion of a Jordan algebra originated in the seminal paper of Pascual Jordan \cite{Jordan} (see also Pascual Jordan, John von Neumann, and Eugene Wigner \cite{Jordan1934}.) \textit{Jordan Operator Algebras} by Hanche-Olsen and St{\o}rmer \cite{HancheOlsen1984} is an indispensable resource for Jordanic theory. McCrimmon's \textit{A Taste of Jordan Algebras} \cite{McCrimmon2004} is also very good, and provides a particularly beautiful review of the history of the subject. We also point the reader to Alfsen and Shultz \cite{Alfsen2003} and to an earlier monograph by Topping \cite{Topping1965} for additional enlightening perspectives. In this section, we collect elements of Jordanic theory required for the balance of this thesis.

\begin{definition} \textit{Let $\mathcal{A}$ be a set. Let $\mathbb{F}$ be a field. A} vector space \textit{is a quadruple $\big(\mathcal{A},\mathbb{F},+,\cdot\big)$ where $\mathcal{A}$ is regarded as an abelian group with respect to} addition $+:\mathcal{A}\times\mathcal{A}\longrightarrow\mathcal{A}::(a,b)\longmapsto a+b$ \textit{and a monoid with respect to} $\mathbb{F}$-scalar multiplication $\cdot:\mathcal{A}\times\mathbb{F}\longrightarrow\mathcal{A}:(a,\alpha)\longmapsto a\alpha$\textit{. The addition and $\mathbb{F}$-scalar multiplication operations are required to satisfy the distributivity axioms:} $\cdot(a+b,\alpha)=a\alpha+b\alpha$ \textit{and} $+(a\alpha,a\beta)=a(\alpha+\beta)$\textit{. Additionally, multiplication in $\mathbb{F}$ and $\mathbb{F}$-scalar multiplication are required to commute:} $\cdot(a\alpha,\beta)=a(\alpha\beta)$\textit{. One says that} $\big(\mathcal{A},\mathbb{F},+,\cdot\big)$ \textit{is a vector space over $\mathbb{F}$.}
\end{definition}

\begin{definition}\label{algDef} \textit{Let $(\mathcal{A},\mathbb{F},+,\cdot)$ be a vector space. An } algebra \textit{is a quintuple $(\mathcal{A},\mathbb{F},+,\cdot,\jProd)$ where $\mathcal{A}$ is regarded to be equipped with} multiplication\textit{ $\jProd:\mathcal{A}\times\mathcal{A}\longrightarrow\mathcal{A}::(a,b)\longmapsto a\jProd b$. Multiplication is required to obey the distributivity axioms: $\jProd(a+b,c)=a\jProd c+b\jProd c$ and $\jProd(a,b+c)=a\jProd b+a\jProd c$ and $\jProd(a\lambda,b\mu)=\jProd(a,b)\lambda\mu$, where $\lambda,\mu\in\mathbb{F}$. An algebra is said to be} unital \textit{if there exists $u\in\mathcal{A}:\forall a\in\mathcal{A}:\jProd(u,a)=a$. An algebra is said to be} commutative i\textit{f $\forall a,b\in\mathcal{A}:\jProd(a,b)=\jProd(b,a)$. One says that $(\mathcal{A},\mathbb{F},+,\cdot,\jProd)$ is an algebra over $\mathbb{F}$, or an $\mathbb{F}$-algebra.}
\end{definition}
\noindent From now on, we abuse notation and write $\mathcal{A}$ for an $\mathbb{F}$-algebra $(\mathcal{A},\mathbb{F},+,\cdot,\jProd)$. 

\begin{definition}\label{jAlgDef} \textit{A} Jordan algebra \textit{is a commutative unital $\mathbb{R}$-algebra $\mathcal{A}$ wherein} $\forall a,b\in\mathcal{A}$
\begin{equation} 
(a\jProd a)\jProd(a\jProd b)=a\jProd\big((a\jProd a)\jProd b\big)\textit{.}
\label{jordanID}
\end{equation}
\end{definition}

\noindent Eq.~\eqref{jordanID} is the \textit{Jordan identity}, in light of which we write $a^{2}$ for $a\jProd a$ and similarly for higher exponentials. For Jordan algebras, we shall refer to $\jProd$ as the \textit{Jordan product}. Except for trivial cases: Jordan algebras are \textit{not} associative! The canonical example of a Jordan algebra is the set of $n\times n$ self-adjoint matrices with entries from $\mathbb{C}$, denoted $\mathcal{M}_{n}(\mathbb{C})_{\text{sa}}$, equipped with Jordan product $a\jProd b=(ab+ba)/2$, where juxtaposition denotes usual matrix multiplication. The Jordan algebra $\mathcal{M}_{n}(\mathbb{C})_{\text{sa}}$ is simply the ambient space for quantum cones wherein one represents self-adjoint linear endomorphisms as matices. It is easy to see that $\mathcal{M}_{n}(\mathbb{C})_{\text{sa}}$ is of the following type.

\begin{definition}\label{frDef} \textit{A} formally real Jordan algebra \textit{is a Jordan algebra} $\mathcal{A}$ \textit{such that} $\forall a_{1},\dots,a_{n}\in\mathcal{A}$
\begin{equation}
a_{1}^{2}+a_{2}^{2}+\dots+a_{n}^{2}=0\iff a_{1}=a_{2}=\dots=a_{n}=0\textit{.}
\end{equation}
\end{definition}

\noindent In this thesis, we will be exclusively concerned with \textit{finite dimensional} Jordan algebras, that is Jordan algebras $\mathcal{A}$ such that as a vector space $\text{dim}_{\mathbb{R}}\mathcal{A}$ is finite. In this case, formally real Jordan algebras are the same as Euclidean Jordan algebras. We define the latter presently, and then proceed with a proof.

\begin{definition}\label{ejaDef}\textit{A} Euclidean Jordan algebra \textit{(}\textsc{eja}\textit{)} \textit{is a finite dimensional Jordan algebra equipped with an inner product} $\langle\cdot|\cdot\rangle:\mathcal{A}\times\mathcal{A}\longrightarrow\mathbb{R}$ \textit{such that} $\forall a,b,c\in\mathcal{A}$
\begin{equation}
\langle a\jProd b|c\rangle=\langle a|b\jProd c\rangle\textit{.}
\label{euclidProp}
\end{equation}
\end{definition} 

\noindent Let $\mathcal{A}$ be an \textsc{eja}. Suppose there exists $a_{j}$ in $\mathcal{A}$ such that $\sum_{j}a_{j}^{2}=0$. By virtue of being Euclidean, $\mathcal{A}$ is endowed with an inner product $\langle\cdot|\cdot\rangle$, which in particular enjoys linearity so, $\langle 0|a\rangle=\langle c0|a\rangle=c\langle0|a\rangle$ for any real number $c$ and any $a$ in $\mathcal{A}$, so $\langle 0|a\rangle=0$ for any $a$ in $\mathcal{A}$, in particular $\langle 0|1\rangle=0$, with $1$ the unit. It follows that $\langle\sum_{j}a_{j}^{2}|1\rangle=0$. Now, by virtue of the Euclidean property expressed in Eq.~\eqref{euclidProp}, we have that $\langle a_{j}^{2}|1\rangle=\langle a_{j}|a_{j}\rangle$, so to conclude we have that $\sum_{j}\langle a_{j}|a_{j}\rangle=0$, implying each $a_{j}=0$, so $\mathcal{A}$ is formally real. The converse is easy to show in light of the Jordan-von Neumann-Wigner Classification Theorem, which we recall below as \cref{jvwClass}. 

\noindent Prior to the statement of \cref{jvwClass}, it will be helpful to introduce some preliminary notions. By the \textit{direct sum} of two Jordan algebras $\mathcal{A}$ and $\mathcal{B}$ one means the Jordan algebra $\mathcal{A}\oplus\mathcal{B}$ spanned by $a_{j}\oplus b_{j}$, with $a_{j}$ and $b_{j}$ basis elements for their respective Jordan algebras, equipped with component-wise Jordan multiplication. By the words `Jordan isomorphic' one means that there exists a vector space isomorphism preserving the relevant Jordan products. More generally, a \textit{Jordan homomorphism} is a linear map $f:\mathcal{A}\longrightarrow\mathcal{B}$ such that $\forall a_{1},a_{2}\in\mathcal{A}:f\big(a_{1}\jProd a_{2}\big)=f(a_{1})\jProd f(a_{2})$; furthermore $f$ must map the algebraic unit of $\mathcal{A}$ to the algebraic unit of $\mathcal{B}$. The terms \textit{Jordan isomorphism}, \textit{Jordan epimorphism}, \textit{Jordan monomorphism}, \textit{Jordan endomorphism}, and \textit{Jordan automorphism} are used respectively for bijective, surjective, injective, endomorphic, and automorphic Jordan homomorphisms.

\begin{theorem} (Jordan-von Neumann-Wigner \cite{Jordan1934}) \textit{Let} $n\in\mathbb{N}$\textit{. Let} $k\in\mathbb{N}_{>1}$\textit{. Let} $\cdot:\mathbb{R}^{n}\times\mathbb{R}^{n}\longrightarrow\mathbb{R}$ \textit{be the usual inner product on $\mathbb{R}^{n}$. Let} $\mathcal{A}$ \textit{be a finite dimensional formally real Jordan algebra. Then} $\mathcal{A}$ \textit{is Jordan isomorphic to a direct sum of Jordan algebras from the following list:}
\begin{itemize}\label{jvwClass}
\item $\mathcal{M}_{n}(\mathbb{R})_{\text{sa}}$ \textit{with} $a\jProd b=(ab+ba)/2$\textit{.}
\item $\mathcal{M}_{n}(\mathbb{C})_{\text{sa}}$ \textit{with} $a\jProd b=(ab+ba)/2$\textit{.}
\item $\mathcal{M}_{n}(\mathbb{H})_{\text{sa}}$ \textit{with} $a\jProd b=(ab+ba)/2$\textit{.}
\item $\mathcal{V}_{k}\cong \mathbb{R}\oplus\mathbb{R}^{k}$ \textit{with} $(\lambda_{0}\oplus\vec{\lambda})\jProd(\mu_{0}\oplus\vec{\mu}) =(\lambda_{0}\mu_{0}+\vec{\lambda}\cdot\vec{\mu})\oplus (\vec{\lambda}\mu_{0}+\vec{\mu}\lambda_{0})$\textit{.}
\item $\mathcal{M}_{3}(\mathbb{O})_{\text{sa}}$ \textit{with} $a\jProd b=(ab+ba)/2$\textit{.}
\end{itemize}
\textit{where} $\mathbb{R},\mathbb{C},\mathbb{H},\mathbb{O}$ \textit{are the reals, complexes, quaternions, and octonions, respectively, and wherein} $\mathcal{M}_{n}(\mathbb{D})_{\text{sa}}$ \textit{denotes the Jordan algebra of} $n\times n$ \textit{self-adjoint matrices over} $\mathbb{D}\in\{\mathbb{R},\mathbb{C},\mathbb{H},\mathbb{O}\}$\textit{.}
\end{theorem}

\noindent The algebras $\mathcal{M}_{n}(\mathbb{R})_{\text{sa}}$ and $\mathcal{M}_{n}(\mathbb{C})_{\text{sa}}$ are ubiquitous, and one should point out that by a self-adjoint real matrix one means the same thing as a real symmetric matrix. A review of the algebras $\mathcal{M}_{n}(\mathbb{H})_{\text{sa}}$ can be found in the author's MSc thesis \cite{Graydon2011}. As a matter of terminology, one refers to $\mathcal{M}_{n}(\mathbb{R})_{\text{sa}}$, $\mathcal{M}_{n}(\mathbb{C})_{\text{sa}}$, and $\mathcal{M}_{n}(\mathbb{H})_{\text{sa}}$ as the \textit{Jordan matrix algebras.} The Jordan matrix algebras are easily rendered Euclidean by virtue of the natural inner product $(a,b)\longmapsto\langle a|b\rangle\equiv\mathrm{Tr}(a\jProd b)$, with $\mathrm{Tr}$ the matrix trace. The \textit{exceptional Jordan algebra} of $3\times 3$ self-adjoint matrices with octonionic entries can also be rendered Euclidean in this way. Baez provides an accessible and comprehensive review of the octonions in \cite{Baez2002}. The remaining items, $\mathcal{V}_{k}$, on the list in \cref{jvwClass} are called \textit{spin factors}, which for the reader's convenience we discuss further in \cref{spinsApp}. The usual Euclidean inner product on the direct sum $\mathbb{R}^{k}\oplus\mathbb{R}$ renders $\mathcal{V}_{k}$ Euclidean. Therefore, we shall now use the term Euclidean Jordan Algebra (\textsc{eja}) for the finite dimensional formally real Jordan algebras equipped with the aforementioned inner products. The \textsc{eja}s listed in \cref{jvwClass} are \cite{Stormer1966} simple in the following sense; moreover by \cref{jvwClass} they are the only simple \textsc{eja}s.

\begin{definition}\textit{Let} $\mathcal{A}$ \textit{be a Jordan algebra. A} Jordan ideal \textit{is a Jordan subalgebra} $\mathcal{I}\subseteq\mathcal{A}$ \textit{such that} $b\in\mathcal{I}$ \textit{and} $a\in\mathcal{A}$ \textit{implies that} $b\jProd a\in\mathcal{I}$\textit{. A} trivial \textit{Jordan ideal} is $\mathcal{I}\subseteq\mathcal{A}$ \textit{such that} $\mathcal{I}\in\{\emptyset,\mathcal{A}\}$\textit{; otherwise $\mathcal{I}$ is} nontrivial\textit{. A} simple \textit{Jordan algebra is a Jordan algebra admitting no nontrivial Jordan ideals.}
\end{definition}

\noindent A \textit{projection} in an \textsc{eja} $\mathcal{A}$ is an element $a \in \mathcal{A}$ with $a^{2}=a$. If $p,q$ are projections with $p\jProd q=0$, then one says that $p$ and $q$ are \textit{orthogonal}. In this case, $p+q$ is another projection. A projection not representable as a sum of other projections is said to be \textit{primitive}. A  \textit{Jordan frame} is a set $E\subseteq \mathcal{A}$ of pairwise orthogonal primitive projections that sum to the algebraic unit.  The \textit{spectral theorem} for \textsc{eja}s \cite{FK} asserts that for every element $a \in \mathcal{A}$ there exists a Jordan frame $\{x_{j}\}$ and real numbers $\alpha_{j}\in\mathbb{R}$ such that $a$ can be expanded as the linear combination $a =\sum_{j} x_{j}\alpha_{j}$. The group of Jordan automorphisms on $\mathcal{A}$ acts transitively on the set of Jordan frames \cite{FK}.  Hence all Jordan frames for a given \textsc{eja} $\mathcal{A}$ have the same number of elements.  This number is called the \textit{dimension} of $\mathcal{A}$, and is denoted $\text{dim}_{\mathbb{R}}\mathcal{A}$. Simple counting arguments yield
\begin{eqnarray}
\text{dim}_{\mathbb{R}}\mathcal{M}_{n}(\mathbb{R})_{\text{sa}}&=&n(n+1)/2\text{,}\\
\text{dim}_{\mathbb{R}}\mathcal{M}_{n}(\mathbb{C})_{\text{sa}}&=&n^{2}\text{,}\\
\text{dim}_{\mathbb{R}}\mathcal{M}_{n}(\mathbb{H})_{\text{sa}}&=&n(2n-1)\text{,}\\
\text{dim}_{\mathbb{R}}\mathcal{V}_{k}&=&k+1\text{,}\\
\text{dim}_{\mathbb{R}}\mathcal{M}_{3}(\mathbb{O})_{\text{sa}}&=&27\text{,}\\
\end{eqnarray}
and so by \cref{jvwClass}, together with our forthcoming Eqs.~\eqref{tDefC},\eqref{pOdd}, and \eqref{pEven}, it follows that
\begin{eqnarray}
\mathcal{V}_{2}\cong\mathcal{M}_{2}(\mathbb{R})_{\text{sa}}\text{,}\label{v2iso}\\
\mathcal{V}_{3}\cong\mathcal{M}_{2}(\mathbb{C})_{\text{sa}}\text{,}\label{v3iso}\\
\mathcal{V}_{5}\cong\mathcal{M}_{2}(\mathbb{H})_{\text{sa}}\text{,}
\end{eqnarray}
which with the trivial isomorphisms $\mathbb{R}\cong\mathcal{M}_{1}(\mathbb{R},\mathbb{C},\mathbb{H})_{\text{sa}}$ constitute the only redundancies in \cref{jvwClass}. In order to discuss \textsc{eja}s in more detail, we shall require a brief foray into ordered vector spaces \cite{ASbasic}.

\noindent Let $\mathcal{A}$ be a vector space over $\mathbb{R}$. Let $\mathcal{K}\subseteq\mathcal{A}$ be a generating pointed convex cone as in \cref{coneDef}. In this case, $\mathcal{K}$ induces a partial ordering of $\mathcal{A}$, given by $a \leq b$ iff $b - a \in \mathcal{K}$; this is \textit{translation invariant}, \textit{i.e}.\ $a \leq b$ implies $a+c\leq b+c$ for all $a,b,c\in\mathcal{A}$, and \textit{homogeneous}, \textit{i.e}.\ $a\leq b$ implies $a\alpha \leq b\alpha$ for all $\alpha\in\mathbb{R}_{\geq 0}$. Conversely, such an ordering determines a cone, namely $\mathcal{K} =\{a\;\boldsymbol{|}a\geq 0\}$. Accordingly, an {\em ordered vector space} is a vector space $\mathcal{A}$ over $\mathbb{R}$ equipped with a designated cone $\mathcal{A}_{+}$  of positive elements. If $\mathcal{A}$ and $\mathcal{B}$ are ordered vector spaces, a linear function $f:\mathcal{A}\longrightarrow \mathcal{B}$ is \textit{positive} if and only if  $f(\mathcal{A}_{+})\subseteq\mathcal{B}_{+}$. If $f$ is bijective and $f(\mathcal{A}_{+})=\mathcal{B}_{+}$, then $f^{-1}(\mathcal{B}_{+})=\mathcal{A}_{+}$, so that$f^{-1}$ is also positive. In this case, we say that $f$ is an \textit{order isomorphism}.  An \textit{order automorphism} of $\mathcal{A}$ is an order isomorphism from $\mathcal{A}$ to itself; these form the corresponding \textit{order automorphism group} where the group operation is composition. Denoting the dual space of $\mathcal{A}$ by $\mathcal{A}^{\star}$, the \textit{dual cone}, denoted $\mathcal{A}^{\star}_{+}$, is the set of positive linear functionals on $\mathcal{A}$.

\begin{definition}\label{hsdDef} \textit{Let} $\mathcal{A}$ \textit{be an ordered vector space over} $\mathbb{R}$\textit{. The cone} $\mathcal{A}_{+}$ \textit{is} self-dual \textit{when there exists an inner product} $\langle\cdot|\cdot\rangle:\mathcal{A}\times\mathcal{A}\longrightarrow\mathbb{R}$ \textit{such that the} internal dual cone\textit{,} $A_{+}^{*\text{int}}\equiv\{a\in\mathcal{A}\;\boldsymbol{|}\;\forall b\in\mathcal{A}_{+}\langle a|b\rangle\geq 0\}$\textit{, satisfies} $A_{+}^{*\text{int}}=\mathcal{A}_{+}$. 
\end{definition}

\begin{definition}
\textit{Let} $\mathcal{A}$ \textit{be an ordered vector space over} $\mathbb{R}$\textit{. The cone} $\mathcal{A}_{+}$ \textit{is} homogeneous \textit{when its order automorphism group acts transitively on its interior.}
\end{definition}

\noindent We now recall a truly remarkable result proved independently by Koecher \cite{Koecher1958} and Vinberg \cite{Vinberg1960}.

\begin{theorem}(Koecher \cite{Koecher1958} and Vinberg \cite{Vinberg1960}) \textit{The only finite dimensional homogeneous self-dual cones are the positive cones of the Euclidean Jordan algebras.}
\end{theorem}

\noindent A very simple characterization of the positive cone of an \textsc{eja} $\mathcal{A}$ can \cite{HancheOlsen1984} be given as follows: $\mathcal{A}_{+}$ is precisely the set of squares of elements in $\mathcal{A}$, which is to say that $\mathcal{A}_{+}=\{a^{2}\;\boldsymbol{|}\; a\in\mathcal{A}\}$. One of our goals is to construct composites in physical theories based on \textsc{eja}s. In order to accomplish this goal, we will appeal to aspects of the well known tensor product structure of C$^{*}$\!-algebras; hence our \cref{compositesEJA} rests heavily on the theory of representations of Jordan algebras, to be reviewed presently.

\noindent We begin by recalling that a \textit{Banach space} is a complete normed vector space over $\mathbb{F}$ with respect to a norm $\|{\cdot}\|$. Let $\mathcal{A}$ be an associative algebra such that, as a vector space, $\mathcal{A}$ is Banach. One says that $\mathcal{A}$ is a \textit{Banach algebra} if $\forall x,y\in\mathcal{A}:\|{xy}\|\leq\|{x}\|\|{y}\|$. Furthermore, one says that $\mathcal{A}$ is a \textit{Banach $^{*}$\!-algebra} when $\mathcal{A}$ comes equipped with an \textit{isometric involution} taking $x\longmapsto x^{*}$ --- \textit{i.e.} $\forall x,y\in\mathcal{A}$ and $\forall\lambda,\mu\in\mathbb{F}$ one has that $(\lambda x+\mu y)^{*}=\overline{\lambda}x^{*}+\overline{\mu}y^{*}$, $(xy)^{*}=y^{*}x^{*}$, and $\|x^{*}\|=\|x\|$.

\begin{definition}\label{cStarDef}\textit{A} C$^{*}$\!-algebra \textit{is a Banach $^{*}$\!-algebra $\mathfrak{B}$ over} $\mathbb{C}$ \textit{such that} $\forall x\in\mathfrak{B}$
\begin{equation}
\|x^{*}x\|=\|x\|^{2}\text{.}
\end{equation}
\end{definition}

\noindent In this thesis, we are interested the finite-dimensional C$^{*}$-algebras $\mathcal{M}_{n}\big(\mathbb{C}\big)$, and their self-adjoint Jordan subalgebras\footnote{By a \textit{Jordan subalgebra} of the self-adjoint part $\mathfrak{B}_{\text{sa}}$ of a C$^{*}$\!-algebra $\mathfrak{B}$ one means a subset $\mathcal{A}\subseteq\mathfrak{B}_{\text{sa}}$ that is a Jordan algebra with respect to the Jordan product $a\jProd b\equiv(ab+ba)/2$, where juxtaposition denotes the associative product in $\mathfrak{B}$.} $\mathcal{M}_{n}(\mathbb{C})_{\text{sa}}$. These furnish representations of all \textit{special} \textsc{eja}s, which are defined as follows.

\begin{definition}\label{repDef}\textit{Let} $\mathcal{A}$ \textit{be a Euclidean Jordan algebra as in \cref{ejaDef}. A} representation \textit{of} $\mathcal{A}$ \textit{is a Jordan monomorphism} $\pi:\mathcal{A}\longrightarrow\mathcal{M}_{n}(\mathbb{C})_{\text{sa}}$\textit{ for some} $n\in\mathbb{N}$\textit{. A} special \textsc{eja} \textit{is a Euclidean Jordan algebra for which there exists such a representation.} 
\end{definition}

\noindent With the sole exception of $\mathcal{M}_{3}(\mathbb{O})_{\text{sa}}$, all simple \textsc{eja}s are special; hence, $\mathcal{M}_{3}(\mathbb{O})_{\text{sa}}$ is called \textit{exceptional}. Indeed, one has the following \textit{standard} representations (where $i^{2}=j^{2}=k^{2}=ijk=-1$, and where we write an arbitrary element of $\mathcal{M}_{n}(\mathbb{H})_{\text{sa}}$ as $a=\Gamma_{1}+\Gamma_{2}j$ with $\Gamma_{1}=\Gamma_{1}^{*}\in\mathcal{M}_{n}(\mathbb{C})_{\text{sa}}$ and $\Gamma_{2}=-\Gamma_{2}^{\text{T}}\in\mathcal{M}_{n}(\mathbb{C})$)
\begin{eqnarray}
\pi_{\mathbb{R}}:\mathbb{R}\longrightarrow C^{*}_{s}(\mathbb{R})\equiv\mathbb{C}&::&a\longmapsto a\label{unitEmbed}\\[0.2cm]
\pi_{\mathcal{M}_{n}(\mathbb{R})_{\text{sa}}}:\mathcal{M}_{n}(\mathbb{R})_{\text{sa}}\longrightarrow C^{*}_{s}(\mathcal{M}_{n}(\mathbb{R})_{\text{sa}})\equiv \mathcal{M}_{n}(\mathbb{C})&::&a\longmapsto a \label{realEmbed}\\[0.2cm]
\pi_{\mathcal{M}_{n}(\mathbb{C})_{\text{sa}}}:\mathcal{M}_{n}(\mathbb{C})_{\text{sa}}\longrightarrow C^{*}_{s}(\mathcal{M}_{n}(\mathbb{C})_{\text{sa}})\equiv \mathcal{M}_{n}(\mathbb{C})&::&a\longmapsto a\label{complexEmbed}\\
\pi_{\mathcal{M}_{n}(\mathbb{H})_{\text{sa}}}:\mathcal{M}_{n}(\mathbb{H})_{\text{sa}}\longrightarrow C^{*}_{s}(\mathcal{M}_{n}(\mathbb{H})_{\text{sa}})\equiv \mathcal{M}_{2n}(\mathbb{C})&::&\Gamma_{1}+\Gamma_{2}j\longmapsto \begin{pmatrix} \;\;\;\Gamma_{1} & \Gamma_{2} \\ -\overline{\Gamma_{2}} & \overline{\Gamma_{1}} \end{pmatrix}\label{quatEmbed}
\end{eqnarray}

\noindent For any\footnote{The standard representation of a direct sum is defined as the direct sum of the summands' standard representations.} special \textsc{eja} $\mathcal{A}$, we call $C^{*}_{s}(\mathcal{A})$ its \textit{standard C$^{*}$\!-algebra}, with $\pi_{\mathcal{A}}$ denoting the \textit{standard representation} thereon. The Jordan homomorphisms in Eq.~\eqref{unitEmbed}, Eq.~\eqref{realEmbed}, and Eq.~\eqref{complexEmbed} are just inclusion. The quaternionic embedding in Eq.~\eqref{quatEmbed} is the symplectic one described in detail in the author's MSc thesis \cite{Graydon2011}, which merits a little discussion. Define 
\begin{equation}
\Phi:\mathcal{M}_{2n}(\mathbb{C})\longrightarrow\mathcal{M}_{2n}(\mathbb{C})::x\longmapsto -(JxJ)^{\text{T}}
\label{phiQuat}
\end{equation} 
with $(\cdot)^{\text{T}}$ the transpose and
\begin{equation}
J\equiv\begin{pmatrix} \mathbf{0}_{n} & \mathds{1}_{n} \\ -\mathds{1}_{n} & \mathbf{0}_{n}\end{pmatrix}\text{.}
\label{jMatrix}
\end{equation}
where $\mathbf{0}_{n}$ and $\mathds{1}_{n}$ are the $n\times n$ zero and identity matrices, respectively. The \textit{symplectic representation} of $\mathcal{M}_{n}(\mathbb{H})_{\text{sa}}$ is defined by the Jordan monomorphism in Eq.~\eqref{quatEmbed}, and one can prove, with $\mathcal{M}_{2n}(\mathbb{C})_{\text{sa}}^{\Phi}$ denoting the self-adjoint fixed points of $\Phi$, the Jordan isomorphism $\mathcal{M}_{2n}(\mathbb{C})_{\text{sa}}^{\Phi}\cong\mathcal{M}_{n}(\mathbb{H})_{\text{sa}}$.

\noindent The standard representation of $\mathcal{V}_{k}$ (recall $k\in\mathbb{N}_{>1}$) introduced by Barnum-Graydon-Wilce in \cite{Barnum2016}, is involved. To begin, recall the \textit{complex Pauli matrices}
\begin{equation}
\sigma_{o}=\begin{pmatrix}1&\hspace{0.2cm}0\\0&\hspace{0.2cm}1\end{pmatrix}\qquad
\sigma_{z}=\begin{pmatrix}1&\hspace{0.2cm}0\\0&-1\end{pmatrix}\qquad \sigma_{x}=\begin{pmatrix}0&\hspace{0.2cm}1\\1&\hspace{0.2cm}0\end{pmatrix}\qquad\sigma_{y}=\begin{pmatrix}0&-i\\i&\hspace{0.2cm}0\end{pmatrix}\text{.}
\end{equation}
$\forall p\in\mathbb{N}$ with $1\leq p\leq k$, and with $\left\lfloor{\cdot}\right\rfloor$ and $\left\lceil{\cdot}\right\rceil$ the usual floor and ceiling functions, define
\begin{eqnarray}
t_{p}=\begin{cases}
\begin{cases}
\sigma_{y}^{\otimes^{\left\lceil{\frac{p}{2}}\right\rceil-1}}\otimes\sigma_{z}\otimes\sigma_{o}^{\otimes^{\left\lfloor{\frac{k}{2}}\right\rfloor-\left\lceil{\frac{p}{2}}\right\rceil}} & p \text{ odd} \\ 
\sigma_{y}^{\otimes^{\left\lceil{\frac{p}{2}}\right\rceil-1}}\otimes\sigma_{x}\otimes\sigma_{o}^{\otimes^{\left\lfloor{\frac{k}{2}}\right\rfloor-\left\lceil{\frac{p}{2}}\right\rceil}} & p \text{ even}
\end{cases} & k \text{ even} \\[1cm]
\begin{cases}
\sigma_{y}^{\otimes^{\left\lceil{\frac{p}{2}}\right\rceil-1}}\otimes\sigma_{z}\otimes\sigma_{o}^{\otimes^{\left\lceil{\frac{k}{2}}\right\rceil-\left\lceil{\frac{p}{2}}\right\rceil}} & p \text{ odd} \\ 
\sigma_{y}^{\otimes^{\left\lceil{\frac{p}{2}}\right\rceil-1}}\otimes\sigma_{x}\otimes\sigma_{o}^{\otimes^{\left\lceil{\frac{k}{2}}\right\rceil-\left\lceil{\frac{p}{2}}\right\rceil}} & p \text{ even}
\end{cases} & k \text{ odd}
\end{cases}
\label{tDefC}\end{eqnarray}
where our notation is such that $x^{\otimes^{0}}=1\in\mathbb{R}$, $x\otimes 1=x=x^{\otimes^{1}}$, $x^{\otimes^{2}}=x\otimes x$, and so on. 

\noindent One can easily check that for each $k>1$, $\{t_{1},\dots,t_{k}\}$ generates a spin factor of dimension $1+k$ with $t_{p}\jProd t_{q}=(t_{p}t_{q}+t_{q}t_{p})/2$. Next, we define $v_{p}=t_{p}$ when $k$ is even, and when $k$ is odd, we define $\forall 1\leq p <k$
\begin{eqnarray}
v_{p}&=&\begin{cases}
\sigma_{y}^{\otimes^{\left\lceil{\frac{p}{2}}\right\rceil-1}}\otimes\sigma_{z}\otimes\sigma_{o}^{\otimes^{\left\lfloor{\frac{k}{2}}\right\rfloor-\left\lceil{\frac{p}{2}}\right\rceil}} &  p  \text{ odd} \label{pOdd}\\ 
\sigma_{y}^{\otimes^{\left\lceil{\frac{p}{2}}\right\rceil-1}}\otimes\sigma_{x}\otimes\sigma_{o}^{\otimes^{\left\lfloor{\frac{k}{2}}\right\rfloor-\left\lceil{\frac{p}{2}}\right\rceil}} & p  \text{ even}
\end{cases}\\
v_{k}&=&\sigma_{y}^{\otimes^{\left\lfloor{\frac{n}{2}}\right\rfloor}}\label{pEven}
\end{eqnarray}
and we embed $\mathcal{V}_{k}$ into its standard C$^{*}$\!-algebra via the following Jordan monomorphism
\begin{equation}
\pi_{k}:\mathcal{V}_{k}\longrightarrow\mathcal{M}_{2}(\mathbb{C})_{\text{sa}}^{\otimes^{\left\lfloor{\frac{k}{2}}\right\rfloor}}::s_{p}\longmapsto v_{p}\text{.}
\end{equation}
where $s_{p}$ are the abstract generators from \cref{ssDef}. One has the following examples of such concrete generators (\textit{i.e}.\ $v_{p}$) for the ambient spaces associated with real, complex, and quaternionic quantum ``bits.''

\noindent\textit{\textbf{Rebit Example}} ($k=2$)
\begin{eqnarray}
v_{1}&=&\sigma_{z}\\
v_{2}&=&\sigma_{x}
\end{eqnarray}
\textit{\textbf{Qubit Example}} ($k=3$)
\begin{eqnarray}
v_{1}&=&\sigma_{z}\\
v_{2}&=&\sigma_{x}\\
v_{3}&=&\sigma_{y}
\end{eqnarray}
\newpage
\noindent\textit{\textbf{Quabit Example\footnote{The present representation differs from the symplectic one only in complex phases; the two are projectively equivalent.}}} ($k=5$)
\begin{eqnarray}
v_{1}&=&\sigma_{z}\otimes\sigma_{o}\\
v_{2}&=&\sigma_{x}\otimes\sigma_{o}\\
v_{3}&=&\sigma_{y}\otimes\sigma_{z}\\
v_{4}&=&\sigma_{y}\otimes\sigma_{x}\\
v_{5}&=&\sigma_{y}\otimes\sigma_{y}
\end{eqnarray}

\noindent $\mathcal{V}_{2}$, $\mathcal{V}_{3}$, and $\mathcal{V}_{5}$ are \cite{HancheOlsen1983} the only spin factors enjoying reversibility, a property to be defined presently.

\begin{definition}\label{revDef}\textit{A} reversible \textsc{eja} \textit{is a Euclidean Jordan algebra} $\mathcal{A}$ \textit{as in \cref{ejaDef} such that there exists a representation} $\pi:\mathcal{A}\longrightarrow\mathcal{M}_{n}(\mathbb{C})_{\text{sa}}$ \textit{as in \cref{repDef} such that}\footnote{Notice that for $m=2$ this is just closure under the Jordan product.} $\forall a_{1},\dots,a_{m}\in\mathcal{A}$
\begin{equation}
\pi(a_{1})\pi(a_{2})\cdots\pi(a_{m})+\pi(a_{m})\pi(a_{m-1})\cdots\pi(a_{1})\in\pi\big(\mathcal{A})\text{.}
\label{revCon}
\end{equation}
\textit{One says that} $\pi$ \textit{as such is a} reversible representation\textit{. A} universally reversible \textsc{eja} \textit{is a Euclidean Jordan algebra} $\mathcal{A}$ \textit{for which all of its representations are reversible.}
\end{definition}

\noindent In fact, one has the following theorem.

\begin{theorem}\label{spinRevThm}\cite{HancheOlsen1983} \textit{A spin factor $\mathcal{V}_{k}$ is reversible if and only if $k\in\{2,3,5\}$, and universally reversible if and only if $k\in\{2,3\}$.} 
\end{theorem}

\noindent A proof of \cref{spinRevThm} can be found in \cite{HancheOlsen1983}. Following  Hanche-Olsen, we detail explicit proofs for \cref{v23prop}, \cref{v4prop}, \cref{v5prop}, and \cref{v6prop} (relegated to \cref{spinsApp}) to establish \cref{spinRevThm}.

\noindent The Jordan matrix algebras, $\mathcal{M}_{n}(\mathbb{R})_{\text{sa}}$, $\mathcal{M}_{n}(\mathbb{C})_{\text{sa}}$, and $\mathcal{M}_{n}(\mathbb{H})_{\text{sa}}$, are also reversible. In order to see this, it will be helpful to introduce the following notions. Let $\mathfrak{B},\mathfrak{C}$ be C$^{*}$\!-algebras. A \textit{$*$-homomorphism} is a $\mathbb{C}$-linear function $f:\mathfrak{B}\longrightarrow\mathfrak{C}$ such that $\forall x,y\in\mathfrak{B}$ one has $f(xy)=f(x)f(y)$ and $f(x^{*})=f(x)^{*}$. A \textit{$*$-antihomomorphism} is a linear function $f:\mathfrak{B}\longrightarrow\mathfrak{C}$ such that $\forall x,y\in\mathfrak{B}$ one has $f(xy)=f(y)f(x)$ and $f(x^{*})=f(x)^{*}$. The terms \textit{$*$-isomorphism}, \textit{$*$-epimorphism}, \textit{$*$-monomorphism}, \textit{$*$-endomorphism}, and \textit{$*$-automorphism} are used respectively for bijective, surjective, injective, endomorphic, and automorphic $*$-homomorphisms, with the antihomomorphic varieties termed in a similar manner. An \textit{involutory} $*$-homomorphism or $*$-antihomomorphism is a suitably defined function of period two, \textit{i.e}.\ $f(f(x))=x$. For example, the usual transpose on $\mathcal{M}_{n}(\mathbb{C})_{\text{sa}}$ is an involutory $*$-antiautomorphism, but the usual adjoint is not, because the adjoint is not $\mathbb{C}$-linear. Note that in some texts the word `involution' is reserved for antilinear functions. We make no use of antilinear functions in this thesis; however, with the adjoint we do see that $\mathcal{M}_{n}(\mathbb{C})_{\text{sa}}$ in its canonical representation is reversible, because the adjoint does reverse the order of associative multiplication. Similarly, we see that $\mathcal{M}_{n}(\mathbb{R})_{\text{sa}}$ in its standard representation in $\mathcal{M}_{n}(\mathbb{C})_{\text{sa}}$ is reversible: the transpose is a $*$-antiautomorphism, so Eq.~\eqref{revCon} is trivially valid. The case of $\mathcal{M}_{n}(\mathbb{H})_{\text{sa}}$ is also easy to prove in light of the fact that $\Phi$ as defined in Eq.~\eqref{phiQuat} is an involutory $*$-antiautomorphism. Regarding the universal reversibility of the Jordan matrix algebras, one has the following.

\begin{theorem}\label{urThm} Let $n\in\mathbb{N}_{\geq 3}$\textit{. Let} $\mathbb{D}\in\{\mathbb{R},\mathbb{C},\mathbb{H}\}$\textit{. Then} $\mathcal{M}_{n}(\mathbb{D})_{\text{sa}}$ \textit{is universally reversible.}
\end{theorem}

\noindent A proof of \cref{urThm} can be found in \cite{HancheOlsen1983}. We also point out that by \cref{spinRevThm} and Eq.~\eqref{v2iso} and Eq.~\eqref{v3iso}, one has that $\mathcal{M}_{2}(\mathbb{R})_{\text{sa}}$ and $\mathcal{M}_{2}(\mathbb{C})_{\text{sa}}$ are universally reversible. We now consider universal representations.

\section{Universal Representations}
\label{cStarPrelims}

In this section, we recall the universal C$^{*}$\!-algebras enveloping Euclidean Jordan algebras. These universal representations are vital for the balance of this thesis. All of the material found in this section is collected from Hanche-Olsen and St{\o}rmer \cite{HancheOlsen1984} and from Hanche-Olsen \cite{HancheOlsen1983}. Furthermore, \cite{HancheOlsen1983} and \cite{HancheOlsen1984} actually deal with the case of infinite dimensional Jordan algebras. The infinite dimensional case is very interesting; however, in this thesis, we are concerned only with the finite dimensional case. Therefore, to avoid introducing presently unnecessary mathematical overhead, we translate the relevant results from \cite{HancheOlsen1984} and \cite{HancheOlsen1983} into our finite dimensional setting. We shall not repeat proofs from the literature. We shall, however, provide some elementary proofs of propositions stated without proof in the literature.

\begin{definition}\textit{Let}\label{jGen} $\mathcal{A}$ \textit{be a Jordan algebra as in \cref{jAlgDef}. Let subset} $\mathcal{X}\subseteq\mathcal{A}$\textit{. The} Jordan algebraic closure \textit{of $\mathcal{X}$ is denoted} $\mathfrak{j}(\mathcal{X})$ \textit{and defined to be the smallest Jordan subalgebra of $\mathcal{A}$ containing $\mathcal{X}$.}
\end{definition}

\begin{definition}\textit{Let} $\mathfrak{B}$ \textit{be a C$^{*}$\!-algebra as in \cref{cStarDef}. Let subset} $\mathcal{Y}\subseteq\mathfrak{B}$\textit{. The} C$^{*}$\!-algebraic closure \textit{of $\mathcal{Y}$ is denoted} $\mathfrak{c}(\mathcal{Y})$ \textit{and defined to be the smallest C$^{*}$\!-subalgebra of $\mathfrak{B}$ containing $\mathcal{Y}$.}
\end{definition}

\noindent One says that $\mathfrak{j}(\mathcal{X})$ is the Jordan algebra \textit{generated} by $\mathcal{X}$ in $\mathcal{A}$. Concretely, one computes $\mathfrak{j}(\mathcal{X})$ by finding the Jordan algebra produced from $\mathbb{R}$-linear combinations and Jordan products of elements of $\mathcal{X}$. Similarly, one says that $\mathfrak{c}(\mathcal{Y})$ is the C$^{*}$\!-algebra \textit{generated} by $\mathcal{Y}$ in $\mathfrak{B}$. Again, concretely, one computes $\mathfrak{c}(\mathcal{Y})$ by finding the C$^{*}$\!-algebra produced from $\mathbb{C}$-linear combinations and associative products of elements of $\mathfrak{B}$.

\begin{definition}\label{uRep}\textit{Let} $\mathcal{A}$ \textit{be a Euclidean Jordan algebra. A} universal representation \textit{of $\mathcal{A}$ is a pair} $\big(C^{*}_{u}(\mathcal{A}),\psi_{\mathcal{A}}\big)$ \textit{where} $\psi_{\mathcal{A}}:\mathcal{A}\longrightarrow C^{*}_{u}(\mathcal{A})$ \textit{is a representation as in \cref{repDef} and with} $\big(C^{*}_{u}(\mathcal{A}),\psi_{\mathcal{A}}\big)$ universal \textit{that is, as in the following \cref{existUThm}.}
\end{definition}

\begin{theorem}\label{existUThm} \cite{HancheOlsen1984} \textit{Let} $\mathcal{A}$ \textit{be a special} \textsc{eja}\textit{. Then there exists up to *-isomorphism a unique universal C$^{*}$-algebra} $C^{*}_{u}\big(\mathcal{A}\big)$ \textit{and a Jordan monomorphism} $\psi_{\mathcal{A}}:\mathcal{A}\rightarrow C^{*}_{u}\big(\mathcal{A}\big)_{sa}$ \textit{such that:}
\begin{enumerate}[(i)]
\item $\mathfrak{c}\Big(\psi_{\mathcal{A}}(\mathcal{A})\Big)=C^{*}_{u}\big(\mathcal{A}\big)$\textit{.} 
\item \textit{If} $\mathfrak{B}$ \textit{is a C$^{*}$-algebra and} $\pi:\mathcal{A}\longrightarrow\mathfrak{B}_{sa}$ \textit{is a Jordan homomorphism, then there exists a $^{*}$-homomorphism} $\hat{\pi}:C^{*}_{u}\big(\mathcal{A}\big)\longrightarrow\mathfrak{B}$ \textit{such that} $\pi=\hat{\pi}\circ\psi_{\mathcal{A}}$\textit{.}
\item \textit{There exists\footnote{$\Phi_{\mathcal{A}}$ is unique (see \cite{alfsen1980state} Corollary 5.2).} an involutive  $^{*}$-antiautomorphism} $\Phi_{\mathcal{A}}$ \textit{of} $C^{*}_{u}\big(\mathcal{A}\big)$ \textit{such that} $\Phi_{\mathcal{A}}\big(\psi_{\mathcal{A}}(a)\big)=\psi_{\mathcal{A}}(a)$ $\forall a\in\mathcal{A}$\textit{.}
\end{enumerate}
\end{theorem}
\noindent \cref{existUThm} is nicely captured by the following commutative diagram:
\begin{equation}
\xymatrix{\mathcal{A}\ar[rr]^{\psi_{\mathcal{A}}}\ar[rrd]_{\pi}& &C^{*}_{u}(\mathcal{A})_{\text{sa}}\ar[rr]^{\text{inc}} & & C^{*}_{u}(\mathcal{A})\ar[d]^{\hat{\pi}}\\
&& \mathfrak{B}_{\text{sa}}\ar[rr]_{\text{inc}} & & \mathfrak{B}}
\end{equation}

\noindent 
\noindent An intense proof for \cref{existUThm} is found in \cite{HancheOlsen1984}. The statement of \cref{existUThm} therein is actually much more general, covering, in particular, the infinite dimensional case.

\noindent  As matters of terminology, one says that $C^{*}_{u}(\mathcal{A})$ is the \textit{universal C$^{*}$\!-algebra} (\textit{enveloping}) $\mathcal{A}$, and one says that $\psi_{\mathcal{A}}$ is the \textit{canonical embedding}. One says that involutive  $^{*}$-antiautomorphism $\Phi_{\mathcal{A}}$ is the \textit{canonical involution} on the universal C$^{*}$-algebra enveloping $\mathcal{A}$. 

\noindent We now come to Hanche-Olsen's Theorem 4.4 from \cite{HancheOlsen1983}, which provides a beautiful way to compute universal C$^{*}$\!-algebras enveloping universally reversible \textsc{eja}s. In general, by $\mathfrak{B}_{\text{sa}}^{\Phi}$ one denotes the self-adjoint fixed points of a *-antiautomorphism $\Phi$ on a C$^{*}$\!-algebra $\mathfrak{B}$.

\begin{theorem}\label{hoThm}  \cite{HancheOlsen1983} \textit{With the notation above} $\mathcal{A}$ \textit{is universally reversible if and only if } $\mathcal{A}\cong C^{*}_{u}(\mathcal{A})_{\text{sa}}^{\Phi_{\mathcal{A}}}$\textit{. More generally, let} $\mathcal{A}$ \textit{be a universally reversible} \textsc{eja}\textit{. Let} $\mathfrak{B}$ \textit{be a C$^{*}$\!-algebra. Let} $\theta:\mathcal{A}\longrightarrow\mathfrak{B}_{\text{sa}}$ \textit{be an injective Jordan homomorphism such that} $\mathfrak{c}\big(\theta(\mathcal{A})\big)=\mathfrak{B}$\textit{. If $\mathfrak{B}$ admits a  *-antiautomorphism $\varphi$ such that} $\varphi\circ\theta=\theta$\textit{, then $\theta$ lifts to a $*$-isomorphism} \textit{of $C^{*}_{u}(\mathcal{A})$ onto $\mathfrak{B}$ as follows:}
\begin{equation}
\xymatrix@C=3pc@R=2pc{C^{*}_{u}(\mathcal{A})\ar[r]^{\hat{\theta}}&\mathfrak{B}\ar[r]^{\varphi} & \mathfrak{B} \\
\mathcal{A}\ar[u]^{\psi_{\mathcal{A}}}\ar[r]_{\theta} & \mathfrak{B}_{\text{sa}} \ar[u]_{\text{inc}}\ar[ur]_{\text{inc}}&}
\end{equation}
\textit{which is to say that the $*$-isomorphism} $\hat{\theta}$\textit{ is such that} $\psi_{\mathcal{A}}(\mathcal{A})\cong\mathfrak{B}_{\text{sa}}^{\varphi}$\textit{.}
\end{theorem}

\noindent Concretely, \cref{hoThm} says that if a universally reversible Jordan algebra $\mathcal{A}$ represented within $\mathfrak{B}_{\text{sa}}$ generates $\mathfrak{B}$ as a C$^{*}$\!-algebra \text{and} if there exists a $*$-antiautomorphism $\varphi$ on $\mathfrak{B}$ such that $\mathcal{A}$ as represented therein is precisely the set of its self-adjoint fixed points, \textit{i.e}.\ $\mathfrak{B}_{\text{sa}}^{\varphi}$, then $\mathfrak{B}$ \textit{is} the universal C$^{*}$\!-algebra enveloping $\mathcal{A}$, that is, up to a $*$-isomorphism. The fruits of this theorem are plentiful. In particular, we will call on \cref{hoThm} to prove our main result in the following chapter. At the moment, we have:

\begin{proposition}\label{forFact2}
\begin{equation}
C^{*}_{u}\big(\mathcal{M}_{n}(\mathbb{R})_{sa}\big)=\mathcal{M}_{n}(\mathbb{C})\text{.}\label{rnCstar}
\end{equation}
\begin{proof} $\mathcal{M}_{n}(\mathbb{R})_{sa}$ is universally reversible by \cref{urThm}. Let $|e_{r}\rangle$ denote the standard basis vectors so 
\begin{equation} 
\text{span}_{\mathbb{C}}\big\{|e_{r}\rangle\langle e_{s}|:r,s\in\{1,\dots,n\}\big\}=\mathcal{M}_{n}(\mathbb{C})\text{.}
\end{equation} 
Observe the C$^{*}$\!-algebraic product of the following real symmetric matrices
\begin{equation}
\Big(|e_{r}\rangle\langle e_{r}|\Big)\Big(|e_{r}\rangle\langle e_{s}|+|e_{s}\rangle\langle e_{r}|\Big)=|e_{r}\rangle\langle e_{s}|+\delta_{rs}|e_{r}\rangle\langle e_{r}|. 
\label{smallCalc}
\end{equation}
Obviously the inclusion of $\mathcal{M}_{n}(\mathbb{R})_{\text{sa}}$ into $\mathcal{M}_{n}(\mathbb{C})$ is injective. Eq.~\eqref{smallCalc} shows that $\mathcal{M}_{n}(\mathbb{R})_{\text{sa}}$ generates $\mathcal{M}_{n}(\mathbb{C})$. There exists a *-antiautomorphism on $\mathcal{M}_{n}(\mathbb{C})$ whose self-adjoint fixed points are exactly the real symmetric matrices, namely the transpose! So with $\varphi$ the transpose we complete the proof in light of \cref{hoThm}. 
\end{proof}
\end{proposition}

\begin{proposition}\label{withGM}
\begin{equation}
C^{*}_{u}\big(\mathcal{M}_{n}(\mathbb{C})_{sa}\big)=\mathcal{M}_{n}(\mathbb{C})\oplus\mathcal{M}_{n}(\mathbb{C})\text{.}\label{cnCstar}
\end{equation}
\begin{proof}
Define $\psi:\mathcal{M}_{n}(\mathbb{C})_{sa}\rightarrow\mathcal{M}_{n}(\mathbb{C})\oplus\mathcal{M}_{n}(\mathbb{C})$ for all $a\in\mathcal{M}_{n}(\mathbb{C})_{sa}$ via $\psi(a)=a\oplus a^{t}$. Let $\mathfrak{B}$ be the C$^{*}-$algebra generated by $\psi(\mathcal{M}_{n}(\mathbb{C})_{sa})$. Note that $\mathfrak{B}\subseteq\mathcal{M}_{n}(\mathbb{C})\oplus\mathcal{M}_{n}(\mathbb{C})$. There exists an orthonormal basis $B$ for $\mathcal{M}_{n}(\mathbb{C})$ such that $B\subset\mathcal{M}_{n}(\mathbb{C})_{sa}$ --- for instance, the generalized Gell-Mann matrices. Thus $\{a\oplus a^{t}:a\in\mathcal{M}_{n}(\mathbb{C})\}\subset\mathfrak{B}$. Choose $a,b\in\mathcal{M}_{n}(\mathbb{C})$ such that $x\equiv ab-ba$ admits a multiplicative inverse. Indeed, such $a$ and $b$ exist in all dimensions --- for example, choose $a$ to be the generalized Pauli X operator and $b$ to be the generalized Pauli Z operator, so that $ab-ba=(1-\overline{\omega})XZ$, where $\omega=\text{exp}[\textstyle{\frac{2\pi i}{n}}]$. Observe the following:
\begin{equation}
(a\oplus a^{t})(b\oplus b^{t})-ba\oplus a^{t}b^{t}=ab-ba\oplus 0=x\oplus0\in\mathfrak{B}\text{,}
\end{equation}
\begin{equation}
(a^{t}\oplus a)(b^{t}\oplus b)-a^{t}b^{t}\oplus ba=0\oplus ab-ba= 0\oplus x\in\mathfrak{B}\text{.}
\end{equation}
Note that $yx^{-1},(x^{-1})^{t}y^{t}\in\mathcal{M}_{n}(\mathbb{C})$ for all $y\in\mathcal{M}_{n}(\mathbb{C})$. Therefore $\mathfrak{B}=\mathcal{M}_{n}(\mathbb{C})\oplus\mathcal{M}_{n}(\mathbb{C}).$\\[0.3cm]
Define $\Phi:(\mathcal{M}_{n}(\mathbb{C})\oplus\mathcal{M}_{n}(\mathbb{C}))\rightarrow\mathcal{M}_{n}(\mathbb{C})\oplus\mathcal{M}_{n}(\mathbb{C})$ for all $a\oplus b$ via $\Phi(a\oplus b)=b^{t}\oplus a^{t}$. The proof now follows from \cref{urThm} and \cref{existUThm} and \cref{hoThm}.
\end{proof}
\end{proposition}

\noindent We now come to the case of quaternionic Jordan matrix algebras.

\begin{proposition}\label{qnProp}
\begin{equation}
C^{*}_{u}\big(\mathcal{M}_{n\geq 3}(\mathbb{H})_{sa}\big)=\mathcal{M}_{2n}(\mathbb{C})\text{.}\label{qnCstar}
\end{equation}
\end{proposition} 

\noindent A proof of \cref{qnProp} using \cref{hoThm} can be found in the author's MSc thesis \cite{Graydon2011}.

\noindent  The case of $\mathcal{M}_{2}(\mathbb{H})_{\text{sa}}$ is much more difficult because $\mathcal{M}_{2}(\mathbb{H})_{\text{sa}}\cong\mathcal{V}_{5}$, and the spin factor $\mathcal{V}_{5}$ is not universally reversible. Let us consider $\mathcal{V}_{5}$ in some detail. Let $\{1,i,j,k\}\subset\mathbb{H}$ be the usual quaternionic basis, with $i^{2}=j^{2}=k^{2}=ijk=-1$. The \textit{quaternionic Pauli matrices} are:
\begin{equation}
q_{0}=\begin{pmatrix} 1 & \hspace{0.3cm}0 \\ 0 & \hspace{0.3cm}1 \end{pmatrix},\hspace{0.1cm}q_{1}=\begin{pmatrix} 1 & \hspace{0.3cm}0 \\ 0 & -1 \end{pmatrix},\hspace{0.1cm} q_{2}=\begin{pmatrix} 0 & \hspace{0.3cm}1 \\ 1 & \hspace{0.3cm}0 \end{pmatrix},\hspace{0.1cm} q_{3}=\begin{pmatrix} 0 & -i \\ i & \hspace{0.3cm}0 \end{pmatrix},\hspace{0.1cm} q_{4}=\begin{pmatrix} 0 & -j \\ j & \hspace{0.3cm}0 \end{pmatrix},\hspace{0.1cm} q_{5}=\begin{pmatrix} 0 & -k \\ k & \hspace{0.3cm}0 \end{pmatrix}\text{.}
\end{equation}
Note $\mathcal{M}_{2}(\mathbb{H})_{\text{sa}}=\text{span}_{\mathbb{R}}\{q_{0},\dots,q_{5}\}$; moreover the quaternionic Pauli matrices are $\mathbb{R}$-linearly independent. Thus $\text{dim}_{\mathbb{R}}\mathcal{M}_{2}(\mathbb{H})_{\text{sa}}=6$. Furthermore, one notes that the quaternionic Pauli matrices (excluding $q_{0}$) are a spin system of cardinality $5$, which is to say that $\forall t,v\in\{1,\dots,5\}:q_{t}\jProd q_{v}=q_{0}\delta_{t,v}$, where $\forall a_{1},a_{2}\in\mathcal{M}_{2}(\mathbb{H})_{\text{sa}}$ we define $a_{1}\jProd a_{2}=(a_{1}a_{2}+a_{2}a_{1})/2$, with juxtaposition denoting the usual associative product in the real associative $*$-algebra $\mathcal{M}_{2}(\mathbb{H})$. For completeness, we note that the composition of matrix transposition $a\longmapsto a^{T}$ with quaternionic conjugation (done entry-wise) $a\longmapsto \overline{a}$ defines the $*$-operation on $\mathcal{M}_{2}(\mathbb{H})$; moreover that $\mathcal{M}_{2}(\mathbb{H})_{\text{sa}}=\{a\in\mathcal{M}_{2}(\mathbb{H}):a=a^{*}\}$. The quaternionic Pauli matrices generate $\mathcal{M}_{2}(\mathbb{H})_{\text{sa}}$ as a Jordan algebra. Equipping $\mathcal{M}_{2}(\mathbb{H})_{\text{sa}}$ with the inner product $\langle a_{1},a_{2}\rangle=\mathrm{Tr}(a_{1}\jProd a_{2})$, we see that the quaternionic Pauli matrices are in fact an orthogonal basis for $\mathcal{M}_{2}(\mathbb{H})_{\text{sa}}$; further, that $\mathcal{M}_{2}(\mathbb{H})_{\text{sa}}$ is a Euclidean Jordan algebra. We are now ready for the following.

\begin{proposition}\label{q2Prop}
\begin{equation}
C^{*}_{u}\big(\mathcal{M}_{2}(\mathbb{H})_{sa}\big)=\mathcal{M}_{4}(\mathbb{C})\oplus\mathcal{M}_{4}(\mathbb{C})\text{.}\label{q2Cstar}
\end{equation}
\begin{proof}
Recall the usual complex Pauli matrices in $\mathcal{M}_{2}(\mathbb{C})_{\text{sa}}$:
\begin{equation}
u_{\mathcal{M}_{2}(\mathbb{C})_{\text{sa}}}=\sigma_{o}=\begin{pmatrix} 1 & \hspace{0.3cm}0 \\ 0 & \hspace{0.3cm}1 \end{pmatrix},\hspace{0.1cm}\sigma_{z}=\begin{pmatrix} 1 & \hspace{0.3cm}0 \\ 0 & -1 \end{pmatrix},\hspace{0.1cm} \sigma_{x}=\begin{pmatrix} 0 & \hspace{0.3cm}1 \\ 1 & \hspace{0.3cm}0 \end{pmatrix},\hspace{0.1cm} \sigma_{y}=\begin{pmatrix} 0 & -i \\ i & \hspace{0.3cm}0\end{pmatrix}\text{.}
\end{equation}
Our complex $*$-algebra tensor product convention is such that $\mathcal{M}_{n}(\mathbb{C})\otimes\mathcal{M}_{m}(\mathbb{C})=\mathcal{M}_{m}\Big(\mathcal{M}_{n}(\mathbb{C})\Big)$. As such we have
\begin{eqnarray}
s_{0}&=&\sigma_{o}\otimes\sigma_{o}\otimes\sigma_{o}\in\mathcal{M}_{4}(\mathbb{C})_{\text{sa}}\oplus\mathcal{M}_{4}(\mathbb{C})_{\text{sa}},\nonumber\\
s_{1}&=&\sigma_{z}\otimes\sigma_{o}\otimes\sigma_{o}\in\mathcal{M}_{4}(\mathbb{C})_{\text{sa}}\oplus\mathcal{M}_{4}(\mathbb{C})_{\text{sa}},\nonumber\\
s_{2}&=&\sigma_{x}\otimes\sigma_{o}\otimes\sigma_{o}\in\mathcal{M}_{4}(\mathbb{C})_{\text{sa}}\oplus\mathcal{M}_{4}(\mathbb{C})_{\text{sa}},\nonumber\\
s_{3}&=&\sigma_{y}\otimes\sigma_{z}\otimes\sigma_{o}\in\mathcal{M}_{4}(\mathbb{C})_{\text{sa}}\oplus\mathcal{M}_{4}(\mathbb{C})_{\text{sa}},\nonumber\\
s_{4}&=&\sigma_{y}\otimes\sigma_{x}\otimes\sigma_{o}\in\mathcal{M}_{4}(\mathbb{C})_{\text{sa}}\oplus\mathcal{M}_{4}(\mathbb{C})_{\text{sa}},\nonumber\\
s_{5}&=&\sigma_{y}\otimes\sigma_{y}\otimes\sigma_{z}\in\mathcal{M}_{4}(\mathbb{C})_{\text{sa}}\oplus\mathcal{M}_{4}(\mathbb{C})_{\text{sa}}\text{.}
\end{eqnarray}
Define and extend $\mathbb{R}$-linearly on its domain the following map, where $t\in\{0,\dots,5\}$:
\begin{equation}
\psi:\mathcal{M}_{2}(\mathbb{H})_{\text{sa}}\longrightarrow\mathcal{M}_{4}(\mathbb{C})_{\text{sa}}\oplus\mathcal{M}_{4}(\mathbb{C})_{\text{sa}}:: q_{t}\longmapsto s_{t}\text{.}
\end{equation}
Noting that $\{s_{1},\dots,s_{5}\}$ is a spin system in $\mathcal{M}_{4}(\mathbb{C})_{\text{sa}}\oplus\mathcal{M}_{4}(\mathbb{C})_{\text{sa}}$, we see that $\psi$ is an injective Jordan homomorphism. One can easily check that $\{s_{1},\dots,s_{5}\}$ generates $\mathcal{M}_{4}(\mathbb{C})\oplus\mathcal{M}_{4}(\mathbb{C})$ as a C$^{*}$\!-algebra, simply by computing their $11 + 10 + 5 + 1 = 27$ projectively distinct, pairwise $\mathbb{C}$-linearly independent double, triple, quadruple, and quintuple associative products, respectively. Now, let $\psi^{u}$ be the canonical injection of $\mathcal{M}_{2}(\mathbb{H})_{\text{sa}}$ into its universal C$^{*}$\!-algebra. Then $\psi^{u}(\mathcal{M}_{2}(\mathbb{H})_{\text{sa}})$ generates $C^{*}_{u}(\mathcal{M}_{2}(\mathbb{H})_{\text{sa}})$ as a C$^{*}$-algebra. So $\mathrm{dim}_{\mathbb{C}}C^{*}_{u}(\mathcal{M}_{2}(\mathbb{H})_{\text{sa}})\leq 32$, because $C^{*}_{u}(\mathcal{M}_{2}(\mathbb{H})_{\text{sa}})$ is generated by the images of 5 anticommuting symmetries under Jordan monomorphism. However, by \cref{existUThm} there exists a $*$-homomorphism $\hat{\psi}$ such that the following diagram commutes
\begin{equation}
\xymatrixrowsep{1cm}
\xymatrixcolsep{3cm}
\xymatrix{C^{*}_{u}\big(\mathcal{M}_{2}(\mathbb{H})_{\text{sa}}\big)\ar[r]^{\hat{\psi}} & \mathcal{M}_{4}(\mathbb{C})\oplus\mathcal{M}_{4}(\mathbb{C})\\ \mathcal{M}_{2}(\mathbb{H})_{\text{sa}}\ar[r]_{\psi}\ar[u]^{\psi^{u}} & \mathcal{M}_{4}(\mathbb{C})_{\text{sa}}\oplus\mathcal{M}_{4}(\mathbb{C})_{\text{sa}}\ar[u]_{\text{inc}}}
\end{equation}
Therefore $\hat{\psi}$ must be a $*$-isomorphism, because with $\text{inc}\hspace{0.05cm}\circ\psi=\hat{\psi}\circ\psi^{u}$, there must exist 32 linearly independent elements in $C^{*}_{u}(\mathcal{A})$ corresponding to the images of the quaternionic Pauli matrices and their aforementioned 27 $\mathbb{C}$-linearly independent associative products under $\psi$ in $\mathcal{M}_{4}(\mathbb{C})\oplus\mathcal{M}_{4}(\mathbb{C})$.
\end{proof}
\end{proposition} 

\noindent The universal C$^{*}$\!-algebras enveloping the spin factors will not play a direct role in the sequel. We thus include the following proposition (whose proof can be found in \cite{HancheOlsen1984}) only for completeness.

\begin{proposition}
\begin{equation}
C^{*}_{u}(\mathcal{V}_{k})=\begin{cases} \mathcal{M}_{2^n}(\mathbb{C}) & k=2n \\  \mathcal{M}_{2^n}(\mathbb{C}) \oplus \mathcal{M}_{2^n}(\mathbb{C}) & k=2n+1 \end{cases}
\end{equation}
\end{proposition}

\noindent In closing this chapter, we present the following lemma, to be called upon later at key junctures.

\begin{lemma}\label{dsLem}\textit{Let} $\mathcal{A}$ \textit{and} $\mathcal{B}$ \textit{be Euclidean Jordan algebras. Then}
\begin{equation}
C^{*}_{u}(\mathcal{A}\oplus\mathcal{B})=C^{*}_{u}(\mathcal{A})\oplus C^{*}_{u}(\mathcal{A})\text{.}
\end{equation}
\end{lemma}

\noindent We relegate a proof of \cref{dsLem} to \cref{appendix: extending}.

\chapter{Jordanic Composites}
\label{compositesEJA}

\epigraphhead[40]
	{
		\epigraph{``The beauty of a living thing is not the atoms that go into it, but the way those atoms are put together.''}{---\textit{Carl Sagan}\\ Cosmos: A Personal Voyage (1980)}
	}

\noindent In \cite{Barnum2015} and \cite{Barnum2016b}, the author and H.\ Barnum and A.\ Wilce consider Jordan algebraic physical theories. This chapter centres on the compositional aspects of those publications.

\noindent In quantum theory, the usual tensor product $\otimes$ of finite dimensional vector spaces over $\mathbb{C}$, as defined in \cref{tensorDef}, is a mathematical apparatus for the formulation of quantum physics pertaining to composite systems. Specifically, given two physical systems with associated quantum cones $\mathcal{L}_{\text{sa}}(\mathcal{H}_{d_{\mathrm{A}}})_{+}$ and $\mathcal{L}_{\text{sa}}(\mathcal{H}_{d_{\mathrm{B}}})_{+}$, the cone of unnormalized states and effects for the composite system is  $\mathcal{L}_{\text{sa}}(\mathcal{H}_{d_{\mathrm{A}}}\otimes\mathcal{H}_{d_{\mathrm{B}}})_{+}$, with bipartite quantum channels thereupon corresponding to convex sets of completely positive trace preserving maps. In a general probabilistic theory, however, one must impose a particular compositional structure. There is no canonical choice. Physical principles, however, can single out compositional structures.

\noindent We postulate that a general probabilistic theory ought to peacefully coexist with special relativity. Of course, an arbitrary general probabilistic theory need not be as such; however, a primary motivation for the study of post-quantum theories is to expand the perimeter encompassing quantum theory so as to merge with general relativistic notions, and so our postulate is natural from this point of view. Furthermore, Einstein's postulates for special relativity are so compelling, and his theory so empirically successful, that a sacrifice of the causal structure imposed by special relativity seems premature, at least at present. We therefore posit that a general probabilistic theory ought to be nonsignaling: local operations carried out by Alice on her component of a bipartite physical system can not be used to transmit superluminal signals to Bob. Indeed, quantum theory enjoys nonsignaling, which is implied by the commutation of local operations \cite{Barrett2007}. The notion of nonsignaling post-quantum theories arose in the seminal work of Popescu and Rohrlich \cite{Popescu1994}. In this chapter, we consider a restricted class of post-quantum theories: those for which unnormalized state and effect cones are the positive cones of Euclidean Jordan algebras. These cones have recently been \textit{derived} for physical systems via information-theoretic principles \cite{BMU}\cite{MU}\cite{Wilce09}\cite{Wilce11}\cite{Wilce12}. All of these derivations leave the formulation of composites as an open problem. We impose a nonsignaling compositional structure for these cones: the \textit{canonical tensor product} $\odot$, introduced in \cref{def: canonical tensor product}. 

\noindent We structure the balance of this chapter as follows. In \cref{physMot}, we review general probabilistic theories and define composites thereof. In \cref{sec: composites Jordan}, we specialize to our case of interest: Jordanic composites. We prove, in particular, that the canonical tensor product is always an ideal of Hanche-Olsen's universal tensor product \cite{HancheOlsen1983}. In \cref{canTPcomps}, we compute all canonical tensor products involving reversible \textsc{eja}s.

\newpage
\section{General and Jordan Algebraic Probabilistic Theories}\label{physMot}  
\noindent In this section we review the usual framework for general probabilistic theories, which are built up on ordered vector spaces. We point the reader to Alfsen and Shultz \cite{ASbasic} for a general introduction to ordered vector spaces. The use of ordered vector spaces with order units for general probabilistic theories dates back to at least the work of Ludwig \cite{Ludwig}\cite{LudwigAx}, and additionally the work of Davies and Lewis \cite{Davies-Lewis}, Edwards \cite{Edwards}, and Holevo \cite{Holevo}. For a more recent survey, we refer the reader to \cite{BarnumWilceFoils}. Along the way, we specialize to our case of interest: Jordan algebraic general probabilistic theories.

\noindent Recall from our discussion immediately prior to \cref{hsdDef} that an {\em ordered vector space} is a vector space $\mathcal{A}$ over $\mathbb{R}$ equipped with a designated cone $\mathcal{A}_{+}$  of positive elements. An {\em order unit} in an ordered vector space $\mathcal{A}$ is an element $u\in\mathcal{A}_{+}$ such that, for all $a\in\mathcal{A}$, $a\leq tu$ for some $t\in\R_{\geq 0}$. In finite dimensions, this is equivalent to $u$ belonging to the convex interior of $\mathcal{A}_{+}$ \cite{AT}. An {\em order unit space} is a pair $(\mathcal{A},u)$ where $\mathcal{A}$ is an ordered vector space and $u$ is a designated order unit. An order unit space provides the machinery to discuss probabilistic concepts. A {\em state} on $(\mathcal{A},u)$ is a positive linear functional $\alpha \in \mathcal{A}^{\star}$ with $\alpha(u) = 1$. An {\em effect} is an element $a \in \mathcal{A}_{+}$ with $a\leq u$.  If $\alpha$ is a state and $a$ is an effect, we have $0\leq \alpha(a) \leq 1$: we interpret this as the {\em probability} of the given effect on the given state.  A {\em discrete  observable} on $\mathcal{A}$ with values $\lambda \in \Lambda$ is represented by an indexed family $\{a_{\lambda} | \lambda \in \Lambda\}$ of effects summing to $u$, the effect $a_{\lambda}$ associated with the event of obtaining value $\lambda$ in a measurement of the observable. {\redd Thus, if $\alpha$ is a state, $\lambda \mapsto \alpha(a_{\lambda})$ gives a probability weight on $\Lambda$.} One can extend this discussion to include more general observables by considering effect-valued measures \cite{Edwards}, but we will not. We denote the set of all states of $\mathcal{A}$ by $\Omega$; the set of all effects {\magenta --- the interval between $0$ and $u$ --- is} denoted $[0,u]$. In our present finite-dimensional setting, both are compact convex sets. Extreme points of $\Omega$ are called {\em pure states}.

\noindent One may wish to privilege certain states and/or certain effects of a probabilistic model as being ``physically possible".  One way of doing so is to consider ordered subspaces $\mathcal{E}$ of $\mathcal{A}$, with $u_{\mathcal{A}}\in \mathcal{E}$, and $\mathcal{V}$ of $\mathcal{A}^{\star}$: this picks out the set of states $\alpha \in\mathcal{V}\cap\mathcal{A}^{\star}_{+}$ and the set of effects $a\in\mathcal{E}\cap\mathcal{A}_{+}$, $a\leq u$. The pair $(\mathcal{E},\mathcal{V})$ then serves as a probabilistic model for a system having these
allowed states and effects. However, in quantum theory, and in those theories that concern us in the rest of this thesis, it is always possible to regard {\em all} states in $\mathcal{A}^{\star}$, and {\em all} effects in $\mathcal{A}$, as allowed. Henceforth, then, when we speak of a {\em probabilistic model} --- or, more briefly, a {\em model} --- we simply mean an order unit space $(\mathcal{A},u)$. It will be convenient to adopt the shorthand $\mathcal{A}$ for such a pair, writing $u_{\mathcal{A}}$ for the order unit where necessary.

\noindent By a {\em process} affecting a system associated with a probabilistic model $\mathcal{A}$, we mean a positive linear mapping $\phi:\mathcal{A} \longrightarrow\mathcal{A}$, subject to the condition that $\phi(u_{\mathcal{A}})\leq u_{\mathcal{A}}$. The probability of observing an effect $a$ after the system has been prepared in a state $\alpha$ and then subjected to a process $\phi$ is $\alpha(\phi(a))$. One can regard $\alpha(\phi(u))$ as the probability that the system is not destroyed by the process. We can, of course, replace $\phi:\mathcal{A}\longrightarrow\mathcal{A}$ with the adjoint mapping $\phi^{\ast}:\mathcal{A}^{\star}\longrightarrow \mathcal{A}^{\star}$ given by $\phi^{\ast}(\alpha)=\alpha \circ \phi$, so as to think of a process as a mapping from states to possibly subnormalized states. Thus, we can view processes either as acting on effects (the ``Heisenberg picture"), or on states (the ``Schr\"{o}dinger picture"). Any nonzero positive linear mapping $\phi:\mathcal{A}\longrightarrow\mathcal{A}$ is a nonnegative scalar multiple of a process in the aforementioned sense: since the set of states $\Omega(\mathcal{A})$ is compact, $\{\alpha(\phi(u))|\alpha\in\Omega(A)\}$ is a compact set of real numbers, not all zero, and so, has a maximum value $m(\phi) > 0$; $m(\phi)^{-1}\phi$ is then a process. For this reason, we make little further distinction here between processes and positive mappings. In particular, if $\phi$ is an order automorphism of $\mathcal{A}$, then both $\phi$ and $\phi^{-1}$ are scalar multiples of processes in the above sense: each of these processes ``undoes" the other, {\em up to normalization}, that is with nonzero probability.  A process that can be reversed with probability one (a {\em symmetry} of $\mathcal{A}$) is associated with an order-automorphism $\phi$ such that $\phi(u_{\mathcal{A}}) =u_{\mathcal{A}}$.


\noindent We denote the group of all order-automorphisms of $\mathcal{A}$ by $\Aut(\mathcal{A})$. This is a Lie group \cite{HilgertHoffmanLawson}, with its connected identity component (consisting of those processes that can be obtained by continuously deforming the identity map) is denoted $\Aut_{0}(\mathcal{A})$.  A possible (probabilistically) reversible {\em dynamics} for a system modelled by
$\mathcal{A}$ is a homomorphism $t \mapsto \phi_t$ from $(\R,+)$ to $\Aut(\mathcal{A})$, that is a one-parameter subgroup of $\Aut(\mathcal{A})$. The set of symmetries forms a compact subgroup, $\Sym(\mathcal{A})$ of $\Aut(\mathcal{A})$. One might wish to privilege certain processes as reflecting physically possible motions or evolutions of the system. In that case, one might add to the basic data $(\mathcal{A},u)$ a preferred subgroup $G(\mathcal{A})$ of order automorphisms.  We refer to such a structure as a {\em dynamical} probabilistic model, since the choice of $G(\mathcal{A})$ constrains the permitted probabilistically reversible dynamics of the model.

\noindent An inner product $\IP$ on an ordered vector space $\mathcal{A}$ is {\em positive} if and only if the associated mapping $\mathcal{A} \rightarrow\mathcal{A}^{\star}$, $a \mapsto \langle a|$, is positive, that is if $\langle a | b \rangle \geq 0$ for all $a, b \in \mathcal{A}_{+}$. We say that $\IP$ is {\em self-dualizing} if $a \mapsto \langle a |$ maps $\mathcal{A}_{+}$ {\em onto} $\mathcal{A}^{\star}_{+}$, so that $a \in\mathcal{A}_{+}$ if and only if $\langle a|b\rangle\geq0$ for all $b\in\mathcal{B}$.  We say that $\mathcal{A}$ (or its positive cone) is {\em self-dual} if $\mathcal{A}$ admits a self-dualizing inner product.  If $\mathcal{A}$ is an order unit space, we ordinarily normalize such an inner product so that $\langle u_{\mathcal{A}}|u_{\mathcal{A}} \rangle = 1$. In this case, we can represent states of $\mathcal{A}$ {\em internally}: if $\alpha \in A^{\ast}_+$ with $\alpha(u) = 1$, there is a unique $a \in A_+$ with $\langle a | b \rangle = \alpha(b)$ for all $b \in \mathcal{A}_{+}$. Conversely, if $a \in \mathcal{A}_{+}$ with $\langle a | u \rangle = 1$, then $\langle a |$ is a state. If $\mathcal{A}$ and $\mathcal{B}$ are both self-dual and $\phi:\mathcal{A}\longrightarrow\mathcal{B}$ is a positive linear mapping, we can use  
self-dualizing inner products on $\mathcal{A}$ and $\mathcal{B}$ to represent the mapping $\phi^{\ast}:\mathcal{B}^{\star}\longrightarrow\mathcal{A}^{\star}$ as a positive linear mapping $\phi^{\dagger} :\mathcal{B}\longrightarrow\mathcal{A}$, setting $\langle a | \phi^{\dagger}(b) \rangle = \langle \phi(a) | b \rangle$ for all $a \in \mathcal{A}$ and $b \in \mathcal{B}$. If $\phi : \mathcal{A} \rightarrow \mathcal{A}$ is an order-automorphism, then so is $\phi^{\dagger}$. 

\noindent We now turn our attention to our specific case of interest: probabilistic models based on the Euclidean Jordan algebras (\textsc{eja}s) reviewed in \cref{jordPrelims}. Let $\mathcal{A}$ be an \textsc{eja}. By the spectral theorem, $a = b^2$ for some $b \in \mathcal{A}$ if and only if $a$ has a spectral decomposition $a = \sum_{i} \lambda_i x_i$ in which all the coefficients $\lambda_i$ are non-negative. The Jordan unit $u$ is also an order unit; thus, any \textsc{eja} $\mathcal{A}$ can serve as a probabilistic model: physical \textit{states} correspond to normalized positive linear functionals on $\mathcal{A}$, while measurement-outcomes are associated with \emph{effects}, \textit{i.e}.\, elements $a \in \mathcal{A}_+$ with $0 \leq a \leq u$, and (discrete) observables, by sets $\{e_i\}$ of events with $\sum_i e_i = u$. 

\noindent Indeed, the inner product on $\mathcal{A}$ allows us to represent states internally, \textit{i.e.} for every state $\alpha$ there exists a unique $a \in \mathcal{A}_+$ with $\alpha(x) = \langle a | x \rangle$ for all $x \in \mathcal{A}$; conversely, every vector $a \in \mathcal{A}_{+}$ with $\langle a | u \rangle = 1$ defines a state in this way. Now, if $a$ is a projection, i.e., $a^2 = a$, let $\hat{a} = \|a\|^{-2} a$: 
then 
\begin{eqnarray}
\langle \hat{a}|u \rangle&=& \frac{1}{\|a\|^2} \langle a | u \rangle= \frac{1}{\|a\|^2} \langle a^2 | u \rangle=\frac{1}{\|a\|^2} \langle a | a \rangle= 1\text{.}
\end{eqnarray}
Thus, $\hat{a}$ represents a state. A similar computation shows that $\langle \hat{a} | a \rangle = 1$. Thus, every projection, regarded as an effect, has probability $1$ in some state.


\noindent Let $\mathcal{A}$ be an \textsc{eja}. A \textit{symmetry of} $\mathcal{A}$ is an order-automorphism preserving the unit $u_{\mathcal{A}}$. By \cite{AS} Theorem 2.80, any symmetry of $\mathcal{A}$ is a Jordan automorphism.  Another class of order automorphisms is given by the {\em quadratic representations} of certain elements of $\mathcal{A}$. The quadratic representation of $a \in \mathcal{A}$ is the mapping $U_a : \mathcal{A} \longrightarrow \mathcal{A}$ given by 
\begin{equation}
U_a = 2L_{a}^{2} - L_{a^2}\text{,}
\end{equation}
where $L_{a}:\mathcal{A}\longrightarrow\mathcal{A}::b\longmapsto a\jProd b$ is left-multiplication by $a\in\mathcal{A}$, \textit{i.e}.\ 
\begin{equation}
U_a(x) = 2a\jProd(a\jProd x) - a^2\jProd x\text{.}
\end{equation} 
These mappings have direct physical interpretations as {\em filters} in the sense of \cite{Wilce12}. We collect the following nontrivial as a proposition

\begin{proposition}\label{prop: quadratic representation} \textit{Let} $a \in \mathcal{A}$\textit{. Then} 
\begin{itemize} 
\item[(a)] $U_a$ is a positive mapping; 
\item[(b)] \textit{If} $a$ \textit{lies in the interior of} $\mathcal{A}_{+}$, \textit{then} $U_a$ \textit{is invertible,  with 
inverse given by} $U_{a^{-1}}$\textit{;}
\item[(c)] $e^{L_{a}} = U_{e^{a/2}}$
\end{itemize} 
\end{proposition}

\noindent For proof of (a), see  Theorem 1.25 of \cite{AS}; for (b), \cite{AS} Lemma 1.23 or \cite{FK}, Proposition II.3.1. Part (c) is Proposition II.3.4 in \cite{FK}. Combining (a) and (b), $U_a$ is an order automorphism for every $a$ in the interior of $A_+$. Regarding (c), note that $e^{L_a}$ is the ordinary operator exponential; in other words, $\phi_{t} := e^{tL_{a}} = U_{\frac{t}{2}a}$ is a one-parameter group of order-automorphisms in $G(\mathcal{A})$ with $\phi'(0) = L_{a}$.

\noindent Since $U_a(u_A) = 2a^2 - a^2 = a^2$, it follows that the group of order-automorphisms of $\mathcal{A}$ act transitively on the interior of $\mathcal{A}_+$. Abstractly, an ordered vector space having this property is said to be {\em homogeneous}. It follows that if $\phi$ is any order-automorphism with $\phi(u_A) = a^2 \in \mathcal{A}_+$, then $U_{a}^{-1} \circ \phi$ is a symmetry of $A$. Hence, every order-automorphism of $A$ decomposes as $\phi = U_{a} \circ g$ where $g$ is a symmetry. In fact, one can show that $a$ can be chosen to belong to the interior of $\mathcal{A}_+$, and that, with this choice, the decomposition is unique (\cite{FK}, III.5.1). 

\noindent We now return to the general case and consider composites of probabilistic models.

\noindent If $\mathcal{A}$ and $\mathcal{B}$ are probabilistic models of two physical systems, one may want to construct a model of  the pair of systems considered together. This model is denoted $\mathcal{A}\mathcal{B}$. In the framework of general probabilistic theories, there is no canonical choice for a model of a composite system. However, one can at least say what one {\em means} by a composite of two probabilistic models: at a minimum, one should be able to perform measurements on the two systems separately, and compare the results. More formally, there should be a mapping $\pi:\mathcal{A}
\times \mathcal{B} \longrightarrow \mathcal{AB}$ taking each pair of effects $(a,b) \in \mathcal{A}
\times \mathcal{B}$ to an effect $\pi(a,b) \in \mathcal{AB}$. One would like this to be {\em nonsignaling}, meaning that the probability of obtaining a particular effect on one of the component systems in a state $\omega \in \Omega(\mathcal{AB})$ should be independent of what observable is measured on the other system. One can show that this is equivalent to $\pi$'s being bilinear, with $\pi(u_{\mathcal{A}}, u_{\mathcal{B}}) = u_{\mathcal{AB}}$ \cite{BarnumWilceFoils}. Finally, one would like to be able to prepare $\mathcal{A}$ and $\mathcal{B}$ separately in arbitrary states. Therefore, in summary one has the following definition.

\begin{definition}\label{def: composites}\textit{A} composite \textit{of two probabilistic models} $\mathcal{A}$ \textit{and} $\mathcal{B}$ \textit{is a pair} $(\mathcal{AB},\pi)$ \textit{where} $\mathcal{AB}$ \textit{is a probabilistic model and} $\pi:\mathcal{A}\times\mathcal{B}\longrightarrow\mathcal{AB}$ \textit{is a bilinear mapping such that} 
\begin{enumerate}[(a)]
\item $\pi(a,b)\in(\mathcal{AB})_{+}$ \textit{for all} $a\in\mathcal{A}_{+}$ \textit{and} $b\in\mathcal{B}_{+}$
\item $\pi(u_{\mathcal{A}},u_{\mathcal{B}})=u_{\mathcal{AB}}$ 
\item \textit{For all states} $\alpha\in\Omega(\mathcal{A})$ and $\beta\in\Omega(\mathcal{B})$ \textit{there exists a state} $\gamma\in\Omega(\mathcal{AB})$ \textit{such that} $\gamma(\pi(a,b))=\alpha(a)\beta(b)$\textit{.}
\end{enumerate} 
\end{definition}

\noindent Since $\pi$ is bilinear, it extends uniquely to a linear mapping $\mathcal{A}\otimes\mathcal{B}\longrightarrow\mathcal{AB}$, which 
we continue to denote by $\pi$ (so that $\pi(a \otimes b) = \pi(a,b)$ for $a \in \mathcal{A}, b \in \mathcal{B}$), and where $\mathcal{A}\otimes\mathcal{B}$ is the usual tensor product of vector spaces over $\mathbb{R}$. 

\begin{lemma}\label{injLem} $\pi$ \textit{is injective.} 
\begin{proof} If $\pi(T) = 0$ for some $T \in \mathcal{A} \otimes \mathcal{B}$, then for all states $\alpha, \beta$ on $\mathcal{A}$ and $\mathcal{B}$ we have a state $\gamma$ on $\mathcal{AB}$ with $(\alpha \otimes \beta)(T)=\gamma(\pi(T))=0$. But then $T=0$.
\end{proof}
\end{lemma}

\noindent\cref{injLem} warrants our treating $\mathcal{A}\otimes\mathcal{B}$ as a subspace of $\mathcal{AB}$ and writing $a \otimes b$ for $\pi(a,b)$. Note that if $\omega$ is a state on $\mathcal{AB}$, then $\pi^{\ast}(\omega) := \omega \circ \pi$ defines a joint probability assignment on effects of $\mathcal{A}$ and $\mathcal{B}$: $\pi^{\ast}(\omega)(a,b) = \omega(a \otimes b)$. This gives us marginal states $\omega_{\mathcal{A}}= \omega(u_{\mathcal{A}}\otimes\;\cdot\;)$ and $\omega_{\mathcal{B}} = \omega(\;\cdot\;\otimes u_B)$. Where these are nonzero, we can also define conditional states $\omega_{1|b}(a) \equiv \omega(a \otimes b)/\omega_{\mathcal{B}}(b)$ and $\omega_{2|a}(b) \equiv \omega(a \otimes b)/\omega_{\mathcal{A}}(a)$.

\noindent If the mapping $\pi:\mathcal{A}\otimes\mathcal{B}\longrightarrow\mathcal{AB}$ is surjective, then we can identify $\mathcal{AB}$ with $\mathcal{A}\otimes\mathcal{B}$. The joint probability assigment $\pi^{\ast}(\omega)$ then completely determines $\omega$, so that states on $\mathcal{AB}$ {\em are} such joint probability assignments. In this case, we say that $\mathcal{AB}$ is {\em locally tomographic}, since states of $\mathcal{AB}$ can be determined by the joint statistics of local measurements. In finite dimensions, both classical and quantum composites have this feature, while composites of real quantum systems are not locally tomographic \cite{Araki80}\cite{Hardy-Wootters}.

\noindent When dealing with dynamical probabilistic models, one needs to supplement conditions (a), (b) and (c) with the further condition that it should be possible for $\mathcal{A}$ and $\mathcal{B}$ to evolve independently within the composite $\mathcal{A}\mathcal{B}$; hence the following definition (recall that $G(\mathcal{A})$ denotes an arbitrary preferred subgroup of order automorphisms of $\mathcal{A}$.)

\begin{definition}\label{def: dynamical composites} \textit{A} composite \textit{of dynamical probabilistic models} $\mathcal{A}$ \textit{and} $\mathcal{B}$ \textit{is a composite} $\mathcal{AB}$ \textit{in the sense of \cref{def: composites}, plus a mapping} $\otimes:G(\mathcal{A})\times G(\mathcal{B})\longrightarrow G(\mathcal{AB})$ \textit{selecting for each} $g\in G(\mathcal{A})$ and $h\in G(\mathcal{B})$ an element $g\otimes h\in G(\mathcal{AB})$ \textit{such that }
\begin{enumerate}[(a)]
\item $(g \otimes h)(a \otimes b) = ga \otimes hb$ \textit{for all} $g \in G(\mathcal{A})$, $h \in G(\mathcal{B})$, $a \in \mathcal{A}$ and $b \in \mathcal{B}$ 
\item \textit{for} $g_1, g_2 \in G(\mathcal{A})$ \textit{and} $h_1, h_2 \in G(\mathcal{B})$, 
$(g_1 \circ g_2) \otimes (h_1 \circ h_2) = (g_1 \otimes h_1) \circ (g_2 \otimes h_2)$\textit{.}
\end{enumerate}
\end{definition}

\noindent With regard to \cref{def: dynamical composites}, one notes that since $\mathcal{AB}$ may be larger than the algebraic tensor product $\mathcal{A}\otimes\mathcal{B}$, the order automorphism $(g \otimes h)$ need not be uniquely determined by condition (a).

\noindent A generalized probabilistic theory is more than a collection of models. At a minimum, one also needs the means to describe interactions between physical systems. A natural way of accomplishing this is to treat physical theories as {\em categories}, in which objects are associated with physical systems, and morphisms with processes. In the setting of this thesis, then, it's natural to regard a probabilistic {\em theory} as a category $\Cat$ in which objects are probabilistic models, that is order unit spaces, and in which morphisms give rise to positive linear mappings between these. The reason for this phrasing --- morphisms {\em giving rise to}, as opposed to simply {\em being}, positive linear mappings --- is to allow for the possibility that two abstract processes that behave the same way on effects of their source system, may differ in other ways---even in detectable ways, such as their effect on composite systems of which the source and target systems are
components. Notice that invertible morphisms $A \longrightarrow A$ that preserve the order unit then induce processes in the sense given above, so that every model $\mathcal{A}\in\Cat$ carries a distinguished group of reversible processes: models in $\Cat$, in other words, are automatically {\em dynamical} models.  

\noindent In order to allow for the formation of composite systems, it is natural to ask that $\Cat$ be a symmetric monoidal category, in particular to accommodate nonsignaling. Of course, we want to take $\mathrm{I} = \R$. Moreover, for objects $\mathcal{A},\mathcal{B}\in \Cat$, we want $\mathcal{A}\otimes\mathcal{B}$ to be a composite in the sense of Definitions \ref{def: composites} and \ref{def: dynamical composites} above. In fact, though, every part of those definitions simply {\em follows from} the monoidality of $\Cat$, except for part (b) of \ref{def: composites}; we must add ``by hand" the requirement that $u_{\mathcal{A}} \otimes u_{\mathcal{B}} = u_{\mathcal{A}\otimes\mathcal{B}}$.  The category will also pick out, for each object $\mathcal{A}$, a preferred group $G(\mathcal{A})$, namely, the group of invertible morphisms in $\text{hom}(\mathcal{A},\mathcal{A})$. The monoidal structure then picks out, for $g\in G(\mathcal{A})$ and $h\in G(\mathcal{B})$, a preferred $g \otimes h \in G(\mathcal{AB})$.  

\section{Composites of Euclidean Jordan Algebras} 
\label{sec: composites Jordan}

\noindent In this section, we first formalize the notion of composites within general probabilistic theories based on \textsc{eja}s. We then move to derive the structure of Jordan algebraic composites.

\noindent Let $\mathcal{A}$ be a Euclidean Jordan algebra. Henceforth, we shall take $G(\mathcal{A})$ to be the connected component of the identity in the group $\Aut(\mathcal{A})$ of order-automorphisms of $\mathcal{A}$. Note that if $\phi\in G(\mathcal{A})$, then $\phi^{\dagger}\in G(\mathcal{A})$ as well. Henceforth, we shall treat Jordan models as dynamical models, using $G(\mathcal{A})$ as the dynamical group. This is a reasonable choice. First, elements of $G(\mathcal{A})$ are exactly those automorphisms of $\mathcal{A}_+$ that figure in the system's possible dynamics, as elements of one-parameter groups of automorphisms. This suggests that the physical dynamical group of a dynamical model based on $\mathcal{\mathcal{A}}$ should at least be a subgroup of $G(\mathcal{A})$, so that the latter is the least constrained choice. This suggests, then, the following definition.

\begin{definition}\label{def: new Jordan composite} \textit{A} composite \textit{of} \textsc{eja}s $\mathcal{A}$ \textit{and} $\mathcal{B}$ \textit{is an} \textsc{eja} $\mathcal{AB}$\textit{ plus bilinear} $\pi:\mathcal{A}\otimes\mathcal{B}\longrightarrow\mathcal{AB}$ \textit{such that} 
\begin{enumerate} [(a)]
\item $\pi$ \textit{renders} $(\mathcal{AB},G(\mathcal{AB}))$ \textit{a dynamical composite of} $(\mathcal{A},G(\mathcal{A}))$ \textit{and} $(\mathcal{B},G(\mathcal{B}))$ \textit{as in Definition \ref{def: dynamical composites}.}
\item $(\phi \otimes \psi)^{\dagger} = \phi^{\dagger} \otimes \psi^{\dagger}$ \textit{for all} $\phi\in G(\mathcal{A})$ \textit{and} $\psi\in G(\mathcal{B})$\textit{.}
\item $\mathcal{AB}$ \textit{is generated as a Jordan algebra by the images of pure tensors.} 
\end{enumerate}
\end{definition}

\noindent In light of \cref{injLem} the $\pi:\mathcal{A}\otimes\mathcal{B}\longrightarrow\mathcal{AB}$ is injective; hence, we can, and shall, identify $\mathcal{A}\otimes\mathcal{B}$ with its image under $\pi$ in $\mathcal{AB}$, writing $\pi(a,b)$ as $a \otimes b$. Condition (b) is rather strong, but natural if we keep in mind that our ultimate aim is to construct dagger compact closed categories of \textsc{eja}s. Regarding condition (c), suppose $\pi : \mathcal{A}\times\mathcal{B}\longrightarrow \mathcal{AB}$ satisfied only (a) and (b): letting $\mathcal{A} \odot \mathcal{B}$ denote the Jordan subalgebra of $\mathcal{AB}$ generated by $\pi(\mathcal{A} \otimes \mathcal{B})$, one can show that the corestriction of $\pi$ to $\mathcal{A} \odot \mathcal{B}$ (\textit{i.e}.\ $\pi_{\text{co}}:\mathcal{A}\otimes\mathcal{B}\longrightarrow\mathcal{A}\odot\mathcal{B}$) also satisfies (a) and (b); thus, any composite in the weaker sense defined by (a) and (b) contains a composite satisfying all three conditions. 

\noindent The universal representation (see \cref{uRep}) allows one to define a natural tensor product of special EJAs, which was first studied by H. Hanche-Olsen \cite{HO}.

\begin{definition}(Hanche-Olsen \cite{HO}) \textit{The} universal tensor product \textit{of two special} \textsc{eja}s $\mathcal{A}$ \textit{and} $\mathcal{B}$ \textit{is denoted} $\mathcal{A}\hotimes\mathcal{B}$ \textit{and defined as the Jordan subalgebra of the C$^{*}$\!-algebraic tensor product} $\Cu(\mathcal{A})\otimes \Cu(\mathcal{B})$ \textit{generated by} $\psi_{\mathcal{A}}(\mathcal{A})\otimes_{\mathbb{R}}\psi_{\mathcal{B}}(\mathcal{B})\equiv\text{span}_{\mathbb{R}}\{\psi_{\mathcal{A}}(a)\otimes \psi_{\mathcal{B}}(b)\;\boldsymbol{|}\;a\in\mathcal{A},b\in\mathcal{B}\}$\textit{, i.e.} $\mathcal{A}\tilde{\otimes}\mathcal{B}\equiv\mathfrak{j}(\psi_{\mathcal{A}}(\mathcal{A})\otimes_{\mathbb{R}}\psi_{\mathcal{B}}(\mathcal{B}))$\textit{.}
\end{definition} 

\noindent Important facts about the universal tensor product are collected from \cite{HO}
in the following proposition.

\begin{proposition} \textit{Let} $\mathcal{A}$ \textit{and} $\mathcal{B}$ \textit{and} $\mathcal{C}$ \textit{denote} \textsc{eja}s\textit{.} 
\label{prop: universal tensor product properties}
\begin{enumerate}[(i)]
\item \textit{If} $\phi :\mathcal{A}\longrightarrow\mathcal{C}$, $\psi :\mathcal{B}\longrightarrow\mathcal{C}$ \textit{are unital Jordan homomorphisms with operator commuting ranges\footnote{\textit{i.e}.\ $\forall a\in\mathcal{A}$ and $\forall b\in\mathcal{B}$ and $\forall c\in\mathcal{C}$ one has that $\phi(a)\jProd(\psi(b)\jProd c)=\psi(b)\jProd(\phi(a)\jProd c)$.}, then $\exists$ a unique Jordan homomorphism} $\mathcal{A} \hotimes \mathcal{B} \longrightarrow\mathcal{C}::\psi_{\mathcal{A}}(a)\otimes \psi_{\mathcal{B}}(b)\longmapsto\phi(a)\jProd\psi(b)$ $\forall a \in \mathcal{A}$, $b \in \mathcal{B}$\textit{.}
\item\label{whichTheorem} $\Cu(\mathcal{A} \hotimes \mathcal{B}) = \Cu(\mathcal{A}) \otimes \Cu(\mathcal{B})$ \textit{and} $\Phi_{\mathcal{A} \hotimes \mathcal{B}} = \Phi_{\mathcal{A}} \otimes \Phi_{\mathcal{B}}$\textit{, where} $\Phi_{\mathcal{A} \hotimes \mathcal{B}}$ \textit{and} $\Phi_{\mathcal{A}}$ \textit{and} $\Phi_{\mathcal{B}}$ \textit{are the relevant canonical involutions.} 
\item $\mathcal{A}\hotimes\mathcal{B}$\textit{ is universally reversible unless one of the factors has a one-dimensional summand and the other has a representation onto a spin factor} $\mathcal{V}_n$ \textit{with} $n=4$ \textit{or} $n \ge 6$\textit{.}
\item \textit{If} $\mathcal{A}$ \textit{is universally reversible, then} $\mathcal{A}\hotimes \mathcal{M}_{n}(\C)_{\text{sa}} = (\Cu(\mathcal{A}) \otimes \mathcal{M}_{n}(\C))_{\text{sa}}$\textit{.}\label{8234}
\item $u_{\mathcal{A} \otilde \mathcal{B}} = u_{\mathcal{A}} \otimes u_{\mathcal{B}} = u_{C^{*}_{u}(\mathcal{A} \otilde \mathcal{B})}$\textit{.}
\end{enumerate}
\end{proposition}
\noindent Our next goal is to show that any composite of simple nontrivial \textsc{eja}s is a special universally reversible \textsc{eja} --- \textit{i.e}.\ \cref{thm: composites special}. We shall first require a brief foray into the theory of projections in \textsc{eja}s. Let $\mathcal{A}$ be an \textsc{eja} as in \cref{ejaDef}. The \textit{centre} of $\mathcal{A}$ is the set of elements of $\mathcal{A}$ that operator commute with all of the elements of $\mathcal{A}$. Denoting the centre of $\mathcal{A}$ by $\mathrm{C}(\mathcal{A})$, we then have by definition $\mathrm{C}(\mathcal{A})=\{a\in\mathcal{A}\;\boldsymbol{|}\;\forall b\in\mathcal{A}\; L_{a}\circ L_{b}=L_{b}\circ L_{a}\}$. Let projection $p\in\mathcal{A}$, \textit{i.e}.\ $p\jProd p=p$. The \textit{central cover} of $p$ is denoted by $c(p)$ and defined to be the smallest projection in the centre of $\mathcal{A}$ larger than or equal to $p$, \textit{i.e}.\ $p\leq c(p)$ with respect to the partial order of the projection lattice of $\mathcal{A}$. For further details regarding the projection lattice, we refer the reader to Chapter 5 \cite{HancheOlsen1984}.

\noindent We now recall two results from \cite{AS}: (1) for any projection $p$ in any \textsc{eja} $\mathcal{A}$, the central cover of $p$ exists (Lemma 2.37 \cite{AS}); (2) a subspace $\mathcal{M}$ of an \textsc{eja} $\mathcal{A}$ is a Jordan ideal of $\mathcal{A}$ if and only if $\mathcal{M}=c\jProd\mathcal{A}\equiv\{c\jProd a\;\boldsymbol{|}\;a\in\mathcal{A}\}$ for a (necessarily unique) projection $c\in\mathrm{C}(\mathcal{A})$ (Proposition 2.39 \cite{AS}). Furthermore $c$ is the unit in $\mathcal{M}$, \textit{i.e}.\ $c\jProd m=m$ for all $m\in\mathcal{M}$. Thus, with $\mathcal{A}$ an \textsc{eja} and $p\in\mathcal{A}$ a projection, $c(p)\geq p$ exists, and $\mathcal{M}\equiv c(p)\jProd\mathcal{A}$ is a Jordan ideal of $\mathcal{A}$ with unit $c(p)$. Now suppose that $\mathcal{M}$ is a simple Jordan ideal, which is to say that the only Jordan ideals of $\mathcal{M}$ are the empty set and $\mathcal{M}$ itself. Therefore in light of (2), the only central projections in $\mathrm{C}(\mathcal{M})$ are $0$ and $c(p)$. Therefore, in light of (1), the nonzero projections in $\mathcal{M}$ share the same central cover, namely $c(p)$. Actually, one can say something more general.
\noindent\begin{proposition}\label{for84}\textit{Let $\mathcal{A}$ be an \textsc{eja}. Let projections $p,q\in\mathcal{A}$ such that $c(p)=c(q)$. Let $\mathcal{M}_{\alpha}=c_{\alpha}\jProd\mathcal{A}$ be a simple Jordan ideal of $\mathcal{A}$, where $c_{\alpha}$ is a central projection. Let $\tilde{p}=c\jProd p$ and let $\tilde{q}=c\jProd q$. Then $\tilde{p}\neq 0$ if and only if $\tilde{q}\neq 0$.}
\begin{proof}
From Lemma 4.3.5 in \cite{HancheOlsen1984} we have that $c(c_{\alpha}\jProd p)=c_{\alpha}\jProd c(p)$. So $c(\tilde{p})=c_{\alpha}\jProd c(p)$. Likewise, $c(\tilde{q})=c_{\alpha}\jProd c(q)$. By assumption $c(p)=c(q)$. Therefore $c(\tilde{p})=c(\tilde{q})$. Suppose $\tilde{p}\neq 0$. Then $c(\tilde{p})\neq 0$, so $\tilde{q}\neq 0$ because $c(\tilde{p})=c(\tilde{q})$.
\end{proof}
\end{proposition}
\noindent We now come to address the exchange of projections by symmetries. Let $\mathcal{A}$ be an \textsc{eja}. A \textit{symmetry in} $\mathcal{A}$ (note, this is different than a symemtry \textit{of} $\mathcal{A}$) is $s\in\mathcal{A}$ such that $s^{2}=u$. Let $a\in\mathcal{A}$. Let projections $p,q\in\mathcal{A}$. One says that $p$ and $q$ are \textit{exchanged by a symmetry} when there exists a symmetry $s\in\mathcal{A}$ such that $U_{s}(p)=q$. If there exists a sequence of symmetries of $s_{1},s_{2},\dots,s_{n}\in\mathcal{A}$ such that $(U_{s_{1}}\circ U_{s_{2}}\circ\cdots\circ U_{s_{n}})(p)=q$, then $p$ and $q$ are said to be \textit{equivalent}. Equivalent projections have the same central cover (\cref{AS39}.) Furthermore, if $p_{1},\dots,p_{n}\in\mathcal{A}$ and $q_{1},\dots,q_{m}\in\mathcal{B}$ are projections, then the projections $p_{i}\otimes q_{j}$ in $\mathcal{AB}$ are equivalent (\cref{lemma: exchange}) The fact that $p_{i}\otimes q_{j}$ are projections is \cref{lemma: product projections}. Finally, one notes from Theorem 5.1.5 and Theorem 5.3.5 in \cite{HancheOlsen1984} that any \textsc{eja} $\mathcal{A}$ decomposes into the direct sum of simple ideals. We are now ready for \cref{thm: composites special}.
\begin{theorem}\label{thm: composites special}\textit{ Let 
$\mathcal{AB}$ be a composite of simple, {\red nontrivial} \textsc{eja}s $\mathcal{A}$ and $\mathcal{B}$.   
Then $\mathcal{AB}$ is a special, {\red universally reversible} EJA.}
\begin{proof} 
We shall show that every irreducible direct summand of $\mathcal{AB}$ has rank $\geq 4$, from which the result follows. The result will indeed follow, because of the Jordan-von Neumann-Wigner Classification Theorem (\cref{jvwClass}) and the fact that the Jordan matrix algebras $\mathcal{M}_{n>3}(\mathbb{D})_{\text{sa}}$ are universally reversible (\cref{urThm}.)\\[0.2cm]
Decompose $\mathcal{AB}$ as a direct sum of simple ideals, say $AB = \bigoplus_{\alpha} \mathcal{M}_{\alpha}$. Again, we can do this in light of Theorem 5.1.5 and Theorem 5.3.5 in \cite{HancheOlsen1984}.\\[0.2cm] 
Let $\pi_{\alpha}:\mathcal{AB}\longrightarrow \mathcal{M}_{\alpha}$ be the corresponding projections, and let
$u_{\alpha}\equiv\pi_{\alpha}(u_{\mathcal{AB}})$ be the unit in $\mathcal{M}_{\alpha}$. Suppose now that $\{p_1,...,p_n\}$ is a Jordan frame in $\mathcal{A}$ and $\{q_1,...,q_m\}$ is a Jordan frame in $\mathcal{B}$. This means, in particular, that $p_{1}+\cdots+p_{n}=u_{\mathcal{A}}$ and $q_{1}+\cdots+q_{m}=u_{\mathcal{B}}$.\\[0.2cm]
By Lemma 3.19 in \cite{AS}, there are symmetries in $\mathcal{A}$ exchanging the $p_i$, and there are symmetries in $\mathcal{B}$ exchanging the $q_j$. By \cref{lemma: exchange} therefore, the projections $p_i \otimes q_j$ are pairwise equivalent. By Lemma 3.9 in \cite{AS}, therefore, these projections have the same central cover $c$. This means that for each $\alpha$, the projection $\pi_{\alpha}:
\mathcal{AB}\longrightarrow \mathcal{M}_{\alpha}$ takes none of the projections $p_i \otimes
q_j$ to the zero projection in $\mathcal{M}_{\alpha}$, or it takes all of them
to zero, that is, in light of \cref{for84}.\\[0.2cm]
If $\mathcal{M}_{\alpha}$ is of the first type, $\{\pi_{\alpha}(p_i \otimes q_j)|i = 1,...,n, j = 1,...,m\}$ consists of $nm$ distinct orthogonal
projections in $M_{\alpha}$, summing to the unit $\pi_{\alpha}(u_{\mathcal{AB}})=u_{\alpha}$. Indeed, here, one notes that $\pi_{\alpha}$ is linear, $\pi_{\alpha}(x)=c\jProd x$ for some central projection $c$; hence $\sum_{i,j}\pi_{\alpha}(p_{i}\otimes q_{j})=\pi_{\alpha}(u_{\mathcal{A}}\otimes u_{\mathcal{B}})=\pi_{\alpha}(u_{\mathcal{AB}})$, where the last equality follows by our definition of composites, specifically $u_{\mathcal{AB}}=u_{\mathcal{A}}\otimes u_{\mathcal{B}}$. One also notes that $\pi_{\alpha}((p_{i}\otimes q_{j})\jProd(p_{k}\otimes q_{l}))=\delta_{ik}\delta_{jl}\pi_{\alpha}(p_{i}\otimes q_{j})$.\\[0.2cm] 
Hence, the rank of $\mathcal{M}_{\alpha}$ is at least $nm$. The rank of $\mathcal{M}_{\alpha}$ could exceed $nm$ because $\pi_{\alpha}(p_{i}\otimes q_{j})$ may not be primitive projections. In particular, since $\mathcal{A}$ and $\mathcal{B}$ are nontrivial, $n,m \geq 2$, whence, $\mathcal{M}_{\alpha}$ has rank at least 4, and hence, is special.\\[0.2cm]
\noindent Now let $p, q$ be arbitrary projections in $\mathcal{A}$ and $\mathcal{B}$, respectively:
extending each to a Jordan frame, as above, we see that for all
$\alpha$, if $\pi_{\alpha}(p \otimes q) \not = 0$, then $M_{\alpha}$
is special. Hence, $p \otimes q$ belongs to the direct sum of the
special summands of $\mathcal{AB}$. Since projections $p \otimes q$ generate $\mathcal{AB}$, the latter is 
special. Indeed, $\mathcal{AB}$ is \textit{by definition} generated as the Jordan hull of pure tensors, which coincides with the Jordan hull of pure tensors of minimal projections, since the minimal projections of the components are spanning and the closure under the Jordan hull is closure under linear combinations and Jordan products.\\[0.2cm]
\noindent The argument also shows that each simple direct summand $\mathcal{M}_{\alpha}$, in addition to being special, is not a spin factor (because the rank of any spin factor is 2), and hence, is universally reversible. It follows from this, plus the fact that direct sums of universally reversible EJAs are again universally reversible, that $\mathcal{AB}$ must be universally reversible.
\end{proof}
\end{theorem}

\begin{corollary}\label{cor: no composite with exceptional} 
\textit{If $\mathcal{A}$ is simple and $\mathcal{B}$ is exceptional, there exists 
no composite satisfying \cref{def: new Jordan composite}.}
\begin{proof}
The mapping $\mathcal{B}\longrightarrow \mathcal{AB}$ given by $b \longmapsto u_{\mathcal{A}} \otimes b$ is a Jordan monomorphism, by Proposition \ref{prop: main equation}. But there exists no Jordan monomorphism from $B$ into a special Euclidean Jordan algebra.
\end{proof}
\end{corollary}

\begin{theorem}\label{thm: simple composites ideals} \textit{Let $\mathcal{A}$ and $\mathcal{B}$ be simple, nontrivial \textsc{eja}s. Then $\mathcal{AB}$ is an ideal in $\mathcal{A}\hotimes\mathcal{B}$.}
\begin{proof} 
By \cref{cor: cor to main equation}, we have Jordan homomorphisms $\mathcal{A}, \mathcal{B} \longrightarrow \mathcal{AB}$ with operator-commuting ranges. Since $\mathcal{AB}$ is special, elements of $\mathcal{AB}$ operator-commute if and only if their images in $\Cu(\mathcal{AB})$ operator commute 
 (\cite{HO}, Lemma 5.1). Thus, we have Jordan homomorphisms $\mathcal{A}, \mathcal{B}\longrightarrow \Cu(\mathcal{AB})$ with operator-commuting ranges. The universal property of $\mathcal{A} \hotimes \mathcal{B}$ yields a Jordan homomorphism $\phi : \mathcal{A} \hotimes \mathcal{B} \longrightarrow \Cu(\mathcal{AB})$ taking 
(the image of) $a \otimes b$ in $\mathcal{A} \hotimes \mathcal{B}$ to (the image of) $a \otimes b$ in $\Cu(\mathcal{AB})$. Since both 
$\mathcal{A} \hotimes \mathcal{B}$ and $\mathcal{AB}$ are generated by pure tensors, $\phi$ takes $\mathcal{A}\hotimes \mathcal{B}$ onto $\mathcal{AB}$. Letting $K$ denote the kernel of $\phi$, an ideal of $\mathcal{A} \hotimes \mathcal{B}$, we have $\mathcal{A} \hotimes \mathcal{B} = K' \oplus K$, where $K'$ is the complementary ideal; 
the mapping $\phi$ factors through the projection $\mathcal{A} \hotimes \mathcal{B} \rightarrow K'$ to give an isomorphism 
$K' \simeq \mathcal{AB}$\ffootnote{{\red [AW: I've now checked, and this does work. But I must {\red write down the details!}]}}.
\end{proof}
\end{theorem}

\noindent \cref{thm: simple composites ideals} sharply restricts the possibilities 
for composites of simple \textsc{eja}s. In particular, it follows that {\em 
if $\mathcal{A} \hotimes \mathcal{B}$ is itself simple}, then $\mathcal{AB}
\simeq \mathcal{A} \hotimes \mathcal{B}$. In other words, in this case the {\redd universal tensor product} 
is the only ``reasonable" tensor product (to the extent that we 
think the 
conditions of Definition \ref{def: new Jordan composite} 
constitute reasonableness, in this context). If $\mathcal{A} = \mathcal{B} = \mathcal{M}_{n}(\mathbb{C})_{\text{sa}}$, 
so that\footnote{This calculation is involved and relegated to \cref{2copies}.} $\mathcal{A} \hotimes \mathcal{B} = \mathcal{M}_{n^2}({\mathbb{C}})_{\text{sa}} \oplus \mathcal{M}_{n^2}({\mathbb{C}})_{\text{sa}}$, we have another candidate, \textit{i.e.} the ideal $\mathcal{M}_{n^2}({\mathbb{C}})_{\text{sa}}$ --- the usual composite in quantum theory. If $\mathcal{A} = \mathcal{B} = \mathcal{M}_{2}(\mathbb{H})_{\text{sa}}$ (that is, if $\mathcal{A}$ and $\mathcal{B}$ are two quabits), then we have\footnote{This calculation is involved and relegated to \cref{quabitApp}.} $\mathcal{A} \hotimes \mathcal{B} = \mathcal{M}_{16}(\R)_{\text{sa}} \oplus \mathcal{M}_{16}(\R)_{\text{sa}}  \oplus \mathcal{M}_{16}(\R)_{\text{sa}}  \oplus \mathcal{M}_{16}(\R)_{\text{sa}}$, giving us four possibilities for $\mathcal{AB}$. {\em These exhaust the possibilities for real, complex 
and quaternionic quantum composites!}

\noindent In view of Theorem \ref{thm: composites special} and Corollary
\ref{cor: no composite with exceptional}, we now restrict our
attention to special \textsc{eja}s.  A {\em JC-algebra} is variously defined as
a norm-closed Jordan subalgebra of bounded linear operators on a real or complex
Hilbert space, or as a Jordan algebra that is
Jordan-isomorphic to such an algebra. In finite dimensions, any JC-algebra is Euclidean, and any special Euclidean Jordan algebra is JC
(on the second definition).  Here, we consider \textsc{eja}s
that are embedded, not necessarily in $\mathcal{L}(\mathcal{H})$ for a specific
Hilbert space, but in some definite complex $\ast$-algebra.

\begin{definition}\label{ejcDef}{\em An {\em embedded Euclidean JC-algebra} is a pair $(\mathcal{A},\M_{\mathcal{A}})$ where $\M_{\mathcal{A}}$ is a {\red finite-dimensional} complex $\ast$-algebra and $\mathcal{A}$ is a unital Jordan subalgebra of $(\M_{\mathcal{A}})_{\sa}$.} \end{definition} 

\noindent For aesthetic reasons, we abbreviate the phrase ``embedded Euclidean JC'' by EJC (rather than EEJC), letting the initial E stand simultaneously for {\em embedded} and {\em Euclidean}. We now develop a canonical tensor product for EJC
algebras, {\redd and use this to construct several symmetric monoidal 
categories of EJC-algebras.} 
In one case, we obtain a category of reversible EJAs, with
a monoidal product extending that of ordinary complex matrix algebras;
{\redd in another, which we call $\InvQM$}, we 
restrict attention to EJAs that arise as fixed-point algebras of 
involutions on {\redd complex} $\ast$-algebras (a class that includes all 
universally reversible EJCs, but also includes the quaternionic bit, 
$\mathcal{M}_{2}(\mathbb{H})_{\text{sa}}$, {\redd as symplectically embedded in $\mathcal{M}_{4}(\C)$}. Here, the monoidal structure agrees with the Hanche-Olsen tensor product 
{\redd in the cases in which the factors are universally reversible.} In what follows, recall the definition of the Jordan algebraic closure operation $\mathfrak{j}$ given in \cref{jGen}

\begin{definition} \label{def: canonical tensor product} 
\textit{The} canonical tensor product\textit{ of $(\mathcal{A},\M_{\mathcal{A}})$ and
  $(\mathcal{B},\M_{\mathcal{B}})$ is $(\mathcal{A} \odot \mathcal{B}, \M_{\mathcal{A}} \otimes \M_{\mathcal{B}})$, where 
$\mathcal{A} \odot \mathcal{B}\equiv \mathfrak{j}(\mathcal{A} \otimes \mathcal{B}) \subseteq \M_{\mathcal{A}}
  \otimes \M_{\mathcal{B}}$, the Jordan subalgebra of $(\M_{\mathcal{A}} \otimes \M_{\mathcal{B}})_{\text{sa}}$ generated 
  by $\mathcal{A} \otimes \mathcal{B}$.}
\end{definition}

\noindent Note that this makes it a matter of \textit{definition} that $\M_{\mathcal{A} \odot \mathcal{B}} = \M_{\mathcal{A}} \otimes \M_{\mathcal{B}}$. 
In particular, if $\mathcal{A}$ and $\mathcal{B}$ are universally embedded, so that $\M_{\mathcal{A}} = \Cu(\mathcal{A})$ and $\M_{\mathcal{B}} = \Cu(\mathcal{B})$, then 
$\mathcal{A} \odot \mathcal{B} = \mathcal{A} \hotimes \mathcal{B}$, so that $\M_{\mathcal{A} \hotimes \mathcal{B}} = \Cu(\mathcal{A}) \otimes \Cu(\mathcal{B})$ by definition; 
but the fact that this last is $\Cu(\mathcal{A} \hotimes \mathcal{B})$ {\redd (whence, $\mathcal{A} \odot \mathcal{B}$ is universally reversible)} is a theorem (\textit{i.e}.\ \cref{prop: universal tensor product properties} (\ref{whichTheorem}).) Let us call an EJC $(\mathcal{A},\M_{\mathcal{A}})$ {\em reversible} iff $\mathcal{A}$ is reversible in $\M_{\mathcal{A}}$, \textit{i.e}.\ with respect to the inclusion map $\mathcal{A}\hookrightarrow(\mathbf{M}_{\mathcal{A}})_{\text{sa}}$. Note that if $\mathcal{A}$ is a 
simple EJC {\em standardly} 
embedded in $\M_{\mathcal{A}}$, then $\mathcal{A}$ is reversible iff $A =\mathcal{M}_{n}(\mathbb{D})_{\text{sa}}$ for some $\mathbb{D}\in\{\mathbb{R},\mathbb{C},\mathbb{H}\}$ and for some $n$.

\begin{theorem} \label{prop: canonical product composite}
\textit{If $(\mathcal{A},M_{\mathcal{A}})$ and $(\mathcal{B},M_{\mathcal{B}})$ are {\em \red reversible} EJC-algebras 
then their canonical tensor product $\mathcal{A} \odot \mathcal{B}$ is a 
{\redd composite in the sense of \cref{def: new Jordan composite}.}}
\begin{proof}
We must show that $\mathcal{A} \odot \mathcal{B}$ is a dynamical composite, in the 
sense of Definition \ref{def: dynamical composites}. In particular, it must 
be a composite of probabilistic models in the sense of Definition \ref{def: composites}. 
Conditions (a)-(c) of that definition are easily verified:   
The condition $u_{\mathcal{AB}} = u_{\mathcal{A}} \otimes u_{\mathcal{B}}$ follows from the unitality of 
the embeddings $\mathcal{A}\longmapsto \M_{\mathcal{A}}$, $\mathcal{B}\longmapsto \M_{\mathcal{B}}$ and $\mathcal{AB}\longmapsto \M_{\mathcal{AB}}$. 
For all $a, x \in (\M_\mathcal{A})_{\text{sa}}$ and $b, y \in (\M_{\mathcal{B}})_{\text{sa}}$, we have 
\[\langle a \otimes b | x \otimes y \rangle = \Tr(ax \otimes by) = \Tr(ax)\Tr(by) = \langle a |x \rangle \langle b | y \rangle\]
Thus,  for any states $\alpha = \langle a |$ and $\beta = \langle b |$, where $a \in \mathcal{A}_+$ and $b \in \mathcal{B}_+$, we have 
a state $\gamma = \langle a \otimes b |$ on $\mathcal{AB}$ with $\gamma(x \otimes y) = \alpha(x) \beta(y)$ for 
all $x \in \mathcal{A}_+$, $y \in \mathcal{B}_+$.   

\noindent That $\mathcal{A} \odot \mathcal{B}$ satisfies the additional conditions required
to be a \emph{dynamical} composite in the sense of Definition \ref{def: dynamical composites} is not trivial.  
 However, using a result of Upmeier on
 the extension of derivations on reversible JC-algebras
 \cite{Upmeier}, one can show that any $\phi \in G(\mathcal{A})$ extends to an
 element $\hat{\phi} \in G((\M_{\mathcal{A}})_{\text{sa}})$ that preserves $\mathcal{A}$, and that
 $\hat{\phi} \otimes 1_{M_{\mathcal{B}}}$ is an order-automorphism of $(\M_{\mathcal{A}}
 \otimes \M_{\mathcal{B}})_{\text{sa}}$ preserving $\mathcal{A} \odot \mathcal{B}$.  It follows that 
 $\hat{\phi} \otimes \hat{\psi} = (\hat{\phi} \otimes 1_{\M_{\mathcal{B}}}) \circ (1_{\M_{\mathcal{A}}} \otimes \hat{\psi})$ 
 preserves $\mathcal{A} \odot \mathcal{B}$ as well. It is not difficult to verify that 
 $(\hat{\phi} \otimes \hat{\psi})^{\dagger} = \hat{\phi}^{\dagger} \otimes \hat{\psi}^{\dagger}$. 
 The details are presented in Appendix \ref{appendix: extending}. 
Thus $\mathcal{A} \odot \mathcal{B}$ satisfies
 conditions (a) and (b) of Definition \ref{def: new Jordan composite}.  Condition 
(c) is immediate from the definition of $\mathcal{A} \odot \mathcal{B}$.
\end{proof}
\end{theorem}

\noindent Combining \cref{prop: canonical product composite} with Theorem \ref{thm: composites special}, we see that if $\mathcal{A}$ and $\mathcal{B}$ are simple reversible EJCs for which $\mathcal{A} \hotimes \mathcal{B}$ is also simple, then the {\redd canonical} and universal tensor products of $\mathcal{A}$ and $\mathcal{B}$ coincide. 
This covers all cases except those involving two factors of the form $\mathcal{M}_{n}(\mathbb{C})_{\text{sa}}$ and $\mathcal{M}_{k}(\mathbb{C})_{\text{sa}}$, and those involving 
$\mathcal{M}_{2}(\mathbb{H})_{\text{sa}}$ as a factor. In fact, many of the latter are covered by the following.


\begin{corollary}\label{cor: canonical composites with involutions}
\textit{Let $(\mathcal{A},\M_{\mathcal{A}})$ and $(\mathcal{B},\M_{\mathcal{B}})$ be reversible EJCs with $\mathcal{A}$ generating
$\M_{\mathcal{A}}$ and $\mathcal{B}$ generating $\M_{\mathcal{B}}$ as $\ast$-algebras.  Suppose $\M_{\mathcal{A}}$ and
$\M_{\mathcal{B}}$ carry involutions $\Phi$ and $\Psi$, respectively, with $\Phi$
fixing points of $\mathcal{A}$ and $\Psi$ fixing points of $\mathcal{B}$ (i.e.
  $\mathcal{A} \subseteq \M_{\mathcal{A}}^\Phi$ and $\mathcal{B} \subseteq \M_{\mathcal{B}}^\Psi$). Then $\mathcal{A} \odot
\mathcal{B} = (\M_{\mathcal{A}} \otimes \M_{\mathcal{B}})^{\Phi \otimes \Psi}_{\text{sa}}$, the set of
self-adjoint fixed points of $\Phi \otimes \Psi$.}
\begin{proof} 
By \cref{prop: canonical product composite}, $\mathcal{A} \odot \mathcal{B}$ is a dynamical
composite of $\mathcal{A}$ and $\mathcal{B}$; hence, by  \cref{thm: composites
  special}, $\mathcal{A} \odot \mathcal{B}$ is universally reversible. Since $\Phi \otimes
\Psi$ fixes points of $\mathcal{A} \otimes \mathcal{B}$, it also fixes points of the
Jordan algebra this generates in $\M_{\mathcal{A}} \otimes \M_\mathcal{B}$, \textit{i.e}.\ of $\mathcal{A} \odot
\mathcal{B}$. But then, by Proposition \cref{hoThm}, $\mathcal{A} \odot \mathcal{B}$
is exactly the set of fixed points of $\Phi \otimes \Psi$.
\end{proof}
\end{corollary}


\tempout{\hbcomment{ I think the more serious worry about the claim that $A
  \odot B$ is a composite may concern conditioning.  As mentioned
  above, I think we need to either show, from Definition 
\ref{def: composite probabilistic} or 
\ref{def: new Jordan composite}, or stipulate in a 
modified definition, that conditional
  states and effects are in the correct state space.  If we must
  stipulate it in the definition, then we need to show it in the above
proof.  This means that for \emph{any} states $y \in A \odot B$ and states
$a \in A$ respectively, we need to show 
that the functional $b \mapsto \tr (y (a \otimes b))$ from $B_+$ to $\R$ is
in $B_+^*$. {\blue [AW: I've added a note after the def. of composites explaining why (for 
models with full-dual state spaces) we always do have conditional states. Regarding the functional above: 
$a \otimes b$ is positive in $\mathcal{AB}$ for all fixed positive $a$ and positive $b$; since $y$ is positive 
in $\mathcal{AB}$ and $T(x,y) := \Tr(xy)$ is positive when both factors are positive (i.e., is a positive inner product), 
we do have positivity here.]} }}

\noindent It follows from \cref{prop: canonical product composite}, together with \cref{thm: simple composites ideals} that the canonical and universal tensor products coincide for simple, reversible EJCs whose universal tensor products are simple. This covers all cases except for those in which one factor is $\mathcal{M}_{2}(\Q)_{\text{sa}}$ (the quaternionic bit), and those in which both factors have the form $\mathcal{M}_{n}(\C)_{\text{sa}}$ for some $n$. In the following section we prove the lemmas that directly yield the following theorem.

\begin{theorem}\label{mgRes} \textit{The canonical tensor products of the Jordan matrix algebras with respect to their standard embeddings are as follows:} 
\noindent \begin{center}
  \begin{tabular}{c | c  c  c  c }
  $\odot$ & $\mathbb{R}$ & $\mathcal{M}_{n}(\mathbb{R})_{\text{sa}}$ & $\mathcal{M}_{n}(\mathbb{C})_{\text{sa}}$  & $\mathcal{M}_{n}(\mathbb{H})_{\text{sa}}$ \\
  \hline
  $\mathbb{R}$ & $\mathbb{R}$ & $\mathcal{M}_{n}(\mathbb{R})_{\text{sa}}$ & $\mathcal{M}_{n}(\mathbb{C})_{\text{sa}}$  & $\mathcal{M}_{n}(\mathbb{H})_{\text{sa}}$ \\
  $\mathcal{M}_{m}(\mathbb{R})_{\text{sa}}$ &  & $\mathcal{M}_{nm}(\mathbb{R})_{\text{sa}}$ & $\mathcal{M}_{nm}(\mathbb{C})_{\text{sa}}$ & $\mathcal{M}_{nm}(\mathbb{H})_{\text{sa}}$ \\
  $\mathcal{M}_{m}(\mathbb{C})_{\text{sa}}$ & & & $\mathcal{M}_{nm}(\mathbb{C})_{\text{sa}}$ & $\mathcal{M}_{2nm}(\mathbb{C})_{\text{sa}}$ \\
  $\mathcal{M}_{m}(\mathbb{H})_{\text{sa}}$ & &  & & $\mathcal{M}_{4nm}(\mathbb{R})_{\text{sa}}$
  \end{tabular}
\end{center}
\end{theorem}

\section{Explicit Computations}\label{canTPcomps}
\noindent In this section, we provide explicit proofs of the lemmas which yield our proof of \cref{mgRes} 

\noindent We will require the following propositions, wherein to remind the reader $\mathfrak{c}(X)$ denotes the C$^{*}$\!-subalgebra of $\mathcal{M}_{n}(\mathbb{C})$ generated\footnote{\textit{i.e}.\ the C$^{*}$\!-algebraic closure} (as a C$^{*}$\!-algebra) by $X\subseteq\mathcal{M}_{n}(\mathbb{C})_{\text{sa}}$, and $X\otimes_{\mathbb{R}}Y\equiv\text{span}_{\mathbb{R}}\big\{x\otimes y:x\in X\text{ and } y\in Y\big\}\subseteq\mathcal{M}_{n}(\mathbb{C})_{\text{sa}}\otimes\mathcal{M}_{m}(\mathbb{C})_{\text{sa}}=\mathcal{M}_{nm}(\mathbb{C})_{\text{sa}}$. Furthermore, we denote the C$^{*}$\!-algebraic unit of $\mathcal{M}_{n}(\mathbb{C})$ by $\mathds{1}_{n}$, which of course coincides with the $\mathbb{R}$-algebraic unit of $\mathcal{M}_{n}(\mathbb{C})_{\text{sa}}$.

\begin{proposition}\label{lem1}\textit{ Let $\mathds{1}_{n}\in X\subseteq\mathcal{M}_{n}(\mathbb{C})_{\text{sa}}$ and $\mathds{1}_{m}\in Y\subseteq\mathcal{M}_{m}(\mathbb{C})_{\text{sa}}$ such that $\mathfrak{c}(X)=\mathcal{M}_{n}(\mathbb{C})$ and $\mathfrak{c}(Y)=\mathcal{M}_{m}(\mathbb{C})$. Then} 
\begin{equation}
\mathfrak{c}(X\otimes_{\mathbb{R}} Y)=\mathcal{M}_{nm}(\mathbb{C})\text{.}
\end{equation}
\begin{proof} Note that $\mathfrak{c}\big(X\otimes_{\mathbb{R}}\{\mathds{1}_{m}\}\big)=\mathcal{M}_{n}(\mathbb{C})\otimes \mathds{1}_{m}=\big\{x\otimes \mathds{1}_{m}:x\in\mathcal{M}_{n}(\mathbb{C})\big\}$. Likewise $\mathfrak{c}\big(\{\mathds{1}_{n}\}\otimes_{\mathbb{R}}Y\big)=\mathds{1}_{n}\otimes\mathcal{M}_{m}(\mathbb{C})$. Thus $\mathcal{M}_{n}(\mathbb{C})\otimes\mathcal{M}_{m}(\mathbb{C})\subseteq\mathfrak{c}\big(X\otimes_{\mathbb{R}} Y\big)$. Moreover, by construction $\mathfrak{c}\big(X\otimes_{\mathbb{R}} Y\big)\subseteq\mathcal{M}_{n}(\mathbb{C})\otimes\mathcal{M}_{m}(\mathbb{C})$.
\end{proof}
\end{proposition}

\noindent For the next proposition, let $\mathcal{A}^{\Phi_{\mathcal{A}}+}$ (respectively, $\mathcal{A}^{\Phi_{\mathcal{A}}-}$) denote the $+1$ (respectively, $-1$) eigenspace of a real orthogonal transformation $\Phi_{\mathcal{A}}$ with $\Phi_{\mathcal{A}}\circ\Phi_{\mathcal{A}}=\mathtt{id}_{\mathcal{A}}:\mathcal{A}\longrightarrow\mathcal{A}::a\longmapsto a$ on a real vector space $\mathcal{A}$ (so $\mathcal{A}=\mathcal{A}^{\Phi_{\mathcal{A}}+}\oplus\mathcal{A}^{\Phi_{\mathcal{A}}-}$).

\begin{proposition}\label{lem2}\textit{Let $\mathcal{A},\mathcal{B}$ be vector spaces over $\mathbb{R}$ equipped with inner products $\langle\cdot,\cdot\rangle_{\mathcal{A}}$ and $\langle\cdot,\cdot\rangle_{\mathcal{B}}$. Let $\Phi_{\mathcal{A}}:\mathcal{A}\longrightarrow\mathcal{A}$ and $\Phi_{\mathcal{B}}:\mathcal{B}\longrightarrow\mathcal{B}$ be real orthogonal transformations\footnote{\textit{i.e}.\ $\forall a_{1},a_{2}\in\mathcal{A}$ one has $\langle a_{1},a_{2}\rangle_{\mathcal{A}}=\big\langle\Phi_{\mathcal{A}}(a_{1}),\Phi_{\mathcal{A}}(a_{2})\big\rangle_{\mathcal{A}}$ and $\forall b_{1},b_{2}\in\mathcal{B}$ one has $\langle b_{1},b_{2}\rangle_{\mathcal{B}}=\big\langle \Phi_{\mathcal{B}}(b_{1}),\Phi_{\mathcal{B}}(b_{2})\big\rangle_{\mathcal{B}}$.} with $\Phi_{\mathcal{A}}\circ\Phi_{\mathcal{A}}=\mathtt{id}_{\mathcal{A}}$ and $\Phi_{\mathcal{B}}\circ\Phi_{\mathcal{B}}=\mathtt{id}_{\mathcal{B}}$. Then}
\begin{equation}
\mathrm{dim}_{\mathbb{R}}\Big(\big(\mathcal{A}\otimes_{\mathbb{R}}\mathcal{B}\big)^{(\Phi_{\mathcal{A}}\otimes\Phi_{\mathcal{B}})+}\Big)=\mathrm{dim}_{\mathbb{R}}\Big(\mathcal{A}^{\Phi_{\mathcal{A}}+}\Big)\mathrm{dim}_{\mathbb{R}}\Big(\mathcal{B}^{\Phi_{\mathcal{B}}+}\Big)+\mathrm{dim}_{\mathbb{R}}\Big(\mathcal{A}^{\Phi_{\mathcal{A}}-}\Big)\mathrm{dim}_{\mathbb{R}}\Big(\mathcal{B}^{\Phi_{\mathcal{B}}-}\Big)
\end{equation}
\begin{proof} The eigenvalues of $\Phi_{\mathcal{A}}\otimes\Phi_{\mathcal{B}}$ are the $\mathbb{R}$-multiplicative products of the eigenvalues of $\Phi_{\mathcal{A}}$ and $\Phi_{\mathcal{B}}$. Any $\mathbb{R}$-linearly independent subsets $\{a_{1},\dots,a_{n}\}\subset\mathcal{A}$ and $\{b_{1},\dots,b_{m}\}\subset\mathcal{B}$ are such that $\big\{a_{k}\otimes b_{l}:k\in\{1,\dots,n\}\wedge l\in\{1,\dots,m\}\big\}$ is an $\mathbb{R}$-linearly independent subset of $\mathcal{A}\otimes_{\mathbb{R}}\mathcal{B}$. Finally, note $\text{dim}_{\mathbb{R}}\big(\mathcal{A}\otimes_{\mathbb{R}}\mathcal{B}\big)=\text{dim}_{\mathbb{R}}\mathcal{A}\text{dim}_{\mathbb{R}}\mathcal{B}$.
\end{proof}
\end{proposition}

\noindent We shall also require the following elementary propostions concerning generation, the first being trivial.

\begin{proposition}\label{fac1}$\mathfrak{c}\Big(\pi_{\mathbb{R}}\big(\mathbb{R}\big)\Big)=\mathbb{C}$\textit{.}
\end{proposition}

\begin{proposition}\label{fac2} $\mathfrak{c}\Big(\pi_{\mathcal{M}_{n}(\mathbb{R})_{\text{sa}}}\big(\mathcal{M}_{n}(\mathbb{R})_{\text{sa}}\big)\Big)=\mathcal{M}_{n}(\mathbb{C})$\textit{.}
\begin{proof} Immediate from our proof of \cref{forFact2}.
\end{proof}
\end{proposition}

\begin{proposition}\label{fac3} $\mathfrak{c}\Big(\pi_{\mathcal{M}_{n}(\mathbb{C})_{\text{sa}}}\big(\mathcal{M}_{n}(\mathbb{C})_{\text{sa}}\big)\Big)=\mathcal{M}_{n}(\mathbb{C})$\textit{.}
\begin{proof}$\forall X\in\mathcal{M}_{n}(\mathbb{C})$ one has  $X=A+B$ with $A=(X+X^{*})/2\in\mathcal{M}_{n}(\mathbb{C})_{\text{sa}}$ and $B=(X-X^{*})/2$. The problem thus reduces to proving that all skew-self-adjoint matrices can be generated by the self-adjoint matrices as a C$^{*}$\!-algebra. Let $B$ be any skew-self-adjoint matrix. Then $iB$ is self-adjoint. So we can generate the matrix $B$ just be multiplying $iB$ by the complex scalar $-i$. $\square$
\end{proof}
\end{proposition}

\begin{proposition}\label{fac4} $\mathfrak{c}\Big(\pi_{\mathcal{M}_{n}(\mathbb{H})_{\text{sa}}}\big(\mathcal{M}_{n}(\mathbb{H})_{\text{sa}}\big)\Big)=\mathcal{M}_{2n}(\mathbb{C})$\textit{.}
\end{proposition}

\noindent A proof of \cref{fac4} is found in the authors MSc thesis \cite{Graydon2011}. We are now ready for our main lemmas.

\begin{lemma}\label{realArblem}
\begin{equation}
\mathbb{R}\odot\mathcal{M}_{n}(\mathbb{R})_{\text{sa}}\cong\mathcal{M}_{n}(\mathbb{R})_{\text{sa}}\
\label{RodotMR}
\end{equation}
\begin{proof}
By definition $\pi_{\mathbb{R}}\big(\mathbb{R}\big)\otimes_{\mathbb{R}}\pi_{\mathcal{M}_{n}(\mathbb{R})_{\text{sa}}}\big(\mathcal{M}_{n}(\mathbb{R})_{\text{sa}}\big)\cong\mathcal{M}_{n}(\mathbb{R})_{\text{sa}}$: a closed Jordan algebra!.
\end{proof}
\end{lemma}

\begin{lemma}\label{realRealLem}
\begin{equation}
\mathcal{M}_{m}(\mathbb{R})_{\text{sa}}\odot\mathcal{M}_{n}(\mathbb{R})_{\text{sa}}\cong\mathcal{M}_{nm}(\mathbb{R})_{\text{sa}}\
\label{MRodotMR}
\end{equation}
\begin{proof} We will begin by proving that 
\begin{equation}
C^{*}_{u}\Big(\mathcal{M}_{m}(\mathbb{R})_{\text{sa}}\odot\mathcal{M}_{n}(\mathbb{R})_{\text{sa}}\Big)\cong\mathcal{M}_{nm}(\mathbb{C})\text{.}
\label{CstarMRMR}
\end{equation}
By \cref{lem2} and \cref{fac2} $\pi_{\mathcal{M}_{m}(\mathbb{R})_{\text{sa}}}\big(\mathcal{M}_{m}(\mathbb{R})_{\text{sa}}\big)\otimes_{\mathbb{R}}\pi_{\mathcal{M}_{n}(\mathbb{R})_{\text{sa}}}\big(\mathcal{M}_{n}(\mathbb{R})_{\text{sa}}\big)\subset\mathcal{M}_{m}(\mathbb{R})_{\text{sa}}\odot\mathcal{M}_{n}(\mathbb{R})_{\text{sa}}\subset\mathcal{M}_{nm}(\mathbb{C})_{\text{sa}}$ in fact generates $\mathcal{M}_{nm}(\mathbb{C})$ as a C$^{*}$\!-algebra. Furthermore, there exists an antiautomorphism of $\mathcal{M}_{nm}(\mathbb{C})$ that fixes the inclusion of $\mathcal{M}_{m}(\mathbb{R})_{\text{sa}}\odot\mathcal{M}_{n}(\mathbb{R})_{\text{sa}}$ --- namely the transpose! Note that $\mathcal{M}_{m}(\mathbb{R})_{\text{sa}}\odot\mathcal{M}_{n}(\mathbb{R})_{\text{sa}}$ is not a spin factor because it has rank at least 4. Therefore by Hanche-Olsen's Theorem Eq.~\eqref{CstarMRMR} holds. Our knowledge of the universal C$^{*}$\!-algebras of all items in the Jordan-von Neumann-Wigner Classification Theorem now already implies that $\mathcal{M}_{m}(\mathbb{R})_{\text{sa}}\odot\mathcal{M}_{n}(\mathbb{R})_{\text{sa}}$ must be Jordan isomorphic to a simple real or quaternionic matrix algebra. It is real. For the proof, we simply apply \cref{lem2} with $\mathcal{A}\cong\mathcal{M}_{m}(\mathbb{C})_{\text{sa}}$ and $\mathcal{B}\cong\mathcal{M}_{n}(\mathbb{C})_{\text{sa}}$, with $\Phi_{\mathcal{A}}$ and $\Phi_{\mathcal{B}}$ the local transposes, and count dimensions:
\begin{eqnarray}
\mathrm{dim}_{\mathbb{R}}\Big(\mathcal{M}_{m}(\mathbb{R})_{\text{sa}}\odot\mathcal{M}_{n}(\mathbb{R})_{\text{sa}}\Big)&=&\frac{m(m+1)}{2}\frac{n(n+1)}{2}+\frac{m(m-1)}{2}\frac{n(n-1)}{2}\\
&=&\frac{nm}{4}(nm+m+n+1+nm-m-n+1)\\
&=&\frac{nm(nm+1)}{2}\\[0.1cm]
&=&\mathrm{dim}_{\mathbb{R}}\mathcal{M}_{nm}(\mathbb{R})_{\text{sa}}\text{.}
\end{eqnarray}
The proof is complete because we know that the self-adjoint fixed points of the canonical antiautomorphism of the universal C$^{*}$\!-algebra of a  universally reversible \textsc{eja} are \textit{precisely} the images of the \textsc{eja} under canonical injection.
\end{proof}
\end{lemma}

\begin{lemma}\label{realQuatLem}
\begin{equation}
\mathcal{M}_{m}(\mathbb{R})_{\text{sa}}\odot\mathcal{M}_{n}(\mathbb{H})_{\text{sa}}\cong\mathcal{M}_{nm}(\mathbb{H})_{\text{sa}}
\label{MRodotMH}
\end{equation}
\begin{proof} The proof is entirely similar to the previous case; however let us keep a full record. First, we will show that
\begin{equation}
C^{*}_{u}\Big(\mathcal{M}_{m}(\mathbb{R})_{\text{sa}}\odot\mathcal{M}_{n}(\mathbb{H})_{\text{sa}}\Big)\cong\mathcal{M}_{2nm}(\mathbb{C})\text{.}
\label{CstarMRMH}
\end{equation}
By \cref{lem1} together with \cref{fac2} and \cref{fac4}, 
\begin{equation}
\pi_{\mathcal{M}_{m}(\mathbb{R})_{\text{sa}}}\big(\mathcal{M}_{m}(\mathbb{R})_{\text{sa}}\big)\otimes_{\mathbb{R}}\pi_{\mathcal{M}_{n}(\mathbb{H})_{\text{sa}}}\big(\mathcal{M}_{n}(\mathbb{H})_{\text{sa}}\big)\subset\mathcal{M}_{m}(\mathbb{R})_{\text{sa}}\odot\mathcal{M}_{n}(\mathbb{H})_{\text{sa}}\subset\mathcal{M}_{2nm}(\mathbb{C})_{\text{sa}}
\end{equation} 
in fact generates $\mathcal{M}_{2nm}(\mathbb{C})$ as a C$^{*}$\!-algebra. Recall $\mathcal{M}_{2n}(\mathbb{C})$ admits the *-antiautomorphism $\mathcal{J}::x\longmapsto -Jx^{T}J$ with $J^{2}=-\mathds{1}_{2n}$ the usual symplectic matrix; moreover, the set of self-adjoint fixed points of $\mathcal{J}$ is precisely the image of $\mathcal{M}_{n}(\mathbb{H})_{\text{sa}}$ in standard representation. Recall also that $\mathcal{M}_{m}(\mathbb{C})$ enjoys the usual transpose antiautomorphism, say $\mathcal{T}$, whose set of self-adjoint fixed points is precisely the real symmetric matrices. Therefore $\mathcal{M}_{m}(\mathbb{R})_{\text{sa}}\odot\mathcal{M}_{n}(\mathbb{H})_{\text{sa}}\subset\mathcal{M}_{2nm}(\mathbb{C})_{\text{sa}}$ is pointwise fixed by the antiautomorphism $\mathcal{T}\otimes\mathcal{J}$ (again, recalling our general arguments in the notes mentioned in the previous proof.) Noting that $\mathcal{M}_{m}(\mathbb{R})_{\text{sa}}\odot\mathcal{M}_{n}(\mathbb{H})_{\text{sa}}$ has rank at least 4, we conclude that Eq.~\eqref{CstarMRMH} holds in light of Hanche-Olsen's Theorem. Once again, our knowledge of the universal C$^{*}$\!-algebras of all items in the Jordan-von Neumann-Wigner Classification \cite{Jordan1934} now already implies that $\mathcal{M}_{m}(\mathbb{R})_{\text{sa}}\odot\mathcal{M}_{n}(\mathbb{H})_{\text{sa}}$ must be Jordan isomorphic to a simple real or quaternionic matrix algebra. It is quaternionic. For the proof, we simply apply \cref{lem2} with $\mathcal{A}\cong\mathcal{M}_{m}(\mathbb{C})_{\text{sa}}$ and $\mathcal{B}\cong\mathcal{M}_{2n}(\mathbb{C})_{\text{sa}}$, with $\Phi_{\mathcal{A}}=\mathcal{T}$ and $\Phi_{\mathcal{B}}=\mathcal{J}$ the local involutions and count dimensions:
\begin{eqnarray}
\mathrm{dim}_{\mathbb{R}}\Big(\mathcal{M}_{m}(\mathbb{R})_{\text{sa}}\odot\mathcal{M}_{n}(\mathbb{H})_{\text{sa}}\Big)&=&\frac{m(m+1)}{2}n(2n-1)+\frac{m(m-1)}{2}(4n^{2}-2n^{2}+n)\\
&=&\frac{nm}{2}\Big((m+1)(2n-1)+(m-1)(2n+1)\Big)\\
&=&\frac{nm}{2}(2nm+2n-m-1+2nm-2n+m-1)\\[0.2cm]
&=&nm(2nm-1)\\[0.2cm]
&=&\mathrm{dim}_{\mathbb{R}}\mathcal{M}_{nm}(\mathbb{H})_{\text{sa}}\text{.}
\end{eqnarray}
\end{proof}
\end{lemma}

\begin{lemma}\label{quatQuatLem}
\begin{equation}
\mathcal{M}_{m}(\mathbb{H})_{\text{sa}}\odot\mathcal{M}_{n}(\mathbb{H})_{\text{sa}}\cong\mathcal{M}_{4nm}(\mathbb{R})_{\text{sa}}\textit{.}
\label{MHodotMH}
\end{equation}
\begin{proof} The proof is entirely similar to the proofs of \cref{realRealLem} and \cref{realQuatLem}. In particular,
\begin{equation}
C^{*}_{u}\Big(\mathcal{M}_{m}(\mathbb{H})_{\text{sa}}\odot\mathcal{M}_{n}(\mathbb{H})_{\text{sa}}\Big)\cong\mathcal{M}_{4nm}(\mathbb{C})\text{.}
\end{equation}
And
\begin{eqnarray}
\mathrm{dim}_{\mathbb{R}}\Big(\mathcal{M}_{m}(\mathbb{H})_{\text{sa}}\odot\mathcal{M}_{n}(\mathbb{H})_{\text{sa}}\Big)&=&m(2m-1)n(2n-1)+\big(4m^{2}-2m^{2}+m\big)\big(4n^{2}-2n^{2}+n\big)\\
&=&nm(4nm-2m-2n+1)+nm(2m+1)(2n+1)\\[0.2cm]
&=&nm(4nm-2m-2n+1+4nm+2m+2n+1)\\[0.2cm]
&=&nm(8nm+2)\\[0.2cm]
&=&\frac{4nm(4nm+1)}{2}\\
&=&\mathrm{dim}_{\mathbb{R}}\mathcal{M}_{4nm}(\mathbb{R})_{\text{sa}}\text{.}
\end{eqnarray}
\end{proof}
\end{lemma}
\noindent Next, we consider nontrivial canonical tensor products with respect to the standard embeddings where at least one factor is complex. Our technique used for the proofs of \cref{realRealLem}, \cref{realQuatLem}, and \cref{quatQuatLem} can not be applied in this case, because there is no *-antiautomorphism on $\mathcal{M}_{n}(\mathbb{C})$ whose fixed points are exactly the self-adjoint matrices (note that the usual adjoint is not an antiautomorphism because it is not $\mathbb{C}$-linear.) Instead, our proofs of the following lemmas rely on the structure of the Generalized Gell-Mann matrices reviewed in \cref{gmApp}.

\begin{lemma}\label{realComLem} $\mathcal{M}_{m}(\mathbb{R})_{\text{sa}}\odot\mathcal{M}_{n}(\mathbb{C})_{\text{sa}}=\mathcal{M}_{nm}(\mathbb{C})_{\text{sa}}$\textit{.}
\begin{proof}
By definition $\mathcal{M}_{m}(\mathbb{R})_{\text{sa}}\odot\mathcal{M}_{n}(\mathbb{C})_{\text{sa}}$ is the Jordan algebraic closure of 
\begin{equation}
\mathcal{M}_{m}(\mathbb{R})_{\text{sa}}\otimes_{\mathbb{R}}\mathcal{M}_{n}(\mathbb{C})_{\text{sa}}=\text{span}_{\mathbb{R}}\big\{a\otimes b:a\in\mathcal{M}_{m}(\mathbb{R})_{\text{sa}}\text{ and } b\in\mathcal{M}_{m}(\mathbb{C})_{\text{sa}}\big\}\subseteq\mathcal{M}_{nm}(\mathbb{C})_{\text{sa}}\text{.}
\end{equation}
Our method of proof will be to obtain the following inclusion by Jordan-algebraic generation
\begin{equation}
\mathcal{M}_{nm}(\mathbb{C})_{\text{sa}}\subseteq\mathcal{M}_{m}(\mathbb{R})_{\text{sa}}\odot\mathcal{M}_{n}(\mathbb{C})_{\text{sa}}\text{.}
\end{equation}
For reference, recall the Generalized Gell-Mann Matrices for $r,s\in\{1,\dots,m\}$:
\begin{equation}
G_{r,s}=
\begin{cases}
\textstyle{\frac{1}{\sqrt{2}}}\big(E_{r,s}+E_{s,r}\big)\in\mathcal{M}_{m}(\mathbb{R})_{\text{sa}}\subset\mathcal{M}_{m}(\mathbb{C})_{\text{sa}} & r<s \\
\textstyle{\frac{i}{\sqrt{2}}}\big(E_{r,s}-E_{s,r}\big)\in\mathcal{M}_{m}(\mathbb{C})_{\text{sa}} & s<r \\
\textstyle{\frac{1}{\sqrt{r(r+1)}}}\big(-rE_{r+1,r+1}+\sum_{k=1}^{r}E_{k,k}\big)\in\mathcal{M}_{m}(\mathbb{R})_{\text{sa}}\subset\mathcal{M}_{m}(\mathbb{C})_{\text{sa}}  & s=r\neq n
\end{cases}
\label{GMm}
\end{equation}
For convenience introduce for $j,k\in\{1,\dots,n\}$:
\begin{equation}
H_{r,s}=
\begin{cases}
\textstyle{\frac{1}{\sqrt{2}}}\big(F_{j,k}+F_{k,j}\big)\in\mathcal{M}_{n}(\mathbb{C})_{\text{sa}} & j<k \\
\textstyle{\frac{i}{\sqrt{2}}}\big(F_{j,k}-F_{k,j}\big)\in\mathcal{M}_{n}(\mathbb{C})_{\text{sa}} & k<j \\
\textstyle{\frac{1}{\sqrt{j(j+1)}}}\big(-jF_{j+1,j+1}+\sum_{l=1}^{j}F_{k,k}\big)\in\mathcal{M}_{n}(\mathbb{C})_{\text{sa}}  & j=k\neq n
\end{cases}
\label{GMc}
\end{equation}
where $F_{j,k}$ are defined just like $E_{r,s}$ (a symbolic distinction that we introduce to signal a possibly different underlying dimensionality.) 

\noindent We will now generate $i\big(E_{r,s}-E_{s,r}\big)/\sqrt{2}$ to appear in the left factor. Observe that for any $r<s\in\{1,\dots,m\}$ and for any $j<k\in\{1,\dots,n\}$
\begin{eqnarray}
Z_{r,s}&=E_{r,r}-E_{s,s}&\in\mathcal{M}_{m}(\mathbb{R})_{\text{sa}}\\
W_{j,k}&=F_{j,j}-F_{k,k}&\in\mathcal{M}_{n}(\mathbb{C})_{\text{sa}}
\end{eqnarray} 
and one therefore has $Z_{r,s}\otimes W_{j,k}\in\mathcal{M}_{m}(\mathbb{R})_{\text{sa}}\odot\mathcal{M}_{n}(\mathbb{C})_{\text{sa}}$. One also has for any $r<s\in\{1,\dots,m\}$ and for any $j<k\in\{1,\dots,n\}$ that $G_{r,s}\otimes H_{j,k}\in\mathcal{M}_{m}(\mathbb{R})_{\text{sa}}\odot\mathcal{M}_{n}(\mathbb{C})_{\text{sa}}$. Their Jordan algebraic product is contained in the canonical tensor product as well, explicitly
\begin{eqnarray}
\big(Z_{r,s}\otimes W_{j,k}\big)\jProd \big(G_{r,s}\otimes H_{j,k}\big)&=&\Big(Z_{r,s}G_{r,s}\otimes W_{j,k}H_{j,k}+G_{r,s}Z_{r,s}\otimes H_{j,k}W_{j,k}\Big)/2\nonumber\\
&=&\Big(\big(E_{r,r}-E_{s,s}\big)\big(E_{r,s}+E_{s,r}\big)\otimes\big(F_{j,j}-F_{k,k}\big)\big(F_{j,k}+F_{k,j}\big)\Big)/4\nonumber\\
&+&\Big(\big(E_{r,s}+E_{s,r}\big)\big(E_{r,r}-E_{s,s}\big)\otimes\big(F_{j,k}+F_{k,j}\big)\big(F_{j,j}-F_{k,k}\big)\Big)/4\nonumber\\
&=&\Big(\big(E_{r,s}-E_{s,r}\big)\otimes\big(F_{j,k}-F_{k,j}\big)\Big)/4\nonumber\\
&+&\Big(\big(-E_{r,s}+E_{s,r}\big)\otimes\big(-F_{j,k}+F_{k,j}\big)\Big)/4\nonumber\\
&=&\Big(\big(E_{r,s}-E_{s,r}\big)\otimes\big(F_{j,k}-F_{k,j}\big)\Big)/2\in\mathcal{M}_{m}(\mathbb{R})_{\text{sa}}\odot\mathcal{M}_{n}(\mathbb{C})_{\text{sa}}
\end{eqnarray}
So $\big(E_{r,s}-E_{s,r}\big)\otimes\big(F_{j,k}-F_{k,j}\big)\in\mathcal{M}_{m}(\mathbb{R})_{\text{sa}}\odot\mathcal{M}_{n}(\mathbb{C})_{\text{sa}}$. Also, $\mathds{1}_{m}\otimes i\big(F_{j,k}-F_{k,j}\big)\in\mathcal{M}_{m}(\mathbb{R})_{\text{sa}}\odot\mathcal{M}_{n}(\mathbb{C})_{\text{sa}}$. Observe
\begin{eqnarray}
&&\Big(\big(E_{r,s}-E_{s,r}\big)\otimes\big(F_{j,k}-F_{k,j}\big)\Big)\jProd\Big(\mathds{1}_{m}\otimes i\big(F_{j,k}-F_{k,j}\big)\Big)\nonumber\\
&=&i\big(E_{r,s}-E_{s,r}\big)\otimes\big((F_{j,k}-F_{k,j})(F_{j,k}-F_{k,j})\big)\nonumber\\
&=&i\big(E_{r,s}-E_{s,r}\big)\otimes\big(F_{j,j}-F_{k,k}\big)\in\mathcal{M}_{m}(\mathbb{R})_{\text{sa}}\odot\mathcal{M}_{n}(\mathbb{C})_{\text{sa}}
\end{eqnarray}
Of course, $\mathds{1}_{m}\otimes F_{k,k}\in\mathcal{M}_{m}(\mathbb{R})_{\text{sa}}\odot\mathcal{M}_{n}(\mathbb{C})_{\text{sa}}$. So, 
\begin{equation}
i\big(E_{r,s}-E_{s,r}\big)\otimes\big(F_{j,j}-F_{k,k}\big)\jProd(\mathds{1}_{m}\otimes F_{j,j})=i\big(E_{r,s}-E_{s,r}\big)\otimes F_{j,j}\in\mathcal{M}_{m}(\mathbb{R})_{\text{sa}}\odot\mathcal{M}_{n}(\mathbb{C})_{\text{sa}}\text{.}
\end{equation} 
Note $\sum_{j=1}^{n}F_{j,j}=\mathds{1}_{n}$. Thus, by $\mathbb{R}$-linear generation:
\begin{equation}
\textstyle{\frac{i}{\sqrt{2}}}\big(E_{r,s}-E_{s,r}\big)\otimes\mathds{1}_{n}\in\mathcal{M}_{m}(\mathbb{R})_{\text{sa}}\odot\mathcal{M}_{n}(\mathbb{C})_{\text{sa}}\text{.}
\end{equation}
Therefore, for \textit{any} $r,s\in\{1,\dots,m\}$, and \textit{any} $j,k\in\{1,\dots,n\}$, one has that $\{G_{r,s}\}\otimes\mathds{1}_{n}$ and $\mathds{1}_{m}\otimes\{H_{j,k}\}$ are both in the canonical tensor product. Of course, so too is $\mathds{1}_{m}\otimes\mathds{1}_{n}$.
\end{proof}
\end{lemma}

\begin{lemma}\label{comComLem} $\mathcal{M}_{m}(\mathbb{C})_{\text{sa}}\odot\mathcal{M}_{n}(\mathbb{C})_{\text{sa}}=\mathcal{M}_{nm}(\mathbb{C})_{\text{sa}}$\textit{.}
\begin{proof} Trivial, by definition $\mathcal{M}_{m}(\mathbb{C})_{\text{sa}}\otimes_{\mathbb{R}}\mathcal{M}_{n}(\mathbb{C})_{\text{sa}}\cong\mathcal{M}_{nm}(\mathbb{C})_{\text{sa}}$.
\end{proof}
\end{lemma}

\begin{lemma}\label{quatComLem} $\mathcal{M}_{m}(\mathbb{H})_{\text{sa}}\odot\mathcal{M}_{n}(\mathbb{C})_{\text{sa}}=\mathcal{M}_{2nm}(\mathbb{C})_{\text{sa}}$\textit{.}
\begin{proof} We begin with a quaternionic generalization of the Generalized Gell-Mann Matrices, which we will call the \textit{quaternionic Gell-Mann matrices}. Let $r,s\in\{1,\dots,m\}$. Define
\begin{eqnarray}
Q_{o_{r,s}}=\textstyle{\frac{1}{\sqrt{2}}}\big(E_{r,s}+E_{s,r}\big)&\hspace{0.3cm}\text{ for } r<s\\
Q_{i_{r,s}}=\textstyle{\frac{i}{\sqrt{2}}}\big(E_{r,s}-E_{s,r}\big)&\hspace{0.3cm}\text{ for } s<r\\
Q_{j_{r,s}}=\textstyle{\frac{j}{\sqrt{2}}}\big(E_{r,s}-E_{s,r}\big)&\hspace{0.3cm}\text{ for } s<r\\
Q_{k_{r,s}}=\textstyle{\frac{k}{\sqrt{2}}}\big(E_{r,s}-E_{s,r}\big)&\hspace{0.3cm}\text{ for } s<r\\
Q_{l_{r,r}}=\textstyle{\frac{1}{\sqrt{r(r+1)}}}\big(-rE_{r+1,r+1}+\sum_{k=1}^{r}E_{k,k}\big)&\hspace{0.3cm}\text{ for } r=s\neq m
\end{eqnarray}
We have defined $4m(m-1)/2+(m-1)=2m^2-m-1=m(2m-1)-1$ traceless $\mathbb{R}$-linearly independent self-adjoint $m\times m$ quaternionic matrices. Thus the cardinality of the $\mathbb{R}$-linearly independent set $\{Q_{o_{r,s}},Q_{i_{r,s}},Q_{j_{r,s}},Q_{k_{r,s}},Q_{l_{r,r}},\mathds{1}_{m}\}$ is $m(2m-1)$, which is equal to the dimension of $\mathcal{M}_{m}(\mathbb{H})_{\text{sa}}$. Exclamation point! Let us now consider the symplectic (\textit{i.e}.\ \textit{standard}) embedding of the quaternionic Gell-Mann matrices. Let $\pi$ be the symplectic embedding. Recall \cite{Graydon2011}:
\begin{equation}
\pi:\mathcal{M}_{m}(\mathbb{H})_{\text{sa}}\longrightarrow\mathcal{M}_{2m}(\mathbb{C})_{\text{sa}}::\Gamma_{1}+\Gamma_{2}j\longmapsto\begin{pmatrix} \;\;\;\Gamma_{1} & \Gamma_{2} \\ -\overline{\Gamma_{2}} & \overline{\Gamma_{1}} \end{pmatrix}\text{.}
\end{equation}
Let $F_{p,q}=|f_{p}\rangle\langle f_{q}|\in\mathcal{M}_{2m}(\mathbb{C})$ where $\big\{|f_{p}\rangle:j\in\{1,\dots,2m\}\big\}$ is the standard orthonormal basis for $\mathbb{C}^{2m}$. Then
\begin{eqnarray}
Q_{o_{r,s}}&\stackrel{\pi}{\longmapsto}&\textstyle{\frac{1}{\sqrt{2}}}\big(F_{r,s}+F_{s,r}+F_{r+m,s+m}+F_{s+m,r+m}\big)\\
Q_{i_{r,s}}&\stackrel{\pi}{\longmapsto}&\textstyle{\frac{i}{\sqrt{2}}}\big(F_{r,s}-F_{s,r}-F_{r+m,s+m}+F_{s+m,r+m}\big)\\
Q_{j_{r,s}}&\stackrel{\pi}{\longmapsto}&\textstyle{\frac{1}{\sqrt{2}}}\big(F_{r,s+m}-F_{s,r+m}-F_{r+m,s}+F_{s+m,r}\big)\\
Q_{k_{r,s}}&\stackrel{\pi}{\longmapsto}&\textstyle{\frac{i}{\sqrt{2}}}\big(F_{r,s+m}-F_{s,r+m}+F_{r+m,s}-F_{s+m,r}\big)\\
Q_{l_{r,r}}&\stackrel{\pi}{\longmapsto}&\textstyle{\frac{1}{\sqrt{r(r+1)}}}\big(-rF_{r+1,r+1}+\sum_{k=1}^{r}F_{k,k}-rF_{r+1+m,r+1+m}+\sum_{k=1}^{r}F_{k+m,k+m}\big)
\end{eqnarray}
Let $G_{p,q}$ be the Generalized Gell-Mann Matrices for $\mathcal{M}_{2m}(\mathbb{C})$; with $p,q\in\{1,\dots,2m\}$:
\begin{equation}
G_{p,q}=
\begin{cases}
\textstyle{\frac{1}{\sqrt{2}}}\big(F_{p,q}+F_{q,p}\big) & p<q \\
\textstyle{\frac{i}{\sqrt{2}}}\big(F_{p,q}-F_{q,p}\big) & q<p \\
\textstyle{\frac{1}{\sqrt{p(p+1)}}}\big(-pE_{p+1,p+1}+\sum_{k=1}^{p}E_{k,k}\big) & p=q\neq 2m
\end{cases}
\label{GM2m}
\end{equation}
Let $L_{t,v}=|g_{t}\rangle\langle g_{v}|\in\mathcal{M}_{n}(\mathbb{C})$ where $\big\{|g_{t}\rangle:t\in\{1,\dots,n\}\big\}$ is the standard orthonormal basis for $\mathbb{C}^{n}$. Let $H_{t,v}$ be the Generalized Gell-Mann Matrices for $\mathcal{M}_{n}(\mathbb{C})$; with $t,v\in\{1,\dots,n\}$:
\begin{equation}
H_{t,v}=
\begin{cases}
\textstyle{\frac{1}{\sqrt{2}}}\big(L_{t,v}+L_{v,t}\big) & t<v \\
\textstyle{\frac{i}{\sqrt{2}}}\big(L_{t,v}-L_{v,t}\big) & v<t \\
\textstyle{\frac{1}{\sqrt{t(t+1)}}}\big(-tL_{t+1,t+1}+\sum_{k=1}^{t}L_{k,k}\big) & t=v\neq n
\end{cases}
\label{GMn}
\end{equation}
Now, by definition $\mathcal{M}_{m}(\mathbb{H})_{\text{sa}}\odot\mathcal{M}_{n}(\mathbb{C})_{\text{sa}}\subseteq\mathcal{M}_{2nm}(\mathbb{C})_{\text{sa}}$. We will establish the reverse inclusion, explicitly put: $\mathcal{M}_{2nm}(\mathbb{C})_{\text{sa}}\subseteq\mathcal{M}_{m}(\mathbb{H})_{\text{sa}}\odot\mathcal{M}_{n}(\mathbb{C})_{\text{sa}}$, by proving that one can generate $\{G_{p,q}\}\otimes\mathds{1}_{n}$ --- obviously $\mathds{1}_{2m}\otimes\{H_{t,v}\}$ is in the canonical tensor product; indeed the real vector space tensor product to begin with. So the task becomes generating the $G_{p,q}$ from the images of the $m\times m$ self-adjoint quaternionic matrices under $\pi$. This will be tedious.

\noindent To start, we compute
\begin{eqnarray}
&&-8\Big(\pi\big(Q_{o_{r,s}}\big)\otimes H_{t<v,v}\Big)\jProd\Big(\pi\big(Q_{i_{r,s}}\big)\otimes H_{t>v,v}\Big)\nonumber\\
&=&\big(-F_{r,r}+F_{s,s}+F_{r+m,r+m}-F_{s+m,s+m}\big)\otimes\big(-L_{t,t}+L_{v,v}\big)\nonumber\\
&+&\big(F_{r,r}-F_{s,s}-F_{r+m,r+m}+F_{s+m,r+m}\big)\otimes\big(L_{t,t}-L_{v,v}\big)\nonumber\\
&=&\big(-F_{r,r}+F_{s,s}+F_{r+m,r+m}-F_{s+m,s+m}\big)\otimes\big(-L_{t,t}+L_{v,v}\big)\nonumber\\
&+&(-1)^{2}\big(-F_{r,r}+F_{s,s}+F_{r+m,r+m}-F_{s+m,s+m}\big)\otimes\big(-L_{t,t}+L_{v,v}\big)\nonumber\\
&=&2\big(-F_{r,r}+F_{s,s}+F_{r+m,r+m}-F_{s+m,s+m}\big)\otimes\big(-L_{t,t}+L_{v,v}\big)\nonumber\\
&\equiv&A\in\mathcal{M}_{m}(\mathbb{H})_{\text{sa}}\odot\mathcal{M}_{n}(\mathbb{C})_{\text{sa}}
\end{eqnarray}
Then, using $\sum_{v=1}^{n}L_{v,v}=\mathds{1}_{n}$, we compute
\begin{equation}
\textstyle{\frac{1}{4}}\sum_{v=1}^{n} A\jProd\big(\mathds{1}_{2m}\otimes L_{v,v}\big)=\big(-F_{r,r}+F_{s,s}+F_{r+m,r+m}-F_{s+m,s+m}\big)\otimes\mathds{1}_{n}\equiv B\in\mathcal{M}_{m}(\mathbb{H})_{\text{sa}}\odot\mathcal{M}_{n}(\mathbb{C})_{\text{sa}}\text{.}
\end{equation}
We will need the following image: $E_{r,r}\stackrel{\pi}\longmapsto F_{r,r}+F_{r+m,r+m}$. With this, we compute
\begin{equation}
B\jProd\Big(\big(F_{r,r}+F_{r+m,r+m}\big)\otimes\mathds{1}_{n}\Big)=\big(-F_{r,r}+F_{r+m,r+m}\big)\otimes\mathds{1}_{n}\equiv C\in\mathcal{M}_{m}(\mathbb{H})_{\text{sa}}\odot\mathcal{M}_{n}(\mathbb{C})_{\text{sa}}
\end{equation}
Thus
\begin{eqnarray}
\textstyle{\frac{1}{2}}\Big(C+\pi\big(E_{r,r}\big)\otimes\mathds{1}_{n}\Big)&=&F_{r+m,r+m}\otimes\mathds{1}_{n}\in\mathcal{M}_{m}(\mathbb{H})_{\text{sa}}\odot\mathcal{M}_{n}(\mathbb{C})_{\text{sa}}\label{Frpm}\\
-\textstyle{\frac{1}{2}}\Big(C-\pi\big(E_{r,r}\big)\otimes\mathds{1}_{n}\Big)&=&F_{r,r}\label{Fr}\otimes\mathds{1}\in\mathcal{M}_{m}(\mathbb{H})_{\text{sa}}\odot\mathcal{M}_{n}(\mathbb{C})_{\text{sa}}\text{.}
\end{eqnarray}
So
\begin{eqnarray}
2\Big(\pi\big(Q_{o_{r,s}}\big)\otimes\mathds{1}_{n}\Big)\jProd\Big(F_{r,r}\otimes\mathds{1}_{n}\Big)&=&\textstyle{\frac{1}{\sqrt{2}}}\big(F_{r,s}+F_{s,r}\big)\otimes\mathds{1}_{n}\in\mathcal{M}_{m}(\mathbb{H})_{\text{sa}}\odot\mathcal{M}_{n}(\mathbb{C})_{\text{sa}}\label{L1}\\
2\Big(\pi\big(Q_{o_{r,s}}\big)\otimes\mathds{1}_{n}\Big)\jProd\Big(F_{r+m,r+m}\otimes\mathds{1}_{n}\Big)&=&\textstyle{\frac{1}{\sqrt{2}}}\big(F_{r+m,s+m}+F_{s+m,r+m}\big)\otimes\mathds{1}_{n}\in\mathcal{M}_{m}(\mathbb{H})_{\text{sa}}\odot\mathcal{M}_{n}(\mathbb{C})_{\text{sa}}\label{L2}\\
2\Big(\pi\big(Q_{i_{r,s}}\big)\otimes\mathds{1}_{n}\Big)\jProd\Big(F_{r,r}\otimes\mathds{1}_{n}\Big)&=&\textstyle{\frac{i}{\sqrt{2}}}\big(F_{r,s}-F_{s,r}\big)\otimes\mathds{1}_{n}\in\mathcal{M}_{m}(\mathbb{H})_{\text{sa}}\odot\mathcal{M}_{n}(\mathbb{C})_{\text{sa}}\label{L3}\\
-2\Big(\pi\big(Q_{i_{r,s}}\big)\otimes\mathds{1}_{n}\Big)\jProd\Big(F_{r+m,r+m}\otimes\mathds{1}_{n}\Big)&=&\textstyle{\frac{i}{\sqrt{2}}}\big(F_{r+m,s+m}-F_{s+m,r+m}\big)\otimes\mathds{1}_{n}\in\mathcal{M}_{m}(\mathbb{H})_{\text{sa}}\odot\mathcal{M}_{n}(\mathbb{C})_{\text{sa}}\label{L4}
\end{eqnarray}
Next, we compute
\begin{eqnarray}
&&-8\Big(\pi\big(Q_{o_{r,s}}\big)\otimes H_{t<v,v}\Big)\jProd\Big(\pi\big(Q_{k_{r,s}}\big)\otimes H_{t>v,v}\Big)\nonumber\\
&=&\big(-F_{r,r+m}+F_{s,s+m}-F_{r+m,r}+F_{s+m,s}\big)\otimes\big(-L_{t,t}+L_{v,v}\big)\nonumber\\
&+&\big(F_{r,r+m}-F_{s,s+m}+F_{r+m,r}-F_{s+m,s}\big)\otimes\big(L_{t,t}-L_{v,v}\big)\nonumber\\
&=&\big(-F_{r,r+m}+F_{s,s+m}-F_{r+m,r}+F_{s+m,s}\big)\otimes\big(-L_{t,t}+L_{v,v}\big)\nonumber\\
&+&(-1)^{2}\big(-F_{r,r+m}+F_{s,s+m}-F_{r+m,r}+F_{s+m,s}\big)\otimes\big(-L_{t,t}+L_{v,v}\big)\nonumber\\
&=&2\big(-F_{r,r+m}+F_{s,s+m}-F_{r+m,r}+F_{s+m,s}\big)\otimes\big(-L_{t,t}+L_{v,v}\big)\nonumber\\
&\equiv& A'\in\mathcal{M}_{m}(\mathbb{H})_{\text{sa}}\odot\mathcal{M}_{n}(\mathbb{C})_{\text{sa}}\text{.}
\end{eqnarray}
Again, using $\sum_{v=1}^{n}L_{v,v}=\mathds{1}_{n}$, we compute
\begin{equation}
\textstyle{\frac{1}{4}}\sum_{v=1}^{n}A'\jProd\big(\mathds{1}_{2m}\otimes L_{v,v}\big)=\big(-F_{r,r+m}+F_{s,s+m}-F_{r+m,r}+F_{s+m,s}\big)\otimes\mathds{1}_{n}\equiv B'\in\mathcal{M}_{m}(\mathbb{H})_{\text{sa}}\odot\mathcal{M}_{n}(\mathbb{C})_{\text{sa}}\text{.}
\label{Bprime}
\end{equation}
From Eqs.~\eqref{L1} and \eqref{Fr} we then compute using Eq.~\eqref{Bprime}
\begin{eqnarray}
&&4\left(B'\jProd\textstyle{\frac{1}{\sqrt{2}}}\big(F_{r,s}+F_{s,r}\big)\otimes\mathds{1}_{n}\right)\jProd\Big(F_{r,r}\otimes\mathds{1}_{n}\Big)\nonumber\\
&=&2\left(\textstyle{\frac{1}{\sqrt{2}}}\big(-F_{r+m,s}+F_{s+m,r}+F_{r,s+m}-F_{s,r+m}\big)\otimes\mathds{1}_{n}\right)\jProd\Big(F_{r,r}\otimes\mathds{1}_{n}\Big)\nonumber\\
&=&\textstyle{\frac{1}{\sqrt{2}}}\big(F_{r,s+m}+F_{s+m,r}\big)\otimes\mathds{1}_{n}\in\mathcal{M}_{m}(\mathbb{H})_{\text{sa}}\odot\mathcal{M}_{n}(\mathbb{C})_{\text{sa}}\text{.}
\label{L5}
\end{eqnarray}
From Eqs.~\eqref{L2} and \eqref{Frpm} we then compute using Eq.~\eqref{Bprime}
\begin{eqnarray}
&&4\left(B'\jProd\textstyle{\frac{1}{\sqrt{2}}}\big(F_{r+m,s+m}+F_{s+m,r+m}\big)\otimes\mathds{1}_{n}\right)\jProd\Big(F_{r+m,r+m}\otimes\mathds{1}_{n}\Big)\nonumber\\
&=&2\left(\textstyle{\frac{1}{\sqrt{2}}}\big(-F_{r,s+m}+F_{s,r+m}+F_{r+m,s}-F_{s+m,r}\big)\otimes\mathds{1}_{n}\right)\jProd\Big(F_{r+m,r+m}\otimes\mathds{1}_{n}\Big)\nonumber\\
&=&\textstyle{\frac{1}{\sqrt{2}}}\big(F_{r+m,s}+F_{s,r+m}\big)\otimes\mathds{1}_{n}\in\mathcal{M}_{m}(\mathbb{H})_{\text{sa}}\odot\mathcal{M}_{n}(\mathbb{C})_{\text{sa}}\text{.}
\label{L6}
\end{eqnarray}
From Eqs.~\eqref{L3} and \eqref{Fr} we then compute using Eq.~\eqref{Bprime}
\begin{eqnarray}
&&4\left(B'\jProd\textstyle{\frac{i}{\sqrt{2}}}\big(F_{r,s}-F_{s,r}\big)\otimes\mathds{1}_{n}\right)\jProd\Big(F_{r,r}\otimes\mathds{1}_{n}\Big)\nonumber\\
&=&2\left(\textstyle{\frac{i}{\sqrt{2}}}\big(-F_{r+m,s}-F_{s+m,r}+F_{r,s+m}+F_{s,r+m}\big)\otimes\mathds{1}_{n}\right)\jProd\Big(F_{r,r}\otimes\mathds{1}_{n}\Big)\nonumber\\
&=&\textstyle{\frac{i}{\sqrt{2}}}\big(F_{r,s+m}-F_{s+m,r}\big)\otimes\mathds{1}_{n}\in\mathcal{M}_{m}(\mathbb{H})_{\text{sa}}\odot\mathcal{M}_{n}(\mathbb{C})_{\text{sa}}\text{.}
\label{L7}
\end{eqnarray}
From Eqs.~\eqref{L4} and \eqref{Frpm} we then compute using Eq.~\eqref{Bprime}
\begin{eqnarray}
&&4\left(B'\jProd\textstyle{\frac{i}{\sqrt{2}}}\big(F_{r+m,s+m}-F_{s+m,r+m}\big)\otimes\mathds{1}_{n}\right)\jProd\Big(F_{r+m,r+m}\otimes\mathds{1}_{n}\Big)\nonumber\\
&=&2\left(\textstyle{\frac{i}{\sqrt{2}}}\big(-F_{r,s+m}-F_{s,r+m}+F_{r+m,s}+F_{s+m,r}\big)\otimes\mathds{1}_{n}\right)\jProd\Big(F_{r+m,r+m}\otimes\mathds{1}_{n}\Big)\nonumber\\
&=&\textstyle{\frac{i}{\sqrt{2}}}\big(F_{r+m,s}-F_{s,r+m}\big)\otimes\mathds{1}_{n}\in\mathcal{M}_{m}(\mathbb{H})_{\text{sa}}\odot\mathcal{M}_{n}(\mathbb{C})_{\text{sa}}\text{.}
\label{L8}
\end{eqnarray}
Recalling that $r,s\in\{1,\dots,m\}$, Eqs.~\eqref{L1}\eqref{L2}\eqref{L3}\eqref{L4}\eqref{L5}\eqref{L6}\eqref{L7}\eqref{L8} establish that for any $p\neq q\in\{1,\dots,2m\}$ one has that $\{G_{p,q}\otimes\mathds{1}_{n}\big\}\in\mathcal{M}_{m}(\mathbb{H})_{\text{sa}}\odot\mathcal{M}_{n}(\mathbb{C})_{\text{sa}}$. Exclamation point! And with Eqs.~\eqref{Frpm} and \eqref{Fr}, we see that $\{G_{p,p}\otimes\mathds{1}_{n}\}$ is also in the canonical product. Indeed, this follows simply from $\mathbb{R}$-linear generation, because $F_{r,r}\otimes\mathds{1}_{n}$ and $F_{r+m,r+m}\otimes\mathds{1}_{n}$ are all in the canonical tensor product. So $\big\{\{G_{p,q}\}\otimes\mathds{1}_{n}\big\}\subset\mathcal{M}_{m}(\mathbb{H})_{\text{sa}}\odot\mathcal{M}_{n}(\mathbb{C})_{\text{sa}}$. Just to record our earlier remark: $\big\{\mathds{1}_{2m}\otimes \{H_{t,v}\}\big\}\subset\mathcal{M}_{m}(\mathbb{H})_{\text{sa}}\odot\mathcal{M}_{n}(\mathbb{C})_{\text{sa}}$. So too, of course, is $\mathds{1}_{2m}\otimes\mathds{1}_{n}$.
\end{proof}
\end{lemma}

\noindent With the proofs of the foregoing lemmas, our \cref{mgRes} follows immediately. 

\chapter{Jordanic Categories}
\label{categoriesEJA}

\epigraphhead[40]
	{
		\epigraph{``Outside in the distance a wildcat did growl\\Two riders were approaching, the wind began to howl''}{---\textit{Bob Dylan}\\ All Along The Watchtower (1968)}
	}

\noindent In \cite{Barnum2015} and \cite{Barnum2016b}, the author and H.\ Barnum and A.\ Wilce consider Jordan algebraic physical theories. This chapter centres on the categorical aspects of those publications.

\noindent Real quantum theory and quaternionic quantum theory date back to the work of Stueckelberg \cite{stueckelberg1960} and Finkelstein-Jauch-Speiser \cite{finkelstein1979}, respectively. In fact, much earlier, in \cite{birkhoff1936}, Birkhoff and von Neumann showed that the lattice of orthogonal projections onto Hilbert modules over $\mathbb{R}$, $\mathbb{C}$, and $\mathbb{H}$ could be taken as a models for experimental propositions. The formulation of a unified $\mathbb{R}$-$\mathbb{C}$ theory is mathematically obvious in light of the usual tensor product defined in \cref{tensorDef}, restricted to real subspaces when required. Our canonical tensor product (see \cref{def: canonical tensor product}) provides a mathematical apparatus to construct arbitrary composites for real, complex, \textit{and} quaternionic systems. In this chapter, we prove that the resulting theory admits the structure of a dagger compact closed category. We call this category $\mathbf{InvQM}$ because it involves involutions. Within $\mathbf{InvQM}$, the structure of complex composites is slightly different than in usual quantum theory: one obtains a direct sum of two copies of the usual tensor product. Of course, it is well known that quantum theory admits the structure of a dagger compact closed category \cite{Abramsky2009}\cite{Selinger2005}. Our work, therefore, provides a concrete example of a general probabilistic theory distinct from quantum theory with identical categorical structure to that of quantum theory. As we declared in \cref{prologue}, quantum theory is, however, the only subtheory wherein tomographic locality holds; moreover, preservation of purity fails outside of the real-complex subtheory. In light of dagger compact closure of $\mathbf{InvQM}$, tomographic locality and preservation of purity are evidently not required for sound categorical compositional structures.

\noindent Beyond $\mathbb{R}$-$\mathbb{C}$-$\mathbb{H}$ quantum theories within a Jordan algebraic framework, we are left with the spin factors $\mathcal{V}_{k}$ and the exceptional Jordan algebra $\mathcal{M}_{3}(\mathbb{O})_{\text{sa}}$. Our \cref{cor: no composite with exceptional} rules out the inclusion of the exceptional Jordan algebra in Jordanic theories enjoying our notion of composition given by \cref{def: new Jordan composite}. In this chapter, we shall see that the spin factors (save for the rebit $\mathcal{M}_{2}(\mathbb{R})_{\text{sa}}\cong\mathcal{V}_{2}$, the qubit $\mathcal{M}_{2}(\mathbb{C})_{\text{sa}}\cong\mathcal{V}_{3}$, and the quabit $\mathcal{M}_{2}(\mathbb{H})_{\text{sa}}\cong\mathcal{V}_{5}$) are ruled out on categorical grounds, specifically by \cref{ex: no states}. $\mathbf{InvQM}$ is, therefore, the largest possible categorical unification of Jordan algebraic physical theories respecting our notion of composition whilst retaining dagger compact closure. 

\noindent We structure the balance of this chapter as follows. In \cref{jpmSec}, we establish that the canonical tensor product is associative, and we describe the canonical tensor product of direct sums. In \cref{sec: categories EJC}, within the framework of EJC-algebras, we first generalize the usual notion of complete positivity. We then construct various categories of EJC-algebras.

\newpage
\section{Associativity and Direct Sums}\label{jpmSec}
\noindent In this section, we first prove that the canonical tensor product $\odot$ defined in \cref{def: canonical tensor product} is associative, setting the stage for the construction of symmetric monoidal categories of EJC-algebras in the sequel. We then describe the canonical tensor product of direct sums. Recall that an \textit{EJC}-\textit{algebra} is a pair $(\mathcal{A},\mathbf{M}_{\mathcal{A}})$ where $\mathbf{M}_{\mathcal{A}}$ is a finite dimensional complex *-algebra and $\mathcal{A}$ is a unital Jordan subalgebra of the self-adjoint part of $\mathbf{M}_{\mathcal{A}}$. For any subset $\mathcal{X}\subseteq\mathbf{M}_{\text{sa}}$ of the self-adjoint part of a finite dimensional complex *-algebra $\mathbf{M}$, recall that we write $\mathfrak{j}(\mathcal{X})$ for Jordan subalgebra of $\mathbf{M}_{\text{sa}}$ generated by $\mathcal{X}$. Put otherwise, $\mathfrak{j}(\mathcal{X})$ is the smallest Jordan subalgebra of $\mathbf{M}_{\text{sa}}$ containing the subset $\mathcal{X}$. We denote associative multiplication in $\mathbf{M}$ via juxtaposition, and we denote Jordan multiplication in $\mathbf{M}_{\text{sa}}$ via $\jProd$, \textit{i.e}.\ $x_{1}\jProd x_{2}=(x_{1}x_{2}+x_{2}x_{1})/2$. If $\mathcal{X}$ and $\mathcal{Y}$ are subsets of the self-adjoint parts of finite dimensional complex *-algebras $\mathbf{M}$ and $\mathbf{N}$, respectively, recall that then we define
\begin{equation}
\mathcal{X}\otimes\mathcal{Y}\equiv\text{span}_{\mathbb{R}}\Big\{x\otimes y\;\boldsymbol{|}\;x\in\mathcal{X} \text{ and } y\in\mathcal{Y}\Big\}\subseteq\big(\mathbf{M}\otimes\mathbf{N}\big)_{\text{sa}}\text{.}
\end{equation}
\begin{proposition}\label{newProp1} \textit{Let $\mathbf{M}$ be a finite dimensional complex *-algebra. Let subset $\mathcal{A}\subseteq\mathbf{M}_{\text{sa}}$. Let $u$ be the Jordan algebraic unit of $\mathfrak{j}(\mathcal{A})$, \textit{i.e}.\ $u\jProd x=x$ for all $x\in\mathfrak{j}(\mathcal{A})$. Then for all $x\in\mathfrak{j}(\mathcal{A})$ the ambient associative products $ux=x=xu$.}
\begin{proof}
Let $x\in\mathfrak{j}(\mathcal{A})$. Then according the premise of the proposition
\begin{equation}
x=u\jProd x=(ux+xu)/2\text{.}
\label{toHit}
\end{equation}
Notice that $u\jProd u=u$ so $u=uu$. Carrying out left associative multiplication by $u$ on Eq.~\eqref{toHit} thus yields
\begin{equation}
ux=uux+uxu/2=(ux+uxu)/2\implies ux=uxu\text{.}
\label{hitLeft}
\end{equation}
Carrying out right associative multiplication by $u$ on Eq.~\eqref{toHit} yields
\begin{equation}
xu=uxu+xuu/2=(uxu+xu)/2\implies xu=uxu\text{.}
\label{hitRight}
\end{equation}
In light of Eq.~\eqref{hitLeft} and Eq.~\eqref{hitRight} we have that $xu=ux$. Therefore $x=u\jProd x=(ux+xu)/2=ux$ and similarly $x=xu$.
\end{proof}
\end{proposition}
\begin{proposition}\label{newProp2}\textit{Let $\mathbf{M}$ be a finite dimensional complex *-algebra. Let subset $\mathcal{A}\subseteq\mathbf{M}_{\text{sa}}$. Let $e_{\mathcal{A}}\in\mathcal{A}$ be such that the ambient associative products $e_{\mathcal{A}}a=a=ae_{\mathcal{A}}$ for all $a\in\mathcal{A}$. Then for all $x\in\mathfrak{j}(\mathcal{A})$ the associative products $e_{\mathcal{A}}x=x=xe_{\mathcal{A}}$.}
\begin{proof} Let $x,y\in\mathcal{A}$. Let $\lambda\in\mathbb{R}$. Then 
\begin{equation}
e_{\mathcal{A}}(x\jProd y)=e_{\mathcal{A}}(xy+yx)/2=\big((e_{\mathcal{A}}x)y+(e_{\mathcal{A}}y)x\big)/2=(xy+yx)/2\text{.}
\label{jords}
\end{equation}
\begin{equation}
e_{\mathcal{A}}(x+y)=e_{\mathcal{A}}x+e_{\mathcal{A}}y=x+y\text{.}
\label{adds}
\end{equation}
\begin{equation}
e_{\mathcal{A}}(x\lambda)=(e_{\mathcal{A}}x)\lambda=x\lambda\text{.}
\label{scales}
\end{equation}
In light of Eq.~\eqref{jords} and Eq.~\eqref{adds} and Eq.~\eqref{scales} we see that left associative multiplication by $e_{\mathcal{A}}$ commutes with the Jordan product, vector space addition, and scalar multiplication operations on $\mathcal{A}$. Similarly, right associative multiplication commutes. Recalling that these are precisely the operations used to generate $\mathfrak{j}(\mathcal{A})$ from $\mathcal{A}$, we complete the proof.
\end{proof}
\end{proposition}

\begin{lemma} \label{lemma: associativity} \textit{Let} $\mathbf{M}$ \textit{and} $\mathbf{N}$ \textit{be finite dimensional complex *-algebras, and let} $\mathcal{A}$ \textit{and} $\mathcal{B}$ \textit{be subsets of} $\mathbf{M}_{\text{sa}}$ \textit{and} $\mathbf{N}_{\text{sa}}$\textit{, respectively.} \textit{Suppose} $\exists e_{\mathcal{A}}\in\mathcal{A}$ \textit{such that} $e_{\mathcal{A}} a=a=ae_{\mathcal{A}}$ $\forall a\in\mathcal{A}$ \textit{. Suppose} $\exists e_{\mathcal{B}}\in\mathcal{B}$ \textit{such that} $e_{\mathcal{B}}b=b=be_{\mathcal{B}}$ $\forall b\in\mathcal{B}$\textit{.} \textit{Then} $\mathfrak{j}(\mathcal{A}\otimes\mathfrak{j}(\mathcal{B}))=\mathfrak{j}(\mathcal{A}\otimes\mathcal{B})$=$\mathfrak{j}(\mathfrak{j}(\mathcal{A})\otimes\mathcal{B})$\textit{.}
\begin{proof}\footnote{Note that `$\subseteq$' means ``is a subset of.''}  We first show that $\mathfrak{j}(\mathcal{A}\otimes\mathcal{B})\subseteq\mathfrak{j}(\mathcal{A}\otimes\mathfrak{j}(\mathcal{B}))$. By definition we have $\mathcal{B}\subseteq\mathfrak{j}(\mathcal{B})$; hence $\mathcal{A}\otimes \mathcal{B}\subseteq \mathcal{A}\otimes\mathfrak{j}(\mathcal{B})$. Observe that if $\mathcal{X}\subseteq\mathcal{Y}\subseteq\mathbf{O}_{\text{sa}}$ are subsets of the self-adjoint part of a finite dimensional complex  *-algebra $\mathbf{O}$, then $\mathfrak{j}(\mathcal{X})\subseteq\mathfrak{j}(\mathcal{Y})$, because then by definition $\mathfrak{j}(\mathcal{X})=\mathfrak{j}\{y\in\mathcal{Y}\;\boldsymbol{|}\;y\in\mathcal{X}\}\subseteq\mathfrak{j}\{y\in\mathbf{O}_{\text{sa}}\;\boldsymbol{|}\;y\in\mathcal{Y}\}=\mathfrak{j}(\mathcal{Y})$. Therefore $\mathfrak{j}(\mathcal{A}\otimes\mathcal{B})\subseteq\mathfrak{j}(\mathcal{A}\otimes\mathfrak{j}(\mathcal{B}))$.\\[0.2cm]
\noindent We now show that $\mathfrak{j}(\mathcal{A}\otimes\mathfrak{j}(\mathcal{B}))\subseteq\mathfrak{j}(\mathcal{A}\otimes\mathcal{B})$. In light of our observation in the first part of this proof, it suffices to show that $\mathcal{A}\otimes\mathfrak{j}(\mathcal{B})\subseteq\mathfrak{j}(\mathcal{A}\otimes\mathcal{B})$, because $\mathfrak{j}(\mathcal{A}\otimes\mathcal{B})$ is closed, \textit{i.e.} $\mathfrak{j}(\mathfrak{j}(\mathcal{A}\otimes\mathcal{B}))=\mathfrak{j}(\mathcal{A}\otimes\mathcal{B})$. Notice that
\begin{equation}
e_{\mathcal{A}}\otimes\mathfrak{j}(\mathcal{B})=\big\{e_{\mathcal{A}}\otimes y\in(\mathbf{M}\otimes\mathbf{N})_{\text{sa}}\;\boldsymbol{|}\;y\in\mathfrak{j}(\mathcal{B})\big\}\subseteq\mathfrak{j}(\mathcal{A}\otimes\mathcal{B})
\label{eAjB}
\end{equation}
follows from bilinearity of $\otimes$ since $(e_{\mathcal{A}}\otimes b_{1})\jProd(e_{\mathcal{A}}\otimes b_{2})=(e_{\mathcal{A}}\otimes b_{1}b_{2}+e_{\mathcal{A}}\otimes b_{2}b_{1})/2=e_{\mathcal{A}}\otimes (b_{1}b_{2}+b_{2}b_{1})/2=e_{\mathcal{A}}\otimes (b_{1}\jProd b_{2})$.
Now, by definition,
\begin{equation}
\mathcal{A}\otimes e_{\mathcal{B}}=\big\{a\otimes e_{\mathcal{B}}\in(\mathbf{M}\otimes\mathbf{N})_{\text{sa}}\;\boldsymbol{|}\;a\in\mathcal{A}\big\}\subseteq\mathcal{A}\otimes\mathcal{B}\subseteq\mathfrak{j}(\mathcal{A}\otimes\mathcal{B})\text{.}
\end{equation}
$\mathfrak{j}(\mathcal{A}\otimes\mathcal{B})$ is by definition a Jordan algebra. In particular, the Jordan product of any two elements in $\mathfrak{j}(\mathcal{A}\otimes\mathcal{B})$ is again in $\mathfrak{j}(\mathcal{A}\otimes\mathcal{B})$. Now let $a\in\mathcal{A}$ and $y\in\mathfrak{j}(\mathcal{B})$. Then
\begin{equation}
(a\otimes e_{\mathcal{B}})\jProd(e_{\mathcal{A}}\otimes y)=(ae_{\mathcal{A}}\otimes e_{\mathcal{B}}y+e_{\mathcal{A}}a\otimes ye_{\mathcal{B}})/2=a\otimes y\in\mathfrak{j}(\mathcal{A}\otimes\mathcal{B})
\label{forConcT}
\end{equation}
follows from \cref{newProp2}; moreover $\mathfrak{j}(\mathcal{A}\otimes\mathcal{B})$ is by definition closed under the real linear span of such elements. We have shown that
\begin{equation}
\mathcal{A}\otimes\mathfrak{j}(\mathcal{B})=\text{span}_{\mathbb{R}}\big\{a\otimes y\in(\mathbf{M}\otimes\mathbf{N})_{\text{sa}}\;\boldsymbol{|}\;a\in\mathcal{A}\text{ and } y\in\mathfrak{j}(\mathcal{B})\big\}\subseteq\mathfrak{j}(\mathcal{A}\otimes\mathcal{B})\text{.}
\end{equation}
We therefore conclude
\begin{equation}
\mathfrak{j}(\mathcal{A}\otimes\mathcal{B})=\mathfrak{j}\big(\mathcal{A}\otimes\mathfrak{j}(\mathcal{B})\big)\text{.}
\end{equation}
A proof that $\mathfrak{j}(\mathcal{A}\otimes\mathcal{B})=\mathfrak{j}(\mathfrak{j}(\mathcal{A})\otimes\mathcal{B})$ follows \textit{mutatis mutandis}.
\end{proof}
\end{lemma}
\begin{proposition}\label{prop: associativity} \textit{Let $(\mathcal{A},\M_{\mathcal{A}})$, $(\mathcal{B},\M_{\mathcal{B}})$, and $(\mathcal{C},\M_{\mathcal{C}})$ be EJC algebras. Then the associator mapping of the symmetric monoidal category of finite dimensional complex *-algebras} $\alpha:\M_{\mathcal{A}}\otimes(\M_{\mathcal{B}}\otimes \M_{\mathcal{C}}) \longrightarrow(\M_{\mathcal{A}}\otimes\M_{\mathcal{B}})\otimes\M_{\mathcal{C}}$ \textit{carries} $\mathcal{A}\odot(\mathcal{B}\odot\mathcal{C})$ \textit{isomorphically onto} $(\mathcal{A}\odot\mathcal{B})\odot\mathcal{C}$\textit{.}
\begin{proof}Let $(\mathcal{A},\M_{\mathcal{A}})$,
$(\mathcal{B},\M_{\mathcal{B}})$ and $(\mathcal{C},\M_{\mathcal{C}})$ be EJC algebras. We need to show that the associator mapping $\alpha : \M_{\mathcal{A}} \otimes (\M_{\mathcal{B}}\otimes \M_{\mathcal{C}}) \longrightarrow (\M_{\mathcal{A}} \otimes \M_{\mathcal{B}}) \otimes \M_{\mathcal{C}}$ carries $\mathcal{A}\odot (\mathcal{B} \odot \mathcal{C})$ onto $(\mathcal{A} \odot \mathcal{B}) \odot \mathcal{C}$. $\mathcal{A}$ is a EJC algebra. Therefore $\exists e_{\mathcal{A}}\in\mathcal{A}$ such that $e_{\mathcal{A}}\jProd a=a$ for all $a\in\mathcal{A}$, namely the Jordan algebraic unit! By \cref{newProp1} we have the associative products $e_{\mathcal{A}}a=a=ae_{\mathcal{A}}$ for all $a\in\mathcal{A}$. $\mathcal{B}$ and $\mathcal{C}$ are also EJC algebras, with units that we shall denote by $e_{\mathcal{B}}$ and $e_{\mathcal{C}}$ respectively. Define $e_{\mathcal{BC}}=e_{\mathcal{B}}\otimes e_{\mathcal{C}}\in\mathcal{B}\otimes\mathcal{C}$. Then $e_{BC}Y=Y=Ye_{BC}$ for all $Y\in\mathcal{B}\otimes\mathcal{C}$. Similarly, by defining $e_{\mathcal{AB}}=e_{\mathcal{A}}\otimes e_{\mathcal{B}}\in\mathcal{A}\otimes\mathcal{B}$ we have the associative products $e_{AB}X=X=Xe_{AB}$ for all $X\in\mathcal{A}\otimes\mathcal{B}$. We may therefore apply \cref{lemma: associativity} as we deduce
\begin{equation}
\mathcal{A} \odot (\mathcal{B} \odot \mathcal{C}) = \mathfrak{j}(\mathcal{A} \otimes \mathfrak{j}(\mathcal{B} \otimes \mathcal{C})) = \mathfrak{j}(\mathcal{A} \otimes (\mathcal{B} \otimes \mathcal{C}))
\end{equation}
and 
\begin{equation}
(\mathcal{A}\odot \mathcal{B}) \odot \mathcal{C} = \mathfrak{j}(\mathfrak{j}(\mathcal{A} \otimes \mathcal{B}) \otimes \mathcal{C}) = \mathfrak{j}((\mathcal{A} \otimes \mathcal{B}) \otimes \mathcal{C})\text{.}
\end{equation} 
The associator mapping is a $\ast$-isomorphism, and carries $\mathcal{A} \otimes (\mathcal{B} \otimes \mathcal{C})$ to 
$(\mathcal{A} \otimes \mathcal{B}) \otimes \mathcal{C}$. Hence, it also carries $\mathfrak{j}(\mathcal{A} \otimes (\mathcal{B} \otimes \mathcal{C}))$ onto $\mathfrak{j}((\mathcal{A} \otimes \mathcal{B}) \otimes \mathcal{C})$. 
\end{proof}
\end{proposition}

\noindent Prior to contructing our categories in the following section, we note that if $(\mathcal{A},\M_{\mathcal{A}})$ and $(\mathcal{B},\M_{\mathcal{B}})$ are two EJC-algebras, we define 
\begin{equation}
(\mathcal{A},\M_{\mathcal{A}}) \oplus (\mathcal{B},\M_{\mathcal{B}})) = (\mathcal{A} \oplus \mathcal{B}, \M_{\mathcal{A}} \oplus \M_{\mathcal{B}})\text{,} 
\end{equation}
where the embedding of $\mathcal{A} \oplus \mathcal{B}$ in $\M_{\mathcal{A}} \oplus \M_{\mathcal{B}}$ is the obvious one. One can easily check that 
for sets $X \subseteq \M_{\mathcal{A}}$ and $Y \subseteq \M_{\mathcal{B}}$, $\mathfrak{j}(X \oplus Y) = \mathfrak{j}(X) \oplus \mathfrak{j}(Y)$. Using this, and 
the distributivity of tensor products over direct sums in the contexts of vector spaces and $\ast$-algebras, we have 
\begin{eqnarray*}
\mathcal{A} \odot (\mathcal{B} \oplus \mathcal{C}) = \mathfrak{j}(\mathcal{A} \otimes (\mathcal{B} \oplus \mathcal{C})) & = & \mathfrak{j}((\mathcal{A} \otimes \mathcal{B}) \oplus (\mathcal{A} \otimes \mathcal{C})) \\
& = & \mathfrak{j}(\mathcal{A} \otimes \mathcal{B}) \oplus \mathfrak{j}(\mathcal{A} \otimes \mathcal{C}) \\
& = & (\mathcal{A} \odot \mathcal{B}) \oplus (\mathcal{A} \odot \mathcal{C})
\end{eqnarray*} 
(where $\mathfrak{j}$ refers variously to generated Jordan subalgebras of $\M_{\mathcal{A}}
\otimes (\M_{\mathcal{B}} \oplus \M_{\mathcal{C}})$, $(\M_{\mathcal{A}} \otimes \M_{\mathcal{B}}) \oplus
(\M_{\mathcal{A}} \otimes \M_{\mathcal{C}})$, $\M_{\mathcal{A}} \otimes \M_{\mathcal{B}}$, and $\M_{\mathcal{A}} \otimes
\M_{\mathcal{C}}$.)
\section{Categories of EJC-algebras}
\label{sec: categories EJC}

\noindent In this section we construct various categories of EJC-algebras (see \cref{ejcDef}). An obvious candidate for a category in which objects are such algebras is the following.

\begin{definition}\label{cjpMorph} \textit{Let} $(\mathcal{A},\M_{\mathcal{A}})$ and $(\mathcal{B},\M_{\mathcal{B}})$ \textit{be EJC-algebras. A} Jordan preserving map \textit{is a linear function} $\phi:\M_{\mathcal{A}} \longrightarrow \M_{\mathcal{B}}$ \textit{such that} $\phi(\mathcal{A}) \subseteq \mathcal{B}$\textit{. The category} $\EJC$ \textit{has, as objects, 
EJC-algebras, and, as morphisms, completely positive Jordan-preserving maps.}
\end{definition} 

\noindent In view the associativity of $\odot$ (\cref{prop: associativity}), one might guess that $\EJC$ is symmetric-monoidal under $\odot$. There is certainly a natural choice for the monoidal unit, namely the EJC-algebra $\mathrm{I} = (\R, \C)$. However, the following propositions show that tensor products of $\EJC$ morphisms are generally not morphisms.

\begin{proposition}\label{ex: no states} {\em There exist simple, nontrivial, universally embedded EJC-algebras} $(\mathcal{A},\Cu(\mathcal{A}))$ \textit{and} 
$(\mathcal{B},\Cu(\mathcal{B}))$ such that $\alpha \otimes \id_{\mathcal{B}}$ \textit{is not Jordan-preserving for any state} $\alpha:\Cu(\mathcal{A})\longrightarrow\mathbb{C}$\textit{.}
\begin{proof} 
Suppose that $\mathcal{B}$ is not universally reversible. Suppose, further, that $\mathcal{A} \hotimes \mathcal{B}$ is irreducible. For instance, let $\mathcal{B}=\mathcal{V}_4$ and $\mathcal{A}=\mathcal{M}_{n}(\mathbb{R})_{\text{sa}}$ \cite{JamjoomTensorJW}. Let $\hat{\mathcal{B}}$ be the set of fixed points of the canonical involution $\Phi_{\mathcal{B}}$. Then by \cref{thm: simple composites ideals}, $\mathcal{A} \odot \mathcal{B} = \mathcal{A} \hotimes \mathcal{B}$, the set self-adjoint of fixed points of $\Phi_{\mathcal{A}} \otimes \Phi_{\mathcal{B}}$. In particular, $u_{\mathcal{A}} \otimes \hat{\mathcal{B}}$ is contained in $\mathcal{A} \odot \mathcal{B}$. Now let $\alpha$ be a state on $\Cu(\mathcal{A})$: this is completely positive, and trivially Jordan-preserving, and so, a morphism in $\EJC$. But \[(\alpha \otimes \id_{\mathcal{B}})(u_{\mathcal{A}} \otimes \hat{\mathcal{B}}) = \alpha(u_{\mathcal{A}})\hat{\mathcal{B}} = \hat{\mathcal{B}},\] which {\magenta by \cref{hoThm} is larger than $\mathcal{B}$ because $\mathcal{B}$ is not universally reversible}. So $\alpha\otimes \id_{\mathcal{B}}$ is not Jordan-preserving.
\end{proof}
\end{proposition}

\noindent Another proof of \cref{ex: no states} runs as follows. Consider the standard embeddings of $\mathcal{M}_{n}(\mathbb{C})_{\text{sa}}$ and $\mathcal{M}_{k}(\mathbb{R})_{\text{sa}}$, \textit{i.e}.\ consider the EJC-algebras $(\mathcal{M}_{n}(\mathbb{C})_{\text{sa}}, \mathcal{M}_{n}(\mathbb{C}))$ and $(\mathcal{M}_{k}(\mathbb{R})_{\text{sa}}, \mathcal{M}_{k}(\mathbb{C}))$. Then we have \[(\mathcal{M}_{n}(\mathbb{C})_{\text{sa}}, \mathcal{M}_{n}(\mathbb{C})) \odot (\mathcal{M}_{k}(\mathbb{R})_{\text{sa}}, \mathcal{M}_{k}(\mathbb{C}))= (\mathcal{M}_{nk}(\mathbb{C})_{\text{sa}},\mathcal{M}_{nk}(\mathbb{C}))\]
where the embedding of $\mathcal{M}_{nk}(\mathbb{C})_{\text{sa}}$ in $\mathcal{M}_{nk}(\mathbb{C})$ is the {\redd standard} one. Now let $\alpha$ be a state on $\mathcal{M}_{n}(\mathbb{C})$, and let $b$ be any self-adjoint matrix in $\mathcal{M}_{nk}(\mathbb{C})_{\text{sa}}$. Then $\mathds{1}_n \otimes b$ is self-adjoint in $\mathcal{M}_{n}(\mathbb{C}) \otimes \mathcal{M}_{k}(\mathbb{C}) = \mathcal{M}_{nk}(\mathbb{C})$, and $(\alpha \otimes \id)(\mathds{1}_n \otimes b) = b$. Since $b$ needn't belong to $\mathcal{M}_{k}(\mathbb{R})_{\text{sa}}$, the mapping $\alpha \otimes \id$ is not Jordan-preserving.

\noindent \cref{ex: no states} suggests the following adaptation of the notion of complete positivity to our setting.

\begin{definition}\label{cjpDef} \textit{A} completely Jordan preserving map \textit{is a linear function} $\phi:(\mathcal{A},\M_{\mathcal{A}})\longrightarrow (\mathcal{B},\M_{\mathcal{B}})$ \textit{such that for any EJC-algebra} $(\mathcal{C},\M_{\mathcal{C}})$ \textit{the function} $\phi \otimes\id_{\M_C}$ \textit{is positive and takes} $\mathcal{A} \odot \mathcal{C}$ \textit{into} $\mathcal{B} \odot \mathcal{C}$\textit{.}
\end{definition}

\noindent We note that \cref{cjpDef} implies that a completely Jordan preserving map $\phi$ is both Jordan-preserving (take $\mathcal{C} = \R$) and completely positive (take 
$\mathcal{C} = \mathcal{M}_{n}(\mathbb{C})_{\text{sa}}$ for any $n$.)

\begin{lemma}\label{lemma:completely Jordan preserving1}  \textit{If} $\phi:\M_{\mathcal{A}}\longrightarrow\M_{\mathcal{B}}$ \textit{is completely Jordan preserving, then for any} $(\mathcal{C},\M_{\mathcal{C}})$, $\phi \otimes \id_{\M_{\mathcal{C}}}$ \textit{is again completely Jordan preserving.}
\begin{proof} If $(\mathcal{D},\M_{\mathcal{D}})$ is another EJC-algebra, then consider $(\mathcal{C} \odot \mathcal{D}, \M_{\mathcal{C}} \otimes \M_{\mathcal{D}})$. The associativity of $\odot$ yields  
\[(\phi \otimes \id_{\M_{\mathcal{C}}}) \otimes \id_{\M_{\mathcal{D}}} = \phi \otimes (\id_{\M_{\mathcal{C}} \otimes \M_{\mathcal{D}}}).\]
Since $\phi$ is completely Jordan preserving, the latter {\redd carries} $\mathcal{A} \odot (\mathcal{C} \odot \mathcal{D})$ into $\mathcal{B} \odot (\mathcal{C} \odot \mathcal{D})$. Since $\odot$ is associative, this tells us that {\redd $(\phi \otimes \id_{\mathcal{C}}) \otimes \id_{\mathcal{D}}$ carries} $(\mathcal{A} \odot \mathcal{C}) \odot \mathcal{D}$  into $(\mathcal{B} \odot \mathcal{C}) \odot \mathcal{D}$. 
\end{proof}
\end{lemma}

\noindent While all completely Jordan preserving morphisms are completely positive, our proofs for \cref{ex: no states} show that the converse is false. On the other hand, the
class of completely Jordan preserving morphisms is still quite large.

\begin{proposition} \textit{Let $\phi : \M_{\mathcal{A}} \rightarrow \M_{\mathcal{B}}$ be a {\redd $\ast$-homomorphism} taking $\mathcal{A}$ to $\mathcal{B}$. Then $\phi$ is completely Jordan preserving.}
\begin{proof}
For the proof, note that $\phi \otimes \id_{\M_{\mathcal{C}}}$ is again a $\ast$-homomorphism, and hence, takes the Jordan subalgebra generated by $\mathcal{A} \otimes \mathcal{C}$ to that generated by $\mathcal{B} \otimes \mathcal{C}$, \textit{i.e}.\ sends $\mathcal{A} \odot \mathcal{C}$ into $\mathcal{B} \odot \mathcal{C}$. So $\phi$ is completely Jordan preserving.
\end{proof}
\end{proposition}

\begin{proposition} \textit{Let $(\mathcal{A},\M_{\mathcal{A}})$ be an EJC-algebra, and let $a \in \mathcal{A}$. The mapping $U_a : \M_{\mathcal{A}} \rightarrow \M_{\mathcal{A}}$ given by $U_a(b) = aba$ is completely Jordan preserving.}
\begin{proof}
For the proof, note that $U_{a}$ is completely positive, and can be expressed in terms of the Jordan product on $\M_{\mathcal{A}}$ as
\begin{equation}
U_a(b) = 2 a \jProd (a \jProd b) - (a^2)\jProd b\text{,}
\end{equation}
that is $U_a = 2 L_{a}^{2} - L_{a^2}$, where $L_a$ is the operator of left Jordan multiplication. Since $\mathcal{A}$ is a Jordan subalgebra of
$(\M_{\mathcal{A}})_{\text{sa}}$, $U_a(b) \in \mathcal{A}$ for all $a, b \in \mathcal{A}$. Thus, $U_a$ is a morphism.  Now if $(\mathcal{C},\M_{\mathcal{C}})$ is another EJC-algebra, we have
\begin{equation}
U_{a} \otimes U_{c} = U_{a \otimes c}
\end{equation}
for all $c \in \M_{\mathcal{C}}$; in particular, 
\begin{equation}
U_a \otimes \id_{\M_C} = U_a \otimes U_{\mathds{1}} = U_{a \otimes \mathds{1}}\text{.}
\end{equation}
Since $a \otimes \mathds{1} \in \mathcal{A} \odot \mathcal{C}$, it follows that $U_{a}$ is completely Jordan preserving.
\end{proof}
\end{proposition}

\begin{proposition} 
\label{prop:completely Jordan preserving composes and tensors} \textit{Let $(\mathcal{A},\M_{\mathcal{A}})$, {\redd $(\mathcal{A}', \M_{\mathcal{A}'})$, $(\mathcal{B},\M_{\mathcal{B}})$, $(\mathcal{B}', \M_{\mathcal{B}'})$ and and $(\mathcal{C},\M_{\mathcal{C}})$  be EJC-algebras.} Then}
\begin{itemize} 
\item[(i)] \textit{If} $\phi : \M_{\mathcal{A}}\longrightarrow \M_{\mathcal{B}}$ \textit{and} $\psi : \M_{\mathcal{B}} \longrightarrow \M_{\mathcal{C}}$ \textit{are completely Jordan preserving, then so is} $\psi \circ \phi$\textit{;}
\item[(ii)] \textit{If} $\phi : \M_{\mathcal{A}} \longrightarrow \M_{\mathcal{B}}$ and $\psi : \M_{\mathcal{A}'} \longrightarrow \M_{\mathcal{B}'}$ \textit{are completely Jordan preserving then so is} $\phi \otimes \psi : \M_{\mathcal{A} \odot \mathcal{A}'} = \M_{\mathcal{A}} \otimes \M_{\mathcal{A}'}\longrightarrow \M_{\mathcal{B}} \otimes \M_{\mathcal{B}'} = \M_{\mathcal{B} \odot \mathcal{B}'}$\textit{.} 
\end{itemize} 
\begin{proof}
(i) is obvious, and (ii) follows from (i) and \cref{lemma:completely Jordan preserving1} by noting that
\begin{equation}
\phi \otimes \psi = (\phi \otimes \id_{\M_{B'}}) \circ (\id_{\M_{A}} \otimes \psi)
\end{equation}
\end{proof}
\end{proposition}

\noindent In light of the foregoing propositions, EJC-algebras and completely Jordan preserving maps form a symmetric monoidal category, which we will call $\CJP$, with tensor unit $\mathrm{I} =(\R,\C)$. The category $\CJP$, however, has an undesirable feature: any positive mapping $f : \M_{\mathcal{A}}\longrightarrow \R$ is automatically Jordan preserving, yet, as \cref{ex: no states} illustrates, for a non-universally reversible $\mathcal{B}$, $f \otimes \id_{\mathcal{B}}$ {\redd need not be}. In other words, there {\redd need be} no completely Jordan preserving morphisms $\mathcal{A} \longrightarrow \mathrm{I}$. In particular, states of $\mathcal{A}$ {\redd need not be} completely Jordan preserving morphisms. In some cases, we can remedy this difficulty by relativising the definition of completely Jordan preserving mappings to a particular class $\mathscr{C}$ of EJC-algebras. In what follows, assume that $\mathscr{C}$ is closed under $\odot$ and contains the tensor unit $\mathrm{I} = (\R,\C)$.

\begin{definition}{\em Let $(\mathcal{A},\M_{\mathcal{A}})$ and $(\mathcal{B},\M_{\mathcal{B}})$ belong to $\mathscr{C}$. A positive linear mapping $\phi : \M_{\mathcal{A}}\longrightarrow \M_{\mathcal{B}}$ is {\em relatively completely Jordan preserving} with respect to $\mathscr{C}$ if, for all $(\mathcal{C},\M_{\mathcal{C}}) \in \mathscr{C}$, the mapping $\phi \otimes \id_{\M_{\mathcal{C}}}$ is positive and maps $\mathcal{A} \odot \mathcal{C}$ into $\mathcal{B} \odot \mathcal{C}$. We denote the set of all such maps by $\CJP_{\mathscr{C}}(\mathcal{A},\mathcal{B})$.} 
\end{definition}

\noindent If $\phi$ is relatively completely Jordan preserving with respect to $\mathscr{C}$, this does not imply that $\phi$ is completely positive, unless 
$\mathscr{C}$ contains $\mathcal{M}_{n}(\mathbb{C})_{\text{sa}}$ for every $n \in\mathbb{N}$. Nevertheless, exactly as in the proof of \cref{prop:completely Jordan preserving composes and tensors}, we see that if $\mathcal{A},\mathcal{B},\mathcal{C} \in \mathscr{C}$ and $\phi \in \CJP_{\mathscr{C}}(\mathcal{A},\mathcal{B})$ and $\psi \in \CJP_{\mathscr{C}}(\mathcal{B},\mathcal{C})$, then $\psi \circ \phi \in \CJP_{\mathscr{C}}(\mathcal{A},\mathcal{C})$, and also that if $\mathcal{A},\mathcal{B},\mathcal{C},\mathcal{D} \in \mathscr{C}$ and 
$\phi \in \CJP_{\mathscr{C}}(\mathcal{A},\mathcal{B})$ and $\psi \in \CJP_{\mathscr{C}}(\mathcal{C},\mathcal{D})$, then $\phi \otimes \psi \in \CJP_{\mathscr{C}}(\mathcal{A} \otimes \mathcal{C},  \mathcal{B} \otimes \mathcal{D})$. In other words, $\CJP_{\mathscr{C}}$ becomes\footnote{We prove this explicitly in \cref{SMCproofApp}.} a symmetric monoidal category with relatively completely Jordan preserving mappings with respect to $\mathscr{C}$ as morphisms. We denote this category by $\CJP_{\mathscr{C}}$. Here are three important ones. 

\begin{definition}[The category $\CQM$]\label{ex: CQM} 
{\em Let $\mathscr{C}$ be the class of hermitian parts of 
complex $*$-algebras with standard embeddings. Then $\phi$ 
belongs to $\mathbf{CJP}_{\mathscr{C}}$ 
iff $\phi$ is completely positive.  Evidently, {\redd the category $\mathbf{CJP}_{\mathscr{C}}$ for this choice of $\mathcal{C}$} is essentially orthodox, mixed-state quantum mechanics with
  superselection rules. From now on, we shall call this category $\CQM$.
  }
\end{definition} 

\begin{definition}[The category $\RSE$] \label{ex: RSE}{\em Let $\mathscr{C}'$ be the class of reversible EJCs with standard embeddings. We denote the category $\CJP_{\mathscr{C}'}$ by $\RSE$.}
\end{definition}

\noindent By \cref{mgRes}, plus the distributivity of $\odot$ over direct sums, $\RSE$ is closed under $\odot$. The category $\RSE$ represents a kind of unification of finite-dimensional real, complex and quaternionic quantum mechanics, in so far as its objects are the Jordan algebras associated with real, complex and quaternionic quantum systems, and direct sums of these. Moreover, as restricted to complex (or real) systems, its compositional structure is the standard one. However, our discussion following \cref{ex: no states} shows that not every quantum-mechanical process --- in particular, not even processes that prepare states --- will count as a morphism in $\RSE$, 
so this unification comes at a high cost. 
 
\begin{definition}[The category $\URUE$]\label{ex: URUE} {\em Let $\mathscr{C}''$ be the class of universally reversible EJCs with universal embeddings. We denote the category $\CJP_{\mathscr{C}''}$ by $\URUE$.}
\end{definition} 

\noindent $\URUE$ is closer to being a legitimate ``unified" quantum theory, although it omits the quaternionic bit (which is reversible, but not universally so). Even as restricted to complex quantum systems, however, it differs from orthodox quantum theory in two interesting ways. First, and most conspicuously, the tensor product is not the usual one: $\mathcal{M}_{n}(\mathbb{C})_{\text{sa}} \hotimes\mathcal{M}_{k}(\mathbb{C})_{\text{sa}} = \mathcal{M}_{nk}(\mathbb{C})_{\text{sa}} \oplus \mathcal{M}_{nk}(\mathbb{C})_{\text{sa}}$, rather than $\mathcal{M}_{nk}(\mathbb{C})_{\text{sa}}$. Secondly, it allows some processes that orthodox QM does not. Any completely positive mapping $\phi : \Cu(\mathcal{A}) \rightarrow \Cu(\mathcal{B})$ intertwining the involutions $\Phi_{\mathcal{A}}$ and $\Phi_{\mathcal{B}}$ (that is, $\Phi_{\mathcal{B}} \circ \phi = \phi \circ \Phi_{\mathcal{A}}$) is relatively completely Jordan preserving, i.e., a morphism in $\URUE$. In particular, the mapping on $\Cu(\mathcal{M}_{n}(\mathbb{C})_{\text{sa}}) = \mathcal{M}_{n}(\mathbb{C}) \oplus \mathcal{M}_{n}(\mathbb{C})$ that swaps the two summands is a morphism. Since the image of $\mathcal{M}_{n}(\mathbb{C})_{\text{sa}}$ in $\Cu(\mathcal{M}_{n}(\mathbb{C})_{\text{sa}})$ consists of pairs $(a,a^T)$, this mapping effects the transpose automorphism on $\mathcal{M}_{n}(\mathbb{C})_{\text{sa}}$. This is not permitted in orthodox quantum theory, as the transpose is not a completely positive mapping on $\mathcal{M}_{n}(\mathbb{C})$. In spite of its divergences from orthodoxy, $\URUE$ is in many respects a well behaved probabilistic theory. We shall now improve on it, by considering a category of EJC-algebras having a slightly larger class of objects (in particular, it includes $\mathcal{M}_{2}(\mathbb{H})_{\text{sa}}$), but a slightly more restricted set of morphisms. 
 
\noindent Recall that we write $\M^{\Phi}$ for the set of fixed-points of an involution $\Phi$  on a complex $\ast$-algebra $\M$. One can show that the involution $\Phi$ on $\mathcal{M}_{n}(\mathbb{C})$ with  $\mathcal{M}_{n}(\mathbb{R})_{\text{sa}} = \mathcal{M}_{n}(\C)^{\Phi}_{\text{sa}}$ (the transpose) and  on $\mathcal{M}_{2n}(\C)$ with $\mathcal{M}_{n}(\mathbb{H})_{\text{sa}} = \mathcal{M}_{2n}(\C)^{\Phi}_{\text{sa}}$, namely $\Phi(a) = -Ja^T J$ where $J$ is the unitary in Eq.~\eqref{jMatrix}, are both unitary with respect to the 
trace inner product $\langle a, b \rangle := \Tr(a^{*}b)$. The canonical involution on $\Cu(\mathcal{M}_{n}(\mathbb{C})_{\text{sa}}) = \mathcal{M}_{n}(\mathbb{C}) \oplus \mathcal{M}_{n}(\mathbb{C})$, namely, 
the mapping $(a,b) \longmapsto (b^T, a^T)$, {\redd is likewise unitary}. We relegate explicit proofs of these statements to \cref{epSec}. Presently, let us introduce the following definitions, and we remind the reader of \cref{existUThm}.

\begin{definition} \textit{An} involutive EJC-algebra \textit{is a triple} $(\mathcal{A},\M_{\mathcal{A}},\Phi)$ \textit{where} $(\mathcal{A},\M_{\mathcal{A}})$ \textit{is an EJC-algebra, and where $\Phi$ is a unitary involution on $\M_{\mathcal{A}}$ with $\mathcal{A} = (\M_{\mathcal{A}})^{\Phi}_{\text{sa}}$.}
\end{definition}
\begin{definition}\label{invQMDef}
\textit{$\InvQM$ is the category where the objects are involutive EJC-algebras, and where the morphisms $\phi:\mathcal{A}\longrightarrow\mathcal{B}$ completely positive mappings from $\M_{\mathcal{A}}$ to  $\M_{\mathcal{B}}$ intertwining $\Phi_{\mathcal{A}}$ and $\Phi_{\mathcal{B}}$, i.e, $\Phi_{\mathcal{B}} \circ \phi = \phi \circ \Phi_{\mathcal{A}}$.}                        
\end{definition} 

\noindent As we pointed out earlier, the condition that $\mathcal{A}$ be the set of self-adjoint fixed points of an involution makes $\mathcal{A}$ reversible in $\M_{\mathcal{A}}$. Thus, the class of involutive EJC-algebras contains no ``higher'' ($n = 4$ or $n > 5$) spin factors. In fact, it contains exactly {\redd direct sums of} the universally embedded universally reversible EJCs $(\mathcal{A},\Cu(\mathcal{A}))$ where $\mathcal{A}=\mathcal{M}_{n}(\mathbb{R})_{\text{sa}}$ or $\mathcal{M}_{n}(\mathbb{C})_{\text{sa}}$, with $n$ arbitrary, or $\mathcal{M}_{n}(\mathbb{H})_{\text{sa}}$ with $n  > 2$, {\em together with} the {\em standardly} embedded quabit, \textit{i.e}.\ $(\mathcal{M}_{2}(\mathbb{H})_{\text{sa}},\mathcal{M}_{4}(\C))$, with $\mathcal{M}_{2}(\mathbb{H})_{\text{sa}}$ the self-adjoint fixed points of the involution $\Phi(a) = -Ja^{T}J$, where $J$ is as in Eq.~\eqref{jMatrix}.  In other words, $\InvQM$ includes exclusively systems over the three division rings $\R, \C$ and $\Q$, albeit with the complex systems represented in their universal embeddings. Note that $(\R, \mathcal{M}_{1}(\R)_{\text{sa}})$ counts as an involutive EJC: since $\mathcal{M}_{1}(\R)_{\text{sa}} = \R$ is commutative, the identity map provides the necessary involution.  

\noindent It is easy to see that completely positive mappings $\M_{\mathcal{A}} \rightarrow \M_{\mathcal{B}}$ intertwining $\Phi_{\mathcal{A}}$ and $\Phi_{\mathcal{B}}$ are automatically relatively completely Jordan preserving for the class of involutive EJCs. Indeed, by \cref{cor: canonical composites with involutions}
the canonical tensor product of involutive EJCs is again involutive, as
\begin{equation} 
\mathcal{A} \odot \mathcal{B} = (\M_{\mathcal{A}} \otimes \M_{\mathcal{B}})^{\Phi_{\mathcal{A}} \otimes \Phi_{\mathcal{B}}}_{\text{sa}}
\label{forNewProp9}
\end{equation}
for all $\mathcal{A}, \mathcal{B}\in\InvQM$, so one has the following.
\begin{proposition}$\mathbf{InvQM}$\textit{ morphisms are relatively completely Jordan preserving with respect to the class of objects} $\text{ob}(\mathbf{InvQM})$\textit{.}
\begin{proof}
Let $\text{hom}(\mathbf{InvQM})\ni\phi:\mathcal{A}\longrightarrow\mathcal{B}$. By \cref{invQMDef} $\phi$ is a completely positive linear map from $\mathbf{M}_{\mathcal{A}}$ into $\mathbf{M}_{\mathcal{B}}$, so $\phi$ is of course then positive. Now, let $(\mathcal{C},\mathbf{M}_{\mathcal{C}},\Phi_{\mathcal{C}})$ be any involutive EJC-algebra. Then from our previous observation $\phi\otimes\text{id}_{\mathbf{M}_{\mathcal{C}}}$ is again positive. Furthermore, with arbitary $X\in\mathcal{A}\odot\mathcal{C}$ one has
\begin{eqnarray}
\Phi_{\mathcal{B}}\otimes\Phi_{\mathcal{C}}\Big(\phi\otimes\text{id}_{\mathbf{M}_{\mathcal{C}}}(X)\Big)&=&\big(\Phi_{\mathcal{B}}\circ\phi\big)\otimes\big(\Phi_{\mathcal{C}}\circ\text{id}_{\mathbf{M}_{\mathcal{C}}}\big)(X)\nonumber\\
&=&\big(\phi\circ\Phi_{\mathcal{A}}\big)\otimes\big(\text{id}_{\mathbf{M}_{\mathcal{C}}}\circ\Phi_{\mathcal{C}}\big)(X)\nonumber\\
&=&\phi\otimes\text{id}_{\mathbf{M}_{\mathcal{C}}}\Big(\Phi_{\mathcal{A}}\otimes\Phi_{\mathcal{C}}(X)\Big)\nonumber\\
&=&\phi\otimes\text{id}_{\mathbf{M}_{\mathcal{C}}}(X)
\label{insideNewProp9}
\end{eqnarray}
where the final equality comes from \cref{forNewProp9}. So, again from \cref{forNewProp9}, Eq.~\eqref{insideNewProp9} shows that $\phi\otimes\text{id}_{\mathbf{M}_{\mathcal{C}}}::\mathcal{A}\odot\mathcal{C}\longrightarrow\mathcal{B}\odot\mathcal{C}$.
\end{proof}
\end{proposition} 
\noindent Since invoutive EJCs are reversible, \cref{prop: canonical product composite} implies 
that $\mathcal{A} \odot \mathcal{B}$ is a dynamical composite of $\mathcal{A}$ and $\mathcal{B}$. Composites and tensor products of intertwining completely positive maps 
are also such {\redd (as are associators, unit-introductions and the swap mapping)}, so $\InvQM$ is a symmetric monoidal category --- indeed, a {\redd monoidal} subcategory of the category of involutive EJC-algebras and relatively completely Jordan preserving maps.

\noindent In the special case in which $\mathcal{A}$ and $\mathcal{B}$ are universally embedded complex qantum systems, say 
$\mathcal{A}=\mathcal{M}_{n}(\mathbb{C})_{\text{sa}}$ and $\mathcal{B}=\mathcal{M}_{k}(\mathbb{C})_{\text{sa}}$, we have $\M_{\mathcal{A}} = \Cu(\mathcal{M}_{n}(\mathbb{C})_{\text{sa}}) = \mathcal{M}_{n}(\mathbb{C}) \oplus \mathcal{M}_{n}(\mathbb{C})$ and 
similarly $\M_{\mathcal{B}} = \mathcal{M}_{k}(\mathbb{C}) \oplus \mathcal{M}_{k}(\mathbb{C})$. The involutions $\Phi_{\mathcal{A}}$ are given by $\Phi_{\mathcal{A}}(a,b) = (b^T, a^T)$, and similarly for $\Phi_{\mathcal{B}}$. In this case, the interwining completely positive-maps are sums of mappings of the two forms: $(a, b) \longmapsto (\phi(a), \phi^T(b))$ and $(a, b) \longmapsto (\phi^T(b), \phi(a))$, where $\phi$ is a completely positive mapping $\mathcal{M}_{n}(\mathbb{C})\longrightarrow \mathcal{M}_{k}(\mathbb{C})$ and $\phi^T$ is determined by the condition $\phi^T(x^T) = (\phi(x))^T$, \textit{i.e}.\ $\phi^T := T \circ \phi \circ T$. 

\noindent Two special cases of $\InvQM$-morphisms are worth emphasising on physical grounds. 

\begin{corollary}\label{cor: things that are InvQM morphisms} \textit{Let $(\mathcal{A},\M_{\mathcal{A}},\Phi_{\mathcal{A}})$ belong to $\InvQM$. Then}
\begin{itemize} 
\item[(a)] \textit{for every} $a \in \mathcal{A}$\textit{, the corresponding linear mapping} $a : \R \longrightarrow \M_{\mathcal{A}}$ \textit{determined by} $a(1) = a$\textit{, belongs to }$\InvQM(\mathrm{I},\mathcal{A})$\textit{;} 
\item[(b)] \textit{every positive linear functional (in particular, every state) on} $\M_{\mathcal{A}}$ \textit{of the form} 
$| a \rangle$,  $a \in \mathcal{A}_{+}$, belongs to $\InvQM(\mathcal{A},\mathrm{I})$ 
\end{itemize}
\begin{proof}
{\redd Part (a) is immediate from the fact that $a$ is fixed by 
$\Phi_{A}$}; part (b) {\redd also} follows from this, plus the unitarity of $\Phi_A$. 
\end{proof}
\end{corollary}

\noindent The category $\InvQM$ provides a unification of finite-dimensional real, complex and quaternionic quantum mechanics, but with the same  important caveats that apply to $\URUE$: the representation of orthodox, complex quantum systems $\mathcal{M}_{n}(\mathbb{C})_{\text{sa}}$ in $\InvQM$ is through the universal embedding 
\begin{equation}
\psi : \mathcal{M}_{n}(\mathbb{C})_{\text{sa}} \longrightarrow \Cu(\mathcal{M}_{n}(\mathbb{C})_{\text{sa}}) = \mathcal{M}_{n}(\mathbb{C}) \oplus \mathcal{M}_{n}(\mathbb{C})::a \longmapsto(a,a^T)\text{.}
\end{equation} 
\noindent As a consequence, the composite of two complex quantum systems in $\InvQM$ is a direct sum of two copies of their standard composite --- equivalently, is the standard composite, combined with a classical bit.  Moreover, the mapping that swaps the direct summands of $\Cu(\mathcal{M}_{n}(\mathbb{C})_{\text{sa}})$, a perfectly good morphism in $\InvQM$, acts as the transpose on $\psi(a) = (a, a^T)$. We now move to compact closure.

\noindent In \cite{Abramsky-Coecke}, it is shown that a large number of information-processing protocols, including 
{\magenta in particular} conclusive teleportation and entanglement-swapping, hold in any compact closed symmetric monoidal category, if we interpret objects as systems and morphisms as physically allowed processes. We shall now see that our category $\InvQM$ is compact closed. More exactly, we
shall show that it inherits a compact structure from the {\magenta natural compact structure on the} category $\stalg$ of finite-dimensional complex $\ast$-algebras, {\redd which we now review}.

\noindent Recall \cref{ccCatDef}. A \textit{compact structure} on a symmetric monoidal category $\Cat$ is a choice, for every object $A \in \Cat$, of a {\em dual object}: a triple $(\mathcal{A}',\eta_{\mathcal{A}},\epsilon_{\mathcal{A}})$ consisting of an object $\mathcal{A}' \in \C$, a {\em co-unit} $\eta_A : A \otimes A' \rightarrow I$ and a {\em unit} $\epsilon_{\mathcal{A}} : \mathrm{I} \rightarrow \mathcal{A}' \otimes \mathcal{A}$ obeying the commutative diagrams expressed in  \cref{ccCatDef}.\footnote{Our usage is slightly perverse. The usual convention is to denote the {\em unit} by $\eta_A$ and the 
co-unit by $\epsilon_A$. Our choice is motivated in part by the desire to represent states as morphisms $A \longrightarrow \mathrm{I}$ 
and effects as morphisms $\mathrm{I} \longrightarrow A$, rather than the reverse, together with the convention that 
takes the unit to correspond to the maximally entangled state.}

\noindent If $\M$ is a finite-dimensional complex $\ast$-algebra, let $\Tr$ denote the
canonical trace on $\M$, regarded as acting on itself by left
multiplication (so that $\Tr(a) = \tr(L_a)$, $L_a : \M \rightarrow \M$
being $L_a(b) = ab$ for all $b \in \M$). This induces an inner product\ffootnote{\red double check!! Checked.}
on $\M$, given by $\langle a , b \rangle_{\M} = \Tr(ab^{\ast})$\footnote{We are now following the convention 
that complex inner products are linear in the first argument.}.  Note
that this inner product is self-dualizing, \textit{i.e}.\ $a \in \M_+$ if and only if 
$\langle a , b \rangle \geq 0$ for all $b \in \M_+$. 

\noindent Now let $\bar{\M}$ be the conjugate algebra, writing $\bar{a}$ for $a \in \M$ when regarded as belonging to $\bar{\M}$ (so that $\bar{ca} =
\bar{c}~\bar{a}$ for scalars $c \in \C$ and $\bar{a} \bar{b} =
\bar{ab}$ for $a, b \in \M$).  Note that $\langle \bar{a}, \bar{b} \rangle = \langle b ,
a \rangle$.  Now define
\begin{equation}
\epsilon_{\M} = \sum_{e \in E} e \otimes \bar{e} \in \M \otimes \bar{\M}
\end{equation}
where $E$ is any orthonormal basis for $\M$ with respect to
$\IP_{\M}$. {\redd Then straightforward computations show that $\epsilon_{\M} \in (\M \otimes \bar{\M})_{+}$, and that, 
for all $a, b \in \M$,
\begin{equation} 
\langle a \otimes \bar{b}, \epsilon_{\M} \rangle = \langle a, b \rangle = \Tr(ab^{\ast})\text{,}
\end{equation}
where the inner product on the left is the trace inner product on $\M \otimes \M^{\ast}$. Now define 
$\eta_{\M} : \M \otimes \bar{\M} \rightarrow \C$ by $\eta_{\M} = | \epsilon_{\M} \rangle$, noting that this functional 
is positive (so, up to normalization, a state) since $\epsilon_{\M}$ is positive in $\M \otimes \bar{\M}$. }

\tempout{
Then a computation shows  that
\[ \langle (a \otimes \bar{b}) f_{\M} | f_{\M} \rangle_{\M \otimes \bar{\M}} = \langle a | b \rangle_{\M}.\]
Since the left-hand side defines a positive linear functional on $\M
\otimes \bar{\M}$, so does the right (remembering here that pure
tensors generate $\M \otimes \bar{\M}$, as we're working in finite
dimensions).  Call this functional $\eta_{\M}$. That is,
\[\eta_{\M} : \M \otimes \bar{\M} \rightarrow \C \ \ \mbox{is given by} \ \ \eta_{\M}(a \otimes \bar{b}) = \langle a | b \rangle = \Tr(ab^{\ast})\]
and is, up to normalization, a state on $\M \otimes \bar{\M}$. A further computation now shows that 
\[\langle a \otimes \bar{b} | f_{\M} \rangle_{\M \otimes \bar{\M}} = \eta(a \otimes \bar{b}).\]
 It follows that $f_{\M}$ belongs to the positive cone of
$\M \otimes \bar{\M}$, by self-duality of the latter. 
}

\noindent A final computation  shows that, for any states $\alpha$ and
$\bar{\alpha}$ on $\M$ and $\bar{\M}$, respectively, and any $a \in
\M, \bar{a} \in \bar{\M}$, we have
\begin{eqnarray}
(\eta_{\M} \otimes \alpha)(a \otimes \epsilon_{\bar{\M}}) = \alpha(a)\\
(\bar{\alpha} \otimes \eta_{\M})(\epsilon_{\bar{M}} \otimes \bar{a}) = \bar{\alpha}(\bar{a})\text{.}
\end{eqnarray}
Thus, $\eta_{\M}$ and $\epsilon_{\bar{\M}}$ define a compact structure on $\stalg$, for which the dual object of $\M$ is given by $\bar{\M}$.

\begin{definition} {\em The} conjugate \textit{of a EJC-algebra $(\mathcal{A},\M_{\mathcal{A}},\Phi_{\mathcal{A}})$ is $(\bar{\mathcal{A}},\bar{\M}_{\mathcal{A}},\bar{\Phi})$, where 
$\bar{\mathcal{A}} = \{ \bar{a} | a \in \mathcal{A}\}$. We write $\eta_{\mathcal{A}}$ for $\eta_{\M_{\mathcal{A}}}$ and $\epsilon_{\mathcal{A}}$ for $\epsilon_{\M_{\mathcal{A}}}$.}
\end{definition} 
\noindent Any linear mapping $\phi : \M \longrightarrow \N$ between $\ast$-algebras $\M$ and $\N$ gives rise to a 
linear mapping 
\begin{equation}
\bar{\phi} : \bar{\M} \longrightarrow \bar{\N}::\bar{a}\longmapsto \bar{\phi(a)}\text{,}
\end{equation} 
for $a \in \M$. It is straightforward {\redd to show} that if $\Phi$ is a unitary involution on $\M_{\mathcal{A}}$ with $\mathcal{A} = {\M_{\mathcal{A}}}^{\Phi}_{\text{sa}}$, then $\bar{\Phi} : \bar{\M} \longrightarrow \bar{\M}$ is also a unitary involution with
\begin{equation}
{\bar{\M}_{\mathcal{A}}}^{\bar{\Phi}}_{\text{sa}} = \bar{\mathcal{A}}\text{.}
\end{equation} 
Thus, the class of involutive EJCs is closed under the formation of conjugates. 

\begin{lemma} \textit{Let} $(\mathcal{A},\M_{\mathcal{A}},\Phi)$ \textit{belong  {\red to $\InvQM$.} Then }
$\epsilon_{\mathcal{A}} \in \mathcal{A} \odot \bar{\mathcal{A}}$\textit{.}
\begin{proof}
{\red By assumption,} there is a unitary involution $\Phi$ on $\M_{\mathcal{A}}$ such that $\mathcal{A} = (\M_{\mathcal{A}})^{\Phi}_{\text{sa}}$; by \cref{cor: canonical composites with involutions}, $\mathcal{A} \odot \bar{\mathcal{A}}$ is then the set of self-adjoint fixed points of $\Phi\otimes \bar{\Phi}$. Since $\Phi$ is unitary,  if $E$ is an orthonormal basis for $\M_\mathcal{A}$, then so is $\{\Phi(e) | e \in E\}$; hence, as $\epsilon_{\mathcal{A}}$ is independent of the choice of orthonormal basis, $\epsilon_{\mathcal{A}}$ is invariant under $\Phi \otimes \bar{\Phi}$. 
Since $\epsilon_{\mathcal{A}}$ is also self-adjoint, it belongs to $(\M_{\mathcal{A}} \otimes \M_{\bar{\mathcal{A}}})^{\Phi \otimes \bar{\Phi}}_{\text{sa}}$, \textit{i.e}.\ to $\mathcal{A} \odot \bar{\mathcal{A}}$.
\end{proof}
\end{lemma}
\noindent It follows now from part (b) of \cref{cor: things that are InvQM morphisms} that the 
functional 
\begin{equation}
{\redd \eta_{\mathcal{A}}} = | \epsilon_{\bar{\mathcal{A}}} \rangle : \M_{\bar{\mathcal{A}}} \otimes \M_{\mathcal{A}} \longrightarrow \R
\end{equation}
is an $\InvQM$ morphism.  Hence, $\InvQM$ inherits the compact structure from $\stalg$, as promised. This gives 
us our main theorem.

\begin{theorem}\label{thm: InvQM is compact closed} $\InvQM$ is compact closed. \end{theorem} 
 
\noindent In fact, we can do a bit better. Recall \cref{dagCom}. It is not difficult to show that $\stalg$ is dagger compact closed category, where, if $\M$ and $\N$ are finite-dimensional $\ast$-algebras and $\phi : \M \rightarrow \N$ is a linear mapping, $\phi^{\dagger}$ is the hermitian adjoint 
of $\phi$ with respect to the natural trace inner products on $\M$ and $\N$. If $(\mathcal{A},\M_{\mathcal{A}})$ and $(\mathcal{B},\M_{\mathcal{B}})$ 
are involutive EJC-algebras with given unitary involutions $\Phi_{\mathcal{A}}$ and $\Phi_{\mathcal{B}}$, then for any intertwiner $\phi : \M_{\mathcal{A}} \longrightarrow \M_{\mathcal{B}}$, $\phi^{\dagger}$ also intertwines $\Phi_{\mathcal{A}}$ and $\Phi_{\mathcal{B}}$. Hence, we have

\begin{theorem}\label{cor: InvQM is dagger compact} $\InvQM$ is dagger compact closed category.\end{theorem}

\chapter{Conclusion (Part II)}
\label{conclusionPartII}

\epigraphhead[40]
	{
		\epigraph{``Our imagination is stretched to the utmost, not, as in fiction, to imagine things which are not really there, but just to comprehend those things which \textit{are} there.''}{---\textit{Richard P.\ Feynman}\\ The Character of Physical Law (1965)}
	}

\noindent We have found two {\redd categories of probabilistic models} ---the categories $\RSE$ and
$\InvQM$---that, in different ways, unify finite-dimensional real,
complex and quaternionic quantum mechanics.  In each case, there is a
price to be paid for this unification.  For $\RSE$, this price is
steep: $\RSE$ is a symmetric monoidal category, but one in which states (for
instance) on complex systems don't count as physical processes. In particular, $\RSE$ is
very far from being compact closed.

\noindent {\magenta In contrast, 
$\InvQM$ is clearly 
a well-behaved --- indeed, dagger compact closed --- probabilistic theory, in which the states, as well
as the effects, of real, 
complex, and quaternionic Euclidean Jordan algebras appear as morphisms. 
  On the other hand, $\InvQM$ 
 admits the transpose
  automorphism on the complex Hermitian Jordan algebra,} and requires
complex quantum systems to compose in a nonstandard way.
Nevertheless, by virtue of being dagger compact, ${\redd \InvQM}$ 
continues
to enjoy many of the information-processing properties of standard
complex QM, for example the existence of conclusive teleportation and
entanglement-swapping protocols \cite{Abramsky-Coecke}. {\magenta Also, composites in $\InvQM$
satisfy the Cirel'son bound on correlations owing to the way that,
by construction, these composites are embedded within a tensor product
of complex matrix algebras.}

\noindent All of this is in spite of the fact that composites in $\InvQM$ are 
not locally tomographic: 
the canonical composite $\mathcal{A} \odot \mathcal{B} $ is larger than the vector 
space tensor product $\mathcal{A}  \otimes \mathcal{B} $.
Local tomography is well known to separate complex QM from its 
real and quaternionic variants, so its failure in $\URUE$ and $\RSE$ 
is hardly surprising,  but it is noteworthy that we are able to construct
(non-locally tomographic) composites in $\URUE$ in all of the non-real cases, and
certain composites involving quaternions even in $\RSE$.     
Even more interesting is the fact that, for 
quaternionic systems $\mathcal{A}$ and $\mathcal{B}$, the {\em information capacity} --- 
the number of sharply distinguishable states --- of $\mathcal{A}\odot \mathcal{B}$ is {\em larger} than the  
product of the capacities of $\mathcal{A}$ and $\mathcal{B}$. A related point is that, for 
quaternionic quantum systems $\mathcal{A}$ and $\mathcal{B}$, the product of a pure state 
of $\mathcal{A}$ and a pure state of $\mathcal{B}$ will generally be a {\em mixed} state 
in $\mathcal{A} \odot \mathcal{B}$. This is simply because the symplectic representation embeds pure quaternionic quantum states as rank-2 complex quantum projectors (for details see the author's MSc thesis \cite{Graydon2011}).

\noindent The category $\InvQM$ contains interesting compact closed
subcategories.  In particular, real and quaternionic quantum systems,
taken together, form a (full) monoidal sub-category of $\InvQM$ closed under
composition.  We  {\magenta conjecture} 
that this is exactly what one gets by applying Selinger's CPM construction \cite{Selinger2005} to Baez'
(implicit) category of pairs $(\mathcal{H},J)$, $\mathcal{H}$ a finite dimensional complex
Hilbert space and $J$ an anti-unitary with $J^2 = \pm 1$ \cite{BaezDivAlg}.

\noindent Another compact closed subcategory of $\mathbf{InvQM}$, which we will shall call $\mathbf{InvCQM}$, consists of universally embedded {\em complex} quantum systems $\mathcal{M}_{n}(\mathbb{C})_{\text{sa}}$. It is interesting to note that, in a hypothetical universe described by $\mathbf{InvQM}$, the subcategory $\mathbf{InvCQM}$ acts as a kind of ``ideal'', in that if $\mathcal{A} \in \mathbf{InvQM}$ and $\mathcal{B}\in \mathbf{InvCQM}$, then $\mathcal{A} \odot \mathcal{B} \in \mathbf{InvCQM}$ as well. Summoning the spirit of the physical interpretation of quantum theory described at the onset of \cref{introPartI}, the canonical tensor product is of course, at bottom, an epistemological construct. Adopting that view, if one associates, say, $\mathcal{M}_{n}(\mathbb{H})_{\text{sa}}$ with system $\mathsf{A}$, and, say, $\mathcal{M}_{m}(\mathbb{C})_{\text{sa}}$ with system $\mathsf{B}$ --- presumably for physical reasons beyond the scope of the present analysis --- then via those associations one constructs the arena for staking degrees of belief regarding local evolutions and local measurements on the local systems in question via post-quantum state assignments, possible effects, and so forth (all epistemic devices.) Furthermore, working in the category $\mathbf{InvQM}$, one has that the arena for staking beliefs regarding the composite $\mathsf{AB}$ is just the ambient space $\mathcal{M}_{n}(\mathbb{H})_{\text{sa}}\odot\mathcal{M}_{m}(\mathbb{C})_{\text{sa}}= \mathcal{M}_{2nm}(\mathbb{C})_{\text{sa}}$, with the equality following from \cref{prop: universal tensor product properties} \ref{8234} and a small calculation. Put otherwise, global states, effects, and transformations regarding the composite system are exactly from usual complex quantum theory. Therefore, if one restricts one's attention to $\mathsf{AB}$, then that universe looks complex-quantum. Furthermore, it is only when one explicitly considers the components that one sees a departure from usual quantum theory; either via the compositional structure of $\mathsf{AB}$ --- which is not locally tomographic since $n(2n-1)m^2<4n^{2}m^{2}$ --- or by, in this case, the local state space for $\mathsf{A}$ itself, which is quaternionic. Furthermore, even if one considers a composite of two \textit{complex} systems in $\mathbf{InvQM}$, then the global states, effects, and transformations remain complex-quantum; however, the composite structure differs in that one has \textit{two} copies of the usual tensor product: the direct summands differ by a transpose, which might best be interpreted as time running forwards in the first direct summand, and backwards in the second. This final thought warrants the attention of further investigation, and may provide a physical justification for allowing the transpose automorphism to factor through the universal C$^{*}$\!-algebra enveloping $\mathcal{M}_{n}(\mathbb{C})_{\text{sa}}$.

\noindent It worth emphasizing that the most marked physical distinctions between quantum theory and $\mathbf{InvQM}$ are manifest in the structure of composites. Indeed, the following informal diagram sketches how quantum theories over the Jordan matrix algebras \textit{simulate} each other locally, or put otherwise: \textit{at the level of single systems}. The `inc' arrows represent the inclusion of one theory of single systems within another via a simple superselection rule. The `MMG' arrow refers to the simulation constructed by McKaugue-Mosca-Gisin in \cite{McKague2009}. The `q' arrow refers to the simulation constructed by the author in \cite{Graydon2013}. The arrow `MMG$\circ$q' refers to the composition of these simulations. Furthermore, in light the standard representation of spin factors expounded by Barnum-Graydon-Wilce in \cite{Barnum2016b} and detailed in \cref{jordPrelims}, theories of single spins over $\mathcal{V}_{k}$ are readily embedded into the framework of usual complex quantum theory. The reader will notice the absence of the exceptional Jordan algebra $\mathcal{M}_{3}(\mathbb{O})_{\text{sa}}$, which admits no faithful representation inside the self-adjoint part of any C$^{*}$\!-algebra.
\begin{equation}
\xymatrixrowsep{0.8cm}
\xymatrixcolsep{0.8cm}
\xymatrix{
 \mathcal{V}_{k}\ar[d]_{\text{inc}} & & \\
 \mathbb{C}\mathrm{QT}\ar@{=}[rr]\ar@<-1ex>[dd]_{\text{MMG}} & &  \mathbb{C}\mathrm{QT}\ar@<-1ex>[dd]_{\text{inc}}\\
 & &\\
 \mathbb{R}\mathrm{QT}\ar@<-1ex>[uu]_{\text{inc}}\ar@<1ex>[rr]^{\text{inc}} & &  \mathbb{H}\mathrm{QT}\ar@<-1ex>[uu]_{\text{q}}\ar@<1ex>[ll]^{\text{MMG}\circ\text{q}}
}
\end{equation}

\noindent Although it is not compact closed, the 
category $\RSE$ of reversible, standardly embedded \textsc{ejc}s remains of interest. This is 
still a monoidal category, and contains, in addition to real and quaternionic quantum systems, orthodox complex 
quantum systems in their standard embedding (and composing in the normal way). 
Indeed, these form a monoidal subcategory, $\CQM$, which, again, 
functions as an ``ideal".

\noindent It is worth noting that the set of quaternionic quantum systems 
does {\em not} form a monoidal subcategory of either $\RSE$ or 
$\InvQM$, as the composite of two quaternionic systems is {\em real}. 
Efforts to construct a free-standing quaternionic quantum theory 
have had to contend with the absence of a suitable {\em quaternionic} 
composite of quaternionic systems. For instance, as pointed out by Araki 
\cite{Araki80}, the obvious candidate for the composite 
of $\mathcal{A} = \mathcal{M}_{n}(\mathbb{H})_{\text{sa}}$ and $\mathcal{B} = \mathcal{M}_{m}(\mathbb{H})_{\text{sa}}$, $\mathcal{M}_{nm}(\mathbb{H})_{\text{sa}}$, 
does not have a large enough dimension to accommodate the real
vector space tensor product $\mathcal{A} \otimes_{\mathbb{R}} \mathcal{B}$, 
{\magenta causing difficulty for}  the 
representation of product effects.\footnote{Attempts to interpret
  the quaternionic Hilbert module $\mathbb{H}^{mn}$ as a tensor 
  product of $\mathbb{H}^m$
  and $\mathbb{H}^n$ raise at least the possibility of signaling via the
  noncommutativity of scalar multiplication.  This noncommutativity
  underlies the the argument in \cite{McKague} that
  stronger-than-quantum correlations are achievable in such a model.}
 In our approach, the issue simply doe not arise. It seems that quaternionic quantum mechanics is best seen as an 
inextricable part of a larger theory. Essentially 
the same point has also been made by Baez \cite{BaezDivAlg}. 
 
\noindent The category $\InvQM$ is somewhat mysterious. It encompasses 
real and quaternionic QM in a completely natural way; 
however, while it also contains complex quantum systems, these compose 
in an exotic way: {\redd as pointed out above}, the composite of two complex quantum systems 
in $\InvQM$ comes with an extra classical bit ---
equivalently, $\{0,1\}$-valued superselection rule. This functions to make the transpose automorphism of 
$\mathcal{M}_{n}(\mathbb{C})_{\text{sa}}$ count as a morphism.   
The extra classical bit is flipped by the Jordan transpose 
(swap of C$^{*}$ summands) on either factor of {\redd such} a composite, but unaffected 
if {\em both} parties {\redd implement} the Jordan transpose (which does, of course, 
effect a Jordan transpose on the composite).
The precise physical significance of this is a subject for
further study. 

\noindent As \cref{ex: no states} shows, there is no way to enlarge
  $\InvQM$ so as to include higher spin factors, without either
  sacrificing compact closure {\magenta (and even rendering the set
  $\cc(\mathcal{A},\mathrm{I})$, which might naturally be thought to represent states,
  trivial)} or venturing outside the ambient category of
  \textsc{ejc}-algebras, to make use of morphisms that are not {\redd (relatively)} completely
  Jordan-preserving maps.  Our second proof of \cref{ex: no states} shows, more strikingly, that
  there is no way to construct a category {\redd of the form $\CJP_{\mathscr{C}}$} that contains
  {\em standardly embedded} complex quantum systems {\em and} real 
  systems, without, again, sacrificing compact closure
  (indeed, the representation of states by morphisms). 
\bookmarksetup{startatroot}
\addtocontents{toc}{\vspace{0.5cm}}
\chapter{Epilogue}
\label{epilogue}
\epigraphhead[40]
	{
		\epigraph{``Dreamers, they never learn.''}{---\textit{Radiohead}\\ Daydreaming (2016)}
	}
	
\noindent The zeitgeist of the current quantum foundations revolution revolves around the conception of quantum theory as a \textit{theory of information}. This spirit of thought underpins a mutualistic symbiosis enjoyed by the fields of quantum information and quantum foundations. Transcending philosophy, quantum theory has in fact recently been \textit{derived} from information-theoretic principles \cite{Chiribella2011}\cite{Masanes2011}\cite{Hardy2011}\cite{Dakic2011} (references given being representative only.) These derivations exemplify the immense power of deep, yet simple physical ideas; ideas akin to Einstein's more compelling postulates for relativity. There is no \textit{a priori} reason to expect that Einstein's relativity and quantum theory can be reconciled within a unified framework for spacetime and quantum physics. One cannot, however, resist speculating that relativity and quantum theory rise from a deeper level of physical notions. Embracing an information-theoretic perspective of physics as a whole, one ultimately seeks to uncover such common ground for all varieties of physical experience. 

\noindent An obvious approach to the formulation of new fundamental physics is to first develop a deeper understanding of quantum theory itself. Indeed, while the derivations cited above exactly yield the full abstract mathematical structure of quantum theory, there remain many important unresolved aspects pertaining to the character of this probability calculus. In particular, the geometry of quantum state space is far from fully understood. In \cref{partI} of this thesis, we introduced novel shapes of high symmetry inscribed within quantum cones of arbitrary dimension, namely \textit{conical designs} \cite{Graydon2015a}. These shapes shed new light on the structure of quantum information; moreover, they are naturally adapted to an information-theoretic description entanglement \cite{Graydon2016}. 

\noindent A tautologically obvious route towards new physics is the explicit formulation of novel physical theories. The vast landscape of general probabilistic theories provides a rich foundation; however, it is certainly not obvious where and how to depart from the twin pillars. Indeed, relativity and quantum theory are remarkably, nay, shockingly successful from an empirical point of view. Therefore, a conservative and logical approach is to consider physical theories sharing some essential characteristics with our current conceptions of nature. At a minimum, one feels that the notion of nonsignaling ought to be preserved, for as Einstein emphasized, the notion of (quasi-)closed systems is necessary for the establishment of empirically testable physics \cite{Einstein1948}. Together with nonsignaling, additional desiderata narrow the field of general probabilistic theories encompassing quantum theory. In \cref{partII} of this thesis, we constructed dagger compact closed symmetric monoidal categories for Jordan-algebraic physics \cite{Barnum2015}\cite{Barnum2016b}.

\newpage
\noindent Eugene Wigner famously described the unreasonable effectiveness of mathematics in the natural sciences \cite{Wigner1960}. \textit{Prima facie}, abstract thinking is very powerful: from the general conception of a class, one can capture the essence of all particular instantiations. This has been especially true in physics, which is truly remarkable. Abstract mathematics alone, however, could never hope to provide a solid conceptual foundation for physics. Rather, one requires deep physical principles concerning the nature of nature. It is our hope that the present thesis draws one closer towards \textit{dem Geheimnis des Alten}.


\bibliographystyle{mg}
\cleardoublepage 
\phantomsection  
\renewcommand*{\bibname}{References}

\cleardoublepage
\phantomsection
\addcontentsline{toc}{chapter}{\textbf{References}}
\label{theBib}
\bibliography{mg-ethesis}
\appendix
\chapter{Appendices of Part I}
\label{appP1}
\section{Proof of Lemma 2.2.1}\label{lem2p2p1Proof}
In this appendix, we recall the statement of the \cref{irrepLem} for convenience, and provide a full proof.
\begin{lemma}\label{irrepLemA}
\textit{The restrictions of the product representation of $\mathrm{U}(\mathcal{H}_{d})$ to $\mathcal{H}_{\text{sym}}$ and $\mathcal{H}_{\text{asym}}$ are irreducible.}
\end{lemma}
\begin{proof}Let subspaces $\mathcal{X}\subseteq\mathcal{H}_{\text{sym}}$ and $\mathcal{Y}\subseteq\mathcal{H}_{\text{asym}}$ be invariant under the product representation of $\mathrm{U}(\mathcal{H}_{d})$. We must prove that $\mathcal{X}\in\{\{0\},\mathcal{H}_{\text{sym}}\}$ and $\mathcal{Y}\in\{\{0\},\mathcal{H}_{\text{asym}}\}$. The zero cases are trivial. Before proceeding with the nontrivial cases, it will be useful to make some preparatory observations. Let $t\in\{1,\dots,d\}$ and $|e_{t}\rangle$ be the orthonormal basis for $\mathcal{H}_{d}$ arbitrarily chosen to define the orthonormal basis elements $|f^{+}_{r,s}\rangle,|f^{+}_{r,r}\rangle$ and $|f^{-}_{r,s}\rangle$ for $\mathcal{H}_{\text{sym}}$ and $\mathcal{H}_{\text{asym}}$ in Eqs.~\eqref{hSym} and \eqref{hAsym}, respectively. Introduce the following families of unitary linear endomorphisms on $\mathcal{H}_{d}$ defined for arbitrary fixed $p<q\in\{1,\dots,d\}\ni l$ via
\begin{eqnarray}
A_{l}|e_{t}\rangle=|e_{t}\rangle(1-2\delta_{t,l})\text{,}\hspace{1cm}
B_{p,q}|e_{t}\rangle=|e_{t}\rangle(1-2\delta_{t,p}-2\delta_{t,q})\text{,}
\end{eqnarray}
where $\delta_{t,l}$ is the usual Kronecker delta function. The action of $A_{l}\otimes A_{l}$ for $l\in\{p,q:p<q\}$ and $B_{p,q}\otimes B_{p,q}$ on $|f^{\pm}_{r,s}\rangle$ is to selectively introduce a phase of $-1$. This action is summarized in the following table organized by mutually exclusive and exhaustive cases of the indices $r$ and $s$. We also include a trivial column, for reasons that will become clear.
\begin{center}
\begin{small}{
\begin{tabular}{c||c|c|c|c}
&$\mathds{1}_{d}\otimes \mathds{1}_{d}|f^{\pm}_{r,s}\rangle$ & $A_{p}\otimes A_{p}|f^{\pm}_{r,s}\rangle(-1)$ & $A_{q}\otimes A_{q}|f^{\pm}_{r,s}\rangle(-1)$ & $B_{p,q}\otimes B_{p,q}|f^{\pm}_{r,s}\rangle$
\\\hline
$\hspace{0.575cm}r<p\;\wedge\;s\notin\{p,q\}$ & $+$ & $-$ & $-$ & $+$ \\\hline
$r<p\;\wedge\;s=q$ & $+$ & $-$ & $+$ & $-$ \\\hline
$r<p\;\wedge\;s=p$ & $+$ & $+$ & $-$ & $-$ \\\hline
$r=p\;\wedge\;s=q$ & $+$ & $+$ & $+$ & $+$ \\\hline
$r=p\;\wedge\;s\neq q$ & $+$ & $+$ & $-$ & $-$ \\\hline
$r=q\;\wedge\;s\neq q$ & $+$ & $-$ & $+$ & $-$ \\\hline
$q\neq r>p\;\wedge\;s=q$ & $+$ & $-$ & $+$ & $-$ \\\hline
$q\neq r>p\;\wedge\;s\neq q$ & $+$ & $-$ & $-$ & $+$ \\\hline
\end{tabular}
}\end{small}
\end{center}
Next, introduce a third family $C_{p,q}\in\mathrm{U}(\mathcal{H}_{d})$ in Eq.~\eqref{Cpq}, and observe the consequent Eq.~\eqref{CpqOb}
\begin{eqnarray}
&C_{p,q}|e_{t}\rangle\sqrt{2}=\Big\{|e_{t}\rangle\sqrt{2}\text{ if } t\notin\{p,q\};\hspace{0.3cm}|e_{p}\rangle+|e_{q}\rangle\text{ if } t=p;\hspace{0.3cm}|e_{p}\rangle-|e_{q}\rangle\text{ if } t=q\Big\}\label{Cpq}\text{,}\\
\implies&\big(C_{p,q}\otimes C_{p,q}\big)|f^{+}_{p,p}\rangle\sqrt{2}-\big(|f^{+}_{p,p}\rangle+|f^{+}_{q,q}\rangle\big)2^{-1/2}=|f^{+}_{p,q}\rangle\text{.}\label{CpqOb}
\end{eqnarray}
This concludes our preparatory work. Our method of proof will now be to demonstrate that from arbitrary nonzero elements $|x\rangle\in\mathcal{X}\subseteq\mathcal{H}_{\text{sym}}$ and $|y\rangle\in\mathcal{Y}\subseteq\mathcal{H}_{\text{asym}}$, one can generate all of the basis elements in Eqs.~\eqref{hSym} and \eqref{hAsym}, respectively, via vector space operations and unitary linear endomorphisms of the form $U\otimes U$. Indeed, $\mathcal{X}$ and $\mathcal{Y}$ are by assumption subspaces invariant under the product representation, and therefore closed under such operations and endomorphisms. We start with such a $|y\rangle$ and find
\begin{eqnarray}
&|y\rangle&=|f^{-}_{p,q}\rangle\lambda_{p,q}+\sum_{r=1}^{d}\sum_{s>r}|f^{-}_{r,s}\rangle\lambda_{r,s}(1-\delta_{r,p}\delta_{s,q})\hspace{0.5cm}\text{with }\; 0\neq \lambda_{p,q}\in\mathbb{C}\label{arbY}\\
\implies\hspace{0.35cm}&\mathcal{Y}\ni|f^{-}_{p,q}\rangle&=\frac{1}{4\lambda_{p,q}}\Big(\mathds{1}_{d}\otimes\mathds{1}_{d}-A_{p}\otimes A_{p}-A_{q}\otimes A_{q}+B_{p,q}\otimes B_{p,q}\Big)|y\rangle\text{.}\label{arbYReduced}
\end{eqnarray} 
The remaining basis elements for $\mathcal{H}_{\text{asym}}$ can then be generated under the action of $U_{\mathrm{p}}\otimes U_{\mathrm{p}}$ for suitably chosen permutations in $\mathrm{p}\in\mathrm{S}_{d}$. Thus $\mathcal{Y}=\mathcal{H}_{\text{asym}}$. Now, for such an $|x\rangle$, the situation is more involved: either $\exists p<q\in\{1,\dots,d\}$ with $0\neq\mu_{p,q}\in\mathbb{C}$ or $\exists v\in\{1,\dots,d\}$ with $0\neq \nu_{v}\in\mathbb{C}$, or both, in
\begin{eqnarray}
|x\rangle=|f^{+}_{p,q}\rangle\mu_{p,q}+\sum_{r=1}^{d}\sum_{s>r}|f^{+}_{r,s}\rangle\mu_{r,s}(1-\delta_{r,p}\delta_{s,q})+|f^{+}_{v,v}\rangle\nu_{v}+\sum_{r=1}^{d}|f^{+}_{r,r}\rangle\nu_{r}(1-\delta_{r,v})\text{.}
\end{eqnarray}
There are now three mutually exclusive and exhaustive subcases to consider.

\noindent Subcase 1: if $!\exists\nu_{v}\neq0$, then our analysis of the antisymmetric case with $|f^{-}_{r,s}\rangle$ now replaced by $|f^{+}_{r,s}\rangle$ yields $\mathcal{X}=\mathcal{H}_{\text{sym}}$. 

\noindent Subcase 2: if $!\exists\mu_{r,s}\neq0$, then introduce then $D_{v}\in\mathrm{U}(\mathcal{H}_{d})$ as follows and observe that
\begin{eqnarray}
D_{v}|e_{l}\rangle=|e_{l}\rangle i^{\delta_{l,v}}\hspace{0.3cm}\text{where } i^{2}=-1\label{Dv}\implies \mathcal{X}\ni|f^{+}_{v,v}\rangle&=\frac{1}{2\nu_{v}}\big(\mathds{1}_{d}\otimes\mathds{1}_{d}-D_{v}\otimes D_{v}\big)|x\rangle\text{.}\label{DvOb}
\end{eqnarray}
All of the remaining remaining $|f^{+}_{r,r}\rangle$ can then be generated from $|f^{+}_{v,v}\rangle$ via $U_{\mathrm{p}}\otimes U_{\mathrm{p}}$ with permutations $\mathrm{p}\in\mathrm{S}_{d}$ suitably chosen, and from $|f^{+}_{r,r}\rangle$, one can generate $|f^{+}_{r,s}\rangle$ via Eq.~\eqref{CpqOb}. Thus $\mathcal{X}=\mathcal{H}_{\text{sym}}$. 

\noindent Subcase 3: if $\exists \nu_{v}\neq 0$ and $\mu_{r,s}\neq 0$, then one can proceed as follows
\begin{eqnarray}
&\mathcal{X}\ni|x'\rangle&\equiv\frac{1}{4\mu_{p,q}}\Big(\mathds{1}_{d}\otimes\mathds{1}_{d}-A_{p}\otimes A_{p}-A_{q}\otimes A_{q}+B_{p,q}\otimes B_{p,q}\Big)|x\rangle\\
&&=|f^{+}_{p,q}\rangle+|f^{+}_{v,v}\rangle\frac{\nu_{v}}{4\mu_{p,q}}+\frac{1}{4\mu_{p,q}}\sum_{r=1}^{d}|f^{+}_{r,r}\rangle\nu_{r}(1-\delta_{r,v})\\
\implies&\mathcal{X}\ni|x''\rangle&\equiv\big(\mathds{1}_{d}\otimes\mathds{1}_{d}+A_{q}\otimes A_{q}\big)|x'\rangle\\
&&=|f^{+}_{v,v}\rangle\frac{\nu_{v}}{2\mu_{p,q}}+\frac{1}{2\mu_{p,q}}\sum_{r=1}^{d}|f^{+}_{r,r}\rangle\nu_{r}(1-\delta_{r,v})
\end{eqnarray}
With $|x''\rangle$ we return to our analysis of Subcase 2 to conclude once again that $\mathcal{X}=\mathcal{H}_{\text{sym}}$.
\end{proof}
\newpage
\section{Proof of Lemma 4.1.5}\label{conProof}
In this appendix, we state \cref{conLem} in equivalent terms, and provide a full proof.
\begin{lemma}\textit{Define}\footnote{Recall that $\forall\sigma\in\mathcal{Q}\big(\mathcal{H}_{d}\big):\frac{1}{d}\leq\mathrm{Tr}\big(\sigma^{2}\big)\leq1$, so $f$ is inherently nonnegative.} 
\begin{equation}
f:\mathcal{Q}\big(\mathcal{H}_{d}\big)\rightarrow\mathbb{R}_{+}::\sigma\mapsto f(\sigma)=\sqrt{1-\mathrm{Tr}\big(\sigma^{2}\big)}
\end{equation}
\textit{It follows that}
\begin{eqnarray}
\text{i}.&\forall U\in\mathrm{U}(\mathcal{H}_{d}):f\big(U\sigma U^{*}\big)=f(\sigma)\text{.} \label{pi}\\
\text{ii}.&\forall\lambda\in[0,1]\forall\sigma_{1},\sigma_{2}\in\mathcal{Q}\big(\mathcal{H}_{d}\big):f\Big(\sigma_{1}\lambda+\sigma_{2}(1-\lambda)\Big)\geq f(\sigma_{1})\lambda+f(\sigma_{2})(1-\lambda)\text{.}\label{pii}
\end{eqnarray}
\begin{proof}
That property i. holds follows directly from cyclicity of the trace:
\begin{eqnarray}
f\big(U\sigma U^{\dagger}\big)&=&\sqrt{1-\mathrm{Tr}\Big(\big(U\sigma U^{*}\big)^{2}\Big)}\nonumber\\[0.2cm]
&=&\sqrt{1-\mathrm{Tr}\big(U\sigma U^{\dagger}U\sigma U^{*}\big)}\nonumber\\[0.2cm]
&=&\sqrt{1-\mathrm{Tr}\big(U^{*}U\sigma^{2}\big)}\nonumber\\[0.2cm]
&=&\sqrt{1-\mathrm{Tr}\big(\sigma^{2}\big)}\nonumber\\[0.2cm]
&=&f(\sigma)\text{.}
\end{eqnarray}
Regarding property ii., we first compute the square of the \textsc{lhs} of \eqref{pii}
\begin{eqnarray}
L&\equiv&\Big(f\big(\sigma_{1}\lambda+\sigma_{2}(1-\lambda)\big)\Big)^{2}\nonumber\\[0.2cm]
&=&1-\mathrm{Tr}\Big(\big(\sigma_{1}\lambda+\sigma_{2}(1-\lambda)\big)^{2}\Big)\nonumber\\[0.2cm]
&=&1-\lambda^{2}\mathrm{Tr}\big(\sigma_{1}^{2}\big)-(1-\lambda)^{2}\mathrm{Tr}\big(\sigma_{2}^{2}\big)-2\lambda(1-\lambda)\mathrm{Tr}\big(\sigma_{1}\sigma_{2}\big)\text{.}
\end{eqnarray}
and the square of the \textsc{rhs} of \eqref{pii},
\begin{eqnarray}
R&\equiv&\Big(\lambda f(\sigma_{1})+(1-\lambda)f(\sigma_{2})\Big)^{2}\nonumber\\[0.2cm]
&=&\lambda^{2}\big(f(\sigma_{1})\big)^{2}+(1-\lambda)^{2}\big(f(\sigma_{2})\big)^{2}+2\lambda(1-\lambda)f(\sigma_{1})f(\sigma_{2})\nonumber\\[0.2cm]
&=&\lambda^{2}\Big(1-\mathrm{Tr}\big(\sigma_{1}^{2}\big)\Big)+(1-\lambda)^{2}\Big(1-\mathrm{Tr}\big(\sigma_{2}^{2}\big)\Big)+2\lambda(1-\lambda)\sqrt{1-\mathrm{Tr}\big(\sigma_{1}^{2}\big)}\sqrt{1-\mathrm{Tr}\big(\sigma_{2}^{2}\big)}\text{.}
\end{eqnarray}
Noting that 
\begin{equation}
1=\lambda+(1-\lambda)\implies\lambda^{2}+(1-\lambda)^{2}=1-2\lambda(1-\lambda),
\end{equation}
we then have
\begin{eqnarray}
L-R&=&1-\lambda^{2}-(1-\lambda)^{2}-2\lambda(1-\lambda)\Big(\mathrm{Tr}\big(\sigma_{1}\sigma_{2}\big)+\sqrt{1-\mathrm{Tr}\big(\sigma_{1}^{2}\big)}\sqrt{1-\mathrm{Tr}\big(\sigma_{2}^{2}\big)}\Big)\nonumber\\
&=&2\lambda(1-\lambda)\Big(1-X\Big)\text{,}
\end{eqnarray}
with
\begin{equation}
X\equiv\mathrm{Tr}\big(\sigma_{1}\sigma_{2}\big)+\sqrt{1-\mathrm{Tr}\big(\sigma_{1}^{2}\big)}\sqrt{1-\mathrm{Tr}\big(\sigma_{2}^{2}\big)}\text{.}
\end{equation}
We seek to show that $X\leq 1$. Without the loss of generality, assume that $\mathrm{Tr}\big(\sigma_{1}^{2}\big)\leq\mathrm{Tr}\big(\sigma_{2}^{2}\big)$. On that innocuous assumption, it follows that
\begin{eqnarray}
X&\leq&\mathrm{Tr}\big(\sigma_{1}\sigma_{2}\big)+1-\mathrm{Tr}\big(\sigma_{1}^{2}\big)\nonumber\\[0.2cm]
&=&1+Y\text{,}
\end{eqnarray}
with
\begin{equation}
Y\equiv\mathrm{Tr}\big(\sigma_{1}\sigma_{2}\big)-\mathrm{Tr}\big(\sigma_{1}^{2}\big)
\end{equation}
Now, from the Cauchy-Schwarz inequality:
\begin{eqnarray}
\mathrm{Tr}\big(\sigma_{1}\sigma_{2}\big)\leq\sqrt{\mathrm{Tr}\big(\sigma_{1}^{2}\big)\mathrm{Tr}\big(\sigma_{2}^{2}\big)}\leq\mathrm{Tr}\big(\sigma_{1}^{2}\big)\text{.}
\end{eqnarray}
We conclude that
\begin{equation}
Y\leq0\implies X\leq 1\implies R\leq L\implies \text{ii. holds\text{.}}
\end{equation}
\end{proof}
\end{lemma}
\newpage
\section{Permutation Notation for Lemma 4.2.4}\label{permApp}

Let $\mathcal{S}=\{r;s;u;v\}$ be a fixed ordered set with $r,s,u,v\in\{1,\dots,d\}$. Let $\text{Perm}(\mathcal{S})$ be the set of all unique ordered permutations of $\mathcal{S}$. We denote that cardinality of $\text{Perm}(\mathcal{S})$ by $|\text{Perm}(\mathcal{S})|$. We label the elements of $\text{Perm}(\mathcal{S})$ by $p$ with $p\in\{1,\dots,|\text{Perm}(\mathcal{S})|\}$. Introduce the notation $p(\mathfrak{a})$ for the $\mathfrak{a}^{th}$ element of the ordered set labeled by $p$. With this notation, we hope to make our standing assumption in Eq.~\eqref{standing} crystal clear. Recall that we wrote Eq.~\eqref{standing} as
\begin{equation}
\sum_{p}X_{r_{p}s_{p},u_{p}v_{p}}^{\mathds{1}}X_{r_{p}s_{p},u_{p}v_{p}}^{\mathds{1}}=\sum_{p}e^{i(\theta_{u_{p}}+\theta_{v_{p}}-\theta_{r_{p}}-\theta_{s_{p}})}X_{r_{p}s_{p},u_{p}v_{p}}^{U}X_{r_{p}s_{p},u_{p}v_{p}}^{V}
\label{standingRecall}
\end{equation}
In the notation mentioned just above, this is equivalent to
\begin{eqnarray}
&&\sum_{p\in\text{Perm}\big(\{r;s;u;v\}\big)}X_{p(1)p(2),p(3)p(4)}^{\mathds{1}}X_{p(1)p(2),p(3)p(4)}^{\mathds{1}}\\
&=&\sum_{p\in\text{Perm}\big(\{r;s;u;v\}\big)}e^{i(\theta_{p(3)}+\theta_{p(4)}-\theta_{p(1)}-\theta_{p(2)})}X_{p(1)p(2),p(3)p(4)}^{U}X_{p(1)p(2),p(3)p(4)}^{V}
\label{standingRewritten}
\end{eqnarray}
For example, let $\mathcal{S}=\{r;r;s;s\}$. Then
\begin{equation} 
\text{Perm}(\mathcal{S})=\big\{\{r;r;s;s\},\{s;s;r;r\},\{r;s;r;s\},\{s;r;s;r\},\{r;s;s;r\},\{s;r;r;s\}\big\}\text{,} 
\end{equation}
so $|\text{Perm}(\mathcal{S})|=6$, and Eq.~\eqref{standingRewritten} reads
\begin{eqnarray}
&&X_{rr,ss}^{\mathds{1}}X_{rr,ss}^{\mathds{1}}+X_{ss,rr}^{\mathds{1}}X_{ss,rr}^{\mathds{1}}+X_{rs,rs}^{\mathds{1}}X_{rs,rs}^{\mathds{1}}+X_{sr,sr}^{\mathds{1}}X_{sr,sr}^{\mathds{1}}+X_{rs,sr}^{\mathds{1}}X_{rs,sr}^{\mathds{1}}+X_{sr,rs}^{\mathds{1}}X_{sr,rs}^{\mathds{1}}\nonumber\\
&=&e^{2i(\theta_{s}-\theta_{r})}X_{rr,ss}^{U}X_{rr,ss}^{V}+e^{2i(\theta_{r}-\theta_{s})}X_{ss,rr}^{U}X_{ss,rr}^{V}\nonumber\\
&+&X_{rs,rs}^{U}X_{rs,rs}^{V}+X_{sr,sr}^{U}X_{sr,sr}^{V}+X_{rs,sr}^{U}X_{rs,sr}^{V}+X_{sr,rs}^{U}X_{sr,rs}^{V}
\end{eqnarray}
\newpage
\section{Applying Proposition 4.2.2}\label{appDMLemma}
Let $\boldsymbol{\theta}\in[0,2\pi]^{d}$. Write $\boldsymbol{\theta}=(\theta_{1},\theta_{2},\dots,\theta_{i},\dots,\theta_{d})$. Let $\mathbf{n}_{k}\in\mathbb{Z}^{d}$. Write $\mathbf{n}_{k}=(n_{k_{1}},n_{k_{2}},\dots,n_{k_{j}},\dots,n_{k_{d}})$. Lemma 1.2 states that if $\mathbf{n}_{1},\dots,\mathbf{n}_{m}$ are distinct, then the functions $e^{i\mathbf{n}_{k}\cdot\boldsymbol{\theta}}$ are linearly independent in $C_{0}([0,2\pi]^{d})$. For example let $r\neq s\in\{1,\dots,d\}$. Then the following are distinct elements of $\mathbb{Z}^{d}$
\begin{eqnarray}
\mathbf{n}_{1}=(0,\dots,0,\underbrace{1}_{s^{\text{th}}\text{ position}},0,\dots,0,\underbrace{-1}_{r^{\text{th}}\text{ position}},0,\dots,0)\\
\mathbf{n}_{2}=(0,\dots,0,\underbrace{-1}_{s^{\text{th}}\text{ position}},0,\dots,0,\underbrace{1}_{r^{\text{th}}\text{ position}},0,\dots,0)
\end{eqnarray}
Then
\begin{eqnarray}
\mathbf{n}_{1}\cdot\boldsymbol{\theta}=\theta_{s}-\theta_{r}\\
\mathbf{n}_{2}\cdot\boldsymbol{\theta}=\theta_{r}-\theta_{s}
\end{eqnarray}
and
\begin{eqnarray}
e^{i\mathbf{n}_{1}\cdot\boldsymbol{\theta}}=e^{i(\theta_{s}-\theta_{r})}\\
e^{i\mathbf{n}_{2}\cdot\boldsymbol{\theta}}=e^{i(\theta_{r}-\theta_{s})}
\end{eqnarray}
are linearly independent functions. This example is for Case 2 considered in the main body of this note.\\[0.3cm]
Let us consider another example. Let $r,s,u$ be distinct elements of $\{1,\dots,d\}$. Then the following are distinct elements of $\mathbb{Z}^{d}$
\begin{eqnarray}
\mathbf{n}_{1}=(0,\dots,0,\underbrace{1}_{s^{\text{th}}\text{ position}},0,\dots,0,\underbrace{1}_{u^{\text{th}}\text{ position}},0,\dots,0,\underbrace{-2}_{r^{\text{th}}\text{ position}},0,\dots,0)\\
\mathbf{n}_{2}=(0,\dots,0,\underbrace{-1}_{s^{\text{th}}\text{ position}},0,\dots,0,\underbrace{-1}_{u^{\text{th}}\text{ position}},0,\dots,0,\underbrace{2}_{r^{\text{th}}\text{ position}},0,\dots,0)\\
\mathbf{n}_{3}=(0,\dots,0,\underbrace{-1}_{s^{\text{th}}\text{ position}},0,\dots,0,\underbrace{1}_{u^{\text{th}}\text{ position}},0,\dots,0,\underbrace{0}_{r^{\text{th}}\text{ position}},0,\dots,0)\\
\mathbf{n}_{4}=(0,\dots,0,\underbrace{1}_{s^{\text{th}}\text{ position}},0,\dots,0,\underbrace{-1}_{u^{\text{th}}\text{ position}},0,\dots,0,\underbrace{0}_{r^{\text{th}}\text{ position}},0,\dots,0)
\end{eqnarray}
and $e^{i\mathbf{n}_{k}\cdot\boldsymbol{\theta}}$ are linearly independent function by Lemma 1.2. This example is for Case 4. Case 5 is handled similarly.
\chapter{Appendices of Part II}
\label{appP2}
\section{Review of Spin Factors}\label{spinsApp}
\noindent In this appendix we detail the spin factors.

\begin{definition}\label{ssDef} \textit{Let $\mathcal{A}$ be a Jordan algebra as in \cref{jAlgDef}. Let $k\in\mathbb{N}$ such that $k\geq 2$. Let $\mathcal{P}_{k}=\{s_{1},s_{2},\dots,s_{k}\}\subset\mathcal{A}$ such that $\forall a,b\in\{1,\dots,k\}:s_{a}\jProd s_{b}=\delta_{a,b}\mathbf{1}$ and $s_{a}\neq\pm\mathbf{1}$. One says that} $\mathcal{P}_{k}$ is a spin system \textit{of cardinality $k$.} 
\end{definition}

\noindent Consider the set $\mathcal{V}_{k}$ generated (via Jordan multiplication and vector space operations in $\mathcal{A}$) by $\mathcal{P}_{k}$. One writes $\mathcal{V}_{k}=\mathfrak{j}(\mathcal{P}_{k})$. It follows that 
\begin{equation}\label{spinBasis}
\forall v\in\mathcal{V}_{k}\;\exists\lambda_{0},\lambda_{1},\dots,\lambda_{k}\in\mathbb{R}:v=\mathbf{1}\lambda_{0}+s_{1}\lambda_{1}+\dots s_{k}\lambda_{k}\text{.}
\end{equation} 
On that view, we regard $\mathcal{V}_{k}=\mathbb{R}\oplus\mathbb{R}^{k}$ and write we $v=\lambda_{0}\oplus\vec{\lambda}$. Let $\langle\cdot|\cdot\rangle$ be the usual inner product on $\mathbb{R}^{k}$. Equipping $\mathcal{V}_{k}$ with mutltiplication $\jProd:\mathcal{V}_{k}\times\mathcal{V}_{k}\longrightarrow\mathcal{V}_{k}$ defined
\begin{equation}\label{spinProd}
\forall v_{1},v_{2}\in\mathcal{V}_{k},v_{1}\jProd v_{2}=\Big(\lambda_{0_{1}}\lambda_{0_{2}}+\langle\vec{\lambda}_{1}|\vec{\lambda}_{2}\rangle\Big)\oplus\Big(\vec{\lambda}_{2}\lambda_{0_{1}}+\vec{\lambda}_{1}\lambda_{0_{2}}\Big)
\end{equation}
we render $\mathcal{V}_{k}$ a Jordan algebra. For the proof, see \cref{spinJordProp}. One says that $\mathcal{V}_{k}$ is a \textit{spin factor} (of cardinality $k$). Hence the following.

\begin{definition}\textit{A} spin factor of cardinality $k$\textit{, denoted} $\mathcal{V}_{k}$\textit{, is the Jordan algebraic closure of a spin system} $\mathcal{P}_{k}$ \textit{as in \cref{ssDef}, equipped with the Jordan product in Eq.~\eqref{spinProd}.}
\end{definition} 

\noindent We note here that $\mathrm{dim}_{\mathbb{R}}\mathcal{V}_{k}=k+1$. For the sake of notational convienience, let us also introduce the notation $\forall v=\lambda_{0}\oplus\vec{\lambda}\in\mathcal{V}_{k}$:
\begin{eqnarray}
\mathfrak{R}(v)&=&\lambda_{0}\\
\mathfrak{I}(v)&=&\vec{\lambda}
\end{eqnarray}

\begin{proposition}\label{spinJordProp} $\mathcal{V}_{k}$ \textit{is a Jordan algebra.}
\begin{proof} Let us bserve from direct calculation 
\begin{eqnarray}
v_{1}\jProd v_{1}&=&\Big(\lambda_{0_{1}}^{2}+\langle\vec{\lambda}_{1}|\vec{\lambda}_{1}\rangle\Big)\oplus\vec{\lambda}_{1}2\lambda_{0_{1}}\
\end{eqnarray}
Therefore
\begin{eqnarray}
\mathfrak{R}\big(v_{1}^{2}\jProd(v_{1}\jProd v_{2})\big)&=&\Big(\lambda_{0_{1}}^{2}+\langle\vec{\lambda}_{1}|\vec{\lambda}_{1}\rangle\Big)\Big(\lambda_{0_{1}}\lambda_{0_{2}}+\langle\vec{\lambda}_{1}|\vec{\lambda}_{2}\rangle\Big)+\Big\langle \vec{\lambda}_{1}2\lambda_{0_{1}}|\big(\vec{\lambda}_{2}\lambda_{0_{1}}+\vec{\lambda}_{1}\lambda_{0_{2}}\big)\Big\rangle\nonumber\\
&=&\lambda_{0_{1}}^{3}\lambda_{0_{2}}+\lambda_{0_{1}}^{2}\langle\vec{\lambda}_{1}|\vec{\lambda}_{2}\rangle+\langle\vec{\lambda}_{1}|\vec{\lambda}_{1}\rangle\lambda_{0_{1}}\lambda_{0_{2}}\nonumber\\
&+&\langle\vec{\lambda}_{1}|\vec{\lambda}_{1}\rangle\langle\vec{\lambda}_{1}|\vec{\lambda}_{2}\rangle+\langle\vec{\lambda}_{1}|\vec{\lambda}_{2}\rangle 2\lambda_{0_{1}}^{2}+\langle\vec{\lambda}_{1}|\vec{\lambda}_{1}\rangle 2\lambda_{0_{1}}\lambda_{0_{2}}\\[0.3cm]
\mathfrak{I}\big(v_{1}^{2}\jProd(v_{1}\jProd v_{2})\big)&=&\bigg(\vec{\lambda}_{1}2\lambda_{0_{1}}\Big(\lambda_{0_{1}}\lambda_{0_{2}}+\langle\vec{\lambda}_{1}|\vec{\lambda}_{2}\rangle\Big)+\Big(\vec{\lambda}_{2}\lambda_{0_{1}}+\vec{\lambda}_{1}\lambda_{0_{2}}\Big)\Big(\lambda_{0}^{2}+\langle\vec{\lambda}_{1}|\vec{\lambda}_{1}\rangle\Big)\bigg)\nonumber\\
&=&\vec{\lambda}_{1}\Big(3\lambda_{0_{1}}^{2}\lambda_{0_{2}}+\langle\vec{\lambda}_{1}|\vec{\lambda}_{2}\rangle2\lambda_{0_{1}}+\langle\vec{\lambda}_{1}|\vec{\lambda}_{1}\rangle\lambda_{0_{2}}\Big)+\vec{\lambda}_{2}\Big(\lambda_{0_{1}}^{3}+\langle\vec{\lambda}_{1}|\vec{\lambda}_{1}\rangle\lambda_{0_{1}}\Big)\\[0.5cm]
\mathfrak{R}\big(v_{1}^{2}\jProd v_{2}\big)&=&\big(\lambda_{0_{1}}^{2}+\langle\vec{\lambda}_{1}|\vec{\lambda}_{1}\rangle\big)\lambda_{0_{2}}+\langle\vec{\lambda}_{1}2\lambda_{0_{1}}|\vec{\lambda}_{2}\rangle\\
\mathfrak{I}\big(v_{1}^{2}\jProd v_{2}\big)&=&\vec{\lambda}_{2}\big(\lambda_{0_{1}}^{2}+\langle\vec{\lambda}_{1}|\vec{\lambda}_{1}\rangle\big)+\vec{\lambda}_{1}2\lambda_{0_{1}}\lambda_{0_{2}}\\[0.5cm]
\mathfrak{R}\big(
v_{1}\jProd(v_{1}^{2}\jProd v_{2})\big)&=&\lambda_{0_{1}}\Big(\big(\lambda_{0_{1}}^{2}+\langle\vec{\lambda}_{1}|\vec{\lambda}_{1}\rangle\big)\lambda_{0_{2}}+\langle\vec{\lambda}_{1}2\lambda_{0_{1}}|\vec{\lambda}_{2}\rangle\Big)+\bigg\langle\vec{\lambda}_{1}\bigg|\Big(\vec{\lambda}_{2}\big(\lambda_{0_{1}}^{2}+\langle\vec{\lambda}_{1}|\vec{\lambda}_{1}\rangle\big)+\vec{\lambda}_{1}2\lambda_{0_{1}}\lambda_{0_{2}}\Big)\bigg\rangle \nonumber\\
&=&\lambda_{0_{1}}^{3}\lambda_{0_{2}}+\langle\vec{\lambda}_{1}|\vec{\lambda}_{1}\rangle\lambda_{0_{1}}\lambda_{0_{2}}+\langle\vec{\lambda}_{1}|\vec{\lambda}_{2}\rangle2\lambda_{0_{1}}^{2}\nonumber\\
&+&\langle\vec{\lambda}_{1}|\vec{\lambda}_{2}\rangle\lambda_{0_{1}}^{2}+\langle\vec{\lambda}_{1}|\vec{\lambda}_{2}\rangle\langle\vec{\lambda}_{1}|\vec{\lambda}_{1}\rangle+\langle\vec{\lambda}_{1}|\vec{\lambda}_{1}\rangle2\lambda_{0_{1}}\lambda_{0_{2}}\\[0.3cm]
\mathfrak{I}\big(
v_{1}\jProd(v_{1}^{2}\jProd v_{2})\big)&=&\Big(\vec{\lambda}_{2}\big(\lambda_{0_{1}}^{2}+\langle\vec{\lambda}_{1}|\vec{\lambda}_{1}\rangle\big)+\vec{\lambda}_{1}2\lambda_{0_{1}}\lambda_{0_{2}}\Big)\lambda_{0_{1}}+\vec{\lambda}_{1}\Big(\big(\lambda_{0_{1}}^{2}+\langle\vec{\lambda}_{1}|\vec{\lambda}_{1}\rangle\big)\lambda_{0_{2}}+\langle\vec{\lambda}_{1}2\lambda_{0_{1}}|\vec{\lambda}_{2}\rangle\Big)\nonumber\\
&=&\vec{\lambda}_{2}\Big(\lambda_{0_{1}}^{3}+\langle\vec{\lambda}_{1}|\vec{\lambda}_{1}\rangle\lambda_{0_{1}}\Big)+\vec{\lambda}_{1}\Big(2\lambda_{0_{1}}^{2}\lambda_{0_{2}}+\lambda_{0_{1}}^{2}\lambda_{0_{2}}+\langle\vec{\lambda}_{1}|\vec{\lambda}_{1}\rangle\lambda_{0_{2}}+\langle\vec{\lambda}_{1}|\vec{\lambda}_{2}\rangle2\lambda_{0_{1}}\Big)
\end{eqnarray}
The Jordan identity therefore holds. Symmetry of the usual inner product on $\mathbb{R}^{k}$ establishes that multiplication is commutative; moreover multiplication distributes with vector space operations and commutes with $\mathbb{R}$-multiplication:
\begin{eqnarray}
v_{1}\jProd(v_{2}+v_{3}\big)&=&\bigg(\lambda_{0_{1}}(\lambda_{0_{2}}+\lambda_{0_{3}})+\big\langle\vec{\lambda}_{1}|(\vec{\lambda}_{2}+\vec{\lambda}_{3})\big\rangle\bigg)\oplus\bigg((\vec{\lambda}_{2}+\vec{\lambda}_{3})\lambda_{0_{1}}+\vec{\lambda_{1}}(\lambda_{0_{2}}+\lambda_{0_{3}})\bigg)\nonumber\\
&=&(v_{1}\jProd v_{2})+(v_{1}\jProd v_{3})\label{dist}\\
v_{1}\alpha_{1}\jProd v_{2}\alpha_{2}&=&\Big(\lambda_{0_{1}}\lambda_{0_{2}}\alpha_{1}\alpha_{2}+\langle\vec{\lambda}_{1}\alpha_{1}|\vec{\lambda}_{2}\alpha_{2}\rangle\Big)\oplus\Big(\vec{\lambda}_{2}\alpha_{2}\lambda_{0_{1}}\alpha_{1}+\vec{\lambda}_{1}\alpha_{1}\lambda_{0_{2}}\alpha_{2}\Big)\nonumber\\
&=&(v_{1}\jProd v_{2})\alpha_{1}\alpha_{2}
\label{fieldcom}
\end{eqnarray}
where we have included Eq.~\eqref{dist} and Eq.~\eqref{fieldcom} for the sake of completeness. 
\end{proof}
\end{proposition}

\begin{proposition}\label{spinUnique} \textit{Let $\mathcal{V}_{k}$ and $\mathcal{W}_{k}$ be spin factors of finite cardinality $k$. There exists a jordan isomophism from $\mathcal{V}_{k}\longrightarrow\mathcal{W}_{k}$.}
\begin{proof} By definition, there exist Jordan algebras $\mathcal{A}$ and $\mathcal{B}$ such that $\mathcal{P}_{k}=\{s_{1},\dots,s_{k}\}\subset\mathcal{A}$ and $\mathcal{R}_{k}=\{t_{1},\dots,t_{k}\}\subset\mathcal{B}$ generate $\mathcal{V}_{k}$ and $\mathcal{W}_{k}$, respectively. In Eq.~\eqref{spinBasis} we stated that 
\begin{equation}
\forall v\in\mathcal{V}_{k}\;\exists\lambda_{0},\dots,\lambda_{k}\in\mathbb{R}:v=\mathbf{1}\lambda_{0}+\sum_{r=1}^{k}s_{k}\lambda_{k}\text{.}
\label{basis}
\end{equation}
It is actually easy to see that \eqref{basis} holds. Let us do so now. The anticommuting symmetries are linearly independent --- for the proof, let $\mathbb{Z}^{*}_{k\setminus i}=\{1,\dots,k\}\setminus\{i\}$ and suppose $\exists\vec{r}\in(\mathbb{Z}^{*}_{k\setminus i})^{\times k-1}$ with $\vec{r}\neq 0$ such that 
\begin{equation}
s_{i}=\sum_{j\in\mathbb{Z}^{*}_{k\setminus i}}s_{j}r_{j}\implies\forall l\in\mathbb{Z}^{*}_{k\setminus i}: s_{l}\jProd s_{i}=\mathbf{1}r_{l}\implies\forall l\in\mathbb{Z}^{*}_{k\setminus i}:r_{l}=0 \;\mbox{\Lightning}
\end{equation}
--- so the linear hull of $\mathcal{P}_{k}$ admits a basis $\{s_{1},\dots,s_{k}\}$. The jordan hull of the linear hull admits a basis $\{\mathbf{1},s_{1},\dots,s_{k}\}$, which is evident from the anticommutativity of the symmetries $s_{j}$ and that fact that
\begin{equation}
\mathbf{1}=\sum_{j=1}^{k}s_{j}r_{j}\implies \forall l\in\{1,\dots,k\}:\mathbf{1}\jProd s_{l}=s_{l}=\mathbf{1}r_{l}\implies\mathbf{1}=\mathbf{1}r_{l}^{2}\implies r_{l}=\pm 1\;\mbox{\Lightning}
\end{equation}
the contradiction coming by the construction of $s_{l}$, which are defined to be distinct from $\pm\mathbf{1}$. The linear hull of the jordan hull of the linear hull is again a Jordan algebra with basis $\{\mathbf{1},s_{1},\dots,s_{k}\}$, etcetera. Therefore one has that
\begin{eqnarray}
\mathcal{V}_{k}&=&\text{lin}_{\mathbb{R}}\left\{\mathbf{1}_{\mathcal{A}},s_{1},s_{2},\dots,s_{k}\right\}\text{.}\\
\mathcal{W}_{k}&=&\text{lin}_{\mathbb{R}}\left\{\mathbf{1}_{\mathcal{B}},t_{1},t_{2},\dots,t_{k}\right\}\text{.}
\end{eqnarray}
Define $f:\mathcal{V}_{k}\longrightarrow\mathcal{W}_{k}::s_{j}\longmapsto t_{j}$ and $f(\mathbf{1}_{\mathcal{A}})=\mathbf{1}_{\mathcal{B}}$. Then $\forall v_{1},v_{2}\in\mathcal{V}_{k}$ with $v_{1}=\lambda_{0}\oplus\vec{\lambda}$ and $v_{2}=\mu_{0}\oplus\vec{\mu}$ one computes
\begin{eqnarray}
f(v_{1}\jProd v_{2})&=&f\Big(\big(\langle\vec{\lambda}|\vec{\mu}\rangle+\lambda_{0}\mu_{0}\big)\oplus\big(\vec{\lambda}\mu_{0}+\vec{\mu}\lambda_{0}\big)\Big)\nonumber\\
&=&\mathbf{1}_{\mathcal{B}}\big(\langle\vec{\lambda}|\vec{\mu}\rangle+\lambda_{0}\mu_{0}\big)+\sum_{j=1}^{k}t_{j}(\lambda_{j}\mu_{0}+\mu_{j}\lambda_{0})\nonumber\\
&=&f(v_{1})\jProd f(v_{2})
\end{eqnarray} 
Thus $f$ is a jordan homomorphism. Further, by construction, $f$ is a bijection. Therefore \textit{the} spin factor of cardinality $k$ is unique (up to a jordan isomorphism).
\end{proof}
\end{proposition}


\begin{proposition}\label{v23prop} \textit{$\mathcal{V}_{2}$ and $\mathcal{V}_{3}$ are universally reversible.}
\begin{proof} We will prove the result for $\mathcal{V}_{3}$, the case of $\mathcal{V}_{2}$ following in a similar manner. Let $\mathcal{H}$ be a complex Hilbert space and let $\pi:\mathcal{V}_{3}\longrightarrow\mathfrak{B}(\mathcal{H})_{sa}::v\longmapsto u$ be a jordan monomorphism. Recall that we demand our homomorphisms to be unital: $\pi(\mathbf{1}_{\mathcal{V}_{3}})=\mathbf{1}_{\mathfrak{B}(\mathcal{H})_{sa}}$. For notational reasons, let $s_{0}:=\mathbf{1}_{\mathcal{V}_{3}}$ and $t_{0}:=\mathbf{1}_{\mathfrak{B}(\mathcal{H})_{sa}}$. Define $t_{a}=\pi(s_{a})$ for all $a\in\{1,2,3\}$. One has then $(t_{a}t_{b}+t_{b}t_{a})/2=\pi(s_{a}\jProd s_{b})=\pi(\delta_{a,b}\mathbf{1})=\delta_{a,b}\mathds{1}$ --- so $t_{a}$ anticommute. Furthermore $t_{a}^{2}=\pi(s_{a}^{2})=\pi(\mathbf{1})=\mathds{1}$. Now, let $v_{1},v_{2},\dots,v_{n}\in\mathcal{V}_{3}$ with $v_{j}=\lambda_{0_{j}}+\vec{\lambda_{j}}$. Let $\mathbb{Z}_{3}=\{0,1,2,3\}$. Then, by linearity,
\begin{eqnarray}
u_{1}\cdots u_{n}+u_{n}\cdots u_{1}
&=&\prod_{j=1}^{n}\sum_{a=0}^{3}t_{a}\lambda_{a_{j}}+\prod_{j=1}^{n}\sum_{a=0}^{3}t_{a}\lambda_{a_{n+1-j}}\nonumber\\
&=&\sum_{\vec{p}\in(\mathbb{Z}^{3})^{\times n}}\big(t_{p_{1}}t_{p_{2}}\cdots t_{p_{n}}+t_{p_{n}}t_{p_{n-1}}\cdots t_{p_{1}}\big)\lambda_{p_{1_{1}}}\lambda_{p_{2_{2}}}\cdots\lambda_{p_{n_{n}}}
\end{eqnarray}
Thus, by linearity of $\pi$, it suffices to prove that $\forall\vec{p}\in(\mathbb{Z}^{3})^{\times n}$ $\exists v\in\mathcal{V}_{3}$:
\begin{equation}
\pi(v)=t_{p_{1}}t_{p_{2}}\cdots t_{p_{n}}+t_{p_{n}}t_{p_{n-1}}\cdots t_{p_{1}}\text{.}
\end{equation}
Incidentally, notice that for $n=2$ one has $(t_{p_{1}}t_{p_{2}}+t_{p_{2}}t_{p_{1}})/2=\pi(s_{p_{1}}\jProd s_{p_{2}})$. Next, consider the case $n=3$. Should it arise that at least one $p_{j}=0$, then we reduce to the case of $n=2$; hence the desired $v$ exists. If it is not the case that at least one $p_{j}=0$, then by anticommutativity one has that $t_{p_{1}}t_{p_{2}}t_{p_{3}}+t_{p_{3}}t_{p_{2}}t_{p_{2}}=0$. Finally, consider the case $n\geq 4$. In this case, one has that \textit{at most three} $p_{j}$ are distinct elements of $\{1,2,3\}$. Furthermore, by anticommutivity, for those $j,j'\in\{1,\dots,n\}$ such that $p_{j}=p_{j'}$, one anticommutes the correspoding $t_{j}$ through an equal number of steps to meet $t_{j'}$ in both $t_{p_{1}}\cdots t_{p_{n}}$ and $t_{p_{n}}\cdots t_{p_{1}}$. Hence $\mathcal{V}_{3}$ is universally reversible (and so is $\mathcal{V}_{2})$. 
\end{proof}
\end{proposition}
\begin{proposition}\label{v4prop} \textit{$\mathcal{V}_{4}$ is not reversible.}\
\begin{proof} We follow Hanche-Olsen \cite{HancheOlsen1983}. Let $\pi:\mathcal{V}_{4}\longrightarrow\mathfrak{B}(\mathcal{H})_{sa}:v\longmapsto u$ be a jordan monomorphism. The question is not, of course, whether such monomorphims exist (they do, and all such monomorphisms take $s_{r}$ to anticommuting $t_{r}$); rather, one is about to prove that $\pi(\mathcal{V}_{4})$ is never such that $\pi(s_{1})\pi(s_{2})\pi(s_{3})\pi(s_{4})+\pi(s_{4})\pi(s_{3})\pi(s_{2})\pi(s_{1})=\pi(v)$ for some $v\in\mathcal{V}_{4}$. Indeed, suppose such a $v\in\mathcal{V}_{4}$ exists. By anticommutativity: $(t_{1}t_{2}t_{3}t_{4}+t_{4}t_{3}t_{2}t_{1})/2=t_{1}t_{2}t_{3}t_{4}$. Thus $\pi(v)^{2}=t_{1}t_{2}t_{3}t_{4}t_{1}t_{2}t_{3}t_{4}=-t_{1}t_{2}t_{3}t_{1}t_{2}t_{3}=-t_{1}t_{2}t_{1}t_{2}=\mathds{1}\implies v\ne 0$. However, one also has $\forall j\in\{1,2,3,4\}$ that $\pi(v\jProd s_{j})=t_{1}t_{2}t_{3}t_{4}t_{j}+t_{j}t_{4}t_{3}t_{2}t_{1}=0$. Therefore $v\ne\mathbf{1}\alpha$ for $\alpha\in\mathbb{R}$. Thus $\text{dim}_{\mathbb{R}}\mathcal{V}_{4}>5 \;\mbox{\Lightning}$.
\end{proof}
\end{proposition}

\begin{proposition}\label{v5prop} \textit{$\mathcal{V}_{5}$ is reversible, but \textit{not} universally so.}
\begin{proof} Define $\pi:\mathcal{V}_{5}\longrightarrow\mathcal{M}_{4}\big(\mathbb{C}\big)_{sa}::s_{r}\longmapsto t_{r}$ in terms of the complex Pauli matrices as follows:
\begin{eqnarray}
t_{0}=\sigma_{o}\otimes\sigma_{o}\hspace{0.4cm}t_{1}=\sigma_{z}\otimes\sigma_{o}\hspace{0.4cm}t_{2}=\sigma_{x}\otimes\sigma_{o}\hspace{0.4cm}
t_{3}=\sigma_{y}\otimes\sigma_{x}\hspace{0.4cm}t_{4}=\sigma_{y}\otimes\sigma_{z}\hspace{0.4cm}t_{5}=\sigma_{y}\otimes\sigma_{y}
\end{eqnarray}
By a similiar argument to one given in our proof of \cref{v23prop}, it suffices to observe that $(t_{1}t_{2}t_{3}t_{4}t_{5}+t_{5}t_{4}t_{3}t_{2}t_{1})/2=t_{1}t_{2}t_{3}t_{4}t_{5}=t_{0}=\pi(\mathbf{1})$ and $(t_{1}t_{2}t_{3}t_{4}+t_{4}t_{3}t_{2}t_{1})/2=t_{5}=\pi(s_{5})$. Hence $\mathcal{V}_{5}$ is reversible. Next, \textit{a la} Hanche-Olsen, define the jordan monomorphism $\psi:\mathcal{V}_{4}\longrightarrow\mathcal{M}_{4}\big(\mathbb{C}\big)_{sa}\oplus\mathcal{M}_{4}\big(\mathbb{C}\big)_{sa}$ via $\forall j\in{1,2,3,4}:\psi(s_{j})=t_{j}\oplus t_{j}$ and $\psi(s_{5})=t_{j}\oplus -t_{j}$. One then has $\psi(s_{1})\psi(s_{2})\psi(s_{3})\psi(s_{4})=t_{5}\oplus t_{5}$. Suppose there exists $v\in\mathcal{V}_{5}$ such that $\psi(v)=t_{5}\oplus t_{5}$. Thus $\psi(v)^{2}=\mathds{1}\implies v\neq0$. However, for all $j\in\{1,2,3,4\}$ one has that and $\pi(v\jProd s_{j})=0$; thus $v=\mathbf{1}\alpha+s_{5}\beta$ for some $\alpha,\beta\in\mathbb{R}$ simply from the vector space structure of $\mathcal{V}_{5}$. In fact, from $\psi(v)^{2}=\mathds{1}$ we are left with two mutually exclusive cases: $v=\mathbf{1}\alpha$ or $v=s_{5}\alpha$, both of which are impossible in light of the foregoing observations.
\end{proof}
\end{proposition}
\begin{proposition}\label{v6prop} \textit{Let $k\geq 6$. Then $\mathcal{V}_{k}$ is not reversible.}
\begin{proof} We follow Hanche-Olsen \cite{HancheOlsen1983}. Let $\pi:\mathcal{V}_{k}\longrightarrow\mathfrak{B}(\mathcal{H})_{sa}::s_{r}\longmapsto t_{r}$. Define $t_{5}=t_{1}t_{2}t_{3}t_{4}=(t_{1}t_{2}t_{3}t_{4}+t_{4}t_{3}t_{2}t_{1})/2$. Observe that $t_{5}^{2}=\mathds{1}$ and $\forall j\in\{1,2,3,4\}$ $t_{5}\jProd t_{j}=0$. \textit{Suppose} $t_{5}\in\pi(\mathcal{V}_{5})$. Then $\pi^{-1}(t_{5})=0\oplus\vec{\lambda}_{5}$ with unit vector $\vec{\lambda}_{5}$ orthogonal to $\vec{\lambda}_{j}$ for $\forall j\in\{1,2,3,4\}$. Consider $s_{6}=0\oplus\vec{\lambda}_{6}\in\mathcal{V}_{k}$ with $\langle\vec{\lambda}_{6}|\vec{\lambda}_{5}\rangle=0$. Dimensional grounds permit such consideration. Then $s_{6}\jProd s_{5}=0$, which implies $t_{5}t_{6}=-t_{6}t_{5}$. However, $t_{5}t_{6}=t_{1}t_{2}t_{3}t_{4}t_{6}=t_{6}t_{1}t_{2}t_{3}t_{4}=t_{6}t_{5}$ Thus $t_{5}t_{6}=0$, which is impossible in light of $t_{5}t_{6}t_{5}t_{6}=t_{5}t_{6}t_{1}t_{2}t_{3}t_{4}t_{6}=t_{5}t_{1}t_{2}t_{3}t_{4}t_{6}^{2}=t_{5}^{2}=\mathds{1}$.
\end{proof}
\end{proposition}
\newpage
\section{Review of the Generalized Gell-Mann Matrices}\label{gmApp}
This appendix reviews the Generalized Gell-Mann matrices. Our primary reference is \cite{Kalev2014}. The details collected herein are elementary, though extremely important for our proofs in \cref{canTPcomps}.

\noindent We consider the real vector space $\mathcal{M}_{n}(\mathbb{C})_{\text{sa}}$. Let $E_{r,s}=|e_{r}\rangle\langle e_{s}|$ where $\{|e_{r}\rangle\}$ is the standard orthonormal basis for $\mathbb{C}^{n}$ --- \textit{i.e.} with no complex entries, just zeros with a single entry of unity; thus also the standard orthonormal basis for $\mathbb{R}^{n}$ should one wish to regard it as such. Just for illustration, with $n=3$ one has
\begin{equation}
E_{1,2}=\begin{pmatrix} 0 & 1 & 0 \\ 0 & 0 & 0 \\ 0 & 0 & 0 \end{pmatrix}
\end{equation}
The set $\big\{E_{r,s}:r,s\in\{1,\dots, n\}\big\}\subset\mathcal{M}_{n}(\mathbb{C})$ is clearly a basis for the complex vector space $\mathcal{M}_{n}(\mathbb{C})$; $\text{dim}_{\mathbb{C}}\mathcal{M}_{n}(\mathbb{C})=n^{2}$. We are looking for a basis for the real vector space of self-adjoint $n\times n$ complex matrices. The Generalized Gell-Mann Matrices are very nice! Following \cite{Kalev2014}, let $r,s\in\{1,\dots,n\}$ and define the following $d^{2}-1$ matrices
\begin{equation}
G_{r,s}=
\begin{cases}
\textstyle{\frac{1}{\sqrt{2}}}\big(E_{r,s}+E_{s,r}\big) & r<s \\
\textstyle{\frac{i}{\sqrt{2}}}\big(E_{r,s}-E_{s,r}\big) & s<r \\
\textstyle{\frac{1}{\sqrt{r(r+1)}}}\big(-rE_{r+1,r+1}+\sum_{k=1}^{r}E_{k,k}\big) & s=r\neq n
\end{cases}
\label{GM}
\end{equation}
\begin{proposition}\label{GMprop1} $G_{r,s}\in\mathcal{M}_{n}(\mathbb{C})_{\text{sa}}\textit{.}$
\begin{proof} Trivial by inspection (note $G_{r,r}$ is real diagonal.)
\end{proof}
\end{proposition}
\begin{proposition}\label{GMprop2} $\forall r,s\in\{1,\dots,n\}$ \textit{(if $r=s$ then $r=s\neq n$):}
\begin{equation}
\mathrm{Tr}\big(G_{r,s}\big)=0\textit{.}
\end{equation} 
\begin{proof} Let $r<s$. Then 
\begin{equation}
\mathrm{Tr}\big(G_{r,s}\big)=\textstyle{\frac{1}{\sqrt{2}}}(\langle e_{s}|e_{r}\rangle+\langle e_{r}|e_{s}\rangle)=0\text{.}
\end{equation} 
Let $s<r$. Then 
\begin{equation}
\mathrm{Tr}\big(G_{r,s}\big)=\textstyle{\frac{i}{\sqrt{2}}}(\langle e_{s}|e_{r}\rangle-\langle e_{r}|e_{s}\rangle)=0\text{.}
\end{equation} 
Finally, for arbitrary $r\in\{1,\dots,n-1\}$:
\begin{equation}
\mathrm{Tr}\big(G_{r,r}\big)=\mathrm{Tr}\left(\textstyle{\frac{1}{\sqrt{r(r+1)}}}\big(-rE_{r+1,r+1}+\sum_{k=1}^{r}E_{k,k}\big)\right)=\textstyle{\frac{1}{\sqrt{r(r+1)}}}(r-r)=0\text{.}
\end{equation}
This concludes analysis of all cases.
\end{proof}
\end{proposition}

\begin{proposition}\label{GMprop3} $\forall r,s\in\{1,\dots,n\}$\textit{ (if $r=s$ then $r=s\neq n$):}
\begin{equation}
\mathrm{Tr}\big(G_{r,s}^{2}\big)=1\textit{.}
\end{equation}
\begin{proof} Let $r<s$. Then
\begin{eqnarray}
\mathrm{Tr}\big(G_{r,s}^{2}\big)&=&\mathrm{Tr}\left(\textstyle{\frac{1}{2}}(E_{r,s}+E_{s,r})(E_{r,s}+E_{s,r})\right)\nonumber\\
&=&\textstyle{\frac{1}{2}}\mathrm{Tr}\Big(E_{r,s}E_{r,s}+E_{r,s}E_{s,r}+E_{s,r}E_{r,s}+E_{s,r}E_{s,r}\Big)\nonumber\\
&=&\textstyle{\frac{1}{2}}(0+1+1+0)\nonumber\\
&=&1\text{.}
\end{eqnarray}
Now let $s<r$. Then
\begin{eqnarray}
\mathrm{Tr}\big(G_{r,s}^{2}\big)&=&\mathrm{Tr}\left(\textstyle{-\frac{1}{2}}(E_{r,s}-E_{s,r})(E_{r,s}-E_{s,r})\right)\nonumber\\
&=&\textstyle{\frac{1}{2}}\mathrm{Tr}\Big(E_{r,s}E_{r,s}-E_{r,s}E_{s,r}-E_{s,r}E_{r,s}+E_{s,r}E_{s,r}\Big)\nonumber\\
&=&-\textstyle{\frac{1}{2}}(0-1-1+0)\nonumber\\
&=&1\text{.}
\end{eqnarray}
Finally let $r=s\neq n$. Then
\begin{eqnarray}
\mathrm{Tr}\big(G_{r,r}^{2}\big)&=&\textstyle{\frac{1}{r(r+1)}}\mathrm{Tr}\Big(\big(-rE_{r+1,r+1}+\sum_{k=1}^{r}E_{k,k}\big)\big(-rE_{r+1,r+1}+\sum_{k=1}^{r}E_{k,k}\big)\Big) \nonumber\\
&=&\textstyle{\frac{1}{r(r+1)}}\mathrm{Tr}\left(r^{2}E_{r+1,r+1}-r\sum_{k=1}^{r}\big(E_{r+1,r+1}E_{k,k}+E_{k,k}E_{r+1,r+1}\big)+\sum_{k=1}^{r}\sum_{j=1}^{r}E_{k,k}E_{j,j}\right) \nonumber\\
&=&\textstyle{\frac{1}{r(r+1)}}(r^{2}+0+0+r) \nonumber \\
&=&1\text{.}
\end{eqnarray}
This concludes analysis of all cases.
\end{proof}
\end{proposition}
\begin{proposition}\label{GMprop4} $\forall r,s,t,v\in\{1,\dots,n\}$ \textit{(if $r=s$ then $r=s\neq n$ and if $t=v$ then $t=v\neq n$) with $(r,s)\neq (t,v)$:} 
\begin{equation}
\mathrm{Tr}\big(G_{r,s}G_{t,v}\big)=0\text{.}
\end{equation}
\begin{proof} We will perform a mutually exclusive and exhaustive case analysis.\\[0.2cm]
\underline{Case A: r=t}:\\[0.2cm]
$\clubsuit$ Subcase A.1:  $r<s$ and $r=t<v$. Then 
\begin{eqnarray}
\mathrm{Tr}\big(G_{r,s}G_{r,v}\big)&=&\textstyle{\frac{1}{2}}\mathrm{Tr}\Big(\big(E_{r,s}+E_{s,r}\big)\big(E_{t,v}+E_{v,t}\big)\Big)\nonumber\\
&=&\textstyle{\frac{1}{2}}\big(\langle e_{v}|e_{r}\rangle\langle e_{s}|e_{t}\rangle+\langle e_{t}|e_{r}\rangle\langle e_{s}|e_{v}\rangle+\langle e_{t}|e_{s}\rangle\langle e_{r}|e_{v}\rangle+\langle e_{v}|e_{s}\rangle\langle e_{r}|e_{t}\rangle\big)\nonumber\\
&=&0\text{.}
\end{eqnarray}
The first and third terms are zero because $r<v$, the second and fourth because $s\neq v$ by assumption since $r=t$.\\[0.2cm]
$\clubsuit$ Subcase A.2: $r>s$ and $r=t<v$. Then 
\begin{eqnarray}
\mathrm{Tr}\big(G_{r,s}G_{r,v}\big)&=&\textstyle{\frac{i}{2}}\mathrm{Tr}\Big(\big(E_{r,s}-E_{s,r}\big)\big(E_{t,v}+E_{v,t}\big)\Big)\nonumber\\
&=&\textstyle{\frac{i}{2}}\big(\langle e_{v}|e_{r}\rangle\langle e_{s}|e_{t}\rangle+\langle e_{t}|e_{r}\rangle\langle e_{s}|e_{v}\rangle-\langle e_{t}|e_{s}\rangle\langle e_{r}|e_{v}\rangle-\langle e_{v}|e_{s}\rangle\langle e_{r}|e_{t}\rangle\big)\nonumber\\
&=&0\text{.}
\end{eqnarray}
The reasoning for the zeros is exactly as in Subcase A.1.\\[0.2cm]
$\clubsuit$ Subcase A.3: $r=s$ and $r=t<v$. Then
\begin{eqnarray}
\mathrm{Tr}\big(G_{r,r}G_{r,v}\big)&=&\textstyle{\frac{1}{\sqrt{2r(r+1)}}}\mathrm{Tr}\Big(\big(-rE_{r+1,r+1}+\sum_{k=1}^{r}E_{k,k}\big)\big(E_{t,v}+E_{v,t}\big)\Big)\nonumber\\
&=&\textstyle{\frac{1}{\sqrt{2r(r+1)}}}\big(-r\langle e_{v}|e_{r+1}\rangle\langle e_{r+1}|e_{t}\rangle-r\langle e_{t}|e_{r+1}\rangle\langle e_{r+1}|e_{v}\rangle\big)\nonumber\\
&+&\textstyle{\frac{1}{\sqrt{2r(r+1)}}}\big(\sum_{k=1}^{r}(\langle e_{v}|e_{k}\rangle\langle e_{k}|e_{t}\rangle+\langle e_{v}|e_{k}\rangle\langle e_{k}|e_{t})\rangle\big)\nonumber\\
&=&0\text{.}
\end{eqnarray} 
The first and second terms are zero because $r=t$, the third because $r<v$.\\[0.2cm]
$\clubsuit$ Subcase A.4: $r<s$ and $r=t>v$. Then
\begin{eqnarray}
\mathrm{Tr}\big(G_{r,s}G_{r,v}\big)&=&\textstyle{\frac{i}{2}}\mathrm{Tr}\Big(\big(E_{r,s}+E_{s,r}\big)\big(E_{t,v}-E_{v,t}\big)\Big)\nonumber\\
&=&\textstyle{\frac{i}{2}}\big(\langle e_{v}|e_{r}\rangle\langle e_{s}|e_{t}\rangle-\langle e_{t}|e_{r}\rangle\langle e_{s}|e_{v}\rangle+\langle e_{t}|e_{s}\rangle\langle e_{r}|e_{v}\rangle-\langle e_{v}|e_{s}\rangle\langle e_{r}|e_{t}\rangle\big)\nonumber\\
&=&0\text{.}
\end{eqnarray}
The first and third terms are zero because $r>v$, the second and fourth because $s\neq v$ by assumption since $r=t$.\\[0.2cm]
$\clubsuit$ Subcase A.5: $r>s$ and $r=t>v$. Then
\begin{eqnarray}
\mathrm{Tr}\big(G_{r,s}G_{r,v}\big)&=&\textstyle{-\frac{1}{2}}\mathrm{Tr}\Big(\big(E_{r,s}-E_{s,r}\big)\big(E_{t,v}-E_{v,t}\big)\Big)\nonumber\\
&=&\textstyle{-\frac{1}{2}}\big(\langle e_{v}|e_{r}\rangle\langle e_{s}|e_{t}\rangle-\langle e_{t}|e_{r}\rangle\langle e_{s}|e_{v}\rangle-\langle e_{t}|e_{s}\rangle\langle e_{r}|e_{v}\rangle+\langle e_{v}|e_{s}\rangle\langle e_{r}|e_{t}\rangle\big)\nonumber\\
&=&0
\end{eqnarray}
The reasoning for the zeros is exactly as in Subcase A.4.\\[0.2cm]
$\clubsuit$ Subcase A.6: $r=s$ and $r=t>v$. Then
\begin{eqnarray}
\mathrm{Tr}\big(G_{r,r}G_{r,v}\big)&=&\textstyle{\frac{i}{\sqrt{2r(r+1)}}}\mathrm{Tr}\Big(\big(-rE_{r+1,r+1}+\sum_{k=1}^{r}E_{k,k}\big)\big(E_{t,v}-E_{v,t}\big)\Big)\nonumber\\
&=&\textstyle{\frac{i}{\sqrt{2r(r+1)}}}\big(-r\langle e_{v}|e_{r+1}\rangle\langle e_{r+1}|e_{t}\rangle+r\langle e_{t}|e_{r+1}\rangle\langle e_{r+1}|e_{v}\rangle\big)\nonumber\\
&+&\textstyle{\frac{i}{\sqrt{2r(r+1)}}}\big(\sum_{k=1}^{r}(\langle e_{v}|e_{k}\rangle\langle e_{k}|e_{t}\rangle-\langle e_{v}|e_{k}\rangle\langle e_{k}|e_{t})\rangle\big)\nonumber\\
&=&0\text{.}
\end{eqnarray}
Since $r=t>v$ the first two terms are zero. The sum is zero by the commutativity of multiplication in $\mathbb{C}$ (\textit{i.e}.\ the summand itself is zero.)\\[0.2cm]
\noindent$\clubsuit$ Subcase A.7: $r<s$ and $r=t=v$. Then
\begin{eqnarray}
\mathrm{Tr}\big(G_{r,s}G_{r,r}\big)&=&\textstyle{\frac{1}{\sqrt{2r(r+1)}}}\mathrm{Tr}\Big(\big(E_{r,s}+E_{s,r}\big)\big(-rE_{r+1,r+1}+\sum_{k=1}^{r}E_{k,k}\big)\Big)\nonumber\\
&=&\textstyle{\frac{1}{\sqrt{2r(r+1)}}}\big(-r\langle e_{r+1}|e_{r}\rangle\langle e_{s}|e_{r+1}\rangle+\sum_{k=1}^{r}\langle e_{k}|e_{r}\rangle\langle e_{s}|e_{k}\rangle\big)\nonumber\\
&+&\textstyle{\frac{1}{\sqrt{2r(r+1)}}}\big(-r\langle e_{r+1}|e_{s}\rangle\langle e_{r}|e_{r+1}\rangle+\sum_{k=1}^{r}\langle e_{k}|e_{s}\rangle\langle e_{r}|e_{k}\rangle\rangle\big)\nonumber\\
&=&0\text{.}
\end{eqnarray}
The first and third terms are zero because $r<r+1$, the second and fourth because $s>r$.\\[0.2cm]
$\clubsuit$ Subcase A.8: $r>s$ and $r=t=v$. Then
\begin{eqnarray}
\mathrm{Tr}\big(G_{r,s}G_{r,r}\big)&=&\textstyle{\frac{1=i}{\sqrt{2r(r+1)}}}\mathrm{Tr}\Big(\big(E_{r,s}-E_{s,r}\big)\big(-rE_{r+1,r+1}+\sum_{k=1}^{r}E_{k,k}\big)\Big)\nonumber\\
&=&\textstyle{\frac{1}{\sqrt{2r(r+1)}}}\big(-r\langle e_{r+1}|e_{r}\rangle\langle e_{s}|e_{r+1}\rangle+\sum_{k=1}^{r}\langle e_{k}|e_{r}\rangle\langle e_{s}|e_{k}\rangle\big)\nonumber\\
&+&\textstyle{\frac{1}{\sqrt{2r(r+1)}}}\big(r\langle e_{r+1}|e_{s}\rangle\langle e_{r}|e_{r+1}\rangle-\sum_{k=1}^{r}\langle e_{k}|e_{s}\rangle\langle e_{r}|e_{k}\rangle\rangle\big)\nonumber\\
&=&0\text{.}
\end{eqnarray}
The first and third terms are zero because $r<r+1$, the second and fourth because $s<r$. This exhausts all possible cases when $r=t$ --- remember, we are considering $(r,s)\neq(t,v)$.\\[0.2cm]
\underline{Case B: $r\neq t$}:\\[0.2cm] 
$\clubsuit$ Subcase B.1: $r<s$ and $t<v$. Then
\begin{eqnarray}
\mathrm{Tr}\big(G_{r,s}G_{t,v}\big)&=&\textstyle{\frac{1}{2}}\mathrm{Tr}\Big(\big(E_{r,s}+E_{s,r}\big)\big(E_{t,v}+E_{v,t}\big)\Big)\nonumber\\
&=&\textstyle{\frac{1}{2}}\big(\langle e_{v}|e_{r}\rangle\langle e_{s}|e_{t}\rangle+\langle e_{t}|e_{r}\rangle\langle e_{s}|e_{v}\rangle+\langle e_{v}|e_{s}\rangle\langle e_{r}|e_{t}\rangle+\langle e_{t}|e_{s}\rangle\langle e_{r}|e_{v}\rangle\big)\nonumber\\
&=&0\text{.}
\end{eqnarray}
The second and third terms vanish because $r\neq t$. We will prove that the first term vanishes by contradiction. \textit{Suppose} $\langle e_{v}|e_{r}\rangle\langle e_{s}|e_{t}\rangle\neq 0$. Then $v=r$ and $s=t$. Thus from $r<s$ we get $v<s$: a contradiction; thus the first term vanishes. The same argument holds for the fourth term.\\[0.2cm]
$\clubsuit$ Subcase B.2: $r>s$ and $t<v$. Then
\begin{eqnarray}
\mathrm{Tr}\big(G_{r,s}G_{t,v}\big)&=&\textstyle{\frac{i}{2}}\mathrm{Tr}\Big(\big(E_{r,s}-E_{s,r}\big)\big(E_{t,v}+E_{v,t}\big)\Big)\nonumber\\
&=&\textstyle{\frac{1}{2}}\big(\langle e_{v}|e_{r}\rangle\langle e_{s}|e_{t}\rangle+\langle e_{t}|e_{r}\rangle\langle e_{s}|e_{v}\rangle-\langle e_{v}|e_{s}\rangle\langle e_{r}|e_{t}\rangle-\langle e_{t}|e_{s}\rangle\langle e_{r}|e_{v}\rangle\big)\nonumber\\
&=&0\text{.}
\end{eqnarray}
The second and third terms vanish because $r\neq t$. The first and fourth terms cancel because $\langle e_{j}|e_{k}\rangle\in\mathbb{R}$.\\[0.2cm]
$\clubsuit$ Subcase B.3: $r=s$ and $t<v$. Then
\begin{eqnarray}
\mathrm{Tr}\big(G_{r,r}G_{t,v}\big)&=&\textstyle{\frac{1}{\sqrt{2r(r+1)}}}\mathrm{Tr}\Big(\big(-rE_{r+1,r+1}+\sum_{k=1}^{r}E_{k,k}\big)\big(E_{t,v}+E_{v,t}\big)\Big)\nonumber\\
&=&\textstyle{\frac{1}{\sqrt{2r(r+1)}}}\big(-r\langle e_{v}|e_{r+1}\rangle\langle e_{r+1}|e_{t}\rangle-r\langle e_{t}|e_{r+1}\rangle\langle e_{r+1}|e_{v}\rangle\big)\nonumber\\
&+&\textstyle{\frac{1}{\sqrt{2r(r+1)}}}\big(\sum_{k=1}^{r}(\langle e_{v}|e_{k}\rangle\langle e_{k}|e_{t}\rangle+\langle e_{v}|e_{k}\rangle\langle e_{k}|e_{t}\rangle)\big)\nonumber\\
&=&0\text{.}
\end{eqnarray}
All terms vanish because $t<v$.\\[0.2cm]
$\clubsuit$ Subcase B.4: $r<s$ and $t>v$. Then
\begin{eqnarray}
\mathrm{Tr}\big(G_{r,s}G_{t,v}\big)&=&\textstyle{\frac{i}{2}}\mathrm{Tr}\Big(\big(E_{r,s}+E_{s,r}\big)\big(E_{t,v}-E_{v,t}\big)\Big)\nonumber\\
&=&\textstyle{\frac{i}{2}}\big(\langle e_{v}|e_{r}\rangle\langle e_{s}|e_{t}\rangle-\langle e_{t}|e_{r}\rangle\langle e_{s}|e_{v}\rangle+\langle e_{v}|e_{s}\rangle\langle e_{r}|e_{t}\rangle-\langle e_{t}|e_{s}\rangle\langle e_{r}|e_{v}\rangle\big)\nonumber\\
&=&0\text{.}
\end{eqnarray}
The reasoning for the zeros is exactly as in Subcase B.2.\\[0.2cm]
$\clubsuit$ Subcase B.5: $r>s$ and $t>v$. Then
\begin{eqnarray}
\mathrm{Tr}\big(G_{r,s}G_{t,v}\big)&=&\textstyle{-\frac{1}{2}}\mathrm{Tr}\Big(\big(E_{r,s}-E_{s,r}\big)\big(E_{t,v}-E_{v,t}\big)\Big)\nonumber\\
&=&\textstyle{-\frac{1}{2}}\big(\langle e_{v}|e_{r}\rangle\langle e_{s}|e_{t}\rangle-\langle e_{t}|e_{r}\rangle\langle e_{s}|e_{v}\rangle-\langle e_{v}|e_{s}\rangle\langle e_{r}|e_{t}\rangle+\langle e_{t}|e_{s}\rangle\langle e_{r}|e_{v}\rangle\big)\nonumber\\
&=&0\text{.}
\end{eqnarray}
The reasoning for the zeros is exactly as in Subcase B.1.\\[0.2cm]
$\clubsuit$ Subcase B.6: $r=s$ and $t>v$. Then
\begin{eqnarray}
\mathrm{Tr}\big(G_{r,r}G_{t,v}\big)&=&\textstyle{\frac{i}{\sqrt{2r(r+1)}}}\mathrm{Tr}\Big(\big(-rE_{r+1,r+1}+\sum_{k=1}^{r}E_{k,k}\big)\big(E_{t,v}-E_{v,t}\big)\Big)\nonumber\\
&=&\textstyle{\frac{i}{\sqrt{2r(r+1)}}}\big(-r\langle e_{v}|e_{r+1}\rangle\langle e_{r+1}|e_{t}\rangle+r\langle e_{t}|e_{r+1}\rangle\langle e_{r+1}|e_{v}\rangle\big)\nonumber\\
&+&\textstyle{\frac{i}{\sqrt{2r(r+1)}}}\sum_{k=1}^{r}(\langle e_{v}|e_{k}\rangle\langle e_{k}|e_{t}\rangle-\langle e_{t}|e_{k}\rangle\langle e_{k}|e_{v})\rangle\big)\nonumber\\
&=&0\text{.}
\end{eqnarray}
All terms vanish because $t>v$.
\noindent $\clubsuit$ Subcase B.7: $r<s$ and $t=v$. Then
\begin{eqnarray}
\mathrm{Tr}\big(G_{r,s}G_{t,t}\big)&=&\textstyle{\frac{1}{\sqrt{2r(r+1)}}}\mathrm{Tr}\Big(\big(E_{r,s}+E_{s,r}\big)\big(-tE_{t+1,t+1}+\sum_{k=1}^{t}E_{k,k}\big)\Big)\nonumber\\
&=&\textstyle{\frac{1}{\sqrt{2r(r+1)}}}\big(-t\langle e_{t+1}|e_{r}\rangle\langle e_{s}|e_{t+1}\rangle-t\langle e_{t+1}|e_{s}\rangle\langle e_{r}|e_{t+1}\rangle\big)\nonumber\\
&+&\textstyle{\frac{1}{\sqrt{2r(r+1)}}}\big(\sum_{k=1}^{r}(\langle e_{k}|e_{r}\rangle\langle e_{s}|e_{k}\rangle+\langle e_{k}|e_{s}\rangle\langle e_{r}|e_{k})\rangle\big)\nonumber\\
&=&0\text{.}
\end{eqnarray}
All terms vanish because $r<s$.\\[0.2cm]
$\clubsuit$ Subcase B.8: $r>s$ and $t=v$. Then
\begin{eqnarray}
\mathrm{Tr}\big(G_{r,s}G_{t,t}\big)&=&\textstyle{\frac{i}{\sqrt{2r(r+1)}}}\mathrm{Tr}\Big(\big(E_{r,s}-E_{s,r}\big)\big(-tE_{t+1,t+1}+\sum_{k=1}^{t}E_{k,k}\big)\Big)\nonumber\\
&=&\textstyle{\frac{1}{\sqrt{2r(r+1)}}}\big(-t\langle e_{t+1}|e_{r}\rangle\langle e_{s}|e_{t+1}\rangle+t\langle e_{t+1}|e_{s}\rangle\langle e_{r}|e_{t+1}\rangle\big)\nonumber\\
&+&\textstyle{\frac{1}{\sqrt{2r(r+1)}}}\big(\sum_{k=1}^{r}(\langle e_{k}|e_{r}\rangle\langle e_{s}|e_{k}\rangle-\langle e_{k}|e_{s}\rangle\langle e_{r}|e_{k})\rangle\big)\nonumber\\
&=&0\text{.}
\end{eqnarray}
All terms vanish because $r>s$.\\[0.2cm]
$\clubsuit$ Subcase B.9: $r=s$ and $t=v$. Our considerations are such that, then: $s=r\neq t=v$. So,

\begin{eqnarray}
\mathrm{Tr}\big(G_{r,r}G_{t,t}\big)&=&\textstyle{\frac{1}{\sqrt{rt(r+1)(t+1)}}}\mathrm{Tr}\Big(\big(-rE_{r+1,r+1}+\sum_{k=1}^{r}E_{k,k}\big)\big(-tE_{t+1,t+1}+\sum_{j=1}^{t}E_{j,j}\big)\Big)\nonumber\\
&=&\textstyle{\frac{1}{\sqrt{rt(r+1)(t+1)}}}\big(rt\langle e_{t+1}|e_{r+1}\rangle\langle e_{r+1}|e_{t+1}\rangle-r\sum_{j=1}^{t}\langle e_{j}|e_{r+1}\rangle\langle e_{r+1}|e_{j}\rangle\big)\nonumber\\
&+&\textstyle{\frac{1}{\sqrt{rt(r+1)(t+1)}}}\big(-t\sum_{k=1}^{r}\langle e_{t+1}|e_{k}\rangle\langle e_{k}|e_{t+1}\rangle+\sum_{k=1}^{r}\sum_{j=1}^{t}\langle e_{j}|e_{k}\rangle\langle e_{k}|e_{j}\rangle\big) \nonumber\\
&=&0\text{.}
\end{eqnarray} 
The first term vanishes because we must have $r\neq t$. To show that the other three terms vanish, we will need to consider some sub-subcases.\\[0.2cm]
First, consider the sub-subcase $r<t$. Then $r+1$ is at most $t$, and the second term contributes a factor of $-r$; the third zero, and the fourth a factor of $+r$. Let us prove that the fourth term indeed contributes a factor of $+r$. Let $\Pi_{t}$ be the rank-$t$ projector onto the subspace spanned by $\big\{|e_{j}\rangle:j\in\{1,\dots,t\}\big\}$. Similarly, let $\Pi_{r}$ be the rank-$r$ projector onto the subspace spanned by $\big\{|e_{k}\rangle:k\in\{1,\dots,r\}\big\}$. Then we can write the forth term (ignoring the overall multiplicative constant $\textstyle{\frac{1}{\sqrt{rt(r+1)(t+1)}}}$) as
\begin{equation}
\sum_{k=1}^{r}\sum_{j=1}^{t}\langle e_{j}|e_{k}\rangle\langle e_{k}|e_{j}\rangle=\sum_{j=1}^{t}\langle e_{j}|\Pi_{r}||e_{j}\rangle=\mathrm{Tr}\big(\Pi_{t}\Pi_{r}\big)=\text{dim}_{\mathbb{C}}\Big(\text{span}_{\mathbb{C}}\big\{|e_{k}\rangle:k\in\{1,\dots,r\}\big\}\Big)=r\text{.}
\end{equation}
The second sub-subcase ($r>t$) can be handled in an entirely similar manner. This concludes our analysis, and, incidentally, the proof.
\end{proof}
\end{proposition}
\begin{proposition}\label{GMprop5} $\{G_{r,s}\}$ \textit{constitute an orthonormal basis for the real vector space of traceless self-adjoint $n\times n$ complex matrices.}
\begin{proof} Immediate from \cref{GMprop1}, \cref{GMprop2}, \cref{GMprop3}, and \cref{GMprop4}. 
\end{proof}
\end{proposition}
\begin{proposition}\label{GMprop6} $\{\mathds{1}_{n},G_{r,s}\}$ \textit{constitute an orthonormal basis for the real vector space} $\mathcal{M}_{n}(\mathbb{C})_{\text{sa}}$\textit{.}
\begin{proof} Immediate consequence of \cref{GMprop5} and $\mathrm{Tr}(\mathds{1}_{n}G_{r,s})=\mathrm{Tr}(G_{r,s})=0$ from \cref{GMprop2}.
\end{proof}
\end{proposition}
\newpage
\section{Universal Tensor Product of Two Qudits}\label{2copies}
In this appendix, we explicitly compute the following universal tensor product
\begin{equation}
\mathcal{M}_{n}(\mathbb{C})_{\text{sa}}\tilde{\otimes}\mathcal{M}_{m}(\mathbb{C})_{\text{sa}}\cong\mathcal{M}_{nm}(\mathbb{C})_{\text{sa}}\oplus\mathcal{M}_{nm}(\mathbb{C})_{\text{sa}}
\end{equation}
To begin, we compress our notation so that $\mathcal{A}=\mathcal{M}_{n}(\mathbb{C})_{\text{sa}}$ and $\mathcal{B}=\mathcal{M}_{m}(\mathbb{C})_{\text{sa}}$. Let $\mathfrak{C}$ be the \textsc{eja} generated by $\psi_{\mathcal{A}}(\mathcal{A})\otimes\psi_{\mathcal{B}}(\mathcal{B})$. One notes that $\mathfrak{C}$ is isomorphic to a Jordan subalgebra of $\mathcal{M}_{nm}\big(\mathbb{C}\big)_{\text{sa}}\oplus\;\mathcal{M}_{nm}\big(\mathbb{C}\big)_{\text{sa}}$ generated by elements of the form $(a\otimes b)\oplus (a\otimes b^{t})$, where $a\in\mathcal{A}$ and $b\in\mathcal{B}$. There exist orthonormal bases for $\mathcal{A}$ and $\mathcal{B}$--- for instance, the sets of respective generalized Gell-Mann matrices --- whose tensor products form an orthonormal basis for $\mathcal{M}_{nm}(\mathbb{C})_{\text{sa}}$. Thus, $\{x\oplus x^{1\otimes t}:x\in\mathcal{M}_{nm}(\mathbb{C})_{\text{sa}}\}\subset\mathfrak{C}$. Now, let $\{|e_{r}\rangle\in\mathbb{C}^{n}\;|\;\forall r,s\in\{0,\dots,n-1\}:\langle e_{r}|e_{s}\rangle=\delta_{rs}\}$ and $\{|f_{u}\rangle\in\mathbb{C}^{m}\;|\;\forall u,v\in\{0,\dots,m-1\}:\langle f_{u}|f_{v}\rangle=\delta_{uv}\}$ be orthonormal bases $\mathbb{C}^{n}$ and $\mathbb{C}^{m}$, respectively. Without the loss of generality, assume that $n\geq m$. Let $r,s\in\{0,\dots,m-1\}$ such that $r\neq s$. Let
\begin{equation}
|\Phi^{+}_{r,s}\rangle=\frac{1}{\sqrt{2}}\big(|e_{r}\rangle\otimes |f_{r}\rangle+|e_{s}\rangle\otimes|f_{s}\rangle\big)\text{.}
\end{equation}
Let 
\begin{equation}
|\Psi^{+}_{r,s}\rangle=\frac{1}{\sqrt{2}}\big(|e_{r}\rangle\otimes |f_{s}\rangle+|e_{s}\rangle\otimes|f_{r}\rangle\big)\text{.}
\end{equation}
Define
\begin{eqnarray}
x_{\Phi^{+}_{r,s}}&\equiv&|\Phi^{+}_{r,s}\rangle\langle\Phi^{+}_{r,s}|\nonumber\\
&=&\frac{1}{2}\Big(|e_{r}\rangle\langle e_{r}|\otimes|f_{r}\rangle\langle f_{r}|+|e_{s}\rangle\langle e_{s}|\otimes|f_{s}\rangle\langle f_{s}|+|e_{s}\rangle\langle e_{r}|\otimes|f_{s}\rangle\langle f_{r}|+|e_{r}\rangle\langle e_{s}|\otimes|f_{r}\rangle\langle f_{s}|\Big)
\end{eqnarray}
Define
\begin{eqnarray}
x_{\Psi^{+}_{r,s}}&\equiv&|\Psi^{+}_{r,s}\rangle\langle\Psi^{+}_{r,s}|\nonumber\\
&=&\frac{1}{2}\Big(|e_{r}\rangle\langle e_{r}|\otimes|f_{s}\rangle\langle f_{s}|+|e_{s}\rangle\langle e_{s}|\otimes|f_{r}\rangle\langle f_{r}|+|e_{r}\rangle\langle e_{s}|\otimes|f_{s}\rangle\langle f_{r}|+|e_{s}\rangle\langle e_{r}|\otimes|f_{r}\rangle\langle f_{s}|\Big)
\end{eqnarray}
By virtue of the orthogonality of $|\Phi^{+}_{r,s}\rangle$ and $|\Psi^{+}_{r,s}\rangle$, one has that $x_{\Phi^{+}_{r,s}}\jProd x_{\Psi^{+}_{r,s}}=0$.
Therefore
\begin{equation}
\left\{0\oplus\big(x_{\Phi^{+}_{r,s}}^{1\otimes t}\jProd x_{\Psi^{+}_{r,s}}^{1\otimes t}\big),\big(x_{\Phi^{+}_{r,s}}^{1\otimes t}\jProd x_{\Psi^{+}_{r,s}}^{1\otimes t}\big)\oplus 0\right\}\subset\mathfrak{C}\text{.}
\end{equation} 
Next, observe that
\begin{equation}
x_{\Phi^{+}_{r,s}}^{1\otimes t}=\frac{1}{2}\Big(|e_{r}\rangle\langle e_{r}|\otimes|f_{r}\rangle\langle f_{r}|+|e_{s}\rangle\langle e_{s}|\otimes|f_{s}\rangle\langle f_{s}|+|e_{s}\rangle\langle e_{r}|\otimes|f_{r}\rangle\langle f_{s}|+|e_{r}\rangle\langle e_{s}|\otimes|f_{s}\rangle\langle f_{r}|\Big)\text{.}
\end{equation}
Observe that
\begin{equation}
x_{\Psi^{+}_{r,s}}^{1\otimes t}=\frac{1}{2}\Big(|e_{r}\rangle\langle e_{r}|\otimes|f_{s}\rangle\langle f_{s}|+|e_{s}\rangle\langle e_{s}|\otimes|f_{r}\rangle\langle f_{r}|+|e_{r}\rangle\langle e_{s}|\otimes|f_{r}\rangle\langle f_{s}|+|e_{s}\rangle\langle e_{r}|\otimes|f_{s}\rangle\langle f_{r}|\Big)\text{.}
\end{equation}
We now calculate
\begin{equation}
x_{\Phi^{+}_{r,s}}^{1\otimes t}x_{\Psi^{+}_{r,s}}^{1\otimes t}=\frac{1}{4}\Big(|e_{r}\rangle\langle e_{s}|\otimes|f_{r}\rangle\langle f_{s}|+|e_{s}\rangle\langle e_{r}|\otimes|f_{s}\rangle\langle f_{r}|+|e_{s}\rangle\langle e_{r}|\otimes|f_{r}\rangle\langle f_{s}|+|e_{r}\rangle\langle e_{s}|\otimes|f_{s}\rangle\langle f_{r}|\Big).
\end{equation}
Notice that $\Big(x_{\Phi^{+}_{r,s}}^{1\otimes t}x_{\Psi^{+}_{r,s}}^{1\otimes t}\Big)^{*}=x_{\Phi^{+}_{r,s}}^{1\otimes t}x_{\Psi^{+}_{r,s}}^{1\otimes t}$. Thus, given that $(x_{\Phi^{+}_{r,s}}^{1\otimes t})^{*}=x_{\Phi^{+}_{r,s}}^{1\otimes t}$ and $(x_{\Psi^{+}_{r,s}}^{1\otimes t})^{*}=x_{\Psi^{+}_{r,s}}^{1\otimes t}$, one has that
\begin{equation}
x_{\Phi^{+}_{r,s}}^{1\otimes t}\jProd x_{\Psi^{+}_{r,s}}^{1\otimes t}=\frac{1}{4}\Big(|e_{r}\rangle\langle e_{s}|\otimes|f_{r}\rangle\langle f_{s}|+|e_{s}\rangle\langle e_{r}|\otimes|f_{s}\rangle\langle f_{r}|+|e_{s}\rangle\langle e_{r}|\otimes|f_{r}\rangle\langle f_{s}|+|e_{r}\rangle\langle e_{s}|\otimes|f_{s}\rangle\langle f_{r}|\Big).
\end{equation}
Next, we calculate
\begin{eqnarray}
&&\Big(x_{\Phi^{+}_{r,s}}^{1\otimes t}\jProd x_{\Psi^{+}_{r,s}}^{1\otimes t}\Big)\jProd\Big(x_{\Phi^{+}_{r,s}}^{1\otimes t}\jProd x_{\Psi^{+}_{r,s}}^{1\otimes t}\Big)\nonumber\\
&=&\frac{1}{16}\Big(|e_{r}\rangle\langle e_{r}|\otimes|f_{r}\rangle\langle f_{r}|+|e_{s}\rangle\langle e_{s}|\otimes|f_{s}\rangle\langle f_{s}|+|e_{s}\rangle\langle e_{s}|\otimes|f_{r}\rangle\langle f_{r}|+|e_{r}\rangle\langle e_{r}|\otimes|f_{s}\rangle\langle f_{s}|\Big).
\end{eqnarray}
Define
\begin{eqnarray}
&&y\nonumber\\
&\equiv&16\sum_{r=0}^{m-1}\sum_{s=r+1}^{m-1}\Big(x_{\Phi^{+}_{r,s}}^{1\otimes t}\jProd x_{\Psi^{+}_{r,s}}^{1\otimes t}\Big)\jProd\Big(x_{\Phi^{+}_{r,s}}^{1\otimes t}\jProd x_{\Psi^{+}_{r,s}}^{1\otimes t}\Big)\\
&=&\sum_{r=0}^{m-1}\sum_{s=r+1}^{m-1}\Big(|e_{r}\rangle\langle e_{r}|\otimes|f_{r}\rangle\langle f_{r}|+|e_{s}\rangle\langle e_{s}|\otimes|f_{s}\rangle\langle f_{s}|+|e_{s}\rangle\langle e_{s}|\otimes|f_{r}\rangle\langle f_{r}|+|e_{r}\rangle\langle e_{r}|\otimes|f_{s}\rangle\langle f_{s}|\Big)
\end{eqnarray}
Observe that
\begin{equation}
\sum_{r=0}^{m-1}\sum_{s=r+1}^{m-1}\Big(|e_{r}\rangle\langle e_{r}|\otimes|f_{r}\rangle\langle f_{r}|\Big)=\sum_{r=0}^{m-1}\Big(|e_{r}\rangle\langle e_{r}|\otimes|f_{r}\rangle\langle f_{r}|\Big)(m-1-r)\text{,}
\end{equation}
and
\begin{equation}
\sum_{r=0}^{m-1}\sum_{s=r+1}^{m-1}\Big(|e_{s}\rangle\langle e_{s}|\otimes|f_{s}\rangle\langle f_{s}|\Big)=\sum_{s=1}^{m-1}\Big(|e_{s}\rangle\langle e_{s}|\otimes|f_{s}\rangle\langle f_{s}|\Big)s=\sum_{r=0}^{m-1}\Big(|e_{r}\rangle\langle e_{r}|\otimes|f_{r}\rangle\langle f_{r}|\Big)r\text{.}
\end{equation}
Furthermore, observe that
\begin{equation}
\sum_{r=0}^{m-1}\sum_{s=r+1}^{m-1}\Big(|e_{s}\rangle\langle e_{s}|\otimes|f_{r}\rangle\langle f_{r}|+|e_{r}\rangle\langle e_{r}|\otimes|f_{s}\rangle\langle f_{s}|\Big)=\sum_{r=0}^{m-1}\sum_{s=0}^{m-1}\Big(|e_{s}\rangle\langle e_{s}|\otimes|f_{r}\rangle\langle f_{r}|\Big)(1-\delta_{rs})\text{.}
\end{equation}
Therefore,
\begin{equation}
y=\sum_{r=0}^{m-1}\Big(|e_{r}\rangle\langle e_{r}|\otimes|f_{r}\rangle\langle f_{r}|\Big)(m-1)+\sum_{r=0}^{m-1}\sum_{s=0}^{m-1}\Big(|e_{s}\rangle\langle e_{s}|\otimes|f_{r}\rangle\langle |f_{r}|\Big)(1-\delta_{rs})\text{.}
\end{equation}
By linearity of the generation of $\mathfrak{C}$, one has that
\begin{equation}
\left\{0\oplus y,y\oplus 0\right\}\subset\mathfrak{C}\text{.}
\end{equation} 
Now, $\forall r\in\{0,\dots,m-1\}$, define $z_{r}^{*}=z_{r}=z_{r}^{1\otimes t}\equiv(|e_{r}\rangle\langle e_{r}|\otimes |f_{r}\rangle\langle f_{r}|)$. Observe that
\begin{equation}
z\equiv y-\frac{m-2}{m-1}\sum_{r=0}^{m-1}z_{r}\jProd y=\sum_{r=0}^{m-1}\sum_{s=0}^{m-1}\Big(|e_{r}\rangle\langle e_{r}|\otimes |f_{s}\rangle\langle f_{s}|\Big)\text{.}
\end{equation}
Thus
\begin{equation}
\left\{0\oplus z,z\oplus 0\right\}\subset\mathfrak{C}\text{.}
\end{equation} 
In the case $m=n$, note that $z=\mathds{1}_{\mathbb{C}^{n^{2}}}$.\\[0.3cm] 
Now, suppose that $n>m$. In this case, $\forall r\in\{m,\dots,n-1\}$ and $\forall s\in\{0,\dots,m-1\}$, define
\begin{equation}
w_{r,s}^{*}=w_{r,s}=w_{r,s}^{1\otimes t}\equiv\Big(\big(|e_{r}\rangle\langle e_{0}|+|e_{0}\rangle\langle e_{r}|\big)\otimes |f_{s}\rangle\langle f_{s}|\Big)\text{.}
\end{equation}
We calculate 
\begin{equation}
w_{r,s}\jProd z=\frac{1}{2}w_{r,s}\text{,} 
\end{equation}
and
\begin{equation}
w_{r,s}\jProd w_{r,s}=\frac{1}{2}\Big(|e_{r}\rangle\langle e_{r}|\otimes |f_{s}\rangle\langle f_{s}|+|e_{0}\rangle\langle e_{0}|\otimes |f_{s}\rangle\langle f_{s}|\Big)\text{.}
\end{equation}
Note that
\begin{equation}
\Big(|e_{0}\rangle\langle e_{0}|\otimes |f_{s}\rangle\langle f_{s}|\Big)\jProd z=\Big(|e_{0}\rangle\langle e_{0}|\otimes |f_{s}\rangle\langle f_{s}|\Big)\text{.}
\end{equation}
Thus,
\begin{equation}
z+8\sum_{r=m}^{n-1}\sum_{s=0}^{m-1}\big(w_{r,s}\jProd z\big)\jProd\big(w_{r,s}\jProd z\big)-\sum_{s=0}^{m-1}\Big(\big(|e_{0}\rangle\langle e_{0}|\otimes |f_{s}\rangle\langle f_{s}|\big)\jProd z\Big)=\mathds{1}_{\mathbb{C}^{nm}}\text{.}
\end{equation}
Therefore,
\begin{equation}
\left\{0\oplus \mathds{1}_{\mathbb{C}^{nm}},\mathds{1}_{\mathbb{C}^{nm}}\oplus 0\right\}\subset\mathfrak{C}\text{.}
\end{equation} 
Recalling that $\{x\oplus x^{1\otimes t}:x\in\mathcal{M}_{nm}(\mathbb{C})_{\text{sa}}\}\subset\mathfrak{C}$, we complete the proof.
\begin{flushright}
$\square$
\end{flushright}
\newpage
\section{Univesal Tensor Product of Two Quabits}
\label{quabitApp}
\noindent In this appendix, we explictly show that
\begin{equation}
\mathcal{M}_{2}(\mathbb{H})_{\text{sa}}\tilde{\otimes}\mathcal{M}_{2}(\mathbb{H})_{\text{sa}}\cong\mathcal{M}_{16}(\mathbb{R})_{\text{sa}}\oplus\mathcal{M}_{16}(\mathbb{R})_{\text{sa}}\oplus\mathcal{M}_{16}(\mathbb{R})_{\text{sa}}\oplus\mathcal{M}_{16}(\mathbb{R})_{\text{sa}}
\end{equation}
To begin, let $x=\mathcal{M}_{n}(\mathbb{C})$ and $y\in\mathcal{M}_{m}(\mathbb{C})$. We shall adopt the following conventions (with $\mathbf{0}_{n,m}$ the $m\times n$ zero matrix):
\begin{equation}
x\oplus y=\begin{pmatrix} x & \mathbf{0}_{n,m} \\ \mathbf{0}_{m,n} & y \end{pmatrix}\in\mathcal{M}_{n+m}(\mathbb{C})\text{,}\hspace{1cm}x\otimes y=\begin{pmatrix} y_{1,1}x  & \dots & y_{1,m}x \\ \vdots  & \ddots  & \vdots \\ y_{m,1}x & \dots & y_{m,m}x\end{pmatrix}\in\mathcal{M}_{m}\Big(\mathcal{M}_{n}(\mathbb{C})\Big)\text{,}
\end{equation}
The following is a well known. We provide a proof for completeness.
\begin{proposition}\label{prop1p1} \textit{As complex *-algebras, with} $n_{1},n_{2},m_{1},m_{2}\in\mathbb{N}$,
\begin{eqnarray} &&\Big(\mathcal{M}_{n_{1}}(\mathbb{C})\oplus\mathcal{M}_{m_{1}}(\mathbb{C})\Big)\otimes\Big(\mathcal{M}_{n_{2}}(\mathbb{C})\oplus\mathcal{M}_{m_{2}}(\mathbb{C})\Big)\nonumber\\
&\cong& \Big(\mathcal{M}_{n_{1}}(\mathbb{C})\otimes\mathcal{M}_{n_{2}}(\mathbb{C})\Big)\oplus\Big(\mathcal{M}_{n_{1}}(\mathbb{C})\otimes\mathcal{M}_{m_{2}}(\mathbb{C})\Big)\oplus\Big(\mathcal{M}_{m_{1}}(\mathbb{C})\otimes\mathcal{M}_{n_{2}}(\mathbb{C})\Big)\oplus\Big(\mathcal{M}_{m_{1}}(\mathbb{C})\otimes\mathcal{M}_{m_{2}}(\mathbb{C})\Big)\text{.}
\end{eqnarray}
\begin{proof} Let $\{e_{1},\dots,e_{n_{1}+m_{1}}\}\subset\mathbb{C}^{n_{1}+m_{1}}$ be the orthonormal basis ($e_{k}^{*}e_{j}=\delta_{j,k}$) underlying $\mathcal{M}_{n_{1}}(\mathbb{C})\oplus\mathcal{M}_{m_{1}}(\mathbb{C})$; hence
\begin{equation}
\mathcal{M}_{n_{1}}(\mathbb{C})\oplus\mathcal{M}_{m_{1}}(\mathbb{C})\ni x\oplus y =\sum_{j=1}^{n_{1}}\sum_{k=1}^{n_{1}}e_{j}x_{j,k}e_{k}^{*}+\sum_{t=n_{1}+1}^{n_{1}+m_{1}}\sum_{v=n_{1}+1}^{n_{1}+m_{1}}e_{t}y_{t-n_{1},v-n_{1}}e_{v}^{*}\text{,}
\end{equation}
with $x\oplus y$ arbitrary. Let $\{f_{1},\dots,f_{n_{2}+m_{2}}\}\subset\mathbb{C}^{n_{2}+m_{2}}$ be the orthonormal basis underlying $\mathcal{M}_{n_{2}}(\mathbb{C})\oplus\mathcal{M}_{m_{2}}(\mathbb{C})$; hence
\begin{equation}
\mathcal{M}_{n_{2}}(\mathbb{C})\oplus\mathcal{M}_{m_{2}}(\mathbb{C})\ni w\oplus z =\sum_{p=1}^{n_{2}}\sum_{q=1}^{n_{2}}f_{p}w_{p,q}f_{q}^{*}+\sum_{r=n_{2}+1}^{n_{2}+m_{2}}\sum_{s=n_{2}+1}^{n_{2}+m_{2}}f_{r}z_{r-n_{2},s-n_{2}}f_{s}^{*}\text{.}
\end{equation}
with $w\oplus z$ arbitrary. Then $X\equiv (x\oplus y)\otimes (w\oplus z)\in\mathcal{M}_{(n_{1}+m_{1})(n_{2}+m_{2})=n_{1}n_{2}+n_{1}m_{2}+m_{1}n_{2}+m_{1}m_{2}}(\mathbb{C})$ is as follows
\begin{eqnarray}
X&=&\sum_{j=1}^{n_{1}}\sum_{k=1}^{n_{1}}\sum_{p=1}^{n_{2}}\sum_{q=1}^{n_{2}}(e_{j}e_{k}^{*}\otimes f_{p}f_{q}^{*})x_{j,k}w_{p,q}+\sum_{t=n_{1}+1}^{n_{1}+m_{1}}\sum_{v=n_{1}+1}^{n_{1}+m_{1}}\sum_{p=1}^{n_{2}}\sum_{q=1}^{n_{2}}(e_{t}e_{v}^{*}\otimes f_{p}f_{q}^{*})y_{t-n_{1},v-n_{1}}w_{p,q}\nonumber \\
&+&\sum_{j=1}^{n_{1}}\sum_{k=1}^{n_{1}}\sum_{p=n_{2}+1}^{n_{2}+m_{2}}\sum_{q=n_{2}+1}^{n_{2}+m_{2}}(e_{j}e_{k}^{*}\otimes f_{p}f_{q}^{*})x_{j,k}z_{p-n_{2},q-n_{2}}\nonumber\\
&+&\sum_{t=n_{1}+1}^{n_{1}+m_{1}}\sum_{v=n_{1}+1}^{n_{1}+m_{1}}\sum_{n_{2}+1}^{n_{2}+m_{2}}\sum_{q=n_{2}+1}^{n_{2}+m_{2}}(e_{t}e_{v}^{*}\otimes f_{p}f_{q}^{*})y_{t-n_{1},v-n_{1}}z_{p-n_{2},q-n_{2}} \nonumber\\
&=&\sum_{j=1}^{n_{1}}\sum_{k=1}^{n_{1}}e_{j}x_{j,k}e_{k}^{*}\otimes\left(\sum_{p=1}^{n_{2}}\sum_{q=1}^{n_{2}}f_{p}w_{p,q}f_{q}^{*}+\sum_{r=n_{2}+1}^{n_{2}+m_{2}}\sum_{q=n_{2}+1}^{n_{2}+m_{2}}f_{r}z_{r-n_{2},s-n_{2}}f_{s}^{*}\right)\nonumber\\
&+&\sum_{t=n_{1}+1}^{n_{1}+m_{1}}\sum_{v=n_{1}+1}^{n_{1}+m_{1}}e_{t}y_{t-n_{1},v-n_{1}}e_{v}^{*}\otimes\left(\sum_{p=1}^{n_{2}}\sum_{q=1}^{n_{2}}f_{p}w_{p,q}f_{q}^{*}+\sum_{r=n_{2}+1}^{n_{2}+m_{2}}\sum_{q=n_{2}+1}^{n_{2}+m_{2}}f_{r}z_{r-n_{2},s-n_{2}}f_{s}^{*}\right)
\label{expanded}
\end{eqnarray}
Where we have appealed to the bilinearity of \begin{equation}
\otimes::\mathcal{M}_{n_{1}+m_{1}}(\mathbb{C})\times\mathcal{M}_{n_{2}+m_{2}}(\mathbb{C})\longrightarrow\mathcal{M}_{n_{1}n_{2}+n_{1}m_{2}+m_{1}n_{2}+m_{2}n_{2}}(\mathbb{C}).
\end{equation}
\noindent Obviously 
\begin{equation}
\big\{e_{j}\otimes f_{p}\;\boldsymbol{|}\;j\in\{1,\dots,n_{1}+m_{1}\}\wedge\boldsymbol{|}\;p\in\{1,\dots,n_{2}+m_{2}\}\big\}\subset\mathbb{C}^{n_{1}n_{2}+n_{1}m_{2}+m_{1}n_{2}+m_{1}m_{2}}
\end{equation} 
is an orthonormal basis. We shall introduce another. Define $g_{\alpha}$ with $\alpha\in\{1,\dots,n_{1}n_{2}+n_{1}m_{2}+m_{1}n_{2}+m_{1}m_{2}\}$ via
\begin{equation}
g_{\alpha}=\begin{cases} 
e_{j}\otimes f_{p} & \text{with}\; 1\leq j \leq n_{1}\;\wedge\; 1\leq p\leq n_{2}\;\text{and}\;\alpha\equiv (p-1)n_{1}+j\\ 
e_{j}\otimes f_{n_{2}+p} & \text{with}\; 1\leq j \leq n_{1}\;\wedge\; 1\leq p\leq m_{2}\;\text{and}\;\alpha\equiv n_{1}n_{2}+(p-1)n_{1}+j\\
e_{n_{1}+j}\otimes f_{p} & \text{with}\; 1\leq j \leq m_{1}\;\wedge\; 1\leq p\leq n_{2}\;\text{and}\;\alpha\equiv n_{1}n_{2}+n_{1}m_{2}+(p-1)m_{1}+j\\
e_{n_{1}+j}\otimes f_{n_{2}+p} & \text{with}\; 1\leq j \leq m_{1}\;\wedge\; 1\leq p\leq m_{2}\;\text{and}\;\alpha\equiv n_{1}n_{2}+n_{1}m_{2}+m_{1}n_{2}+(p-1)m_{1}+j
\end{cases}
\label{newBasis}
\end{equation}
Any two orthonormal bases are related via a unitary transformation. Let $U::\{e_{j}\otimes f_{p}\}\longmapsto \{g_{l}\}$. Then from Eq.~\eqref{expanded},  
\begin{eqnarray}
UXU^{*}&=&\underbrace{\sum_{\alpha=1}^{n_{1}n_{2}}\sum_{\beta=1}^{n_{1}n_{2}}g_{\alpha}X_{\alpha\beta}g_{\beta}^{*}}_{\equiv A}\nonumber\\
&+&\underbrace{\sum_{\alpha=n_{1}n_{2}+1}^{n_{1}n_{2}+n_{1}m_{2}}\sum_{\beta=n_{1}n_{2}+1}^{n_{1}n_{2}+n_{1}m_{2}}g_{\alpha}X_{\alpha\beta}g_{\beta}^{*}}_{\equiv B}\nonumber\\
&+&\underbrace{\sum_{\alpha=n_{1}n_{2}+n_{1}m_{2}+1}^{n_{1}n_{2}+n_{1}m_{2}+m_{1}n_{2}}\sum_{\beta=n_{1}n_{2}+n_{1}m_{2}+1}^{n_{1}n_{2}+n_{1}m_{2}+m_{1}n_{2}}g_{\alpha}X_{\alpha\beta}g_{\beta}^{*}}_{\equiv C}\nonumber\\
&+&\underbrace{\sum_{\alpha=n_{1}n_{2}+n_{1}m_{2}+m_{1}n_{2}+1}^{n_{1}n_{2}+n_{1}m_{2}+m_{1}n_{2}+m_{1}m_{2}}\sum_{\alpha=n_{1}n_{2}+n_{1}m_{2}+m_{1}n_{2}+1}^{n_{1}n_{2}+n_{1}m_{2}+m_{1}n_{2}+m_{1}m_{2}}g_{\alpha}X_{\alpha\beta}g_{\beta}^{*}}_{\equiv D}\nonumber\\[0.5cm]
&=& A\oplus B\oplus C\oplus D\text{,}
\end{eqnarray}
Notice from Eq.~\eqref{newBasis} that $\alpha$ and $\beta$ are uniquely determined from $(p,q)$ according to our tensor product convention, specifically $\alpha=1\iff (p,q)=(1,1)$, $\alpha=2\iff (p,q)=(1,2)$,\dots, $\alpha=n_{2}\iff (p,q)=(1,n_{2})$, and so. Explicitly, then, we have from Eq.~\eqref{expanded}
\begin{eqnarray}
1\leq\alpha,\beta\leq n_{1}n_{2}&\implies& \big(1\leq j,k\leq n_{1}\;\wedge\;1\leq p,q\leq n_{2}\big)\nonumber\\&\implies& X_{\alpha,\beta}=x_{j,k}w_{p,q}\\[0.3cm]
n_{1}n_{2}+1\leq\alpha,\beta\leq n_{1}n_{2}+n_{1}m_{2}&\implies& \big(1\leq j,k\leq n_{1}\;\wedge\;1\leq p,q\leq m_{2}\big)\nonumber\\&\implies& X_{\alpha,\beta}=x_{j,k}z_{p,q}\\[0.3cm]
n_{1}n_{2}+n_{1}m_{2}+1\leq\alpha,\beta\leq n_{1}n_{2}+n_{1}m_{2}+m_{1}n_{2}&\implies& \big(1\leq j,k\leq m_{1}\;\wedge\;1\leq p,q\leq n_{2}\big)\nonumber\\ &\implies& X_{\alpha,\beta}=y_{j,k}w_{p,q}\\[0.3cm]
n_{1}n_{2}+n_{1}m_{2}+m_{1}n_{2}\leq\alpha,\beta\leq n_{1}n_{2}+n_{1}m_{2}+m_{1}n_{2}+m_{1}m_{2}&\implies& \big(1\leq j,k\leq m_{1}\;\wedge\;1\leq p,q\leq m_{2}\big)\nonumber\\ &\implies& X_{\alpha,\beta}=y_{j,k}w_{p,q}
\end{eqnarray}
Therefore, with respect to $\{g_{l}\}$,
\begin{eqnarray}
UXU^{*}&=&\begin{pmatrix} 
x\otimes w & \mathbf{0}_{n_{1}m_{2}} & \mathbf{0}_{m_{1}n_{2}} & \mathbf{0}_{m_{1}m_{2}}\\ 
\mathbf{0}_{n_{1}n_{2}} & x\otimes z & \mathbf{0}_{m_{1}n_{2}} & \mathbf{0}_{m_{1}m_{2}}\\
\mathbf{0}_{n_{1}n_{2}} & \mathbf{0}_{n_{1}m_{2}} & y\otimes w & \mathbf{0}_{m_{1}m_{2}}\\
\mathbf{0}_{n_{1}n_{2}} & \mathbf{0}_{n_{1}m_{2}} & \mathbf{0}_{m_{1}n_{2}} & y\otimes z
\end{pmatrix}\\[0.3cm]
&=&(x\otimes w)\oplus(x\otimes z)\oplus(y\otimes w)\oplus(y\otimes z)\text{.}
\end{eqnarray}
It remains only to note that
\begin{eqnarray}
\mathcal{U}&:\mathcal{M}_{n_{1}n_{2}+n_{1}m_{2}+m_{1}n_{2}+m_{1}m_{2}}(\mathbb{C})\longrightarrow\mathcal{M}_{n_{1}n_{2}+n_{1}m_{2}+m_{1}n_{2}+m_{1}m_{2}}(\mathbb{C})::X\longmapsto UXU^{*}
\end{eqnarray}
is clearly a complex *-algebra automorphism --- in detail, $\forall X,Y$ one has $\mathcal{U}(X)=\mathcal{U}(Y)\implies X=Y$, and $\forall Z$ $\exists W$ such that $\mathcal{U}(W)=Z$, namely $W=U^{*}ZU$; furthermore $\mathcal{U}(X)\mathcal{U}(Y)=UXU^{*}UYU^{*}=UXYU^{*}=\mathcal{U}(XU)$, and $\mathcal{U}(X)^{*}=(UXU^{*})^{*}=UX^{*}U^{*}=\mathcal{U}(X^{*})$ --- and from all of the foregoing we have shown that
\begin{eqnarray}
\mathcal{U}&::&\Big(\mathcal{M}_{n_{1}}(\mathbb{C})\oplus\mathcal{M}_{m_{1}}(\mathbb{C})\Big)\otimes\Big(\mathcal{M}_{n_{2}}(\mathbb{C})\oplus\mathcal{M}_{m_{2}}(\mathbb{C})\Big)
\nonumber\\ &\longmapsto&\Big(\mathcal{M}_{n_{1}}(\mathbb{C})\otimes\mathcal{M}_{n_{2}}(\mathbb{C})\Big)\oplus\Big(\mathcal{M}_{n_{1}}(\mathbb{C})\otimes\mathcal{M}_{m_{2}}(\mathbb{C})\Big)\oplus\Big(\mathcal{M}_{m_{1}}(\mathbb{C})\otimes\mathcal{M}_{n_{2}}(\mathbb{C})\Big)\oplus\Big(\mathcal{M}_{m_{1}}(\mathbb{C})\otimes\mathcal{M}_{m_{2}}(\mathbb{C})\Big)\nonumber\\
&&
\label{autoMorph}
\end{eqnarray}
\end{proof}
\end{proposition}
\begin{corollary}\label{cor1p2} $\Big(\mathcal{M}_{4}(\mathbb{C})\oplus\mathcal{M}_{4}(\mathbb{C})\Big)\otimes\Big(\mathcal{M}_{4}(\mathbb{C})\oplus\mathcal{M}_{4}(\mathbb{C})\Big)\cong\mathcal{M}_{16}(\mathbb{C})\oplus\mathcal{M}_{16}(\mathbb{C})\oplus\mathcal{M}_{16}(\mathbb{C})\oplus\mathcal{M}_{16}(\mathbb{C})$\textit{.}
\begin{proof} 
Immediate consequence of Proposition \hyperref[prop1p1]{1.1}.
\end{proof}
\end{corollary}
We shall need the following.
\begin{proposition} \textit{Let} $\mathcal{A}\subseteq\mathcal{M}_{n}(\mathbb{C})_{\text{sa}}$\textit{. Let} $\mathcal{U}:\mathcal{M}_{n}(\mathbb{C})\longmapsto\mathcal{M}_{n}(\mathbb{C})$ \textit{be a complex *-algebra automorphism. Let} $\mathfrak{J}(\mathcal{A})$ \textit{denote the Jordan closure of} $\mathcal{A}$ \textit{in} $\mathcal{M}_{n}(\mathbb{C})_{\text{sa}}$ \textit{under the Jordan product} $a\jProd b=(ab+ba)/2$\textit{. Then, as Jordan algebras,}
\begin{equation}
\mathfrak{J}\Big(\mathcal{U}(\mathcal{A})\Big)\cong\mathfrak{J}(\mathcal{A})\text{.}
\end{equation}
\begin{proof} $\mathcal{U}$ is a linear bijection, hence it suffices to observe that
\begin{equation}
\mathcal{U}(X)\jProd\mathcal{U}(Y)=\big(\mathcal{U}(X)\mathcal{U}(Y)+\mathcal{U}(Y)\mathcal{U}(X)\big)/2=\big(\mathcal{U}(XY)+\mathcal{U}(YX)\big)/2=\mathcal{U}\Big((XY+YX)/2\Big)=\mathcal{U}(X\jProd Y)\text{.}
\end{equation}
\end{proof}
\end{proposition}

For reference, recall the complex Pauli matrices in $\mathcal{M}_{2}(\mathbb{C})_{\text{sa}}\subset\mathcal{M}_{2}(\mathbb{C})$, these are
\begin{equation}
\sigma_{o}=\begin{pmatrix} 1 & \hspace{0.23cm}0 \\ 0 & \hspace{0.23cm}1 \end{pmatrix}=\mathds{1}_{2}\text{,}\hspace{0.4cm}\sigma_{z}=\begin{pmatrix} 1 & \hspace{0.23cm}0 \\ 0 & -1 \end{pmatrix}\text{,}\hspace{0.4cm}\sigma_{x}=\begin{pmatrix} 0 & \hspace{0.23cm}1 \\ 1 & \hspace{0.23cm}0 \end{pmatrix}\text{,}\hspace{0.4cm}\sigma_{y}=\begin{pmatrix} 0 & -i \\ i & \hspace{0.23cm}0 \end{pmatrix}\text{.}
\end{equation}
Define and observe the following (where $\mathbf{0}_{n}$ is the $n\times n$ zero matrix)
\begin{eqnarray}
s_{0}\equiv\sigma_{o}\otimes\sigma_{o}\otimes\sigma_{o}=\begin{pmatrix} \sigma_{o}\otimes\sigma_{o} & \mathbf{0}_{4} \\ \mathbf{0}_{4} & \sigma_{0}\otimes\sigma_{0}\end{pmatrix} =(\sigma_{o}\otimes\sigma_{o})\oplus(\sigma_{o}\otimes\sigma_{o})\equiv t_{0}\oplus t'_{0}\label{s0}\\
s_{1}\equiv\sigma_{z}\otimes\sigma_{o}\otimes\sigma_{o}=\begin{pmatrix} \sigma_{z}\otimes\sigma_{o} & \mathbf{0}_{4} \\ \mathbf{0}_{4} & \sigma_{z}\otimes\sigma_{o} \end{pmatrix} =(\sigma_{z}\otimes\sigma_{o})\oplus(\sigma_{z}\otimes\sigma_{o})\equiv t_{1}\oplus t'_{1}\label{s1}\\
s_{2}\equiv\sigma_{x}\otimes\sigma_{o}\otimes\sigma_{o}=\begin{pmatrix} \sigma_{x}\otimes\sigma_{o} & \mathbf{0}_{4} \\ \mathbf{0}_{4} & \sigma_{x}\otimes\sigma_{o} \end{pmatrix} =(\sigma_{x}\otimes\sigma_{o})\oplus(\sigma_{x}\otimes\sigma_{o})\equiv t_{2}\oplus t'_{2}\label{s2}\\
s_{3}\equiv\sigma_{y}\otimes\sigma_{z}\otimes\sigma_{o}=\begin{pmatrix} \sigma_{y}\otimes\sigma_{z} & \mathbf{0}_{4} \\ \mathbf{0}_{4} & \sigma_{y}\otimes\sigma_{z} \end{pmatrix} =(\sigma_{y}\otimes\sigma_{z})\oplus(\sigma_{y}\otimes\sigma_{z})\equiv t_{3}\oplus t'_{3}\label{s3}\\
s_{4}\equiv\sigma_{y}\otimes\sigma_{x}\otimes\sigma_{o}=\begin{pmatrix} \sigma_{y}\otimes\sigma_{x} & \mathbf{0}_{4} \\ \mathbf{0}_{4} & \sigma_{y}\otimes\sigma_{x}\end{pmatrix} =(\sigma_{y}\otimes\sigma_{x})\oplus(\sigma_{y}\otimes\sigma_{x})\equiv t_{4}\oplus t'_{4}\label{s4}\\
s_{5}\equiv\sigma_{y}\otimes\sigma_{y}\otimes\sigma_{z}=\begin{pmatrix} \sigma_{y}\otimes\sigma_{y} & \mathbf{0}_{4} \\ \mathbf{0}_{4} & -\sigma_{y}\otimes\sigma_{y} \end{pmatrix} =(\sigma_{y}\otimes\sigma_{y})\oplus\big(-(\sigma_{y}\otimes\sigma_{y})\big)\equiv t_{5}\oplus t'_{5}\label{s5}\text{,}
\end{eqnarray}
with $t_{0}=t'_{0},t_{1}=t'_{1},t_{2}=t'_{2},t_{3}=t'_{3},t_{4}=t'_{4}$ and $t_{5}=-t'_{5}$. Recall that $\psi_{5}:V_{5}\longrightarrow C^{*}_{u}(V_{5})::v_{r}\longmapsto s_{r}$ is the canonical embedding of $V_{5}$ into its universal C$^{*}$-algebra. Recall also that, by definition (noting linear hull $\subseteq$ Jordan hull),
\begin{equation}
V_{5}\tilde{\otimes}V_{5}=\mathfrak{J}\Big(\psi_{V_{5}}(V_{5})\circledcirc\psi_{V_{5}}(V_{5})\Big)\text{,}
\end{equation}
where $\psi_{V_{5}}(V_{5})\circledcirc\psi_{V_{5}}(V_{5})$ is just the set of pure tensors. Also, recall from \cite{Barnum2016} Appendix B (wherein our tensor product convention is different), our Corollary B.5:
\begin{equation}
\mathfrak{J}\Big(\{t_{r}\otimes t_{s}\}\Big)=\mathcal{Q}_{2}\odot\mathcal{Q}_{2}=\mathcal{M}_{16}(\mathbb{R})_{\text{sa}}.
\label{stp}
\end{equation}
We are now ready for the main result of this section.
\begin{proposition} $V_{5}\tilde{\otimes}V_{5}\cong\mathcal{M}_{16}(\mathbb{R})_{\text{sa}}\oplus\mathcal{M}_{16}(\mathbb{R})_{\text{sa}}\oplus\mathcal{M}_{16}(\mathbb{R})_{\text{sa}}\oplus\mathcal{M}_{16}(\mathbb{R})_{\text{sa}}$\textit{.} 
\begin{proof} By definition,
\begin{equation}
V_{5}\tilde{\otimes}V_{5}=\mathfrak{J}\Big(\psi_{5}(V_{5})\circledcirc\psi_{5}(V_{5})\Big)\text{.}
\end{equation}
Let $\mathcal{U}:\mathcal{M}_{64}(\mathbb{C})\longrightarrow\mathcal{M}_{64}(\mathbb{C})$ be the complex *-algebra automorphism defined in Eq.~\eqref{autoMorph}. By Corollary \hyperref[cor1p2]{1.2}, we have
\begin{equation}
V_{5}\tilde{\otimes}V_{5}\cong\mathfrak{J}\Bigg(\mathcal{U}\Big(\psi_{5}(V_{5})\circledcirc\psi_{5}(V_{5})\Big)\Bigg)
\end{equation}
Then, with $t_{p},t'_{p}$ defined for $p\in\{0,\dots,5\}$ according to Eqs.~\eqref{s0}\eqref{s1}\eqref{s2}\eqref{s3}\eqref{s4}\eqref{s5},
\begin{equation}
V_{5}\tilde{\otimes}V_{5}\cong\mathfrak{J}\Bigg\{\big(t_{p}\otimes t_{q}\big)\oplus\big(t_{p}\otimes t'_{q}\big)\oplus\big(t'_{p}\otimes t_{q}\big)\oplus\big(t'_{p}\otimes t'_{q}\big)\;\boldsymbol{\big|}\;p,q\in\{1,\dots,5\}\Bigg\}\equiv\mathcal{A}\text{,}
\label{utp}
\end{equation}
where
\begin{equation}
t_{0}=\sigma_{o}\otimes\sigma_{o}\text{,}\hspace{0.2cm}t_{1}=\sigma_{z}\otimes\sigma_{o}\text{,}\hspace{0.2cm}t_{2}=\sigma_{x}\otimes\sigma_{o}\text{,}\hspace{0.2cm}t_{3}=\sigma_{y}\otimes\sigma_{z}\text{,}\hspace{0.2cm}t_{4}=\sigma_{y}\otimes\sigma_{x}\text{,}\hspace{0.2cm}t_{5}=\sigma_{y}\otimes\sigma_{y}\text{.}
\end{equation}
Immediately from Eq.~\eqref{stp}, one has that 
\begin{equation}
\mathcal{A}\text{ is a Jordan subalgebra of } \mathcal{M}_{16}(\mathbb{R})_{\text{sa}}\oplus\mathcal{M}_{16}(\mathbb{R})_{\text{sa}}\oplus\mathcal{M}_{16}(\mathbb{R})_{\text{sa}}\oplus\mathcal{M}_{16}(\mathbb{R})_{\text{sa}}\equiv\mathcal{B}\text{.} 
\label{jSub}
\end{equation}
We will prove that $\mathcal{A}\cong\mathcal{B}$. Again in light of \eqref{stp}, with \eqref{jSub} and \eqref{utp} we see that it suffices to prove that 
\begin{equation}
\{\mathds{1}_{16}\oplus\mathbf{0}_{16}\oplus\mathbf{0}_{16}\oplus\mathbf{0}_{16},\;\mathbf{0}_{16}\oplus\mathds{1}_{16}\oplus\mathbf{0}_{16}\oplus\mathbf{0}_{16},\;\mathbf{0}_{16}\oplus\mathbf{0}_{16}\oplus\mathds{1}_{16}\oplus\mathbf{0}_{16},\;\mathbf{0}_{16}\oplus\mathbf{0}_{16}\oplus\mathbf{0}_{16}\oplus\mathds{1}_{16}\}\subset\mathcal{A}.
\label{done}
\end{equation}
We shall require the following observations.
\begin{eqnarray}
(t_{4}\otimes t_{2})\jProd (t_{5}\otimes t_{3})&=&\Big((t_{4}t_{5}\otimes t_{2}t_{3})+(t_{5}t_{4}\otimes t_{3}t_{2})\Big)/2\nonumber\\
&=&\Big((\sigma_{o}\otimes i\sigma_{z}\otimes i\sigma_{z}\otimes\sigma_{z})+(\sigma_{o}\otimes-i\sigma_{z}\otimes -i\sigma_{z}\otimes\sigma_{z})\Big)/2\nonumber\\[0.2cm]
&=&-\sigma_{o}\otimes\sigma_{z}\otimes\sigma_{z}\otimes\sigma_{z}\text{.}
\label{forA1}
\end{eqnarray}
\begin{eqnarray}
(t_{2}\otimes t_{4})\jProd (t_{3}\otimes t_{5})&=&\Big((t_{2}t_{3}\otimes t_{4}t_{5})+(t_{3}t_{2}\otimes t_{5}t_{4})\Big)/2\nonumber\\
&=&\Big((i\sigma_{z}\otimes\sigma_{z}\otimes \sigma_{o}\otimes i\sigma_{z})+(-i\sigma_{z}\otimes\sigma_{z}\otimes\sigma_{o}\otimes -i\sigma_{z})\Big)/2\nonumber\\[0.2cm]
&=&-\sigma_{z}\otimes\sigma_{z}\otimes\sigma_{o}\otimes\sigma_{z}\text{.}
\label{forB1}
\end{eqnarray}
\begin{eqnarray}
t_{0}\otimes t_{1}\jProd\big(\sigma_{o}\otimes\sigma_{z}\otimes\sigma_{z}\otimes\sigma_{z}\big)=\sigma_{o}\otimes\sigma_{z}\otimes\sigma_{o}\otimes\sigma_{z}\text{.}
\label{forA2}
\end{eqnarray}
\begin{eqnarray}
t_{1}\otimes t_{0}\jProd\big(\sigma_{z}\otimes\sigma_{z}\otimes\sigma_{o}\otimes\sigma_{z}\big)=\sigma_{o}\otimes\sigma_{z}\otimes\sigma_{o}\otimes\sigma_{z}\text{.}
\label{forB2}
\end{eqnarray}
\begin{eqnarray}
(t_{4}\otimes t_{4})\jProd(\sigma_{o}\otimes\sigma_{z}\otimes\sigma_{o}\otimes\sigma_{z})&=&\big(\sigma_{y}\otimes\sigma_{x}\otimes\sigma_{y}\otimes\sigma_{x}\big)\jProd\big(\sigma_{o}\otimes\sigma_{z}\otimes\sigma_{o}\otimes\sigma_{z}\big)\nonumber\\
&=&\Big(\big(\sigma_{y}\otimes-i\sigma_{y}\otimes\sigma_{y}\otimes-i\sigma_{y}\big)+\big(\sigma_{y}\otimes i\sigma_{y}\otimes\sigma_{y}\otimes i\sigma_{y}\big)\Big)/2\nonumber\\
&=&-t_{5}\otimes t_{5}\text{.}
\label{forA3B3}
\end{eqnarray}
Thus
\begin{eqnarray}
\mathcal{A}\ni&&\big(t_{4}\otimes t_{4}\big)^{\oplus^{4}}\jProd\Bigg(\Big(t_{0}\otimes t_{1}\Big)^{\oplus^{4}}\jProd\Big(\big(t_{4}\otimes t_{2}\big)^{\oplus^{4}}\jProd\big((t_{5}\otimes t_{3})\oplus(t_{5}\otimes t_{3})\oplus(-t_{5}\otimes t_{3})\oplus(-t_{5}\otimes t_{3})\big)\Big)\Bigg)\nonumber\\
&=&\big(t_{5}\otimes t_{5}\big)\oplus\big(t_{5}\otimes t_{5}\big)\oplus\big(-t_{5}\otimes t_{5}\big)\oplus\big(-t_{5}\otimes t_{5}\big)
\label{ppmm}
\end{eqnarray}
\begin{eqnarray}
\mathcal{A}\ni&&\big(t_{4}\otimes t_{4}\big)^{\oplus^{4}}\jProd\Bigg(\Big(t_{1}\otimes t_{0}\Big)^{\oplus^{4}}\jProd\Big(\big(t_{2}\otimes t_{4}\big)^{\oplus^{4}}\jProd\big((t_{3}\otimes t_{5})\oplus(-t_{3}\otimes t_{5})\oplus(t_{3}\otimes t_{5})\oplus(-t_{3}\otimes t_{5})\big)\Big)\Bigg)\nonumber\\
&=&\big(t_{5}\otimes t_{5}\big)\oplus\big(-t_{5}\otimes t_{5}\big)\oplus\big(t_{5}\otimes t_{5}\big)\oplus\big(-t_{5}\otimes t_{5}\big)
\label{pmpm}
\end{eqnarray}
Directly from Eq.~\eqref{utp} with $p=q=5$ we also have that
\begin{equation}
\mathcal{A}\ni\big(t_{5}\otimes t_{5}\big)\oplus\big(-t_{5}\otimes t_{5}\big)\oplus\big(-t_{5}\otimes t_{5}\big)\oplus\big(t_{5}\otimes t_{5}\big)
\label{pmmp}
\end{equation}
We will now show that $\mathcal{A}\ni(t_{5}\otimes t_{5})^{\oplus^{4}}$. Noting that $t_{1},t_{2},t_{3},t_{4}$ mutually anticommute in $\mathcal{M}_{4}(\mathbb{C})_{\text{sa}}$, observe the following.
\begin{eqnarray}
\Big(\big(t_{1}\otimes t_{1}\big)\jProd\big(t_{2}\otimes t_{2}\big)\Big)\jProd\Big(\big(t_{3}\otimes t_{3}\big)\jProd\big(t_{4}\otimes t_{4}\big)\Big)&=&\big(t_{1}t_{2}\otimes t_{1}t_{2}\big)\jProd\big(t_{3}t_{4}\otimes t_{3}t_{4}\big)\nonumber\\
&=&t_{1}t_{2}t_{3}t_{4}\otimes t_{1}t_{2}t_{3}t_{4}\nonumber\\
&=&t_{5}\otimes t_{5}\text{.}
\end{eqnarray}
Therefore
\begin{eqnarray}
&&\mathcal{A}\ni\Big(\big(t_{1}\otimes t_{1}\big)^{\oplus^{4}}\jProd\big(t_{2}\otimes t_{2}\big)^{\oplus^{4}}\Big)\jProd\Big(\big(t_{3}\otimes t_{3}\big)^{\oplus^{4}}\jProd\big(t_{4}\otimes t_{4}\big)^{\oplus^{4}}\Big)\nonumber\\
&=&(t_{5}\otimes t_{5})\oplus(t_{5}\otimes t_{5})\oplus(t_{5}\otimes t_{5})\oplus(t_{5}\otimes t_{5})\text{.}
\label{pppp}
\end{eqnarray}
Notice that \big(\eqref{ppmm}+\eqref{pmpm}+\eqref{pmmp}+\eqref{pppp}\big)/4 yeilds
\begin{equation}
\mathcal{A}\ni t_{5}\otimes t_{5}\oplus\mathbf{0}_{16}\oplus\mathbf{0}_{16}\oplus\mathbf{0}_{16}\implies \mathcal{A}\ni\mathds{1}_{16}\oplus\mathbf{0}_{16}\oplus\mathbf{0}_{16}\oplus\mathbf{0}_{16}\text{,}
\label{ozzz}
\end{equation}
where the implication follows from $t_{5}^{2}=\mathds{1}_{4}$. Next, notice that \big(\eqref{ppmm}+\eqref{pmpm}\big)/2 yields
\begin{equation}
\mathcal{A}\ni (t_{5}\otimes t_{5})\oplus\mathbf{0}_{16}\oplus\mathbf{0}_{16}\oplus(-t_{5}\otimes t_{5})\implies \mathcal{A}\ni\mathds{1}_{16}\oplus\mathbf{0}_{16}\oplus\mathbf{0}_{16}\oplus\mathds{1}_{16}\implies\mathbf{0}_{16}\oplus\mathbf{0}_{16}\oplus\mathbf{0}_{16}\oplus\mathds{1}_{16}\text{,}
\label{zzzo}
\end{equation}
where the final implication follows from Eq.~\eqref{ozzz}, \textit{i.e}.\ $\mathcal{A}\ni\big(\mathds{1}_{16}\oplus\mathbf{0}_{16}\oplus\mathbf{0}_{16}\oplus\mathds{1}_{16}\big)-\big(\mathds{1}_{16}\oplus\mathbf{0}_{16}\oplus\mathbf{0}_{16}\oplus\mathbf{0}_{16}\big)$. Furthermore, notice that \big(\eqref{pmpm}-\eqref{pmmp}\big)/2 yields
\begin{equation}
\mathcal{A}\ni\mathbf{0}_{16}\oplus\mathbf{0}_{16}\oplus(t_{5}\otimes t_{5})\oplus(-t_{5}\otimes t_{5})\implies\mathcal{A}\ni\mathbf{0}_{16}\oplus\mathbf{0}_{16}\oplus\mathds{1}_{16}\oplus\mathds{1}_{16}\implies\mathcal{A}\ni\mathbf{0}_{16}\oplus\mathbf{0}_{16}\oplus\mathds{1}_{16}\oplus\mathbf{0}_{16}\text{,}
\end{equation}
where the final implication follows from Eq~\eqref{zzzo}, \textit{i.e}.\ $\mathcal{A}\ni\big(\mathbf{0}_{16}\oplus\mathbf{0}_{16}\oplus\mathds{1}_{16}\oplus\mathds{1}_{16}\big)-\big(\mathbf{0}_{16}\oplus\mathbf{0}_{16}\oplus\mathbf{0}_{16}\oplus\mathds{1}_{16}\big)$. Finally, notice that \big(\eqref{pppp}+\eqref{ppmm}\big)/2 yields
\begin{equation}
\mathcal{A}\ni(t_{5}\otimes t_{5})\oplus(t_{5}\otimes t_{5})\oplus\mathbf{0}_{16}\oplus\mathbf{0}_{16}\implies\mathcal{A}\ni\mathds{1}_{16}\oplus\mathds{1}_{16}\oplus\mathbf{0}_{16}\oplus\mathbf{0}_{16}\implies\mathcal{A}\ni\mathbf{0}_{16}\oplus\mathds{1}_{16}\oplus\mathbf{0}_{16}\oplus\mathbf{0}_{16}\text{,}
\end{equation}
where the final implication follows from Eq.~\eqref{ozzz}, \textit{i.e}.\ $\mathcal{A}\ni\big(\mathds{1}_{16}\oplus\mathds{1}_{16}\oplus\mathbf{0}_{16}\oplus\mathbf{0}_{16}\big)-\big(\mathds{1}_{16}\oplus\mathbf{0}_{16}\oplus\mathbf{0}_{16}\oplus\mathbf{0}_{16}\big)$. Taking stock, \eqref{done} is true.
\end{proof}
\end{proposition}
\newpage
\section{Supporting Details for Chapter 8} 
\label{appendix: extending}
\noindent This appendix collects technical details supporting our proofs in \cref{compositesEJA}.

\noindent Let $\mathcal{A}$ be an EJC-algebra, that is, $\mathcal{A}$ a Jordan subagebra of $\M_{\text{sa}}$ for a finite-dimensional complex $\ast$-algebra $\M = \M_{\mathcal{A}}$. Recall that $G(\mathcal{A})$ is the connected component of the identity in the group of order-automorphisms of $\mathcal{A}$.

\begin{lemma} \label{lemma: derivations extend} Let $\mathcal{A}$ be reversible. 
Then any one-parameter group of order automorphisms of $\mathcal{A}$ extends to a
one-parameter group of order-automorphisms of $\M$.  
\begin{proof}
If $\{\phi(t)\}_{t \in \R}$ is a one-parameter
group of order-automorphisms of $\mathcal{A}$, then $\phi(t)(a) = e^{tX}a$ where
$X = \phi'(0)$ is a linear operator on $\mathcal{A}$. By definition, $X$ is an
order-derivation of $\mathcal{A}$. Now, order-derivations come in two basic types:
those having the form $X
= L_a$ for some $a \in \mathcal{A}$, 
and those having the property $Xu =
0$ (which turns out to be the same thing as being a Jordan derivation). 
The former are self-adjoint with respect to the canonical inner product on $\mathcal{A}$,  
by the definition of a Euclidean Jordan algebra, while the latter are skew-adjoint 
\ffootnote{\red For AS, that self-adjoint means $X = L_a$ for some $a \in \mathcal{A}$
  is a definition; in FK, we learn that this is equivalent to
  self-adjointness with respect to the inner product on $\mathcal{A}$.  Guess:
  ditto for skew-adjoint? Check AS, FK, and or Satake.  From AS, any 
derivation is the sum of a skew and an sa one, so we likely just need to 
find in FK or Satake where the selfadjoint ones are symmetric and the
skew ones antisymmetric}.  Moreover,
every order-derivation has the form $\delta = L_{a} + \delta'$ where
$\delta'$ is skew (\cite{AS}, Proposition 1.60).

\noindent Now, $L_a$ obviously extends from $\mathcal{A}$ to $\M$, simply because $a \in \M$
and the Jordan product on $\mathcal{A}$ is the restriction of that on $\M$. By
(\cite{Upmeier}, Theorem 2.5),  {\red if $\mathcal{A}$ is reversible} 
\ffootnote{\red But it looks,on reading his Example 2.3, as though
  finite-dimensional spin factors also have the extension
  property. Can we find a better (f.d.) reference?} in $\M$, $\delta'$
also extends to a Jordan derivation $\delta''$ on $\M$. Thus, we have
an extension $\hat{\delta} = L_a + \delta''$ on $\M$. In particular,
$\delta''(\mathcal{A}) \subseteq \mathcal{A}$.

\noindent It follows that we have an order-automorphism $\hat{\phi}(t) = e^{t \hat{\delta}}$ of $\M_+$. Note that this preserves $\mathcal{A}$, since 
\[\hat{\phi}(t)x = \sum_{k=1}^{\infty} \frac{t^k}{k!} \hat{\delta}^{k} x\]
and $\hat{\delta}x = (L_a + \delta'')x = a\dot x + \delta''(x)$, 
which belongs to $\mathcal{A}$ if $x$ does.
\end{proof}
\end{lemma}

\begin{corollary}\label{cor: automorphisms extend} 
\textit{If $\mathcal{A}$ is reversible, every element of $G(\mathcal{A})$ extends to an element of
$G((\M_{\mathcal{A}})_{\text{sa}})$.} 
\end{corollary}

\noindent Now let $\mathcal{B}$ also be a reversible EJC-algebra. 

\begin{lemma}\label{lemma: fixing} \textit{Let $\delta$ be any $\ast$-derivation of $\M_{\mathcal{A}}$ fixing $\mathcal{A}$. Then $\delta \otimes \mathbf{1}$ is a $\ast$-derivation of $\M_{\mathcal{A}} \otimes \M_{\mathcal{A}}$ fixing $\mathcal{A} \odot \mathcal{B}$.} 
\begin{proof} Let $\M$ and $\N$ be $\ast$-algebras, and let $a,
b \in \M$ and $x,y \in \N$. Then
\[(a \otimes x) \jProd (b \otimes y) = \frac{1}{2} ( ab \otimes xy + ba \otimes yx).\]
If $\delta$ is a $\ast$-derivation of $\M$, then it is straighforward to check that $\delta \otimes \mathbf{1}$ is a $\ast$-derivation of $\M \otimes \N$, 
and that for all $a, b \in \M$ and $x, y \in \N$,  
\[(\delta \otimes \mathbf{1})((a \otimes x) \jProd (b \otimes y)) = (a \otimes x) \jProd (\delta(b) \otimes y) + (\delta(a) \otimes x) \jProd (b \otimes y).\]
In particular, if $\mathcal{A} \subseteq \M$ and $\delta(\mathcal{A}) \subseteq \mathcal{A}$, it follows that $(\delta \otimes \mathbf{1})(\mathcal{A} \otimes \mathcal{A}) \subseteq (\mathcal{A} \otimes \mathcal{A}) \jProd (\mathcal{A} \otimes \mathcal{A})$ 
for any $B \subseteq \N$. It follows easily that, 
{\redd where $\mathcal{A}$ and $\mathcal{B}$ are EJCs and $\M = \M_{\mathcal{A}}$ and $\N = \M_{\mathcal{B}}$,} $\delta \otimes \mathbf{1}$ preserves $\mathcal{A} \odot \mathcal{B}$.
\footnote{The details: let $\delta$ be a $\ast$-derivation on a $\ast$-algebra $\M$, and let $X \subseteq M_{\text{sa}}$ with $\delta(X) \subseteq X$. Let $Y = \{ a \in \mathfrak{j}(X) | \delta(a) \in \mathfrak{j}(X)\}$. Evidently $X \subseteq Y$. Now if $a, b \in Y$ and $t \in \R$, then $\delta(ta + b) = t\delta(a) + \delta(b) \in \mathfrak{j}(X)$, so $\mathfrak{j}(X)$ is a subspace of $M$. If $a, b \in Y$ then $\delta(a \dot b) = a \dot \delta(b) + \delta(a) \dot b \in \mathfrak{j}(X)$. Thus, $Y$ is a Jordan subalgebra of $\M_{\text{sa}}$, containing 
$X$, and contained in $\mathfrak{j}(X)$. Ergo, $Y = \mathfrak{j}(X)$, and $\delta(\mathfrak{j}(X)) \subseteq \mathfrak{j}(X)$.}
\end{proof}
\end{lemma} 

\begin{proposition} \textit{If $\phi$ and $\psi$ are order-automorphisms of $\mathcal{A}$ and $\mathcal{B}$, respectively, then $\phi \otimes \psi$ extends to an order-automorphism of $\mathcal{A} \odot \mathcal{B}$.}
\begin{proof} By \cref{cor: automorphisms extend}, we can assume that $\phi$ is an order-automorphism of $\M_{\text{sa}}$ fixing
$\mathcal{A}$. Then $\phi = e^{\delta}$ for an order-derivation $\delta$ of $\M$,
which (taking derivatives) must also fix $\mathcal{A}$.  It follows that
\[\phi \otimes \mathbf{1} = e^{t\delta} \otimes \mathbf{1}  = \sum_{n=0}^{\infty} \frac{t^n}{n!} \delta^n \otimes \mathbf{1} = \sum_{n=0}^{\infty} \frac{t^{n}}{n!} (\delta \otimes \mathbf{1})^n = e^{t(\delta \otimes \mathbf{1})}.\]
By \cref{lemma: fixing}, $\delta \otimes \mathbf{1}$ fixes $\mathcal{A} \odot \mathcal{B}$; thus, so does the
series at right, whence, so does $\phi \otimes \mathbf{1}$.  It follows that
if $\phi$ is an order-automorphism of $\M_{\mathcal{A}}$ fixing $\mathcal{A}$, then $\phi
\otimes \mathbf{1}$ is an order-automorphism of $\mathcal{A} \otimes \mathcal{B}$ fixing $\mathfrak{j}(\mathcal{A}
\otimes \mathcal{B}) = \mathcal{A} \odot \mathcal{B}$. Hence, if $\phi$ and $\psi$ are
order-automorphisms of $\M_{\mathcal{A}}$ and $\M_{\mathcal{B}}$, respectively fixing $\mathcal{A}$ and
$\mathcal{B}$, then $\phi \otimes \psi = (\phi \otimes \mathbf{1}) \circ (\mathbf{1} \otimes
\psi)$ fixes $\mathcal{A} \odot \mathcal{B}$.
\end{proof}
\end{proposition} 

\noindent We now collect some basic facts about direct sums of \textsc{eja}s. We shall omit calligraphic mathematical fonts.

\begin{definition} \textit{The} direct sum \textit{of EJAs $A$ and $B$ is $A \oplus B := A \times B$, equipped with the slotwise operations, 
so that the canonical projections $\pi_1 : A \times B \rightarrow A$ and $\pi_2 : A \times B \rightarrow B$ 
are unital Jordan homomorphisms.}
\end{definition}

\noindent Identifying  $A$ and $B$ with $A \times \{0\}$ and $\{0\} \times B$, respectively, we  write 
$a + b$ for $(a,0) + (0,b)$.  Note that $A$ and $B$ are then ideals in $A \oplus B$, and 
that $B = A^{\perp} := \{ z \in A \oplus B | \langle a, z \rangle = 0 \ \forall a \in A\}$.  Conversely, 
it's easy to check that if $E$ is an EJA and $A$ is an ideal in $A$, then $A^{\perp}$ is also an ideal, and $ab = 0$ for 
all $a \in A, b \in A^{\perp}$; hence, $E \simeq A \oplus A^{\perp}$. 

\noindent Suppose $E$ is an EJA and $A \leq E$ is an ideal: let $B = A^{\perp}$. Then for all $z \in E$, 
\[\langle a, zb \rangle = \langle az, b \rangle = 0\]
since $az \in A$. Thus, $B$ is also an ideal, and $E = A \oplus B$ as a vector space. Finally, 
if $a \in A$ and $b \in B$, then $ab \in A \cap B = \{0\}$. Hence, 
if $a, x \in A$ and $b, y \in B$ then $(a + b)(x + y) = ax + by$, i.e., in the 
representation of $A \oplus B$ as $A \times B$, operations are slotwise.  What is not entirely obvious is 
that $A$ contains a unit element.  

\begin{lemma} Let $A$ be an ideal in an EJA $E$. Then there exists a projection $p \in E$ 
such that $pa = a$ for every $a \in A$. Thus, $A = pA$, and $E = pA \oplus p'A$, where $p' = 1 - p$. \end{lemma}

\noindent For a proof in the more general setting of JBW algebras, see  \cite{AS}, Propositions 2.7, 2.39 
and 2.41. 

\noindent The {\em center} of an EJA $E$ is the set of elements operator-commuting with all other elements. Denote this by $C(E)$
If $E = A \oplus B$, and $p$ is the unit of $A$, so that $A = pA$, then it's easy to check 
that $p  \in C(E)$.  Conversely, if $p$ is a central projection, then $pA$ is an ideal, with unit element $p$. 
If $p$ is a {\em minimal} central projection, then $pA$ is a minimal direct summand of $E$.  If $E$ is simple, 
then its only central projections are $0$ and $1$, and conversely. 


\noindent One can show that for every projection $e$ in an EJA $E$, there exists a unique miniml projection $c(e) \in C(E)$, 
the {\em central cover} of $p$, such that $e \leq c(e)$. Then $A := c(p)E$ is an ideal of $E$, in which 
$c(p)$ is the unit. If $A$ is a minimal ideal, then  
elements of $A$ are exactly those with central cover $c(p)$. [\cite{AS}, 2.37, 2.39]. More generally, 
two elements of $E$ having the same central cover have nonzero components in exactly the same ideals of $E$.


\noindent Recall that a {\em symmetry} of $A$ is an element $s \in A$ with $s^2 = u$. 
Projections $e, f$ in $A$ are {\em exchanged by a symmetry} $s$ iff $U_{s}(e) = f$. If there 
exists a sequence of symmetries $s_1,...,s_n$ with $U_{s_n} \circ \cdots \circ U_{s_1} (e) = f$, then 
$e$ and $f$ are {\em equivalent}. The following is Lemma 3.9 from \cite{AS}. 

\begin{lemma}[\cite{AS}, Lemma 3.9]\label{AS39}
Equivalent projections have the same central cover. 
\end{lemma}

\noindent Recall that a sequence of vector spaces and linear maps, or of Jordan algebras and Jordan homomorphisms, or of 
$C^{\ast}$ algebras and $\ast$-homomorphisms 
\[ A \stackrel{\alpha}{\longrightarrow} B \stackrel{\beta}{\longrightarrow} C\]
is said to be {\em exact at $B$} iff the image of $\alpha$ is the kernel of $\beta$. A {\em short exact sequence} is 
one of the form 
\[0 \longrightarrow A \stackrel{\alpha}{\longrightarrow} B \stackrel{\beta}{\longrightarrow} C \longrightarrow 0\]
that is exact at $A$, $B$ and $C$ (with the maps on the ends being the only possible ones). This means that 
$\alpha$ is injective (its kernel is $0$), while $\beta$ is surjective (its image is the kernel of the zero map, 
i.e., all of $C$). 

\noindent Let $\EJA$ and $\Cstar$ be the categories of EJAs and Jordan homomorphisms, and of $C^{\ast}$-algebras and 
$\ast$-homomorphisms, respectively.

\begin{theorem}[\cite{HO}] 
$A \mapsto \Cu(A)$ is an exact functor from $\EJA$ to $\Cstar$. 
In other words, if $A \stackrel{\alpha}{\longrightarrow} B \stackrel{\beta}{\longrightarrow} C$ is an exact sequence in $\EJA$, then $\Cu(A) \stackrel{\Cu(\alpha)}{\longrightarrow} B \stackrel{\Cu(\beta)}{\longrightarrow} \Cu(C)$ is an exact sequence in $\Cstar$. 
\end{theorem}

\noindent We are going to use this to show that $\Cu(A \oplus B) = \Cu(A) \oplus \Cu(B)$. We need some preliminaries. 
The following is standard: 

\begin{lemma}\label{exact sequences} Let 
\[0 \longrightarrow A \stackrel{\alpha}{\longrightarrow} C \stackrel{\beta}{\longrightarrow} B \longrightarrow 0\]
be a short exact sequence of vector spaces. Then the following are equivalent:
\begin{itemize} 
\item[(a)] There is an isomorphism $\phi : A \oplus B \simeq C$ such that $\alpha$ and $\beta$ are respectively the canonical injection and surjection given by  
\[\alpha(a) = \phi(a,0) \ \ \mbox{and} \ \ \beta(\phi(a,b)) = b\]
\item[(b)] The sequence is {\em split} at $B$: there exists a linear mapping $\gamma : B \rightarrow C$ such 
that $\beta \circ \gamma = \id_{B}$.
\end{itemize}
\end{lemma}

\noindent The idea is that, given $\phi$, we can define $\gamma$ by $\gamma(b) = \phi(0,b)$ and, given $\gamma$, we can 
define $\phi$ by $\phi(a,b) = \alpha(a) + \gamma(b)$. 

\tempout{
\noindent{\em Proof:} Given (a), we can define $\gamma(b) = \phi(0,b)$: then $\beta(\gamma(b)) = \beta(\phi(0,b)) = b$, so 
the sequence is split at $B$. Conversely, suppose $\gamma : B \rightarrow C$ is a right inverse for $\beta$. 
Then define $\phi : A \oplus B \rightarrow C$ by $\phi(a,b) = \alpha(a) + \gamma(b)$. Since $\gamma$ is linaer, 
we have $\phi(a,0) = \alpha(a)$. Also 
\[\beta(\phi(a,b)) = \beta(\alpha(a)) + \beta(\gamma(b)) = b\]
since $\alpha(a)$ belongs to $\ker(\beta)$ and $\beta \circ \gamma = \id_{B}$. It remains to show that $\phi$ is 
an isomorphism. If $\phi(a,b) = 0$, then $\alpha(a) + \gamma(b) = 0$, i.e., 
$\alpha(a) = -\gamma(b)$. It follows that $\gamma(b) \in \ker(\beta)$, whence, $b = \beta(\gamma(b)) = 0$. 
Hence, $\alpha(a) = 0$. But $\alpha$ is injective, so $a = 0$. This shows that $\phi$ is injective. To see it's surjective, let $c \in C$, and let $c_2 = \gamma(\beta(c))$: then $\beta(c_2) = \beta(c)$. Setting $c_1 = c = c_2$, 
we have $\beta(c_2) = 0$: by exactness, $c_2$ belongs to the image of $\alpha$, i.e., $c_2 = \alpha(a)$ for 
some $a \in A$. Taking $b = \beta(c)$, we have $\phi(a,b) = \alpha(a) + \gamma(b) = c_1 + c_2 = c$. $\Box$ }

\noindent If $A$, $B$ and $C$ are Jordan algebras or $C^{\ast}$ algebras, the implication from (a) to (b) is obviously valid, 
but the converse requires additional assumptions.

\begin{lemma}\label{exact sequences and ideals} 
Let \[0 \longrightarrow A \stackrel{\alpha}{\longrightarrow} C \stackrel{\beta}{\longrightarrow} B \longrightarrow 0.\]
be a short exact sequence of $\ast$-algebras and $\ast$-homomorphisms, split by a $\ast$-homomorphism 
$\gamma :  B \rightarrow C$ with $\beta \circ \gamma = \id_{B}$. Let $\phi : A \oplus B \rightarrow C$ be 
as defined above, {\redd i.e, $\phi(a,b) = \alpha(a) + \gamma(b)$ for $a \in A$, $b \in B$.}  Then the following are equivalent: 
\begin{itemize} 
\item[(a)] $\gamma(B)$ is a (2-sided) $\ast$-ideal in $C$;
\item[(b)] $\phi$ is multiplicative, and thus a $\ast$-isomorphism;  
\item[(c)] There exists a $\ast$-homomorphism $\delta : C \rightarrow A$ with 
\[0 \longleftarrow A \stackrel{\delta}{\longleftarrow} C \stackrel{\gamma}{\longleftarrow} B \longleftarrow 0\]
exact.
\end{itemize} 
\begin{proof} (a) $\Rightarrow$ (b). It is easy to see that a $C^{\ast}$-algebra $C$ is the direct sum of two $\ast$-ideals $A, B \leq C$ iff $A + B = C$ and $A \cap C = \{0\}$, i.e., iff $C$ is the vector-space 
direct sum of $A$ and $B$.  We already know that $\alpha(A) + \beta(B) = C$ (since $\phi$ is a linear 
isomorphism). It therefore suffices to show that $\alpha(A)$ and $\gamma(B)$ are $\ast$-ideals with 
zero intersection. We are assuming that $\gamma(B)$ is a $\ast$ ideal. As it's the kernel of a $\ast$-homomorphism, 
$\alpha(A)$ is automatically a $\ast$-ideal. To see that $\alpha(A) \cap \gamma(C) = \{0\}$, 
let $c \in C$ with $c = \alpha(a) = \gamma(b)$ for some $a \in A$ and $b \in B$. Then we have 
\[b = \beta(\gamma(b)) = \beta(\alpha(a)) = 0\]
whence, $c = \gamma(0) = 0$. 


(b) $\Rightarrow$ (c). If $\phi$ is a $\ast$-isomorphism, then let $\delta = \pi_{A} \circ \phi^{-1}$ where 
$\pi_{A} : A \oplus B \rightarrow A$ is the projection $\pi_{A}(a,b) = a$. Note that $\delta$ is the composition 
of two $\ast$-homomorphisms, and thus, a $\ast$-homomorphism. To verify exactness, 
note that as $\phi(a,b) = \alpha(a) + \gamma(b)$, 
we have $\phi(0,b) = \gamma(b)$, whence, $\phi^{-1}(\gamma(b)) = (0,b)$. Thus, 
$\delta(\gamma(b)) {\redd = \pi_{A}(0,b)} = 0$. 

(c) $\Rightarrow$ (a). By exactness, $\gamma(C)$ is the kernel of the $\ast$-homomorphism $\delta$, and hence, 
a $\ast$-ideal.
\end{proof}
\end{lemma}

\noindent Now let $E = A \oplus B$. Then we have a short exact sequence 
\[0 \longrightarrow A \stackrel{j}{\longrightarrow} A \oplus B \stackrel{p}{\longrightarrow} B \longrightarrow 0.\]
where $j(a) = (a,0)$ and $p(a,b) = b$. This is split by the homomorphism $k : A \rightarrow A \oplus B$ given 
by $k(b) = (0,b)$. Hanche-Olsen's exactness theorem (2.1) 
gives us a short exact sequence 
\[0 \longrightarrow \Cu(A) \stackrel{C^{\ast}(j)}{\longrightarrow} \Cu(A \oplus B) \stackrel{C^{\ast}(p)}{\longrightarrow} \Cu(B) \longrightarrow 0.\]
By functoriality, $\Cu(p) \circ \Cu(j) = \id_{\Cu(B)}$, so this is again split. Thus, {\em regarded as a vector space}, $\Cu(A \oplus B)$ is canonically isomorphic to 
$\Cu(A) \oplus \Cu(B)$. On the other hand, we {\em also} have an exact sequence 
\[0 \longleftarrow A \stackrel{q}{\longleftarrow} A \oplus B \stackrel{k}{\longleftarrow} B \longleftarrow 0\]
where $q(a,b) = a$. By the same argument, then, we have a short exact sequence 
\[0 \longleftarrow \Cu(A) \stackrel{\Cu(q)}{\longleftarrow} \Cu(A \oplus B) \stackrel{\Cu(k)}{\longleftarrow} \Cu(B) \longleftarrow 0.\]
Applying the preceding Lemma, we have

\begin{proposition} \textit{If $A$ and $B$ are EJAs, then  $\Cu(A \oplus B) \simeq \Cu(A) \oplus \Cu(B)$.}
\end{proposition} 

\noindent Notice that if $\Phi_A$ and $\Phi_B$ are, respectively, the canonical involutions on $\Cu(A)$ and $\Cu(B)$ fixing points of $A$ and $B$, then $\Phi_{A} \oplus \Phi_{B}$ is a $\ast$-involution on $\Cu(A) \oplus \Cu(B)$ fixing points of 
$A \oplus B$. But there is only one such  $\ast$-involution on $\Cu(A \oplus B)$, the canonical one. In other words, 
in identifying $\Cu(A \oplus B)$ with $\Cu(A) \oplus \Cu(B)$, we also identify $\Phi_{A \oplus B}$ with $\Phi_{A} \oplus \Phi_{B}$. 

\noindent Recalling now the fact (\cite{HO}, Lemma 4.2) that an EJA
$A$ is universally reversible (UR) iff $A$ coincides with the set of
self-adjoint fixed points in $\Cu(A)$ of the canonical
$\ast$-involution $\Phi_A$, we have the

\begin{corollary} 
\textit{If $A$ and $B$ are UR, then so is $A \oplus B$.}
\begin{proof} Let $\Phi = \Phi_{A} \oplus \Phi_{B}$ be the canonical involution on $\Cu(A \oplus B) = \Cu(A) \oplus \Cu(B)$. 
For $(a,b) \in \Cu(A) \oplus \Cu(B)$, we have 
 $\Phi(a,b) = (\Phi_{A}(a), \Phi_{B}(b)) = (a,b)$ iff $\Phi_{A}(a) = a$ and $\Phi_{B}(b) = b$. Since $A$ and $B$ 
are UR, this holds iff $a \in A$ and $b \in B$, i.e., iff $(a,b) \in A \oplus B$. Thus, $A \oplus B$ is exactly the 
set of fixed-points of $\Phi$, and so, is UR.
\end{proof}
\end{corollary}

\noindent In \ref{prop: canonical product composite}, we show that $A \hotimes B$ is a composite
in the sense of Definition \ref{def: new Jordan composite}.  We will show that 
any such composite $\mathcal{AB}$ 
is a direct summand of $A \hotimes B$. We rely heavily on the fact (\cref{lemma: product projections}) that $p\otimes q$ is a projection in $\mathcal{AB}$ when $p$ and $q$ are projections in $\mathcal{A}$ and $\mathcal{B}$, respectively. In order to prove \cref{lemma: product projections}, we shall require some significant preparatory work. To centre ourselves, recall that, by definition, any \textsc{eja} $\mathcal{A}$ is equipped with a self-dualizing inner product $\langle\cdot|\cdot\rangle:\mathcal{A}\times\mathcal{A}\longrightarrow\mathbb{R}$ such that $\forall a,b,c\in\mathcal{A}$ we have that
\begin{equation}
\langle a|b\jProd c\rangle=\langle a\jProd b|c\rangle\text{.}
\end{equation}
Recall that a \textit{Jordan frame} in $\mathcal{A}$ is a set $\{x_{1},\dots,x_{n}\}$ of primitive orthogonal projections with $x_{1}+x_{2}+\dots+x_{n}=u_{\mathcal{A}}$. All Jordan frames have the same size, namely $n$, which is the \textit{rank} of $\mathcal{A}$. 
We shall now recall the \textit{spectral theorem} from \cite{FK}.
\begin{theorem}\label{specThm} [Theorem III.1.2 \cite{FK}] \textit{Let} $\mathcal{A}$ \textit{be an} \textsc{eja}\textit{ of rank} $n$\textit{. Let} $a\in\mathcal{A}$\textit{. Then there exists unique real numbers} $\lambda_{1},\dots,\lambda_{n}$ \textit{and a Jordan frame} $\{x_{1},\dots,x_{n}\}$ \textit{such that }
\begin{equation}
a=x_{1}\lambda_{1}+x_{2}\lambda_{2}+\dots+x_{n}\lambda_{n}\text{.}
\label{specRes}
\end{equation}
\textit{The real numbers} $\lambda_{j}$ \textit{are referred to as the} eigenvalues \textit{of} $a$\textit{, and the} \textsc{rhs} \textit{of Eq.}~\eqref{specRes}\textit{ is called a} spectral resolution \textit{of} $a$\textit{. The} spectral radius \textit{of} $a$ \textit{is denoted} $\rho(a)$ \textit{and defined via} $\rho(a)=\text{max}\{|\lambda_{1}|,\dots,|\lambda_{n}|\}$\textit{.}
\end{theorem}
\noindent It will be useful for us to record some easy consequences of the spectral theorem.
\begin{proposition}\label{prop1ppp} \textit{Let} $\mathcal{A}$ \textit{be an} \textsc{eja}\textit{. Let} $a\in\mathcal{A}_{+}$\textit{. Then the eigenvalues of $a$ are nonnegative.}
\begin{proof} Let
\begin{equation}
a=x_{1}\lambda_{1}+x_{2}\lambda_{2}+\dots+x_{n}\lambda_{n}
\end{equation}
be a spectral resolution of $a$. $\mathcal{A}_{+}$ is self-dual, so $\forall j\in\{1,\dots,n\}$ we have $\langle x_{j}|a\rangle\geq 0$. Therefore
\begin{equation}
0\leq \langle x_{j}|a\rangle=\sum_{k=1}^{n}\lambda_{k}\langle x_{j}|x_{k}\rangle=\sum_{k=1}^{n}\lambda_{k}\langle u_{\mathcal{A}}\jProd x_{j}|x_{k}\rangle=\sum_{k=1}^{n}\lambda_{k}\langle u_{\mathcal{A}}|x_{j}\jProd x_{k}\rangle=\sum_{k=1}^{n}\lambda_{k}\langle x_{k}|x_{k}\rangle\delta_{j,k}=\lambda_{j}\langle x_{j}|x_{j}\rangle\text{.}
\label{series}
\end{equation}
Now, $\langle\cdot|\cdot\rangle$ is an inner product, and therefore positive definite, so from Eq.~\eqref{series} we have that $\lambda_{j}\geq 0$.
\end{proof}
\end{proposition}
\begin{proposition}\label{prop2ppp}\textit{Let} $\mathcal{A}$ \textit{be an} \textsc{eja}\textit{. Let} $a\jProd a=a\in\mathcal{A}$\textit{. Then the eigenvalues of} $a$ \textit{are zero or unity.}
\begin{proof} Let $a=x_{1}\lambda_{1}+\dots+x_{n}\lambda_{n}$ be a spectral decomposition of $a$. Then $a\jProd a=a$ means
\begin{equation}
\sum_{j=1}^{n}x_{j}\lambda_{j}=\sum_{j=1}^{n}x_{j}\lambda_{j}^{2}\text{,}
\label{forProjs}
\end{equation}
because $\forall j,k\in\{1,\dots,n\}$ the Jordan products of the Jordan frame elements are $x_{j}\jProd x_{k}=\delta_{j,k}x_{j}$. In fact, on that observation left Jordan multiply each side of Eq.~\eqref{forProjs} by $x_{k}$ for arbitrary $k$ to get $x_{k}\lambda_{k}=x_{k}\lambda_{k}^{2}$. So $\lambda_{k}\in\{0,1\}$.
\end{proof}
\end{proposition}
\begin{proposition}\label{prop3ppp}\textit{Let} $\mathcal{A}$ \textit{be an} \textsc{eja}\textit{. Let} $a\in\mathcal{A}_{+}$ \textit{such that} $\langle u_{\mathcal{A}}|a\rangle=\langle u_{\mathcal{A}}|a\jProd a\rangle$ \textit{and} $\rho(a)\leq 1$\textit{. Then} $a\jProd a=a$\textit{.}
\begin{proof} Apply \cref{specThm} to write $a$ in terms of a Jordan frame $\{x_{1},\dots,x_{n}\}\subset\mathcal{A}_{+}$ and $\{\lambda_{1},\dots,\lambda_{n}\}\subset\mathbb{R}_{\geq 0}^{n}$,
\begin{equation}
a=\sum_{j=1}^{n}x_{j}\lambda_{j}\text{.}
\end{equation}
$\{x_{1},\dots,x_{n}\}$ is a Jordan frame, so $x_{1}+\dots x_{n}=u_{\mathcal{A}}$. Therefore (including intermediate steps just to be careful)
\begin{equation}
\langle u_{\mathcal{A}}|a\rangle=\left\langle\sum_{k=1}^{n}x_{k}\Bigg|\sum_{j=1}^{n}x_{j}\lambda_{j}\right\rangle=\sum_{k=1}^{n}\sum_{j=1}^{n}\lambda_{j}\langle x_{k}|x_{j}\rangle=\sum_{k=1}^{n}\sum_{j=1}^{n}\lambda_{j}\langle x_{j}|x_{j}\rangle\delta_{k,j}=\sum_{j=1}^{n}\lambda_{j}\langle x_{j}|x_{j}\rangle\text{.}
\label{singleTrace}
\end{equation}
Now, $\mathcal{A}$ is an \textsc{eja}, so $\langle u_{\mathcal{A}}|a\jProd a\rangle=\langle u_{\mathcal{A}}\jProd a|a\rangle=\langle a|a\rangle$, and so we compute
\begin{equation}
\langle u_{\mathcal{A}}|a\jProd a\rangle=\langle a|a\rangle=\left\langle\sum_{k=1}^{n}x_{k}\lambda_{k}\Bigg|\sum_{j=1}^{n}x_{j}\lambda_{j}\right\rangle=\sum_{k=1}^{n}\sum_{j=1}^{n}\lambda_{k}\lambda_{j}\langle x_{k}|x_{j}\rangle=\sum_{k=1}^{n}\sum_{j=1}^{n}\lambda_{k}\lambda_{j}\langle x_{j}|x_{j}\rangle\delta_{k,j}=\sum_{j=1}^{n}\lambda_{j}^{2}\langle x_{j}|x_{j}\rangle\text{.} 
\end{equation}
The assumption $\langle u_{\mathcal{A}}|a\rangle=\langle u_{\mathcal{A}}|a\jProd a\rangle$ in the statement of the proposition therefore yields
\begin{eqnarray}
\sum_{j=1}^{n}\lambda_{j}\langle x_{j}|x_{j}\rangle=\sum_{j=1}^{n}\lambda_{j}^{2}\langle x_{j}|x_{j}\rangle\text{.}
\label{traces}
\end{eqnarray}
We can rewrite Eq.~\eqref{traces} as follows:
\begin{equation}
\sum_{j=1}^{n}\lambda_{j}(1-\lambda_{j})\langle x_{j}|x_{j}\rangle=0\text{.}
\end{equation} 
Let us consider an arbitrary $j\in\{1,\dots,n\}$. We now see that the assumption $\rho(a)\leq 1$ in the statement of the proposition means $\lambda_{j}(1-\lambda_{j})\langle x_{j}|x_{j}\rangle\geq 0$, since $\langle\cdot|\cdot\rangle$ is an inner product and therefore positive definite, and since $a\in\mathcal{A}_{+}$ implies $\lambda_{j}\geq 0$ by \cref{prop1ppp}. Eq.~\eqref{traces} therefore yields $\lambda_{j}\in\{0,1\}$, which is to say that $a$ is a projection in light of \cref{prop2ppp}.
\end{proof}
\end{proposition}
\noindent We are now ready for the following lemma.
\begin{lemma}\label{lemma: product projections}
\textit{Let $p \in A$ and $q \in B$ be projections. Then $p \otimes q$ is a
projection in $\mathcal{AB}$, for any composite $\mathcal{AB}$ of $A$ and $B$.}
\begin{proof}  We first show that $\langle u_{\mathcal{A}}\otimes u_{\mathcal{B}}|p\otimes q\rangle=\langle u_{\mathcal{A}}\otimes u_{\mathcal{B}}|(p\otimes q)\jProd(p\otimes q)\rangle$. By assumption $p\jProd p=p$ and $q\jProd q=q$, so
\begin{eqnarray}
\langle u_{\mathcal{A}}\otimes u_{\mathcal{B}}|p\otimes q\rangle&=&\langle u_{\mathcal{A}}\otimes u_{\mathcal{B}}|p\jProd p\otimes q\rangle\\
&=&\langle u_{\mathcal{A}}\otimes u_{\mathcal{B}}|(p\otimes u_{\mathcal{B}})\jProd (p\otimes q)\rangle\label{fourSix1}\\
&=&\langle (u_{\mathcal{A}}\otimes u_{\mathcal{B}})\jProd(p\otimes u_{\mathcal{B}})| p\otimes q\rangle\label{eja1}\\
&=&\langle p\otimes u_{\mathcal{B}}| p\otimes q\rangle\label{unit1}\\
&=&\langle p\otimes u_{\mathcal{B}}| p\otimes q\jProd q\rangle\\
&=&\langle p\otimes u_{\mathcal{B}}| (u_{\mathcal{A}}\otimes q)\jProd (p\otimes q)\rangle\label{fourSix2}\\
&=&\langle (p\otimes u_{\mathcal{B}})\jProd (u_{\mathcal{A}}\otimes q)|p\otimes q\rangle\label{eja2}\\
&=&\langle p\otimes q|p\otimes q\rangle\label{fourSix3}\\
&=&\langle (u_{\mathcal{A}}\otimes u_{\mathcal{B}})\jProd(p\otimes q)|p\otimes q\rangle\label{unit2}\\
&=&\langle u_{\mathcal{A}}\otimes u_{\mathcal{B}}|(p\otimes q\rangle)\jProd(p\otimes q)\rangle\label{eja3}\text{,}
\end{eqnarray}
where Eqs.~\eqref{fourSix1}, \eqref{fourSix2}, and \eqref{fourSix3} follow immediately from \cref{prop: main equation}; where Eqs.~\eqref{eja1}, \eqref{eja2}, and \eqref{eja3} follow from \cref{def: new Jordan composite} (\textit{i.e}.\ $\mathcal{AB}$ is an \textsc{eja}, so $\langle a|b\jProd c\rangle=\langle a\jProd b|c\rangle$ for all $a,b,c\in\mathcal{AB}$); and where Eq.~\eqref{unit1} and Eq.~\eqref{unit2} also follow from \cref{def: new Jordan composite} (\textit{i.e}.\ $u_{\mathcal{A}}\otimes u_{\mathcal{B}}$ is the unit in $\mathcal{AB}$). We record this last fact
\begin{equation}
u_{\mathcal{AB}}=u_{\mathcal{A}}\otimes u_{\mathcal{B}}\text{.}
\label{unitsTounit}
\end{equation}
So,
\begin{equation}
\langle u_{\mathcal{AB}}|p\otimes q\rangle=\langle u_{\mathcal{AB}}|(p\otimes q)\jProd(p\otimes q)\rangle\text{.}
\end{equation}
We will now prove that the spectral radius $\rho(p\otimes q)\leq 1$, from which the desired result will follow from \cref{prop3ppp}.\\[0.2cm]
\noindent A \textit{state} on $\mathcal{AB}$ is by definition a positive linear functional $\gamma:\mathcal{AB}_{+}\longrightarrow\mathbb{R}_{\geq 0}$ such that $\gamma(u_{\mathcal{AB}})=1$. Now, $\mathcal{AB}$ is an \textsc{eja}; hence $\mathcal{AB}_{+}$ is self-dual. So, for any state $\gamma$ there exists $Y\in\mathcal{AB}_{+}$ such that $\forall X\in\mathcal{AB}_{+}$ we have that $\gamma(X)=\langle Y|X\rangle$. Let arbitrary $\gamma$ and $Y$ be as such, which means in particular that 
\begin{equation}
\langle Y|u_{\mathcal{AB}}\rangle=1\text{.}
\label{whyU}
\end{equation} 
Now, as in our previous note, choose projection $\tilde{p}\in\mathcal{A}_{+}$ such that $p+\tilde{p}=u_{\mathcal{A}}$. Likewise, choose projection $\tilde{q}\in\mathcal{B}_{+}$ such that $q+\tilde{q}=u_{\mathcal{B}}$. Recalling \eqref{unitsTounit} we then have from Eq.~\eqref{whyU} that
\begin{eqnarray}
1&=&\langle Y|u_{\mathcal{A}}\otimes u_{\mathcal{B}}\rangle\\
&=&\langle Y| (p+\tilde{p})\otimes(q+\tilde{q})\rangle\\
&=&\langle Y| p\otimes q\rangle+\langle Y| p\otimes \tilde{q}\rangle+\langle Y| \tilde{p}\otimes q\rangle+\langle Y| \tilde{p}\otimes \tilde{q}\rangle\label{fourTerms}\text{,}
\end{eqnarray}
where we have appealed to bilinearity of $\otimes$ and $\langle\cdot|\cdot\rangle$. Now, from \cref{def: new Jordan composite} we have that $\mathcal{AB}$ is, in particular, a composite in the sense of \cref{def: composites}. So, from \cref{def: composites} (a) we have that $p\otimes q,p\otimes \tilde{q},\tilde{p}\otimes q,\tilde{p}\otimes \tilde{q}\in\mathcal{AB}_{+}$. Each term on the \textsc{rhs} of Eq.~\eqref{fourTerms} is therefore nonnegative. Therefore, for any state $\gamma$ represented internally by $Y$ we have that
\begin{equation}
1\geq \langle Y|p\otimes q\rangle\text{.}
\end{equation}
Let us now apply the spectral theorem to write
\begin{equation}
p\otimes q=\sum_{j=1}^{n}x_{j}\lambda_{j}\text{,}
\end{equation}
where $\{x_{1},\dots,x_{n}\}$ is a Jordan frame for $\mathcal{AB}$ and $\forall j$ we have that $\lambda_{j}\geq 0$ from $p\otimes q\in\mathcal{AB}_{+}$ and Proposition 1. Choose\footnote{Indeed, $\frac{x_{k}}{\langle x_{k}|x_{k}\rangle}$ is a state for any projection $x_{k}$, since $\left\langle \frac{x_{k}}{\langle x_{k}|x_{k}\rangle}\Big|u\right\rangle=\langle x_{k}\jProd x_{k}|u\rangle\frac{1}{\langle x_{k}|x_{k}\rangle}=\langle x_{k}|x_{k}\rangle\frac{1}{\langle x_{k}|x_{k}\rangle}=1$.} $Y=\frac{x_{k}}{\langle x_{k}|x_{k}\rangle}$ to get
\begin{equation}
1\geq \left\langle \frac{x_{k}}{\langle x_{k}|x_{k}\rangle}\Bigg|\sum_{j=1}^{n}x_{j}\lambda_{j}\right\rangle=\sum_{j=1}^{n}\frac{\lambda_{j}}{\langle x_{k}|x_{k}\rangle}\langle x_{k}|x_{j}\rangle=\sum_{j=1}^{n}\frac{\lambda_{j}}{\langle x_{k}|x_{k}\rangle}\langle u_{\mathcal{AB}}\jProd x_{k}|x_{j}\rangle=\sum_{j=1}^{n}\frac{\lambda_{j}}{\langle x_{k}|x_{k}\rangle}\langle u_{\mathcal{AB}}|x_{k}\jProd x_{j}\rangle=\lambda_{k}\text{.}
\end{equation}
\cref{prop3ppp} now completes the proof.
\end{proof}
\end{lemma}

\noindent \cref{lemma: product projections} has an important consequence regarding the inner product on $\mathcal{AB}$:\ffootnote{Does this really belong here?}

\begin{proposition}\label{prop: inner}  \textit{Let $AB$ be a composite of EJAs $A$ and $B$. Then 
for all $a, x \in A$ and all $b, y \in B$, }
\begin{equation}
\langle a \odot x | b \odot y \rangle = \langle a | x \rangle \langle b | y \rangle\textit{.}
\end{equation}
\begin{proof} We begin with the case in which $a$ and $b$ are minimal projections. Since $a \odot b$ is then 
also a projection, we know that that 
\[\hat{a \odot b} = \|a \odot b \|^{-2}(a \odot b)\]
defines a state. Now define a positive bilinear form $\omega : A \times B \rightarrow \R$ 
by setting 
\[\omega(x,y) = \langle \hat{a \odot b} | x \odot y \rangle = \langle a \odot b | x \odot y \rangle\]
for all $ x \in A$, $y \in B$. Note that $\omega$ is normalized, i.e., a state in the maximal tensor product 
$A \otimes B$. 
\noindent Now evaluate the first marginal of $\omega$ at $a$: 
\begin{eqnarray*}
\omega_{1}(a) = \omega(a \otimes u_B) & = & \langle \hat{a \odot b} | a \otimes u_B\rangle \\
& = & \langle \hat{a \odot b} | a \odot b + a \odot b' \rangle\\
& = & (\langle \hat{a \odot b} | a \odot b \rangle + \langle \hat{a \odot b} | a \odot b' \rangle. 
\end{eqnarray*}
\noindent Since $\langle \hat{a \odot b} | a \odot b \rangle = 1$, the secondond summand at the end is $0$, and we have 
$\omega_{1}(a) = 1$. Since $a$ is minimal, there is only one such state: $\omega_1(a) = \langle \hat{a} |$. Moreover, 
this is a {\em pure} state. The same argument shows that $\omega_2(b) = 1$, so that $\omega_2 = \langle \hat{b} |$. 
As is well known, if a non-signaling state has pure marginals, then it's the product of these marginals \cite{BBLW}.
 Thus, $\omega = \langle \hat{a} | \otimes \langle \hat{b} |$. This gives us 
 \[\langle a \odot b | x \odot y \rangle = c \langle a | x \rangle \langle b | y \rangle.\]
where $c := \frac{\|a \odot b \|^2}{\|a\|^2 \|b\|^2}$. 

\noindent We want to show that $c = 1$.  As a first step, note that $c$ is independent of the choice of minimal projections $a$ and $b$ (Argue by symmetry: since $A$ and $B$ are simple, there are symmetries $\phi$ and $\psi$ taking any given 
pair $a,b$ to any other such pair $a', b'$; hence, there's a symmetry taking $a \odot b$ to $a' \odot b'$ 
\ref{lemma: exchange}. This is 
the only use we make of condition (c) above.) Extend $a$ and $b$ to 
orthogonal decompositions
 of $u_A$ and $u_B$ as sums of projections:  $u_A = \sum_{a_i} a_i$ with $a = a_1$, and $u_B = \sum_j b_j$ with $b = b_1$ Then we have 
 $u_{AB} = \sum_{i,j} a_i \odot b_j$,  so 
 \[\sum_{i,j, k, l} \langle a_i \odot b_j | a_k \odot b_l \rangle = \langle u_{AB} | u_{AB} \rangle = 1\]
 but also, noting that all $a_i$ and $b_j$ here are minimal projections, so that the constant $c$ above is 
 the same for all choices of $a_i \odot b_j$, 
 \[\sum_{i,j,k,l} \langle a_i \odot b_j | a_i \odot b_l \rangle = \sum_{i,j,k,l} c \langle a_i \odot b_j | a_k \odot b_l \rangle = c.\] 
 Hence, $c = 1$. 

\noindent Now suppose $a$ and $b$ are arbitrary elements of $A$ and $B$, respectively. Spectrally decomposing $a$ and $b$ as 
\[a \ = \ \sum_i t_i a_i \ \mbox{and} \  b \ = \ \sum_j s_j b_j\]
where $a_i$ and $b_j$ are pairwise orthogonal families of minimal projections, we have 
\[\langle a \odot b | = \langle \sum_{i, j} t_i s_j a_i \odot b_j |  = \sum_{i,j} t_i s_j \langle a_i \odot b_j |.\]
Hence, for all $x, y$,
\begin{eqnarray*} 
\langle a \odot b | x \odot y \rangle & = & \sum_{i,j} t_i s_j \langle a_i \odot b_j | x \odot y \rangle \\
& = & \sum_{i,j} t_i s_j \langle a_i | x \rangle \langle b_j | y \rangle \\
& = & \langle a | x \rangle \langle b | y \rangle
\end{eqnarray*}
as advertised.
\end{proof}
\end{proposition}

\noindent It now follows that, for $a, a' \in A$ and $b, b' \in B$, $\langle a \otimes b, a' \otimes b' \rangle = 0$ iff 
either $\langle a, a' \rangle = 0$ or $\langle b, b' \rangle = 0$.

\begin{proposition}\label{prop: operator commutation} 
\textit{For all $a \in A$, $b \in B$, $a \otimes u_B$ and $u_A \otimes b$
operator-commute in $\mathcal{AB}$.}
\begin{proof} Suppose $p \in A$ and $q \in B$ are projections, and let $p' = u_A - p$ and $q' = u_B - q$.
Then we have 
\[u_{AB} = u_{A} \otimes u_{B} = (p + p') \otimes (q + q') = p \otimes q + p' \otimes q + p \otimes q' + p' \otimes q'.\]
The four projections appearing on the right are
mutually orthogonal, by \ref{prop: inner}, and sum to the unit in $\mathcal{AB}$.  Hence, $p \otimes
u_B = p \otimes q + p \otimes q'$ and $u_A \otimes q = p \otimes q +
p' \otimes q$ operator commute by \cite{AS}, Lemma 1.48.  Now let $a \in A$
and $b \in B$ be arbitrary: by spectral theory, we have $a = \sum_i
t_i p_i$ and $b = \sum_j s_j q_j$ for pairwise orthogonal projections
$p_i$ and $q_j$. Since $p_i \otimes u_B$ and $u_A \otimes q_j$
operator commute for all $i, j$, it follows that
\[{\redd a \otimes u_B} = \sum_i t_i p_i \otimes u_B \ \ \mbox{and} \ \ {\redd u_A \otimes b} = \sum_j s_j u_A \otimes q_j\]
also operator commute.
\end{proof}
\end{proposition}

\noindent To this point, we have limited information---mainly Proposition 
\ref{prop: operator commutation}---about how the Jordan
structure of a composite $\mathcal{AB}$ interacts with the Jordan structures of 
$A$ and $B$. Our goal is to now establish, for
any $a \in A$ and any $x, y \in B$, the identity $(a \otimes u) \jProd 
(x \otimes y) = (a \jProd x) \otimes y$ --- in other words, that $L_{a
  \otimes u_B}$ acts on $A \otimes B \leq AB$ as $L_{a} \otimes \mathbf{1}_B$,
where $\mathbf{1}_B$ is the identity operator on $B$. This is non-trivial, 
since $\mathcal{AB}$ need not be spanned by $A \otimes B$.

\noindent It will be
helpful first to recall some basic facts about operator exponentials,
or, equivalently, one-parameter groups of linear operators on
finite-dimensional spaces (see, e.g., \cite{Curtis}). Let $V$ be a finite-dimensional real vector
space, and $X$, a linear operator on $V$. Recall that 
$\phi(t) := e^{tX}$ is the unique function $\R \rightarrow {\cal
  L}(\V)$ satisfying the initial-value problem
\[\phi'(t) = X \phi(t) ; \ \ \phi(0) = \mathbf{1}\]
(where $\mathbf{1}$ is the identity operator on $V$). In particular, $\phi'(0) =
X$. The function $\phi$ satisfies $\phi(t + s) = \phi(t)\phi(s)$ 
{\redd and hence, $\phi(t)\phi(-t) = \phi(0) = \mathbf{1}$,} 
 {\redd hence, as $\phi(0) = \mathbf{1}$,  
$\phi(t)$ is invertible, with $\phi(t)^{-1} = \phi(-t)$. In other 
words,} $\phi$ is a one-parameter group of linear operators on $V$.
Conversely, if $\phi : \R \rightarrow {\cal L}(V)$ is any continuous
one-parameter group of linear operators on $V$, then $\phi$ is
differentiable, and $\phi(t) = e^{tX}$ where $X = \phi'(0)$. Notice,
also, that in such a case we have
\[Xa = \frac{d}{dt} \phi(t)a|_{t = 0}\]
for any vector $a \in V$.

\noindent For later reference, the following lemma collects some standard facts: 

\begin{lemma} \label{lemma: operator commutation and one parameter groups}
Let $X, Y$ be linear operators on a finite-dimensional inner product space $V$. Then 
\begin{itemize} 
\item[(a)] $X$ commutes with $e^{tX}$ for all $t$; 
\item[(b)] If 
$e^{tX}$ commutes with $e^{sY}$ for all $t, s$, then $X$ commutes with $Y$; 
\item[(c)] $(e^{tX})^{\dagger} = e^{tX^{\dagger}}$.
\end{itemize}
\end{lemma}  

\tempout{
\noindent{\em Proof:} (a) is standard. For (b), define $F(s,t) =
e^{tX} e^{sY} = e^{sY} e^{tX}$. By the equality of mixed partials, we
have
\[XY F(s,t) = \frac{\partial^2}{\partial s \partial t} F(s,t) = \frac{\partial^2}{\partial s \partial t} F(s,t) 
= YX F(s,t)\] for all $s, t$. In particular, setting $s, t = 0$, so
that $F(s,t) = \mathbf{1}$ (the identity operator), we have $XY = YX$. For
(c), use the observation that if $\phi(t)$ is a one-parameter group,
then $\langle \phi'(0)x, y \rangle = \frac{d}{dt} \langle \phi(t) x, y
\rangle|_{t = 0}$. $\Box$}

\noindent Note that, by (c), if $\phi(t)$ is a one-parameter group with $\phi'(0) = X$ self-adjoint, then $\phi(t)$ is self-adjoint for all $t$.

\noindent Now let $A$ be an EJA. For $a \in A$, 
define 
\[\phi_{a}(t) := e^{t L_{a}} = e^{L_{ta}},\] 
i.e., $\phi_{a}$ is the solution to the initial-value problem
$\frac{d}{dt}\phi_{a} = L_{a} \phi_{a}, \ \ \phi_{a}(0) = \mathbf{1}$. 
By part (c) of Proposition \ref{prop: quadratic representation}, part (c), $\phi_{a}(t) = U_{e^{ta/2}}$; by part 
(b) of the same Proposition, this last is a positive mapping. Since 
$e^{tL_{a}}$ is invertible with inverse $e^{-tL_{a}} = e^{L_{-ta}}$, $\phi_{a}(t)$ is an order-automorphism 
belonging to $G(A)$. 
\tempout{By changing the parametrization, one gets $\phi_{a}(t)$ an
automorphism for all $t$. (In more detail: set $\psi_{a}(s) :=
\phi_{a}(ts)$, with $t$ fixed; then $\psi_{a}$ is still a one-parameter
group, with $\psi_{a}'(0) = t \phi_{a}'(0) = t L_{a} = L_{ta}$, and
$\psi_{a}(1) = \phi_{a}(t)$ an automorphism.)[Didn't we just establish this directly?]}
{\redd 
It follows that $L_a \in {\mathfrak g}_{A}$, the Lie group of the identity component $G(A)$ of $A$. 
Note that $\langle L_{a} x, y \rangle = \langle a x, y \rangle = \langle x, ay \rangle = \langle x, L_{a} y\rangle$ 
for all $x, y \in A$; that is, $L_a$ is self-adjoint. One can show that, conversely, a self-adjoint 
element of ${\mathfrak g}_{A}$ has the form $L_{a}$ for a unique $a \in A$. (See \cite{FK}, pp. 6 and 49, for the details.)}

\tempout{
\begin{corollary}{\redd Do we use this?}
Let $\phi_{a}$ and $\phi_{b}$ be the one-parameter groups
associated with $L_a$ and $L_b$, as discussed above. If $\phi_{a}(t)$ and
$\phi_{b}(s)$ commute for all $t, s$, then so do $L_a$ and $L_b$,
i.e., $a$ and $b$ operator commute.\footnote{\red [AW: Do we actually use this anywhere?]}
\end{corollary}

\noindent{\em Proof:}
Immediate from Lemma \ref{lemma: operator commutation and one parameter groups}.
$\hfill \Box$
}

\noindent We are now ready for the main result of this appendix. 

\begin{proposition}\label{prop: main equation} 
Let $\mathcal{AB}$ be a composite (in the sense of Definition 
\ref{def: new Jordan composite}) of Jordan algebras $A$ and $B$. For 
all $a, x \in A$ and $b, y \in B$, 
\[(a \otimes u_B)\jProd (x \otimes y) = a \jProd x \otimes y \ \ \mbox{and} \ \ (u_A \otimes b)\jProd (x \otimes y) = x \otimes b \jProd y\]
\begin{proof}
We prove the first identity; the second is 
handled similarly. Let $\phi(t)$ be a
one-parameter group on $A$ with $\phi'(0) = L_a$.  Then $\psi(t) :=
\phi(t) \otimes \mathbf{1}$ is a one-parameter group of automorphisms, by 
condition (c) of \ref{def: dynamical composites}. 
 Let $Y = \psi'(0) {\redd \in {\mathfrak g}_{AB}}$; then, for all $x \in A$ and $y \in B$,  
\begin{eqnarray*} 
Y(x \otimes y) 
& = & \left[ \frac{d}{dt} \psi(t)\right]_{t=0}  (x \otimes y) \\
& = & \left[ \frac{d}{dt} (\psi(t)(x \otimes y))\right]_{t = 0}\\
& = & \left[ \frac{d}{dt} \left( \phi(t)x \otimes y \right) \right]_{t = 0} \\
& = & \left( \left[ \frac{d}{dt} \phi(t)\right ]_{t = 0} x \right) \otimes y 
\ = \ L_{a} x \otimes y \ = \ ax \otimes y.
\end{eqnarray*}
Subject to condition (c) of Definition 
\ref{def: new Jordan composite}, 
we have
\[(\phi_{a}(t) \otimes \mathbf{1})^{\dagger} = \phi_{a}(t)^{\dagger} \otimes \mathbf{1} = \phi_{a}(t) \otimes \mathbf{1}.\]
Hence, $Y$ is self-adjoint. 
{\redd As discussed above,} it follows 
that 
there exists some $v \in AB$ with $Y = L_{v}$ on $A \otimes B$.
\ffootnote{{\green [HB: I still need check over the remainder of this proof once more.]}{\blue [AW: proof now much 
shortened!]}}
Thus, 
\[v (x \otimes y) = L_{v} (x \otimes y) = Y(x \otimes y) = ax \otimes y.\]
Setting $x = u_A$ and $y = u_B$, we have 
 $v = v u_{AB} = v(u_{A} \otimes u_{B}) = au_A \otimes u_B = a \otimes u_B$, which gives the 
 advertised result.
\end{proof}
\end{proposition} 

\tempout{
\[\langle v | x \otimes y \rangle = \Tr(ax \otimes y) = \Tr(ax) \Tr(y) = \Tr(ax)\Tr(uy) = \langle a  \otimes u | x \otimes y \rangle\]
(where the last identity is part of the definition of the composite $\mathcal{AB}$)
Hence, the projection of $v$ onto $A \otimes B \leq AB$ is precisely $a \otimes u$. It follows that 
\beq \label{eq: another equation}
v = (a \otimes u) + z
\eeq
where $z$ is orthogonal to $A \otimes B$, i.e., $\langle z | T \rangle = 0$ for every $T \in A \otimes B$. 
We want to show that $z = 0$. 

We begin with the case in which $a = p$, a projection. By assumption
(see Definition \ref{def: new Jordan composite}), $p \otimes u$ is also a projection in $\mathcal{AB}$.  Then
by (\ref{eq: an equation}) we have
\[v (p \otimes u) = p^2 \otimes u = p \otimes u\]
while by Eq. \ref{eq: another equation} also 
\[v (p \otimes u) = ((p \otimes u) + z)(p \otimes u) = (p \otimes u)^2 + z (p \otimes u) = (p \otimes u) + z(p \otimes u).\]
Subtracting, we have $z (p \otimes u) = 0$. Now consider $p' = u - p$: again, $p' \otimes u$ is a projection, and $p' \otimes u = (u - p) \otimes
u = (u \otimes u) - (p \otimes u) = (p \otimes u)'$: 
By (1),
\[v(p' \otimes u) = pp' \otimes u = 0;\]
but also $v(p' \otimes u) = (p \otimes u)(p' \otimes u) + z(p' \otimes u) = z(p' \otimes u)$. 
{\magenta It follows that} $z(p' \otimes u) = 0$, whence, $z = z(p \otimes u + p' \otimes u) = 0$. 
Hence, $z (p' \otimes u) = 0$. We now have
\[ \ \ z = z (u \otimes u) = z ((p \otimes u) + (p' \otimes u)) = z(p \otimes u) + z(p' \otimes u) = 0 + 0 = 0.\]

\noindent The general case now follows from the spectral theorem, plus
linearity: if $a = \sum_i t_i p_i$, where the $p_i$ are projections,
then by the foregoing, $(p_i \otimes u)(x \otimes y) = p_i x \otimes y$ for $i = 1,...,n$, whence 
\[\ \ \ \ (a \otimes u)(x \otimes y) = \sum_i t_i (p_i \otimes u)(x \otimes y) = \sum_i t_i (p_i x \otimes y) = 
ax \otimes y.\]
The identity $(u_A \otimes b)(x \otimes y) = x \otimes by$ follows by symmetry. $\Box$}

\noindent Recall that, for $a \in A$, the mapping $U_a : A \rightarrow A$ is defined by 
$U_a = 2L_{a}^2 - L_{a^2}$. 

\begin{corollary}\label{cor: cor to main equation} In any composite $\mathcal{AB}$ {\redd of EJAs $A$ and $B$}, and 
for any $a \in A$, $b \in B$, $U_{a \otimes u_B}$ and $U_{u_A \otimes A}$ act on $A \otimes B$ as 
$U_a \otimes \id_B$ and $\id_A \otimes U_b$, respectively. \end{corollary}

\noindent We now show that if $A$ and $B$ are non-trivial {\em simple} EJAs --- that is, if neither has any non-trivial direct summands --- then any composite $\mathcal{AB}$ must be special, universally reversible, and an ideal (a direct summand) of the maximal tensor product $A \hotimes B$. The rough idea is that, since $A$ and $B$ have rank at least two, the fact that 
products of distinguishable effects are distinguishable will yield at least four distinguishable effects in $\mathcal{AB}$. 
If the latter were simple, this would be the end of the story; but we know from the case of universal tensor products 
(which we will ultimately show are dynamical composites in our sense) that composites can have non-trivial direct 
summands. Thus, need to work a bit harder, and show that every irreducible direct summand of $\mathcal{AB}$ has rank at least 4. 
 
\noindent {\redd An element} $s \in A$ is {\redd called} a {\em symmetry} iff $s^2 = u$.\footnote{\redd Not to be confused with a symmetry {\em qua} order-automorphism.} In this case
$U_s$ is a Jordan automorphism of $A$, with $U_{s}^{2} = \id$ 
(\cite{AS}, Prop. 2.34). Also
note that $p := \frac{1}{2}(s + u)$ is a projection, and, conversely,
if $p$ is a projection, then $s := 2p - u$ is a symmetry.

\noindent Two projections $p, q \in A$ are {\em exchanged by a symmetry} $s \in
A$ iff $U_s(p) = q$ (in which case, $p = U_{s}(q)$).  More
generally, $p$ and $q$ are {\em equivalent} iff there exists a finite
sequence of symmetries $s_1,...s_{\ell}$ with $q = (U_{s_{\ell}} \circ
\cdots \circ U_{s_1})(p)$.

\noindent Now let $\mathcal{AB}$ be a composite in the sense of Definition \ref{def: new Jordan composite}.

\begin{lemma}\label{lemma: exchange} Let $s \in A$ be a symmetry exchanging projections 
$p_1, p_2 \in A$, and let $t \in B$ be a symmetry exchanging projections $q_1, q_2 \in B$.  
Then $s \otimes u_B$ and $u_A \otimes t$
  are symmetries in $\mathcal{AB}$, and $U_{u_A \otimes t} U_{s \otimes u_B}
  (p_1 \otimes q_1) = p_2 \otimes q_2$. In particular, the projections
  $p_1 \otimes q_1$ and $p_2 \otimes q_2$ are equivalent.
\begin{proof}
$(s \otimes u_B)^2 = s^2 \otimes u_B = u_A
\otimes u_B = u_{AB}$ by Proposition \ref{prop: main equation}. Similarly for $u_A \otimes
t$. Now by Corollary \ref{cor: cor to main equation}, we have 
\[\ \ U_{u_A \otimes t} U_{s \otimes u_B} (p_1 \otimes q_1) = U_{u_A \otimes t}(U_s (p_1) \otimes q_1) = U_{s}(p_1) \otimes U_{t}(q_1) = p_2 \otimes q_2. \ \]
\end{proof}
\end{lemma}

\newpage
\section{Canonical Involutions}\label{epSec}
In this appendix, we prove that the canonical involutions of the Jordan matrix algebras are self-adjoint and unitary. For the reader's convienience, let us collect the relevant material from \cref{partII} in the form of the following proposition.

\begin{proposition}\label{colProp} \textit{Let $\mathcal{A}$ be a simple universally reversible Euclidean Jordan algebra. Let $a\in\mathcal{A}$, and let $X\in C^{*}_{u}(\mathcal{A})$. Let $\psi_{A}$ and $\Phi_{A}$ be as defined in \cref{existUThm}. Let}
\begin{equation} 
J\equiv\begin{pmatrix} \mathbf{0}_{n} & \mathds{1}_{n} \\ -\mathds{1}_{n} & \mathbf{0}_{n}\end{pmatrix}\textit{.}
\end{equation} 
\textit{Then one has the following table\footnote{Note that case of $\mathcal{M}_{2}(\mathbb{H})_{\text{sa}}$ is similiar to $\mathcal{M}_{n>2}(\mathbb{H})_{\text{sa}}$, although in this case we consider the standard embedding in our constuction of $\mathbf{InvQM}$.}:}
\begin{center}
  \begin{tabular}{ c || c | c | c }
    $\mathcal{A}$ & $C^{*}_{u}(\mathcal{A})$  & $\psi_{\mathcal{A}}$ & $\Phi_{\mathcal{A}}$\\[0.25cm] \hline &  & &  \\
    $\mathcal{M}_{n}\big(\mathbb{R}\big)_{s}$ & $\mathcal{M}_{n}\big(\mathbb{C}\big)$  & $\psi_{\mathcal{A}}(a)=a$ & $\Phi_{\mathcal{A}}(X)=X^{T}$\\[0.5cm]
    $\mathcal{M}_{n}\big(\mathbb{C}\big)_{sa}$ & $\mathcal{M}_{n}\big(\mathbb{C}\big)\oplus\mathcal{M}_{n}\big(\mathbb{C}\big)$  & $\psi_{\mathcal{A}}(a)=a\oplus a^{T}$ & $\Phi_{\mathcal{A}}(X=Y\oplus Z)=Z^{T}\oplus Y^{T}$\\[0.32cm]
    $\mathcal{M}_{n>2}\big(\mathbb{H}\big)_{sa}$ & $\mathcal{M}_{2n}\big(\mathbb{C}\big)$  & $\psi_{\mathcal{A}}\big(a=\Gamma_{1}+\Gamma_{2}j\big)=$\begin{small}$\begin{pmatrix}\;\;\;\Gamma_{1} & \Gamma_{2}\\[1mm] -\overline{\Gamma_{2}} & \overline{\Gamma_{1}}\end{pmatrix}$\end{small} & $\Phi_{\mathcal{A}}(X)=JX^{T}J^{T}$
  \end{tabular}
\end{center}
\end{proposition}

\begin{definition} \textit{Let $\mathcal{X}$ be a finite dimensional C$^{*}$-algebra of dimension $n$. We identify $\mathcal{X}$ with $\mathcal{M}_{n}\big(\mathbb{C}\big)$. The \textit{Hilbert-Schmidt inner product} on $\mathcal{X}$ is denoted and defined for all $X,Y\in\mathcal{X}$ via $\langle X,Y\rangle=\mathrm{Tr}\big(X^{\dagger}Y\big)$, where $\dagger$ denotes the usual composition of the usual transpose and complex conjugation operations.}
\end{definition}

\begin{definition}\textit{Let $\mathcal{X}$ be a finite dimensional C$^{*}$-algebra of dimension $n$. We say that a linear operator $\Phi:\mathcal{X}\rightarrow\mathcal{X}$ is} unitary \textit{if $\langle \Phi(X),\Phi(Y)\rangle=\langle X,Y\rangle$ for all $X,Y\in\mathcal{X}$}
\end{definition}

\begin{lemma}\textit{The involutive *-antiautomorphisms tabulated in \cref{colProp} are unitary.}
\begin{proof}
For the case of $\mathcal{A}=\mathcal{M}_{n}\big(\mathbb{R}\big)_{s}$, note that
\begin{eqnarray}
\langle \Phi_{\mathcal{A}}(X),\Phi_{\mathcal{A}}(Y)\rangle&=&\mathrm{Tr}\Big(\Phi_{\mathcal{A}}(X)^{\dagger}\Phi_{\mathcal{A}}(Y)\Big)\nonumber\\
&=&\mathrm{Tr}\Big(\big(X^{T}\big)^{\dagger}Y^{T}\Big)\nonumber\\
&=&\mathrm{Tr}\Big(\big(X^{\dagger}\big)^{T}Y^{T}\Big)\nonumber\\
&=&\mathrm{Tr}\Big(\big(YX^{\dagger}\big)^{T}\Big)\nonumber\\
&=&\mathrm{Tr}\Big(YX^{\dagger}\Big)\nonumber\\
&=&\mathrm{Tr}\Big(X^{\dagger}Y\Big)\nonumber\\
&=&\langle X,Y\rangle\text{.}
\end{eqnarray}
For the case of $\mathcal{A}=\mathcal{M}_{n}\big(\mathbb{C}\big)_{sa}$, note that
\begin{eqnarray}
\langle \Phi_{\mathcal{A}}(Y\oplus Z),\Phi_{\mathcal{A}}(\tilde{Y}\oplus\tilde{Z})\rangle&=&\mathrm{Tr}\Big(\Phi_{\mathcal{A}}(Y\oplus Z)^{\dagger}\Phi_{\mathcal{A}}(\tilde{Y}\oplus\tilde{Z})\Big)\nonumber\\
&=&\mathrm{Tr}\Big(\big(Z^{T}\oplus Y^{T}\big)^{\dagger}\big(\tilde{Z}^{T}\oplus \tilde{Y}^{T}\big)\Big)\nonumber\\
&=&\mathrm{Tr}\Big(\big(Z^{\dagger}\oplus Y^{\dagger}\big)^{T}\big(\tilde{Z}\oplus \tilde{Y}\big)^{T}\Big)\nonumber\\
&=&\mathrm{Tr}\left(\Big(\big(\tilde{Z}\oplus \tilde{Y}\big)\big(Z^{\dagger}\oplus Y^{\dagger}\big)\Big)^{T}\right)\nonumber\\
&=&\mathrm{Tr}\Big(\big(\tilde{Z}\oplus \tilde{Y}\big)\big(Z^{\dagger}\oplus Y^{\dagger}\big)\Big)\nonumber\\
&=&\mathrm{Tr}\Big(\big(Z^{\dagger}\oplus Y^{\dagger}\big)\big(\tilde{Z}\oplus \tilde{Y}\big)\Big)\nonumber\\
&=&\mathrm{Tr}\Big(\big(Y^{\dagger}\oplus Z^{\dagger}\big)\big(\tilde{Y}\oplus \tilde{Z}\big)\Big)\nonumber\\
&=&\mathrm{Tr}\Big(\big(Y\oplus Z\big)\big)^{\dagger}\big(\tilde{Y}\oplus \tilde{Z}\big)\Big)\nonumber\\
&=&\langle Y\oplus Z,\tilde{Y}\oplus\tilde{Z}\rangle\text{.}
\end{eqnarray}
For the case of $\mathcal{A}=\mathcal{M}_{n>2}\big(\mathbb{H}\big)_{sa}$, note that
\begin{eqnarray}
\langle \Phi_{\mathcal{A}}(X),\Phi_{\mathcal{A}}(Y)\rangle&=&\mathrm{Tr}\Big(\Phi_{\mathcal{A}}(X)^{\dagger}\Phi_{\mathcal{A}}(Y)\Big)\nonumber\\
&=&\mathrm{Tr}\Big(\big(JX^{T}J^{T}\big)^{\dagger}\big(JY^{T}J^{T}\big)\Big)\nonumber\\
&=&\mathrm{Tr}\Big(J(X^{T})^{\dagger}J^{T}JY^{T}J^{T}\Big)\nonumber\\
&=&\mathrm{Tr}\Big((X^{\dagger})^{T}Y^{T}\Big)\nonumber\\
&=&\mathrm{Tr}\Big(\big(Y(X^{\dagger}\big)^{T}\Big)\nonumber\\
&=&\mathrm{Tr}\Big(X^{\dagger}Y\Big)\nonumber\\
&=&\langle X,Y\rangle\text{.}
\end{eqnarray}
\end{proof}
\end{lemma}

\begin{definition}\textit{Let $\mathcal{X}$ be a finite dimensional C$^{*}$-algebra of dimension $n$. For every bounded linear operator $\Phi:\mathcal{X}\rightarrow\mathcal{X}$ one defines the corresponding \textit{adjoint} linear operator --- denoted by $\Phi^{\dagger}$ --- via $\langle X,\Phi(Y)\rangle=\langle \Phi^{\dagger}(X),Y\rangle$ for all $X,Y\in\mathcal{X}$. If $\Phi=\Phi^{\dagger}$, then one says that $\Phi$ is self-adjoint.}
\end{definition}

\begin{lemma}\textit{The involutive *-antiautomorphisms tabulated in \cref{colProp} are self-adjoint.}
\begin{proof}
Let us first recall the elementary fact that $\forall X,Y\in\mathcal{X}$: $\mathrm{Tr}(XY^{T})=\mathrm{Tr}(X^{T}Y)$. Let us also note that $JXJ^{T}=J^{T}XJ$ $\forall X\in\mathcal{M}_{2n}\big(\mathbb{C}\big)$. We now proceed with the proof by cases.\\
For the case of $\mathcal{A}=\mathcal{M}_{n>2}\big(\mathbb{R}\big)_{s}$, note that
\begin{eqnarray}
\langle \Phi_{\mathcal{A}}^{\dagger}(X),Y\rangle&=&\langle X,\Phi_{\mathcal{A}}(Y)\rangle\nonumber\\
&=&\mathrm{Tr}\Big(X^{\dagger}Y^{T}\Big)\nonumber\\
&=&\mathrm{Tr}\Big(\big(X^{\dagger}\big)^{T}Y\Big)\nonumber\\
&=&\mathrm{Tr}\Big(\big(X^{T}\big)^{\dagger}Y\Big)\nonumber\\
&=&\langle \Phi_{\mathcal{A}}(X),Y\rangle\text{.}
\end{eqnarray}
For the case of $\mathcal{A}=\mathcal{M}_{n>2}\big(\mathbb{C}\big)_{sa}$, note that
\begin{eqnarray}
\langle \Phi_{\mathcal{A}}^{\dagger}(Y\oplus Z),\tilde{Y}\oplus\tilde{Z}\rangle&=&\langle Y\oplus Z,\Phi_{\mathcal{A}}(\tilde{Y}\oplus\tilde{Z})\rangle\nonumber\\
&=&\mathrm{Tr}\Big(\big(Y\oplus Z\big)^{\dagger}\big(\tilde{Z}^{T}\oplus \tilde{Y}^{T}\big)\Big)\nonumber\\
&=&\mathrm{Tr}\Big(Y^{\dagger}\tilde{Z}^{T}\Big)+\mathrm{Tr}\Big(Z^{\dagger}\tilde{Y}^{T}\Big)\nonumber\\
&=&\mathrm{Tr}\Big(\big(Y^{\dagger}\big)^{T}\tilde{Z}\Big)+\mathrm{Tr}\Big(\big(Z^{\dagger}\big)^{T}\tilde{Y}\Big)\nonumber\\
&=&\mathrm{Tr}\Big(\big(Y^{T}\big)^{\dagger}\tilde{Z}\Big)+\mathrm{Tr}\Big(\big(Z^{T}\big)^{\dagger}\tilde{Y}\Big)\nonumber\\
&=&\mathrm{Tr}\Big(\big(Z^{T}\oplus Y^{T}\big)^{\dagger}\big(\tilde{Y}\oplus \tilde{Z}\big)\Big)\nonumber\\
&=&\langle \Phi_{\mathcal{A}}(Y\oplus Z),\tilde{Y}\oplus\tilde{Z}\rangle\text{.}
\end{eqnarray}
For the case of $\mathcal{A}=\mathcal{M}_{n>2}\big(\mathbb{H}\big)_{sa}$, note that
\begin{eqnarray}
\langle \Phi_{\mathcal{A}}^{\dagger}(X),Y\rangle&=&\langle X,\Phi_{\mathcal{A}}(Y)\rangle\nonumber\\
&=&\mathrm{Tr}\Big(X^{\dagger}\big(JY^{T}J^{T}\big)\Big)\nonumber\\
&=&\mathrm{Tr}\Big(\big(J^{T}X^{\dagger}J\big)Y^{T}\Big)\nonumber\\
&=&\mathrm{Tr}\Big(\big(J^{T}X^{\dagger}J\big)^{T}Y\Big)\nonumber\\
&=&\mathrm{Tr}\Big(\big(J^{T}X^{T}J\big)^{\dagger}Y\Big)\nonumber\\
&=&\mathrm{Tr}\Big(\big(JX^{T}J^{T}\big)^{\dagger}Y\Big)\nonumber\\
&=&\langle \Phi_{\mathcal{A}}(X),Y\rangle\text{.}
\end{eqnarray}
\end{proof} 
\end{lemma}

\newpage
\section[The Category CJPc]{The Category $\mathbf{CJP}_{\!\mathscr{C}}$}\label{SMCproofApp}
In this appendix, we explicitly show that $\mathbf{CJP}_{\!\mathscr{C}}$ is a symmetric monoidal category. First, recall our definition of $\mathbf{CJP}_{\!\mathscr{C}}$ from \cref{categoriesEJA}.
\begin{definition} \textit{Let} $\mathbf{CJP}_{\!\mathscr{C}}$\textit{ be the category with objects} \textsc{ejc}\textit{-algebras} $(\mathcal{A},\mathbf{M}_{\mathcal{A}})$ \textit{from a class} $\mathscr{C}$ \textit{closed under} $\odot$ \textit{and containing} $\mathrm{I}\equiv(\mathbb{R},\mathbb{C})$\textit{, and morphisms}\footnote{When clear from context, we overload the symbol $\mathcal{A}$ to stand for the object $(\mathcal{A},\mathbf{M}_{\mathcal{A}})$.} $\phi\in\text{hom}(\mathcal{A},\mathcal{B})$\textit{ positive maps} $\phi:\mathbf{M}_{\mathcal{A}}\longrightarrow\mathbf{M}_{\mathcal{B}}$ \textit{such that} $\forall\mathcal{C}\in\mathscr{C}$ $\phi\otimes\text{id}_{\mathbf{M}_{\mathcal{C}}}$ \textit{is positive and maps }$\mathcal{A}\odot\mathcal{C}$ \textit{into} $\mathcal{B}\odot\mathcal{C}$, \textit{i.e}.\ $\phi\otimes\text{id}_{\mathbf{M}_{\mathcal{C}}}::\mathcal{A}\odot\mathcal{C}\longmapsto\mathcal{B}\odot\mathcal{C}$\textit{, where $\otimes$ is the usual monoidal product in the symmetric monoidal category of finite dimensional complex *-algebras.}
\end{definition}
\begin{lemma}\label{catLemAB} $\mathbf{CJP}_{\!\mathscr{C}}$ \textit{is a category.}
\begin{proof} We detail the required properties. We shall write $x,y,z,\dots$ for elements of finite dimensional complex *-algebras and $a,b,c,\dots$ for elements of EJC-algebras (\textit{i.e}.\ Jordan subalgebras of the self-adjoint parts of finite dimensional complex *-algebras.) The obvious identity morphisms are
\begin{equation}
1_{\mathcal{A}}=\text{id}_{\mathbf{M}_{\mathcal{A}}}:\mathbf{M}_{\mathcal{A}}\longrightarrow\mathbf{M}_{\mathcal{A}}::x\longmapsto x\text{.}
\end{equation}
Indeed, with $\phi\in\text{hom}(\mathcal{A},\mathcal{B})$ we then have that $1_{\mathcal{B}}\circ\phi=\phi=\phi\circ1_{\mathcal{A}}$. The composition of morphisms in $\mathbf{CJP}_{\!\mathscr{C}}$ is simply inherited from the composition of morphisms in $\mathbf{*}$-$\mathbf{Alg}$\footnote{Henceforth, we write $\mathbf{*}$-$\mathbf{Alg}$ for the category of complex *-algebras, \textit{i.e}.\ the category with objects finite dimensional complex *-algebras, morphisms linear functions thereof, with composition given by usual functional composition, and with $\alpha,\lambda,\rho,\sigma$ the usual natural isomorphisms (the associator, left and right unit --- \textit{i.e}.\ $\mathbb{C}$ --- introductions, respectively, and symmetor) which render $\mathbf{*}$-$\mathbf{Alg}$ a symmetric monoidal category with respect to the usual monoidal (\textit{i.e}.\ tensor) product $\otimes$.}, and by the proof of \cref{prop:completely Jordan preserving composes and tensors} (i) we have that the composition of $\mathbf{CJP}_{\!\mathscr{C}}$ morphisms are indeed Jordan preserving with respect to $\mathscr{C}$. Associativity of composition is also inherited from $\mathbf{*}$-$\mathbf{Alg}$, indeed, simply from the composition of linear functions.
\end{proof}
\end{lemma}
\begin{lemma}\label{bfLemAB} $\odot:\mathbf{CJP}_{\!\mathscr{C}}\times\mathbf{CJP}_{\!\mathscr{C}}\longrightarrow\mathbf{CJP}_{\!\mathscr{C}}$ \textit{is a bifunctor.}
\begin{proof} First, note that \textit{by definition} $\mathscr{C}$ is closed under $\odot$; in category notation, $\forall(\mathcal{A},\mathcal{B})\in\text{ob}\big(\mathbf{CJP}_{\!\mathscr{C}}\times\mathbf{CJP}_{\!\mathscr{C}}\big)$ one has $\odot(\mathcal{A},\mathcal{B}\big)\in\text{ob}\big(\mathbf{CJP}_{\!\mathscr{C}}\big)$. Next, let $(f_{1},f_{2})\in\text{hom}\big(\mathbf{CJP}_{\!\mathscr{C}}\times\mathbf{CJP}_{\!\mathscr{C}}\big)$. Then 
\begin{equation}
\odot(f_{1},f_{2})\equiv f_{1}\otimes f_{2}:\mathbf{M}_{\text{dom}(f_{1})}\otimes \mathbf{M}_{\text{dom}(f_{2})}\longrightarrow \mathbf{M}_{\text{cod}(f_{1})}\otimes \mathbf{M}_{\text{cod}(f_{2})}
\end{equation} 
is again a morphism in $\mathbf{CJP}_{\!\mathscr{C}}$ in light of our proof of \cref{prop:completely Jordan preserving composes and tensors} (ii), where one notes that $\mathbf{M}_{\text{dom}(f_{1})}\otimes \mathbf{M}_{\text{dom}(f_{2})}$ is by definition $\mathbf{M}_{\text{dom}(f_{1})\odot\text{dom}(f_{2})}$ and similarly for the codomain.\\
One also notes that 
\begin{equation}
\odot(1_{\mathcal{A}\times\mathcal{B}}\big)=\odot\big(1_{\mathcal{A}}\times 1_{\mathcal{B}}\big)=\text{id}_{\mathbf{M}_{\mathcal{A}}}\otimes\text{id}_{\mathbf{M}_{\mathcal{B}}}=1_{\odot(\mathcal{A},\mathcal{B})}\text{.}
\end{equation}
Finally, let $(f_{1},f_{2}):\mathcal{A}_{1}\times\mathcal{A}_{2}\longrightarrow\mathcal{B}_{1}\times\mathcal{B}_{2}$ and $(g_{1},g_{2}):\mathcal{B}_{1}\times\mathcal{B}_{2}\longrightarrow\mathcal{C}_{1}\times\mathcal{C}_{2}$ be composable morphisms in $\mathbf{CJP}_{\!\mathscr{C}}\times\mathbf{CJP}_{\!\mathscr{C}}$. Then 
\begin{equation}
\odot\big((g_{1},g_{2})\circ(f_{1},f_{2})\big)=\odot\big((g_{1}\circ f_{1}),(g_{2}\circ f_{2})\big)
\end{equation} 
because \cref{catLemAB} ensures that $\mathbf{CJP}_{\!\mathscr{C}}\times\mathbf{CJP}_{\!\mathscr{C}}$ is a category. Then by definition 
\begin{equation}
\odot\big((g_{1}\circ f_{1}),(g_{2}\circ f_{2})\big)=(g_{1}\circ f_{1})\otimes(g_{2}\circ f_{2})=\big(g_{1}\otimes g_{2}\big)\circ\big(f_{1}\otimes f_{2}\big)\text{,}
\end{equation} 
the last equality following from the fact that $\otimes:\mathbf{*}\text{-}\mathbf{Alg}\times\mathbf{*}\text{-}\mathbf{Alg}\longrightarrow\mathbf{*}\text{-}\mathbf{Alg}$ is a bifunctor. But 
\begin{equation}
\big(g_{1}\otimes g_{2}\big)\circ\big(f_{1}\otimes f_{2}\big)=\odot(g_{1},g_{2})\circ\odot(f_{1},f_{2})\text{,}
\end{equation}
which completes the proof.
\end{proof}
\end{lemma}
\begin{theorem} $\mathbf{CJP}_{\!\mathscr{C}}$ \textit{is a symmetric monoidal category with unit} $\mathrm{I}\equiv(\mathbb{R},\mathbb{C})$ \textit{and monoidal functor} $\odot$\textit{.}
\begin{proof} We will introduce the appropriate natural transformations and check the required conditions. Let $\mathds{1}:\mathbf{CJP}_{\!\mathscr{C}}\longrightarrow\mathbf{CJP}_{\!\mathscr{C}}$ be the usual identity endofunctor on $\mathbf{CJP}_{\!\mathscr{C}}$. Introduce the following functor morphism 
\begin{equation}
\lambda:\mathds{1}\longrightarrow\odot
\end{equation}
so that
\begin{equation}
\lambda:\text{ob}(\mathbf{CJP}_{\!\mathscr{C}})\longrightarrow\text{hom}(\mathbf{CJP}_{\!\mathscr{C}})::\mathcal{A}\longrightarrow \lambda_{\mathcal{A}}\text{,}
\end{equation}
where the components are
\begin{equation}
\lambda_{\mathcal{A}}:\mathds{1}(\mathcal{A})\longmapsto\odot(\mathbb{R},\mathcal{A})\text{.}
\end{equation}
The components $\lambda_{\mathcal{A}}$ are positive and furthermore Jordan preserving since $\odot(\mathbb{R},\mathcal{A})=\mathcal{A}$ for any $\mathcal{A}\in\mathscr{C}$. Now, let $f\in\text{hom}\big(\mathcal{A},\mathcal{B}\big)$. We have $\forall \mathcal{A}\in\text{ob}(\mathbf{CJP}_{\!\mathscr{C}})$ that $\odot(\mathbb{R},\mathcal{A})=\mathcal{A}$. Thus the following diagram commutes
\begin{equation}
\xymatrixrowsep{1cm}
\xymatrixcolsep{2cm}
\xymatrix{\mathcal{A}\ar@{=}[r]\ar@/^2pc/[rrr]_{f}\ar@{=}[d]  &\mathds{1}(\mathcal{A})\ar[r]^{\mathds{1}(f)}\ar[d]_{\lambda_{\mathcal{A}}} & \mathds{1}(\mathcal{B})\ar@{=}[r]\ar[d]^{\lambda_{\mathcal{B}}} &\mathcal{B}\ar@{=}[d]\\
\mathcal{A}\ar@{=}[r]\ar@/_2pc/[rrr]_{f}\ar@{=}[u] & \odot(\mathbb{R},\mathcal{A})\ar[r]_{\odot(1_{\mathbb{R}},f)} & \odot(\mathbb{R},\mathcal{B})\ar@{=}[r]&\mathcal{B}}
\end{equation}
So $\lambda$ is a natural transformation. Moreover $\exists \lambda^{-1}:\odot\rightarrow\mathds{1}$ such that $\lambda_{\mathcal{A}}\circ\lambda_{\mathcal{A}}^{-1}=1_{\mathds{1}(\mathcal{A})}$ and $\lambda_{\mathcal{A}}^{-1}\circ\lambda_{\mathcal{A}}=1_{\odot(\mathbb{R},\mathcal{A})}$. So $\lambda$ is a natural isomorphism. Reproducing the above \textit{mutatis mutandis} one also has the following natural isomorphism
\begin{equation}
\rho:\text{ob}(\mathbf{CJP}_{\!\mathscr{C}})\longrightarrow\text{hom}(\mathbf{CJP}_{\!\mathscr{C}})::\mathcal{A}\longrightarrow \rho_{\mathcal{A}}\text{,}
\end{equation}
where the components are
\begin{equation}
\rho_{\mathcal{A}}:::\mathds{1}(\mathcal{A})\longmapsto\odot(\mathcal{A},\mathbb{R})\text{.}
\end{equation}
Next, introduce the following functor morphism $\sigma:\odot\longrightarrow\odot$ with
\begin{equation}
\sigma:\text{ob}(\mathbf{CJP}_{\!\mathscr{C}}\times\mathbf{CJP}_{\!\mathscr{C}})\longrightarrow\text{hom}(\mathbf{CJP}_{\!\mathscr{C}})::(\mathcal{A},\mathcal{B})\longrightarrow\sigma_{\mathcal{A},\mathcal{B}}
\end{equation}
with components
\begin{equation}
\sigma_{\mathcal{A},\mathcal{B}}:\odot(\mathcal{A},\mathcal{B})\longmapsto\odot(\mathcal{B},\mathcal{A})
\end{equation}
$\sigma_{\mathcal{A,B}}$ is of course Jordan preserving relative to $\mathscr{C}$, because the swap operation is a positive linear function, and since $\mathbf{M}\otimes\mathbf{N}\cong\mathbf{N}\otimes\mathbf{M}$ as complex *-algebras with $\mathbf{M},\mathbf{N}\in\text{ob}(\mathbf{*}\text{-}\mathbf{Alg})$. The former observation renders $\sigma$ a natural isomorphism. So the following diagrams commute
\begin{eqnarray}
\xymatrix{\odot(\mathcal{A},\mathcal{B})\ar[d]_{\sigma_{\mathcal{A},\mathcal{B}}}\ar[rd]^{1_{\odot(\mathcal{A},\mathcal{B})}} & \\ \odot(\mathcal{B},\mathcal{A})\ar[r]_{\sigma_{\mathcal{B},\mathcal{A}}} & \odot(\mathcal{A},\mathcal{B})} \hspace{2cm} \xymatrix{ & \mathcal{A}\ar[d]^{\rho_{\mathcal{A}}}\ar[dl]_{\lambda_{\mathcal{A}}}\\ \odot(\mathbb{R},\mathcal{A})& \ar[l]^{\sigma_{\mathcal{A},\mathbb{R}}}\odot(\mathcal{A},\mathbb{R})} 
\end{eqnarray}
We now come to associativity.\\
Let $\odot_{12,3}:\mathbf{CJP}_{\!\mathscr{C}}\times\mathbf{CJP}_{\!\mathscr{C}}\times\mathbf{CJP}_{\!\mathscr{C}}\longrightarrow\mathbf{CJP}_{\!\mathscr{C}}$ be the following functor:
\begin{eqnarray}
\forall (\mathcal{A},\mathcal{B},\mathcal{C})\in\text{ob}(\mathbf{CJP}_{\!\mathscr{C}}\times\mathbf{CJP}_{\!\mathscr{C}}\times\mathbf{CJP}_{\!\mathscr{C}})&\text{ define }&\odot_{12,3}(\mathcal{A},\mathcal{B},\mathcal{C})=\odot\Big(\odot(\mathcal{A},\mathcal{B}),\mathcal{C}\Big)\label{ob12three}\text{.}\\
\forall (f,g,h)\in\text{hom}(\mathbf{CJP}_{\!\mathscr{C}}\times\mathbf{CJP}_{\!\mathscr{C}}\times\mathbf{CJP}_{\!\mathscr{C}})&\text{ define }&\odot_{12,3}(f,g,h)=\odot\Big(\odot(f,g),h\Big)\label{hom12three}\text{.}
\end{eqnarray}
Let $\odot_{1,23}:\mathbf{CJP}_{\!\mathscr{C}}\times\mathbf{CJP}_{\!\mathscr{C}}\times\mathbf{CJP}_{\!\mathscr{C}}\longrightarrow\mathbf{CJP}_{\!\mathscr{C}}$ be the following functor:
\begin{eqnarray}
\forall (\mathcal{A},\mathcal{B},\mathcal{C})\in\text{ob}(\mathbf{CJP}_{\!\mathscr{C}}\times\mathbf{CJP}_{\!\mathscr{C}}\times\mathbf{CJP}_{\!\mathscr{C}})&\text{ define }&\odot_{1,23}(\mathcal{A},\mathcal{B},\mathcal{C})=\odot\Big(\mathcal{A},\odot(\mathcal{B},\mathcal{C})\Big)\\
\forall (f,g,h)\in\text{hom}(\mathbf{CJP}_{\!\mathscr{C}}\times\mathbf{CJP}_{\!\mathscr{C}}\times\mathbf{CJP}_{\!\mathscr{C}})&\text{ define }&\odot_{1,23}(f,g,h)=\odot\Big(f,\odot(g,h)\Big)
\end{eqnarray}
Introduce the following functor morphism $\alpha:\odot_{12,3}\longrightarrow\odot_{1,23}$ so that
\begin{equation}
\alpha:\text{ob}(\mathbf{CJP}_{\!\mathscr{C}}\times\mathbf{CJP}_{\!\mathscr{C}}\times\mathbf{CJP}_{\!\mathscr{C}})\longrightarrow\text{hom}(\mathbf{CJP}_{\!\mathscr{C}})::(\mathcal{A},\mathcal{B},\mathcal{C})\longrightarrow\alpha_{\mathcal{A},\mathcal{B},\mathcal{C}}\text{,}
\end{equation}
where the components are
\begin{equation}
\alpha_{\mathcal{A},\mathcal{B},\mathcal{C}}:\odot\Big(\odot(\mathcal{A},\mathcal{B}),\mathcal{C}\Big)\longmapsto\odot\Big(\mathcal{A},\odot(\mathcal{B},\mathcal{C})\Big)\text{.}
\end{equation}
The components are again positive, and furthermore Jordan preserving, in light of the fact that $\otimes$ is a monoidal functor, and so it is associative. Now, we are going to prove that $\alpha$ is a natural transformation. First, let us verify that $\odot_{12,3}$ and $\odot_{1,23}$ are indeed functors. We will verify the case of $\odot_{12,3}$, with verification of the other case following in an entirely similar manner. First, it is to be observed that $\text{ob}(\mathbf{CJP}_{\!\mathscr{C}})$ is by definition closed under the canonical tensor product. Moreover, $\text{hom}(\mathbf{CJP}_{\!\mathscr{C}})$ is closed under the canonical tensor product in light of \cref{bfLemAB}. Next, in light of Eq.~\eqref{hom12three}
\begin{equation}
\xymatrixcolsep{1.5cm}
\xymatrix{1_{(\mathcal{A},\mathcal{B},\mathcal{C})}\ar[r]^{\odot_{12,3}} & \odot\Big(1_{\odot(\mathcal{A},\mathcal{B})},1_{\mathcal{C}}\Big)\ar@{=}[r] & 1_{\odot\big(\odot(\mathcal{A},\mathcal{B}),\mathcal{C}\big)}\ar@{=}[r] &1_{\odot_{12,3}(\mathcal{A},\mathcal{B},\mathcal{C})}}
\end{equation}
with the first equality following from the proof of \cref{bfLemAB} and the second by the definition in Eq.~\eqref{ob12three}. Next, let $f=(f_{1},f_{2},f_{3})$ and $g=(g_{1},g_{2},g_{3})$ be composable in $\text{hom}(\mathbf{CJP}_{\!\mathscr{C}}\times\mathbf{CJP}_{\!\mathscr{C}}\times\mathbf{CJP}_{\!\mathscr{C}})$. Then one has from Eq.~\eqref{hom12three}
\begin{equation}
\xymatrixcolsep{1.5cm}
\xymatrix{g\circ f \ar@{=}[r]&(g_{1}\circ f_{1},g_{2}\circ f_{2},g_{3}\circ f_{3}) \ar[r]^{\odot_{12,3}} & \odot\Big(\odot(g_{1}\circ f_{1},g_{2}\circ f_{2}),g_{3}\circ f_{3}\Big)
}\label{123hom}
\end{equation}
with the first equality coming from the definition of the product category $\mathbf{CJP}_{\!\mathscr{C}}\times\mathbf{CJP}_{\!\mathscr{C}}\times\mathbf{CJP}_{\!\mathscr{C}}$, the second by the definition in Eq.~\eqref{hom12three}. Now by \cref{bfLemAB} we have $\odot(g_{1}\circ f_{1},g_{2}\circ f_{2})=\odot(g_{1},g_{2})\circ\odot(f_{1},f_{2})\equiv G\circ F\in\text{hom}(\mathbf{CJP}_{\!\mathscr{C}})$. And again by \cref{bfLemAB} we have have $\odot(G\circ F,g_{3}\circ f_{3})=\odot(G,g_{3})\circ \odot(F,f_{3})$. But in light of the definition in Eq.~\eqref{hom12three}, and the foregoing observations, we may now continue the diagram in Eq.~\eqref{123hom} as follows
\begin{equation}
\odot\Big(\odot(g_{1}\circ f_{1},g_{2}\circ f_{2}),g_{3}\circ f_{3}\Big)=\odot\Big(\odot(g_{1},g_{2}),g_{3}\Big)\circ\odot\Big(\odot(f_{1},f_{2}),f_{3}\Big)= \odot_{12,3}(g_{1},g_{2},g_{3})\circ\odot_{12,3}(f_{1},f_{2},f_{3})
\end{equation}
So $\odot_{12,3}$ really is a functor. So too is $\odot_{1,23}$, with the proof following in an entirely similar manner as noted above. We now come to recall an important result from Chapter 9, namely \cref{prop: associativity}: the canonical tensor product is associative! Thus the following diagram commutes $\forall(\mathcal{A}_{1},\mathcal{A}_{2},\mathcal{A}_{3})\in\text{ob}(\mathbf{CJP}_{\!\mathscr{C}}\times\mathbf{CJP}_{\!\mathscr{C}}\times\mathbf{CJP}_{\!\mathscr{C}})$ and $\forall (f_{1},f_{2},f_{3})\in\text{hom}(\mathbf{CJP}_{\!\mathscr{C}}\times\mathbf{CJP}_{\!\mathscr{C}}\times\mathbf{CJP}_{\!\mathscr{C}})$ with $\text{cod}(f_{1})=\mathcal{B}_{1}\in\text{ob}(\mathbf{CJP}_{\!\mathscr{C}})$, $\text{cod}(f_{2})=\mathcal{B}_{2}\in\text{ob}(\mathbf{CJP}_{\!\mathscr{C}})$, and $\text{cod}(f_{3})=\mathcal{B}_{3}\in\text{ob}(\mathbf{CJP}_{\!\mathscr{C}})$
\begin{equation}
\xymatrixrowsep{1cm}
\xymatrixcolsep{4cm}
\xymatrix{\odot_{12,3}(\mathcal{A}_{1},\mathcal{A}_{2},\mathcal{A}_{3})\ar[r]^{\odot_{12,3}(f_{1},f_{2},f_{3})}\ar[d]_{\alpha_{\mathcal{A}_{1},\mathcal{A}_{2},\mathcal{A}_{3}}} & \odot_{12,3}(\mathcal{B}_{1},\mathcal{B}_{2},\mathcal{B}_{3})\ar[d]^{\alpha_{\mathcal{B}_{1},\mathcal{B}_{2},\mathcal{B}_{3}}}\\
\odot_{1,23}(\mathcal{A}_{1},\mathcal{A}_{2},\mathcal{A}_{3})\ar[r]_{\odot_{1,23}(f_{1},f_{2},f_{3})} & \odot_{1,23}(\mathcal{B}_{1},\mathcal{B}_{2},\mathcal{B}_{3})
}
\end{equation}
One should note in light of the associativity of $\otimes$ that
\begin{eqnarray}
\odot_{12,3}(f_{1},f_{2},f_{3})\equiv\odot\Big(\odot(f_{1},f_{2}),f_{3}\Big)=\big(f_{1}\otimes f_{2}\big)\otimes f_{3}=f_{1}\otimes\big(f_{2}\otimes f_{3}\big)=\odot_{1,23}(f_{1},f_{2},f_{3})
\end{eqnarray}
One should also note in light of the associativity of $\odot$ on objects that
\begin{equation}
\mathbf{M}_{\Big(\odot\big(\odot(\mathcal{A}_{1},\mathcal{A}_{2}),\mathcal{A}_{3}\big)\Big)}\cong \mathbf{M}_{\Big(\odot\big(\mathcal{A}_{1},\odot(\mathcal{A}_{2},\mathcal{A}_{3})\big)\Big)}
\end{equation}
So $\alpha$ really is a natural transformation; moreover a natural isomorphism with $\alpha^{-1}:\odot_{1,23}\longrightarrow\odot_{1,23}$ defined such that
\begin{eqnarray}
\alpha^{-1}:\text{ob}(\mathbf{CJP}_{\!\mathscr{C}}\times\mathbf{CJP}_{\!\mathscr{C}}\times\mathbf{CJP}_{\!\mathscr{C}})\longrightarrow\text{hom}(\mathbf{CJP}_{\!\mathscr{C}})::(\mathcal{A},\mathcal{B},\mathcal{C})\longrightarrow\alpha^{-1}_{\mathcal{A},\mathcal{B},\mathcal{C}}\text{,}
\end{eqnarray}
where
\begin{equation}
\alpha^{-1}_{\mathcal{A},\mathcal{B},\mathcal{C}}:::\odot\Big(\mathcal{A},\odot(\mathcal{B},\mathcal{C})\Big)\longmapsto\odot\Big(\odot(\mathcal{A},\mathcal{B}),\mathcal{C}\Big)\text{.}
\end{equation}
The remaining coherence conditions for an $\textsc{smc}$ hold in light of our proof of \cref{prop:completely Jordan preserving composes and tensors}. In particular, the following diagram commutes $\forall \mathcal{A},\mathcal{B},\mathcal{C},\mathcal{D}\in\text{ob}(\mathbf{CJP}_{\!\mathscr{C}})$ since $\odot \big(\alpha_{\mathcal{A},\mathcal{B},\mathcal{C}},1_{\mathcal{D}}\big)\in\text{hom}(\mathbf{CJP}_{\!\mathscr{C}})$ and $\odot(1_{\mathcal{A}},\alpha_{\mathcal{B},\mathcal{C},\mathcal{D}}\big)\in\text{hom}(\mathbf{CJP}_{\!\mathscr{C}})$
\begin{equation}
\xymatrixrowsep{1.2cm}
\xymatrixcolsep{2cm}
\xymatrix{ & \odot\Big(\odot\big(\mathcal{A},\mathcal{B}\big),\odot\big(\mathcal{C},\mathcal{D}\big)\Big)\ar[rd]^{\alpha_{\mathcal{A},\mathcal{B},\odot(\mathcal{C},\mathcal{D})}} &\\
\odot\Bigg(\odot\Big(\odot\big(\mathcal{A},\mathcal{B}\big),\mathcal{C}\Big),\mathcal{D}\Bigg)\ar[d]_{\odot \big(\alpha_{\mathcal{A},\mathcal{B},\mathcal{C}},1_{\mathcal{D}}\big)}\ar[ru]^{\alpha_{\odot(\mathcal{A},\mathcal{B}),\mathcal{C},\mathcal{D}}}&& \odot\Bigg(\mathcal{A},\odot\Big(\mathcal{B},\odot\big(\mathcal{C},\mathcal{D}\big)\Big)\Bigg)\\
\odot\Bigg(\odot\Big(\mathcal{A},\odot\big(\mathcal{B},\mathcal{C}\big)\Big),\mathcal{D}\Bigg)\ar[rr]_{\alpha_{\mathcal{A},\odot(\mathcal{B},\mathcal{C}),\mathcal{D}}}&  &\odot\Bigg(\mathcal{A},\odot\Big(\odot\big(\mathcal{B},\mathcal{C}\big),\mathcal{D}\Big)\Bigg)\ar[u]_{\odot(1_{\mathcal{A}},\alpha_{\mathcal{B},\mathcal{C},\mathcal{D}}\big)}}
\end{equation}
Furthermore, in light of all foregoing observations the following diagrams commute
\begin{equation}
\xymatrixrowsep{1.2cm}
\xymatrixcolsep{2cm}
\xymatrix{\odot\Big(\odot\big(\mathcal{A},\mathbb{R}\big),\mathcal{B}\Big)\ar[rr]^{\alpha_{\mathcal{A},\mathbb{R},\mathcal{B}}} & & \odot\Big(\mathcal{A},\odot\big(\mathbb{R},\mathcal{B}\big)\Big)\ar[ld]^{\odot(1_{\mathcal{A}},\lambda_{\mathcal{B}}^{-1})}\\
& \odot(\mathcal{A},\mathcal{B})\ar[lu]^{\odot(\rho_{A},1_{\mathcal{B}})} &}
\end{equation}
\begin{equation}
\xymatrixrowsep{1.2cm}
\xymatrixcolsep{3cm}
\xymatrix{\odot\Big(\odot\big(\mathcal{A},\mathcal{B}\big),\mathcal{C}\Big)\ar[r]^{\alpha_{\mathcal{A},\mathcal{B},\mathcal{C}}}\ar[d]_{\odot(\sigma_{\mathcal{A},\mathcal{B}},1_{\mathcal{C}})} & \odot\Big(\mathcal{A},\odot\big(\mathcal{B},\mathcal{C}\big)\Big)\ar[r]^{\sigma_{\mathcal{A},\odot(\mathcal{B},\mathcal{C})}} & \odot\Big(\odot\big(\mathcal{B},\mathcal{C}\big),\mathcal{A}\Big)\ar[d]^{\alpha_{\mathcal{B},\mathcal{C},\mathcal{A}}}\\
\odot\Big(\odot\big(\mathcal{B},\mathcal{A}\big),\mathcal{C}\Big)\ar[r]_{\alpha_{\mathcal{B},\mathcal{A},\mathcal{C}}}& \odot\Big(\mathcal{B},\odot\big(\mathcal{A},\mathcal{C}\big)\Big)\ar[r]_{\odot(1_{\mathcal{B}},\sigma_{\mathcal{A},\mathcal{C}})} & \odot\Big(\mathcal{B},\odot\big(\mathcal{C},\mathcal{A}\big)\Big)
}
\end{equation}
So $\mathbf{CJP}_{\!\mathscr{C}}$ is a symmetric monoidal category; whereby we abuse notation: the seven-tuple 
\begin{equation}
\Big(\mathbf{CJP}_{\!\mathscr{C}},\mathbb{R},\odot,\alpha,\lambda,\rho,\sigma\Big)
\end{equation} 
is the structure to which we refer. 
\end{proof}
\end{theorem}



\end{document}